\documentclass[11pt,english]{article}
\usepackage{color}
\usepackage{amsfonts}
\usepackage{amsmath, amsfonts, amsthm, amssymb, amscd}
\usepackage[mathcal]{euscript}
\usepackage{a4wide}
\usepackage{slashed}
\usepackage{extarrows}
\usepackage[normalem]{ulem}
\usepackage{mathrsfs,bm}
\usepackage{tikz-cd}
\usepackage[pdfusetitle]{hyperref}
\hypersetup{linktoc = all}
\hypersetup{hidelinks}
\usepackage{hyperref}
\usepackage{multirow}
\usepackage{tabularx}
\usepackage{lscape}
\usepackage{bbm} 

\newtheorem{theorem}{Theorem}[section]
\newtheorem{proposition}[theorem]{Proposition}
\newtheorem{lemma}[theorem]{Lemma}
\newtheorem{corollary}[theorem]{Corollary}
\newtheorem{definition}[theorem]{Definition}
\newtheorem{remark}[theorem]{Remark}
\newtheorem{claim}[theorem]{Claim}
\numberwithin{equation}{section} 

\newcommand \Ccal{\mathcal{C}}

\newcommand \Kcal {\mathcal K}

\newcommand \Hcal {\mathcal H}

\newcommand \Lcal {\mathcal L}

\newcommand \Ecal {\mathcal E}

\newcommand \opDirac {\mathfrak{D}} 
\newcommand \opCurve {\mathfrak{R}}

\newcommand \delb {\overline {\del}}
\newcommand \gb {\overline g}

\newcommand \Hb {\overline{H}}

\newcommand \Phib{\overline{\Phi}}
\newcommand \Psib{\overline{\Psi}}

\newcommand \sigmab{\overline{\sigma}}
\newcommand \betab{\overline{\beta}}
\newcommand \zetab{\overline{\zeta}}

\newcommand \gch{\check{g}}
\newcommand \vch{\check{v}}

\newcommand \del \partial
\newcommand \delu {\uline{\del}}

\newcommand \Hu {\uline{H}}

\newcommand \thetau{\uline{\theta}}
\newcommand \gu{\uline{g}}

\newcommand \Psiu{\uline{\Psi}}
\newcommand \Phiu{\uline{\Phi}}

\newcommand \Abf { {\bf A}}

\newcommand \Ebf { {\bf E}}

\newcommand \source{\mathrm{sour}}

\newcommand \RR{\mathbb{R}}
\newcommand \eps{\varepsilon}
\newcommand {\vep}{\varepsilon}

\newcommand {\dels}{\slashed{\del}}

\newcommand {\ebf}{ {\bf e}}

\newcommand{\diff}{\mathrm{d}}
\makeatletter
\def\hlinew#1{%
\noalign{\ifnum0=`} \fi\hrule \@height #1 \futurelet
\reserved@a\@xhline}
\makeatother

\newcommand \bei {\begin{itemize}}
\newcommand \eei {\end{itemize}}
\newcommand \be {\begin{equation}}
\newcommand \bel {\be\label}
\newcommand \ee {\end{equation}}
\newcommand \bse {\begin{subequations}}
\newcommand \ese {\end{subequations}}
\newcommand \la \langle
\newcommand \ra \rangle

\newcommand \nablas{\nabla^\ourS}
\newcommand \ih{\widehat{\imath}}
\newcommand \jh{\widehat{\jmath}}
\newcommand \kh{\widehat{k}}
\newcommand \lh{\widehat{l}}

\newcommand \Hessian {\textbf{Hess}}

\newcommand \Riem {\mathrm{\bf Riem}}
\newcommand \Ric {\mathrm{\bf Ric}}
\newcommand \ourS {\mathrm{\bf S}}

\newcommand \Fbb {\mathbb F}

\newcommand \Pbb {\mathbb P}

\newcommand \Qbb {\mathbb Q}

\newcommand \Sbb {\mathbb S} 

\newcommand \delN {\del^{\mathcal{N}}}
\newcommand \delsN {\slashed \del^{\mathcal{N}}}
\newcommand \Mscr       {\mathscr M} 

\newcommand \MME {\Mscr^{\mathcal{EM}}}  

\newcommand \E {\mathcal{E}} 

\newcommand \ME {\mathcal{EM}} 
\newcommand \EM {\mathcal{EM}} 
\renewcommand \H {\mathcal H} 

\newcommand \M {\mathcal M} 
\newcommand \Mcal {\mathcal M} 

\newcommand \N {\mathcal N} 

\newcommand \delts  \delsN

\newcommand \Ncal {\mathcal N}

\newcommand \near {\textbf{near}}

\newcommand \xsans {\widetilde x}
\newcommand \rsans {\widetilde r}
\newcommand \Phisans {\widetilde \Phi}
\newcommand \Psisans {\widetilde \Psi}
\newcommand \delsans {\widetilde \del}
\newcommand \etasans {\widetilde \eta}
\newcommand \gsans {{\widetilde g}}
\newcommand \tsans {{\widetilde t}}

\newcommand \Hsans {\widetilde H}
\newcommand \Zsans {\widetilde Z}

\newcommand \Tsans {\widetilde T}
\newcommand \xavec {x}
\newcommand \ravec {r}

\newcommand \delavec {\del}

\newcommand \tavec {t}

\newcommand \Havec {H}

\newcommand \lapsb {\overline{l}}

\newcommand \ord {\mathrm{ord}}
\newcommand \rank{\mathrm{rank}}
\newcommand \gref{{g_\mathrm{r}}}
\newcommand \Href{{H_{\mathrm{r}}}}
\newcommand \Hreff{{h_{\mathrm{r}}}}
\newcommand \greft{\widetilde{g_{\mathrm{r}}}}
\newcommand \Hreft{\widetilde{H_{\mathrm{r}}}}
\newcommand \Lsans{\widetilde{L}}
\newcommand \Omegasans{\widetilde{\Omega}}
\newcommand \Wsans{\widetilde{W}}
\newcommand \Gammasans{\widetilde{\Gamma}}

\newcommand \epss{{\eps^{\star}}}
\newcommand \Mks{{\mathrm{Ms}}}
\newcommand \Kscr{\mathscr{K}}
\newcommand \Ssans{\widetilde{S}}
\newcommand \Scal{\mathcal{S}}

\newcommand \Rbb{\mathbb{R}}
\newcommand \Hcom{{\bf C}^+}
\newcommand \Hwave{{\bf W}}
\newcommand \HDirac{{\bf D}}

\newcommand \Cubic{{\bf C}}
\newcommand \Quart{{\bf Qt}}
\newcommand \epsm{{\eps_s}}
\newcommand \Kref{{K_{\mathrm{r}}}}
\newcommand \ourD{\mathrm{D}}
\newcommand \delS{\del^{\Scal}}
\newcommand \delSs{\slashed{\del}^{\Scal}}
\newcommand \PhiS{{\Phi^{\Scal}}}
\newcommand \PsiS{{\Psi^{\Scal}}}
\newcommand \gScal{{g_{\Scal}}}
\newcommand \HScal{{H_{\Scal}}}
\newcommand \Phihat{{\slashed\Phi}}
\newcommand \Psihat{{\slashed\Psi}}

\begin{document}

\title{The Global Nonlinear Stability of Minkowski Spacetime 
\\
with Self-Gravitating Massive Dirac Fields}

\author{Philippe G. LeFloch\footnote{Laboratoire Jacques-Louis Lions, Sorbonne Universit\'e \& Centre National de la Recherche Scientifique, 4 Place Jussieu, 75252 Paris, France. 
\newline 
Email: {\tt contact@philippelefloch.org}, {\tt weidong.zhang@sorbonne-universite.fr}} \hskip.01cm, 
Yue Ma\footnote{School of Mathematics and Statistics, Xi'an Jiaotong University, Xi'an, Shaanxi 710049, P.R. China.
\newline 
E-mail: {\tt yuemath@xjtu.edu.cn, zwd13892650621@stu.xjtu.edu.cn}.
\newline
{\it Keywords and Phrases.} Einstein equation; massive Dirac field; nonlinear stability; gauge-invariance; light-bending coordinates; Euclidean--hyperboloidal method.} 
\hskip.01cm, 
and Weidong Zhang$^{\ast, \dag}$
}

\date{October 2025}

\maketitle

\begin{abstract}
We consider the Einstein-Dirac system for a massive field, which describes the evolution of self-gravi\-tating massive spinor fields, and we investigate the corresponding \emph{global evolution problem}, when the initial data set is sufficiently close to data describing a spacelike, asymptotically Euclidean slice of the (vacuum) Minkowski spacetime. We establish the \emph{gauge-invariant nonlinear stability} of such fields, namely the existence of a globally hyperbolic development, which remains asymptotic to Minkowski spacetime in future timelike, null, and spacelike directions. Previous results on this problem have been limited to the Einstein-Dirac system in the massless case. Our analysis follows the asymptotically hyperboloidal-Euclidean framework introduced by LeFloch and Y.~Ma for the massive Klein-Gordon-Einstein system. The structure specific to spinor fields and the Dirac equation necessitates significantly new elements in the proof, which are developed here. In contrast with prior approaches in the literature, our treatment of spinor fields and the Dirac equation on curved spacetimes is \emph{fully gauge-invariant,} relying on the formalism of Lorentz Clifford algebras and principal fiber bundles. The core of our analysis is carried out with the spacetime metric expressed in \emph{light-bending wave coordinates}, as we call them. This leads us to the study of a global existence problem for a coupled system of second-order wave equations with constraints (for the spacetime metric) and first-order Klein-Gordon-type equation for the spinor field. We derive $L^2$ estimates for the Dirac equation and its coupling with the Einstein equations, along with $L^\infty$ pointwise estimates. New Sobolev inequalities are introduced to handle spinor fields  in a gauge-invariant manner in the hyperboloidal–Euclidean foliation. The nonlinear coupling between the massive Dirac equation and the Einstein equations is thoroughly investigated, and we establish a hierarchy of bootstrap estimates, which carefully distinguish between the order of differentiation with respect to translations, spatial rotations, and Lorentz boosts.
\end{abstract}

\vfill


\clearpage 


{\small
 
\setcounter{secnumdepth}{2} 
\setcounter{tocdepth}{2}
\tableofcontents

} 


\clearpage 


\section{Introduction}
\label{section=N1}

\subsection{Main objective} 
\label{section=1-1}

\paragraph{Aim of this Monograph.}

We investigate here Einstein's field equation of general relativity, coupled to a massive spinor field satisfying the Dirac equation. We focus on the \emph{global evolution problem} from initial data given by small perturbations of a spacelike, asymptotically Euclidean slice of the (vacuum) Minkowski spacetime. Our main result establishes the nonlinear stability of self-gravitating, massive spinor fields, showing that the coupled Einstein-Dirac system admits a globally hyperbolic development, which remains asymptotically Euclidean in all causal directions---timelike, null, and spacelike.

Whereas earlier works in the literature address only the \emph{massless} case, the presence of mass entails novel analytical challenges, in particular for controlling the long-time behavior of solutions and the intricate interactions between geometry and matter. To overcome these difficulties, we extend the \emph{asymptotically hyperboloidal--Euclidean method} developed by two of the authors for the Einstein--Klein--Gordon system; cf.~LeFloch and Ma~\cite{PLF-YM-CRAS}--\cite{PLF-YM-PDE}. Independently, the global evolution problem for the Einstein--Klein--Gordon system was also solved by Ionescu and Pausader~\cite{IP,IP-two,IP3} by introducing a Fourier-based method.  Wang~\cite{Wang} proposed a geometric approach for solutions having a Schwarzschild exterior. More recently, Chen and Zhou \cite{ChenZhou} studied Einstein-Klein-Gordon spacetimes with $U(1) \times \RR$ symmetry. 

The method introduced in the present monograph  incorporates the algebraic and differential structures specific to spinor fields and exploits a formulation in terms of Lorentzian Clifford algebras and principal fiber bundles, thereby ensuring a fully gauge-invariant treatment of the spinor field throughout the analysis. 


\paragraph{Main techniques.}

A key element of our method in~\cite{PLF-YM-PDE} is the introduction of a geometrically motivated condition---referred to as the \emph{light-bending wave coordinates}---which is obtained here by a choice of coordinates
motivated by the earlier work of Kauffman and Lindblad~\cite{KauffmanLindblad}. In these suitably defined wave coordinates, the Einstein-Dirac system takes the form of a coupled system of second-order nonlinear wave equations (governing the metric components) and a first-order Dirac equation (governing the spinor field), together with algebraic and differential constraints. Within this framework, we handle a rather general class of spacetime metrics and analyze the decay and dispersion of solutions along the hyperboloidal--Euclidean foliation. In short, we establish a comprehensive set of energy and pointwise estimates for the spinor and metric components, adapted to the gauge-invariant structure of Dirac fields. In particular, we derive new Sobolev-type inequalities specifically adapted to the spinorial setting and compatible with the asymptotic geometry of the foliation. A bootstrap argument is then developed, in which we carefully track the behavior of the differentiated spinor field under spacetime symmetries---including translations, rotations, and Lorentz boosts---thereby controlling a full hierarchy of derivatives. The outcome is a geometric and analytic framework for self-gravitating fields in curved spacetimes, yielding the first global existence and stability theorem for massive spinor fields coupled to gravity.


\paragraph{Physical motivations.}

The Einstein-Dirac system arises naturally in the context of general relativity and quantum field theory, where it models the gravitational interaction of spin-$\tfrac{1}{2}$ particles---known as \emph{fermions}. These include physically fundamental entities such as electrons, protons, and neutrons, all of which are described by the Dirac equation, a relativistic wave equation incorporating spin and mass. When coupled to gravity, such fermionic fields influence and are influenced by the curvature of spacetime, giving rise to a highly nonlinear and geometrically rich system.

Among spinor fields, those with non-zero mass are of particular physical relevance. While massless spinors (e.g., neutrinos in idealized models) travel along null geodesics and interact weakly with geometry, massive fermions are central to both particle physics and astrophysical phenomena. Their gravitational self-interaction can lead to significant nonlinear effects, including spacetime collapse, which are crucial in the modeling of compact objects such as neutron stars. Understanding the global behavior of these systems is therefore essential for formulating a coherent theory of quantum matter on curved backgrounds and for exploring the interplay between geometry and quantum field theory.


\paragraph{Related works.}

We will not attempt to review here the vast literature on the subject. First of all, the stability of Minkowski spacetime for the {\sl vacuum} Einstein equations was established geometrically by Christodoulou and Klainerman~\cite{CK} (and then reviewed in Bourguignon \cite{Bourguignon}), while an alternative proof based on a formulation in wave coordinates was given later on by Lindblad and Rodnianski~\cite{LR1,LR2}. Solutions with lower decay at spacelike infinity were constructed by Bieri~\cite{Bieri} (and \cite{BieriZipser}), while the most recent contributions include the work of Hintz and Vasy~\cite{HintzVasy1,HintzVasy2} and Shen~\cite{Shen} (minimal decay assumptions). 

On the other hand, dealing with \textsl{self-gravitating matter fields} has been undertaken only in recent years. The stability of self-gravitating Klein--Gordon fields was proven independently by LeFloch--Ma~\cite{PLF-YM-PDE} and by Ionescu and Pausader~\cite{IP3}, the latter work being based on the notion of spacetime resonances and Fourier-based techniques. Recent work on other types of matter fields includes, among others, contributions by Bigorgne~\cite{Bigorgne}, Bigorgne et al.~\cite{Bigorgne2}, Fajman et al.~\cite{FJS3}, and Lindblad et al.~\cite{LTay}. More recently, the \emph{massless} Dirac equation was studied by Chen~\cite{ChenXuantao}, using a coordinate-based formulation of the Dirac equation. 

Furthermore, as far as the Euclidean--hyperboloidal method is concerned, we point out the contribution by Dong, LeFloch, and Wyatt~\cite{DIPP} on the global evolution of the U(1) Higgs boson model. Extensive, now classical, work is also available on the Dirac equation on a fixed spacetime; cf.~for instance Bachelot~\cite{Bachelot88}, which contains an analysis of the Cauchy problem. For further results on the Dirac equations, we refer the reader to Bejenaru--Herr~\cite{Bejenaru-Herr-2017} (scattering for Dirac-Klein-Gordon), Wang \cite{X.Wang-2015} (critical case), Cai-Dong-Li-Zhao \cite{Cai-Dong-2024} and Dong-Li-Yuan \cite{Dong-Li-Zhao-2023} (large data), Dong-Li \cite{Dong-Li-2022}, Dong-Wyatt \cite{DW} and Dong-Li-Ma-Yuan~\cite{Dong-Li-Ma-Yuan-2024} (lower dimension), Dong-Li-Yuan \cite{Dong-Li-Yuan-2023} (uniform energy bound), Ge et al. \cite{Ge-Jiang-Wang-Zhang-Zhong} (Reissner-Nordstrom-de Sitter geometry), Jia-Li~\cite{Jia-Li-2024} (Dirac radiation field), Ma-Zhang \cite{Ma-Zhang-2022} (Schwarzschild geometry), R\"oken \cite{Roken-2019} (Kerr geometry), Wang-Zhang \cite{Wang-Zhang-2018} (Kerr-Newman-AdS geometry), Zhang-Zhang \cite{Zhang-Zhang-2024} (Kerr-Newman geometry), and Zhao-Wu \cite{Zhao-Wu-2025} (characteristic initial value problem). 


\subsection{Formulation of the problem}

\paragraph{Spinor fields.}

Let us briefly outline the spinorial formalism adopted in this work, while precise definitions are provided in Part~\ref{part-one}, below. We consider a time-oriented, orientable Lorentzian manifold $(\mathcal{M}, g)$ of signature $(-,+,+,+)$, which is globally diffeomorphic to $\RR^+ \times \RR^3$ and admits a foliation by spacelike hypersurfaces $\Sigma_t$ diffeomorphic to $\RR^3$. The manifold $\mathcal{M}$ carries a spin structure $\mathrm{Spin}^+(\Mcal)$, and we introduce the spinor bundle 
\be
\ourS(\mathcal{M}) := \mathrm{Spin}^+(\mathcal{M}) \times_{\kappa} \mathbb{C}^4,
\ee
associated with the Dirac representation $\kappa$ of the Lorentz spin group.

The covariant derivative operator associated with the Lorentz metric $g$ is denoted by $\nablas$ 
( but later in the following discussion, for simplicity of notation it is also denoted by $\nabla$ when there is no risk of ambiguity). A spinor field is then a global section $\Psi: \mathcal{M} \to \ourS(\mathcal{M})$, and the Dirac operator is defined by
\begin{equation}
\opDirac \Psi := g^{\mu\nu} \del_{\mu}\cdot\nablas_\nu \Psi,
\end{equation}
where $\del_{\mu}\cdot$ denotes Clifford multiplication by the basis vectors. More explicitly, in a local orthonormal frame $\{e_{\alpha}\}$ (referred to as a ``tetrad'' or a local gauge), the spinorial covariant derivative acting on the ``coordinate'' of $\Psi$ with respect to $\{e_{\alpha}\}$ is given by (cf. \cite{Hamilton-2017})
\begin{equation} \label{eq3-21-11-2024-Z}
\aligned
\nablas_{e_{\alpha}} \psi =&  \diff \psi(e_{\alpha}) - \tfrac{1}{4} \,\omega_{\mu ab} \,\Sigma^{ab} \psi,
\endaligned
\end{equation}
where $\omega_{\mu ab} = g(e_a, \nabla_\mu e_b)$ is the spin connection associated with the Levi-Civita connection $\nabla$ of the metric $g$, and $\Sigma^{ab}$ are the generators of the spin representation. Upon fixing a local gauge, we may equivalently express the Dirac equation in terms of the standard $4 \times 4$ gamma matrices $\gamma^\mu$, with the covariant derivative defined so as to respect the spinor structure. However, throughout our analysis, we choose to avoid doing so and develop our geometric analysis in a fully gauge-invariant framework.


\paragraph{The Einstein-Dirac system for a massive field.}

We are interested in the Einstein equations coupled to a massive spinor field $\Psi$ of mass $M \geq 0$. This system takes the form
\begin{equation}
\label{main system}
\begin{aligned}
R_{\mu\nu} - \tfrac{1}{2} R g_{\mu\nu} & = T[\Psi]_{\mu\nu}, 
\\
\opDirac \Psi + \mathrm{i}M \Psi & = 0,
\end{aligned}
\end{equation}
where the energy--momentum tensor $T_{\mu\nu}$ depends on the spinor field and its derivatives, and is given by 
\begin{equation}\label{equa-tensorTmunu-0}
\aligned
T[\Psi]_{\mu\nu} 
=& 
\frac{\mathrm{i}}{4} \Big(\la\Psi,\del_\mu \cdot\nablas_\nu \Psi \ra_{\ourD}
+ 
\la\Psi,\del_{\nu} \cdot\nablas_{\mu} \Psi \ra_{\ourD} \Big)
- 
\frac{\mathrm{i}}{4} \Big(\la \del_\mu \cdot\nablas_\nu \Psi,\Psi\ra_{\ourD}
+
\la \del_{\nu} \cdot\nablas_{\mu} \Psi,\Psi\ra_{\ourD}
\Big).
\endaligned
\end{equation}
 It is important to note that this system is fully covariant, both under spacetime diffeomorphisms and under local Lorentz transformations. Moreover, the conservation law
\be
\nabla^\mu T_{\mu\nu} = 0
\ee
follows directly from the contracted Bianchi identities together with the Dirac equation.  


\paragraph{Global Cauchy developments.} 

We are interested in the Cauchy problem when a suitable initial data set is prescribed on a spacelike hypersurface. Namely, we prescribe 
\[
\bigl(\Sigma,\overline{g}, \overline{k}, \overline{\Psi} \bigr),
\] 
consisting of a three-dimensional Riemannian manifold $(\Sigma,\overline g)$, endowed with a symmetric two-tensor field $\overline{k} = (\overline{k}_{ij})$, and a spinor field $\overline{\Psi}$. The tensor $\overline{k}$ stands for the second fundamental form of the initial slice. The initial data set must satisfy the Einstein constraint equations on $\Sigma$, namely the Hamiltonian and momentum constraints obtained by projecting the Einstein equations along and orthogonal to~$\Sigma$, together with the compatibility conditions imposed by the Dirac equation on the initial spinor field.

We seek a solution to the Einstein-Dirac system consisting of a four-dimensional Lorentzian manifold $(\mathcal{M},g)$, endowed with a smooth embedding $\iota:\Sigma\hookrightarrow\mathcal{M}$, such that $\overline{g}$ coincides with the induced Riemannian metric $\iota^\ast g$ on $\Sigma$, $\overline{k}$ coincides with the second fundamental form of $\Sigma$ in $(\mathcal{M},g)$, and the spinor field $\Psi$ on $\mathcal{M}$ restricts to $\overline{\Psi}$ on~$\Sigma$ under a fixed identification of the spin structures. The resulting spacetime must satisfy the coupled equations \eqref{main system} everywhere in $\mathcal{M}$.

Recall that a \emph{Cauchy development} of $(\Sigma,\overline g,\overline{k},\overline{\Psi})$ is a spacetime $(\mathcal{M},g,\Psi)$, together with an embedding $\iota$, enjoying the above properties and such that $\iota(\Sigma)$ is a Cauchy hypersurface for $(\mathcal{M},g)$, that is, every inextendible causal curve in $\mathcal{M}$ intersects $\iota(\Sigma)$ exactly once. A development is said to be \emph{maximal} if it cannot be extended to a strictly larger Cauchy development of the same initial data, in the sense that there is no isometric embedding into another development which preserves the initial hypersurface and extends both the metric and the spinor field.

Our objective, under suitably smallness conditions on the initial data set, is to construct a \emph{global future Cauchy development}, namely the (possibly maximal) Cauchy development in which the evolution of the initial data is defined for the entire future of $\Sigma$, without encountering singularities or breakdown of the gauge conditions. Specifically, our aim is to construct global Cauchy developments for \emph{small perturbations} of Minkowski initial data, and prove that the resulting spacetime is geodesically complete. Indeed, we will show that it remains asymptotically Euclidean along all causal directions. 


\subsection{Reference, coordinate transform and formulation into PDE system}

As we are interested in the perturbation problem, we work near the flat Minkowski spacetime, and we can assume that $\Mcal$ is covered by a globally-defined coordinate chart $\{\xsans^{\alpha}\}$. We work in the domain
\be
\Mcal^+ \simeq \RR_+^{1+3} = \{(\tsans,\xsans)|\xsans\in\RR^3,\tsans \geq 1\}.
\ee
As in previous works, the perturbation under consideration is not defined directly with respect to the Minkowski metric, but with respect to a \emph{reference metric} which carries the \emph{ADM mass} of the spacetime. For example in \cite{LR2} the reference is essentially the Schwarzschild metric. In the present Monograph, our notion of {\sl admissible reference metric} is introduced in Section~\ref{subsec1-20-sept-2025} ---which is comparable with the so-called {\bf Class B} proposed in~\cite{PLF-YM-PDE}. However, in contrast, we \underline{do not assume a priori} that the {\sl light-bending condition} is satisfied. This setup encompasses one of the most significant case, namely, the approximate vacuum solution to the Einstein equations constructed in \cite{LR2}; in other words, we are considering massive perturbations of the vacuum solution. Observe that the notion of admissible reference involve four parameters $(N,\theta,\epss,\ell)$, which correspond to the regularity, decay rate, amplitude, and special decaying rate near the light-cone, respectively.

As mentioned above, we perform a coordinate transformation
\be
(\tsans,\xsans)\rightarrow(t,x), 
\ee
(defined in Proposition~\ref{prop1-01-oct-2025}, below), after a proposal made by Kauffman and Lindblad \cite{KauffmanLindblad}. Crucially, we establish that we can always guarantee that the reference in the nex coordinates satisfies the light-bending reference; see Proposition~\ref{prop1-01-oct-2025}. Denoting by $g$ the spacetime metric and $\gref$ the reference metric, we consider 
\be
u_{\alpha\beta} := g_{\alpha\beta} - \gref_{\alpha\beta}, 
\ee
where $g_{\alpha\beta},\gref_{\alpha\beta}$ are the components in the new coordinates $(t,x)$. Then we formulate the Einstein-Dirac system \eqref{main system} in the form 
\begin{subequations}\label{eq1-22-oct-2025}
\begin{equation}
\aligned
g^{\mu\nu}\del_{\mu}\del_{\nu}u_{\alpha\beta} 
& =\Fbb(g,g;\del u,\del u)_{\alpha\beta} 
+ u^{\mu\nu}\del_{\mu}\del_{\nu}\gref_{\alpha\beta} 
+ \Sbb(g,\del g,\Psi,\del\Psi)_{\alpha\beta}
\\
& \quad + 2\Fbb(g,g;\del\gref,\del u)_{\alpha\beta} 
+ 2\Fbb(u,\gref;\del \gref,\del \gref)_{\alpha\beta} 
+ \Fbb(u,u;\del\gref,\del\gref)_{\alpha\beta}
\\
& \quad -2\big(\mathbb{G}(g,\del g,\Gamma,\del\Gamma)_{\alpha\beta} - \mathbb{G}(\gref,\del \gref,\Gamma_{\mathrm{r}},\del\Gamma_{\mathrm{r}})_{\alpha\beta}\big),
\\
\opDirac\Psi + \mathrm{i}M\Psi & = 0,
\endaligned
\end{equation}
with the generalized coordinate condition
\begin{equation}\label{eq2-22-oct-2025}
-\Box_g x^{\lambda} = \Gamma^{\lambda}  
= {\Gamma_{\mathrm{r}}}^{\lambda} 
+ \del_{\alpha}\big(\Phisans_{\beta}^{\delta}\big)\Phisans_{\delta}^{\lambda} u^{\alpha\beta}.
\end{equation}
\end{subequations}
For further details, we refer to Section~\ref{subsec1-22-oct-2025}. Here, ${\Gamma_{\mathrm{f}}}^{\lambda}$ represents the wave coordinate remainder of the reference spacetime. Roughly speaking, the terms associated with $\Fbb$ are nonlinear terms from the Einstein equations and, more precisely, from the Ricci curvature. The term $\Sbb$ contains the source terms of $\Psi$ acting on the spacetime. The last term contains the errors from the coordinate transformation and the generalized wave coordinates condition.

\subsection{Statement of the nonlinear stability of self-gravitating Dirac spinor fields}

The system \eqref{eq1-22-oct-2025} is hyperbolic and quasilinear in nature, thus admits local solutions when the initial data is given on a spacelike slice and sufficiently small in certain Sobolev norms. The main result of the present work is that, with sufficient smallness conditions, the local solutions extends to time, space and null infinity. Due to the lack of notation at this stage, we only state a semi-quantitative version (quantitative on the metric) and postpone the complete quantitative statement to Section~\ref{section=N15}; cf.~Theorem~\ref{subsec2-22-oct-2025-theo}. 

\begin{theorem}[Nonlinear stability of the Minkwoski spacetime with massive Dirac field, Qualitative statement]
Assume that $\gref$ is a $(N,\theta,\epss,\ell)$ admissible reference with sufficiently large $N$, sufficiently small $\theta,\epss$ and $\ell\in (0,1/2)$. Assume that the spacetime metric $g$ and the spinor field $\Psi$ satisfies the following smallness conditions on the initial slice $\{\tsans=1\}$,
\begin{equation}
\sum_{|I|\leq N} \big\| \la \rsans\ra^{\kappa + |I|} \big( |\delsans_x^I\delsans_x u_0| + |\delsans_x^I \delsans_tu_0| \big) \big\|_{L^2(\RR^3)}\leq \eps, 
\quad u_0=u\big|_{\{\tsans=1\}},\quad\kappa\in(1/2,3/4),
\end{equation}
\begin{equation}
\big\|\Psi|_{\{\tsans=1\}}, \nabla_{\vec{n}}\Psi\big|_{\{\tsans=1\}}\big\|_{N,\mu}\leq \eps,\quad \mu\in(3/4,1)
\end{equation}
where $\|\cdot\|_{N,\mu}$ represents a Sobolev norm of order $N$ with a weight $\la\rsans\ra^{\mu}$ to be specified in Section~\ref{subsec2-22-oct-2025}. Here $\vec{n}$ represents the future oriented unit normal vector of the slice $\{\tsans=1\}$. Then the maximal globally hyperbolic Cauchy development of the data is future causally geodesically complete, remains asymptotically flat in all causal directions. 
\end{theorem}


\paragraph{General strategy of proof.}

Our approach builds on the framework developed earlier by two of the authors for the Einstein--Klein--Gordon system~\cite{PLF-YM-CRAS}--\cite{PLF-YM-PDE}. However, substantial new elements are required, and introduced here, to handle massive spinor fields via a fully gauge-invariant formulation. 

The proof proceeds in generalized wave coordinates satisfying \eqref{eq2-22-oct-2025} chosen so as to enforce the light-bending property. In short, the metric evolution is analyzed via a hierarchy of weighted energy estimates adapted to the Euclidean--hyperboloidal foliation, while the Dirac equation is treated directly in its geometric, gauge-invariant form, thus avoiding any fixed choice of gamma-matrix representation and/or tetrad. This requires a careful study of the spin connection coefficients in the hyperboloidal frame, as well as commutator estimates for the coupled system. 

The analysis is closed by a bootstrap argument in which the metric and spinor estimates are coupled: the metric bounds rely on the \emph{quasi-null structure} and provide the decay needed for the spinor field, while the spinor energy--momentum tensor is shown to satisfy improved decay in suitable directions, thereby contributing to the stability of the geometry. In turn, we establish uniform bounds on all commuted fields, proving global existence, completeness, and the asymptotic approach to Minkowski spacetime. Specific estimates for the metric and the spinor field will be stated. Since we follow the main lines of proof in LeFloch--Ma~\cite{PLF-YM-PDE}, we will not repeat here the full set of metric estimates, and will instead focus most of our analysis on the spinor field and its coupling with and its effect on the geometry. 


\subsection{Structure of this Monograph}

\paragraph{Part I.}

Sections~\ref{section=N2} to \ref{section=N10} build the geometric and analytic framework needed to treat Dirac spinors on a Lorentzian manifold in a \emph{frame-independent} fashion and establish fundamental commutator and decay estimates. Starting from the real Clifford algebra $\mathrm{Cl}(3,1)$ and the spin group $\mathrm{Spin}_{3,1}^+$, global orthonormal frames give rise to spin bundles, and different global frames on $\RR^4$ lead to \emph{isomorphic} bundles. The resulting \emph{spin structure} ---the isomorphism class of these bundles ---does not depend on the chosen tetrad. On the principal-bundle side, an Ehresmann connection on the orthonormal frame bundle induces a connection on every associated bundle, in particular on the spin bundle; the Levi--Civita connection lifts canonically and yields a \emph{spin connection}. Passing to the associated Dirac bundle provides a covariant derivative for spinors and, consequently, a gauge-invariant formulation of the Dirac operator and the Einstein-Dirac equations. Our analysis in the present work treates the spinor field intrinsically, so that changes of frame (local or global) amount only to conjugation by the canonical bundle isomorphisms.

In Part~\ref{part-one}, we also set up the Euclidean--hyperboloidal foliation and associated frames,  weights, and ``good'' derivatives in Section~\ref{section=N4}. We work with a spacelike slicing adapted to the structure of wave-type equations in timelike, space-like, and null directions. We also introduce the energy functionals and hyperboloidal norms tailored to spinor fields (Section~\ref{section=N5}), and we establish Sobolev inequalities on hyperboloidal slices, on order to eventually reach  sharp $t-r$ decay from weighted energies 
(Section~\ref{section=N6}). In addition, we establish pointwise bounds for massive Dirac fields on light-bending backgrounds (Section~\ref{section=N7}), and develop high-order vector-field operators and commutator calculus suited to the Dirac operator (Section~\ref{section=N8}). A refined decomposition of the 
Dirac commutator in the near-light-cone region (Section~\ref{section=N9}) captures the favorable structure needed to control certain borderline terms, and leads to improved pointwise estimates for massive spinors (Section~\ref{section=N10}). Collected together, these results deliver a robust decay and commutation framework, which should be applicable to a broad set of evolution problems. 

\paragraph{Part II.}

Part~\ref{part-two}  reformulates the coupled Einstein-Dirac system in a coordinate system
tailored to be able to conveniently analyze the decay and (quasi-)null structure of the metric and spinor field, via  a global bootstrap scheme. Section~\ref{section=N11} gathers preliminary identities and structural bounds needed for nonlinear energy estimates. Section~\ref{section=N12} analyzes source terms in the hyperboloidal region, isolating the quadratic--cubic decomposition and identifying contributions with favorable decay. In Section~\ref{section=N13} we construct a light-bending, almost wave-coordinate chart and transfer the asymptotic flatness and gauge constraints to the new variables, establishing key bounds on the contracted Christoffel symbols. Section~\ref{section=N14} recasts the full Einstein-Dirac system in generalized wave coordinates, exposing null structures and commutator cancellations. Sections~\ref{section=N15}--\ref{section=N16}
then lay out the global strategy: a hierarchy of hyperboloidal energies, pointwise decay via Sobolev inequalities on the foliation, and a closed bootstrap controlling metric and spinor components in both the interior and exterior regions.

\paragraph{Part III.} Part~\ref{part-three} collects the geometric and analytic background for spin geometry and the technical tools used in the main proofs. Section~\ref{section=N17} reviews the spin structure, Clifford algebra conventions, and basic identities. Section~\ref{section=N18} develops gauge-invariant differentiation for spinors
and its compatibility with the Dirac operator and curvature. Sections~\ref{section=N19}
and~\ref{section=N20} provide additional material on the gauge-invariant framework and detailed properties of the spinorial connection, including covariant commutators. Section~\ref{section=N21} contains the proof of a key auxiliary lemma. Section~\ref{section=N22} establishes the uniformly-spacelike
character of the hyperboloidal slices. Sections~\ref{section=N23}--\ref{section=N24}
gather high-order product/commutator estimates and refined Sobolev embeddings on the foliation. This supplies the algebraic and geometric estimates repeatedly invoked in Parts~\ref{part-one} and \ref{part-two}.



\subsection{Summary of notation}
\label{section-1--4}

We summarize our main notation in three tables: coordinates and frames; metrics and perturbations; and
vector fields and operators.

\begin{table}[h!]
\centering
\caption{Notation: \bf coordinates and frames.}
\begin{tabular}{@{}ll@{}}
\\
$\{\xsans^{\alpha}\}$ & original coordinates with derivatives $\delsans_{\alpha} := \frac{\partial}{\partial \xsans^{\alpha}}$
\\
$\{\xavec^{\alpha}\}$ & light--bending coordinates  with $t=\tsans$, $x^a=(r/\rsans)\xsans^a$
\\
& and $r=\rsans-\Kref\epsm\,\chi(\rsans/\tsans)\,\rsans^{\theta}$
\\
$\Psisans_{\beta}^{\alpha}=\frac{\partial \xavec^{\alpha}}{\partial \xsans^{\beta}}$, $\Phisans=\Psisans^{-1}$ 
& transition matrices with $\del_{\alpha}=\Phisans_{\alpha}^{\beta}\delsans_{\beta}$
\\
$\chi$ & cutoff function
\\
$\delu_a$,\; $\del^{\Ncal}_a$ & tangential (``good'') derivatives in semi-/null frames 
\\
$\Mcal_s$ & slice of the Euclidean--hyperboloidal foliation
\\
$\Mcal^{\Hcal}_{[s_0,s_1]}$, $\Mcal^{\ME}_{[s_0,s_1]}$ & hyperboloidal / exterior slab
\end{tabular}
\\
\end{table}

\begin{table}[h!]
\centering
\caption{Notation: \bf metrics and perturbations.}
\begin{tabular}{@{}ll@{}}
\\
$\eta_{\alpha\beta}$ & Minkowski metric 
\\
$g_{\alpha\beta}$,\; $\gref_{\alpha\beta}$ & unknown and reference metrics
\\
$\Hsans_{\alpha\beta}:=\gsans_{\alpha\beta}-\etasans_{\alpha\beta}$ 
& perturbation in $\{\xsans^{\alpha}\}$
\\ 
$\Hreff_{\alpha\beta}:=\gref_{\alpha\beta}-\eta_{\alpha\beta}$ 
& reference perturbation in $\{x^{\alpha}\}$
\\
$\H^{\Ncal 00}$ & null component $(\diff t-\diff r,\diff t-\diff r)_{\;\cdot}$
\\
$\Gamma^{\lambda}:=g^{\alpha\beta}\Gamma_{\alpha\beta}^{\lambda}$ 
& contracted Christoffel symbols for $g$
\\
\end{tabular}
\end{table}

\begin{table}[h!]
\centering
\caption{Notation: \bf vector fields and operators.}
\begin{tabular}{@{}ll@{}}
\\
$\Kscr=\{\del_{\alpha},L_a,\Omega_{ab},S\}$,\; $\widetilde{\Kscr}$ & Extended admissible families in $\{x^{\alpha}\}$ and $\{\xsans^{\alpha}\}$
\\
$\mathscr{Z}=\{\del_{\alpha},L_a,\Omega_{ab}\}$,\; $\widetilde{\mathscr{Z}}$ &  Admissible families in $\{x^{\alpha}\}$ and $\{\xsans^{\alpha}\}$
\\
$L_a =x^a\del_t+t\del_a$;\; $\Omega_{ab} =x^a\del_b-x^b\del_a$;\; $S =t\del_t+x^a\del_a$ & boosts, rotations, scaling
\\
$\Kscr^I,\mathscr{Z}^I$ of type $(p,k)$ & Extended/Admissible operator of order $p$ and rank $k$ 
\\
$\opDirac$,\; $\widehat{Z}$ & Dirac operator and Clifford--adapted derivative
\\
$[\widehat{Z},\opDirac]$;\; $\pi[Z]$ & commutator and deformation tensor 
\\
$[\,\cdot\,]_{p,k}$,\; $|\cdot|_{p,k}$ & weighted norms 
\\
\end{tabular}
\end{table}


\clearpage 

\part{Gauge-invariant Euclidean--hyperboloidal framework for spinor fields} 
\label{part-one}

\section{Spinorial calculus with the Dirac operator} 
\label{section=N2}

\subsection{ Spinorial covariant derivative}
\label{section===41}

{

\paragraph{Aim of this section.}

We now summarize the calculus rules, and refer to Sections~\ref{section=N17} and~\ref{section=N18} for standard material on spinor fields in a curved spacetime.  For the reader's convenience, although this is a standard result, we present a construction of the Ehresmann connection in~Section~\ref{Appendix--B5}. All of our results concerning spinor fields are independent of the choice of tetrad. However, for the derivation of certain identities and estimates, occasionally it will be convenient to fix a tetrad and deal with the coordinate representation $\psi$ of a spinor $\Psi$, while all of our final identities remain independent of the choice of gauge, according to the general theory in Section~\ref{section=N17}. 

Importantly, the notation and identities introduced in the present section will be used throughout the rest of this work. We use the same symbol $\nabla$ to denote both the covariant derivatives of tensor fields and the spinorial derivatives of spinor fields. This should not lead to confusion, as we typically use capital Latin letters $X, Y, \ldots$ for vector fields and Greek letters $\Psi$ (or $\psi$) for spinor fields (or their coordinates). Recall that the Clifford product is denoted simply by a dot.


\paragraph{Covariant derivative.}

As we have shown in and after Proposition~\ref{proposition-JD44}, it suffices to define the spin structure (and associated spinor bundle, spinorial covariant derivative, Dirac form, etc.) with respect to a single tetrad, as commonly presented in textbooks. For the reader's convenience, we include a brief construction based on the horizontal distribution in Section~\ref{section=N19}.  An explicit construction based on the Levi-Civita connection form can be found in \cite[Sec. 6]{Hamilton-2017}. We now state several properties of the spinorial covariant derivative, two of which follow directly from the general definition \eqref{eq3-08-jan-2025}:
\begin{subequations}
\begin{equation}
\nabla_X(f\Psi)  = \mathrm{d}f(X)\Psi + f\nabla_X\Psi,
\end{equation}
\begin{equation}
\nabla_{fX} \Psi  = f\nabla_X\Psi.
\end{equation}
These two properties hold true for all covariant derivatives defined on vector bundles.
\end{subequations}
In view of the standard construction of the spin connection, the following properties hold specifically for the spinorial covariant derivative and the Dirac form:
\begin{subequations}
\begin{equation}
\nabla_{X}(Y\cdot\Psi)  = \nabla_XY + Y\cdot\nabla_X\Psi,
\end{equation}
where $X,Y$ are vector fields and $\nabla_XY$ denote the Levi-Civita derivative. This is known as the {\it compatibility property with respect to the Levi-Civita connection}(cf.~\cite[Sec.6, Theorem 6.10.13]{Hamilton-2017}), while the following one is also known as a {\sl compatibility property with the scalar product} (cf. \cite[Theorem 6.10.14]{Hamilton-2017}):
\begin{equation}
\Lcal_X\la  \Psi,\Phi \ra_{\ourD} = \la \nabla_X\Psi, \Phi \ra_{\ourD}  + \la  \Psi,\nabla_X\Phi\ra_{\ourD}.
\end{equation}
\end{subequations}

}


\subsection{ Fundamental calculus for spinor fields} 
\label{section===42}

{

\paragraph{Combining tensor and spinor fields.}

We now provide further remarks on the treatment of spinor fields. Let us consider the tensor product of a tensor bundle with the spinor bundle (all of which are real vector bundles), namely
\be
\ourS(\mathcal{M})^r_s := T(\mathcal{M})^r_s\otimes \ourS(\mathcal{M}), 
\ee
which may be understood as the space of multilinear forms on the tangent and cotangent bundles taking values in the spinor bundle. A section of $\ourS(\mathcal{M})^r_s$ is still called a $(r,s)-$ spinor field, or simply a spinor when there is no risk of confusion. When  $T$ is a tensor field on $\mathcal{M}$ and $\Psi$ is a spinor field, we define
\be
\nabla(T\otimes \Psi) := \nabla T\otimes \Psi + T\otimes \nabla\Psi, 
\ee
where the latter occurrence of $\nabla$ acting on $\Psi$ denotes the spinor connection. When $T$ is a scalar field,  we have $\nabla T = \mathrm{d}T$, and the above identity reduces to the Leibniz rule for the covariant derivative. For a general tensor or spinor, say $A$, we use the standard notation on connection 1-form, that is, 
\be
\nabla A(X) = \nabla_X A.
\ee
Then, by the Leibniz rule, the second-order connection form (i.e. the Hessian form) can be expressed as 
\bel{equa46--}
\nabla\nabla A(X,Y) = \nabla_X\nabla_Y A = \nabla_X (\nabla_Y A) - \nabla_{\nabla_XY} A, 
\ee
which is checked by observing that $\nabla_X( \nabla A(Y)) = \nabla_X(\nabla A)(Y) + \nabla A(\nabla_XY)$. 


\paragraph{Link with the standard notions.}

The above notation coincides with the well-known calculus rule for scalar fields:
\be
\nabla_X\nabla_Y\phi = \nabla_X(\nabla_Y\phi) - \nabla_{\nabla_XY} \phi = X(Y\phi) - (\nabla_XY)\phi
\ee
or, in abstract coordinates,
\be
\nabla_{\alpha} \nabla_{\beta} \phi = \del_\alpha \del_{\beta} \phi - \Gamma_{\alpha\beta}^{\gamma} \del_{\gamma} \phi,
\ee
For vector or spinor fields we also write 
\begin{subequations} \label{eq1-24-feb-2025}
\begin{equation}\label{eq1a-24-feb-2025}
\nabla_X\nabla_Y Z = \nabla\nabla Z(X,Y) = \nabla_X(\nabla_Y Z) - \nabla_{\nabla_XY}Z
\end{equation}
and, in abstract coordinates,
\begin{equation}\label{eq1b-24-feb-2025}
\nabla_{\alpha} \nabla_{\beta}Z = \nabla_{\alpha}(\nabla_{\beta}Z) - \nabla_{\nabla_{\alpha} \del_{\beta}}Z = \nabla_{\alpha}(\nabla_{\beta}Z) - \Gamma_{\alpha\beta}^{\delta} \nabla_{\delta}Z.
\end{equation}
\end{subequations}
However, the notation $\nabla_X\nabla_Y Z$ can also be understood as $\nabla_X\circ\nabla_YZ = \nabla_X(\nabla_Y Z)$, which is {\sl not a two form} in $(X,Y)$ (and this does not maintain the consistency with the scalar case). Throughout, we always rely on the definition \eqref{eq1-24-feb-2025}. With this convention, the Riemann curvature acting on $Z$ can be defined as the anti-symmetric part of $\nabla\nabla Z$:
\begin{equation}\label{eq1-31-aout-2025}
R(X,Y)Z = \nabla_X\nabla_Y Z - \nabla_Z\nabla_Y Z = \nabla_X(\nabla_YZ) - \nabla_X(\nabla_YZ) - \nabla_{[X,Y]}Z.
\end{equation}


\paragraph{Abstract index notation.}

Suppose that in abstract local coordinates a spinor or tensor $A$ can be written as $A = A^{\alpha} \otimes \del_{\alpha}$, or $A = A_{\alpha} \otimes \mathrm{d}x^{\alpha}$. Here $A^{\alpha}, A_{\alpha}$ may be scalar functions or lower order spinor fields. Then we denote by
\be
\nabla_X A^{\alpha} = (\nabla_X A)^{\alpha}, \qquad \nabla_X A_{\alpha} = (\nabla_X A)_{\alpha},
\ee
so that
\be
\nabla_X A^{\alpha} \otimes \del_\alpha = \nabla_X A, \qquad \nabla_X A_{\alpha} \otimes \mathrm{d}x^{\alpha} = \nabla_X A.
\ee
Then in abstract coordinates, we have 
\begin{equation}
\nabla_{\beta} A^{\alpha} = \nabla_{\beta}(A^{\alpha}) + A^{\gamma} \Gamma_{\beta\gamma}^{\alpha},
\quad 
\nabla_{\beta} A_{\alpha} = \nabla_{\beta}(A_{\alpha}) - A_{\gamma} \Gamma_{\beta\alpha}^{\gamma}.
\end{equation}
In the above expression, $\nabla_{\beta}(A^{\alpha})$ denotes the covariant derivative on the specific scalar, tensor, or spinor field $A^{\alpha}$. For example, in the purely vectorial case, $A^{\alpha}$ is a scalar function, namely $\mathrm{d}x^{\alpha}(A)$. Thus we have 
\be
\nabla_{\beta} A^{\alpha} = \nabla_{\beta}(A^{\alpha}) + A^{\gamma} \Gamma_{\beta\gamma}^{\alpha} = \del_{\beta} A^{\alpha} + A^{\gamma} \Gamma_{\beta\gamma}^{\alpha}. 
\ee
This is the standard formula when $A^{\alpha},A_{\alpha}$ are scalar and $\nabla_{\beta}$ reduces to the partial derivative $\del_{\beta}$.

With the above notation, we can establish the following property ---commonly referred to as the statement that {\sl  Gamma matrices are covariantly constant} (cf., for instance, \cite{Bertlmann00,ChenXuantao}). 

\begin{lemma} \label{lem1-24-feb-2025}
In any local coordinate chart, for any spinor field $\Psi$ one has 
\begin{equation}
\nabla_{\alpha}(\del_{\beta} \cdot\Psi) = \del_{\beta} \cdot\nabla_{\alpha} \Psi.
\end{equation}
\end{lemma}

\begin{proof} Let us set  $T(Y): = Y\cdot\Psi$ and compute 
$$
\hskip1.cm
\aligned
\nabla_{\alpha}(\del_{\beta} \cdot\Psi) 
& =  \nabla_{\alpha}T(\del_{\beta}) 
= \nabla_{\alpha}(T_{\beta}) - T(\nabla_{\alpha} \del_{\beta}) 
= \nabla_{\alpha}^{\ourS}(\del_{\beta} \cdot\Psi) - \nabla_{\alpha} \del_{\beta} \cdot\Psi 
\\
& =  \nabla_{\alpha} \del_{\beta} \cdot\Psi + \del_{\beta} \cdot\nabla_{\alpha} \Psi 
- \nabla_{\alpha} \del_{\beta} \cdot\Psi
 =  \del_{\beta} \cdot\nabla_\alpha \Psi. \hskip8.Cm  \qedhere 
\endaligned 
$$
Here, the notation $\nabla^{\ourS}_{\alpha}$ refer to the fact that this operation acts on the spinorial component $\del_{\beta}\cdot\Psi$, which does not consider the index $\beta$ as a tensorial index. 
\end{proof}

}


\subsection{ Energy for the Dirac equation in curved spacetime} 
\label{section===43}

{ 

\paragraph{Energy identity.} 

We continue to proceed in a fully gauge-independent manner, and introduce the vector field 
\bel{equa-jd394}
V[\Psi] := \la \Psi,g^{\alpha\beta} \del_\alpha \cdot\Psi\ra_{\ourD} \del_{\beta}.
\ee
which is quadratic in terms of a general spinor field $\Psi$ and plays now the role of an energy  flux for the Dirac operator. Observe that $V[\Psi]$ is a real-valued vector, thanks to the anti-symmetry of Clifford multiplication and the properties of the Dirac form. Alternatively, in an orthonormal frame, we can write  
\bel{equa-jd395}
V[\Psi] = \la \Psi, \eta^{ij}e_j\cdot\Psi\ra_{\ourD} e_i,
\ee
which is independent of the choice of tetrad. Indeed, if $(f_0, f_1, f_2, f_3)$ denotes any other orthonormal frame, satisfying the transformation rule $f_i = T_i^je_j$ for some $T \in \mathrm{SO}_{3,1}^+$, we have 
$$
\aligned
\la \Psi, \eta^{ij}f_j\cdot\Psi\ra_{\ourD} f_i 
& = \eta^{ij} \la \Psi,f_j\cdot\Psi \ra_{\ourD} f_i 
= \eta^{ij} \la \Psi, T_j^{j'}e_{j'} \cdot\Psi\ra_{\ourD} T_i^{i'}e_{i'} \\
& = \eta^{ij}T_i^{i'}T_j^{j'} \la \Psi, e_{j'} \cdot\Psi\ra_{\ourD} e_{i'} 
= \eta^{i'j'} \la \Psi, e_{j'} \cdot\Psi\ra_{\ourD} e_{i'}.
\endaligned
$$
Hence, $V[\Psi]$ in \eqref{equa-jd395} is a globally well-defined geometric vector field, which coincides with the gauge-independent expression \eqref{equa-jd394}.

\begin{proposition}
\label{propo-JDK02}
For any sufficiently regular section $\Psi$ on $\ourS(\mathcal{M})$, one has 
\begin{equation}
\label{equa-4222}
\mathrm{div} (V[\Psi]) = \big\la \opDirac \Psi , \Psi \big\ra_{\ourD} + \big\la \Psi, \opDirac \Psi \big\ra_{\ourD}
= 2 \, \Re\big(\opDirac\Psi,\Psi\big).
\end{equation}
Hence, the divergence vanishes identically when $\Psi$ satisfies the (homogeneous) Dirac equation. Furthermore, the vector field $V[\Psi]$ is timelike and, under suitable normalization, future-oriented. 
\end{proposition}

\paragraph{Derivation of the energy identity.} 

Computing the divergence of $V[\Psi]$ using the Leibniz rule for the spinorial covariant derivative, we find 
\be
\aligned
\mathrm{div} (V[\Psi]) 
& = \nabla_{\beta}(g^{\alpha\beta} \la\Psi,\del_\alpha \cdot\Psi \ra_{\ourD}) 
 =  g^{\alpha\beta} \big(\nabla_{\beta} \la\Psi,\del_\alpha \cdot\Psi\ra_{\ourD} \big)
\\
& = g^{\alpha\beta} \del_{\beta} \la\Psi,\del_\alpha \cdot\Psi\ra_{\ourD} 
- g^{\alpha\beta} \la \Psi,\nabla_{\beta} \del_\alpha \cdot \Psi\ra_{\ourD}
\\ 
& =  \la g^{\alpha\beta} \del_\alpha \cdot\nabla_{\beta} \Psi,\Psi \ra_{\ourD} 
+ \la \Psi,g^{\alpha\beta} \del_\alpha \cdot \nabla_{\beta} \Psi\ra_{\ourD}
= 2 \, \Re\big(\opDirac\Psi,\Psi\big).
\endaligned 
\ee 
\bse
Alternatively, this calculation can be presented in a frame, as follows: 
\be
\aligned
\text{div}(V[\Psi]) & = \eta^{ij} \Lcal_{e_i} \left( \la \Psi, e_j \cdot \Psi \ra_{\ourD} \right) \\
& = \eta^{ij} \left( \la \nabla_{e_i} \Psi, e_j \cdot \Psi \ra_{\ourD} + \la \Psi, e_j \cdot \nabla_{e_i} \Psi \ra_{\ourD} \right).
\endaligned
\ee
We use the anti-symmetry property $\la e_i \cdot \Psi, \Phi \ra_{\ourD} = - \la \Psi, e_i \cdot \Phi \ra_{\ourD}$ and, by combining terms, 
\be
\aligned
\text{div}(V[\Psi]) & = \eta^{ij} \left( \la \nabla_{e_i} \Psi, e_j \cdot \Psi \ra_{\ourD} + \la \Psi, e_j \cdot \nabla_{e_i} \Psi \ra_{\ourD} \right) \\
& = \eta^{ij} \mathcal{L}_{e_i} \bigl( \la \Psi, e_j \cdot \Psi \ra_{\ourD} \big)
- \eta^{ij} \la \Psi, (\nabla_{e_i} e_j) \cdot \Psi \ra_{\ourD}.
\endaligned
\ee
Recalling the definition of the Dirac operator $\opDirac \Psi = \eta^{ij} e_i \cdot \nabla_{e_j} \Psi$ in the frame $e_i$, we arrive at \eqref{equa-4222}. 
\ese
%
 

\paragraph{Calculus in coordinates.}

It is convenient to present the timelike property in coordinates. 
Let us pick a particular tetrad $e = \{e_i\}$. The spinorial covariant derivative can be explicitly expressed via the Levi-Civita connection form associated with $\{e_i\}$ (and, for further details, cf.~for instance~\cite[Proposition 6.10.9]{Hamilton-2017})\footnote{Our expression here is slightly different from the one in~\cite[Proposition 6.10.9]{Hamilton-2017}, but in fact is equivalent to it. Observe that $\omega_{ij} = - \omega_{ji}$, and
$
\omega_{ij}(X)\gamma^i\gamma^j = \frac{1}{2}\omega_{ij}(X)(\gamma^i\gamma^j- \gamma^j\gamma^i) = \omega_{ij}(X)\gamma^{ij}.
$
}:
\begin{equation} \label{eq3-27-feb-2025}
\nabla_X\psi := \diff\psi(X) + \frac{1}{4} \omega_{ij}(X)\gamma^i\gamma^j\psi, 
\end{equation}
where
\be
\omega_{ij}(X)\eta^{ik}e_k = \nabla_X(e_i).
\ee
We then prove that the energy flux $V[\Psi]$ is timelike, as follows. Namely, let $\vec{n}$ be any future-oriented, unit, timelike vector field, together with an adapted tetrad $e = \{e_0,e_{\ih}\}$ such that $e_0 = \vec{n}$. Let $\psi$ be the coordinate of $\Psi$ in this basis $e$. Then, in any local coordinate chart, we have $\vec{n} = n^{\gamma}\del_{\gamma}$ and 
\be
g(V[\Psi],\vec{n}) = \la \Psi, g^{\alpha\beta}\del_{\alpha}\cdot\Psi\ra_{\ourD}g(\del_{\beta},\vec{n}) = \la \Psi, n^{\alpha}\del_{\alpha}\cdot\Psi\ra_{\ourD} = \la \Psi,\vec{n}\cdot\Psi\ra_{\ourD}.
\ee
We express this equation in coordinates in the basis $e$ and find
\be\label{eq2-02-mai-2025}
g(V[\Psi],\vec{n}) = \la \gamma_0\psi, \gamma_0\psi\ra = (\gamma_0\psi)^\dag(\gamma_0\psi)\geq 0,
\ee
which shows that $V[\Psi]$ is timelike and future-oriented.

}

\section{Dirac commutators} 
\label{section=N3}

\subsection{ Notation for the curvature}
\label{section===51}

{

\paragraph{An identity between the Riemann and the Ricci curvatures.}

We now introduce and study curvature-related operators acting on spinor fields and, by recalling the properties of the Riemann and Ricci curvature tensors, we derive explicit and useful curvature-spinor identities involving Clifford multiplication. The antisymmetry properties of the Clifford algebra and the Bianchi identities come together, and the following technical identities provide fundamental calculus rules that are essential in the subsequent analysis. 

The curvature of a Lorentzian manifold was defined in \eqref{eq1-31-aout-2025}. More precisely, the components of the Riemann curvature tensor $\Riem$ are denoted by by \( R_{\alpha\beta\nu\mu} \) (following the convention of \cite{YCB}) with 
\be
R_{\alpha\beta\ \mu}^{\ \ \ \nu}\del_{\nu} = R(\del_{\alpha},\del_{\beta})\del_{\mu}, 
\ee
while the components of the Ricci curvature $\Ric$ are denoted by \( R_{\alpha\beta} \). As usual, the Ricci tensor is obtained by contracting the Riemann tensor and, specifically, 
\begin{equation}
\aligned
R_{\alpha\beta} & = g^{\mu\nu} R_{\mu \alpha\nu \beta}. 
\endaligned
\end{equation}
As usual, we raise and lower indices with the metric, so for instance we have 
\be
R_{\gamma}{}^{\alpha\mu\nu} = g^{\alpha\alpha'} g^{\mu\mu'} g^{\nu\nu'} R_{\gamma\alpha'\mu'\nu'}.
\ee
We recall the antisymmetry in the first two indices $R_{\alpha\beta\gamma\delta} = -R_{\beta\alpha\gamma\delta}$, the antisymmetry in the last two indices $R_{\alpha\beta\gamma\delta} = -R_{\alpha\beta\delta\gamma}$ and the symmetry under interchange of pairs
$R_{\alpha\beta\gamma\delta} = R_{\gamma\delta\alpha\beta}$. For convenience in the presentation, we often prefer the notation $R_{\gamma\alpha'\mu'\nu'}$. Let us now consider identities involving spinor fields and the Clifford product. The following identity is central in all computations involving spinor fields on curved manifolds.

\begin{lemma}
\label{lemma-51}
For any spinor field $\Psi$, the Riemann curvature and the Ricci curvature operators are related via the identity:
\begin{equation}\label{eq3-22-feb-2025}
g^{\mu\nu} g^{\alpha\gamma} g^{\eta\beta} R_{\delta\beta\gamma\nu} \del_\alpha \cdot \del_\mu \cdot \del_{\eta} \cdot \Psi = - 2 \, g^{\alpha\gamma} R_{\delta\gamma} \del_\alpha \cdot \Psi.
\end{equation}
\end{lemma}

\begin{proof} We are going rely on the {\sl anti-commutation property} enjoyed by the Clifford product  
\begin{equation} \label{clifford-anticomm}
\del_\mu \cdot\del_{\eta}+\del_{\eta} \cdot\del_\mu = - 2g_{\mu\eta}, 
\end{equation}
which implies
\begin{equation} \label{clifford-anticomm2}
\aligned
\del_\alpha \cdot \del_\mu \cdot\del_{\eta} 
& = - \del_\mu \cdot \del_\alpha \cdot\del_{\eta} - 2g_{\mu\alpha} \del_\eta
\\
& = \del_\mu \cdot \del_{\eta} \cdot\del_\alpha + 2 g_{\eta \alpha} \del_\mu - 2g_{\mu\alpha} \del_\eta.
\endaligned
\end{equation} 
We will also use the {\sl first Bianchi identity} satisfied by the Riemann curvature, that is,  
\be 
R_{\delta\beta\gamma\nu} = - R_{\delta\gamma\nu\beta} - R_{\delta\nu\beta\gamma} .
\ee
We then rely on these identities after contracting with $g^{\mu\nu} g^{\alpha\gamma} g^{\eta\beta} $ and seek to identify the two key terms 
\be
\aligned
A_\delta 
& := g^{\mu\nu} g^{\alpha\gamma} g^{\eta\beta} R_{\delta\beta\gamma\nu} \del_\alpha \cdot \del_\mu \cdot \del_{\eta} \cdot \Psi,
\qquad
B_\delta := g^{\alpha\gamma} R_{\delta\gamma} \del_\alpha \cdot \Psi.
\endaligned
\ee
Thanks to \eqref{clifford-anticomm2}, we find 
\be
\aligned
A_\delta 
& = g^{\mu\nu} g^{\alpha\gamma} g^{\eta\beta} R_{\delta\beta\gamma\nu} \del_\alpha \cdot \del_\mu \cdot \del_{\eta} \cdot \Psi
\\
& = g^{\mu\nu} g^{\alpha\gamma} g^{\eta\beta} R_{\delta\beta\gamma\nu}
\del_\mu \cdot \del_{\eta} \cdot\del_\alpha   \cdot \Psi
+  
g^{\mu\nu} g^{\alpha\gamma} g^{\eta\beta} R_{\delta\beta\gamma\nu}
\Big( 2 g_{\eta \alpha} \del_\mu - 2g_{\mu\alpha} \del_\eta.
\Big) \cdot \Psi
\\
& 
=: C_\delta + D_\delta. 
\endaligned
\ee
in which 
\be
\aligned
C_\delta 
& :=  g^{\mu\nu} g^{\alpha\gamma} g^{\eta\beta} R_{\delta\beta\gamma\nu}
\del_\mu \cdot \del_{\eta} \cdot\del_\alpha   \cdot \Psi
\\
& =  g^{\alpha\gamma} g^{\eta\beta} g^{\mu\nu} \Big(
- R_{\delta\beta\gamma\nu} - R_{\delta\gamma\nu\beta} \Big) 
\del_\alpha \cdot \del_\mu \cdot\del_{\eta}   \cdot \Psi
\\ 
& = - A_\delta -
g^{\alpha\gamma} g^{\eta\beta} g^{\mu\nu} R_{\delta\gamma\nu\beta} 
\del_\alpha \cdot \del_\mu \cdot\del_{\eta}   \cdot \Psi
=: - A_\delta - E_\delta, 
\endaligned
\ee
after transforming $\mu, \eta, \alpha \to \alpha, \mu, \eta$.  
and using 
$R_{\delta\nu\beta\gamma}
= - R_{\delta\beta\gamma\nu} - R_{\delta\gamma\nu\beta}.
$
We also observe that, again with the anti-commutation property for spinors, 
\be
\aligned
E_\delta 
& := 
g^{\alpha\gamma} g^{\eta\beta} g^{\mu\nu} R_{\delta\gamma\nu\beta} 
\del_\alpha \cdot \del_\mu \cdot\del_{\eta}   \cdot \Psi
\\
& = -  g^{\alpha\gamma} g^{\eta\beta} g^{\mu\nu} R_{\delta\gamma\nu\beta} 
\del_\mu \cdot \del_\alpha \cdot\del_{\eta}   \cdot \Psi
- 2  g^{\alpha\gamma} g^{\eta\beta} g^{\mu\nu} g_{\alpha\mu}  R_{\delta\gamma\nu\beta}  \del_{\eta}   \cdot \Psi
= C_\delta + 2  B_\delta. 
\endaligned
\ee
It remains to compute 
\be
\aligned
D_{\delta}:& = g^{\mu\nu} g^{\alpha\gamma} g^{\eta\beta} R_{\delta\beta\gamma\nu}
\Big( 2 g_{\eta \alpha} \del_\mu - 2g_{\mu\alpha} \del_\eta.
\Big) \cdot \Psi
\\
& = 2g^{\mu\nu}g^{\alpha\gamma}R_{\delta\alpha\gamma\nu}\del_{\mu} 
- 2g^{\alpha\gamma}g^{\beta\beta}R_{\delta\beta\gamma\alpha}\del_{\eta} = -2B_{\delta}.
\endaligned
\ee

We have established that $A_\delta = C_\delta + D_\delta$ and $C_\delta = - A_\delta - E_\delta$ with 
$E_\delta = C_\delta + 2  B_\delta$ 
and 
$D_\delta = -2 \, B_\delta$, 
so that 
$A_\delta = - 2 B_\delta$. 
\end{proof}


We conclude as follows. 

\begin{corollary}[See also Lemma~3.2 of \cite{ChenXuantao}]
The following identity holds:
\begin{equation}
R_{\gamma}^{\ \alpha\mu\nu}\del_{\alpha}\cdot\del_{\mu}\cdot\del_{\nu} = 2g^{\alpha\gamma}R_{\delta\gamma}\del_{\alpha}\cdot\Psi.
\end{equation}
\end{corollary}

\begin{proof} 
\bse
We have 
\be
\aligned
g^{\mu\nu}g^{\alpha\gamma}g^{\eta\beta}R_{\delta\beta\gamma\nu}(\del_{\alpha}\cdot\del_{\mu}\cdot\del_{\eta} + \del_{\alpha}\cdot\del_{\eta}\cdot\del_{\mu}) 
& =-2g^{\mu\nu}g^{\alpha\gamma}g^{\eta\beta}R_{\delta\beta\gamma\nu}g_{\mu\eta}\del_{\alpha}
\\
& =-2g^{\nu\beta}R_{\delta\beta\gamma\nu}g^{\alpha\gamma}\del_{\alpha} = -2g^{\alpha\gamma}R_{\delta\gamma}\del_{\alpha}.
\endaligned
\ee
and 
\be
\aligned
g^{\mu\nu}g^{\alpha\gamma}g^{\eta\beta}R_{\delta\beta\gamma\nu}
(\del_{\alpha}\cdot\del_{\eta}\cdot\del_{\mu} + \del_{\eta}\cdot\del_{\alpha}\cdot\del_{\mu}) 
& =-2g^{\mu\nu}g^{\alpha\gamma}g^{\eta\beta}R_{\delta\beta\gamma\nu}g_{\alpha\eta}\del_{\mu}
\\
& =-2g^{\mu\nu}g^{\gamma\beta}R_{\delta\beta\gamma\nu}\del_{\mu} 
= 2g^{\mu\nu}R_{\delta\nu}\del_{\mu}.
\endaligned
\ee
Therefore we find 
\be
g^{\mu\nu}g^{\alpha\gamma}g^{\eta\beta}R_{\delta\beta\gamma\nu}(\del_{\alpha}\cdot\del_{\mu}\cdot\del_{\eta} - \del_{\eta}\cdot\del_{\alpha}\cdot\del_{\mu}) = -4g^{\alpha\gamma}R_{\delta\gamma}\del_{\alpha}.
\ee
This leads us to the desired result.
\ese
\end{proof}

}

\subsection{ Hessian, curvature, Dirac, and wave operators}
\label{section===52}

{

\paragraph{Hessian operator.}

We now turn our attention to various operators and identities that involve the curvature. First of all, the Hessian of a spinor field $\Psi$ is defined by (by \eqref{equa46--}) 
\be
\Hessian[\Psi](X,Y) := 
\nabla_X\nabla_Y \Psi 
= \nabla_X(\nabla_Y \Psi) - \nabla_{\nabla_XY} \Psi. 
\ee
We easily check the linearity property, for any sufficiently regular function $f$ and vector fields $X,Y$: 
\begin{equation}
\Hessian[\Psi](X,fY) = f \, \Hessian[\Psi](X,Y).
\end{equation}
Indeed, this follows by computing 
\be
\aligned
\Hessian[\Psi](X,fY) 
& = \nabla_X(\nabla_{fY} \Psi) - \nabla_{\nabla_X(fY)} \Psi 
\\
& = f\nabla_X(\nabla_Y \Psi) + \mathrm{d}f(X)\nabla_Y \Psi - \nabla_{\mathrm{d}f(X)\nabla_XY + f\nabla_XY} \Psi
\\
& =  f\nabla_X(\nabla_Y \Psi) - f\nabla_{\nabla_XY} \Psi.
\endaligned
\ee
In fact, $\Hessian[\Psi]$ is a 2-form field taking values in the spinor bundle.  


\paragraph{Curvature operator associated with spinor fields.}
The \emph{curvature operator acting on spinor fields} $\Psi$ is given by the anti-symmetric part of the Hessian operator:
\bel{equa--512}
\aligned
\opCurve_{XY} \Psi :& =  \nabla_X\nabla_Y \Psi - \nabla_Y\nabla_X\Psi = \nabla_X(\nabla_Y\Psi) - \nabla_Y(\nabla_X\Psi) - \nabla_{[X,Y]} \Psi 
\\
& =  [\nabla_X,\nabla_Y]\Psi - \nabla_{[X,Y]} \Psi.
\endaligned
\ee
While the following result is standard, we provide an outline of its proof (using our notation) in Section~\ref{section=N21}. 
\begin{lemma} \label{lem1-01-march-2025}
The Riemann curvature $\Riem$ suitably applied to a spinor can be expressed in terms of the curvature operator $\opCurve_{XY}$ acting on a spinor via the identity 
\begin{equation}\label{equa-26juillet2025a} 
\opCurve_{XY} \Psi = \frac{1}{4}g^{\alpha\mu}g^{\beta\nu} \, \big(R(X,Y)\del_{\mu},\del_{\nu}\big)_g\del_\alpha \cdot\del_{\beta} \cdot\Psi . 
\end{equation}  
in which $X,Y$ are arbitrary vector fields, $\Psi$ is an arbitrary spinor field, and $\del_\alpha$ is an arbitrary coordinate basis. 
\end{lemma} 

We first establish the following result. 

\begin{lemma}
The Ricci curvature $\Ric$ suitably applied to a spinor can be expressed in term of the curvature operator $\opCurve_{X\beta}$ via the identity 
\begin{equation} \label{eq2-24-feb-2025}
\frac{1}{2}g^{\alpha\beta} \Ric (X,\del_{\beta})\del_\alpha \cdot\Psi
= -g^{\alpha\beta} \del_\alpha \cdot \opCurve_{X\beta} \Psi.
\end{equation}
\end{lemma}

\begin{proof} 
\bse
From 
$$
\aligned
\opCurve_{X\beta} \Psi & =  \frac{1}{4}g^{\mu\mu'}g^{\nu\nu'}
(R(X,\del_{\beta})\del_{\mu},\del_{\nu})_g\del_{\mu'} \cdot\del_{\nu'} \cdot\Psi
\\
& = \frac{1}{4}g^{\mu\mu'}g^{\nu\nu'}X^{\delta}  R_{\delta\beta\nu\mu}
\del_{\mu'} \cdot\del_{\nu'} \cdot\Psi, 
\endaligned
$$
we deduce that 
$$
\aligned
g^{\alpha\beta} \del_\alpha \cdot \opCurve_{X\beta} \Psi 
& =  \frac{1}{4}g^{\alpha\beta}g^{\mu\mu'}g^{\nu\nu'}X^{\delta} R_{\delta\beta\nu\mu} \del_\alpha \cdot\del_{\mu'} \cdot\del_{\nu'} \cdot\Psi
\\
& = - \frac{1}{2}g^{\alpha\beta}X^{\delta} R_{\delta\beta}\del_\alpha \cdot\Psi
=  - \frac{1}{2}g^{\alpha\beta}X^{\delta} \Ric (\del_\delta,\del_\beta)\del_\alpha \cdot\Psi, 
\endaligned
$$
where, for the last expression, we applied Lemma~\ref{lemma-51}.
\ese
\end{proof} 


\paragraph{Wave (Laplace) operator.}

The trace of Hessian is nothing but the wave operator, given by 
\begin{equation} \label{eq1-27-feb-2025}
\Box_g\Psi := \text{trace} \bigl( \Hessian[\Psi] \bigr)
= g^{\alpha\beta} \nabla_{\alpha} \nabla_{\beta} \Psi. 
\end{equation}
Observe that, in any \underline{wave coordinate chart} in which the contracted Christoffel symbols $\Gamma^\alpha = 0$ vanish, we find the simpler expression
$$
\Box_g\Psi = g^{\mu\nu} \nabla_{\mu} (\nabla_\nu \Psi) =: \widetilde{\Box}_g\Psi, 
$$
which {\sl formally} coincides with the expression associated with a scalar field $\phi$, namely
\be
\Box_g\phi = g^{\mu\nu} \del_\mu \del_{\nu} \phi =: \widetilde{\Box}_g\phi.
\ee


\paragraph{Dirac operator.}

Let us recall the notation for the Dirac operator  
\begin{equation}
\opDirac \Psi := g^{\mu\nu} \del_\mu \cdot \nabla_\nu \Psi, 
\end{equation}
and observe now its square $\opDirac^2$ is related to the wave operator, as follows. 

\begin{lemma}\label{lem1-21-feb-2025}
For any spinor field $\Psi$ one has 
\begin{equation}
\opDirac^2\Psi = - \Box_g\Psi - \frac{1}{4}R\Psi, 
\end{equation}
where $R$ denotes the scalar curvature.
\end{lemma}

\begin{proof} Observe that from Lemma~\ref{lem1-24-feb-2025}, that is, 
${\nabla_{\beta}(\del_\mu \cdot\nabla_\nu \Psi) = \del_{\mu}\cdot\nabla_{\beta}\nabla_{\nu}\Psi}$.
Next, recalling that the Levi-Civita connection satisfies $\nabla g = 0$, we have 
$$
\aligned
&g^{\alpha\beta} \del_\alpha \cdot \nabla_{\beta}(g^{\mu\nu} \del_\mu \cdot\nabla_\nu \Psi) 
\\
& =  g^{\alpha\beta}g^{\mu\nu} \del_\alpha \cdot\nabla_{\beta}(\del_\mu \cdot\nabla_\nu \Psi) 
= g^{\alpha\beta}g^{\mu\nu} \del_\alpha \cdot\del_\mu \cdot\nabla_{\beta} \nabla_\nu \Psi
\\
& =  \frac{1}{2}g^{\alpha\beta}g^{\mu\nu}(\del_\alpha \cdot\del_{\mu} + \del_{\mu}\cdot\del_{\alpha})\nabla_{\beta} \nabla_\nu \Psi 
+ \frac{1}{2}g^{\alpha\beta}g^{\mu\nu}\del_{\mu}\cdot\del_\alpha \cdot\big(\nabla_\nu \nabla_{\beta} \Psi - \nabla_{\beta} \nabla_\nu \Psi\big)
\\
& =  -g^{\beta\nu} \nabla_{\beta} \nabla_\nu \Psi 
+ \frac{1}{2}g^{\alpha\beta}g^{\mu\nu} \del_\mu \cdot\del_\alpha \cdot [\nabla_{\nu},\nabla_{\beta}] \Psi
\\
& = -g^{\beta\nu} \nabla_{\beta} \nabla_\nu \Psi 
+  \frac{1}{2}g^{\mu\nu}\del_{\mu}\cdot \big(g^{\alpha\beta}\del_{\alpha}\cdot\opCurve(\del_{\nu},\del_{\beta})\Psi\big)
\\
& = -g^{\beta\nu} \nabla_{\beta} \nabla_\nu \Psi 
-  \frac{1}{4}g^{\mu\nu}g^{\alpha\beta}\del_{\mu}\cdot \Ric(\del_{\nu},\del_{\beta})\del_{\alpha}\cdot\Psi
\\
& = -g^{\beta\nu} \nabla_{\beta} \nabla_\nu \Psi 
-  \frac{1}{4}R\Psi.
\endaligned
$$ 
in which we used the notation \eqref{equa--512} and applied \eqref{eq2-24-feb-2025}. 
\end{proof}
}


\subsection{ Commutation relations for the Dirac operator}
\label{section===53}

{

\paragraph{Commutation relation for the Dirac operator.}
We  now turn our attention to computing commutators with the Dirac operator and writing consequences of interest in wave coordinates, since these will be the coordinates we will work with.
In order to reach better commutation relations, it is convenient to introduce a \emph{modified notion of derivative}, as follows. 

\begin{proposition} \label{prop1-22-feb-2025}
For any spinor field $\Psi$ and any vector field $X$, with the notation 
\bel{equa-511m}
\widehat{X} \Psi := \nabla_X\Psi 
- \frac{1}{4}g^{\alpha\gamma} \del_\alpha \cdot\nabla_{\gamma}X\cdot\Psi, 
\ee
one has 
\begin{equation} \label{eq2-26-feb-2025}
[\widehat{X},\opDirac]\Psi = 
- \frac{1}{2} \pi[X]^{\alpha\beta} \del_\alpha \cdot\nabla_{\beta} \Psi -\frac{1}{4}g^{\alpha\beta} \Ric (X,\del_{\alpha})\del_{\beta} \cdot \Psi 
- \frac{1}{4} \Box_gX\cdot\Psi,
\end{equation}
where $\pi[X]$ is the \emph{deformation tensor of $X$}, defined as 
\be
\pi[X]_{\alpha\beta} = (\Lcal_Xg)_{\alpha\beta}.
\ee
\end{proposition}

The proof is decomposed in several observations, as follows.
\begin{lemma} The commutator $[\nabla_X,\opDirac]\Psi$ associated with a vector field and a spinor field is given by  
\begin{equation} \label{eq2-22-feb-2025}
[\nabla_X,\opDirac]\Psi = 
- g^{\alpha\beta} \del_\alpha \cdot\nabla_{\nabla_{\beta}X} \Psi
- \frac{1}{2}g^{\alpha\beta} \Ric (X,\del_{\beta})\del_\alpha \cdot \Psi. 
\end{equation} 
\end{lemma}

Of course, the identity \eqref{eq2-22-feb-2025} can also be expressed in the form 
\begin{equation}
[\nabla_{\mu},\opDirac]\Psi 
= - g^{\alpha\beta} \Gamma_{\beta\mu}^{\nu} \del_\alpha \cdot\nabla_\nu \Psi 
- \frac{1}{2}g^{\alpha\beta} \Ric (\del_{\mu},\del_{\beta})\del_\alpha \cdot\Psi.
\end{equation}

\begin{proof} It suffices to compute 
$$
\aligned
\nabla_X(g^{\alpha\beta} \del_\alpha \cdot\nabla_{\beta} \Psi) 
& =  \nabla_Xg^{\alpha\beta}\del_{\alpha}\cdot\nabla_{\beta}\Psi +  g^{\alpha\beta} \nabla_X(\del_\alpha \cdot\nabla_{\beta} \Psi)
=  g^{\alpha\beta} \nabla_X(\del_\alpha \cdot\nabla_{\beta} \Psi) 
\\
& =  g^{\alpha\beta}\del_{\alpha}\cdot\nabla_{X}\nabla_{\beta}\Psi
\endaligned
$$
where we have applied the fact $\nabla_Xg = 0$ and Lemma~\ref{lem1-24-feb-2025}.
Therefore, we find 
$$
\aligned
\nabla_X(g^{\alpha\beta} \del_\alpha \cdot\nabla_{\beta} \Psi) 
& =  g^{\alpha\beta} \del_\alpha \cdot\nabla_{\beta} \nabla_X\Psi 
+g^{\alpha\beta} \del_\alpha \cdot \opCurve_{X\beta} \Psi
\\
& = g^{\alpha\beta} \del_\alpha \cdot\nabla_{\beta}(\nabla_X\Psi) 
-g^{\alpha\beta} \del_\alpha \cdot\nabla_{\nabla_{\beta}X} \Psi  
+g^{\alpha\beta} \del_\alpha \cdot \opCurve_{X\beta} \Psi
\\
& =  \opDirac(\nabla_X \Psi) 
- g^{\alpha\beta} \del_\alpha \cdot\nabla_{\nabla_{\beta}X} \Psi
+ g^{\alpha\beta} \del_\alpha \cdot \opCurve_{X\beta} \Psi.
\endaligned
$$
For the second term in the right-hand side of the identity, we apply \eqref{eq2-24-feb-2025} and obtain \eqref{eq2-22-feb-2025}. 
\end{proof}


It remains to analyze the effect of the contribution arising in the modified derivative \eqref{equa-511m}. 

\begin{lemma}
With the notation above, one has 
\begin{equation} \label{eq6-22-feb-2025}
\aligned
& \,  [\opDirac,g^{\alpha\gamma} \del_\alpha \cdot\nabla_{\gamma}X\cdot]\Psi 
\\
& = 
2g^{\alpha\gamma} \omega_X^{\beta}(\del_{\gamma})(\del_\alpha \cdot\nabla_{\beta} \Psi - \del_{\beta} \cdot\nabla_{\alpha} \Psi) 
+ g^{\alpha\gamma} \Ric (X,\del_{\gamma})\del_\alpha \cdot\Psi
- \Box_gX \cdot\Psi,
\endaligned
\end{equation}
where
$$
\omega_X^{\beta}(\del_\gamma)\del_{\beta} = \nabla_{\gamma} X.
$$
\end{lemma}
\begin{proof} 
\bse
First of all, we have 
\begin{equation}
g^{\mu\nu} \del_\mu \cdot\del_\alpha \cdot\del_{\beta} \cdot \nabla_\nu \Psi 
+ g^{\mu\nu} \del_\alpha \cdot\del_\mu \cdot\del_{\beta} \cdot\nabla_\nu \Psi = -2\del_{\beta} \cdot\nabla_{\alpha} \Psi,
\end{equation}
\begin{equation}
g^{\mu\nu} \del_\alpha \cdot\del_\mu \cdot\del_{\beta} \cdot\nabla_\nu \Psi 
+ g^{\mu\nu} \del_\alpha \cdot\del_{\beta} \cdot\del_\mu \cdot\nabla_\nu \Psi = -2\del_\alpha \cdot\nabla_{\beta} \Psi.
\end{equation}
taking the difference of these two identities, we obtain:
\begin{equation}
g^{\mu\nu} \del_\mu \cdot\del_\alpha \cdot\del_{\beta} \cdot\nabla_\nu \Psi -  \del_{\alpha}\cdot\del_{\beta}\cdot\opDirac\Psi 
= 2\del_\alpha \cdot\nabla_{\beta} \Psi - 2\del_{\beta} \cdot\nabla_{\alpha} \Psi.
\end{equation}
So, by applying Lemma~\ref{lem1-24-feb-2025} along with the above identity, we find 
\begin{equation} \label{eq4-22-feb-2025}
\aligned
\opDirac(\del_\alpha \cdot\del_{\beta} \cdot\Psi)
& =  g^{\mu\nu} \del_\mu \cdot\del_\alpha \cdot\del_{\beta} \cdot\nabla_\nu \Psi 
\\
& =  \del_\alpha \cdot\del_{\beta} \cdot\opDirac\Psi + 2(\del_\alpha \cdot\nabla_{\beta} \Psi - \del_{\beta} \cdot\nabla_{\alpha} \Psi).
\endaligned
\end{equation}

According to our notation (or by Remark~\ref{remark-calculation}), we have 
\begin{equation} \label{equa-393F} 
\nabla_{\nu}(\omega_X^{\beta}(\del_{\gamma}))\del_{\beta} = \nabla_{\nu}(\omega_X^{\beta}(\del_{\gamma})\del_{\beta}) = \nabla_\nu \nabla_{\gamma}X, 
\end{equation}
and therefore  
\begin{equation}
\aligned
& 
g^{\mu\nu} \del_\mu \cdot\del_\alpha \cdot\nabla_\nu \big(g^{\alpha\gamma} \omega_X^{\beta}(\del_{\gamma})\big)\del_{\beta} \cdot\Psi
\\
& =  g^{\mu\nu}g^{\alpha\gamma}
\del_\mu \cdot\del_\alpha \cdot\nabla_\nu \nabla_{\gamma}X\cdot\Psi
\\
& = \frac{1}{2}g^{\mu\nu}g^{\alpha\gamma} \del_\mu \cdot\del_\alpha \cdot\nabla_\nu \nabla_{\gamma}X\cdot\Psi 
+ \frac{1}{2}g^{\alpha\gamma}g^{\mu\nu} \del_\alpha \cdot\del_\mu \cdot\nabla_{\gamma} \nabla_{\nu}X\cdot\Psi
\\
& = \frac{1}{2}g^{\mu\nu}g^{\alpha\gamma}
(\del_\mu \cdot\del_\alpha + \del_\alpha \cdot\del_{\mu})
\cdot\nabla_\nu \nabla_{\gamma}X\cdot\Psi
+ \frac{1}{2}g^{\mu\nu}g^{\alpha\gamma} \del_\alpha \cdot\del_\mu \cdot R(\del_{\gamma},\del_{\nu})X\cdot\Psi
\\
& =  -g^{\mu\nu} \nabla_{\mu} \nabla_{\nu}X\cdot\Psi
+ \frac{1}{2}g^{\mu\nu}g^{\alpha\gamma} \del_\alpha \cdot\del_\mu \cdot R(\del_{\gamma},\del_{\nu})X\cdot\Psi.
\endaligned
\end{equation}
For the last term, we expand
\begin{equation}
\aligned
& 
g^{\mu\nu}g^{\alpha\gamma} \del_\alpha \cdot\del_\mu \cdot R(\del_{\gamma},\del_{\nu})X
=  X^{\delta}g^{\mu\nu}g^{\alpha\gamma}
 R_{\gamma\nu\ \delta}^{\ \ \,\eta} \del_\alpha \cdot\del_\mu \cdot\del_{\eta} \cdot\Psi
\\
& =  X^{\delta}g^{\mu\nu}g^{\alpha\gamma}g^{\eta\beta}
R_{\gamma\nu\beta\delta}\del_\alpha \cdot\del_\mu \cdot\del_{\eta} \cdot\Psi
=  X^{\delta}g^{\mu\nu}g^{\alpha\gamma}g^{\eta\beta}
R_{\beta\delta\gamma\nu} \del_\alpha \cdot\del_\mu \cdot\del_{\eta} \cdot\Psi
\\
& = 2X^{\delta}g^{\alpha\gamma}R_{\delta\gamma} \del_\alpha \cdot\Psi.
\endaligned
\end{equation}
where in the last equality, we applied \eqref{eq3-22-feb-2025}. This leads us to
\begin{equation} \label{eq5-22-feb-2025}
g^{\mu\nu} \del_\mu \cdot\del_\alpha \cdot\nabla_\nu \big(g^{\alpha\gamma} \omega_X^{\beta}(\del_{\gamma})\big)\del_{\beta} \cdot\Psi 
= - \Box_gX\cdot\Psi
+ g^{\alpha\gamma} \Ric (X,\del_{\gamma})\del_\alpha \cdot\Psi.
\end{equation}
Using \eqref{eq4-22-feb-2025} and \eqref{eq5-22-feb-2025}, we conclude
$$
\aligned
& \opDirac\Big(g^{\alpha\gamma}\del_{\alpha}\cdot\nabla_{\gamma}X\cdot\Psi\Big)
- g^{\alpha\gamma} \del_\alpha \cdot\nabla_{\gamma}X \cdot\opDirac\Psi
\\
& =\opDirac\Big(\del_\alpha \cdot \big(g^{\alpha\gamma} \omega_X^{\beta}(\del_{\gamma})\big)\del_{\beta} \cdot\Psi\Big) 
- \del_\alpha \cdot\big(g^{\alpha\gamma} \omega_X^{\beta}(\del_{\gamma})\del_{\beta}\big) \cdot\opDirac\Psi
\\
& =  g^{\mu\nu}\nabla_{\nu}(g^{\alpha\gamma} \omega_X^{\beta}(\del_{\gamma})) (\del_\alpha \cdot\del_{\beta} \cdot\Psi)
+ g^{\alpha\gamma} \omega_X^{\beta}(\del_{\gamma})\opDirac(\del_\alpha \cdot\del_{\beta} \cdot\Psi)
- g^{\alpha\gamma} \omega_X^{\beta}(\del_{\gamma})\del_\alpha \cdot\del_{\beta} \cdot\opDirac\Psi
\\
& = - \Box_gX\cdot\Psi
+ g^{\alpha\gamma} \Ric (X,\del_{\gamma})\del_\alpha \cdot\Psi
+2g^{\alpha\gamma} \omega_X^{\beta}(\del_{\gamma})(\del_\alpha \cdot\nabla_{\beta} \Psi - \del_{\beta} \cdot\nabla_{\alpha} \Psi).  \qedhere
\endaligned
$$
\ese
\end{proof}


\begin{remark} 
\label{remark-calculation}
The derivation of \eqref{equa-393F} is straightforward, as follows. For the two-tensor
$
T(\mathrm{d}x^{\beta},\del_{\gamma}) := \omega_X^{\beta}(\del_{\gamma}), 
$
we can write 
$$
\nabla_{\nu}(\omega_X^{\beta}(\del_{\gamma})) = \nabla_{\nu}T^{\beta}_{\gamma} = \del_{\nu}T^{\beta}_{\gamma} 
+ T^{\delta}_{\gamma} \Gamma_{\nu\delta}^{\beta} 
- T^{\beta}_{\delta} \Gamma_{\nu\gamma}^{\delta}
$$
and, on the other hand,
$$
\aligned
\nabla_\nu \nabla_{\gamma}X 
& =  \nabla_{\nu}(\nabla_{\gamma}X) - \nabla_{\nabla_\nu \del_{\gamma}}X 
= \nabla_{\nu}(T_{\gamma}^{\beta} \del_{\beta}) - \nabla_{\nabla_\nu \del_{\gamma}}X
\\
& =  \del_{\nu}T^{\beta}_{\gamma} \del_{\beta} 
+ T^{\beta}_{\gamma} \Gamma_{\nu\beta}^{\delta} \del_{\delta}
- \Gamma_{\nu\gamma}^{\delta} \nabla_{\delta}X
= \big(\del_{\nu}T^{\beta}_{\gamma} + T_{\gamma}^{\delta} \Gamma_{\nu\delta}^{\beta} 
- \Gamma_{\nu\gamma}^{\delta}T_{\delta}^{\beta} \big)\del_{\beta}.
\endaligned
$$
\end{remark}


\begin{proof}[Proof of Proposition~\ref{prop1-22-feb-2025}]
We start from the identity
\[
-g^{\alpha\beta} \del_\alpha \cdot\nabla_{\nabla_{\beta}X} \Psi = -g^{\alpha\gamma} \omega_X^{\beta}(\del_{\gamma})\del_\alpha \cdot\nabla_{\beta} \Psi.
\]
Next, combining \eqref{eq2-22-feb-2025} and \eqref{eq6-22-feb-2025}, we obtain:
\[
\aligned
&[\nabla_X,\opDirac]\Psi - \frac{1}{4}[g^{\alpha\gamma} \del_\alpha \cdot\nabla_{\gamma}X\cdot,\opDirac]\Psi
\\
& = -\frac{1}{2}g^{\alpha\beta} \Ric (X,\del_{\beta})\del_\alpha \cdot \Psi 
- g^{\alpha\beta} \del_\alpha \cdot\nabla_{\nabla_{\beta}X} \Psi 
+ \frac{1}{2}g^{\alpha\gamma} \omega_X^{\beta}(\del_{\gamma})(\del_\alpha \cdot\nabla_{\beta} \Psi - \del_{\beta} \cdot\nabla_{\alpha} \Psi)
\\
& \quad - \frac{1}{4} \big(\Box_gX\cdot\Psi
- g^{\alpha\gamma} \Ric (X,\del_{\gamma})\del_\alpha \cdot\Psi\big).
\endaligned
\]
Rearranging terms and using the earlier identity, we arrive at 
\[
\aligned
&[\nabla_X,\opDirac]\Psi - \frac{1}{4}[g^{\alpha\gamma} \del_\alpha \cdot\nabla_{\gamma}X\cdot,\opDirac]\Psi
\\
& = -\frac{1}{4}g^{\alpha\beta} \Ric (X,\del_{\alpha})\del_{\beta} \cdot \Psi 
- \frac{1}{4} \Box_gX\cdot\Psi
\\
& \quad - g^{\alpha\gamma} \omega_X^{\beta}(\del_{\gamma})\del_\alpha \cdot\nabla_{\beta} \Psi
+ \frac{1}{2}g^{\alpha\gamma} \omega_X^{\beta}(\del_{\gamma})(\del_\alpha \cdot\nabla_{\beta} \Psi - \del_{\beta} \cdot\nabla_{\alpha} \Psi).
\endaligned
\]
For the last two terms in the right-hand side, we observe that 
\[
\aligned
- \frac{1}{2}g^{\alpha\gamma} \omega_X^{\beta}(\del_{\gamma})
(\del_\alpha \cdot\nabla_{\beta} \Psi + \del_{\beta} \cdot\nabla_{\alpha} \Psi)
& = - \frac{1}{2} \big(g^{\alpha\gamma} \omega_X^{\beta}(\del_{\gamma}) + g^{\beta\gamma} \omega_X^{\alpha}(\del_{\gamma})\big)\del_\alpha \cdot\nabla_{\beta} \Psi
\\
& = - \frac{1}{2} \pi[X]^{\alpha\beta} \del_\alpha \cdot\nabla_{\beta} \Psi.
\endaligned
\]
Thus, substituting this back, we obtain the desired result.
\end{proof}


\paragraph{Commutators in generalized wave coordinates.}

The wave coordinate condition $\Box_g x^\alpha$ on a choice of coordinates $x^\alpha$ is equivalent to saying 
\be
g^{\alpha\beta} \nabla_{\alpha} \del_{\beta} = 0.
\ee
Here we emphasize that the above expression is \emph{not tensorial,} and pick a (family of) choice(s) of coordinates. Under this condition, the previous identities simplify significantly.  
In fact, we will need a generalized wave coordinate conditions, defined by 
\begin{equation}\label{equa-26juillet2025-b}
g^{\alpha\beta} \nabla_{\alpha} \del_{\beta} = W, 
\end{equation}
where $W=W^{\alpha}\del_{\alpha}$ is a prescribed vector field. 

\begin{lemma}
For any vector field $X$ one has 
\begin{equation}
\aligned
\Box_g X = [X,W] - g^{\alpha\beta}\Ric(X,\del_{\alpha})\del_{\beta}
+ \pi[X]^{\alpha\beta}\nabla_{\alpha}\del_{\beta} 
+ g^{\alpha\beta}\del_{\alpha}\del_{\beta}(X^{\gamma})\del_{\gamma}, 
\endaligned
\end{equation}
where $W = g^{\alpha\beta}\nabla_{\alpha}\del_{\beta}$ is defined in \eqref{equa-26juillet2025-b}. 
\end{lemma}

\begin{proof} We have 
$$
\aligned
&g^{\alpha\beta}\nabla_{\alpha}\nabla_{\beta}X 
=  g^{\alpha\beta}\nabla_{\alpha}(\nabla_{\beta}X) 
- g^{\alpha\beta}\nabla_{\nabla_{\alpha}\del_{\beta}}X
\\
& =  g^{\alpha\beta}\nabla_{\alpha}(\nabla_X\del_{\beta})
+ g^{\alpha\beta}\nabla_{\alpha}([\del_{\beta},X])
- \nabla_WX
\\
& = g^{\alpha\beta}\nabla_X(\nabla_{\alpha}\del_{\beta}) 
+ g^{\alpha\beta}\nabla_{[\del_{\alpha},X]}\del_{\beta}
 + g^{\alpha\beta}\Riem(\del_{\alpha},X)\del_{\beta}
+ g^{\alpha\beta}\nabla_{\alpha}([\del_{\beta},X])
- \nabla_WX
\\
& = \nabla_X(g^{\alpha\beta}\nabla_{\alpha}\del_{\beta})- \nabla_WX 
 +X^{\mu}g^{\alpha\beta}R_{\alpha\mu\ \beta}^{\quad\nu}\del_{\nu}
\\
& \quad -  X(g^{\alpha\beta})\nabla_{\alpha}\del_{\beta}
+ g^{\alpha\beta}\nabla_{[\del_{\alpha},X]}\del_{\beta}
+ g^{\alpha\beta}\nabla_{\alpha}([\del_{\beta},X]), 
\endaligned
$$
therefore 
$$
\aligned
&g^{\alpha\beta}\nabla_{\alpha}\nabla_{\beta}X 
\\
& =  \nabla_XW - \nabla_WX 
+ g^{\alpha\beta}X^{\gamma} R_{\alpha\gamma\ \beta}^{\quad \nu}\del_{\nu}
- X(g^{\alpha\beta})\nabla_{\alpha}\del_{\beta}
+ g^{\alpha\beta}\nabla_{[\del_{\alpha},X]}\del_{\beta}
+ g^{\alpha\beta}\nabla_{\alpha}([\del_{\beta},X])
\\
& = \nabla_XW - \nabla_WX - g^{\alpha\beta}\Ric(X,\del_{\alpha})\del_{\beta}
- X(g^{\alpha\beta})\nabla_{\alpha}\del_{\beta}
+ g^{\alpha\beta}\nabla_{[\del_{\alpha},X]}\del_{\beta}
+ g^{\alpha\beta}\nabla_{\alpha}([\del_{\beta},X]).
\endaligned
$$
We also observe that
$$
\aligned
&X(g^{\alpha\beta})\nabla_{\alpha}\del_{\beta}
- g^{\alpha\beta}\nabla_{[\del_{\alpha},X]}\del_{\beta}
- g^{\alpha\beta}\nabla_{\alpha}([\del_{\beta},X])
\\
& = X(g^{\alpha\beta})\nabla_{\alpha}\del_{\beta}
- g^{\alpha\beta}\del_{\alpha}X^{\gamma}\nabla_{\gamma}\del_{\beta}
- g^{\alpha\beta}\nabla_{\alpha}(\del_{\beta}X^{\gamma}\del_{\gamma})
\\
& = \big(X(g^{\alpha\beta}) - g^{\gamma\beta}\del_{\gamma}X^{\alpha} - g^{\alpha\gamma}\del_{\gamma}X^{\beta}\big)\nabla_{\alpha}\del_{\beta} 
- g^{\alpha\beta}\del_{\alpha}\del_{\beta}(X^{\gamma})\del_{\gamma}
\\
& =  - \pi[x]^{\alpha\beta}\nabla_{\alpha}\del_{\beta} 
- g^{\alpha\beta}\del_{\alpha}\del_{\beta}(X^{\gamma})\del_{\gamma}, 
\endaligned
$$
which gives us the desired result.
\end{proof}

\begin{proposition}[Commutation relation for the Dirac operator in generalized wave gauge]
\label{prop1-16-july-2025}
With the notation $W = W^{\alpha}\del_{\alpha} = g^{\alpha\beta}\nabla_{\alpha}\del_{\beta}$
in \eqref{equa-26juillet2025-b}, for any vector field $X$ defined in (a subset of) $\Mcal_{[s_0,s_1]}$, one has 
\begin{equation}
\aligned
\, [\widehat{X},\opDirac]\Psi  & =   - \frac{1}{2} \pi[X]^{\alpha\beta} \del_\alpha \cdot\nabla_{\beta} \Psi
- \frac{1}{4}[X,W]\cdot\Psi
\\
& \quad - \frac{1}{4}\pi[X]^{\alpha\beta}\nabla_{\alpha}\del_{\beta}\cdot\Psi
- \frac{1}{4}g^{\mu\nu}\del_{\mu}\del_{\nu}(X^{\gamma})\del_{\gamma}\cdot\Psi .
\endaligned
\end{equation}
\end{proposition}

\begin{proof} It suffices to compute 
$$
\aligned
\, [\widehat{X},\opDirac]\Psi 
& =   - \frac{1}{2} \pi[X]^{\alpha\beta} \del_\alpha \cdot\nabla_{\beta} \Psi 
-\frac{1}{4}g^{\alpha\beta} \Ric (X,\del_{\alpha})\del_{\beta} \cdot \Psi 
- \frac{1}{4} \Box_gX\cdot\Psi
\\
& =  - \frac{1}{2} \pi[X]^{\alpha\beta} \del_\alpha \cdot\nabla_{\beta} \Psi 
-\frac{1}{4}g^{\alpha\beta} \Ric (X,\del_{\alpha})\del_{\beta} \cdot \Psi 
\\
& \quad - \frac{1}{4}[X,W]\cdot\Psi
+ \frac{1}{4}g^{\alpha\beta}\Ric(X,\del_{\alpha})\del_{\beta}\cdot\Psi
\\
& \quad - \frac{1}{4}\pi[X]^{\alpha\beta}\nabla_{\alpha}\del_{\beta} \cdot\Psi
- \frac{1}{4}g^{\alpha\beta}\del_{\alpha}\del_{\beta}(X^{\gamma})\del_{\gamma}.
\qedhere
\endaligned
$$
\end{proof}

}


\subsection{ Adapted action on tensor fields}
\label{section===54}

{ 

Another important issue is differentiating the Clifford product with respect to a vector field. For convenience in our analysis, we introduce the {\sl adapted spinorial action} on vector and spinor fields (which is also known as the {\sl Lie derivative on spinor fields} in the flat case), defined by
\begin{equation}\label{eq3-04-july-2025}
Z\Psi := \widehat{Z}\Psi = \nabla_Z\Psi - \frac{1}{4}g^{\mu\nu}\del_{\mu}\cdot\nabla_{\nu}Z\cdot\Psi. 
\end{equation}
\begin{equation}\label{eq2-04-july-2025}
ZX:=\widehat{Z}X = \nabla_ZX 
+ \frac{1}{4}g^{\mu\nu}\big(X\cdot\del_{\mu}\cdot\nabla_{\nu}Z - \del_{\mu}\cdot\nabla_{\nu}Z\cdot X\big),
\end{equation}
where the covariant derivatives are the Levi-Civita and the spinorial ones, respectively. This convention allows us to express the Leibniz rules as follows:
\begin{equation}\label{eq4-04-july-2025}
\aligned
Z(fX) & =   Z(f)X + fZX,
\\
Z(f\Psi) & =  Z(f)\Psi + fZ\Psi,
\\
Z(X\cdot\Psi) & =  ZX\cdot\Psi + X\cdot Z\Psi, 
\endaligned
\end{equation}
where $f$ denotes any real-valued scalar field.
For instance, the later identity is checked by writing 
\begin{equation}
\aligned
Z(X\cdot\Psi) & = \, \nabla_Z(X\cdot\Psi) - \frac{1}{4}g^{\mu\nu}\del_{\mu}\cdot\nabla_{\nu}Z\cdot X\cdot\Psi
\\
& = \nabla_ZX\cdot\Psi + X\cdot\nabla_Z\Psi - \frac{1}{4}g^{\mu\nu}\del_{\mu}\cdot\nabla_{\nu}Z\cdot X\cdot\Psi
\\
& = \,\nabla_ZX\cdot\Psi 
+ \frac{1}{4}g^{\mu\nu}\big(X\cdot\del_{\mu}\cdot\nabla_{\nu}Z - \del_{\mu}\cdot\nabla_{\nu}Z\cdot X\big)\cdot\Psi
+ X\cdot\Big(\nabla_Z\Psi - \frac{1}{4}g^{\mu\nu}\del_{\mu}\cdot\nabla_{\nu}\cdot\Psi\Big)
\\
& =  ZX\cdot\Psi + X\cdot Z\Psi.
\endaligned
\end{equation}
With this notation, we can also establish the following identity relating the notions of adapted derivative, Lie derivative, and deformation tensor.  

\begin{lemma}\label{lem4-06-oct-2025(l)}
For any vector fields $X,Z$ one has 
\begin{equation}\label{eq1-04-july-2025}
ZX = \Lcal_ZX + \frac{1}{2}g^{\mu\nu}X^{\alpha}\pi[Z]_{\alpha\nu}\del_{\mu}.
\end{equation}
\end{lemma}

\begin{proof} It suffices to compute 
$$
\aligned
& \frac{1}{4}g^{\mu\nu}\big(X\cdot\del_{\mu}\cdot\nabla_{\nu}Z - \del_{\mu}\cdot\nabla_{\nu}Z\cdot X\big)
\\
& = \frac{1}{2}g^{\mu\nu}g(X,\nabla_{\nu}Z)\del_{\mu}
+\frac{1}{4}g^{\mu\nu}\big(X\cdot\del_{\mu}\cdot\nabla_{\nu}Z + \del_{\mu}\cdot X\cdot\nabla_{\nu}Z\big)  
\\
& = \frac{1}{2}g^{\mu\nu}g(X,\nabla_{\nu}Z)\del_{\mu} 
- \frac{1}{2}g^{\mu\nu}g(X,\del_{\mu})\nabla_{\nu}Z
\\
& =  \frac{1}{2}g^{\mu\nu}X^{\alpha}\nabla_{\nu}Z^{\beta}g_{\alpha\beta}\del_{\mu}
- \frac{1}{2}g^{\mu\nu}X^{\alpha}g_{\alpha\mu}\nabla_{\nu}Z
\\
& = \frac{1}{2}g^{\mu\nu}X^{\alpha}\big(\pi[Z]_{\nu\alpha} - \nabla_{\alpha}Z^{\beta}g_{\beta\nu}\big)\del_{\mu} - \frac{1}{2}\nabla_XZ
=\frac{1}{2}g^{\mu\nu}X^{\alpha}\pi[Z]_{\alpha\nu}\del_{\mu} - \nabla_XZ, 
\endaligned
$$
where we used
$
\pi[Z]_{\nu\alpha} = \nabla_{\nu}Z^{\beta}g_{\alpha\beta} + \nabla_{\alpha}Z^{\beta}g_{\nu\beta}.
$
\end{proof}

}


\subsection{ Commutators for general vector fields}
\label{section===56}

{ 

We first establish the following identity.

\begin{proposition}\label{prop1-22-june-2025}
Consider vector fields $X,Y$ defined in (a subset of) $\Mcal_{[s_0,s_1]}$. If $\widehat{X},\widehat{Y}$ are defined as in \eqref{equa-511m}, one has 
\begin{equation}
[\widehat{X},\widehat{Y}]\Psi = \widehat{[X,Y]}\Psi - \frac{1}{8}\Big(g^{\alpha\beta}\pi[Y]^{\mu\nu}\pi[X]_{\beta\nu}\del_{\mu}\cdot\del_{\alpha} + \pi[X]^{\alpha\beta}\pi[Y]_{\alpha\beta}\Big)\cdot \Psi.
\end{equation}
\end{proposition}

The proof consists of two steps. First, we show that the commutators contain (counter-intuitively) \emph{no curvature components}. Second, we express the commutator as a quadratic form of the deformation tensor.

\begin{lemma}
Based on the assumptions in Proposition~\ref{prop1-22-june-2025}, one also has 
\begin{equation}\label{eq2-22-june-2025}
\aligned
\, [\widehat{X},\widehat{Y}]\Psi & =  \widehat{[X,Y]}\Psi
- \frac{1}{4}g^{\mu\nu}\del_{\mu}\cdot(\nabla_{\nabla_{\nu}Y}X 
- \nabla_{\nabla_{\nu}X}Y)\cdot\Psi
\\
& \quad + \frac{1}{16}g^{\alpha\beta}g^{\mu\nu}
\big(\del_{\alpha}\cdot\nabla_{\beta}X\cdot\del_{\mu}\cdot\nabla_{\nu}Y
- \del_{\mu}\cdot\nabla_{\nu}Y\cdot\del_{\alpha}\cdot\nabla_{\beta}X\big)\cdot\Psi.
\endaligned
\end{equation}
\end{lemma}

\begin{proof}
\bse
By a direct calculation, we have 
\be
\aligned
\widehat{X}(\widehat{Y}\Psi)
& =  \nabla_X\big(\nabla_Y \Psi\big) 
- \frac{1}{4}\nabla_X(g^{\mu\nu}\del_{\mu}\cdot\nabla_{\nu}Y\cdot\Psi)
\\
& \quad -  \frac{1}{4}g^{\alpha\beta}\del_{\alpha}\cdot\nabla_{\beta}X\cdot
\Big(\nabla_Y\Psi - \frac{1}{4}g^{\mu\nu}\del_{\mu}\cdot\nabla_\nu Y\cdot\Psi\Big)
\\
& =  \nabla_X\big(\nabla_Y \Psi\big) 
- \frac{1}{4}g^{\mu\nu}\del_{\mu}\cdot\nabla_X\nabla_{\nu}Y\cdot\Psi 
- \frac{1}{4}g^{\mu\nu}\del_{\mu}\cdot\nabla_{\nu}Y\cdot\nabla_X\Psi
\\
& \quad - \frac{1}{4}g^{\alpha\beta}\del_{\alpha}\cdot\nabla_{\beta}X\cdot\nabla_Y\Psi
+\frac{1}{16}g^{\alpha\beta}g^{\mu\nu}\del_{\alpha}\cdot\nabla_{\beta}X\cdot\del_{\mu}\cdot\nabla_{\nu}Y\cdot\Psi.
\endaligned
\ee
Thus we deduce that 
\begin{equation}\label{eq2-20-june-2025}
\aligned
\, [\widehat{X},\widehat{Y}]\Psi
& = [\nabla_X,\nabla_Y]\Psi 
- \frac{1}{4}g^{\mu\nu}\del_{\mu}
\cdot\big(\nabla_X\nabla_{\nu}Y - \nabla_Y\nabla_{\nu}X\big)\cdot\Psi
\\
& \quad + \frac{1}{16}g^{\alpha\beta}g^{\mu\nu}
\big(\del_{\alpha}\cdot\nabla_{\beta}X\cdot\del_{\mu}\cdot\nabla_{\nu}Y
- \del_{\mu}\cdot\nabla_{\nu}Y\cdot\del_{\alpha}\cdot\nabla_{\beta}X\big)\cdot\Psi.
\endaligned
\end{equation}
The first term in the right-hand side above can be written in terms of the curvature, namely 
\begin{equation}\label{eq1-22-june-2025}
\aligned
\, [\nabla_X,\nabla_Y]\Psi & =  \frac{1}{4}g^{\mu\nu}g^{\gamma\gamma'}X^{\alpha}Y^{\beta} R_{\alpha\beta\gamma'\nu}\del_{\mu}\cdot\del_{\gamma}\cdot\Psi + \nabla_{[X,Y]}\Psi
\\
& = - \frac{1}{4}g^{\mu\nu}g^{\gamma\gamma'}X^{\alpha}Y^{\beta} R_{\gamma'\nu\beta\alpha}\del_{\mu}\cdot\del_{\gamma}\cdot\Psi + \nabla_{[X,Y]}\Psi.
\endaligned
\end{equation}
For the second term in the right-hand side of \eqref{eq2-20-june-2025}, we observe that
$$
\aligned
\nabla_X\nabla_{\nu}Y - \nabla_Y\nabla_{\nu}X
& = 
\nabla_X\nabla_{\nu}Y - \nabla_{\nu}\nabla_XY + \nabla_{\nu}\nabla_XY 
- \nabla_{\nu}\nabla_YX 
\\
& \quad +  \nabla_{\nu}\nabla_YX - \nabla_Y\nabla_{\nu}X
\\
& = 
R(X,\del_{\nu})Y + R(\del_{\nu},Y)X 
+ \nabla_{\nu}([X,Y]) + \nabla_{\nabla_{\nu}Y}X - \nabla_{\nabla_{\nu}X}Y.
\endaligned
$$
We focus particularly on the curvature terms and write 
$$
g^{\mu\nu}\del_{\mu}\cdot R(X,\del_{\nu})Y 
= X^{\alpha}Y^{\beta}g^{\mu\nu}R_{\alpha\nu\ \beta}^{\quad\gamma}\del_{\mu}\cdot\del_{\gamma}
= g^{\mu\nu}g^{\gamma\gamma'}X^{\alpha}Y^{\beta} R_{\alpha\nu\gamma'\beta}\del_{\mu}\cdot\del_{\gamma}
$$
and 
$$
g^{\mu\nu}\del_{\mu}\cdot R(\del_{\nu},Y)X 
= X^{\alpha}Y^{\beta}g^{\mu\nu}R_{\nu\beta\ \alpha}^{\quad\gamma}\del_{\mu}\cdot\del_{\gamma}
= g^{\mu\nu}g^{\gamma\gamma'}X^{\alpha}Y^{\beta} R_{\nu\beta\gamma'\alpha}\del_{\mu}\cdot\del_{\gamma}.
$$
Then we obtain 
\begin{equation}
\aligned
&
- \frac{1}{4}g^{\mu\nu}\big(\nabla_X\nabla_{\nu}Y - \nabla_Y\nabla_{\nu}X\big)
\\
& =  - \frac{1}{4}g^{\mu\nu}g^{\gamma\gamma'}X^{\alpha}X^{\beta}(R_{\gamma'\beta\alpha\nu} +  R_{\gamma'\alpha\nu\beta})\del_{\mu}\cdot\del_{\gamma}\cdot\Psi
\\
& \quad - \frac{1}{4}g^{\mu\nu}\del_{\mu}\cdot\nabla_{\nu}([X,Y])\cdot\Psi
- \frac{1}{4}g^{\mu\nu}\del_{\mu}\cdot(\nabla_{\nabla_{\nu}Y}X 
- \nabla_{\nabla_{\nu}X}Y)\cdot\Psi.
\endaligned
\end{equation}
This identity, together with \eqref{eq1-22-june-2025} leads us to
$$
\aligned
\, &[\widehat{X},\widehat{Y}]\Psi 
\\
& =  \widehat{[X,Y]}\Psi
- \frac{1}{4}g^{\mu\nu}g^{\gamma\gamma'}X^{\alpha}X^{\beta}
 (R_{\gamma'\nu\beta\alpha} + R_{\gamma'\beta\alpha\nu} + R_{\gamma'\alpha\nu\beta})\del_{\mu}\cdot\del_{\gamma}\cdot\Psi
\\
& \quad - \frac{1}{4}g^{\mu\nu}\del_{\mu}\cdot(\nabla_{\nabla_{\nu}Y}X 
- \nabla_{\nabla_{\nu}X}Y)\cdot\Psi
\\
& \quad + \frac{1}{16}g^{\alpha\beta}g^{\mu\nu}
\big(\del_{\alpha}\cdot\nabla_{\beta}X\cdot\del_{\mu}\cdot\nabla_{\nu}Y
- \del_{\mu}\cdot\nabla_{\nu}Y\cdot\del_{\alpha}\cdot\nabla_{\beta}X\big)\cdot\Psi.
\endaligned
$$
For the second term in the right-hand side, we apply the Bianchi identity, and obtain the desired result.
\ese
\end{proof} 


It remains to handle the last term in the right-hand side of \eqref{eq2-22-june-2025}.

\begin{lemma}
With the notation and assumption in Proposition~\ref{prop1-22-june-2025}, one has 
\begin{equation}\label{eq5-22-june-2025}
\aligned
& \frac{1}{16}g^{\alpha\beta}g^{\mu\nu}
\big(\del_{\alpha}\cdot\nabla_{\beta}X\cdot\del_{\mu}\cdot\nabla_{\nu}Y
- \del_{\mu}\cdot\nabla_{\nu}Y\cdot\del_{\alpha}\cdot\nabla_{\beta}X\big)\cdot\Psi
\\
& = 
- \frac{1}{8}g^{\alpha\beta}\pi[Y]^{\mu\gamma}\pi[X]_{\beta\gamma}\del_{\mu}\cdot\del_{\alpha} 
- \frac{1}{4}\pi[X]^{\alpha\beta}g(\nabla_{\alpha}Y,\del_{\beta})
+\frac{1}{4}g^{\mu\nu}\del_{\mu}\cdot(\nabla_{\nabla_{\nu}Y}X - \nabla_{\nabla_{\nu}X}Y).
\endaligned
\end{equation}
\end{lemma}

\begin{proof} 
\bse
First of all, we compute the left-hand side of \eqref{eq5-22-june-2025} (up to a multiplicative factor $16$), that is, 
\begin{equation}\label{eq4-22-june-2025}
\aligned
A & := g^{\alpha\beta}g^{\mu\nu}
\big(\del_{\alpha}\cdot\nabla_{\beta}X\cdot\del_{\mu}\cdot\nabla_{\nu}Y
- \del_{\mu}\cdot\nabla_{\nu}Y\cdot\del_{\alpha}\cdot\nabla_{\beta}X\big)
\\
& = \,g^{\alpha\beta}g^{\mu\nu}
\big(- \del_{\alpha}\cdot\del_{\mu}\cdot\nabla_{\beta}X\cdot\nabla_{\nu}Y 
-2g(\nabla_{\beta}X,\del_{\mu})\del_{\alpha}\cdot\nabla_{\nu}Y\big)
\\
& \quad + g^{\alpha\beta}g^{\mu\nu}\big(\del_{\mu}\cdot\del_{\alpha}\cdot\nabla_{\nu}Y\cdot\nabla_{\beta}X+2g(\nabla_{\nu}Y,\del_{\alpha})\del_{\mu}\cdot\nabla_{\beta}X\big)
\\
& = g^{\alpha\beta}g^{\mu\nu}\del_{\mu}\cdot\del_{\alpha}\cdot
\big(\nabla_{\beta}X\cdot\nabla_{\nu}Y + \nabla_{\nu}Y\cdot\nabla_{\beta}X\big)
+2g^{\alpha\beta}g^{\mu\nu}g_{\alpha\mu}\nabla_{\beta}X\cdot\nabla_{\nu}Y 
\\
& \quad - 2g^{\alpha\beta}g^{\mu\nu}g(\nabla_{\beta}X,\del_{\mu})\del_{\alpha}\cdot\nabla_{\nu}Y
+2g^{\alpha\beta}g^{\mu\nu}g(\nabla_{\nu}Y,\del_{\alpha})\del_{\mu}\cdot\nabla_{\beta}X
\\
& = - 2g^{\alpha\beta}g^{\mu\nu}g(\nabla_{\nu}Y,\nabla_{\beta}X)\del_{\mu}\cdot\del_{\alpha}
+2g^{\alpha\beta}g^{\mu\nu}g_{\alpha\mu}\nabla_{\beta}X\cdot\nabla_{\nu}Y 
\\
& \quad - 2g^{\alpha\beta}g^{\mu\nu}g(\nabla_{\beta}X,\del_{\mu})\del_{\alpha}\cdot\nabla_{\nu}Y
+2g^{\alpha\beta}g^{\mu\nu}g(\nabla_{\nu}Y,\del_{\alpha})\del_{\mu}\cdot\nabla_{\beta}X
\\
& =: -2A_1 + 2A_2 - 2A_3 + 2A_4. 
\endaligned
\end{equation}

With the notation $\nabla_{\alpha}X^{\beta}\del_{\beta} = \nabla_{\alpha}X$, we obtain 
\begin{equation}
\aligned
&g_{\gamma\beta}\nabla_{\alpha}X^{\gamma} + g_{\gamma\alpha}\nabla_{\beta}X^{\gamma} = \pi[X]_{\alpha\beta},
\\
&g^{\gamma\beta}\nabla_{\gamma}X^{\alpha} + g^{\gamma\alpha}\nabla_{\gamma}X^{\beta}
= \pi[X]^{\alpha\beta},
\endaligned
\end{equation}
Now for the first term in the right-hand side of \eqref{eq4-22-june-2025}, we find 
$$
\aligned
A_1 & =  g^{\alpha\beta}g^{\mu\nu}g(\nabla_{\nu}Y,\nabla_{\beta}X)\del_{\mu}\cdot\del_{\alpha}
\\
& = g^{\alpha\beta}g^{\mu\nu}\nabla_{\beta}X^{\delta}\nabla_{\nu}Y^{\gamma}g_{\gamma\delta}\del_{\mu}\cdot\del_{\alpha}
\\
& = g^{\alpha\beta}g^{\mu\nu}(\nabla_{\beta}X^{\delta}g_{\gamma\delta} + \nabla_{\gamma}X^{\delta}g_{\beta\delta})\nabla_{\nu}Y^{\gamma}\del_{\mu}\cdot\del_{\alpha}
-g^{\alpha\beta}g^{\mu\nu}\nabla_{\gamma}X^{\delta}g_{\beta\delta}\nabla_{\nu}Y^{\gamma}\del_{\mu}\cdot\del_{\alpha}
\\
& = g^{\alpha\beta}g^{\mu\nu}\pi[X]_{\beta\gamma}\nabla_{\nu}Y^{\gamma}\del_{\mu}\cdot\del_{\alpha} - g^{\mu\nu}\nabla_{\gamma}X^{\alpha}\nabla_{\nu}Y^{\gamma}\del_{\mu}\cdot\del_{\alpha}
\\
& = g^{\alpha\beta}g^{\mu\nu}\pi[X]_{\beta\gamma}\nabla_{\nu}Y^{\gamma}\del_{\mu}\cdot\del_{\alpha} - g^{\mu\nu}\del_{\mu}\cdot\nabla_{\nabla_{\nu}Y}X
\\
& = g^{\alpha\beta}\big(g^{\mu\nu}\nabla_{\nu}Y^{\gamma}
+ g^{\gamma\nu}\nabla_{\nu}Y^{\mu}\big)\pi[X]_{\beta\gamma}\del_{\mu}\cdot\del_{\alpha} 
\\
&\quad- g^{\alpha\beta}g^{\gamma\nu}\nabla_{\nu}Y^{\mu}\pi[X]_{\beta\gamma}\del_{\mu}\cdot\del_{\alpha} 
- g^{\mu\nu}\del_{\mu}\cdot\nabla_{\nabla_{\nu}Y}X
\\
& = g^{\alpha\beta}\pi[Y]^{\mu\gamma}\pi[X]_{\beta\gamma}\del_{\mu}\cdot\del_{\alpha} 
- g^{\alpha\beta}g^{\gamma\nu}\pi[X]_{\beta\gamma}\nabla_{\nu}Y\cdot\del_{\mu}\cdot\del_{\alpha}
- g^{\mu\nu}\del_{\mu}\cdot\nabla_{\nabla_{\nu}Y}X,
\endaligned
$$
while
\be
\aligned
A_2 = g^{\mu\nu}\nabla_{\mu}X\cdot\nabla_{\nu}Y 
& =  g^{\mu\nu}\nabla_{\mu}X^{\gamma}\nabla_{\nu}Y^{\delta}\del_{\gamma}\cdot\del_{\delta}
\\
& = \big(g^{\mu\nu}\nabla_{\mu}X^{\gamma} + g^{\mu\gamma}\nabla_{\mu}X^{\nu}\big)\nabla_{\nu}Y^{\delta}\del_{\gamma}\cdot\del_{\delta}
-g^{\mu\gamma}\nabla_{\mu}X^{\nu}\nabla_{\nu}Y^{\delta}\del_{\gamma}\cdot\del_{\delta}
\\
& = \pi[X]^{\nu\gamma}\del_{\gamma}\cdot\nabla_{\nu}Y
- g^{\mu\gamma}\del_{\gamma}\cdot\nabla_{\nabla_{\mu}X}Y, 
\endaligned
\ee
and 
\be
\aligned
&A_3 = g^{\alpha\beta}g^{\mu\nu}g(\nabla_{\beta}X,\del_{\mu})\del_{\alpha}\cdot\nabla_{\nu}Y
=g^{\mu\nu}g^{\alpha\beta}\nabla_{\beta}X^{\gamma}g_{\gamma\mu}\del_{\alpha}\cdot\nabla_{\nu}Y
=g^{\alpha\beta}\del_{\alpha}\cdot\nabla_{\nabla_{\beta}X}Y,
\\
&A_4 = g^{\alpha\beta}g^{\mu\nu}g(\nabla_{\nu}Y,\del_{\alpha})\del_{\mu}\cdot\nabla_{\beta}X 
=g^{\mu\nu}g^{\alpha\beta}\nabla_{\nu}Y^{\gamma}g_{\gamma\alpha}\del_{\mu}\cdot\nabla_{\beta}X
=g^{\mu\nu}\del_{\mu}\cdot\nabla_{\nabla_{\nu}Y}X.
\endaligned
\ee
We thus return to $A$ and deduce that 
$$
\aligned
{A \over 2}  
& = -g^{\alpha\beta}\pi[Y]^{\mu\gamma}\pi[X]_{\beta\gamma}\del_{\mu}\cdot\del_{\alpha} 
+ g^{\alpha\beta}g^{\gamma\nu}\pi[X]_{\beta\gamma}\nabla_{\nu}Y\cdot\del_{\alpha} 
+ g^{\mu\nu}\del_{\mu}\cdot\nabla_{\nabla_{\nu}Y}X
\\
& \quad +  \pi[X]^{\nu\gamma}\del_{\gamma}\cdot\nabla_{\nu}Y
- g^{\mu\gamma}\del_{\gamma}\cdot\nabla_{\nabla_{\mu}X}Y
-g^{\alpha\beta}\del_{\alpha}\cdot\nabla_{\nabla_{\beta}X}Y
+g^{\mu\nu}\del_{\mu}\cdot\nabla_{\nabla_{\nu}Y}X
\\
& = -g^{\alpha\beta}\pi[Y]^{\mu\gamma}\pi[X]_{\beta\gamma}\del_{\mu}\cdot\del_{\alpha} 
+ g^{\alpha\beta}g^{\gamma\nu}\pi[X]_{\beta\gamma}\nabla_{\nu}Y\cdot\del_{\alpha} 
+\pi[X]^{\nu\gamma}\del_{\gamma}\cdot\nabla_{\nu}Y
\\
& \quad + 2g^{\mu\nu}\del_{\mu}\cdot(\nabla_{\nabla_{\nu}Y}X - \nabla_{\nabla_{\nu}X}Y)
\\
& = \pi[X]^{\nu\gamma}(\nabla_{\nu}Y\cdot\del_{\gamma} + \del_{\gamma}\cdot\nabla_{\nu}Y)
-g^{\alpha\beta}\pi[Y]^{\mu\gamma}\pi[X]_{\beta\gamma}\del_{\mu}\cdot\del_{\alpha}  
+2g^{\mu\nu}\del_{\mu}\cdot(\nabla_{\nabla_{\nu}Y}X - \nabla_{\nabla_{\nu}X}Y)
\\
& = -2\pi[X]^{\alpha\beta}g(\nabla_{\alpha}Y,\del_{\beta})
-g^{\alpha\beta}\pi[Y]^{\mu\gamma}\pi[X]_{\beta\gamma}\del_{\mu}\cdot\del_{\alpha}  
+2g^{\mu\nu}\del_{\mu}\cdot(\nabla_{\nabla_{\nu}Y}X - \nabla_{\nabla_{\nu}X}Y).
\endaligned
$$
\ese
\end{proof}


\begin{proof}[Proof of Proposition~\ref{prop1-22-june-2025}]
Thanks to the symmetry property of the deformation tensor $\pi[X]^{\alpha\beta}$, we have 
\be
\aligned
\pi[X]^{\alpha\beta}g(\nabla_{\alpha}Y,\del_{\beta})
=\frac{1}{2}\pi[X]^{\alpha\beta}\big(\nabla_{\alpha}Y^{\gamma}g_{\gamma\beta} + \nabla_{\beta}Y^{\gamma}g_{\gamma\alpha}\big)
=\frac{1}{2}\pi[X]^{\alpha\beta}\pi[Y]_{\alpha\beta}.
\endaligned
\ee
Substituting this result and \eqref{eq5-22-june-2025} into \eqref{eq2-22-june-2025}, we obtain the desired result.
\end{proof}

}


\section{Geometry of the Euclidean-hyperboloidal foliation}
\label{section=N4}

\subsection{ The Euclidean-hyperboloidal foliation}
\label{section===71}

{

\paragraph{Time function and foliation.}

We work with a general spacetime $(\mathcal{M},g)$ of dimension $3+1$, equipped with a global coordinate chart $(x^\alpha)$ (with $\alpha=0,1,2,3\}$) and we set $t = x^0$ and $x = (x^1,x^2,x^3)$. The notion of Euclidean-hyperboloidal foliation proposed in \cite{PLF-YM-PDE} is based on a choice of a \emph{time function} $T=T(s,r)$ satisfying the differential equation 
\begin{equation}\label{eq3-28-july-2025}
\del_rT(s,r) = \frac{r \, \xi(s,r)}{\sqrt{s^2+r^2}},
\qquad T(s,0) = s, 
\end{equation}
where $\xi=\xi(s,r) \in [0,1]$ is a \emph{cut-off function} satisfying 
\begin{equation}
\xi(s,r) =
\begin{cases}
1, \quad &r<r^{\Hcal}(s), 
\\
0, \qquad &r>r^{\Ecal}(s).
\end{cases}
\end{equation}
Here, we have introduced the hyperboloidal radius $r^{\Hcal}(s) := \frac{1}{2}(s^2-1)$ and the Euclidean radius  $r^{\Ecal}(s) = \frac{1}{2}(s^2+1)$, as we call them. The hypersurfaces
\be
\mathcal{M}_s := \{(t,x) \, / \,  t = T(s,r)\} 
\ee
we remark that
\begin{equation}\label{eq1-19-march-2025}
0\leq \del_rT<1.
\end{equation}
The collection of these hypersurfaces form a foliation of the spacetime
\be
\mathcal{M}_{[s_0,+\infty)} = \bigcup_{s\in [s_0,+\infty)} \mathcal{M}_s.
\ee
whose past boundary is $\mathcal{M}_{s_0}$. 
Let us denote by $\vec{n}$ the future oriented, normal unit vector on each slice $\mathcal{M}_s$, by $\sigmab$ the (Riemannian) restriction of the Lorentzian metric $g$, and by $\mathrm{Vol}_{\sigmab}$ the associated volume form. 

In addition, we decompose $\Mcal_s$ into three regions, as follows: 
\be
\mathcal{M}_s^{\Hcal}:=\mathcal{M}_s\cap\{r\leq r^{\Hcal}(s)\}, \quad 
\mathcal{M}_s^{\mathcal{M}} := \mathcal{M}_s\cap\{r^{\Hcal}(s)\leq r\leq r^{\Ecal}(s)\}, \quad
\mathcal{M}_s^{\Ecal}:=\mathcal{M}_s\cap \{r\geq r^{\Ecal}(s)\}, 
\ee
and we set  
\be
\mathcal{M}_{[s_0,s_1]}^{\Hcal} := \bigcup_{s_0\leq s\leq s_1} \mathcal{M}_s^{\Hcal}.
\ee
We also use the notation $\mathcal{M}^{\EM}_s = \mathcal{M}_s^{\mathcal{M}}\cup \mathcal{M}_s^{\Ecal}$, together with $\mathcal{M}_{[s_0,+\infty)}^{\EM}$, defined analogously. 

For technical reason, we also introduce the near-light-cone region:
\begin{equation}\label{eq6-05-oct-2025}
\Mcal^{\near}_{\ell,[s_0,s_1]} := \Mcal^{\ME}_{[s_0,s_1]}\cap\{r\leq (1- \ell)^{-1}t\},\quad 0< \ell < 1/2.
\end{equation}

\paragraph{Basic properties of the foliation.}

The above construction introduces a parameterization in the variable $(s,x^a)$. The Jacobian between $(s,x^a)$ and the Cartesian coordinates $(t, x^a)$ reads 
\begin{equation}\label{eq6-10-april-2025}
\aligned
\Phib & :=  \frac{\del(t,x)}{\del(s,x)} 
=
\left(
\begin{array}{cc}
\del_sT & \del_bT
\\
\del_sx^a& \del_{b}x^a
\end{array}
\right)
=
\left(
\begin{array}{cc}
\del_sT &(x^b/r)\del_rT
\\
0 & \mathrm{I}_3
\end{array}
\right),
\\
\Psib :=
\Phib^{-1} & = \frac{\del(s,x)}{\del(t,x)} 
=
\left(
\begin{array}{cc}
(\del_sT)^{-1} & -(x^b/r)\del_rT
\\
0 & \mathrm{I}_3
\end{array}
\right),
\endaligned
\end{equation}
and we also introduce the Jacobian determinant $J := \det\frac{\del(s,x)}{\del(t,x)}$.

With these notation, we write that $\del_rT$ as $\delb_rT$ with former  $T$ regarded as a function of $(s,r)$ and the latter $T$ regarded as a scalar field defined in $\Mcal_{[s_0,s_1]}$ From now on, we apply the latter one. The same principle is applied on $\del_sT$ and $\delb_s T$.

From~\cite[Lemma 3.3]{PLF-YM-PDE}, we recall the following estimate: 
\begin{equation}\label{eq1-21-march-2025}
J = \delb_sT =\del_sT
\begin{cases}
=s/t = \frac{s}{\sqrt{s^2+r^2}}, \qquad &r\leq r^{\Hcal},
\\
\in \left[\frac{s\xi}{\sqrt{s^2+r^2}} + \frac{3}{5}(1- \xi)s, \frac{s\xi}{\sqrt{s^2+r^2}} + 2(1- \xi)s\right], \qquad &r^{\Hcal}\leq r\leq r^{\Ecal},
\\
\in \left[\frac{3s}{5},2s\right], \qquad & r^{\Ecal}\leq r. 
\end{cases}
\end{equation}
The flat volume form of $\Mcal_s$ is defined as 
\be
\zeta^2 := 1- |\delb_rT|^2 = \frac{s^2}{s^2+r^2} + (1- \xi^2(s,r))\frac{r^2}{s^2+r^2}.
\ee
This used to be called the ``energy coefficient'' in our previous work.
We then observe that 
\begin{equation}\label{eq4-21-march-2025}
0\leq 
J \, \zeta^{-1}
\begin{cases}
=1, \qquad & r\leq r^{\Hcal},
\\
\leq \xi + 2(1- \xi)^{1/2}s, \qquad & r^{\Hcal}\leq r\leq r^{\Ecal},
\\
\in [3s/5,2s], \qquad &r^{\Ecal}\leq r.
\end{cases}
\end{equation}

\bse
It can also be checked that, in the transition interval $r^{\Hcal}\leq r\leq r^{\Ecal}$, 
\be
\aligned
\frac{\delb_sT}{\sqrt{1- |\delb_rT|^2}} 
& = \sqrt{\frac{s^2+r^2}{(1- \xi^2)r^2+s^2}}\Big(\frac{s\xi}{\sqrt{s^2+r^2}} + \lambda (1- \xi)s\Big)
\\
& = s\xi\sqrt{\frac{1}{(1- \xi^2)r^2+s^2}} + \lambda(1- \xi) s \sqrt{1 + \frac{\xi^2r^2}{(1- \xi^2)r^2+s^2}}, 
\endaligned
\ee
in which the first term is (universally) bounded by $\xi$. For the second term, we derive the upper bound $(1- \xi)^{1/2}$ by writing  
\be
\aligned
(1- \xi)\sqrt{1+\frac{\xi^2r^2}{(1- \xi)(1+\xi)r^2+s^2}} 
& =  (1- \xi)^{1/2}\sqrt{(1- \xi) + \frac{(1- \xi)\xi^2r^2}{(1- \xi)(1+\xi)r^2 + s^2}}
\\
& \leq (1- \xi)^{1/2}\sqrt{(1- \xi) + 1}\lesssim (1- \xi)^{1/2}.
\endaligned
\ee
\ese

The following estimate was stated and proven in~\cite[Lemma 3.5]{PLF-YM-PDE}.

\begin{lemma}\label{lem1-21-march-2025}
In the merging and Euclidean domains, the
weight $\zeta$ is controlled by the Jacobian, as follows: 
\be
K_1 \zeta^2 s\leq J\leq  K_2 \zeta^2s \text{ in }\mathcal{M}^{\mathcal{ME}},
\ee
where $K_1$, $K_2$ are universal constants.
\end{lemma}

\begin{proof} 
Indeed, ~\cite[Lemma 3.5]{PLF-YM-PDE} guarantees the estimate in the domain $\mathcal{M}^{\mathcal{M}}$. On the other hand, in the Euclidean domain $\mathcal{M}^{\Ecal}$, we only need to recall \eqref{eq1-21-march-2025} and use that $\zeta = 1$.
\end{proof}


}

We will establish the following properties on the cut-off functions.
\begin{lemma}
\label{lem1-05-aout-2025}
There exists a cut-off function $\xi(s,r)$, such that
\begin{equation}\label{eq4-05-aout-2025}
|\del_r\xi(s,r)| + |\del_r\del_r\xi(s,r)|\lesssim_{\delta}\zeta^{2- \delta},\quad \forall 1\geq \delta>0.
\end{equation}
\end{lemma}
\begin{proof}
This is based on a refined property of the cut-off function. In fact we need to construct a smooth cut-off function $\xi$ with $\chi(\rho)\equiv 0$ for $\rho\leq 0$ and $\chi(\rho) \equiv 1$ for $\rho\geq 1$. Then let 
$$
\xi(s,r) := 1- \chi\big(r-r^{\Hcal}(s)\big),\quad r^{\Hcal}(s) = \frac{s^2-1}{2}.
$$
We then hope to have \eqref{eq4-05-aout-2025}. It is direct that
$$
\del_r\xi(s,r) = - \chi'\big(r-r^{\Hcal}(s)\big).
$$
On the other hand, we recall 
\begin{equation}\label{eq9-09-aout-2025}
\zeta^2 = 1 - |\delb_rT|^2 = \frac{s^2}{s^2+r^2} + \frac{(1- \xi^2(s,r))r^2}{s^2+r^2}
= \frac{s^2}{s^2+r^2} + \chi\big(r-r^{\Hcal}(s)\big)\frac{(1+\xi(s,r))r^2}{s^2+r^2}.
\end{equation}
Then for the estimate on $\del_r\xi$, it is thus sufficient to guarantee the following property:
\begin{equation}\label{eq2-05-aout-2025}
\chi'(\rho)\lesssim_{\delta} |\chi(\rho)|^{1- \delta}.
\end{equation}
To this end, we consider
\be
\theta(\rho):= 
\begin{cases}
0,\quad &x\leq 0, 
\\
\exp{\big(-\rho^{-1}(1- \rho)^{-1}\big)},\quad &0<x<1, 
\\
0,\quad &x\geq 1, 
\end{cases}
\ee
and
\begin{equation}\label{eq3-18-aout-2025}
\chi(\rho):= K^{-1}\int_{- \infty}^{\rho}\theta(\lambda)\diff \lambda,\quad K = \int_{\RR}\theta(\lambda)\diff \lambda.
\end{equation}
It is easily checked that $\theta,\chi$ are well defined and smooth. Furthermore, on the interval $(0,1)$ we have 
\be
\theta'(\rho) = \theta(\rho) \frac{1-2\rho}{\rho^2(1- \rho)^2}.
\ee
The key is that the polynomial singularity in the right-hand side near $\rho\rightarrow0^+$ can be absorbed by an exponential decreasing term $|\theta(\rho)|^{\delta}$ with any $\delta>0$. Thus
\begin{equation}
0\leq \theta'(\rho)\lesssim_{\delta}\theta^{1- \delta}(\rho),\quad\text{for}\quad 0< \rho\leq 1/2.
\end{equation}
Integrate the above inequality on $[0,\rho]$, and apply H\"older's inequality ($1/p = 1- \delta$, $1/q = \delta$),
$$
\aligned
0&\leq \chi'(\rho) = \theta(\rho) \lesssim_{\delta}\int_{[0,\rho]}\theta(\lambda)^{1- \delta}\diff\lambda
\leq \big(\int_{[0,\rho]}\theta(\lambda)\diff\lambda\Big)^{1- \delta}
\Big(\int_{[0,\rho]}1\diff \lambda\Big)^{\delta}
\leq \rho^{\delta}\chi^{1- \delta}(\rho) 
\\
&\leq |\chi(\rho)|^{1- \delta}
\endaligned
$$
which guarantees \eqref{eq2-05-aout-2025} in $[0,1/2]$. When $\rho\geq 1/2$, \eqref{eq2-05-aout-2025} is trivial because $0\leq \chi'(\rho)\leq 1$ and $\chi(\rho)\geq \chi(1/2)>0$. The estimate on $\del_r\del_r\xi$ is similar, we omit the details.
\end{proof}


\begin{remark}
We have established the following estimate
\begin{equation}\label{eq4-28-july-2025}
|\del_r\xi(s,r)| + |\del_r^2\xi(s,r)|\lesssim \zeta,
\quad
s^{-1}\lesssim \zeta.
\end{equation}
in \cite[Lemma 7.5]{PLF-YM-PDE}, who is a special case of  \eqref{eq4-05-aout-2025}.
\end{remark}


We also recall the following estimate which we implicitly derived the special case \eqref{eq6b-09-aout-2025} in \cite{PLF-YM-PDE}.

\begin{lemma}
\label{lem1-28-july-2025}
In the region $\Mcal^{\EM}_{[s_0,s_1]}$, one has 
\begin{equation}\label{eq6-09-aout-2025}
|\delb_r\zeta|\lesssim_{\delta} \zeta^{1- \delta}\delb_rT,\quad |\delb_r\delb_rT|\lesssim_{\delta} \zeta^{2- \delta},\quad \forall 1\geq \delta>0.
\end{equation}
In particular
\begin{equation}\label{eq6b-09-aout-2025}
|\delb_r\zeta|\lesssim \delb_rT,\quad |\delb_r\delb_rT|\lesssim \zeta.
\end{equation}
\end{lemma}

\begin{proof} 
\bse
The result is based on \eqref{eq4-05-aout-2025}. 
We differentiate \eqref{eq3-28-july-2025} and obtain
\be
\delb_r^2T = \frac{\xi(s,r) + r\del_r\xi(s,r)}{\sqrt{s^2+r^2}} - \frac{r^2\xi(s,r)}{(s^2+r^2)^{3/2}}, 
\ee
which, in combination with \eqref{eq4-05-aout-2025} leads us to the desired estimate.
\ese
\end{proof}

We also need the following estimates.

\begin{lemma}\label{lem2-05-aout-2025}
In $\Mcal^{\ME}_{[s_0,s_1]}$,
\begin{equation}
|\del_t\delb_rT| = |J^{-1}\delb_rJ| \lesssim_{\delta} \zeta^{- \delta},\quad \Mcal^{\ME}_{[s_0,s_1]}.
\end{equation}
\end{lemma}

\begin{proof}
\bse
Observe that
\be
\del_t\delb_rT = J^{-1}\delb_s\delb_rT = J^{-1}\delb_r\delb_sT = J^{-1}\delb_rJ.
\ee
Then we recall \cite[(A.1)]{PLF-YM-PDE} (or recompute directly) and obtain
\begin{equation}\label{eq1-05-aout-2025}
\del_t\delb_rT 
= -J^{-1}\Big(\frac{sr\del_r\xi(s,r)}{(s^2+r^2)^{1/2}} + \frac{sr\xi(s,r)}{(s^2+r^2)^{3/2}}\Big)
\end{equation}
Now we apply \eqref{eq1-21-march-2025}:
\begin{equation}\label{eq3-05-aout-2025}
0\leq J^{-1}\frac{-sr\del_r\xi(s/r)}{\sqrt{s^2+r^2}}
\lesssim \frac{s|\del_r\xi(s,r)|}{(1- \xi(s,r))s} = \frac{\del_r\xi(s,r)}{1- \xi(s,r)},
\end{equation}
and
\begin{equation}
0\leq J^{-1}\frac{sr\xi(s,r)}{(s^2+r^2)^{3/2}}\lesssim1.
\end{equation}
Then we apply Lemma~\ref{lem1-05-aout-2025} on \eqref{eq3-05-aout-2025}, more precisely, \eqref{eq4-05-aout-2025}, 
\be
0\leq J^{-1}\frac{-sr\del_r\xi(s/r)}{\sqrt{s^2+r^2}}
\lesssim_{\delta} \zeta^{-2}\zeta^{2- \delta}\leq \zeta^{- \delta}.
\ee
Here we have applied the fact $0\leq1- \zeta\leq \zeta^2$, and this leads us to the desired estimate.
\ese
\end{proof}


\paragraph{Metric.} 

{

Throughout this Monograph, we focus on almost Euclidean metrics, i.e., 
\begin{equation}
g_{\alpha\beta} = g(\del_{\alpha,\del_{\beta}}) = \eta_{\alpha\beta} + H_{\alpha\beta}, 
\end{equation}
where  $\big(\eta_{\alpha\beta}\big) = \mathrm{diag}(-1,1,1,1)$ and $H_{\alpha\beta}$ enjoy certain smallness properties to be specified latter on. We denote by $\big(g^{\alpha\beta}\big)$ the inverse of $\big(g_{\alpha\beta}\big)$, and also set
\begin{equation}
g^{\alpha\beta} =: \eta^{\alpha\beta} + H^{\alpha\beta}
\end{equation}

\begin{lemma}\label{lem1-16-june-2025}
When $|H|:=\max_{\alpha\beta}|H_{\alpha\beta}|\leq \eps_s$ for some sufficiently small $\eps_s$, one has 
\begin{equation}
\big(g^{\alpha\beta}\big)^{-1} := \big(g_{\alpha\beta}\big)^{-1} = \eta^{\alpha\beta} - \eta^{\alpha\mu}H_{\mu\nu}\eta^{\nu\beta} + O(|H|^2),
\end{equation}
and, more precisely,
\begin{equation}
H^{\alpha\beta} := g^{\alpha\beta} - \eta^{\alpha\beta} = O(|H|^2) + 
\begin{cases}
-H_{00},\quad &(\alpha,\beta) = (0,0);
\\
H_{0a},\quad &(\alpha,\beta) = (0,a),\quad a=1,2,3;
\\
-H_{ab},\quad &(\alpha,\beta) = (a,b),\quad a,b = 1,2,3.
\end{cases}
\end{equation}
\end{lemma}

\begin{proof} The result follows from the identity
\be
(A+B)^{-1} = \sum_{k=0}^{\infty}\big(-A^{-1}B\big)^kA^{-1}, 
\ee
in which $A, B$ are $n\times n$ complex matrices, provided the right-hand side converges. 
\end{proof}


\paragraph{Frame adapted to the foliation.}

Let us introduce the vector fields
\begin{equation}
\delb_s := \delb_0 = \delb_sT\del_t = J\del_t,
\qquad 
\delb_a:= \frac{\del T}{\del r} \frac{x^a}{r} \del_t + \del_a, 
\end{equation}
so that $\{\delb_0,\delb_a\}$ forms a frame.
The dual frame is written as
$$
\diff s = J^{-1}\big(\diff t - (x^a/r)\delb_rT \diff x^a\big),\quad
\diff x^a = \diff x^a.
$$
Let $S$ be any two-tensor  and consider its components 
\be
\overline{S}_{ab} = S(\delb_a,\delb_b), \qquad 
\overline{S}_{0a} = S(\delb_0,\delb _a), \qquad 
\overline{S}_{00} = S(\delb_0,\delb_0). 
\ee
With the notation $H_{\alpha\beta} := g_{\alpha\beta} - \eta_{\alpha\beta}$, we find 
\begin{equation}\label{eq1-18-march-2025}
\aligned
& \sigmab_{ab} = \gb_{ab} := (\delb_a,\delb_b)_g = \overline{\eta}_{ab} + \Hb_{ab} = \delta_{ab} - \Big(\frac{\del T}{\del r} \Big)^2\frac{x^ax^b}{r^2} +\Hb_{ab}.
\\
& \gb_{0a} = - \frac{x^a}{r} J \delb_rT  + \Hb_{0a}, \qquad \gb_{00} = - J^2 + \Hb_{00}.
\endaligned
\end{equation}
\begin{equation}\label{eq10-03-aout-2025}
\aligned
\gb_{00} & =J^2g_{00} = -J^2+J^2H_{00},
\\
\gb_{0a} & =J\frac{x^a}{r}\del_rTg_{00} + Jg_{a0} = -J\frac{x^a}{r}\del_rT + J\frac{x^a}{r}\del_rTH_{00} + JH_{a0},
\\
\sigmab_{ab} =: \gb_{ab} & =\frac{x^ax^b}{r^2}(\delb_rT)^2g_{00} 
+ \frac{x^a}{r}\delb_rTg_{b0} + \frac{x^b}{r}\delb_rTg_{a0}
+ g_{ab}
\\
& = \delta_{ab} - \frac{x^ax^b}{r^2}\big(\delb_rT\big)^2
+ \frac{x^ax^b}{r^2}(\delb_rT)^2H_{00} 
+ \frac{x^a}{r}\delb_rTH_{b0} + \frac{x^b}{r}\delb_rTH_{a0}
+ H_{ab}.
\endaligned
\end{equation}
More precisely, we have 
\begin{equation}
\aligned
& \Hb_{00} = J^2H_{00}, \quad \Hb_{0a} = J(x^a/r)\delb_rT H_{00} + JH_{0a}, 
\\
& \Hb_{ab} = (x^ax^b/r^2)(\delb_rT)^2H_{00} + (x^a/r)\delb_rH_{b0} + (x^b/r)\delb_rTH_{a0} + H_{ab}. 
\endaligned
\end{equation}
This leads us to, recalling that $|H| = \max_{\alpha,\beta}\{|H_{\alpha\beta}|\}$
\begin{equation}\label{eq3-14-july-2025}
|\Hb_{00}|\lesssim J^2|H|,\quad |\Hb_{0a}|\lesssim J|H|,\quad |\Hb_{ab}|\lesssim |H|.
\end{equation}
In combination with \eqref{eq2-14-june-2025}, this leads us to the estimate, provided that $|H|\lesssim 1$,
\begin{equation}\label{eq4-14-july-2025}
|\gb_{00}|\lesssim J^2,\quad |\gb_{0a}|\lesssim J,\quad |\gb_{ab}|\lesssim 1.
\end{equation}
We write $\big(\gb^{\alpha\beta}\big)$ for the inverse of $\big(\gb_{\alpha\beta}\big)$, with
\begin{equation}
\gb^{\alpha\beta} = \bar{\eta}^{\alpha\beta} + \Hb^{\alpha\beta}, 
\end{equation}
and we also set $\big(\sigmab^{ab}\big) = \big(\sigmab_{ab}\big)^{-1}$.
\begin{equation}\label{eq11-03-aout-2025}
\aligned
\gb^{00} & =J^{-2}\Big(g^{00} - 2\delb_rT\frac{x^a}{r}g^{a0} + (\delb_rT)^2\frac{x^ax^b}{r^2}g^{ab}\Big)
\\
& =-J^{-2}\zeta^2 + J^{-2}\Big(H^{00} - 2\delb_rT\frac{x^a}{r}H^{a0} + (\delb_rT)^2\frac{x^ax^b}{r^2}H^{ab}\Big),
\\
\gb^{0a} & =J^{-1}\Big(g^{0a} - \delb_rT\frac{x^b}{r}g^{ab}\Big)
=-J^{-1}\delb_rT\frac{x^a}{r} + J^{-1}H^{0a} - J^{-1}\delb_rT\frac{x^b}{r}H^{ab},
\\
\gb^{ab} & =g^{ab}.
\endaligned
\end{equation}

By definition of the Jacobi matrices  
$\Phib_{\alpha}^{\beta} = (\Phib)_{\beta\alpha}$ and $\Psib_{\alpha}^{\beta} = (\Psib)_{\beta\alpha}$ defined in \eqref{eq6-10-april-2025}, we have 
\begin{equation}
(\delb_0,\delb_1,\delb_2,\delb_3) = (\del_0,\del_1,\del_2,\del_3)\Phib,
\qquad
(\del_0,\del_1,\del_2,\del_3) = (\delb_0,\delb_1,\delb_2,\delb_3)\Psib, 
\end{equation}
therefore 
\be
\delb_{\alpha} = \Phib_{\alpha}^{\beta}\del_{\beta}, \qquad \del_{\alpha} = \Psib_{\alpha}^{\beta}\delb_{\beta}, 
\ee
together with 
\begin{subequations}\label{eq1-13-june-2025}
\begin{equation}
\Big(\overline{\Phi}_{\alpha}^{\beta}\Big)_{\beta\alpha}
=\left(\begin{array}{cccc}
J & \frac{x^1}{r}\delb_rT & \frac{x^2}{r}\delb_rT & \frac{x^3}{r}\delb_rT
\\
0 &1 &0 &0
\\
0 &0 &1 &0
\\
0 &0 &0 &1
\end{array}\right),
\end{equation}
\begin{equation}
\Big(\overline{\Psi}_{\alpha}^{\beta}\Big)_{\beta\alpha}
=\left(\begin{array}{cccc}
J^{-1} & \quad - \frac{x^1}{r}J^{-1}\delb_rT & \quad - \frac{x^2}{r}J^{-1}\delb_rT & \quad - \frac{x^3}{r}J^{-1}\delb_rT
\\
0 &1 &0 &0
\\
0 &0 &1 &0
\\
0 &0 &0 &1
\end{array}\right).
\end{equation}
\end{subequations}


\paragraph{Null frame for the merging-Euclidean domain.}
Next, we also introduce the so-called \emph{null frame}
\begin{equation}
\delN_0 := \del_t, \qquad \delN_a := \frac{x^a}{r}\del_t + \del_a
\end{equation}
defined out of a fixed cone $\{r>\lambda t\}$ with $\lambda >0$.
The components of any two-tensor $S$ read
\be
S^{\Ncal}_{ab} = S(\delN_a,\delN_b), \qquad S^{\Ncal}_{a0} = S(\delN_a,\delN_0), \qquad 
S^{\Ncal}_{00} = S(\delN_0,\delN_0).
\ee
With the notation $H_{\alpha\beta} := g_{\alpha\beta} - \eta_{\alpha\beta}$, we have 
\begin{equation}\label{eq2-30-march-2025}
\aligned
&g^{\Ncal}_{ab} := (\delN_a,\delN_b)_g = \eta^{\Ncal}_{ab} + H^{\Ncal}_{ab} = \delta_{ab} - \frac{x^ax^b}{r^2} + H^{\Ncal}_{ab}.
\\
&g^{\Ncal}_{0a} = - \frac{x^a}{r} + H^{\Ncal}_{0a}, \qquad g^{\Ncal}_{00} = - 1 + H^{\Ncal}_{00}.
\endaligned
\end{equation}
\bse
The transition matrices between $\{\delN_{\alpha}\}$ and $\{\del_{\alpha}\}$, namely 
\be
\delN_{\alpha} = {\Phi^{\Ncal}}_{\alpha}^{\beta}\del_{\beta},
\qquad \del_{\alpha} = {\Psi^{\Ncal}}_{\alpha}^{\beta}\delN_{\beta},
\ee
are given by 
\be
\Big({\Phi^{\Ncal}}_{\alpha}^{\beta} \Big)_{\beta\alpha}
=
\left(
\begin{array}{cccc}
1 &(x^1/r) &(x^2/r) &(x^3/r)
\\
0 &1 &0 &0
\\
0 &0 &1 &0
\\
0 &0 &0 &1
\end{array}
\right)
\ee
and
\be
\Big({\Psi^{\Ncal}}_{\alpha}^{\beta} \Big)_{\beta\alpha}{\xout.}
=
\left(
\begin{array}{cccc}
1 & -(x^1/r) & -(x^2/r) & -(x^3/r)
\\
0 &1 &0 &0
\\
0 &0 &1 &0
\\
0 &0 &0 &1
\end{array}
\right).
\ee
We also observe that
\begin{equation}\label{eq1-14-july-2025}
\aligned
S^{\Ncal}_{\alpha\beta} & =  S_{\mu\nu}{\Phi^{\Ncal}}_{\alpha}^{\mu}{\Phi^{\Ncal}}_{\beta}^{\nu}
= S_{00} + 2(x^a/r)S_{a0} + (x^ax^b/r^2)S_{ab},
\quad
\\
S^{\Ncal \alpha\beta} & =  S^{\mu\nu}{\Psi^{\Ncal}}_{\mu}^{\alpha}{\Psi^{\Ncal}}_{\nu}^{\beta}
= S^{00} - 2(x^a/r)S^{a0} + (x^ax^b/r^2)S^{ab}.
\endaligned
\end{equation}
\ese
An important feature of this null frame (as well as the semi-hyperboloidal frame to be recalled later) is that the coefficients of the transition matrices are homogeneous of degree zero (out of a fixed cone). Thus
\begin{equation}
|S^{\Ncal}_{\alpha\beta}|\lesssim |S_{\alpha}|,\quad
|S^{\Ncal\alpha\beta}|\lesssim |S^{\alpha\beta}|.
\end{equation}
\paragraph{Semi-hyperboloidal frame in the hyperboloidal region.}

We recall that in $\Mcal_{[s_0,s_1]}\cap\{r\leq 3t\}$, the semi-hyperboloidal frame $\{\delu_{\alpha}\}$ is defined by
\begin{equation}\label{eq7-29-july-2025}
\delu_0 := \del_t, \qquad \delu_a: = \delb_a = (x^a/t)\del_t+\del_a.
\end{equation}
and the associated co-frame reads
\be
\thetau^0 = \diff t - (x^a/t)\diff x^a, \qquad \thetau^a = \diff x^a.
\ee
\bse
The transition matrices 
\begin{equation}\label{eq2-13-june-2025}
\Big(\Phiu_{\alpha}^{\beta}\Big)_{\beta\alpha} =
\left(\begin{array}{cccc}
1 & \frac{x^1}{t} & \frac{x^2}{t} & \frac{x^3}{t}
\\
0 &1 &0 &0
\\
0 &0 &1 &0
\\
0 &0 &0 &1
\end{array}\right),
\quad 
\Big(\Psiu_{\alpha}^{\beta}\Big)_{\beta\alpha} =
\Phiu^{-1} =
\left(\begin{array}{cccc}
1 & \quad - \frac{x^1}{t} & \quad - \frac{x^2}{t} & \quad - \frac{x^3}{t}
\\
0 &1 &0 &0
\\
0 &0 &1 &0
\\
0 &0 &0 &1
\end{array}\right)
\end{equation}
allow us to write 
\be
\delu_{\alpha} = \Phiu_{\alpha}^{\beta}\del_{\beta},
\qquad 
\del_{\alpha} = \Psiu_{\alpha}^{\beta}\delu_{\beta}.
\ee
\ese
A two-tensor $S$ can be written in the semi-hyperboloidal frame, that is, 
\be
\underline{S}_{\alpha\beta} = S(\delu_{\alpha},\delu_{\beta}), 
\qquad
\underline{S}^{\alpha\beta} = S(\theta^{\alpha},\theta^{\beta}).
\ee
We can also derive certain relations between the adapted frame and the semi-hyperboloidal frame. We compare \eqref{eq1-13-june-2025} and \eqref{eq2-13-june-2025} and obtain 
\be
\Psib_{\alpha}^0 = (s/t)^{-1}\Psiu_{\alpha}^0,
\qquad 
\Psib_{\alpha}^a = \Psiu_{\alpha}^a.
\ee
Thus we arrive at, in $\Mcal^{\Hcal}_{[s_0,s_1]}$
\begin{equation}\label{eq3-13-june-2025}
\overline{S}^{ab} = \underline{S}^{ab},
\qquad 
\overline{S}^{a0} = \overline{S}^{0a} = (s/t)^{-1}\underline{S}^{a0},
\qquad
\overline{S}^{00} = (s/t)^{-2}\underline{S}^{00}, 
\end{equation}
and, on the other hand, 
$$
\Phib_0^{\alpha} = (s/t)\Phiu_0^{\alpha},
\qquad
\Phib_a^{\alpha} = \Phiu_a^{\alpha}, 
$$
therefore
\begin{equation}\label{eq4-13-june-2025}
\overline{S}_{ab} = \underline{S}_{ab},
\qquad 
\overline{S}_{a0} = \overline{S}_{0a} = (s/t)\underline{S}_{a0},
\qquad
\overline{S}_{00} = (s/t)^2\underline{S}_{00} = (s/t)^2S_{00}.
\end{equation}
As mentioned for the null frame, for this frame, the elements of the transition matrices are \emph{homogeneous of degree zero}, and are uniformly bounded in $\Mcal^{\Hcal}_{[s_0,s_1]}$ ---particularly near the light-cone $\{r=t-1\}$. Thus we obtain with $|I| = p-k$ and $|J| = k$
\begin{equation}\label{eq1-18-june-2025}
|\del^IL^J\underline{S}_{\alpha\beta}|\lesssim \max_{\alpha,\beta}|S_{\alpha\beta}|_{p,k},
\qquad
|\del^IL^J\underline{S}^{\alpha\beta}|\lesssim \max_{\alpha,\beta}|S^{\alpha\beta}|_{p,k}. 
\end{equation}


\paragraph{Tangent-orthogonal frame.}

Finally, we will also use a \emph{metric-dependent frame}, consisting of $\{\delb_a\}$ and the future-oriented, unit normal  $\vec{n}$ to $\Mcal_s$. Namely, we define 
\begin{equation}\label{eq4-12-june-2025}
\vch_0 = \vec{n}, \qquad \vch_a = \delb_a
\end{equation}
and
\begin{equation}\label{eq5-12-june-2025}
\vec{n} = \lapsb^{-1}(\delb_0 - \betab^b\delb_b)
\end{equation}
Then, given that $\big(\sigmab^{ab}\big)$ (the inverse of $\big(\gb_{ab}\big)$), and the functions
\begin{equation}\label{eq6-12-june-2025}
\betab^b = \gb_{0a}\sigmab^{ab}, 
\quad
- \lapsb^2 = \gb_{00} - \gb_{0a}\sigmab^{ab}\gb_{b0} = \gb_{00} - \gb_{0a}\betab^a, 
\end{equation}
are called the {\bf shift function} and the {\bf lapse function}, respectively. The components of any two-tensor $S$ read
\be
\check{S}_{\alpha\beta} = S(\vch_{\alpha},\vch_{\beta}).
\ee
Furthermore, the dual frame of this tangent-orthogonal frame reads 
\begin{equation}\label{eq1-24-june-2025}
\check{\omega}^0 = \lapsb \diff s, \qquad \check{\omega}^a = \diff x^a + \betab^a\diff s.
\end{equation}

}


\subsection{ Estimates on the volume of the Euclidean-hyperboloidal slices}
\label{section===72}

{ 

\paragraph{Expressions in flat geometry.}

It is convenient to list here several expressions in the flat case.

\begin{lemma}[Expressions associated with the adapted frame. Flat case]
\label{lem1-13-june-2025}
\bse
Suppose that $g=\eta$. One has 
\begin{equation}\label{eq3-18-mai-2025}
\aligned
& \overline{\eta}_{00} = -J^2 = -|\delb_sT|^2, \quad 
&& \overline{\eta}_{a0} = -(x^a/r)J\delb_rT,\
\\
& \overline{\eta}_{ab} = \delta_{ab} - |\delb_rT|^2(x^ax^b/r^2),
\\
& \overline{\eta}^{00} = - \zeta^2J^{-2}, \qquad 
&& \overline{\eta}^{0a} = -J^{-1}\delb_rT(x^a/r),
\quad 
\overline{\eta}^{ab} = \delta^{ab}.
\endaligned
\end{equation}
Furthermore, the unit normal vector is expressed as 
\begin{equation}
\aligned
\vec{n}_{\eta} & =  \zeta^{-1} \del_t + \zeta^{-1}\delb_rT(x^a/r)\del_a 
\\
& =  \zeta\del_t + \zeta^{-1}\delb_rT(x^a/r)\delb_a 
= J^{-1}\zeta\delb_s + \zeta^{-1}\delb_rT(x^a/r)\delb_a.
\endaligned
\end{equation}
Concerning the geometry of the slices $\Mcal_s$, the restricted metric $\sigmab_{\eta}$ is written as
\begin{equation}\label{eq1-17-mai-2025}
\sigmab_{\eta,ab} = \overline{\eta}_{ab} = \delta_{ab} - |\delb_rT|^2(x^ax^b/r^2),
\qquad
\sigmab_{\eta}^{ab} = \zeta^{-2}\frac{x^ax^b}{r^2}\big(\delb_rT\big)^2 + \delta^{ab} 
\end{equation}
and, consequently, 
\begin{equation}\label{eq7-13-june-2025}
|\sigmab_{\eta}^{ab}|\lesssim \zeta^{-2}.
\end{equation}
Moreover, the shift vector reads
\begin{equation}\label{eq1-27-june-2025}
\betab_{\eta}^a = -(x^b/r)\delb_rT J\zeta^{-2}.
\end{equation}
\ese
while the second fundamental form of $\Mcal_s$ takes the form 
\begin{equation}\label{eq1-03-aout-2025}
\Pi_{\eta,ab} := -(\vec{n}_{\eta},\nabla_{\delb_a}\delb_b)_g 
= \zeta^{-1}\Big(r^{-1}\delb_rT\big(\delta_{ab} - (x^ax^b/r^2)\big) + \del_r^2T(x^ax^b/r^2)\Big). 
\end{equation}
\end{lemma}

\begin{lemma}[Expressions associated with the semi-hyperboloidal frame in the hyperboloidal domain. Flat case]
\label{lem2-13-june-2025}
\bse
Suppose that $g=\eta$ and consider the hyperboloidal domain. One has 
\begin{equation}
J = \zeta = (s/t), \quad \delb_rT = (r/t), 
\end{equation}
\begin{equation}
\underline{\eta}_{00} = -1,
\qquad 
\underline{\eta}_{a0} = \underline{\eta}_{0a} = - \frac{x^a}{t},
\qquad
\underline{\eta}_{ab} = - \frac{x^ax^b}{t^2} + \delta_{ab}, 
\end{equation}
\begin{equation}
\underline{\eta}^{00} = -(s/t)^2,
\qquad 
\underline{\eta}^{a0} = \underline{\eta}^{0a} = - \frac{x^a}{t},
\qquad
\underline{\eta}^{ab} = \delta^{ab}. 
\end{equation}
\ese
\bse
Concerning the geometry of the slices $\Mcal_s$, one has 
\begin{equation}
\vec{n}_{\eta} = (t/s)\del_t + (x^a/s)\del_a = (s/t)\del_t + (x^a/s)\delu_a, 
\end{equation}
\begin{equation}
\nabla_{\delu_a}\delu_b = t^{-1}\underline{\eta}_{ab}\del_t,
\end{equation}
\begin{equation}
\Pi_{\eta,ab} = -(\vec{n}_{\eta},\nabla_{\delb_a}\delb_b)_g = s^{-1}\underline{\eta}_{ab}.
\end{equation}
\ese
\end{lemma}


\paragraph{Volume form of Euclidean-hyperboloidal slices.}

Before we can analyze any given system of PDEs within the Euclidean-hyperboloidal framework, a preliminary and important issue is guaranteeing that $\Mcal_s$ are spacelike and enjoy a certain non-degeneracy property, measured in terms of the volume form $\det{(\sigmab_{ab})}\diff x$. When $g$ coincides with $\eta$, a direct calculation leads us to
\begin{equation}\label{eq2-19-march-2025}
\det(\overline{\eta}_{ab}) = 1- \Big(\frac{\del T}{\del r} \Big)^2 = \zeta^2
\qquad \text{ when } g=\eta, 
\end{equation}
so that 
\begin{equation}
\sqrt{\det \overline{\eta}_{ab}}= \zeta
=\begin{cases}
(s/t), \qquad & \mathcal{M}^{\Hcal},
\\
1, \qquad & \mathcal{M}^{\Ecal}, 
\end{cases}
\qquad \text{ when } g=\eta. 
\end{equation}
A natural way to guarantee this property for the curved metric is to proceed as follows. 

{
\begin{lemma}
\label{lem1a-02-march-2025}
There exists a universal constant $\eps_s>0$, such that if
\begin{equation} \label{eq1-05-03-2025}
|\Hb| := \max_{a,b}\{|\Hb_{ab}|\}\leq \eps_s\zeta^2
\qquad \text{ in the domain } \mathcal{M}_{[s_0,s_1]}, 
\end{equation} 
then one has 
\begin{equation}
\big|\sqrt{\det \sigmab} - \zeta\big|\lesssim \zeta^{-1}|\Hb_{ab}|\leq \frac{1}{2}\zeta.
\end{equation}
\end{lemma}

\begin{proof} We easily compute 
$
\det\sigmab =  \det\sigmab_{\eta}
+ T_1(\Hb) 
+ T_2(\Hb) 
+ T_3(\Hb),
$
in which $T_1(\Hb)$ are linear terms, $T_2(\Hb)$ are quadratic terms, and $T_3(\Hb)$ are cubic terms.  More precisely, the coefficients of these terms are bounded. Then we find 
$$
\big\vert\det\sigmab - \det\sigmab_{\eta} \big\vert \lesssim 
C \, \big\vert\Hb_{ab} \big\vert\leq \frac{1}{2} \zeta, 
$$
provided that $\eps_s$ sufficiently small. This implies that $\det\sigmab>0$ and
$$\hfill 
\big|\sqrt{\det\sigmab} - \sqrt{\det\sigmab_{\eta}}\big|
= \Big|\frac{\det\sigmab - \det\sigmab_{\eta}}{\sqrt{\det\sigmab} + \sqrt{\det\sigmab_{\eta}}}\Big|
\leq \frac{C\vert\Hb_{ab} \vert}{\sqrt{\det\sigmab_{\eta}}}. \hfill \qedhere
$$
\end{proof}
}


While the above estimate would in principe be sufficient for the application in the present Monograph (but would make the proofs considerably more complicated), we prefer to propose a more precise estimate which is wider interest, as follows. 

\begin{proposition}\label{prop1-14-june-2025}
Let $g_{\alpha\beta} = \eta_{\alpha\beta} + H_{\alpha\beta}$ be a metric defined in (a subset of) $\Mcal_{[s_0,s_1]}$. Assume the following uniform spacelike condition: 
\begin{subequations}\label{eq-USA-condition}
\begin{equation}\label{eq1-14-june-2025}
H^{\Ncal 00} := g(\diff t- \diff r, \diff t- \diff r)<0
\qquad \text{in }\{3t/4\leq r\leq r^{\Ecal}(s)\}\cap \Mcal_{[s_0,s_1]}
\end{equation}
with, for a sufficiently small $\eps_s>0$, 
\begin{equation}\label{eq2-14-june-2025}
|H| := \max_{\alpha,\beta}|H|\leq \eps_s\zeta, \quad \text{in }\Mcal_{[s_0,s_1]}. 
\end{equation}
\end{subequations}
Then the $3\times3$ matrix $\gb_{ab}$ is invertible and satisfies
\begin{equation}\label{eq5-15-june-2025}
C_0(\zeta^2+|H^{\Ncal00}|)\geq\det(\gb_{ab})\geq c_0\big(\zeta^{2} + |H^{\Ncal 00}|\big)
\end{equation}
with universal constants $C_0>c_0>0$.
\end{proposition}

A detailed proof of this result is provided in Section~\ref{section=N22}. Here, we content with a heuristic argument. Let $X,Y$ be a (locally defined) orthonormal frame of a Euclidean sphere (with respect to the induced metric on the sphere form the Minkowski metric), which is a sub-manifold of a slice $\Mcal_s$. Let $\delb_r = (x^a/r)\delb_a$. Then $\{\delb_r,X,Y\}$ forms a (locally defined) orthogonal frame of $\Mcal_s$. Here remark that $|\delb_r|_{\eta} = \zeta^2$. Thus to guarantee that $g(\delb_r,\delb_r)>0$, one is only permitted to perturb $g$ within a size of $\eps_s\zeta^2$ (which is Lemma~\ref{lem1a-02-march-2025}). However, the light-bending condition \eqref{eq1-14-june-2025} in fact says that, modulo higher order terms, $g(\delb_r,\delb_r)>0$. We can thus relax considerably the demand on the decreasing rate of the perturbation (comparing \eqref{eq2-14-june-2025} and \eqref{eq1-05-03-2025}).

Based on Proposition~\ref{prop1-14-june-2025}, we can establish  several basic estimates on various geometric objects associated with the Euclidean-hyperboloidal slices, stated now in Claims~\ref{lem1-02-march-2025} and \ref{cla1-16-june-2025}, as well as in Claim~\ref{cor1-16-june-2025}, below. 

Our first observation provides us with a global estimate on the volume form of curved slices $\Mcal_s$.
}

\begin{claim}\label{lem1-02-march-2025}
Assume that the uniform spacelike condition \eqref{eq-USA-condition} holds for a sufficiently small $\eps_s$. Then, for some constants $K_1>K_0>0$, one has 
\begin{equation}\label{eq5-16-june-2025}
K_0 \big(\zeta + |H^{\Ncal00}|^{1/2}\big) \leq\sqrt{\det(\gb_{ab})}\leq K_1\big(\zeta + |H^{\Ncal00}|^{1/2}\big).
\end{equation}
\end{claim}

\begin{proof}
We only need to apply \eqref{eq5-15-june-2025} in the region $\{3t/4\leq r\leq r^{\Ecal}(s)\}$, and Lemma~\ref{lem1a-02-march-2025} for the remaining region, because there $\zeta\geq \sqrt{7}/4$ and $\big|H^{\Ncal00}\big|\lesssim \eps_s$.
\end{proof}

For the simplicity of expression, we denote by 
\begin{equation}\label{eq9-08-aout-2025}
\zetab := \sqrt{\zeta^2 + |H^{\Ncal00}|}.
\end{equation}
and it is clear that
\begin{equation}
\zeta\lesssim \zetab.
\end{equation}

\begin{claim}\label{cla1-16-june-2025}
Let  $\big(\sigmab^{ab}\big)$ be the inverse of $\big(\gb_{ab}\big)$. Under the uniform spacelike condition \eqref{eq-USA-condition} for a sufficiently small $\eps_s$, one has 
\begin{equation}\label{eq2-16-june-2025}
|\sigmab^{ab}|\lesssim \zetab^{-2}
\end{equation}
Furthermore, one has 
\begin{equation}\label{eq4-16-june-2025}
|\sigmab^{ab}\delb_b|\lesssim \zetab^{-1}.
\end{equation}
\end{claim}

\begin{proof}
We recall the relation
$$
(\sigmab^{ab})_{ab} = \det((\gb)_{ab})^{-1}\mathrm{adj}((\gb)_{ab}).
$$
We only need to point out that
$$
\big|\mathrm{adj}((\gb)_{ab}) - \mathrm{adj}((\overline{\eta})_{ab})\big|\lesssim |H|.
$$
Then, in view of \eqref{eq5-16-june-2025}, \eqref{eq2-16-june-2025} follow. For \eqref{eq4-16-june-2025}, we note that
$$
g(\sigmab^{ab}\delb_b,\sigmab^{ac}\delb_c) = \sigmab^{ab}\sigmab^{ac}\gb_{bc} = \sigmab^{aa}, 
$$
which is bounded by $\zetab^{-2}$.
\end{proof}


\subsection{Estimates on lapse and shift functions}

{

We relate the volume form and the lapse, as follows. 

\begin{lemma}[Decomposition of the volume form]\label{lem1-25-mai-2025} 
Suppose that $\big(\gb_{ab}\big)$ is invertible. Then, with the induced orientation on $\Mcal_s$, one has 
\begin{equation}\label{eq1-24-mai-2025}
\mathrm{Vol}_g = \lapsb \, \mathrm{Vol}_{\sigmab}\wedge\diff s
\end{equation}
\end{lemma}

\begin{proof}
\bse
We write  
\be
\mathrm{Vol}_g = \sqrt{|\det{\gb}|} \, {\diff x}\wedge\diff s,
\ee
with 
$\gb = \big(\gb_{\alpha\beta}\big)_{\alpha,\beta}$ and ${\diff x}:=\diff x^1\wedge\diff x^2 \wedge \diff x^3$. 
With $V = (\gb_{10}, \gb_{20}, \gb_{30})^{\mathrm{T}}$, we have 
\be
W = - \gb^{-1} V = - \big(\sigmab^{1b}\gb_{b0},\sigmab^{2b}\gb_{b0},\sigmab^{3b}\gb_{b0}\big)^{\mathrm{T}}
\ee
with 
\begin{equation}\label{eq2-24-mai-2025}
\left(
\begin{array}{cc}
1 &W^{\mathrm{T}}
\\
0 &1
\end{array}
\right)
\left(
\begin{array}{cc}
\gb_{00} &V^{\mathrm{T}}
\\
V & \big(\gb_{ab}\big)
\end{array}
\right)
\left(
\begin{array}{cc}
1 &0 
\\
W&1
\end{array}
\right)
= 
\left(
\begin{array}{cc}
\gb_{00} + W^{\mathrm{T}} V &0
\\
0 & \big(\gb_{ab}\big)
\end{array}
\right).
\end{equation}
We then obtain 
\begin{equation}\label{eq2-18-oct-2025(l)}
\gb_{00} + W^{\mathrm{T}}V = \gb_{00} - \gb_{0a}\sigmab^{ab}\gb_{b0} = - \lapsb^2.
\end{equation}
Taking the determinant of  \eqref{eq2-24-mai-2025}, we find
\begin{equation}
\lapsb^2 \det \sigmab = - \det\gb, 
\end{equation}
which leads us to the identity 
\begin{equation}
\sqrt{|\det \gb|}{\diff x}\wedge\diff s = \lapsb  \sqrt{\det \sigmab}\,{\diff x}\wedge\diff s
\end{equation}
and we reach \eqref{eq1-24-mai-2025}.
\ese
\end{proof}


As a by-product of our analysis, we now state an identity that we will use later on. 

\begin{lemma}\label{lem1-12-june-2025}
Provided $(\gb_{ab})$ is invertible, one has 
\begin{equation}\label{eq6-03-aout-2025}
\gb^{00} = - \lapsb^{-2},
\end{equation}
\begin{equation}
\sigmab^{ab} = \lapsb^2\gb^{0a}\gb^{0b} + \gb^{ab}.
\end{equation}
\end{lemma}

\begin{proof}
\bse
The result follows from \eqref{eq2-24-mai-2025}. Indeed, by taking the inverse on both sides of \eqref{eq2-24-mai-2025}, we obtain 
\begin{equation}
\left(
\begin{array}{cc}
1 &0
\\
-W &1
\end{array}
\right)
\left(
\begin{array}{cc}
\gb^{00} &U^\mathrm{T}
\\
U &(\gb^{ab})
\end{array}
\right)
\left(
\begin{array}{cc}
1 & \quad - W^{\mathrm{T}}
\\
0 &1
\end{array}
\right)
=
\left(
\begin{array}{cc}
- \lapsb^{-2} &0
\\
0 & \big(\sigmab^{ab}\big)
\end{array}
\right)
\end{equation}
with $U^{\mathrm{T}} = (\gb^{10}, \gb^{20}, \gb^{30})^{\mathrm{T}}$. This leads us to
\be
\gb^{00} = - \lapsb^{-2},
\quad
U - \gb^{00}W = 0,
\quad 
-(U- \gb^{00}W)W^{\mathrm{T}} - WU^{\mathrm{T}} + \gb^{ab} = \sigmab^{ab},
\ee
and the desired result is reached. 
\ese
\end{proof}


We are now in a position to estimate the lapse function (thanks to Proposition~\ref{prop1-14-june-2025}).

\begin{claim}\label{cor1-16-june-2025}
Suppose that the uniform spacelike condition \eqref{eq-USA-condition} holds for a sufficiently small  $\eps_s$. Then, one has 
\begin{equation}\label{eq1a-16-june-2025}
K_0\lapsb_{\eta}\leq \lapsb \, \leq K_1\lapsb_{\eta},\quad \text{ in }  \{r\leq 3t/4\}\cap \{r\geq r^{\Ecal}(s)\},
\end{equation}
and
\begin{equation}\label{eq2-06-aout-2025}
K_0\lapsb_{\eta}\leq (\zetab/\zeta)\lapsb\leq K_1\lapsb_{\eta},\quad \text{ in }  \Mcal_{[s_0,s_1]}.
\end{equation}
with $0<K_0<K_1$. In particular, one has 
\begin{equation}
\lapsb\lesssim \lapsb_{\eta}.
\end{equation} 
\end{claim}

\begin{proof} 
\bse
Observe that
\be
\diff s = J^{-1}\diff t - J^{-1}(x^a/r)\delb_rT\diff r, 
\ee
so that, by Lemma~\ref{lem1-12-june-2025},  
\begin{equation}\label{eq1-08-aout-2025}
\aligned
-{\lapsb^{-2}} & = \, \gb^{00} = g(\diff s,\diff s) = J^{-2}g(\diff t - \delb_rT\diff r,\diff t - \delb_rT\diff r)
\\
& = \,-J^{-2}\zeta^2 + J^{-2}H(\diff t - \delb_rT\diff r,\diff t - \delb_rT\diff r)
\\
& = \,- \lapsb_{\eta}^{-2} 
+ J^{-2}H(\diff t - \diff r + (1- \delb_rT)\diff r,\diff t - \diff r + (1- \delb_rT)\diff r)
\\
& = \,- \lapsb_{\eta}^{-2} + J^{-2}H^{\Ncal00} 
+ J^{-2}\zeta^2\Big(\frac{2H(\diff t- \diff r,\diff r)}{1+\delb_rT} + \frac{1- \delb_rT}{1+\delb_rT}H(\diff r,\diff r)\Big), 
\endaligned
\end{equation}
hence 
\begin{equation}\label{eq4-23-july-2025}
\aligned
- \gb^{00} = {\lapsb^{-2}} = \lapsb_{\eta}^{-2} - J^{-2}H^{\Ncal00}  - J^{-2}\zeta^2R[H]
= \lapsb_{\eta}^{-2}\big(1- \zeta^{-2}H^{\Ncal00} - R[H]\big).
\endaligned
\end{equation}
Then, we obtain 
\begin{equation}\label{eq1-16-june-2025}
(1+\zeta^{-2}|H^{\Ncal00}|)\lapsb^2 = \frac{1+\zeta^{-2}|H^{\Ncal00}|}{1- \zeta^{-2}H^{\Ncal00} - R[H]}\lapsb_{\eta}^2,
\end{equation}
and by Lemma~\ref{lem1-07-aout-2025}, stated next, we arrive at \eqref{eq2-06-aout-2025}. For \eqref{eq1a-16-june-2025}, we only need to point out that, when $\{r\leq 3t/4\}\cup\{r\geq r^{\Ecal}(s)\}$ and provided that \eqref{eq-USA-condition} holds with a sufficiently small $\eps_s$, we have 
$0<C_1\leq(\zeta/\zetab)\leq 1$ where $C_1$ a constant.
\ese
\end{proof} 


\begin{lemma}\label{lem1-07-aout-2025}
Assume that \eqref{eq-USA-condition} holds with a sufficiently small $\eps_s$. Then one has 
\begin{equation}\label{eq1-07-aout-2025}
\frac{1}{2}\big(\zeta^2+|H^{\Ncal00}|\big)\leq \zeta^2-H^{\Ncal00} - \zeta^2R[H]\leq 2(\zeta^2 + |H^{\Ncal00}|).
\end{equation}
\end{lemma}

\begin{proof}
In the region $\{3t/4\leq r\leq r^{\Ecal}(s)\}$, we have $H^{\Ncal00}<0$. Thus \eqref{eq2-14-june-2025} leads us to \eqref{eq1-07-aout-2025}. Outside the region $\{3t/4\leq r \leq r^{\Ecal}(s)\}$, we have $\zeta \geq \sqrt{7}/4$. Then the double inequality \eqref{eq1-07-aout-2025} also holds thanks to \eqref{eq2-14-june-2025}.
\end{proof}

\begin{remark}
In view of \eqref{eq4-23-july-2025} and \eqref{eq9-08-aout-2025}, the estimate \eqref{eq1-07-aout-2025} can also written as
\begin{equation}\label{eq10-aout-2025}
\zetab^2\lesssim-J^2\gb^{00}\lesssim \zetab^2.
\end{equation}
\end{remark}


}

Finally, for application later on, we also establish several estimates on the derivatives of the lapse function. For this purpose, we apply Lemma~\ref{lem1-12-june-2025}, and obtain
\begin{equation}\label{eq1-13-aout-2025}
\betab^b = \lapsb^2\gb^{0b}.
\end{equation}
Indeed, this is a consequence of 
\be
\aligned
\betab^b & =  \gb_{0a}\sigmab^{ab} = \gb_{0a}(\lapsb^2\gb^{a0}\gb^{0b} + \gb^{ab})
= \lapsb^2\big(\gb_{0\alpha}\gb^{\alpha0} - \gb_{00}\gb^{00}\big)\gb^{0b} 
+ \big(\gb_{0\alpha}\gb^{\alpha b} - \gb_{00}\gb^{0b}\big)
\\
& = \lapsb^2\big(1 + \lapsb^{-2}\gb_{00}\big)\gb^{0b}  - \gb_{00}\gb^{0b}
=\lapsb^2\gb^{0b}.
\endaligned
\ee

\begin{lemma}\label{lem1-26-july-2025}
Assume that the uniform spacelike condition \eqref{eq-USA-condition} holds for a sufficiently small $\eps_s$. Then  in $\Mcal^{\Hcal}_{[s_0,s_1]}$, one has 
\begin{equation}
\aligned
|L_a\lapsb|& \lesssim (s/t)^{-2}\lapsb^3\big(|H| + |LH|\big)\lesssim (s/t)^{-2}\big(|H| + |LH|\big),
\\
|\del_{\alpha}\lapsb|& \lesssim (s/t)^{-2}\lapsb^3|\del H| + (s/t)^{-3}s^{-1}\lapsb^3|H|
\lesssim (s/t)^{-2}|\del H| + (s/t)^{-3}s^{-1}|H|
\endaligned
\end{equation}
where
\be
|H| := \max_{\alpha\beta}\{|H_{\alpha\beta}|\},
\quad 
|LH|:=\max_{a,\alpha\beta}\{|L_{a}H_{\alpha\beta}|\},
\quad
|\del H|:=\max_{\gamma,\alpha,\beta} :=\{|\del_{\gamma}H_{\alpha\beta}|\}.
\ee
\end{lemma}

\begin{proof} 
\bse
Recall \eqref{eq6-03-aout-2025} and \eqref{eq11-03-aout-2025}. A direct calculation shows that
\begin{equation}
\big|\del_\alpha(\lapsb^{-2})\big|\lesssim (s/t)^{-3}s^{-1}|H| + (s/t)^{-2}|H|,
\end{equation}
\begin{equation}
\big|L_a(\lapsb^{-2})\big|\lesssim (s/t)^{-2}\big(|H| + |LH|\big).
\end{equation}
Here we have used the following fact (a high-order version will be established later)
$$
|XH^{\alpha\beta}|\lesssim |X H|
$$
with $X = L_a,\del_{\alpha}$, provided that \eqref{eq2-14-june-2025}. In the above calculation we have noticed that in $\{r\leq t-1\}$,
\begin{equation}
\big|\del_{\alpha}(s/t)\big|\lesssim s^{-1},
\quad
\big|L_a(s/t)\big|\lesssim (s/t).
\end{equation}
Finally, combining Claim~\ref{cor1-16-june-2025} and the relation
\be
X(\lapsb) = \frac{1}{2}\lapsb^3X(- \lapsb^{-2}),
\ee
we obtain the desired estimate.
\ese
\end{proof}

We also need a finer estimate which uses the condition $H^{\Ncal00}>0$.

\begin{lemma}\label{lem1-08-aout-2025}
Assume that the uniform spacelike condition \eqref{eq-USA-condition} holds for a sufficiently small $\eps_s$. Then one has 
\begin{equation}\label{eq3-08-aout-2025}
|L_a\lapsb|\lesssim \lapsb\frac{\zeta^{-2}|H^{\Ncal00}|_1}{1+|\zeta^{-2}H^{\Ncal00}|} + \lapsb|H|_1,\quad \text{in}\quad\Mcal^{\Hcal}_{[s_0,s_1]},
\end{equation}
\begin{equation}\label{eq4-08-aout-2025}
\big|\Omega_{ab}\lapsb\big|\lesssim \lapsb\frac{\zeta^{-2}|\Omega_{ab}H^{\Ncal00}|}{1+|\zeta^{-2}H^{\Ncal00}|} + \lapsb |H|_1,\quad \text{in}\quad\Mcal^{\ME}_{[s_0,s_1]}, 
\end{equation}
where
\be
|H|_1 = \max_{\alpha,\beta,a,b}\big\{|H_{\alpha\beta}|, |L_aH_{\alpha\beta}|, |\Omega_{ab}H_{\alpha\beta}|\}. 
\ee
\end{lemma}

\begin{proof}
\bse
These are also based on \eqref{eq6-03-aout-2025} and \eqref{eq11-03-aout-2025}. First of all, we remark that all the estimates are much easier out of the region $\{3t/4\leq r\leq r^{\E}(s)\}$, because $\zeta\geq \sqrt{7}/4>0$ and $|1+\zeta^{-2}|H^{\Ncal00}||\lesssim 1$. In this case, recall that
$$
\lapsb = (- \gb^{00})^{-1/2} = J\Big(\zeta^2 - \big(H^{00} - 2(x^a/r)\delb_rTH^{a0} + (\delb_rT)^2(x^ax^b/r^2)H^{ab}\big)\Big)^{-1/2}.
$$
Then in $\Mcal^{\Hcal}_{[s_0,s_1]}$, $\zeta = J =  s/t$, $(x^a/r)\delb_rT = x^a/t$. By the homogeneity of these functions and the fact that $|L_a(s/t)|\lesssim (s/t)$, we obtain
$$
|L_a(\lapsb)|\lesssim |H|_1,
$$ 
which fits \eqref{eq3-08-aout-2025}. When in $\Mcal^{\Ecal}_{[s_0,s_1]}$, we remark that $\Omega_{ab}(\zeta) = \Omega_{ab}(J) = \Omega_{ab}(\delb_rT) = \delb_rT = 0$. Then \eqref{eq4-08-aout-2025} follows.


Then we treat the region $\{3t/4\leq r\leq r^{\E}(s)\}\cap \Mcal_{[s_0,s_1]}$. For any vector field $X$, we have 
\begin{equation}
XH^{\alpha\beta} = Xg^{\alpha\beta} 
= -g^{\alpha\mu}\big(XH_{\mu\nu}\big)g^{\nu\beta}.
\end{equation}
Also, provided that $|H|$ is sufficiently small, $g^{\alpha\beta}$ are uniformly bounded. Then we find 
\begin{equation}\label{eq4-26-july-2025}
|X g^{\alpha\beta}| = |XH^{\alpha\beta}|\lesssim |XH|:=\max_{\alpha\beta}|XH_{\alpha\beta}|.
\end{equation}
We then recall \eqref{eq1-08-aout-2025}. In fact in $\Mcal_{[s_0,s_1]}$, we have 
\begin{equation}\label{eq2-08-aout-2025}
\aligned
- \lapsb^{-2} = -J^{-2}\zeta^2 + J^{-2}\big(H^{\Ncal00} + \zeta^2R[H]\big)
\endaligned
\end{equation}
with
$$
R[H] = \frac{2H(\diff t- \diff r,\diff r)}{1+\delb_rT} + \frac{1- \delb_r T}{1+\delb_rT}H(\diff r,\diff r).
$$
That is,
\begin{equation}
\lapsb = J\zeta^{-1}\big(1- \zeta^{-2}H^{\Ncal00} - R[H]\big)^{-1/2}.
\end{equation}
For any vector field $X$,
$$
\aligned
X\lapsb & =\frac{X(J\zeta^{-1})}{(1- \zeta^{-2}H^{\Ncal00}-R[H])^{1/2}} 
- \frac{-J\zeta^{-1}X(\zeta^{-2}H^{\Ncal00}+R[H])}{2(1- \zeta^{-2}H^{\Ncal00}-R[H])^{3/2}}
\\
& = \frac{X(J\zeta^{-1})}{(1- \zeta^{-2}H^{\Ncal00}-R[H])^{1/2}} 
+ \frac{1}{(- \gb^{00})^{1/2}}\frac{X(\zeta^{-2}H^{\Ncal00}+R[H])}{2(1- \zeta^{-2}H^{\Ncal00}-R[H])}
\\
& = \frac{X(J\zeta^{-1})}{(1- \zeta^{-2}H^{\Ncal00}-R[H])^{1/2}} 
+ \frac{\lapsb}{2}\frac{X(\zeta^{-2}H^{\Ncal00}+R[H])}{1- \zeta^{-2}H^{\Ncal00}-R[H]}.
\endaligned
$$
When we are in $\Mcal^{\Hcal}_{[s_0,s_1]}$, thus $J\zeta^{-1} \equiv 1$. Furthermore, for $X=L_a$,
$$
\big|X(\zeta^{-2}H^{\Ncal00})\big| = \big|X((s/t)^{-2}H^{\Ncal00})\big|\lesssim (s/t)^{-2}|H^{\Ncal00}|_1\lesssim \zeta^{-2}|H^{\Ncal00}|_1.
$$
Thus thanks to Lemma~\ref{lem1-07-aout-2025} and the homogeneity property of the coefficients in $R[H]$, we obtain \eqref{eq3-08-aout-2025}.

When we are in $\Mcal^{\Ecal}_{[s_0,s_1]}$, $\Omega_{ab}(J\zeta^{-1}) = 0$ because the function is radial. Then
$$
\Omega_{ab}\lapsb 
= \frac{\lapsb}{2}\frac{\zeta^{-2}\Omega_{ab}H^{\Ncal00} + \Omega_{ab}(R[H])}{1- \zeta^{-2}H^{\Ncal00}-R[H]} .
$$
This leads us to \eqref{eq4-08-aout-2025}.
\ese
\end{proof}

\subsection{ Perturbations associated with geometric objects}
\label{section===73}

{ 

To proceed with the analysis of the curved background metric, we need to control the ``difference'' (for diverse geometric objects) between the curved case and the flat case and bound it in terms of the amplitude of the deviation of the curved metric from the Minkowski metric. 

\begin{proposition}
\label{prop1-27-june-2025}
Assume that the uniform spacelike \eqref{eq-USA-condition} holds for a sufficiently small $\eps_s$. Then in the region $\{3t/4\leq r\leq r^{\Ecal}(s)\}\cap\Mcal_{[s_0,s_1]}$, one has 
\begin{equation}
\sigmab^{ab} = Q^{ab} + M_c^aQ^{cb} 
\end{equation}
with
\begin{equation}
Q^{ab} = \sigmab_{\eta}^{ab} 
- \frac{|H^{\Ncal00}|}{\zeta^2(\zeta^2+|H^{\Ncal00}|)}\frac{x^ax^b}{r^2},
\end{equation}
\begin{equation}
|M_c^a|\lesssim \zetab^{-2}|H|, \qquad |M_c^aQ^{cb}|\lesssim \zetab^{-2}|H|.
\end{equation}
\end{proposition}
The proof is postponed to Section~\ref{subsec1-29-sept-2025}. Then we turn our attention to the shift vector.

\begin{claim}
\label{clm-1-29-june-2025}
Assume that the uniform spacelike condition \eqref{eq-USA-condition} holds for a sufficiently small $\eps_s$. Then in the region $\{3t/4\leq r\leq r^{\Ecal}(s)\}\cap\Mcal_{[s_0,s_1]}$, one has 
\begin{equation}\label{eq2-27-june-2025}
\betab^b =  \frac{\zeta^2}{\zeta^2-H^{\Ncal00}}\betab_{\eta}^b 
+ \Hb_{0a}\sigmab^{ab} 
 + \bar{\eta}_{0a}\,M_c^aQ^{cb} 
\end{equation}
where, by recalling \eqref{eq1-27-june-2025}, $\betab_{\eta}^b = -(x^b/r)\delb_rTJ\zeta^{-2}$. In particular, it follows that 
\begin{equation}\label{eq8-31-july-2025}
|\betab^a|\lesssim \zetab^{-2}(J + |H|). 
\end{equation}
\end{claim}

\begin{proof} It is direct to compute  
$$
\betab^b- \betab_{\eta}^b = \gb_{0a}\sigmab^{ab} - \bar{\eta}_{0a}\sigmab_{\eta}^{ab}
= \Hb_{0a}\sigmab^{ab} + \bar{\eta}_{0a}(\sigmab^{ab} - \sigmab_{\eta}^{ab}).
$$
For the second term, we apply Proposition~\ref{prop1-27-june-2025} and $\bar{\eta}_{0a} = -J(x^a/r)\delb_rT$,
and obtain 
$$
\aligned
\bar{\eta}_{0a}(\sigmab^{ab} - \sigmab_{\eta}^{ab})
& =  J\frac{x^a}{r}\delb_rT\frac{|H^{\Ncal00}|}{\zeta^2(\zeta^2+|H^{\Ncal00}|)}\frac{x^ax^b}{r^2}
+\bar{\eta}_{0a}M_c^aQ^{cb}
\\
& = \frac{|H^{\Ncal00}|}{\zeta^2(\zeta^2+|H^{\Ncal00}|)}J\delb_rT\frac{x^b}{r}
+\bar{\eta}_{0a}M_c^aQ^{cb}
\\
& = - \frac{|H^{\Ncal00}|}{\zeta^2+|H^{\Ncal00}|}\betab_{\eta}^b + \bar{\eta}_{0a}M_c^aQ^{cb}.
\endaligned
$$
The desired result follows, provided that $H^{\Ncal00}<0$, which is assumed.
\end{proof}

\begin{claim}
Assume that the uniform spacelike condition \eqref{eq-USA-condition} holds for a sufficiently small $\eps_s$. Then in the region $\{3t/4\leq r\leq r^{\Ecal}(s)\}\cap\Mcal_{[s_0,s_1]}$, one has 
\begin{equation}
\big|Q^{ab}\delb_b\big|_{\sigmab}\lesssim \zetab^{-1}\lesssim \zeta^{-1}.
\end{equation}
\end{claim}

\begin{proof} 
\bse
A direct calculation gives us 
$$
\aligned
Q^{ab}\delb_b 
& =  \frac{x^a}{r}\zeta^{-2}\Big((\delb_rT)^2 - \frac{|H^{\Ncal00}|}{\zeta^2+|H^{\Ncal00}|}\Big)\delb_r + \delb_a
\\
& = - \frac{x^a}{r}\zeta^{-2}\big(1-(\delb_rT)^2\big)\delb_r + \frac{x^a}{r}\zeta^{-2}\Big(1- \frac{|H^{\Ncal00}|}{\zeta^2+|H^{\Ncal00}|}\Big)\delb_r + \delb_a
\\
& = - \frac{x^a}{r}\delb_r + \frac{x^a}{r}\frac{1}{\zeta^2+|H^{\Ncal00}|}\delb_r + \delb_a.
\endaligned
$$
Recalling Lemma~\ref{lem1-14-june-2025}, we have 
\begin{equation}
\big| |\delb_r|_{\sigmab}^2 - (\zeta^2 + |H^{\Ncal00}|)\big|\lesssim \eps_s\zeta^2, 
\end{equation}
which leads us to
\begin{equation}
|\delb_r|_{\sigmab}\lesssim \zeta + |H^{\Ncal00}|^{1/2}.
\end{equation}
On the other hand, we have 
\begin{equation}\label{eq2-30-june-2025}
|\delb_a|_{\sigmab}^2 = |\gb_{aa}| = |\bar{\eta}_{aa}| + |\Hb_{aa}|\lesssim 1
\end{equation}
and we arrive at the desired result.
\ese
\end{proof}
}

\begin{claim}\label{clm1-12-aout-2025}
Assume that the uniform spacelike condition \eqref{eq-USA-condition} holds for a sufficiently small $\eps_s$. Then 
\begin{equation}\label{eq3-12-aout-2025}
\big|\sigmab^{ab} - \sigmab_{\eta}^{ab}\big|\lesssim \zeta^{-2}\zetab^{-2}|H|.
\end{equation}
\end{claim}
\begin{proof}
Thanks to Proposition~\ref{prop1-27-june-2025}, in $\{3t/4\leq r\leq r^{\Ecal}(s)\}\cap\Mcal_{[s_0,s_1]}$,
$$
\sigmab^{ab} - \sigmab_{\eta}^{ab} = - \frac{|H^{\Ncal00}|}{\zeta^2(\zeta^2+|H^{\Ncal00}|)}\frac{x^ax^b}{r^2} + M_c^aQ^{cb}.
$$
Then \eqref{eq3-12-aout-2025} follows. Out of the region $\{3t/4\leq r\leq r^{\Ecal}(s)\}$, the estimate is trivial because $\zetab\geq \zeta\geq \sqrt{7}/4$.
\end{proof}


\section{Energy functionals and norms of interest}
\label{section=N5}

\subsection{ Weighted energy identity for Dirac operator}
\label{section===81}

{ 

It is convenient to introduce the notation
\be
\|u\|_{L^1_{\sigmab}(\Mcal_s)} := \int_{\Mcal_s}|u|\mathrm{Vol}_{\sigmab},
\quad
\|u\|_{L^1(\Mcal_s)} := \int_{\Mcal_s}|u|\diff x.
\ee

\begin{proposition}\label{prop1-04-april-2025}
Suppose that $g_{\alpha\beta} = \eta_{\alpha\beta} + H_{\alpha\beta}$ is a sufficiently regular metric defined in $\mathcal{M}_{[s_0,s_1]}$. Then any $C^1$ solution $\Psi$ to the Dirac equation
\be
\opDirac\Psi + \mathrm{i}M\Psi = \Phi, 
\ee
defined in the domain $\mathcal{M}_{[s_0,s_1]}$ and vanishing sufficiently fast at the spatial infinity, satisfies the energy identity  
\begin{equation}\label{eq1a-05-april-2025}
\aligned
& \int_{\mathcal{M}_{s_1}} w\la \Psi,\vec{n} \cdot\Psi\ra_{\ourD} \mathrm{Vol}_{\sigmab}
+\int_{s_0}^{s_1}\int_{\Mcal_s}
\la \Psi,\mathrm{grad}(w)\cdot\Psi\ra_{\ourD}\,\lapsb \, \mathrm{Vol}_{\sigmab}\diff s
\\
& =  \int_{\mathcal{M}_{s_0}} w\la \Psi,\vec{n} \cdot\Psi\ra_{\ourD}\mathrm{Vol}_{\sigmab}
- 2\int_{s_0}^{s_1}\int_{\Mcal_s}w\Re\big(\la \Phi, \Psi \ra_{\ourD} \big)\,\lapsb \, \mathrm{Vol}_{\sigmab}\diff s. 
\endaligned
\end{equation}
Here, $\vec{n}$ is the future oriented, unit normal vector to $\Mcal_s$, and $w$ is a sufficiently regular scalar weight. Furthermore, one has 
\begin{equation}\label{eq1-10-oct-2025}
\aligned
& \int_{\mathcal{M}_{s_1}} w\la t^{-1}\Psi,\vec{n} \cdot t^{-1}\Psi\ra_{\ourD} \mathrm{Vol}_{\sigmab}
\\
& \quad +\int_{s_0}^{s_1}\int_{\Mcal_s}
\big(\la t^{-1}\Psi,\mathrm{grad}(w)\cdot t^{-1}\Psi\ra_{\ourD} 
+ t^{-1}g^{0\mu}\la t^{-1}\Psi,\del_{\mu}\cdot t^{-1}\Psi \ra_{\ourD}\big)
\,\lapsb \, \mathrm{Vol}_{\sigmab}\diff s
\\
& =  \int_{\mathcal{M}_{s_0}} w\la t^{-1}\Psi,\vec{n} \cdot t^{-1}\Psi\ra_{\ourD}\mathrm{Vol}_{\sigmab}
- 2\int_{s_0}^{s_1}\int_{\Mcal_s}t^{-1}w\Re\big(\la \Phi, \Psi \ra_{\ourD} \big)\,\lapsb \, \mathrm{Vol}_{\sigmab}\diff s. 
\endaligned
\end{equation}
\end{proposition}

\begin{proof} 
\bse
Recall that
\be
\mathrm{grad}(w) =: g^{\alpha\beta}\del_{\alpha}w \del_{\beta}. 
\ee
The identity 
\be
\mathrm{div}(w V[\Psi]) 
= \la\Psi,\mathrm{grad}(w)\cdot\Psi\ra_{\ourD}
+ w\mathrm{div}(V[\Psi])
\ee
implies, thanks to Proposition~\ref{propo-JDK02},  that 
\begin{equation}\label{eq1-13-feb-2025}
\mathrm{div}(w V[\Psi]) 
= 2\Re(\la w\opDirac\Psi,\Psi \ra_{\ourD} )
+ \la \Psi,\mathrm{grad}(w)\cdot\Psi\ra_{\ourD}.
\end{equation}
On the other hand, we  have 
\begin{equation}\label{eq9-05-april-2025}
2\Re(\la\mathrm{i}w M\Psi,\Psi\ra_{\ourD}) 
= w\la \Psi, \mathrm{i}M\Psi\ra_{\ourD} + w\la \mathrm{i}M\Psi,\Psi\ra_{\ourD} = 0.
\end{equation}
We integrate the sum of \eqref{eq1-13-feb-2025} and \eqref{eq9-05-april-2025} in the domain $\mathcal{M}_{[s_0,s_1]}$ and apply Stokes' formula\footnote{See, for instance, \cite[Section 3.2]{Alinhac-book} with, in the convention therein, $s = f$ and $\nabla f/||\nabla f|| = - \vec{n}$.}: 
\begin{equation}
\label{equa-10mai2025}
\aligned
\int_{\mathcal{M}_{s_0}} - \int_{\mathcal{M}_{s_1}}w V[\Psi]\cdot \vec{n} \mathrm{Vol}_{\sigmab}
=  \int_{\mathcal{M}_{[s_0,s_1]}}
\hspace{-.4cm}\la \Psi,\mathrm{grad}(w)\cdot\Psi\ra_{\ourD} \mathrm{Vol}_g
+ 2\int_{\mathcal{M}_{[s_0,s_1]}}\hspace{-.4cm}w \Re\big\la \Phi , \Psi\big\ra_{\ourD}  \mathrm{Vol}_g.
\endaligned
\end{equation}
It remains to apply Lemma~\ref{lem1-25-mai-2025} on $\mathrm{Vol}_g$, and we obtain \eqref{eq1a-05-april-2025}.
\ese
On the other hand, to derive \eqref{eq1-10-oct-2025}, we use the equation
\begin{equation}\label{eq5-10-oct-2025}
\opDirac (t^{-1}{\Psi}) + \mathrm{i}Mt^{-1}\Psi + t^{-1}g^{\mu0}\del_{\mu}\cdot(t^{-1}\Psi) 
= t^{-1}\Phi
\end{equation}
and apply \eqref{eq1a-05-april-2025}.
\end{proof}


Observe in passing that, with respect to the Minkowski metric, in the domain $\Mcal^{\Hcal}_{[s_0,s_1]}$
\be
\vec{n}_{\eta}  := \zeta^{-1} \del_t + \zeta^{-1}\delb_rT(x^a/r)\del_a, \qquad \mathrm{Vol}_{\sigmab} = (s/t)\diff x.
\ee
In the global orthogonal frame $\{e_i\}$ with $e_0 = \vec{n}$, we have 
\be
\int_{\Hcal_{s_1}}|\psi|^2(s/t)\diff x = \int_{\Hcal_{s_0}}  |\psi|^2(s/t)\diff x 
\ee
for the case $\Phi \equiv 0$ and $\Psi$ supported in $\{r<t-1\}$. This is nothing but the standard conservation law for the probability of particles.

We now introduce the weighted energy 
\begin{equation}\label{eq4-19-aout-2025}
\Ebf_{g,w}(s,\Psi) := \int_{\Mcal_s}w\la\Psi,\vec{n}\cdot\Psi\ra_{\ourD} \mathrm{Vol}_{\sigmab}
\end{equation}
and the pointwise norm associated with $\vec{n}$:
\begin{equation}\label{eq3-26-june-2025}
|\Psi|_{\vec{n}}^2 := \la\Psi,\vec{n}\cdot\Psi \ra_{\ourD}.
\end{equation}
From a physical perspective, the expression $|\Psi|_{\vec{n}}^2$ represents the probability density of a fermion detected by the observer whose worldline is tangent to $\vec{n}$.

}

\subsection{ Norms associated with timelike vector fields}
\label{section===82}
\paragraph{Decomposition of vector fields.}
{ 
In \eqref{eq3-26-june-2025}, we introduced a norm on spinor fields, which is now extended as follows.  Let $\vec{s}$ be a (non-vanishing) future-oriented timelike vector field defined in (possibly a subset of) $\Mcal_{[s_0,s_1]}$. For any vector field $X$, we then introduce the orthogonal decomposition
\bse
\be
X^{\perp}_{\vec{s}} := \frac{(X,\vec{s})_g\vec{s}}{(\vec{s},\vec{s})_g},\quad X^{\top}_{\vec{s}} = X - X^{\perp}_{\vec{s}}.
\ee
\ese
It is obvious that we define a norm by 
\begin{equation}\label{eq4-26-june-2025}
|X|_{\vec{s}} := \big||X^{\perp}_{\vec{s}}|_g\big| + |X^{\top}_{\vec{s}}|_{\sigmab}.
\quad
\end{equation}
In particular, this decomposition applies with the choice of the future oriented unit normal vector $\vec{n}$  to $\Mcal_s$, and then we use the short-hand notation 
\begin{equation}\label{eq4-27-june-2025}
X^{\perp} = X^{\perp}_{\vec{n}}, \qquad X^{\top} = X^{\top}_{\vec{n}}.
\end{equation}

\paragraph{Dirac form.}

We begin with the following result.
\begin{lemma}[Cauchy-Schwartz inequality]
\label{prop2-04-mai-2025}
Let $\vec{n}$ be a future oriented, timelike unit vector field. Then for any spinor fields $\Phi,\Psi$ defined in $\Mcal_{[s_0,s_1]}$, one has 
\begin{equation}
\big|\la \Phi,\vec{n}\cdot\Psi\ra_{\ourD}\big|\lesssim |\Phi|_{\vec{n}}|\Psi|_{\vec{n}}.
\end{equation}
\end{lemma}

\begin{proof} We only need to check that, for $t\in \RR$,
$$
f(t) := \la \Phi -t\Psi,\vec{n}\cdot (\Phi- t\Psi) \ra_{\ourD}\geq 0, 
$$
which is equivalent to saying that the discriminant is always negative, that is, 
$$
|\Re\la \Phi,\vec{n}\cdot\Psi\ra_{\ourD}|^2\leq |\Phi|_{\vec{n}}|\Psi|_{\vec{n}}.
$$
We can similarly consider
$$
g(t) := \la \Phi - \mathrm{i}t\Psi, \vec{n}(\Phi - \mathrm{i}t\Psi)\ra_{\ourD}\geq 0, 
$$
and the above argument on the discriminant also leads us to
$|\Im \la \Phi,\vec{n}\cdot\Psi\ra_{\ourD}|^2\leq |\Phi|_{\vec{n}}|\Psi|_{\vec{n}}$. 
\end{proof}

\begin{lemma}\label{lem1-06-april-2025}
For any vector field $X$ defined in (an open subset of) $\Mcal_{[s_0,s_1]}$, one has 
\begin{equation}\label{eq1-06-april-2025}
|\la \Phi,X\cdot\Psi\ra_{\ourD}|\lesssim |X|_{\vec{n}}|\Phi|_{\vec{n}}|\Psi|_{\vec{n}}.
\end{equation}
\end{lemma}

\begin{proof}
\bse
Let $\{e_0,e_1,e_2,e_3\}$ be a tetrad on $\Mcal_{[s_0,s_1]}$ with $e_0 = \vec{n}$ and\footnote{If $X^{\top}=0$, then take any positive oriented frame $\{e_1,e_2,e_3\}$ tangent to $\Mcal_s$.} $e_1 = X^{\top}/|X^{\top}|$. Let $\phi,\psi$ be the coordinate of $\Phi,\Psi$ in this tetrad. Then we find 
$$
X = X^0\vec{n} + X^1e_1, \qquad |X^0| = ||X^{\perp}|_g|, \qquad |X^1| = |X^{\top}|_{\sigmab} 
$$
and
$$
\la \Phi,X\cdot\Psi\ra_{\ourD} = X^0\phi^{\dag}\psi + X^1\phi^{\dag}\gamma_0\gamma_1\psi, 
$$
which leads us to 
$$
|\la \Psi,X\cdot\Psi\ra_{\ourD}|\leq |X^0||\phi||\psi| + C|X^1||\phi||\psi|. \qedhere
$$
\ese
\end{proof}

\begin{corollary}\label{cor1-04-mai-2025}
In the flat case,  one has 
\begin{equation}
\zeta\big|\la\Phi,\del_t\cdot\Psi \ra_{\ourD}\big|\lesssim |\Phi|_{\vec{n}}|\Psi|_{\vec{n}}
\qquad \text{ when } g=\eta. 
\end{equation}
\end{corollary}

\begin{proof} We only need to observe that
$\zeta\del_t = \vec{n}_{\eta} - \zeta^{-1}\delb_rT (x^a/r)\delb_a$, in which 
$$
|\vec{n}|_\eta = 1, \qquad |(x^a/r)\delb_a|_{\eta} = \zeta. \qedhere
$$
\end{proof}


\paragraph{A norm based on Clifford multiplication of spinor fields.}
From another perspective, a vector $X$ can be related to a linear map between spinor fields via the Clifford multiplication, namely 
$
X: \Psi\mapsto X\cdot\Psi.
$
Thus we can define a (pointwise) norm on $X$ by setting 
\be
|X|_{\vec{s}} := \sup_{|\Psi|_{\vec{s}}=1}|X\cdot\Psi|_{\vec{s}}.
\ee
This may seem, at first, to lead to a conflict of notation with respect to~\eqref{eq4-26-june-2025}: however,   they do define the same object, as we now show. 

\begin{lemma}\label{lem1-16-july-2025}
For any vector $X$ and any future-oriented timelike unit vector $\vec{s}$, one has 
\begin{equation}
\sup_{|\Psi|_{\vec{s}}=1}|X\cdot\Psi|_{\vec{s}} = ||X^{\perp}_{\vec{s}}|_g| + |X^{\top}_{\vec{s}}|_{\sigmab}.
\end{equation}
\end{lemma}

\begin{proof} 
We are going to establish first an upper bound and then construct a specific spinor that achieves the bound. 

\bse
{\bf 1. Derivation of an upper bound.} 
For any vector $X$ interpreted as an element of the Clifford algebra, we consider (here for simplicity, in this proof we write $X^{\perp}$ for $X^{\perp}_{\vec{s}}$)
\be
\aligned
X\cdot\vec{s}\cdot X & =  (X^{\perp} + X^{\top})\cdot\vec{s}\cdot(X^{\perp} + X^{\top}) 
\\
& =  X^{\perp}\cdot\vec{s}\cdot X^{\perp} + X^{\top}\cdot\vec{s}\cdot X^{\top} 
+ X^{\perp}\cdot\vec{s}\cdot X^{\top} + X^{\top}\cdot\vec{s}\cdot X^{\perp}, 
\endaligned
\ee
and we introduce $a, b, \vec{f}$ by 
\be
X^{\perp} =:  a\vec{s}, \qquad X^{\top} =: b\vec{f}, 
\ee
where $a,b\in \RR$, $\vec{f}$ is orthogonal to $\vec{s}$ and $|\vec{f}|_g = 1$. We can then compute 
\be
X\cdot\vec{s}\cdot X
= a^2\vec{s}\cdot\vec{s}\cdot\vec{s} + ab\vec{s}\cdot\vec{s}\cdot\vec{f} + ba\vec{f}\cdot\vec{s}\cdot\vec{s} + b^2\vec{f}\cdot\vec{s}\cdot\vec{f}.
\ee
Since $\vec{s} \cdot \vec{s} \cdot \vec{s} = (\vec{s} \cdot \vec{s}) \cdot \vec{s} = \vec{s}$ for a unit timelike vector. And we also used that $v \cdot w + w \cdot v = -2 g(v, w) \cdot 1 = 0$ for orthogonal vectors $v,w$, hence $\vec{s}\cdot\vec{f} = - \vec{f}\cdot\vec{s}$ since $\vec{s} \perp \vec{f}$. That is, we have  that $\vec{f} \cdot \vec{s} \cdot \vec{f} = - \vec{f} \cdot \vec{f} \cdot \vec{s} = g(\vec{f}, \vec{f}) \vec{s} = \vec{s}$. Thus
\begin{equation}
X\cdot\vec{s}\cdot X = (a^2+b^2)\vec{s} + 2ab\vec{f}.
\end{equation}
\ese

Consequently, we deduce that 
\bse
\be
\aligned
\la X\cdot\Psi,\vec{s}\cdot X\cdot\Psi\ra_{\ourD} 
 = \la \Psi,X\cdot\vec{s}\cdot X\cdot\Psi\ra_{\ourD}
 = (a^2+b^2)\la \Psi,\vec{s}\cdot\Psi\ra_{\ourD} + 2ab\la\Psi,\vec{f}\cdot \Psi\ra_{\ourD}.
\endaligned
\ee
By introducing the coordinate $\psi$ of $\Psi$ in the orthonormal frame $\{e_0=\vec{s},e_1=\vec{f},e_2,e_3\}$,
we obtain 
\be
\la \Psi,\vec{f}\cdot\Psi\ra_{\ourD} = \psi^{\dag}\gamma_0\gamma_1\psi. 
\ee
We note that
\begin{equation}
\gamma_0\gamma_1 = 
\left(
\begin{array}{cc}
\sigma_1&0
\\
0& \quad - \sigma_1
\end{array}
\right).
\end{equation}
This is a Hermitian matrix with eigenvalues $\{\pm1\}$. Thus we find 
\be
|\psi^{\dag}\gamma_0\gamma_1\psi|\leq |\psi|^2, 
\ee
that is,
\be
|X\cdot\Psi|_{\vec{s}}^2 = \la X\cdot\Psi,\vec{s}\cdot X\cdot\Psi\ra_{\ourD} \leq (|a|+|b|)^2 |\psi|^2 
= (|a|+|b|)^2|\Psi|_{\vec{s}}^2,
\ee
so that
\be
|X|_{\vec{s}} \leq |a| + |b| = ||X^{\perp}|_g| + |X^{\top}|_{\overline \sigma}|. 
\ee
\ese
\bse
{\bf 2. Saturation of the upper bound.}
We now express the Dirac sesquilinear forms in the orthonormal frame consisting of 
\be
e_0 = \vec{s}, \qquad e_1=\vec{f}, \qquad  e_2, e_3\text{ tangent to } \{e_0,e_1\}.
\ee
Let $\psi\in\mathbb{C}^4$ be the coordinate of $\Psi$ in this frame. Then we have 
\be
\la X\cdot\Psi,\vec{s}\cdot X\cdot\Psi\ra_{\ourD} = (a^2+b^2)\psi^{\dag}\psi 
+ 2ab\psi\gamma_0\gamma_1\psi.
\ee
Suppose that $\psi^{\dag}\psi = 1$, which is equivalent to saying that $|\Psi|_{\vec{s}}=1$. We observe that
\be
\gamma_0\gamma_1 = 
\left(
\begin{array}{cc}
\sigma_1&0
\\
0& \quad - \sigma_1
\end{array}
\right).
\ee
We can thus take $\psi$ an eigenvector of $\gamma_0\gamma_1$ associated to the eigenvalue $|ab|/ab$ (the case $ab=0$ is trivial), which leads us to 
\be
\aligned
\la X\cdot\Psi,\vec{s}\cdot X\cdot \Psi\ra_{\ourD} 
& = a^2+b^2 + 2|ab| = (|a|+|b|)^2.
\endaligned
\ee
We thus find
\be
\sup_{|\Psi|_{\vec{s}}=1}|X\cdot\Psi|_{\vec{s}} \geq |a|+|b|, 
\ee
which provides us with the desired result.
\ese
\end{proof}
   

This leads us to the fact that $T_x\Mcal$ equipped with the Clifford multiplication together with the norm $|\cdot|_{\vec{s}}$ forms a normed algebra. More precisely, we can state the following property. 

\begin{corollary}\label{cor1-18-july-2025}
If $\Psi$ is a spinor field and $X_1,X_2,\cdots, X_n$ are finitely many vector fields, one has 
\begin{equation}
\Big|\Big(\prod_{\imath =1}^nX_{\imath}\Big)\cdot\Psi\Big|_{\vec{s}}\leq \prod_{\imath=1}^n|X_{\imath}|_{\vec{s}}|\Psi|_{\vec{s}}.
\end{equation}
\end{corollary}

\begin{corollary}\label{cor1-18-aout-2025}
For any two spinor fields, one has 
\begin{equation}
\big|\la \Phi,\Psi\ra_{\ourD}\big|\lesssim |\Phi|_{\vec{n}}|\Psi|_{\vec{n}}.
\end{equation}
\end{corollary}

\begin{proof}
\bse
We write
\be
\la \Phi,\Psi\ra_{\ourD} = - \la\vec{n}\cdot\Phi,\vec{n}\cdot\Psi \ra_{\ourD}.
\ee
Then thanks to Lemma~\ref{prop2-04-mai-2025} and Corollary~\ref{cor1-18-july-2025},
\be
\big|\la\Phi,\Psi \ra_{\ourD}\big|\lesssim |\vec{n}\cdot\Phi|_{\vec{n}}|\Psi|_{\vec{n}}
\lesssim |\Phi|_{\vec{n}}|\Psi|_{\vec{n}}.
\ee
\ese
\end{proof}


\paragraph{Norms associated with frames.}

In our forthcoming analysis, we will mostly use the norm associated with $\vec{n}$. However, occasionally, we will also use the norm associated to $\del_t$. The following estimate will be necessary. 

\begin{proposition}\label{prop1-09-aout-2025}
Assume \eqref{eq-USA-condition} with a sufficiently small $\eps_s$.  Then there exist a positive constant $K_1$, such that
\begin{equation}
\zetab|\del_t|_{\vec{n}} \leq K_1.
\end{equation}
Thus, by \eqref{eq1-06-april-2025}, one has 
\begin{equation}
 \zetab|\Psi|_{\del_t} \lesssim |\Psi |_{\vec{n}}. 
\end{equation}
\end{proposition}

\begin{proof}
\bse
We first establish the following estimates. There are positive constants $K_0,K_1$, such that
\begin{equation}
K_0\zeta\leq|\delb_a|_{\vec{n}}\leq K_1.
\end{equation}
To see this, we  note that
\be
|\delb_a|_{\vec{n}}^2 = |\delb_a|_{\sigmab}^2 = \gb_{aa} = \bar{\eta}_{aa} + \Hb_{aa} 
= 1- \Big(\frac{x^a}{r}\Big)^2(\delb_rT)^2 + \Hb_{aa}.
\ee
Then when $|H|\lesssim \zeta\eps_s$ with $\eps_s$ sufficiently small, \eqref{eq3-14-july-2025} leads us to the estimate on $|\delb_a|_{\vec{n}}$ in \eqref{eq2-17-july-2025}.

We then note the identity
\begin{equation}\label{eq2-07-aout-2025}
\aligned
\del_t =J^{-1}\lapsb\vec{n} + J^{-1}\betab^b\delb_b 
& = \big(\zeta^2 - H^{\Ncal00} - \zeta^2R[H]\big)^{-1/2}\vec{n} 
- \frac{\delb_rT}{\zeta^2 - H^{\Ncal00} - \zeta^2 R[H]}\delb_r
\\
& \quad- \frac{-H^{b0} + (x^c/r)\delb_rT H^{cb}}{\zeta^2 - H^{\Ncal00} - \zeta^2 R[H]}\delb_b,
\endaligned
\end{equation}
where we have applied \eqref{eq4-23-july-2025} and \eqref{eq1-13-aout-2025} and \eqref{eq11-03-aout-2025}.
\ese
\bse
Based on \eqref{eq1-07-aout-2025}, the first term in the right-hand side of \eqref{eq2-07-aout-2025} is bounded as expected. For the last term, we remark that, thanks to \eqref{eq2-14-june-2025},
\be
\Big|\frac{-H^{b0} + (x^c/r)\delb_rT H^{cb}}{\zeta^2 - H^{\Ncal00} - \zeta^2 R[H]}\delb_b\Big|_{\vec{n}}\leq \frac{\zeta}{\zeta^2+|H^{\Ncal00}|}\eps_s\leq (\zeta^2+|H^{\Ncal00}|)^{-1/2}\eps_s. 
\ee
For the second term in the right-hand side of \eqref{eq2-07-aout-2025}, we recall Lemma~\ref{lem1-14-june-2025}
\be
\big|(\delb_r,\delb_r)_g - \big(\zeta^2 - H^{\Ncal00}\big)\big|\lesssim \zeta^2|H| + |H|^2
\lesssim \zeta^2\eps_s\lesssim (\zeta^2+|H^{\Ncal00}|)\eps_s.
\ee
Thus we obtain
\be
\Big|\frac{\delb_rT}{\zeta^2 - H^{\Ncal00} - \zeta^2 R[H]}\delb_r\Big|_{\vec{n}}
\lesssim(\zeta^2+|H^{\Ncal00}|)^{-1/2}.
\ee
Then we arrive at the expected estimate by sum up the above results.
\ese
\end{proof}


\begin{lemma}\label{lem1-23-july-2025}
Assume that \eqref{eq-USA-condition} holds for some sufficiently small $\eps_s$. Then one has 
\begin{equation}\label{eq2-17-july-2025}
K_0J\zetab^{-1}\leq |\delb_0|_{\vec{n}}\leq K_1J\zetab^{-1},
\quad
K_0\zeta\leq|\delb_a|_{\vec{n}}\leq K_1.
\end{equation}
\end{lemma}

\begin{proof}
\bse
We only need to estimate $|\delb_0|_{\vec{n}}$, which is direct the from the decomposition
\be
\delb_0 = \lapsb\vec{n} + \betab^a\delb_a, 
\ee
from which we see that, thanks to the fact that $\delb_a\perp\vec{n}$, 
\be
|\delb_0|_{\vec{n}}\geq \lapsb \geq K_0J\zetab^{-1}, 
\ee
due to Claim~\ref{cor1-16-june-2025}. The other half of control is direct form Proposition~\ref{prop1-09-aout-2025}.
\ese
\end{proof}

We now establish the following estimate.

\begin{lemma}\label{lem2-13-july-2025}
Assume that \eqref{eq-USA-condition} holds for a sufficiently small $\eps_s$.  
There exist constants $K_0<K_1$ such that, for any $X = X^{\alpha}\del_{\alpha}$, 
\begin{equation}\label{eq8-14-july-2025}
|X|_{\vec{n}}\lesssim K_1\zetab^{-1}\max_{\alpha}|X^{\alpha}|.
\end{equation}
\end{lemma}

\begin{proof}
Given any $X = X^{\alpha}\del_{\alpha}$, we observe that from Proposition~\ref{prop1-09-aout-2025}, Lemma~\ref{lem1-23-july-2025}, and the relation $\del_a = -(x^a/r)\delb_rT \del_t + \delb_a$,
\be
|\del_t|_{\vec{n}}\lesssim \zetab^{-1},
\quad
|\del_a|_{\vec{n}}\leq |\del_t|_{\vec{n}} + \sqrt{\gb_{aa}}\lesssim \zetab^{-1}. 
\ee
Then we find 
\begin{equation}\label{eq4-13-july-2025}
|X|_{\vec{n}}\leq |X^0\del_t|_{\vec{n}} + |X^a\del_a|_{\vec{n}}\lesssim \max_{\alpha}|X^{\alpha}|\zetab^{-1}.
\end{equation}
This establishes the inequality in \eqref{eq8-14-july-2025}.
\end{proof}


Once the norm is defined for vectors, we generalize its definition to tensors through the standard way. First, for the co-vectors (one forms), we define the norm by duality, i.e.,
\begin{equation}
|\omega|_{\vec{s}} := \sup_{|X|_{\vec{s}}=1}\{|\omega(X)|\}.
\end{equation}
In addition, for a general $(r,s)$ type tensor $T$, we define
\be
|T|_{\vec{s}} : = \sup_{|X_k|_{\vec{s}}=1\atop|\theta_j|_{\vec{s}} = 1} \{|T(\theta_1,\theta_2,\cdots,\theta_r,X_1,X_2,\cdots,X_s)|\}.
\ee
It is easily checked that $|\cdot|_{\vec{n}}$ is a norm. We then point out the following properties, whose proof is immediate from our definitions. 

\begin{lemma}
For any tensors $T,U$ of type $(r_1,s_1)$ and $(r_2,s_2)$, one has 
\begin{equation}\label{eq5-04-july-2025}
|T\otimes U|_{\vec{s}}\leq |T|_{\vec{s}}|U|_{\vec{s}},
\end{equation}
\begin{equation}\label{eq6-04-july-2025}
|\mathrm{tr}_i^j\, T|_{\vec{s}}\leq 4|T|_{\vec{s}}.
\end{equation}
\end{lemma}

 The above properties can also be written in components. For example, if $T$ is a type $(2,0)$ tensor and $X$ is a vector, we have 
\begin{equation}
|T_{\alpha\beta}X^{\alpha}\diff x^{\beta}|_{\vec{n}}\lesssim |T|_{\vec{n}}|X|_{\vec{n}}.
\end{equation}

}


\subsection{ Weighted energy estimates in the merging-Euclidean domain}

{ 
We pick a weight function
\begin{equation}\label{equa-weight1} 
\omega := \aleph(r-t), 
\end{equation}
where $\aleph$ is a smooth and non-decreasing function, satisfying
\be
\aligned
& \aleph(y) = 0 \quad\text{for}\quad y \leq -2 \quad\text{and}\quad \aleph(y) = y+2 \quad \text{for}\quad y \geq -1,
\\
& \aleph'(y)>0 \quad \text{for}\quad y\geq -1.
\endaligned
\ee
Observe that this weight satisfies
\be
\omega \simeq 
\begin{cases}
0, \quad & \quad \Mcal^{\Hcal}_{[s_0,s_1]},
\\
(2+r-t),\quad  & \quad\Mcal^{\ME}_{[s_0,s_1]}.
\end{cases}
\ee
We also need the energy density 
\bse\label{equa-28-sept-2025a}
\be
\ebf_{\kappa}[\Psi] := (1+\omega^{2\kappa})\zetab \, |\Psi|_{\vec{n}}^2 
= (1+\omega^{2\kappa})\zetab\la \Psi,\vec{n}\cdot\Psi\ra_{\ourD} 
\ee
and the associated energy functional 
\begin{equation}\label{eq5-06-oct-2025}
\Ebf_{\kappa}(s,\Psi) := \int_{\Mcal_s}\ebf_{\kappa}[\Psi] \, \diff x.  
\end{equation}
\ese
Thanks to Claim~\ref{lem1-02-march-2025}, it is clear that $\Ebf_{\kappa}(s,\Psi)$ is equivalent to $\Ebf_{g,w}(s,\Psi)$ introduced earlier in~\eqref{eq4-19-aout-2025}. In~\eqref{equa-28-sept-2025a}, we have used the weight \eqref{equa-weight1}.

For future reference, we establish first the following result.

\begin{lemma}[Timelike property of the weight]
\label{lem1-18-aout-2025}
Assume that \eqref{eq-USA-condition} holds for a sufficiently small $\eps_s$, and
\begin{equation}\label{eq2-09-oct-2025(l)}
H^{\Ncal00} < 0\quad \text{in }\{r\geq 3t/4\}\cap \Mcal_{[s_0,s_1]}.
\end{equation}  
Then $\mathrm{grad}(r-t)$ is future-oriented timelike, thus
\begin{equation}
\la \Psi,\mathrm{grad}\big(\omega^{2\kappa}\big)\cdot\Psi\ra_{\ourD}\geq 0.
\end{equation}
\end{lemma}

\begin{proof} 
\bse
A direct calculation leads us to
\be
\aligned
\big(g^{\alpha\beta}\del_{\alpha}(t-r)\del_{\beta},g^{\alpha'\beta'}\del_{\alpha'}(t-r)\del_{\beta'}\big)_g 
& =  g^{\alpha\alpha'}\del_{\alpha}(t-r)\del_{\alpha'}(t-r) 
= \big(\diff(t-r),\diff(t-r)\big)_g
\\
& =  H^{\Ncal00}
\endaligned
\ee
and, therefore,
\be
\mathrm{grad}\big(\omega^{2\kappa}\big)
= 2\kappa\omega^{2\kappa-1}\aleph'(r-t)g^{\alpha\beta}\del_{\alpha}(r-t)\del_{\beta}
= 2\kappa\omega^{2\kappa-1}\aleph'(r-t)\mathrm{grad}(r-t). 
\ee
Consequently, we have 
\be
\big(\mathrm{grad}\big(\omega^{2\kappa}\big),\mathrm{grad}\big(\omega^{2\kappa}\big)\big)_g<0,\quad \text{ in }  \{r\geq 3t/4\}\cap\Mcal_{[s_0,s_1]}.
\ee
Thus $\mathrm{grad}\big(\omega^{2\kappa}\big)$ is timelike. 
\ese
\bse
On the other hand, we have 
\be
\aligned
\mathrm{grad}(\omega^{2\kappa}) 
& =  2\kappa\omega^{2\kappa-1}\aleph'(r-t)g^{\alpha\beta}\del_{\alpha}(r-t)\del_{\beta}
\\
& =  2\kappa\omega^{2\kappa-1}\aleph'(r-t)\eta^{\alpha\beta}\del_{\alpha}(r-t)\del_{\beta} 
+ 2\kappa\omega^{2\kappa-1}\aleph'(r-t)H^{\alpha\beta}\del_{\alpha}(r-t)\del_{\beta}.
\endaligned
\ee
Provided that~\eqref{eq2-14-june-2025} holds with a sufficiently small $\eps_s$ (which is implied by our assumptions~\eqref{eq-US-condition}- \eqref{eq-bending-condition}) we have 
\be
(1-H^{00} + (x^a/r)H^{a0}\big)>0, 
\ee
thus $\mathrm{grad}\big(\omega^{2\kappa}\big)$ is \textit{future oriented}. Then it follows that 
\be
\la \Psi,\mathrm{grad}\big(\omega^{2\kappa}\big)\cdot\Psi\ra_{\ourD}\geq0.
\ee
To see this, let us fix a global orthonormal frame $\{e_i\}$ satisfying 
\be
e_0 := \frac{\mathrm{grad}\big(\omega^{2\kappa}\big)}{\big||\mathrm{grad}\big(\omega^{2\kappa}\big)|_g\big|},
\ee
and let $\psi$ be the coordinate of $\Psi$ in this frame. Then we find
\be
\la\Psi,\mathrm{grad}\big(\omega^{2\kappa}\big)\cdot\Psi \ra_{\ourD} 
= \big||\mathrm{grad}\big(\omega^{2\kappa}\big)|_g\big|\psi^{\dag}\gamma_0^{\dag}\gamma_0\psi\geq 0.
\qedhere 
\ee
\ese
\end{proof}


We treat here a setup that is closely linked to the setup introduced in~\cite{PLF-YM-PDE}, i.e., $\kappa\in(1/2,1)$ and $\mu\in(3/4,1)$. In this framework, we need to distinguish between the hyperboloidal domain $\Mcal^{\Hcal}_{[s_0,s_1]}$ and the merging-Euclidean domain $\Mcal^{\EM}_{[s_0,s_1]}$. We thus define
\begin{equation}\label{eq8-04-oct-2025}
\Ebf^{\Hcal,p,k}(s,\Psi) = \int_{\Mcal^{\Hcal}_s}\ebf_{\kappa}^{p,k}[\Psi]\diff x, 
\quad
\Ebf^{\ME,p,k}_{\kappa}(s,\Psi) = \int_{\Mcal^{\ME}_s}\ebf_{\kappa}^{p,k}[\Psi]\diff x.
\end{equation}
Let us point out that $\Ebf^{\Hcal}(s,\Psi)$ remains equivalent when the index of weight $\kappa$ change. So we do not need to specify it.

We also need to consider the integration of the energy flux 
\be
V[\Psi] = \la\Psi,g^{\alpha\beta}\del_{\alpha}\cdot\Psi\ra_{\ourD}\del_{\beta}
\ee
along the conical boundary between $\Mcal^{\Hcal}_{[s_0,s_1]}$ and $\Mcal^{\ME}_{[s_0,s_1]}$ :
\begin{equation}
\Ccal_{[s_0,s_1]} = \Big\{(t,x) \, \big| \, r=t-1, \frac{(s_0^1+1)^{1/2}}{2}\leq t\leq \frac{(s_1^2+1)^{1/2}}{2}\Big\}.
\end{equation}
In this context, by bootstrap we will show that the above hypersurface is spacelike, so that we can decouple the analysis in $\Mcal^{\ME}_{[s_0,s_1]}$ from the global analysis. More precisely, we have the following property. 

\begin{proposition}\label{prop1-05-oct-2025}
Assume that \eqref{eq-USA-condition} holds with a sufficiently small $\eps_s$, and
\begin{equation}
H^{\Ncal00} < 0\quad \text{in }\{r\geq 3t/4\}\cap \Mcal_{[s_0,s_1]}.
\end{equation}  
When $\Ccal_{[s_0,s_1]}$ is spacelike, the following energy estimates hold for any solution to the equation
\begin{equation}
\opDirac\Psi + \mathrm{i}M\Psi = \Phi
\end{equation}
with sufficient decay rate at spatial infinity: 
\begin{equation}\label{eq4-10-oct-2025(l)}
\aligned
& \Ebf_{\kappa}^{\ME}(s_1,u) - \Ebf_{g,\kappa}^{\ME}(s_0,u)
+\int_{s_0}^{s_1}s\int_{\Mcal^{\ME}_s}
\frac{\big(\omega^{\kappa}\zeta|\Psi|_{\vec{\gamma}}\big)^2}{\la r-t\ra}\diff x\diff s
+\Ebf^{\Ccal}(s_0,s_1;\Psi)
\\
& \lesssim \int_{s_0}^{s_1}s\int_{\Mcal_s}\omega^{2\kappa}\zeta^2\big|\la \Phi, \Psi \ra_{\ourD}\big|\, \diff x\diff s  
\lesssim\int_{s_0}^{s_1}s\Ebf_{\kappa}^{\ME}(s,\Psi)^{1/2}
\|\omega^{\kappa}\zeta^2\zetab^{-1/2}|\Phi|_{\vec{n}}\|_{L^2(\Mcal^{\ME}_s)}.
\endaligned
\end{equation}
where $\vec{\gamma} := \mathrm{grad}(r-t)$, $|\Psi|_{\vec{\gamma}}^2 = \la\Psi,\vec{\gamma}\cdot\Psi\ra_{\ourD}$, and
\begin{equation}
\Ebf_{\kappa}^{\ME}(s,\Psi) := \int_{\Mcal^{\ME}_s}\ebf_{\kappa}[\Psi]\diff x,
\quad
\Ebf^{\Ccal}(s_0,s_1;\Psi) := \int_{\Ccal_{[s_0,s_1]}}
\la \Psi,\vec{n}_{\Ccal}\cdot\Psi\ra_{\ourD}\mathrm{Vol}_{\Ccal}\geq 0,
\end{equation}
and $\vec{n}_{\Ccal}$ is the future-oriented unit normal vector of $\Ccal_{[s_0,s_1]}$, while $\mathrm{Vol}_{\Ccal}$ is the associated volume form. Furthermore, one has 
\begin{equation}\label{eq2-10-oct-2025}
\aligned
& \Ebf_{\kappa}^{\ME}(s_1,t^{-1}\Psi) - \Ebf_{\kappa}^{\ME}(s_0,t^{-1}\Psi) + \Ebf^{\Ccal}(s_0,s_1;t^{-1}\Psi) 
\\
& \quad +
\int_{s_0}^{s_1}s\int_{\Mcal^{\ME}_s}t^{-1}\big(\omega^{\kappa}\zeta|t^{-1}\Psi|_{\del_t}\big)^2
+ 
\la r-t \ra^{-1}\big(\omega^{\kappa}\zeta|t^{-1}\Psi|_{\vec{\gamma}}^2\big)^2\,\diff x\diff s
\\
& \lesssim \int_{s_0}^{s_1}s\Ebf_{\kappa}^{\ME}(s,t^{-1}\Psi)^{1/2}
\|t^{-1}\omega^{\kappa}\zeta^2\zetab^{-1/2}|\Phi|_{\vec{n}}\|_{L^2(\Mcal^{\ME}_s)}\diff s
\\
& \quad+\int_{s_0}^{s_1}s\|t^{-1}|H|\|_{L^{\infty}(\Mcal^{\ME}_s)}\Ebf_{\kappa}^{\ME}(s,t^{-1}\Psi) \,\diff s.
\qedhere 
\endaligned
\end{equation}
\end{proposition}

\begin{proof} 
\bse
We proceed as in the proof of Proposition~\ref{prop1-04-april-2025}. We apply Stokes' formula in the domain $\Mcal^{\ME}_{[s_0,s_1]}$. When $\Ccal_{[s_0,s_1]}$ is spacelike, we observe that
$
\la \Psi,\vec{n}_{\Ccal}\cdot\Psi\ra_{\ourD}>0, 
$
which leads us to the positivity of $\Ebf^{\Ccal}(s,\Psi)$. Integrate \eqref{eq1-13-feb-2025} and \eqref{eq9-05-april-2025}, we obtain, similar to \eqref{equa-10mai2025},
\begin{equation}\label{eq2-08-oct-2025}
\aligned
& \int_{\Mcal^{\ME}_s}\hspace{-0.5cm}
\omega^{2\kappa}\la\Psi,\vec{n}\cdot\Psi \ra_{\ourD}\mathrm{Vol}_{\sigmab}
 + \int_{\Ccal_{[s_0,s_1]}}\hspace{-0.8cm}
\omega^{2\kappa}\la \Psi,\vec{n}_{\Ccal}\cdot\Psi\ra_{\ourD}\mathrm{Vol}_{\Ccal}
+ \int_{s_0}^{s_1}\int_{\Mcal^{\ME}_s}\hspace{-0.4cm}
\la\Psi,\mathrm{grad}(\omega^{2\kappa})\cdot\Psi \ra_{\ourD}\,\lapsb\,\mathrm{Vol}_{\sigmab}\,\diff s 
\\
& = \int_{\Mcal^{\ME}_s}\hspace{-0.5cm}
\omega^{2\kappa}\la\Psi,\vec{n}\cdot\Psi \ra_{\ourD}\mathrm{Vol}_{\sigmab} -2\int_{s_0}^{s_1}\int_{\Mcal^{\ME}_s}\hspace{-0.5cm}\omega^{2\kappa}\Re(\la\Phi,\Psi\ra_{\ourD})\,\lapsb\,\mathrm{Vol}_{\sigmab}\,\diff s. 
\endaligned
\end{equation}
Thanks to Claim~\ref{lem1-02-march-2025}, the first term is equivalent to $\Ebf_{\kappa}^{\ME}(s,u)$. For the second term, we observe that $\omega$ is constant along $\Ccal_{[s_0,s_1]}$. For the third term, we observe that 
\begin{equation}\label{eq5-08-oct-2025}
\aligned
& \int_{s_0}^{s_1}\int_{\Mcal^{\ME}_s}
\la\Psi,\mathrm{grad}(\omega^{2\kappa})\cdot\Psi\ra_{\ourD}
\,\lapsb\,\mathrm{Vol}_{\sigmab}\,\diff s 
\\
& = 
\int_{s_0}^{s_1}\int_{\Mcal^{\ME}_s}
\frac{2\kappa}{\omega}\aleph'(r-t)\omega^{2\kappa}\la\Psi,\mathrm{grad}(r-t)\cdot\Psi\ra_{\ourD}
\,\lapsb\,\mathrm{Vol}_{\sigmab}\,\diff s
\\
& = \int_{s_0}^{s_1}\int_{\Mcal^{\ME}_s}
\frac{2\kappa\omega^{2\kappa}}{2+r-t}\la\Psi,\mathrm{grad}(r-t)\cdot\Psi\ra_{\ourD}
\,\lapsb\,\mathrm{Vol}_{\sigmab}\,\diff s
\endaligned
\end{equation}
Recall Claims~\ref{lem1-02-march-2025}, \ref{cor1-16-june-2025} and Lemma~\ref{lem1-21-march-2025}, in $\Mcal^{\ME}_{[s_0,s_1]}$ these are positive universal constants $K_0,K_0',K_1,K_1'$
\begin{equation}\label{eq3-08-oct-2025}
s\zeta^2\diff x \leq K_0'J \diff x\leq\lapsb\,\mathrm{Vol}_{\sigmab}\leq K_1' J \diff x\leq s\zeta^2 \diff x
\end{equation}
in the sens of positive measure.
Thus the third term in the left-hand side of \eqref{eq2-08-oct-2025} is equivalent to
\begin{equation}
\int_{s_0}^{s_1}s\int_{\Mcal^{\ME}_s}\zeta^2\la r-t\ra^{-1}\omega^{2\kappa}\la \Psi,\mathrm{grad}(r-t)\cdot\Psi\ra_{\ourD}\diff x \diff s.
\end{equation}
Also thanks to \eqref{eq3-08-oct-2025}, the last term in the right-hand side of \eqref{eq2-10-oct-2025} can be controlled as
\begin{equation}\label{eq6-11-oct-2025}
\aligned
\Big|\int_{s_0}^{s_1}\int_{\Mcal^{\ME}_s}\hspace{-0.5cm}\omega^{2\kappa}\Re(\la\Phi,\Psi\ra_{\ourD})\,\lapsb\,\mathrm{Vol}_{\sigmab}\,\diff s\Big|
& \lesssim 
\int_{s_0}^{s_1}s\int_{\Mcal^{\ME}_s}\omega^{2\kappa}|\Phi|_{\vec{n}}|\Psi|_{\vec{n}}\zeta^2\diff x\diff s
\\
& \lesssim \int_{s_0}^{s_1}s\|\omega^{\kappa}\zeta^2\zetab^{-1/2}|\Phi|_{\vec{n}}\|_{L^2(\Mcal^{\ME}_s)}
\|\omega^{\kappa}\zetab^{1/2}\Psi\|_{L^2(\Mcal^{\ME}_s)}\diff s
\\
& \lesssim \int_{s_0}^{s_1}\Ebf_{\kappa}^{\ME}(s,\Psi)^{1/2}\|s\omega^{\kappa}\zeta^2\zetab^{-1/2}|\Phi|_{\vec{n}}\|_{L^2(\Mcal^{\ME}_s)}.
\endaligned
\end{equation}


Then we turn our attention to \eqref{eq2-10-oct-2025}. We rely on \eqref{eq4-10-oct-2025(l)}, applying  \eqref{eq5-10-oct-2025}, and emphasize one important issue associated with the decomposition
\be
\aligned
& \int_{s_0}^{s_1}\int_{\Mcal^{\ME}_s}\omega^{2\kappa}t^{-1}
\la t^{-1}\Psi,g^{0\mu}\del_{\mu}\cdot t^{-1}\Psi \ra_{\ourD}
\,\lapsb \, \mathrm{Vol}_{\sigmab}\diff s
\\
& =\int_{s_0}^{s_1}\int_{\Mcal^{\ME}_s}t^{-1}\omega^{2\kappa}
\la t^{-1}\Psi,\del_t\cdot t^{-1}\Psi\ra_{\ourD}
\,\lapsb \, \mathrm{Vol}_{\sigmab}\diff s
\\
& \quad +\int_{s_0}^{s_1}\int_{\Mcal^{\ME}_s}t^{-1}\omega^{2\kappa}
\la t^{-1}\Psi,H^{0\mu}\del_{\mu}\cdot t^{-1}\Psi\ra_{\ourD}
\,\lapsb \, \mathrm{Vol}_{\sigmab}\diff s
=: L_{21} + L_{22}.
\endaligned
\ee
The term $L_{22}$ can be bounded, thanks to Lemma~\ref{lem2-13-july-2025} and \eqref{eq2-10-oct-2025},
as follows: 
\begin{equation}
\aligned
|L_{22}| & \lesssim \int_{s_0}^{s_1}\int_{\Mcal^{\ME}_s}t^{-1}
\omega^{2\kappa}\zeta^{-1}|H||t^{-1}\Psi|_{\vec{n}}^2 J\diff x\diff s
\\
& \lesssim \int_{s_0}^{s_1}s\|t^{-1}|H|\|_{L^{\infty}(\Mcal^{\ME}_s)}\Ebf_{\kappa}^{\ME}(s,t^{-1}\Psi) \,\diff s.
\qedhere 
\endaligned
\end{equation}
\ese
\end{proof}

}

\subsection{ Admissible vector fields, high-order norms and homogeneous coefficients}
\label{section===83}

{ 

\paragraph{Admissible vector fields.}

We now define the Lie algebra of admissible vector fields and introduce a multi-index notation for iterated applications. Namely, we denote by $\mathscr{Z} := \{\del_{\alpha},L_a,\Omega_{ab}|\alpha=0,1,2,3; a,b=1,2,3\}$ the Lie algebra composed of
\be
\del_{\alpha}, \qquad L_a = t\del_a + x^a\del_t, \qquad \Omega_c = x^a\del_b - x^b\del_a
\ee
with $c\neq a, c\neq b$. They are called in the present article the {\sl admissible vector fields} or {\sl admissible operators}. A vector field contained in $\mathscr{Z}$ is generally denoted by $Z$.

Let $I = (\imath_1,\imath_2,\cdots,\imath_p)$ be a multi-index with $\imath_j\in \{0,1,2,\cdots,10\}$. Then we recall the adapted action defined by \eqref{eq3-04-july-2025}, and define its high-order version:
\be
\mathscr{Z}^I \Psi := Z_{\imath_1}Z_{\imath_2}\cdots Z_{\imath_p}\Psi 
= \widehat{Z}_{\imath_1}\widehat{Z}_{\imath_2}\cdots \widehat{Z}_{\imath_p}\Psi.
\ee
Here, $\mathscr{Z}^I$ is a $p$-th-order admissible operator with $Z_\alpha = \del_{\alpha}$ for $\alpha=0,1,2,3$; $Z_{3+a} = L_a$ and $Z_{6+a} = \Omega_a$ for $a=1,2,3$. In the same manner, for vector fields, we also recall \eqref{eq2-04-july-2025} and introduce the high-order adapted action:
\be
\mathscr{Z}^I X := Z_{\imath_1}Z_{\imath_2}\cdots Z_{\imath_p}X = \widehat{Z}_{\imath_1}\widehat{Z}_{\imath_2}\cdots \widehat{Z}_{\imath_p} X.
\ee
We also write the high-order Lie derivative 
\be
\Lcal_{\mathscr{Z}}^I X = \Lcal_{Z_{\imath_1}}\circ\Lcal_{Z_{\imath_2}}\circ\cdots\circ\Lcal_{Z_{\imath_p}}X. 
\ee
Moreover, if a $p$-order operator $\mathscr{Z}^I$ is composed of $(p-k)$ partial derivatives and $k$ boosts/rota\-tions (i.e., $L_a,\Omega_{ab}$), we write
\be
\ord(I) = p,\quad \rank(I) = k
\ee
and $I$ (or $\mathscr{Z}^I$) is said to be of type $(p,k)$. 


\paragraph{High-order norms.}

We now consider high-order norms $|\cdot|_{p,k}$ that distinguish between derivatives and commutators. We have introduced the notation $|u|_{p,k}$ for scalar fields in our earlier work, that is, 
\be
|u|_{p,k} = \max_{\ord(I)\leq p\atop \rank(I)\leq k}|\mathscr{Z}^I u|,
\quad
|u|_p = |u|_{p,p}.
\ee
See Definition~5.1 in \cite{PLF-YM-PDE}. For spinor and vector fields, we also introduce
\begin{equation}
\aligned
&[\Phi]_{p,k} = \max_{\ord(I)\leq p\atop \rank(I)\leq k}|\mathscr{Z}^I\Phi|_{\vec{n}},\quad
&&[X]_{p,k} = \max_{\ord(I)\leq p\atop \rank(I)\leq k}|\mathscr{Z}^IX|_{\vec{n}}
\\
&[\Phi]_p = [\Phi]_{p,p},
\quad
&&[X]_p = [X]_{p,p}. 
\endaligned
\end{equation}
For convenience in the discussion, for vector and tensor fields\footnote{Earlier, we introduced the notation $|H|$ and $|LH|$, which can be understood as zero-order norms of $H_{\alpha\beta}$,$L_{a}H_{\alpha\beta}$.}, we also introduce 
\be
\aligned
&|X|_{p,k} := \max_{\ord(I)\leq p\atop \rank(I)\leq k}\{|\mathscr{Z}^IX^{\alpha}|\},
\quad 
&&|T|_{p,k} :=  \max_{\ord(I)\leq p\atop \rank(I)\leq k}\{|\mathscr{Z}^IT_{\alpha\beta}|\},
\\
&|X|_{p} := |X|_{p,p},
\quad 
&&|T|_{p} :=  |T|_{p,p}.
\endaligned
\ee
Here, $\mathscr{Z}^IX^{\alpha}$ denotes the derivative of $\mathscr{Z}^I$ acting on the component $X^{\alpha}$, which is regarded as a scalar field. We recall that the Lie derivative of the metric is referred to as the deformation tensor. 


In the following discussion, we will frequently write diverse objects as finite linear combinations with functional coefficients which enjoy certain homogeneity. We have introduced these types of coefficients in our previous work in the hyperboloidal region (\cite{PLF-YM-one}) and in the merging-Euclidean region (\cite{PLF-YM-PDE}). Here we recall the following terminology. 

\begin{definition}
A (real-valued) function $f$ defined in (any subset of) $\Mcal_{[s_0,s_1]}$ is said to be homogeneous of degree $k$, if it is smooth, and
\bei

\item[1.] $\quad f(\lambda t,\lambda x) = \lambda^k f(t,x)$,

\item[2.] $\quad |\mathscr{Z}^I f(2,x)|\leq C_I t^{k-l}$ for $|x|\leq 1$ with $l = \rank{\mathscr{Z}^I} - \ord{\mathscr{Z}^I}$, and 

\item[3.]  $\quad |\mathscr{Z}^I f(t,\omega)|\leq C_I r^{k-l}$ for $|\omega|=1,\; 0<t\leq 1$ with $l = \rank{\mathscr{Z}^I} - \ord{\mathscr{Z}^I}$.
\eei

\noindent A smooth function satisfying 1.\ and 2.\ is called \textbf{interior-homogeneous}, while a function satisfying 1.\ and 3.\ is called \textbf{exterior-homogeneous}.
\end{definition}

Typical homogeneous functions are $(x^a/r)$ in the merging-Euclidean region and $(x^a/t)$ in the hyperboloidal region. In general, one can verify directly the following properties.

\begin{lemma}
Let $f,g$ be (interior-/exterior-)homogeneous functions of degree $l,m$. Then
\begin{equation}
|f|_{p,k}\lesssim 
\begin{cases}
t^{l-(p-k)}, \qquad & \text{in } \Mcal^{\Hcal}_{[s_0,s_1]}, \qquad \text{provided that $f$ is interior-homogeneous}, 
\\
r^{l-(p-k)}, \qquad & \text{in } \Mcal^{\ME}_{[s_0,s_1]}, \qquad \text{provided that $f$ is exterior-homogeneous}, 
\\
(r+t)^{l-(p-k)}, \qquad & \text{in } \Mcal_{[s_0,s_1]}, \qquad \text{provided that $f$ is homogeneous}.
\end{cases}
\end{equation}
Furthermore, one has: 
\\
$\bullet\quad$ the function $fg$ is (interior-/exterior-)homogeneous of degree $l+m$, and 
\\
$\bullet\quad$ the function $\alpha f + \beta g$ is (interior-/exterior-)homogeneous of degree $l$, provided that $l=m$ and $\alpha,\beta\in\RR$. 
\end{lemma}

Let $u,v$ be functions defined in $\Mcal_{[s_0,s_1]}$ (or $\Mcal^{\Hcal}_{[s_0,s_1]}, \Mcal^{\ME}_{[s_0,s_1]}$ respectively) and let $f,g$ be (interior-/exterior-)homogeneous functions of degree zero. With our notation we have immediately
\begin{equation}\label{eq4-11-july-2025}
|fu + gv|_{p,k}\lesssim_{p} |u|_{p,k} + |v|_{p,k}.
\end{equation}
This estimate also holds for finite linear combinations with homogeneous coefficients.

}


\subsection{ High-order operators acting on vector fields}
\label{section===84}

{ 

We now study the action of high-order admissible operators on vector fields, and introduce the notion of admissible partition of a multi-index. This allows us to control the application of these operators to products such as $X\cdot\Psi$. Namely, let $I = (\imath_1,\imath_2,\cdots,\imath_p)$ be a multi-index. Then for $I_1 = (\jmath_1,\jmath_2,\cdots,\jmath_{p_1})$ and $I_2 = (k_1,k_2,\cdots, k_{p_2})$ with $p_1+p_2=p$ and there exist two injections, namely 
\be
i_1: \{1,2,\cdots,p_1\}\mapsto \{1,2,\cdots p\}, \qquad i_2:\{1,2,\cdots,p_2\}\mapsto \{1,2,\cdots,p\}
\ee
such that
\bei 

\item  $\mathrm{Im}(i_1)\cap \mathrm{Im}(i_2) = \emptyset$ and $\mathrm{Im}(i_1)\cup \mathrm{Im}(i_2) = \{1,2,\cdots,p\}$,

\item $i_1$ and $i_2$ are order-preserved, i.e., $i_{1,2}(j)<i_{1,2}(k)$ provided that $j<k$,

\item $\jmath_j = \imath_{i_1(j)}$ and $k_l = \imath_{i_2(l)}$.
\eei
\noindent 
That is, we separate the indices $\{1,2,\cdots,p\}$ into two groups, and we arrange them in each group according to the given order. Such a partition leads us to a decomposition of the index $I$. This is called an {\bf admissible partition} of $I$ and we write
\begin{equation}\label{eq9-14-july-2025}
I = I_1\odot I_2.
\end{equation}
Similarly, we define inductively the notion of admissible partition into more than two sub-multi-indices. Then we have the following result, whose proof is immediate and is based on \eqref{eq4-04-july-2025} and an induction on $|I|$.

\begin{lemma}
Let $\mathscr{Z}^I$ be an admissible operator. Then for any vector field $X$ and spinor field $\Psi$ one has 
\begin{equation}
\mathscr{Z}^I(X\cdot\Psi) = \sum_{I_1\odot I_2 = I}\mathscr{Z}^{I_1}X\cdot \mathscr{Z}^{I_2}\Psi, 
\end{equation}
where the sum is taken over all admissible partitions of the multi-index  $I$.
\end{lemma}

Moreover, this type of expression also applies to scalar products involving vectors and spinors. With this notation, we now estimate $\mathscr{Z}^IX$, as follows. 

\begin{lemma}
If $f$ is an interior/exterior-homogeneous function of degree $k$ in $\Mcal^{\Hcal}_{[s_0,s_1]}/ \Mcal^{\ME}_{[s_0,s_1]}$, then one has 
\begin{equation}
|\mathscr{Z}^I(fX)|_{\vec{n}}\lesssim 
\begin{cases}
t^{k}|X|_{p,k}\quad & \text{in }\Mcal^{\Hcal}_{[s_0,s_1]}\text{ and }f\text{ interior-homogeneous},
\\
r^{k}|X|_{p,k}\quad & \text{in }\Mcal^{\ME}_{[s_0,s_1]}\text{ and }f\text{ exterior-homogeneous}.
\end{cases}
\end{equation} 
\end{lemma}

\begin{proof} We treat the domain $\Mcal^{\Hcal}_{[s_0,s_1]}$ and assume that $f$ is interior-homogeneous of degree $k$. We observe that 
$
L_af
$
is also homogeneous of degree $k$, while $\del_{\alpha}f$ is homogeneous of degree $(k-1)$. Then for any admissible operator $\mathscr{Z}^I$, $\mathscr{Z}^I f$ is also homogeneous of degree $\leq k$, thus bounded by $C_It^{-k}$. Thus we find 
\begin{equation}\label{eq1-15-july-2025}
\mathscr{Z}^I(fX) = \sum_{I_1\odot I_2 =I}\mathscr{Z}^{I_1}f\,\mathscr{Z}^{I_2}X, 
\end{equation}
which provides us with the desired result in the interior region. For the merging-Euclidean region, the argument is similar and we omit the details.
\end{proof}

}


\subsection{ Zero-order estimate}
\label{section===85}

{


\begin{lemma}\label{lem1'-18-july-2025}
Assume \eqref{eq-USA-condition} holds for a sufficiently small $\eps_s$. Consider any multi-index $I$ with elements in $\{0,1,2,3\}$, and set 
\be
\aligned
& \Big(\prod_I\del\Big)\cdot\Psi := \del_{\jmath_1}\cdot\del_{\jmath_2}\cdots\cdot\del_{\jmath_p}\cdot\Psi,\quad \text{with}\quad \jmath=0,1,2,3;
\\
& \Big(\prod_I\delb\Big)\cdot\Psi := \delb_{\imath_1}\cdot\delb_{\imath_2}\cdots\cdot\delb_{\imath_p}\cdot\Psi,\quad \text{with}\quad \imath=1,2,3.
\endaligned
\ee
Then one has 
\begin{equation}\label{eq17'-18-july-2025}
\Big|\Big(\prod_I\del\Big)\cdot\Psi\Big|_{\vec{n}}\lesssim_p \zetab^{-1}|\Psi|_{\vec{n}}, 
\end{equation}
\begin{equation}
\Big|\Big(\prod_J\delb\Big)\cdot\Psi\Big|_{\vec{n}}\lesssim_p |\Psi|_{\vec{n}}.
\end{equation}
\end{lemma}

\begin{proof}
We will establish the decomposition
\begin{equation}\label{eq1-23-july-2025}
\prod_I\del = J^{-1}\sum_{|J|\leq p}\Theta^J_I\delb_0\cdot\prod_J\delb 
+ \sum_{|J|\leq p}\Gamma^J_I\prod_J\delb
\end{equation}
where $\Theta^J_I$ and $\Gamma^J_I$ are uniformly bounded functions defined in $\Mcal_{[s_0,s_1]}$. This is by induction on $p=|I|$. When $|I|=1$, we remark that
\be
\del_t = J^{-1}\delb_0,\quad \del_a = - J^{-1}\frac{x^a}{r}\delb_rT \delb_0 + \delb_a.
\ee
Then let us consider the product
$
\del_{\alpha}\cdot\prod_{I}\del.
$
When $\del_{\alpha} = \del_t = J^{-1}\delb_0$, we find 
\be
\aligned
\del_{\alpha}\cdot\prod_{I}\del & =  J^{-1}\delb_0\cdot\bigg(J^{-1}\sum_{|J|\leq p}\Theta^J_I\delb_0\cdot\prod_J\delb + \sum_{|J|\leq p}\Gamma^J_I\prod_J\delb\bigg)
\\
& = J^{-2}\gb_{00}\sum_{|J|\leq p}\Theta_I^J\prod_J\delb
+ J^{-1}\sum_{|J|\leq p}\Gamma^J_I\delb_0\cdot\prod_J\delb.
\endaligned
\ee
We observe that
\be
\gb_{00} = \bar{\eta}_{00} + \Hb_{00} = -J^2 + \Hb_{00},
\ee
and
\be
\Hb_{00} = H(\delb_0,\delb_0) = J^2H_{00}.
\ee
By recalling \eqref{eq2-14-june-2025}, the induction is completed.

Next, we turn our attention to the proof of \eqref{eq17'-18-july-2025}. Based on \eqref{eq1-23-july-2025},  Lemma~\ref{lem1-23-july-2025} and Corollary~\ref{cor1-18-july-2025}, we have 
\be
\Big|\Big(\prod_J\delb\Big)\cdot\Psi\Big|_{\vec{n}}\lesssim_p |\Psi|_{\vec{n}}.
\ee
On the other hand, recalling \eqref{eq2-17-july-2025} we have 
\be
\Big|\delb_0\cdot\Big(\prod_J\delb\Big)\cdot\Psi\Big|_{\vec{n}}\lesssim_p J\zetab^{-1}|\Psi|_{\vec{n}}.
\ee
Since $0<\zeta\leq 1$, we obtain the desired estimate.
\end{proof}


We then list several estimates (which will be frequently used), that are direct consequences of the above lemmas.

\begin{lemma}
\label{lem1-01-aout-2025}
Assume that \eqref{eq-USA-condition} holds with a sufficiently small $\eps_s$. Then one has 
\begin{equation}\label{eq1-31-july-2025}
|\del_{\mu}\cdot\nabla_{\nu}\del_{\alpha}\cdot\Psi|_{\vec{n}}\lesssim \zetab^{-1}|\del H||\Psi|_{\vec{n}}.
\end{equation}
\begin{equation}\label{eq2-31-july-2025}
|\del_{\mu}\cdot\nabla_{\nu}L_a\cdot\Psi|_{\vec{n}}
\lesssim \zetab^{-1}\big(1+(t+r)|\del H|\big)|\Psi|_{\vec{n}}.
\end{equation}
\begin{equation}\label{eq3-31-july-2025}
|\del_{\mu}\cdot\nabla_{\nu}L_a\cdot\del_{\mu'}\cdot \nabla_{\nu'} \del_{\alpha}\cdot\Psi|_{\vec{n}}\lesssim \zetab^{-1}\big(|\del H| + (t+r)|\del H|^2\big)|\Psi|_{\vec{n}},
\end{equation}
\begin{equation}\label{eq4-31-july-2025}
|\del_{\mu}\cdot\nabla_{\nu}L_b\cdot \del_{\mu'}\cdot\nabla_{\nu'}L_a\cdot\Psi|_{\vec{n}}
\lesssim \zetab^{-1}(1+(t+r)^2|\del H|^2)|\Psi|_{\vec{n}}.
\end{equation}
\begin{equation}\label{eq6-31-july-2025}
\big|\nabla_{\alpha}\big(g^{\mu\nu}\del_{\mu}\cdot\nabla_{\nu}\del_{\beta}\big)\cdot\Psi\big|_{\vec{n}}
\lesssim \zetab^{-1}\big(|\del\del H| + |\del H| + |\del H|^2\big)|\Psi|_{\vec{n}}.
\end{equation}
\begin{equation}\label{eq7-31-july-2025}
\big|\nabla_{\alpha}\big(g^{\mu\nu}\del_{\mu}\cdot\nabla_{\nu}L_b\big)\cdot\Psi\big|_{\vec{n}}
\lesssim \zetab^{-1}|\del H||\Psi|_{\vec{n}} + \zetab^{-1}(t+r)\big( |\del\del H| + |\del H|^2\big)|\Psi|_{\vec{n}}.
\end{equation}
\end{lemma}

\begin{proof} 
\bse
We deal first with  \eqref{eq1-31-july-2025} and define 
$
\Phi_{\nu\alpha}^{\gamma} := \Gamma_{\nu\alpha}^{\gamma}\Psi$, 
where $\Gamma_{\nu\alpha}^{\gamma}\del_{\gamma} = \nabla_{\nu}\del_{\alpha}$ are the Christoffel symbols. Then we have 
\be
\del_{\mu}\cdot\nabla_{\nu}\del_{\alpha}\cdot \Psi = \del_{\mu}\cdot\del_{\gamma}\cdot\Phi_{\nu\alpha}^{\gamma}.
\ee
By Lemma~\ref{lem1'-18-july-2025}, we deduce that 
\be
\big|\del_{\mu}\cdot\nabla_{\nu}\del_{\alpha}\cdot \Psi\big|_{\vec{n}}
\lesssim \zetab^{-1}\big|\Phi_{\nu\alpha}^{\gamma}\big|_{\vec{n}}
\lesssim \zetab^{-1}|\Gamma_{\nu\alpha}^{\gamma}||\Phi|_{\vec{n}}
\ee
and we observe that
\be
\Gamma_{\nu\alpha}^{\gamma} 
= \frac{1}{2}g^{\gamma\delta}\big(\del_{\nu}g_{\alpha\delta} + \del_{\alpha}g_{\nu\delta} - \del_{\delta}g_{\nu\alpha}\big)
\ee
and $|g^{\gamma\delta}|\lesssim 1$ (provided that \eqref{eq2-16-june-2025} holds with a sufficiently small $\eps_s$). Thus we find 
\begin{equation}\label{eq5-31-july-2025}
|\Gamma_{\alpha\beta}^{\gamma}|\lesssim |\del H|, 
\end{equation} 
which is the desired estimate.
\ese
\bse
The proof of \eqref{eq2-31-july-2025}, \eqref{eq3-31-july-2025} and \eqref{eq4-31-july-2025} are similar. We estimate the (products of) the Christoffel symbols and then apply Lemma~\ref{lem1'-18-july-2025}. For $\nabla_{\nu}L_a$, we note that
\begin{equation}\label{eq9-22-july-2025}
\nabla_{\nu}L_a = \nabla_{\nu}(t\del_a + x^a\del_t) 
= \delta_{\nu}^0\del_a + \delta_{\nu}^a\del_t + t\nabla_{\nu}\del_a + x^a\nabla_{\nu}\del_t.
\end{equation}
Thus we have 
\begin{equation}\label{eq3-30-july-2025}
|X^{\alpha}|\lesssim 1 + (t+r)|\del H|,\quad \text{with}\quad \nabla_{\nu}L_a = X^{\alpha}\del_{\alpha}, 
\end{equation}
which establishes \eqref{eq2-31-july-2025}.

The last two estimates are somewhat different. For any vector field $X = X^{\alpha}\del_{\alpha}$, thanks to Lemma~\ref{lem1-24-feb-2025} and the fact that $\nabla_{\alpha}g\equiv 0$, we have 
\begin{equation}
\nabla_{\alpha}\big(g^{\mu\nu}\del_{\mu}\cdot\nabla_{\nu}X\big) 
= g^{\mu\nu}\del_{\mu}\cdot\nabla_{\alpha}\nabla_{\nu}X
= g^{\mu\nu}\del_{\mu}\cdot\nabla_{\alpha}\big(\nabla_{\nu}X\big) - g^{\mu\nu}\del_{\mu}\cdot\nabla_{\nabla_{\alpha}\del_{\nu}}X.
\end{equation}
The above calculation leads us to
\begin{equation}\label{eq16-16-aout-2025}
\aligned
\nabla_{\alpha}\big(g^{\mu\nu}\del_{\mu}\cdot\nabla_{\nu}X\big) 
& =  g^{\mu\nu}\big(\del_{\alpha}\del_{\nu}X^{\beta} 
+ \del_{\nu}X^{\gamma}\Gamma_{\alpha\gamma}^{\beta} 
+ \del_{\alpha}(X^{\gamma}\Gamma_{\nu\gamma}^{\beta})
+ X^{\gamma}\Gamma_{\nu\beta}^{\delta}\Gamma_{\alpha\delta}^{\beta}\big)\del_{\mu}\cdot\del_{\beta}
\\
& \quad - g^{\mu\nu}\Gamma_{\alpha\nu}^{\gamma}\del_{\gamma}X^{\beta} \del_{\mu}\cdot\del_{\beta}.
\endaligned
\end{equation}
Now for $X = \del_{\beta}$ or $X = L_a$, we have $\del_{\alpha}\del_{\gamma}X^{\beta} = 0$. We thus arrive at \eqref{eq6-31-july-2025}, and \eqref{eq7-31-july-2025}
\ese
\end{proof}

}


\section{Klainerman-Sobolev inequalities for spinor fields on Eucli\-dean-hyper\-boloidal slices}
\label{section=N6}

\subsection{ Objective and main result}
\label{section===91}

{

\paragraph{Objective for this section.}

The goal of this section is to establish Sobolev-type inequalities for spinor fields on a curved spacetime foliated by Euclidean-hyperboloidal slices. These estimates are particularly tailored to the geometry of the foliation under consideration, and will be used for controlling the pointwise behavior of spinor fields in terms of weighted $L^2$ norms of suitable derivatives. The originality of our approach lies in adapting the classical vector field method to \emph{gauge-independent} spinorial quantities on a Euclidean-hyperboloidal background satisfying the \emph{uniform spacelike condition}. 

Indeed, our results and derivation concern the gauge-independent contraction $\la\Psi,\vec{n}\cdot\Psi\ra_{\ourD}$ and the analysis of its derivatives with respect to admissible vector fields. A crucial technical challenge for dealing with the curved setting is the derivation of estimates for the geometric error terms induced by commutators and deformation tensors, as defined in Proposition~\ref{prop1-08-june-2025}, below. We compute and rely on geometric identities involving spinor covariant derivatives together with curvature-related and foliation-related contributions, in order to isolate and control these error terms. The resulting inequalities (cf.~Theorem~\ref{thm2-05-april-2025} stated next) provide sharp pointwise bounds and will serve as a fundamental ingredient in our subsequent nonlinear analysis of the Einstein equations. In this section, we thus combine a series of techniques (geometric arguments, regularization, localization) and emphasize that our treatment does not rely on exact Killing symmetries of the underlying spacetime, which are absent in a general spacetime.


\paragraph{Main result for this section.}

We use the notation in Section~\ref{section-1--4}, and work with the previously constructed parameterization of the Euclidean-hyperboloidal slices and we also consider the natural frame $\{\delb_0,\delb_a\}$. Recall that $\vec{n}$ denotes the future oriented unit normal vector along the slices $\mathcal{M}_s$. Recall also that $\widehat{X}$ denotes the Clifford-corrected derivative  \eqref{equa-511m} of any admissible vector field $X$. 

\begin{theorem}[Klainerman-Sobolev inequalities for spinor fields on a curved spacetime]
\label{thm2-05-april-2025}Suppose that uniform spacelike \eqref{eq-USA-condition} holds for some sufficiently small $\eps_s$. Consider any sufficiently regular spinor field $\Psi$ defined in the spacetime slab $\mathcal{M}_{[s_0,s_1]}$. 
\\
\noindent$\bullet$ \textbf{Hyperboloidal domain.} Provided the metric satisfies the decay condition
\begin{equation}\label{eq5a-05-april-2025}
(s/t)^{-1}|H|_1 + \frac{(s/t)^{-2}|H^{\Ncal00}|_1}{1+(s/t)^{-2}|H^{\Ncal00}|}\lesssim (s/t)^{- \delta}
\quad
\text{ in } \Mcal^{\Hcal}_{[s_0,s_1]} \quad \text{ for some } 0\leq \delta\leq 1, 
\end{equation}
one has  
\begin{equation}\label{eq9-16-june-2025}
(s/t)^{1/2+2\delta}t^{3/2}|\Psi|_{\vec{n}}\lesssim \sum_{|J|\leq 2} \bigl\|(s/t)^{1/2}|\mathscr{L}^J\Psi|_{\vec{n}} \bigr\|_{L^2(\mathcal{M}^{\Hcal}_s)}
\quad \text{ in } \Mcal^{\Hcal}_{[s_0,s_1]}.
\end{equation}

\noindent$\bullet$ \textbf{Merging-Euclidean domain.}
Provided  the metric satisfies the decay condition 
\begin{equation}\label{eq6-17-aout-2025}
\zeta^{-2}|\del H| + \zeta^{-1}|H|_1 + \frac{\zeta^{-2}|H^{\Ncal00}|_1}{1+\zeta^{-2}|H^{\Ncal00}|}\lesssim \zeta^{- \delta} \quad \text{ in } \Mcal^{\ME}_{[s_0,s_1]}
\quad \text{ for some } 0 < \delta\leq 1, 
\end{equation} 
one has, for $0\leq \kappa\leq 1$,
\begin{equation}\label{eq5-18-aout-2025}
\zeta^{1/2+2\delta}(2+r-t)^{\kappa}(1+r)|\Psi|_{\vec{n}}\lesssim_{\delta} \sum_{j+|K|\leq 2}\|(2+r-t)^{\kappa}\zeta^{1/2}|\widehat{\del}^j\widehat{\Omega}^K\Psi|_{\vec{n}}\|_{L^2(\mathcal{M}^{\ME}_s)}
\quad \text{ in } \Mcal^{\ME}_{[s_0,s_1]}.
\end{equation}
\noindent$\bullet$ \textbf{Euclidean domain.}
Provided \eqref{eq6-17-aout-2025} holds, one has 
\begin{equation}\label{eq11-19-aout-2025}
(2+r-t)^{\kappa}(1+r)|\Psi|_{\vec{n}}\lesssim \sum_{j+|K|\leq 2}\|(2+r-t)^{\kappa}|\widehat{\del}^j\widehat{\Omega}^K\Psi|_{\vec{n}}\|_{L^2(\mathcal{M}^{\Ecal}_s)} 
\quad \text{ in } \Mcal^{\Ecal}_{[s_0,s_1]}.
\end{equation}
\end{theorem}

The remaining of this section is devoted to the proof of the above estimates. Our strategy is to extend standard Klainerman-Sobolev inequalities to the \emph{gauge-invariant contraction} $\la \Psi,\vec{n}\cdot\Psi\ra_{\ourD}$ defined on a \emph{curved background}. The proof consists of two main steps. First, we compute the derivatives of $\la \Psi,\vec{n}\cdot\Psi\ra_{\ourD}$ with respect to the admissible vector fields $L_a,\Omega_{ab}$ and $\delb_r$. In a second step, we estimate the geometric error terms, and we find it necessary to introduce a regularization method before applying the Klainerman-Sobolev inequalities available for scalar-valued functions.

}


\subsection{ Derivatives of the contraction $\la\Psi,\vec{n}\cdot\Psi\ra_{\ourD}$}
\label{section===92}

{ 

\paragraph{An identity.}

Our first task is the study of the function $\la\Psi,\vec{n}\cdot\Psi\ra_{\ourD}$ associated with an arbitrary spinor field. We begin with a technical identity involving the Clifford-corrected derivative \eqref{equa-511m}, namely
$\widehat{X}\Psi = \nabla_X\Psi - \frac{1}{4}g^{\alpha\gamma} \del_\alpha \cdot\nabla_{\gamma}X\cdot\Psi$.

\begin{lemma}
\label{lemma-19juillet2025-a}
For any sufficient regular spinor field $\Psi$ and any vector field $X$ defined in (a subset of) $\Mcal_{[s_0,s_1]}$, one has 
\begin{equation}\label{eq3-06-june-2025}
X \bigl( \la\Phi,\Psi\ra_{\ourD} \bigr)
= \la\widehat{X}\Phi,\Psi\ra_{\ourD} + \la\Phi,\widehat{X}\Psi\ra_{\ourD} 
- \frac{1}{2}\mathrm{div}(X)\la\Phi,\Psi\ra_{\ourD}.
\end{equation}
\end{lemma}

\begin{proof} By a direct calculation, we have 
$$
\aligned
X \bigl( \la \Phi,\Psi\ra_{\ourD} \bigr)  
& =  \la\nabla_X\Phi,\Psi\ra_{\ourD} + \la\Phi,\nabla_X\Psi\ra_{\ourD}
\\
& =  \la\widehat{X}\Phi,\Psi \ra_{\ourD} + \la\Phi,\widehat{X}\Psi\ra_{\ourD} 
+ \frac{1}{4}g^{\alpha\gamma}\la\del_{\alpha}\cdot\nabla_{\gamma}X\cdot\Phi,\Psi \ra_{\ourD}
+ \frac{1}{4}g^{\alpha\gamma}\la\Phi,\del_{\alpha}\cdot\nabla_{\gamma}X\cdot\Psi \ra_{\ourD}
\\
& =  \la\widehat{X}\Phi,\Psi \ra_{\ourD} + \la\Phi,\widehat{X}\Psi\ra_{\ourD} 
+ \frac{1}{4}g^{\alpha\gamma}\big\la\Phi,(\nabla_{\gamma}X\cdot\del_{\alpha} + \del_{\alpha}\cdot\nabla_{\gamma}X)\cdot\Psi \big\ra_{\ourD}
\\
& =  \la\widehat{X}\Phi,\Psi \ra_{\ourD} + \la\Phi,\widehat{X}\Psi\ra_{\ourD} 
- \frac{1}{2}g^{\alpha\gamma}(\del_{\alpha},\nabla_{\gamma}X)_g\la \Phi,\Psi\ra_{\ourD}, 
\endaligned
$$
in which $g^{\alpha\gamma}(\del_\alpha,\nabla_\gamma X)_g = \mathrm{div}(X)$. 
\end{proof}

In particular, in the flat case, the boosts and rotations are divergence-free, that is, 
\begin{equation}
\mathrm{div}_\eta(L_a) = \mathrm{div}_\eta(\Omega_{ab}) = 0 \qquad \text{(flat geometry)}, 
\end{equation}
and we thus have the following result.

\begin{lemma}\label{lem1-09-june-2025}
Suppose that $g=\eta$. 
For the boost vector fields, one has 
\begin{equation}\label{eq1-18-aout-2025}
\widehat{L_a}\Psi = \nabla_{L_a}\Psi + \frac{1}{2}\del_0\cdot\del_a\cdot\Psi 
\qquad \text{(flat geometry)}, 
\end{equation}
\begin{equation}\label{eq7-27-july-2025}
\widehat{\Omega_{ab}}\Psi = \nabla_{\Omega_{ab}}\Psi - \frac{1}{2}\del_a\cdot\del_b\cdot\Psi 
\qquad \text{(flat geometry)}, 
\end{equation}
together with 
\begin{equation}\label{eq3-09-june-2025}
L_a \bigl( \la \Psi,\vec{n}\cdot\Psi\ra_{\ourD} \bigr) 
= \la \widehat{L_a}\Psi,\vec{n}\cdot\Psi\ra_{\ourD} + \la \Psi,\vec{n}\cdot\widehat{L_a}\Psi\ra_{\ourD}
\qquad \text{(flat geometry)}, 
\end{equation}
\begin{equation}\label{eq1-28-july-2025}
\Omega_{ab} \bigl( \la \Psi,\vec{n}\cdot\Psi\ra_{\ourD} \bigr) 
= \la\widehat{\Omega_{ab}}\Psi,\vec{n}\cdot\Psi \ra_{\ourD}
+\la\Psi,\vec{n}\cdot\widehat{\Omega_{ab}}\Psi\ra_{\ourD}
\qquad \text{(flat geometry)}. 
\end{equation}
\end{lemma}

\begin{proof}
\bse
In flat spacetime, the Levi-Civita connection reduces to partial derivatives, namely $\nabla_X = X^\mu \del_\mu$. The identity \eqref{eq1-18-aout-2025} is checked by a direct calculation. Then thanks to Lemma~\ref{lemma-19juillet2025-a}, we have 
\be
\aligned
L_a\big(\la \Psi,\vec{n}\cdot\Psi\ra_{\ourD}\big)
& =  \la\nabla_{L_a}\Psi,\vec{n}\cdot\Psi \ra_{\ourD} 
+ \la \Psi,\nabla_{L_a}\vec{n}\cdot\Psi\ra_{\ourD}
+ \la \Psi,\vec{n}\cdot\nabla_{L_a}\Psi\ra_{\ourD}
\\
& =  \la\nabla_{L_a}\Psi,\vec{n}\cdot\Psi \ra_{\ourD} 
+ (t/s)\la \Psi,\delb_a\cdot\Psi\ra_{\ourD}
+ \la \Psi,\vec{n}\cdot\nabla_{L_a}\Psi\ra_{\ourD}
\\
& =  \la \nabla_{L_a}\Psi,\vec{n}\cdot\Psi\ra_{\ourD} +\frac{1}{2}\la\Psi,\del_a\cdot\del_t\cdot\vec{n}\cdot\Psi\ra_{\ourD}
\\
& \quad+ \la \Psi,\vec{n}\cdot\nabla_{L_a}\Psi\ra_{\ourD} 
+ \frac{1}{2}\la\Psi,\vec{n}\cdot\del_t\cdot\del_a\cdot\Psi\ra_{\ourD}
\\
& 
=\la\widehat{L_a}\Psi,\vec{n}\cdot\Psi \ra_{\ourD} 
+ \la\Psi,\vec{n}\cdot\widehat{L_a}\Psi \ra_{\ourD}.
\endaligned
\ee
Here, for the second equality we used $\nabla_{L_a}\vec{n} = (t/s)\delb_a$, while for the third equality we applied the identity
\be
\aligned
\vec{n}\cdot\del_t\cdot\del_a + \del_a\cdot\del_t\cdot\vec{n}
& =  - \del_t\cdot\vec{n}\cdot\del_a + 2(t/s)\del_a +  \del_a\cdot\del_t\cdot\vec{n}
\\
& =  \del_t\cdot\del_a\cdot\vec{n} + 2(x^a/s)\del_t + 2(t/s)\del_a +  \del_a\cdot\del_t\cdot\vec{n}
\\
& = 2(t/s)\big(\del_a + (x^a/t)\del_t\big) = 2(t/s)\delb_a.
\endaligned
\ee
The calculation concerning the rotation fields $\Omega_{ab}$ is similar. Finally, we also have the identity
\begin{equation}
\nabla_{\Omega_{ab}}\vec{n} = s^{-1}\Omega_{ab},\quad\del_b\cdot\del_a\cdot\vec{n} + \vec{n}\cdot\del_a\cdot\del_b = -2s^{-1}\Omega_{ab}.
\end{equation}
\ese
\end{proof}
\begin{remark}
As far as we are concerned, the relation \eqref{eq3-09-june-2025} is firstly remarked in \cite[Lemma~3.2]{DW} in the setting of $\RR^{1+2}$ with the canonical tetrad $\{\del_{\alpha}\}$. 
\end{remark}


\paragraph{Geometric error and error vector.}

In the flat case and in the hyperboloidal region, the identity \eqref{eq3-09-june-2025} is sufficient in order to establish Theorem~\ref{thm2-05-april-2025}. However, for the curved case, a further observation is required, as follows. Here, we find it convenient to refer to the underlined term $\la\Psi,\vec{w}\cdot\Psi \ra_{\ourD}$ in \eqref{eq10-30-june-2025}, below, as the \textit{ geometric error} and to the vector $\vec{w}[X]$ in \eqref{eq2-26-june-2025}, below, as the \textit{error vector}.

\begin{proposition}\label{prop1-08-june-2025}
Let $\vec{n}$ be the future oriented normal vector to a slice $\Mcal_s$ and $X$ be any tangent vector. Then one has 
\begin{equation}\label{eq10-30-june-2025}
X \bigl( \la\Psi,\vec{n}\cdot\Psi\ra_{\ourD} \bigr) 
= \la\widehat{X}\Psi,\vec{n}\cdot\Psi\ra_{\ourD} + \la\Psi,\vec{n}\cdot\widehat{X}\Psi\ra_{\ourD}
+ \underline{\la\Psi,\vec{w}[X]\cdot\Psi \ra_{\ourD}}, 
\end{equation} 
the vector $\vec{w}[X]$ being defined as 
\begin{equation}\label{eq2-26-june-2025}
\vec{w}[X] = \frac{1}{2}([X,\vec{n}])^{\top}
- \frac{1}{4}g^{\alpha\beta}\pi[X]_{\alpha\beta} \, \vec{n}, 
\end{equation}
where $(X)^{\top}$ denotes the orthogonal component of the orthogonal decomposition of $X$ with respect to $\vec{n}$ (cf.~\eqref{eq4-27-june-2025}), and $\pi[X]$ is the deformation tensor of $X$.  
\end{proposition}

\begin{proof} {\bf 1. A decomposition.} 
\bse
Thanks to the identity \eqref{eq3-06-june-2025} for a general contraction $X \bigl( \la\Phi,\Psi\ra_{\ourD} \bigr)$, we can write 
\be
\aligned
X \bigl( \la \Psi,\vec{n}\cdot\Psi\ra_{\ourD} \bigr) 
& =  \la \widehat{X}\Psi,\vec{n}\cdot\Psi\ra_{\ourD} 
+ \big\la\Psi,\widehat{X}(\vec{n}\cdot\Psi)\big\ra_{\ourD}
- \frac{1}{2}\mathrm{div}(X)\la\Psi,\vec{n}\cdot\Psi\ra_{\ourD}
\\
& = \la \widehat{X}\Psi,\vec{n}\cdot\Psi\ra_{\ourD} 
+ \big\la\Psi,\nabla_X(\vec{n}\cdot\Psi)\big\ra_{\ourD}
- \frac{1}{4}g^{\alpha\gamma}\la\Psi,\del_{\alpha}\cdot\nabla_{\gamma}X\cdot\vec{n}\cdot\Psi\ra_{\ourD}
\\
& \quad 
- \frac{1}{2}\mathrm{div}(X)\la\Psi,\vec{n}\cdot\Psi\ra_{\ourD}
\\
& = \la\widehat{X}\Psi,\vec{n}\cdot\Psi\ra_{\ourD}
+ \la\Psi,\nabla_X\vec{n}\cdot\Psi\ra_{\ourD}
+ \la\Psi,\vec{n}\cdot\nabla_X\Psi\ra_{\ourD}
\\
& \quad - \frac{1}{2}\mathrm{div}(X)\la\Psi,\vec{n}\cdot\Psi\ra_{\ourD}
- \frac{1}{4}g^{\alpha\gamma}\la\Psi,\del_{\alpha}\cdot\nabla_{\gamma}X\cdot\vec{n}\cdot\Psi\ra_{\ourD}, 
\endaligned
\ee
therefore
\be
\aligned
X\la \Psi,\vec{n}\cdot\Psi\ra_{\ourD} 
& = \la\widehat{X}\Psi,\vec{n}\cdot\Psi\ra_{\ourD}
+ \la\Psi,\vec{n}\cdot\widehat{X}\Psi\ra_{\ourD}
- \frac{1}{2}\mathrm{div}(X)\la\Psi,\vec{n}\cdot\Psi\ra_{\ourD}
\\
& \quad + \la\Psi,\nabla_X\vec{n}\cdot\Psi\ra_{\ourD}
+\frac{1}{4}g^{\alpha\gamma}\la\Psi,\vec{n}\cdot\del_{\alpha}\cdot\nabla_{\gamma}X\cdot\Psi\ra_{\ourD}
- \frac{1}{4}g^{\alpha\gamma}\la\Psi,\del_{\alpha}\cdot\nabla_{\gamma}X\cdot\vec{n}\cdot\Psi\ra_{\ourD}. 
\endaligned
\ee
Then, we recall that
\be
\aligned
- \frac{1}{4}g^{\alpha\gamma}\la\Psi,\del_{\alpha}\cdot\nabla_{\gamma}X\cdot\vec{n}\cdot\Psi\ra_{\ourD}
& = 
\frac{1}{4}g^{\alpha\gamma}\la\Psi,\nabla_{\gamma}X\cdot\del_{\alpha}\cdot\vec{n}\cdot\Psi\ra_{\ourD}
+ \frac{1}{2}g^{\alpha\gamma}(\del_{\alpha},\nabla_{\gamma}X)_g\la\Psi,\vec{n}\cdot\Psi\ra_{\ourD}
\\
& =  \frac{1}{4}g^{\alpha\gamma}\la\Psi,\nabla_{\gamma}X\cdot\del_{\alpha}\cdot\vec{n}\cdot\Psi\ra_{\ourD}
+ \frac{1}{2}\mathrm{div}(X)\la\Psi,\vec{n}\cdot\Psi\ra_{\ourD}.
\endaligned
\ee
Thus by introducing the vector 
\begin{equation}\label{eq1-26-june-2025}
\vec{w} = \nabla_X\vec{n} + \frac{1}{4}g^{\alpha\gamma}\vec{n}\cdot\del_{\alpha}\cdot\nabla_{\gamma}X
+\frac{1}{4}g^{\alpha\gamma}\nabla_{\gamma}X\cdot\del_{\alpha}\cdot\vec{n}, 
\end{equation}
we arrive at  
\begin{equation}
X\la \Psi,\vec{n}\cdot\Psi\ra_{\ourD} 
= \la\widehat{X}\Psi,\vec{n}\cdot\Psi\ra_{\ourD}
+ \la\Psi,\vec{n}\cdot\widehat{X}\Psi\ra_{\ourD}
+ \la\Psi,\vec{w}\cdot\Psi\ra_{\ourD}.
\end{equation}
In the expression \eqref{eq1-26-june-2025}, it remains to compute 
\be
\vec{n}\cdot\del_{\alpha}\cdot\nabla_{\gamma}X + \nabla_{\gamma}X\cdot\del_{\alpha}\cdot\vec{n}
= 2(\vec{n},\nabla_{\gamma}X)_g\del_{\alpha} 
- 2(\vec{n},\del_{\alpha})_g\nabla_{\gamma}X
- 2(\del_{\alpha},\nabla_{\gamma}X)_g\vec{n}, 
\ee
which leads us to the following expression of $\vec{w}$: 
\begin{equation}\label{eq1-06-june-2025}
\aligned
\vec{w} 
& = \nabla_X\vec{n} + \frac{1}{2}g^{\alpha\gamma}(\vec{n},\nabla_{\gamma}X)_g\del_{\alpha}
- \frac{1}{2}g^{\alpha\gamma}(\vec{n},\del_{\alpha})_g\nabla_{\gamma}X
- \frac{1}{2}g^{\alpha\gamma}(\del_{\alpha},\nabla_{\gamma}X)_g\vec{n}
\\
& =: \nabla_X\vec{n} + A_2 + A_3 + A_4. 
\endaligned
\end{equation}
\ese

\vskip.3cm

\noindent{\bf 2. Analysis of the terms $A_2,A_3,A_4$.} 
\bse
Next, we analyze the second term $A_2$ in the right-hand side of \eqref{eq1-06-june-2025} and work with the tangent-orthogonal frame (cf. \eqref{eq4-12-june-2025}). The components of $g$ in the frame $\{w_{\alpha}\}$ are denoted by $\gch_{\alpha\beta}$ and $\gch^{\alpha\beta}$ with
\begin{equation}\label{eq1-11-june-2025}
\gch^{00} = \gch_{00} = -1, \qquad \gch_{0a} = \gch^{0a} = 0, \qquad \gch^{ab} = \sigmab^{ab}, \qquad \gch_{ab} = \gb_{ab}.
\end{equation}
We thus find 
$$ 
\aligned
A_2 
& = 
\frac{1}{2}\gch^{\alpha\gamma}(\vec{n},\nabla_{\check{\del}_{\gamma}}X)_gw_{\alpha}
=
\frac{1}{2}\gch^{00}(\vec{n},\nabla_{\vec{n}}X)_g\vec{n} 
+ \frac{1}{2}\sigmab^{ac}(\vec{n},\nabla_{\delb_c}X)_g\delb_a
\\
& = \frac{1}{2}\gch^{00}(\vec{n},\nabla_{\vec{n}}X)_g\vec{n}
+ \frac{1}{2}\sigmab^{ac}(\vec{n},\nabla_X\delb_c)_g\delb_a
+ \frac{1}{2}\sigmab^{ac}(\vec{n},[\delb_c,X])\delb_a
\\
& =  \frac{1}{2}\gch^{00}(\vec{n},\nabla_{\vec{n}}X)_g\vec{n}
- \frac{1}{2}\sigmab^{ac}(\nabla_X\vec{n},\delb_c)_g\delb_a
+ 0, 
\endaligned
$$
therefore 
\begin{equation}\label{eq1-09-june-2025}
\aligned
A_2
& =  \frac{1}{2}\gch^{00}(\vec{n},\nabla_{\vec{n}}X)_g\vec{n}
- \frac{1}{2}\sigmab^{ac}(\nabla_X\vec{n},\delb_c)_g\delb_a
=: A_{2,1} + A_{2,2}. 
\endaligned
\end{equation}
For the first term $A_{2,1}$ above, we observe that
\be
A_{2,1} 
=
- \frac{1}{2}(\vec{n}, \nabla_{\vec{n}}X)_g\vec{n} 
= 
- \frac{1}{2}(\vec{n},\nabla_X\vec{n})_g\vec{n} + \frac{1}{2}(\vec{n},[X,\vec{n}])_g\vec{n} 
=
\frac{1}{2}(\vec{n},[X,\vec{n}])_g\vec{n}.
\ee
For the second term $A_{2,2}$ \eqref{eq1-09-june-2025}, we set 
$\nabla_X\vec{n} =: B^b\delb_b$ and obtain 
\be
\sigmab^{ac}(\nabla_X\vec{n},\delb_c)_g\delb_a = B^b\gb_{bc}\sigmab^{ac}\delb_a = B^a\delb_a = \nabla_X\vec{n}, 
\ee
thus
\be
A_2 = \frac{1}{2}g^{\alpha\gamma}(\vec{n},\nabla_{\gamma}X)_g\del_{\alpha}
=
\frac{1}{2}(\vec{n},[X,\vec{n}])_g\vec{n}
- \frac{1}{2}\nabla_X\vec{n}. 
\ee

For the third term $A_3$, we have 
\be
A_3 = - \frac{1}{2}g^{\alpha\gamma}(\vec{n},\del_{\alpha})_g\nabla_{\gamma}X
=
- \frac{1}{2}\nabla_{\vec{n}}X = - \frac{1}{2}\nabla_X\vec{n} + \frac{1}{2}[X,\vec{n}].
\ee
Finally, for the last term in \eqref{eq1-06-june-2025}, we write 
\be
A_4 = - \frac{1}{2}g^{\alpha\gamma}(\del_{\alpha},\nabla_{\gamma}X)_g\vec{n}
= 
- \frac{1}{4}g^{\alpha\gamma}
\big((\del_{\alpha},\nabla_{\gamma}X)_g + (\nabla_{\alpha}X,\del_{\gamma})_g\big)\vec{n}
=
- \frac{1}{4}g^{\alpha\gamma}\pi[X]_{\alpha\gamma}\vec{n}, 
\ee
where $\pi[X] = \Lcal_Xg$ is the deformation tensor of $X$. Therefore, we have reached  
\begin{equation}
\vec{w} = \frac{1}{2}\big([X,\vec{n}] + (\vec{n},[X,\vec{n}])_g\vec{n}\big)
- \frac{1}{4}g^{\alpha\gamma}\pi[X]_{\alpha\gamma}\vec{n}, 
\end{equation}
which coincides with \eqref{eq2-26-june-2025}. 
\ese
\end{proof}

}


\subsection{ The Lorentzian rotation vectors}

{ 

\paragraph{Analysis of the deformation tensors.}

We now focus on the error vector $\vec{w}$ acting on the Lorentzian rotations $L_a,\Omega_{ab}$. We recall that, for an arbitrary vector field $X$, we consider 
\begin{equation}
\pi[X]_{\alpha\beta} = (\Lcal_Xg)_{\alpha\beta} = X(g_{\alpha\beta}) - g([X,\del_{\alpha}],\del_{\beta}) - g(\del_{\alpha},[X,\del_{\beta}]).
\end{equation}
Observe that
\be
X(g_{\alpha\beta}) = g(\nabla_X\del_{\alpha},\del_{\beta}) + g(\del_{\alpha},\nabla_X\del_{\beta}).
\ee
and the explicit expression
\begin{equation}
\label{eq3-22-june-2025}
\pi[X]_{\alpha\beta} = g(\nabla_\alpha X,\del_{\beta}) + g(\del_{\alpha},\nabla_{\beta}X)
=\nabla_{\alpha}X^{\gamma}g_{\beta\gamma} + \nabla_{\beta}X^{\gamma}g_{\alpha\gamma}.
\end{equation}
We are interested particularly in $\pi[L_a]_{\alpha\beta}, \pi[\Omega_{ab}]_{\alpha\beta}$ and $\pi[\del_{\gamma}]_{\alpha\beta}$. It is well known that, since they are Killing vectors for the Minkowski metric 
\be
\pi[L_a]_{\alpha\beta} = \pi[\Omega_{ab}]_{\alpha\beta} = \pi[\del_{\gamma}]_{\alpha\beta} = 0.
\ee 
Let us spell out the components of the deformation tensors associated with $L_a,\Omega_{ab}$ and $\del_{\gamma}$ in  the natural frame: 
\begin{equation}\label{eq7-30-june-2025}
\aligned
\pi[L_a]_{\alpha\beta} & = \, L_aH_{\alpha\beta} 
+ \delta_{\alpha0}H_{a\beta} + \delta_{\beta0}H_{\alpha a} 
+ \delta_{\beta a}H_{\alpha 0} + \delta_{\alpha a}H_{0\beta},
\\
\pi[\Omega_{ab}]_{\alpha\beta} & = \, \Omega_{ab}H_{\alpha\beta} 
+ \delta_{\alpha a}H_{b\beta} + \delta_{\beta a}H_{\alpha b}
- \delta_{\alpha b}H_{a\beta} - \delta_{\beta b}H_{\alpha a},
\\
\pi[\del_{\gamma}]_{\alpha\beta} & = \, \del_{\gamma}H_{\alpha\beta}, 
\endaligned
\end{equation}
and, in the adapted frame, 
\begin{equation}\label{eq6-02-july-2025}
\aligned
\overline{\pi}[L_a]_{\alpha\beta} & = \, L_a(\Hb_{\alpha\beta}) + \delb_{\alpha}T\Hb_{a\beta} + \delb_{\beta}T\Hb_{\alpha a},\quad && \text{ in }  \Mcal^{\Hcal}_{[s_0,s_1]},
\\
\overline{\pi}[\Omega_{ab}]_{\alpha\beta} & = \, \Omega_{ab}(\Hb_{\alpha\beta})
+ \delta_{\alpha a}\Hb_{b\beta} + \delta_{\beta a}\Hb_{\alpha b}
- \delta_{\alpha b}\Hb_{a\beta} - \delta_{\beta b}\Hb_{\alpha a},\quad &&  \text{ in }  \Mcal_{[s_0,s_1]},
\\
\overline{\pi}[\del_{\gamma}]_{\alpha\beta} & = \,\del_{\gamma}\Hb_{\alpha\beta} 
+ \Hb_{\beta\delta}\delb_{\alpha}\Psib_{\gamma}^{\delta}
+ \Hb_{\alpha\delta}\delb_{\beta}\Psib_{\gamma}^{\delta},\quad && \text{ in }  \Mcal_{[s_0,s_1]}.
\endaligned
\end{equation}

The above expressions directly lead to the following estimates.

\begin{lemma}\label{lem4-08-aout-2025}
Assume that the condition \eqref{eq-USA-condition} holds for a sufficiently small $\eps_s$. Then
one has 
\begin{equation}
\aligned
\big|g^{\alpha\beta}\pi[L_a]_{\alpha\beta}\big|
& \lesssim |H|_1 \quad \text{ in } \Mcal^{\Hcal}_{[s_0,s_1]},
\\
\big|g^{\alpha\beta}\pi[\Omega_{ab}]_{\alpha\beta}\big|
& \lesssim |H|_1 \quad \text{ in } \Mcal_{[s_0,s_1]}.
\endaligned 
\end{equation}
\end{lemma}


\paragraph{Analysis on the Lie derivatives of the unit normal vector.}

\begin{lemma}\label{lem2-08-aout-2025}
Assume that the condition \eqref{eq-USA-condition} holds for a sufficiently small $\eps_s$. Then, in the hyperboloidal domain in $\Mcal^{\Hcal}_{[s_0,s_1]}$, one has
\bse
\begin{equation}\label{eq5-07-aout-2025}
\aligned
{} [L_a,\vec{n}]^{\top} & =   - \lapsb^{-1}\gb^{0c}L_a(\lapsb^2) \delb_c + R[L_a], 
\endaligned
\end{equation}
in which the remainder $R[L_a]$ satisfies 
\begin{equation}
{} |R[L_a]|_{\vec{n}}\lesssim \zetab^{-1}|H|_1, 
\end{equation}
\ese
\bse
together with 
\begin{equation}\label{eq6-08-aout-2025}
\aligned
{} [\Omega_{ab},\vec{n}]^{\top} 
= - \lapsb^{-1}\Omega_{ab}(\lapsb^2)\gb^{0c}\delb_c + R[\Omega_{ab}], 
\endaligned
\end{equation}
in which the remainder $R[\Omega_{ab}]$ satisfies 
\begin{equation}
{} |R[\Omega_{ab}]|_{\vec{n}}\lesssim \zetab^{-1}|H|_1.
\end{equation}
\ese
\end{lemma}

\begin{proof}
\bse
First of all, observe that, thanks to \eqref{eq5-12-june-2025} and \eqref{eq1-13-aout-2025}
\begin{equation}\label{eq4-07-aout-2025}
\aligned
&[L_a,\vec{n}]^{\top} = [t\delb_a,\vec{n}]^{\top} 
\\
& =  - \lapsb^{-1}L_a(\betab^c)\delb_c - \vec{n}(t)\delb_a
=- \lapsb^{-1}L_a(\lapsb^2\gb^{0c})\delb_c 
- \lapsb^{-1}\big(J - (x^c/r)\betab^c\delb_rT\big)\delb_a
\\
& = - \lapsb L_a\gb^{0c}\delb_c - \lapsb^{-1}\gb^{0c}L_a(\lapsb^2)\delb_c 
- \lapsb^{-1}\big(J - (x^c/r)\betab^c\delb_rT\big)\delb_a.
\endaligned
\end{equation}
For the first term in the right-hand side, we remark that in $\Mcal^{\Hcal}_{[s_0,s_1]}$, $\zeta = J = (s/t)$ and $\delb_rT = (r/t)$. Thus we find, thanks to \eqref{eq11-03-aout-2025}
\begin{equation}
\aligned
- \lapsb L_a\gb^{0c}\delb_c & =  - \lapsb L_a\big(-J^{-1}(x^c/r)\delb_rT\big)\delb_c
- \lapsb L_a\big(J^{-1}H^{0c} - J^{-1}(x^b/r)\delb_rT H^{cb}\big)\delb_c
\\
& = \lapsb s^{-1}L_a - \lapsb L_a(t/s) \big(H^{0c} - (x^b/t) H^{bc}\big)\delb_c
- \lapsb (t/s)L_a\big(H^{0c} - (x^b/t) H^{bc}\big)\delb_c.
\endaligned
\end{equation}
For the last term in the right-hand side of \eqref{eq4-07-aout-2025}, thanks to \eqref{eq1-13-aout-2025} and \eqref{eq11-03-aout-2025} we have 
\begin{equation}
\aligned
&  - \lapsb^{-1}\big(J - (x^c/r)\betab^c\delb_rT\big)\delb_a 
\\
& = - \lapsb^{-1}\big(J-(x^c/r)\lapsb^2\gb^{0b}\delb_rT\big)\delb_a
\\
& =  - \lapsb^{-1}\big(J + J^{-1}\lapsb^2(\delb_rT)^2\big)\delb_a 
- \lapsb (t/s)\big((x^d/r)\delb_rT H^{dc} - H^{0c}\big)(x^c/t)\delb_a
\\
& = - \lapsb s^{-1}L_a + (\lapsb - \lapsb^{-1})(s/t)\delb_a
- \lapsb (t/s)\big((x^d/t) H^{dc} - H^{0c}\big)(x^c/t)\delb_a.
\endaligned
\end{equation}
These lead us to
\begin{equation}
\aligned
{} [L_a,\vec{n}]^{\top} 
& =   - \lapsb^{-1}\gb^{0c}L_a(\lapsb^2)\delb_c
+(\lapsb - \lapsb^{-1})(s/t)\delb_a \\
& \quad - \lapsb L_a(t/s) \big(H^{0c} - (x^b/t) H^{bc}\big)\delb_c
- \lapsb (t/s)L_a\big(H^{0c} - (x^b/t) H^{bc}\big)
\\
& \quad - \lapsb (t/s)\big((x^d/t) H^{dc} - H^{0c}\big)(x^c/t)\delb_a.
\endaligned
\end{equation}
Recalling that, thanks to \eqref{eq6-03-aout-2025}, $(\lapsb- \lapsb^{-1}) = \lapsb(1- \lapsb^{-2}) = \lapsb (1+\gb^{00})$ and observing that $\lapsb_{\eta} = J\zeta^{-1} = 1$ in $\Mcal^{\Hcal}_{[s_0,s_1]}$, we obtain
\begin{equation}\label{eq1-18-oct-2025(l)}
\aligned
{} [L_a,\vec{n}]^{\top} & = 
- \lapsb^{-1}\gb^{0c}L_a(\lapsb^2) \delb_c
+\lapsb\Hb^{00}(s/t)\delb_a
- \lapsb L_a(t/s) \big(H^{0c} - (x^b/t) H^{bc}\big)\delb_c
\\
& \quad - \lapsb (t/s)L_a\big(H^{0c} - (x^b/t) H^{bc}\big)
- \lapsb (t/s)\big((x^d/t) H^{dc} - H^{0c}\big)(x^c/t)\delb_a.
\endaligned
\end{equation}
We remark that, thanks to \eqref{eq4-23-july-2025}, 
$$
\big|\Hb^{00}\big| = \big|(s/t)^{-2}H^{\Ncal00}\big| + \big|R[H]\big|\lesssim \zetab^{-2}|H| + |H|^2.
$$
Thus, all terms in the right-hand side of \eqref{eq1-18-oct-2025(l)}, except  the first one, are bounded by $\zetab^{-1}|H|_1$,
where we used $\big|L_a(s/t)\big|\lesssim (s/t)$, $\big|L_a(x^b/r)\big|\lesssim 1$. We have established \eqref{eq5-07-aout-2025}.
\ese
\bse
For $\Omega_{ab}$, a similar calculation provides us with the identity 
\be
{} [\Omega_{ab},\vec{n}]^{\top} 
=- \lapsb^{-1}\Omega_{ab}(\betab^c)\delb_c - \vec{n}(x^a)\delb_b + \vec{n}(x^b)\delb_a.
\ee
Here, $\zeta,T,J$ are purely radial functions which therefore commute with the spatial rotations $\Omega_{ab}$ and, therefore, 
\be
\aligned
& \quad - \vec{n}(x^a)\delb_b + \vec{n}(x^b)\delb_a 
= \lapsb^{-1}\big(\betab^a\delb_b - \betab^b\delb_a\big)
= \lapsb\gb^{0a}\delb_b - \lapsb\gb^{0b}\delb_a
\\
 & = - \lapsb r^{-1}J^{-1}\delb_rT\Omega_{ab} 
+ \lapsb J^{-1}\big(H^{0a}-(x^d/r)\delb_rT H^{da}\big)\delb_b 
- \lapsb J^{-1}\big(H^{0b} - (x^d/r)\delb_rT H^{db}\big)\delb_a.
\endaligned
\ee
We also have 
\be
\aligned
- \lapsb^{-1}\Omega_{ab}(\betab^c)\delb_c
& = - \lapsb^{-1}\Omega_{ab}(\lapsb^2\gb^{0c})\delb_c
=- \lapsb^{-1}\Omega_{ab}(\lapsb^2)\gb^{0c}\delb_c - \lapsb\Omega_{ab}(\gb^{0c})\delb_c
\\
& = - \lapsb^{-1}\Omega_{ab}(\lapsb^2)\gb^{0c}\delb_c
+ \lapsb r^{-1}J^{-1}\delb_rT\Omega_{ab} 
+ \lapsb J^{-1}\Omega_{ab}\big((x^d/r)\delb_rTH^{dc} - H^{0c}\big)\delb_c, 
\endaligned
\ee
and 
\be
\aligned
{} [\Omega_{ab},\vec{n}]^{\top} 
& = - \lapsb^{-1}\Omega_{ab}(\lapsb^2)\gb^{0c}\delb_c
+ \lapsb J^{-1}\Omega_{ab}\big( H^{0c} - (x^d/r)\delb_rTH^{dc}\big)\delb_c
\\
& \quad - \lapsb J^{-1}\big(H^{0a}-(x^d/r)\delb_rT H^{da}\big)\delb_b 
- \lapsb J^{-1}\big(H^{0b} - (x^d/r)\delb_rT H^{db}\big)\delb_a.
\endaligned
\ee
The terms other than the first in the right-hand side are bounded by, 
$
\lapsb J^{-1}|H|_1\lesssim \zetab^{-1}|H|_1.
$
This concludes the derivation of \eqref{eq6-08-aout-2025}.
\ese
\end{proof}


\paragraph{Conclusion.}

\begin{proposition}\label{lem3-08-aout-2025}
Assume that \eqref{eq-USA-condition} holds with a sufficiently small $\eps_s$. 
\\
$\bullet$ In the hyperboloidal domain $\Mcal^{\Hcal}_{[s_0,s_1]}$, one has 
\begin{equation}\label{eq5-08-aout-2025}
\big|\vec{w}[L_a]\big|_{\vec{n}}\lesssim \zetab^{-1}|H|_1 + \frac{(s/t)^{-2}|H^{\Ncal00}|_1}{1+(s/t)^{-2}|H^{\Ncal00}|} \quad \text{ in } \Mcal^{\Hcal}_{[s_0,s_1]}.
\end{equation}
$\bullet$ In the merging-Euclidean domain $\Mcal^{\ME}_{[s_0,s_1]}$, one has 
\begin{equation}\label{eq9-11-aout-2025}
\big|\vec{w}[\Omega_{ab}]\big|_{\vec{n}}\lesssim \zetab^{-1}|H|_1 + \frac{\zeta^{-2}|H^{\Ncal00}|_1}{1+\zeta^{-2}|H^{\Ncal00}|} \quad \text{ in } \Mcal^{\ME}_{[s_0,s_1]}.
\end{equation}
\end{proposition} 

\begin{proof}
\bse
Recalling Lemma~\ref{lem1-08-aout-2025} and \eqref{eq11-03-aout-2025}, for any vector field $X = L_a,\Omega_{ab}$ we find 
\be
\aligned
- \lapsb^{-1}\gb^{0c}X(\lapsb^2)\delb_c & =  -2\gb^{0c}X(\lapsb)\delb_c
= 2J^{-1}\delb_rTX(\lapsb)\,\delb_r + 2J^{-1}X(\lapsb)\big((x^d/r)\delb_rTH^{dc} - H^{0c}\big)\delb_c
\\
=:& T_1 + T_2.
\endaligned
\ee
We have first 
\begin{equation}
|T_2|_{\vec{n}}\lesssim J^{-1}\lapsb |H|\Big(\frac{\zeta^{-2}|H^{\Ncal00}|_1}{1+\zeta^{-2}|H^{\Ncal0}|} + |H|_1\Big)
\lesssim (\zeta/\zetab)\frac{\zeta^{-2}|H^{\Ncal00}|_1}{1+\zeta^{-2}|H^{\Ncal00}|} + (\zeta/\zetab)|H|_1.
\end{equation}
For $T_1$, we rely on Lemma~\ref{lem1-14-june-2025} and write 
\be
|\delb_r|_{\vec{n}}^2 = g(\delb_r,\delb_r) = \zeta^2-H^{\Ncal00} + T_{11}, 
\ee
in which $|T_{11}|\lesssim \zeta^2|H| + |H|^2$. 
Therefore, we obtain 
\begin{equation}\label{eq7-08-aout-2025}
|\delb_r|_{\vec{n}}\lesssim \zetab, 
\end{equation}
and so 
\begin{equation}
\big|- \lapsb^{-1}\gb^{0c}X(\lapsb^2)\delb_c\big|_{\vec{n}}\lesssim \frac{\zeta^{-2}|H^{\Ncal00}|_1}{1+\zeta^{-2}|H^{\Ncal0}|} + |H|_1.
\end{equation}
By recalling Lemmas~\ref{lem2-08-aout-2025} and~\ref{lem4-08-aout-2025}, the desired estimates are established. 
\ese
\end{proof}

}


\subsection{ Derivation of the Sobolev inequality in the hyperboloidal region}
\label{section===94}

{ 

We proceed by assuming the uniform spacelike condition \eqref{eq-USA-condition} for some sufficiently small $\eps_s$ (unless specified otherwise). It is convenient to introduce the short-hand notation 
\be
A[\Psi] := \la\Psi,\vec{n}\cdot\Psi\ra_{\ourD}.
\ee
We assume that the metric satisfies 
\begin{equation}\label{eq1-09-aout-2025}
\zetab^{-1}|H|_1 + \frac{(s/t)^{-2}|H^{\Ncal00}|_1}{1+(s/t)^{-2}|H^{\Ncal00}|}\lesssim (s/t)^{- \delta}
\quad \text{ in } \Mcal^{\Hcal}_{[s_0,s_1]} \quad \text{ for some } 0\leq \delta\leq 1. 
\end{equation}
Then, thanks to \eqref{eq10-30-june-2025} together with Proposition~\ref{lem3-08-aout-2025}, we obtain
\begin{equation}\label{eq3-03-july-2025}
\big| L_a(A[\Psi])\big|\lesssim (s/t)^{- \delta}|\Psi|_{\vec{n}}^2
+|\Psi|_{\vec{n}}\sum_{|I|\leq 1}\big|\widehat{L}^I\Psi\big|_{\vec{n}}.
\end{equation}

Moreover, given any sufficiently regular spinor field $\Psi$, we introduce the functions 
\bel{equa-21juillet2025-d}
u_k = u_k[\Psi] := \sqrt{(s/t)^{1+4\delta}A[\Psi] + k^{-2}},
\qquad
v_k = v_k[\Psi] := \sqrt{(s/t)^{1+2\delta}A[\Psi] + k^{-2}}, 
\ee
which clearly are also regular functions.  It is clear that 
\begin{equation}\label{eq1-10-mai-2025}
2u_k\geq \max\bigl( (s/t)^{1/2+2\delta}|\Psi|_{\vec{n}}, k^{-1}\bigr),
\qquad
2v_k\geq \max\bigl( (s/t)^{1/2+\delta}|\Psi|_{\vec{n}}, k^{-1}\bigr).
\end{equation} 

\begin{lemma}
\label{lem1-10-mai-2025}
Let $\Psi$ be a sufficiently regular spinor field defined in $\Mcal^{\Hcal}_{[s_0,s_1]}$. Suppose that in $\Mcal^{\Hcal}_{[s_0,s_1]}$ 
\begin{equation}\label{eq4-10-mai-2025}
\big|L_a(A[\Psi])\big|
\lesssim |\Psi|_{\vec{n}}\sum_{|I|\leq 1}|\widehat{L}^I\Psi|_{\vec{n}} 
+ (s/t)^{- \delta}|\Psi|_{\vec{n}}^2, \qquad \delta\geq 0.
\end{equation}
Then the functions defined by \eqref{equa-21juillet2025-d} satisfy 
\begin{equation}\label{eq5-15-mai-2025}
\aligned
|L_au_k| & \lesssim  (s/t)^{1/2+\delta}|\Psi|_{\vec{n}} 
+ \sum_{|I|= 1}(s/t)^{1/2+2\delta}|\widehat{L}^I\Psi|_{\vec{n}},
\\
|L_av_k| & \lesssim  (s/t)^{1/2}|\Psi|_{\vec{n}} 
+ \sum_{|I|= 1}(s/t)^{1/2+\delta}|\widehat{L}^I\Psi|_{\vec{n}}.
\endaligned
\end{equation}
\end{lemma}

\begin{proof} 
\bse
Thanks to \eqref{eq4-10-mai-2025}, we have 
\begin{equation}
\big|L_a\big(u_k^2\big)\big|
\lesssim 
(s/t)^{1+4\delta}|\Psi|_{\vec{n}}
\sum_{|I|\leq 1}|\widehat{L}^I\Psi|_{\vec{n}} 
+ (s/t)^{1+3\delta}|\Psi|_{\vec{n}}^2,
\end{equation}
which provides us with 
\be
|u_k||L_au_k|\lesssim(s/t)^{1+4\delta}|\Psi|_{\vec{n}}\sum_{|I|\leq 1}|\widehat{L}^I\Psi|_{\vec{n}} + (s/t)^{1+3\delta}|\Psi|_{\vec{n}}^2.
\ee
In combination with \eqref{eq1-10-mai-2025}, this leads us to the desired result. The derivation of the estimate on $v_k$ is similar. 
\ese
\end{proof}

\begin{lemma}\label{prop1-16-mai-2025}
Let $\Psi$ be a sufficiently regular spinor field defined in $\Mcal^{\Hcal}_{[s_0,s_1]}$. Suppose that in $\Mcal^{\Hcal}_{[s_0,s_1]}$ 
\begin{equation}\label{eq9-11-mai-2025}
\big|L_a(A[\Psi])\big|
\lesssim |\Psi|_{\vec{n}}\sum_{|I|\leq 1}|\widehat{L}^I\Psi|_{\vec{n}} 
+ (s/t)^{- \delta}|\Psi|_{\vec{n}}^2, \qquad \delta\geq 0.
\end{equation}
Then one has the pointwise bound 
\begin{equation}\label{eq10-11-mai-2025}
(s/t)^{1/2+2\delta}t^{3/2}|\Psi|_{\vec{n}}\lesssim 
\sum_{|I|\leq 2}\|(s/t)^{1/2+|I|\delta}|L^I\Psi|_{\vec{n}}\|_{L^2(\Mcal_s)}.
\end{equation}
\end{lemma}

\begin{proof}
The result is based on Lemma~\ref{lem2-15-mai-2025}. Let $(t,x)\in\Mcal^{\Hcal}_s$ with $s = \sqrt{t^2-r^2}$, $r = |x|$. By applying a spatial rotation, we suppose that $x = \hat{x} = - \frac{\sqrt{3}\, r}{3}(1,1,1)$.  We apply \eqref{eq8-15-mai-2025} to $u_k[\Psi]$, and  thanks to \eqref{eq5-15-mai-2025}, we find 
\begin{equation}\label{eq7-16-june-2025}
\aligned
&t\Big(\int_{\Ccal_{t,\hat{x}}}|u_k[\Psi]|^6 \diff x\Big)^{1/6}
\\
& \lesssim 
\Big(\int_{\Mcal^{\Hcal}_s}\big((s/t)^{1+4\delta}|\Psi|_{\vec{n}}^2 + k^{-2}\big)\diff x\Big)^{1/2} 
+ \sum_{a}\Big(\int_{\Mcal^{\Hcal}_s}|L_au_k|^2\diff x\Big)^{1/2}
\\
& \lesssim 
\Big(\int_{\Mcal^{\Hcal}_s}\big((s/t)^{1+2\delta}|\Psi|_{\vec{n}}^2 + k^{-2}\big)\diff x\Big)^{1/2} 
+ \sum_{a}\Big(\int_{\Mcal^{\Hcal}_s}|\widehat{L_a}\Psi|_{\vec{n}}^2(s/t)^{1+4\delta}\diff x\Big)^{1/2}.
\endaligned
\end{equation}
We now take the limit $k\rightarrow \infty$. Observe that $\Mcal^{\Hcal}_s$ is of finite measure with respect to $\diff x$. Thanks to \eqref{eq1-10-mai-2025} (applied to the left-hand side) and Lebesgue's dominant convergence theorem (applied to the first term in the right-hand side), we obtain
\begin{subequations}
\begin{equation}\label{eq8-16-june-2025}
\aligned
t\Big(\int_{\Ccal_{t,\hat{x}}}\big|(s/t)^{1/2+2\delta}|\Psi|_{\vec{n}}\big|^6 \diff x\Big)^{1/6}
& \lesssim  \Big(\int_{\Mcal^{\Hcal}_s}|\Psi|_{\vec{n}}^2(s/t)^{1+2\delta}\diff x\Big)^{1/2}
\\
& \quad + \sum_{a}\Big(\int_{\Mcal^{\Hcal}_s}|\widehat{L_a}\Psi|_{\vec{n}}^2(s/t)^{1+4\delta}\diff x\Big)^{1/2}.
\endaligned
\end{equation}
On the other hand, applying the same argument to $v_k[\Psi]$, we also obtain
\begin{equation}\label{eq5-16-mai-2025}
\aligned
t\Big(\int_{\Ccal_{t,\hat{x}}}\big|(s/t)^{1/2+\delta}|\Psi|_{\vec{n}}\big|^6 \diff x\Big)^{1/6}
& \lesssim  \Big(\int_{\Mcal^{\Hcal}_s}|\Psi|_{\vec{n}}^2(s/t)\diff x\Big)^{1/2}
\\
& \quad + \sum_{a}\Big(\int_{\Mcal^{\Hcal}_s}|\widehat{L_a}\Psi|_{\vec{n}}^2(s/t)^{1+2\delta}\diff x\Big)^{1/2}.
\endaligned
\end{equation}
\end{subequations}

Next, applying \eqref{eq7-15-mai-2025} to $u_k[\Psi]$ and recalling \eqref{eq7-16-june-2025} and \eqref{eq5-15-mai-2025}, we find 
\be
\aligned
t^{1/2}|u_k(t,\hat{x})|& \lesssim  \Big(\int_{\Ccal_{t,\hat{x}}}|u_k|^6\diff x\Big)^{1/6} 
+ \sum_{a}\Big(\int_{\Ccal_{t,\hat{x}}}|L_au_k|^6\diff x\Big)^{1/6}
\\
& \lesssim  
t^{-1}\Big(\int_{\Mcal^{\Hcal}_s}\big((s/t)^{1+2\delta}|\Psi|_{\vec{n}}^2 + k^{-2}\big)\diff x\Big)^{1/2} 
+ t^{-1}\sum_{a}\Big(\int_{\Mcal^{\Hcal}_s}|\widehat{L_a}\Psi|_{\vec{n}}^2(s/t)^{1+4\delta}\diff x\Big)^{1/2}
\\
& \quad +  \Big(\int_{\Mcal^{\Hcal}_s}\big|(s/t)^{1/2+\delta}|\Psi|_{\vec{n}}\big|^6\Big)^{1/6} 
+\sum_{|I|= 1}
\Big(\int_{\Ccal_{t,\hat{x}}}\big|(s/t)^{1/2+2\delta}|\widehat{L}^I\Psi|_{\vec{n}}\big|^6\diff x\Big)^{1/6}.
\endaligned
\ee
For the last two terms, we apply \eqref{eq5-16-mai-2025} and \eqref{eq8-16-june-2025}, respectively:
\be
\aligned
t^{3/2}|u_k(t,\hat{x})|
& \lesssim 
\Big(\int_{\Mcal^{\Hcal}_s}\big((s/t)^{1+2\delta}|\Psi|_{\vec{n}}^2 + k^{-2}\big)\diff x\Big)^{1/2} 
+ \sum_{a}\Big(\int_{\Mcal^{\Hcal}_s}|\widehat{L_a}\Psi|_{\vec{n}}^2(s/t)^{1+4\delta}\diff x\Big)^{1/2}
\\
& \quad + \Big(\int_{\Mcal^{\Hcal}_s}|\Psi|_{\vec{n}}^2(s/t)^{1+2\delta}\diff x\Big)^{1/2}
+\sum_{a}\Big(\int_{\Mcal^{\Hcal}_s}|\widehat{L_a}\Psi|_{\vec{n}}^2(s/t)^{1+4\delta}\diff x\Big)^{1/2}
\\
& \quad +  \Big(\int_{\Mcal^{\Hcal}_s}|\Psi|_{\vec{n}}^2(s/t)\diff x\Big)^{1/2}
+\sum_{a}\Big(\int_{\Mcal^{\Hcal}_s}|\widehat{L_a}\Psi|_{\vec{n}}^2(s/t)^{1+2\delta}\diff x\Big)^{1/2}
\\
& \quad + \sum_{a}\Big(\int_{\Mcal^{\Hcal}_s}|\widehat{L_a}\Psi|_{\vec{n}}^2(s/t)^{1+4\delta}\diff x\Big)^{1/2}
+\sum_{a,b}\Big(\int_{\Mcal^{\Hcal}_s}|\widehat{L_b}\widehat{L_a}\Psi|_{\vec{n}}^2(s/t)^{1+4\delta}\diff x\Big)^{1/2}.
\endaligned
\ee
Then again, by \eqref{eq1-10-mai-2025} and Lebesgue's dominant convergence theorem, we obtain \eqref{eq10-11-mai-2025}.
\end{proof}

\begin{proof}[Proof of Theorem~\ref{thm2-05-april-2025} in the hyperboloidal region]
Since the assumptions \eqref{eq5a-05-april-2025} together with \eqref{eq-USA-condition} guarantee \eqref{eq3-03-july-2025}, we can apply Lemma~\ref{prop1-16-mai-2025} and conclude that the inequality \eqref{eq9-16-june-2025} holds.
\end{proof}

}


\subsection{ Estimates concerning the tangent radial derivative $\delb_r$}
\label{section===96}

{ 

\paragraph{Aim.}

Our next task is establishing the  following estimate.

\begin{proposition}
\label{lem1-11-aout-2025}
Assume that \eqref{eq-USA-condition} holds with a sufficiently small $\eps_s$. Then one has 
\begin{equation}\label{eq6-11-aout-2025}
|\nabla_{\delb_r}\vec{n}|_{\vec{n}}\lesssim_{\delta} (\zeta/\zetab)^2\zeta^{- \delta} + \zetab^{-2}|\del H| \quad \text{ in } \Mcal^{\ME}_{[s_0,s_1]},\quad 0 < \delta \leq 1,
\end{equation}
\begin{equation}\label{eq5-19-aout-2025}
|\nabla_{\delb_r}\vec{n}|_{\vec{n}}\lesssim 1 + |\del H| \quad \text{ in } \Mcal^{\Ecal}_{[s_0,s_1]}.
\end{equation}
\end{proposition}

The proof will rely on several preliminary results. We begin with a general identity. 

\begin{lemma}
\label{lem1-03-aout-2025}
The future-oriented unit normal vector $\vec{n}$ satisfies the following identity for any tangent vector $X = \overline{X}^a\delb_a$ to $\Mcal_s$: 
\begin{equation}
\nabla_X\vec{n} = \overline{X}^a\Pi_{ab}\sigmab^{bc}\delb_c, 
\end{equation}
where $\Pi_{ab} = -(\vec{n},\nabla_{\delb_a}\delb_b)_g$ denotes the second fundamental form of $\Mcal_s$. 
\end{lemma}

\begin{proof} 
\bse
We recall that 
\be
\big(\nabla_X\vec{n},\vec{n}\big)_g = \frac{1}{2}\Lcal_X\big((\vec{n},\vec{n})_g\big) = 0, 
\ee
and we write $\nabla_X\vec{n} =: B^a\delb_a$. Then, we can decompose 
\be
B^a\gb_{ab} = (\nabla_X\vec{n},\delb_b)_g = - (\vec{n},\nabla_X\delb_b)_g = \Pi(X,\delb_b) 
= X^a\Pi_{ab}, 
\ee
which leads us to the desired expression.
\ese
\end{proof}


Next, thanks to Lemma~\ref{lem1-03-aout-2025}, we have 
\begin{equation}\label{eq1-11-aout-2025}
\nabla_{\delb_r}\vec{n} = (x^a/r)\Pi_{ab}\sigmab^{bc}\delb_c 
= \Pi(\delb_r,\delb_b)\sigmab^{bc}\delb_c
= -(\vec{n},\nabla_{\delb_r}\delb_b)_g\sigmab^{bc}\delb_c.
\end{equation}
For convenience, we introduce the notation
\be
\omega_{\alpha}(X) = \omega_{\alpha}^{\beta}(X) \del_{\beta} = \nabla_{X}\del_{\alpha}, 
\ee
so that
\be
\nabla_{\delb_r}\delb_b =(x^b/r)\delb_r^2T\del_t 
+ (x^b/r)\delb_rT\omega_0(\delb_r) + \omega_b(\delb_r).
\ee
It then follows that the radial derivative of the normal vector reads 
\begin{equation}
\label{eq3-11-aout-2025}
\nabla_{\delb_r}\vec{n} = - \big(\vec{n},\del_t\big)_g(x^b/r)\delb_r^2T\sigmab^{bc}\delb_c 
- (x^b/r)\delb_rT(\vec{n},\omega_0(\delb_r))_g\sigmab^{bc}\delb_c
- (\vec{n},\omega_b(\delb_r))_g\sigmab^{bc}\delb_c.
\end{equation}


\paragraph{A calculation in the flat case.}

When $g=\eta$, we have $\omega_{\alpha}\equiv 0$ and, therefore, 
\begin{equation}\label{eq7-09-aout-2025}
\vec{n}_{\eta} = \zeta^{-1}\del_t + \zeta^{-1}(x^b/r)\delb_rT\del_b= \zeta\del_t + \zeta^{-1}\delb_rT\delb_r
\quad \text{(flat geometry).} 
\end{equation}
Then \eqref{eq3-11-aout-2025} together with \eqref{eq1-17-mai-2025} leads us to
\begin{equation}\label{eq2-03-aout-2025}
\nabla_{\delb_r}\vec{n} = \zeta^{-3}\delb_r^2T\delb_r
\quad \text{(flat geometry).} 
\end{equation}
Since $|\delb_r|_{\vec{n}} = \zeta$, we have 
\begin{equation}
|\nabla_{\delb_r}\vec{n}|_{\vec{n}} = \zeta^{-2}|\delb_r^2T|\lesssim_{\delta}\zeta^{- \delta}
\quad \text{(flat geometry),} \quad 0 < \delta \leq 1,
\end{equation}
where we used \eqref{eq6-09-aout-2025}.


\paragraph{Proof of Proposition~\ref{lem1-11-aout-2025}.} 

\bse
We begin with the identity 
\be
\aligned
\big(\vec{n},\del_t\big)_g 
& =  J^{-1}\lapsb^{-1}(\delb_0 - \betab^d\delb_d,\delb_0)_g 
= J^{-1}\lapsb^{-1} (\gb_{00} - \gb_{0d}\betab^d) 
\\
& = J^{-1}\lapsb^{-1} (\gb_{00} - \gb_{0d}\sigmab^{de}\gb_{e0})
= - J^{-1}\lapsb.
\endaligned 
\ee
Here we have applied \eqref{eq2-18-oct-2025(l)}. So that, thanks to Claim~\ref{cor1-16-june-2025},
\begin{equation}\label{eq2-11-aout-2025}
\big|\big(\vec{n},\del_t\big)_g \big|\lesssim \zetab^{-1}.
\end{equation}
Recalling \eqref{eq4-16-june-2025}, thanks to Lemma~\ref{lem1-28-july-2025} we find 
\begin{equation}\label{eq4-11-aout-2025}
\big|- \big(\vec{n},\del_t\big)_g(x^b/r)\delb_r^2T\sigmab^{bc}\delb_c \big|_{\vec{n}}
\lesssim \zetab^{-2}\delb_r^2T\lesssim_{\delta} (\zeta/\zetab)^2\zeta^{- \delta}
\lesssim (\zeta/\zetab)^{2- \delta}\zetab^{- \delta}.
\end{equation}
On the other hand, the connection form is bounded by 
\begin{equation}\label{eq5-11-aout-2025}
|\omega_{\alpha}(\del_{\beta})|_{\vec{n}}\lesssim \zetab^{-1}|\del H|.
\end{equation}
Indeed, this is a consequence of Lemma~\ref{lem2-13-july-2025} and the following calculation:
\be
\omega_{\alpha}(\del_{\beta}) 
= \frac{1}{2}g^{\gamma\delta}(\del_{\alpha}g_{\beta\gamma} + \del_{\beta}g_{\alpha\gamma} - \del_{\gamma}g_{\alpha\beta})\del_{\delta}.
\ee
Since $|H|\lesssim \eps_s$, we have $|g^{\alpha\beta}|\lesssim 1$ and, therefore, \eqref{eq5-11-aout-2025} holds. Consequently, we can substitute \eqref{eq4-11-aout-2025}, \eqref{eq5-11-aout-2025}, and \eqref{eq4-16-june-2025} into \eqref{eq3-11-aout-2025}, and we obtain \eqref{eq6-11-aout-2025}.

\ese

For \eqref{eq5-19-aout-2025}, we only need to track the loss $\zeta^{- \delta}$. It is due to \eqref{eq6-09-aout-2025} on the term $\delb_r^2T$. It is clear that this term vanishes in the Euclidean domain $\Mcal^{\Ecal}_{[s_0,s_1]}$ since $\delb_rT\equiv 0$. Therefore, in view of \eqref{eq4-11-aout-2025} we find 
\be
\big|- \big(\vec{n},\del_t\big)_g(x^b/r)\delb_r^2T\sigmab^{bc}\delb_c \big|_{\vec{n}} = 0.
\ee
Moreover, also in the Euclidean domain we have $\zetab^{-1}\lesssim 1$, thanks to \eqref{eq2-14-june-2025} and the fact that $\zeta \equiv 1$. This establishes \eqref{eq5-19-aout-2025}, and this concludes the proof of Proposition~\ref{lem1-11-aout-2025}. 

}


\subsection{ Derivation of the weighted Sobolev inequality in the merging-Euclidean domain}
\label{section===910}

{ 

We now focus on the merging-Euclidean domain $\Mcal^{\ME}_{[s_0,s_1]}$ and derive the desired Sobolev inequality. It is convenient to proceed by assuming for technical reason that $\Psi$ is supported in $\mathcal{B}_R\cap \Mcal^{\ME}_{[s_0,s_1]}$, where $\mathcal{B}_R = \{(t,x), r\in\RR, x\in\RR^3, |x|\leq R\}$. We also use the notation
$
A[\Psi]:=\la \Psi,\vec{n}\cdot\Psi\ra_{\ourD}$, 
and we introduce the weight
\begin{equation}
w(s,r) := (2+r-t)^{2\kappa},
\end{equation}
together with the auxiliary functions, which are regularization of the (square-root) of $A[\Psi]$: 
\be
u_k := \sqrt{\zeta^{1+4\delta}wA[\Psi] + k^{-2}},
\qquad v_k := \sqrt{\zeta^{1+2\delta}wA[\Psi] + k^{-2}},
\qquad k=1,2,\cdots.
\ee
It is clear that 
\begin{equation}
2u_k\geq \max\bigl( \zeta^{1/2+2\delta}|\Psi|_{\vec{n}},k^{-1}\bigr), 
\quad
2v_k\geq \max\bigl( \zeta^{1/2+\delta}|\Psi|_{\vec{n}},k^{-1}\bigr). 
\end{equation}

Let us assume that the metric satisfies the decay condition 
\begin{equation}\label{eq10-11-aout-2025}
\zetab^{-1}|H|_1 + \frac{\zeta^{-2}|H^{\Ncal00}|_1}{1+\zeta^{-2}|H^{\Ncal00}|}\lesssim \zeta^{- \delta},\quad 0 < \delta\leq 1 \quad \text{ in } \Mcal^{\ME}_{[s_0,s_1]}. 
\end{equation}
Under this condition, the estimate \eqref{eq9-11-aout-2025} provides us with the following estimate for rotation fields: 
\begin{equation}\label{eq11-11-aout-2025}
\big|\Omega_{ab}A[\Psi]\big|
\lesssim \zeta^{- \delta}|\Psi|_{\vec{n}}^2 + |\Psi|_{\vec{n}}\sum_{|I|\leq 1}|\widehat{\Omega}^I\Psi|_{\vec{n}}.
\end{equation}
For the radial derivative $\delb_r$, the argument is more involved. 

\begin{lemma}
Assume that the uniform spacelike condition \eqref{eq-USA-condition} holds with a sufficiently small $\eps_s$. Suppose that, in addition to \eqref{eq10-11-aout-2025}, the first-order derivatives of the metric satisfy 
\begin{equation}\label{eq12-11-aout-2025}
\zetab^{-2}|\del H|\lesssim \zeta^{- \delta},\quad 0<\delta\leq 1 \quad \text{ in } \Mcal^{\ME}_{[s_0,s_1]}.
\end{equation}
Then for any spinor field, one has 
\begin{equation}\label{eq6-19-aout-2025}
|\delb_rA[\Psi]|\lesssim_{\delta} \zeta^{- \delta}|\Psi|_{\vec{n}}^2 + |\Psi|_{\vec{n}}\sum_{|I|\leq 1}|\widehat{\del}^I\Psi|_{\vec{n}}, \quad \text{ in }  \Mcal^{\ME}_{[s_0,s_1]},
\end{equation}
\begin{equation}\label{eq7-19-aout-2025}
|\delb_rA[\Psi]|\lesssim |\Psi|_{\vec{n}}^2 + |\Psi|_{\vec{n}}\sum_{|I|\leq 1}|\widehat{\del}^I\Psi|_{\vec{n}},\quad \text{ in }  \Mcal^{\E}_{[s_0,s_1]}.
\end{equation}
\end{lemma}

\begin{proof}
\bse
In order to establish \eqref{eq6-19-aout-2025}, we write 
\be
\delb_r\big(\la \Psi,\vec{n}\cdot\Psi\ra_{\ourD}\big) 
= \la\nabla_{\delb_r}\Psi,\vec{n}\cdot\Psi \ra_{\ourD}
+ \la\Psi,\vec{n}\cdot\nabla_{\delb_r}\Psi \ra_{\ourD}
+ \la\Psi,\nabla_{\delb_r}\vec{n}\cdot\Psi \ra_{\ourD},
\ee
therefore 
\begin{equation}
\big|\delb_r\big(\la \Psi,\vec{n}\cdot\Psi\ra_{\ourD}\big) \big|
\lesssim |\Psi|_{\vec{n}}\sum_{\alpha}|\nabla_{\alpha}\Psi|_{\vec{n}} 
+ |\nabla_{\delb_r}\vec{n}|_{\vec{n}}|\Psi|_{\vec{n}}^2.
\end{equation}
The second term in the right-hand side is bounded thanks to \eqref{eq6-11-aout-2025}, while for the first term, we have
\be
|\nabla_{\del_{\alpha}}\Psi|_{\vec{n}}\leq |\widehat{\del_{\alpha}}\Psi|_{\vec{n}} 
+ \frac{1}{4}\big|g^{\mu\nu}\del_{\mu}\cdot\nabla_{\nu}\del_{\alpha}\cdot\Psi\big|_{\vec{n}}
\lesssim |\widehat{\del_{\alpha}}\Psi|_{\vec{n}} + \zetab^{-1}|\del H||\Psi|_{\vec{n}}.
\ee
Consequently, it follows that 
\be
\big|\delb_r\big(\la \Psi,\vec{n}\cdot\Psi\ra_{\ourD}\big) \big|
\lesssim_{\delta} \zeta^{- \delta}\big((\zeta/\zetab)^2 + \zetab^{-2}\zeta^{\delta}|\del H|\big)|\Psi|_{\vec{n}}^2 + |\Psi|_{\vec{n}}\sum_{\alpha}|\widehat{\del_{\alpha}}\Psi|_{\vec{n}}, 
\ee
and, in view of \eqref{eq12-11-aout-2025}, we arrive at \eqref{eq6-19-aout-2025}.

On the other hand, to establish \eqref{eq7-19-aout-2025} we only need to apply \eqref{eq5-19-aout-2025} instead of \eqref{eq6-11-aout-2025}, and then observe that the weight $\zetab^{-1}\lesssim 1$ is bounded in $\Mcal^{\ME}_{[s_0,s_1]}$. Indeed, the latter claim follows from \eqref{eq2-14-june-2025} and the fact that $\zeta\equiv 1$.
\ese
\end{proof}

Next, we proceed as in the hyperboloidal region and establish the following estimates.

\begin{lemma}\label{lem1-15-mai-2025}
Consider sufficiently regular spinor fields $\Psi$ defined in $\Mcal^{\ME}_{[s_0,s_1]}$. Then provided, 
in the merging-Euclidean domain $\Mcal^{\ME}_{[s_0,s_1]}$, 
\begin{equation}\label{eq4-15-mai-2025}
\aligned
& \big|\Omega_{ab}A[\Psi]\big|
\lesssim \zeta^{- \delta}|\Psi|_{\vec{n}}^2 + |\Psi|_{\vec{n}}\sum_{|I|\leq 1}|\widehat{\Omega}^I\Psi|_{\vec{n}},
\\
&|\delb_rA[\Psi]|\lesssim_{\delta} \zeta^{- \delta}|\Psi|_{\vec{n}}^2 + |\Psi|_{\vec{n}}\sum_{|I|\leq 1}|\widehat{\del}^I\Psi|_{\vec{n}}, 
\endaligned
\end{equation}
then, for any $X = \Omega_{ab}$ or $X=\delb_r$, one also has in $\Mcal^{\ME}_{[s_0,s_1]}$ 
\begin{equation}\label{eq9-16-mai-2025}
\big|X(u_k)\big|
\lesssim_{\delta} \zeta^{1/2+\delta}w^{1/2}|\Psi|_{\vec{n}} + \zeta^{1/2+2\delta}w^{1/2}
\sum_{|I|\leq 1}
\big(\big|\widehat{\Omega}^I\Psi\big|_{\vec{n}}+\big|\widehat{\del}^I\Psi\big|_{\vec{n}}\big),
\end{equation}
\begin{equation}\label{eq10-16-mai-2025}
\big|X(v_k)\big|
\lesssim_{\delta}
\zeta^{1/2}w^{1/2}|\Psi|_{\vec{n}} + \zeta^{1/2+\delta}w^{1/2}
\sum_{|I|\leq 1}
\big(\big|\widehat{\Omega}^I\Psi\big|_{\vec{n}}+\big|\widehat{\del}^I\Psi\big|_{\vec{n}}\big),
\end{equation}
and, in the Euclidean domain $\Mcal^{\Ecal}_{[s_0,s_1]}$,
\begin{equation}\label{eq8-19-aout-2025}
\big|X(u_k)\big|
\lesssim w^{1/2}|\Psi|_{\vec{n}} + w^{1/2}
\sum_{|I|\leq 1}\big(\big|\widehat{\Omega}^I\Psi\big|_{\vec{n}}+\big|\widehat{\del}^I\Psi\big|_{\vec{n}}\big),
\end{equation}
\begin{equation}\label{eq9-19-aout-2025}
\big|X(v_k)\big|
\lesssim
w^{1/2}|\Psi|_{\vec{n}} + w^{1/2}
\sum_{|I|\leq 1}\big(\big|\widehat{\Omega}^I\Psi\big|_{\vec{n}}+\big|\widehat{\del}^I\Psi\big|_{\vec{n}}\big). 
\end{equation}
\end{lemma}

\begin{proof} We only sketch the proof, since it almost identical to the proof of Proposition~\ref{lem1-10-mai-2025}. The only difference are that we now rely on the bounds 
$$
|X\zeta|\lesssim_{\delta}\zeta^{1- \delta} \qquad \text{ to be compared with } |L_a(s/t)|\lesssim (s/t)
$$
and $|Xw|\lesssim w$. 
\end{proof}

We continue with the same strategy as in the hyperboloidal region.

\begin{lemma}\label{lem2-11-aout-2025}
Consider any sufficiently regular spinor field $\Psi$ defined in $\Mcal^{\ME}_{[s_0,s_1]}$ satisfying, in $\Mcal^{\ME}_{[s_0,s_1]}$, 
\begin{equation}\label{eq11-16-mai-2025}
\aligned
& \big|\Omega_{ab}A[\Psi]\big|
\lesssim \zeta^{- \delta}|\Psi|_{\vec{n}}^2 + |\Psi|_{\vec{n}}\sum_{|I|\leq 1}|\widehat{\Omega}^I\Psi|_{\vec{n}},
\\
&|\delb_rA[\Psi]|\lesssim_{\delta} \zeta^{- \delta}|\Psi|_{\vec{n}}^2 + |\Psi|_{\vec{n}}\sum_{|I|\leq 1}|\widehat{\del}^I\Psi|_{\vec{n}}
\endaligned
\end{equation}
for some exponent $0<\delta\leq 1$. Then, in the merging-Euclidean domain $\Mcal^{\ME}_{[s_0,s_1]}$ for $0\leq \kappa\leq 1$ one has 
\begin{equation}\label{eq12-16-mai-2025}
r(2+r-t)^{\kappa}\zeta^{1/2+2\delta}|\Psi|_{\vec{n}}
\lesssim_{\delta}
\sum_{j+|K|\leq 2}\|(2+r-t)^{\kappa}\zeta|\delb_r^j\Omega^K\Psi|_{\vec{n}}\|_{L^2(\Mcal^{\ME}_s)},
\end{equation}
while, in the Euclidean domain $\Mcal^{\Ecal}_{[s_0,s_1]}$,
\begin{equation}\label{eq10-19-aout-2025}
r(2+r-t)^{\kappa}|\Psi|_{\vec{n}}
\lesssim
\sum_{j+|K|\leq 2}\|(2+r-t)^{\kappa}|\delb_r^j\Omega^K\Psi|_{\vec{n}}\|_{L^2(\Mcal^{\Ecal}_s)},
\end{equation}

\end{lemma}

\begin{proof}
The proof is quite similar to that of Proposition~\ref{prop1-16-mai-2025}. Here we rely on Lemma \ref{lem2-16-mai-2025}. The only difference is that $\Mcal^{\ME}_s$ is {\bf not} of finite measure with respect to $\diff x$. We thus need to proof \eqref{eq12-16-mai-2025} in the case where $\Psi$ is compactly supported. Then again we let $k\rightarrow \infty$ and apply the theorem of dominant convergence. Then we follow a procedure of approximation of a general spinor field by compactly-supported spinor fields, which can be constructed via suitable cut-off functions.
\end{proof}

\begin{proof}[Proof of Theorem~\ref{thm2-05-april-2025} in the merging-Euclidean domain] We recall that Lemma \ref{lem1-15-mai-2025}, under the decay assumptions \eqref{eq10-11-aout-2025} and \eqref{eq12-11-aout-2025}, guarantees \eqref{eq11-16-mai-2025}. Then Lemma~\ref{lem2-11-aout-2025} leads us to the desired estimate.
\end{proof}

}


\section{Pointwise estimate for massive Dirac fields on light-bending spacetimes} 
\label{section=N7}

\subsection{ The Cauchy adapted frame}
\label{section===10-1}

{ 

\paragraph{Aim of this section.} 

Our next purpose is to establish a pointwise estimate for massive Dirac fields, as stated in Proposition~\ref{prop2-14-aout-2025}, below. We work on a light-bending spacetime and seek a control of 
\(
s^{3/2}\!\big((s/t)\sum_{\alpha}|\widehat{\del_\alpha}\Psi|_{\vec n}+\lapsb M\,|\Psi|_{\vec n}\big)
\)
at a spacetime point \((t,x)\) by data transported along the orthogonal generators denoted by \(\gamma_{t,x}\) and we provide explicit gauge/curvature remainders. Precisely, the estimate \eqref{eq8-26-aout-2025} stated below bounds the value at \((t,x)\) by the \emph{initial ray value}
\(
\big(|\nabla_{\vec L}(s^{3/2}\Psi)|_{\vec n} + \zetab^{-1} M\,|s^{3/2}\Psi|_{\vec n}\big)
\)
evaluated at \(\gamma_{t,x}(s_{t,x}^*)\), plus the integral of the effective source \(|F|_{\vec n,t,x}\), and lower–order terms. Here \(F\) collects the inhomogeneity \(\opDirac\Phi- \mathrm{i}M\Phi\), the gauge contribution \(\nabla_W\Psi\), and the curvature/quasi-null remainders \(R[H,\Psi]\) satisfying the perturbative bound \eqref{eq8-22-aout-2025}. The proof proceeds by renormalizing and considering \(s^{3/2}\Psi\), deriving a transport inequality for \(\sum_{\alpha}|\widehat{\del_\alpha}u|_{\vec n}+\lapsb M\,|u|_{\vec n}\) along \(\vec L\), and exploiting the light-bending condition to obtain a favorable sign on the principal part. The commutator identities in generalized wave gauge reduce all error terms to deformation tensors, the prescribed vector \(W\), and curvature (\(H,R_g\)), yielding the terms \(F\) and \(R[H,\Psi]\), while an integration along \(\gamma_{t,x}\) and an integration argument under the smallness \(\eps_s\) then yields \eqref{eq8-26-aout-2025}, below, with constants depending only on \(M\).


\paragraph{A frame of interest.}

Let us introduce the following {\sl Cauchy adapted frame} (following the terminology in~\cite[Section 5.3]{YCB}): 
\begin{equation}\label{eq1-20-july-2025}
\grave{v}_0  = \vec{L} := \lapsb\vec{n} = \delb_0 - \betab^a\delb_a,\quad \grave{v}_a = \delb_a, 
\end{equation}
which can also be expressed in the natural frame within the hyperboloidal domain $\Mcal^{\Hcal}_{[s_0,s_1]}$:
\begin{equation}\label{eq2-20-july-2025}
\grave{v}_0 = \big((s/t) - \betab^a(x^a/t)\big) \del_t - \betab^a\del_a,
\quad 
\grave{v}_a = (x^a/t)\del_t + \del_a.
\end{equation}
We denote by $\grave{\Phi}_{\alpha}^{\beta}$ and $\grave{\Psi}_{\alpha}^{\beta}$ the transition matrices between $\{\grave{v}_{\alpha}\}$ and $\{\del_{\alpha}\}$, in the sense that
\begin{equation}\label{eq3-20-july-2025}
\grave{v}_{\alpha} = \grave{\Phi}_{\alpha}^{\beta}\del_{\beta},
\quad 
\del_{\alpha} = \grave{\Psi}_{\alpha}^{\beta}\grave{v}_{\beta}.
\end{equation}
The corresponding dual frames are then 
\begin{equation}\label{eq4-20-july-2025}
\aligned
\grave{\omega}^0 & =  \big((s/t) - (x^a/t)\betab^a\big)^{-1}\diff t 
- \big((s/t) - (x^a/t)\betab^a\big)^{-1}(x^a/t)\diff x^a,
\\
\grave{\omega}^a & =  \big((s/t) - (x^a/t)\betab^a\big)^{-1}\betab^a \diff t + \diff x^a.
\endaligned
\end{equation}
The spacetime metric can be expressed as
\bse
\begin{equation}\label{eq5-20-july-2025}
\grave{g}_{00} = - \lapsb^2,\quad \grave{g}_{0a} = 0,
\quad
\grave{g}_{ab} = \gb_{ab}, 
\end{equation}
and as 
\begin{equation}\label{eq6-20-july-2025}
\grave{g}^{00} = - \lapsb^{-2},\quad \grave{g}^{a0} = 0,\quad \grave{g}^{ab} = \sigmab^{ab}.
\end{equation}
\ese

Observe that the Cauchy adapted frame depends on the metric. In the flat case, it coincides with the tangent-orthogonal frame and, in this case,
\begin{equation}\label{eq7-20-july-2025}
\big(\grave{\Phi}_{\alpha}^{\beta}\big|_{\eta}\big)_{\beta\alpha}
=\left(
\begin{array}{cc}
t/s & x^a/t
\\
x^b/s & \mathrm{I}_3
\end{array}
\right),
\quad \qquad
\big(\grave{\Psi}_{\alpha}^{\beta}\big|_{\eta}\big)_{\beta\alpha}
=\left(
\begin{array}{cc}
t/s & -x^a/s
\\
-tx^b/s^2 & \delta^{ab} + (x^ax^b)/s^2
\end{array}
\right).
\end{equation}
We also define $\bar{\delta}^b := \betab^b- \betab_{\eta}^b$. In the curved spacetime, we have 
\begin{equation}\label{eq1-21-july-2025}
\aligned
\big(\grave{\Psi}_{\alpha}^{\beta}\big)_{\beta\alpha} 
& = 
\left(
\begin{array}{cc}
t/s & -(x^a/s)
\\
(t/s)\betab^b& \delta^{ab} - (x^a/s)\betab^b
\end{array}\right)
=
\big(\grave{\Psi}_{\alpha}^{\beta}\big)_{\beta\alpha}\Big|_{g=\eta} + 
\left(
\begin{array}{cc}
0 &0
\\
(t/s)\bar{\delta}^b & \quad - (x^a/s)\bar{\delta}^b
\end{array}\right)
\\
& =: \big(\grave{\Psi}_{\alpha}^{\beta}\big)_{\beta\alpha}\Big|_{g=\eta} 
+ \big(\grave{\Delta}_{\alpha}^{\beta}\big)_{\beta\alpha}.
\endaligned
\end{equation}
One of our tasks is to estimate $\grave{\Delta}_{\alpha}^{\beta}$, referred to as the \textbf{frame deformation} of the Cauchy adapted frame. 


\paragraph{Preliminary estimates.}

As a preparation, we establish some  technical estimates, as follows. 

\begin{lemma}\label{lem1-31-july-2025}
Assume that the uniform spacelike condition \eqref{eq-USA-condition} holds with a sufficiently small $\eps_s$. Then one has
\begin{equation}\label{eq10-31-july-2025}
\aligned
\big|\del_{\alpha}\sigmab^{ab}\big|& \lesssim  \zetab^{-4}|\del H| + s^{-1}\zetab^{-3}|H| + \zetab^{-2},
\\
\big|\del_{\alpha}\big(\sigmab^{ab} - \sigmab_{\eta}^{ab}\big)\big|
& \lesssim  (s/t)^{-2}\zetab^{-2}\big(1+\zetab^{-2}|H|\big)|\del H| + (s/t)^{-2}\zetab^{-2}|H|
\endaligned
\end{equation}
and, moreover,  
\begin{equation}\label{eq11-31-july-2025}
\big|\delu_c\sigmab^{ab}\big|\lesssim t^{-1}\zetab^{-4}|L_cH| + s^{-1}(s/t)^{-1}|H| + s^{-1}(s/t)^{-1}\zetab^{-2}|H|^2,
\end{equation}
\begin{equation}\label{eq1-12-aout-2025}
\big|\delu_c\big(\sigmab^{ab} - \sigmab_{\eta}^{ab}\big)\big|
\lesssim s^{-1}(s/t)^{-1}\zetab^{-2}|H|_1 + s^{-1}(s/t)^{-1}\zetab^{-4}|H||H|_1.
\end{equation}
\end{lemma}

\begin{proof} 
\bse
We begin by observing that 
\begin{equation}\label{eq12-31-july-2025}
\del_{\alpha}\sigmab^{ab} = - \sigmab^{ac}\del_{\alpha}\gb_{cd}\sigmab^{db},
\quad\qquad 
\delu_e\sigmab^{ab} = - \sigmab^{ac}\delu_e\gb_{cd}\sigmab^{db}
\end{equation}
and
\begin{equation}\label{eq4'-21-july-2025}
\sigmab^{ab} - \sigmab_{\eta}^{ab} = - \sigmab_{\eta}^{ac}\Hb_{cd}\sigmab^{db}.
\end{equation}
We also have
\begin{equation}
\big|\del_{\alpha}\Hb_{ab}\big|\lesssim t^{-1}|H| + |\del H|,
\quad
\big|\del_{\alpha}\gb_{ab}\big|\lesssim t^{-1} + t^{-1}|H| + |\del H|.
\end{equation}
Then the estimate on $\del_{\alpha}\sigmab^{ab}$ in \eqref{eq10-31-july-2025} follows.
By recalling \eqref{eq2-16-june-2025}, \eqref{eq1-17-mai-2025} and observing that  
\be
\del_{\alpha}(s/t)\lesssim s^{-1}\lesssim (s/t),\quad \text{ in }  \Mcal^{\Hcal}_{[s_0,s_1]},
\ee
the inequality \eqref{eq10-31-july-2025} is immediate.

For the derivation of \eqref{eq11-31-july-2025}, we have (in particular) 
\begin{equation}
\delu_c(x^d/t) = -(x^cx^d/t^3) + \delta_{cd}t^{-1} = t^{-1}\bar{\eta}_{cd}, 
\end{equation}
and therefore
\be
\aligned
\delu_e\Hu_{cd} 
& =  (x^cx^d/t^2)\delu_eH_{00} + (x^c/t)\delu_eH_{0d} + (x^d/t)\delu_eH_{c0} + \delu_eH_{cd}
\\
& \quad +(x^c/t^2)\bar{\eta}_{ed}H_{00} + (x^d/t^2)\bar{\eta}_{ec}H_{00} + t^{-1}\bar{\eta}_{ce}H_{0d} + t^{-1}\bar{\eta}_{ed}H_{c0}
\\
& = \Big((x^cx^d/t^2)\delu_eH_{00} + (x^c/t)\delu_eH_{0d} + (x^d/t)\delu_eH_{c0} + \delu_eH_{cd}\Big)
\\
& \quad +\Big((x^c/t^2)\gb_{ed}H_{00} + (x^d/t^2)\gb_{ec}H_{00} + t^{-1}\gb_{ce}H_{0d} + t^{-1}\gb_{ed}H_{c0}\Big)
\\
& \quad - \Big((x^c/t^2)\Hb_{ed}H_{00} + (x^d/t^2)\Hb_{ec}H_{00} + t^{-1}\Hb_{ce}H_{0d} + t^{-1}\Hb_{ed}H_{c0}\Big)
\\
& = : T_1+T_2+T_3.
\endaligned
\ee
Estimating the terms $T_1$ and $T_3$ is immediate, and we can focus on $T_2$. For instance, we have 
\be
\big|\sigmab^{ac}(x^c/t^2)\gb_{ed}H_{00}\sigmab^{db}\big| = \big|t^{-1}(x^c/t)\sigmab^{ac}H_{00}\delta_e^b\big|\lesssim s^{-1}(s/t)^{-1}|H|,
\ee
while the remaining terms in $T_2$ are bounded in a similar manner. We thus arrive at the estimate \eqref{eq11-31-july-2025}. 

For the derivation of the last estimate \eqref{eq1-12-aout-2025}, thanks to \eqref{eq4'-21-july-2025} we have 
\begin{equation}\label{eq2-12-aout-2025}
\delb_e\big(\sigmab^{ab} - \sigmab_{\eta}^{ab}\big) 
= - \delb_e\big(\sigmab_{\eta}^{ac}\Hb_{cd}\sigmab^{db}\big)
= - \delb_e(\sigmab_{\eta}^{ac})\Hb_{cd}\sigmab^{db} 
- \sigmab_{\eta}^{ac}\delb_e\Hb_{cd}\sigmab^{db} 
- \sigmab_{\eta}^{ac}\Hb_{cd}\delb_e\sigmab^{db}.
\end{equation}
Then we apply \eqref{eq11-31-july-2025} and the fact
\be
\big|\delb_c\sigmab_{\eta}^{ab}\big| = \big|t^{-1}L_c\sigmab_{\eta}^{ab}\big| \lesssim s^{-1}(s/t)^{-1}. 
\ee
In particular, we need to take care of the factor $s^{-1}(s/t)^{-3}|H|^2$ arising from the first term in the right-hand side of \eqref{eq2-12-aout-2025}. We claim that this factor can be absorbed by $s^{-1}(s/t)^{-1}\zetab^{-4}|H||H|_1$. This is indeed possible, since 
\begin{equation}\label{eq3-aout-2025}
\zetab^2 = (s/t)^2+|H^{\Ncal00}|\lesssim (s/t),\quad \text{ in }  \Mcal^{\Hcal}_{[s_0,s_1]}, 
\end{equation}
by virtue of the light-bending condition \eqref{eq2-14-june-2025}.
\ese
\end{proof}


\paragraph{Controlling the frame deformation.}

\begin{lemma}\label{lem2-22-july-2025}
Assume that the uniform spacelike condition \eqref{eq-USA-condition} holds with a sufficiently small $\eps_s$. Then the frame deformation of the Cauchy adapted frame satisfies 
\begin{equation}\label{eq9'-21-july-2025}
\big|g^{\alpha\beta}\del_{\alpha}\big(\grave{\Delta}_{\beta}^c\big)\big|
\lesssim \zetab^{-2}|H|_1\sum_{k=0}^2\big((s/t)^{-2}|H|\big)^k.
\end{equation}
\end{lemma}

\begin{proof}
\bse
We first write 
\be
\aligned
(t/s)\bar{\delta}^b 
= (t/s)(\betab^b- \betab_{\eta}^b) 
& = (t/s)\Hb_{0a}\sigmab^{ab} + (t/s)\bar{\eta}_{0a}(\sigmab^{ab} - \sigmab_{\eta}^{ab})
\\
& = \Hu_{0a}\sigmab^{ab} - (x^a/t)(\sigmab^{ab} - \sigmab_{\eta}^{ab}).
\endaligned
\ee
Recall that, for the Euclidean metric, $\eta_{00} = -1, \eta_{aa} =1$, so 
\be
\grave{\Delta}_{\beta}^c = - \eta_{\gamma\beta}(x^{\gamma}/t)\big(\Hu_{0a}\sigmab^{ac} - (x^a/t)(\sigmab^{ac}- \sigmab_{\eta}^{ac})\big).
\ee
Thus we find
\be
\aligned
g^{\alpha\beta}\del_{\alpha}(\grave{\Delta}_{\beta}^c) 
& = g^{\alpha\beta}\del_{\alpha}\Big(- \eta_{\gamma\beta}(x^{\gamma}/t)\big(\Hu_{0a}\sigmab^{ac} - (x^a/t)(\sigmab^{ac}- \sigmab_{\eta}^{ac})\big)\Big)
\\
& = - \eta_{\gamma\beta}g^{\alpha\beta}(x^{\gamma}/t)\del_{\alpha}\big(\Hu_{0a}\sigmab^{ac} - (x^a/t)(\sigmab^{ac}- \sigmab_{\eta}^{ac})\big)
\\
& \quad - \eta_{\gamma\beta}g^{\alpha\beta}\big(\Hu_{0a}\sigmab^{ac} - (x^a/t)(\sigmab^{ac}- \sigmab_{\eta}^{ac})\big)\del_{\alpha}(x^{\gamma}/t)
=: T_1+T_2.
\endaligned
\ee
The term $T_2$ is easily handled, due to the decaying factor $|\del_{\alpha}(x^{\gamma}/t)|\lesssim t^{-1}$: 
\begin{equation}
|T_2|\lesssim s^{-1}(s/t)^{-1}\zetab^{-2}|H|.
\end{equation}
For the term $T_1$, we observe that
\be
\aligned
& \quad - \eta_{\gamma\beta}g^{\alpha\beta}(x^{\gamma}/t)\del_{\alpha}\big(\Hu_{0a}\sigmab^{ac} - (x^a/t)(\sigmab^{ac}- \sigmab_{\eta}^{ac})\big)
\\
& = -(x^{\alpha}/t)\del_{\alpha}\big(\Hu_{0a}\sigmab^{ac} - (x^a/t)(\sigmab^{ac}- \sigmab_{\eta}^{ac})\big)
- \eta_{\gamma\beta}H^{\alpha\beta}\del_{\alpha}\big(\Hu_{0a}\sigmab^{ac} - (x^a/t)(\sigmab^{ac}- \sigmab_{\eta}^{ac})\big)
\\
& = -(x^{\alpha}/t)(s/t)^2\del_t\big(\Hu_{0a}\sigmab^{ac} - (x^a/t)(\sigmab^{ac}- \sigmab_{\eta}^{ac})\big)
\\
& \quad - (x^{\alpha}/t)(x^a/t)\delu_a\big(\Hu_{0a}\sigmab^{ac} - (x^a/t)(\sigmab^{ac}- \sigmab_{\eta}^{ac})\big)
- \eta_{\gamma\beta}H^{\alpha\beta}\del_{\alpha}\big(\Hu_{0a}\sigmab^{ac} - (x^a/t)(\sigmab^{ac}- \sigmab_{\eta}^{ac})\big)
\\
& =:T_{11} + T_{12} + T_{13}.
\endaligned
\ee
In this calculation we applied the identity 
\be
(x^\alpha/t)\del_{\alpha} = (s/t)^2\del_t + (x^a/t)\delb_a = (s/t)^2\del_t + (x^a/t^2)L_a.
\ee
Then according to Lemma~\ref{lem1-31-july-2025}, we have 
\begin{equation}
\aligned
|T_{11}|& \lesssim  \zetab^{-2}(|H| + |\del H|) + \zetab^{-4}|H||\del H|,
\\
|T_{13}|& \lesssim  (s/t)^{-2}\zetab^{-2}(|H|^2 + |H||\del H|) + (s/t)^{-2}\zetab^{-4}|H|^2|\del H|.
\endaligned
\end{equation}
Finally, we treat $T_{12}$ by relying on \eqref{eq11-31-july-2025}: 
\begin{equation}
\big|T_{12}\big|
\lesssim s^{-1}(s/t)^{-1}\zetab^{-2}|H|_1 + s^{-1}(s/t)^{-1}\zetab^{-4}|H||H|_1.
\end{equation}
\ese
\end{proof}

}


\subsection{ Decomposition of the spinorial D'Alembert operator}
\label{section===10-2}

{ 

Next, we aim at deriving a ``line decomposition'' of the D'Alembert operator applied to spinors, as stated in in Proposition~\ref{prop2-23-july-2025}, below. We start from the identity 
\begin{equation}\label{eq6-29-july-2025}
\Box_g\Psi = g^{\alpha\beta}\nabla_{\alpha}(\nabla_{\beta}\Psi) - \nabla_W\Psi
\end{equation}
where $W := g^{\alpha\beta}\nabla_{\alpha}\del_{\beta}$ denotes the vector associated with the generalized wave gauge condition. We work in the Cauchy adapted frame and, in view of the transition relations \eqref{eq3-20-july-2025} and \eqref{eq6-20-july-2025}, we find 
\begin{equation}\label{eq11-22-july-2025}
\Box_g\Psi = - \lapsb^{-2}\nabla_{\vec{L}}(\nabla_{\vec{L}}\Psi) 
+ \sigmab^{ab}\nabla_{\delb_a}(\nabla_{\delb_b}\Psi)
+ g^{\alpha\beta}\del_{\alpha}\big(\grave{\Psi}_{\beta}^{\gamma}\big)\nabla_{\grave{v}_{\gamma}}\Psi
- \nabla_{W}\Psi.
\end{equation}
We treat first the third term in the above expression.

\begin{lemma}
\label{lem4-23-july-2025}
Assume that the uniform spacelike condition \eqref{eq-USA-condition} holds for a sufficiently small $\eps_s$. Then
one has 
\begin{equation}
g^{\alpha\beta}\del_{\alpha}\big(\grave{\Psi}_{\beta}^{\gamma}\big)\nabla_{\grave{v}_{\gamma}}\Psi
= -3s^{-1}\nabla_{\vec{L}}\Psi + R_3[H,\Psi]
\end{equation}
with
\begin{equation}\label{eq13-31-july-2025}
\aligned
|R_3[H,\Psi]|_{\vec{n}}
& \lesssim s^{-1}t^{-1}\zetab^{-1}[\Psi]_{1,1}
\\
& \quad + t^{-1}\zetab^{-2}|H|_1[\Psi]_{1,1}\sum_{k=0}^2\big((s/t)^{-2}|H|\big)^k
+ t^{-1}\zetab^{-3}|\Psi|_{\vec{n}}|H|_1\sum_{k=0}^2\big((s/t)^{-2}|H|\big)^k
\\
& \quad + \zetab^{-3}|\Psi|_{\vec{n}}|\del H||H|_1|\sum_{k=0}^2\big((s/t)^{-2}|H|\big)^k
+ s^{-2}(s/t)^{-2}\zetab^{-1}|\del H||H||\Psi|_{\vec{n}}.
\endaligned
\end{equation}
\end{lemma}

\begin{proof}
\bse
When the metric is flat, i.e., $g = \eta$, thanks to \eqref{eq7-20-july-2025},\footnote{Equivalently, we can recall (from \cite[3.10]{PLF-YM-one}) 
$
- \Box = \delb_0\delb_0 - \sum_a\delb_a\delb_a + 2\sum_a\frac{x^a}{s}\delb_0\delb_a + \frac{3}{s}\del_0$, 
which leads us to
$
- \Box = \big(\delb_0 + (x^a/s)\delb_a\big)^2 - \big(\delta^{ab} + (x^ax^b/s^2)\big)\delb_a\delb_b + (3/s)\delb_0.
$
}
\begin{equation}
\aligned
\eta^{\alpha\beta}\del_{\alpha}\big(\grave{\Psi}_{\beta}^{\gamma}\big|_{g=\eta}\big)\nabla_{\grave{v}_{\gamma}}\Psi
& =  -3s^{-1}\nabla_{\delb_s}\Psi = -3s^{-1}\nabla_{\vec{L}}\Psi -3s^{-1}\betab^c\nabla_{\delb_c}\Psi.
\endaligned
\end{equation}
When in the curved case, we recall \eqref{eq1-21-july-2025} and write 
\be
\aligned
g^{\alpha\beta}\del_{\alpha}\big(\grave{\Psi}_{\beta}^{\gamma}\big)\nabla_{\grave{v}_{\gamma}}\Psi
&  =  
\eta^{\alpha\beta}\del_{\alpha}\big(\grave{\Psi}_{\beta}^{\gamma}\big|_{\eta}\big)\nabla_{\grave{v}_{\gamma}}\Psi
+H^{\alpha\beta}\del_{\alpha}\big(\grave{\Psi}_{\beta}^{\gamma}\big|_{\eta}\big)\nabla_{\grave{v}_{\gamma}}\Psi
+ g^{\alpha\beta}\del_{\alpha}\big(\grave{\Delta}_{\beta}^{\gamma}\big)\nabla_{\grave{v}_{\gamma}}\Psi
\\
& =:  -3s^{-1}\nabla_{\delb_0}\Psi -3\betab^c\nabla_{\delb_c}\Psi + R_4[H,\Psi], 
\endaligned
\ee
where, recalling that $\grave{\Delta}_{\beta}^0 = 0$,
\be
\aligned
R_4[H,\Psi] 
& = 
H^{\alpha\beta}\del_{\alpha}\big(\grave{\Psi}_{\beta}^{\gamma}\big|_{\eta}\big)\nabla_{\grave{v}_{\gamma}}\Psi
+ t^{-1}g^{\alpha\beta}\del_{\alpha}\big(\grave{\Delta}_{\beta}^a\big)\nabla_{L_a}\Psi
\\
& =  t^{-1}g^{\alpha\beta}\del_{\alpha}\big(\grave{\Delta}_{\beta}^a\big)\widehat{L_a}\Psi 
+H^{\alpha\beta}\del_{\alpha}\big(\grave{\Psi}_{\beta}^{\gamma}\big|_{\eta}\big)\nabla_{\grave{v}_{\gamma}}\Psi
+
\frac{1}{4t}g^{\alpha\beta}\del_{\alpha}\big(\grave{\Delta}_{\beta}^a\big)g^{\mu\nu}\del_{\mu}\cdot\nabla_{\nu}L_a\cdot\Psi
\\
& =:  T_1 + T_2 + T_3.
\endaligned
\ee
The estimate of $T_1$ is based on Lemma~\ref{lem2-22-july-2025}. For $T_3$ term, we apply \eqref{eq2-31-july-2025}. Finally, for $T_2$, we observe that
\be
\aligned
T_2 
& = H^{\alpha\beta}\del_{\alpha}\big(\grave{\Psi}_{\beta}^0|_{\eta}\big)\nabla_{\vec{L}}\Psi
+ H^{\alpha\beta}\del_{\alpha}\big(\grave{\Psi}_{\beta}^a|_{\eta}\big)\nabla_{\delb_a}\Psi
\\
& =   H^{\alpha\beta}\del_{\alpha}\big(\grave{\Psi}_{\beta}^0|_{\eta}\big)(s/t)\nabla_t\Psi
-H^{\alpha\beta}\del_{\alpha}\big(\grave{\Psi}_{\beta}^0|_{\eta}\big)\betab^a\nabla_{\delb_a}\Psi
+ H^{\alpha\beta}\del_{\alpha}\big(\grave{\Psi}_{\beta}^a|_{\eta}\big)\nabla_{\delb_a}\Psi
\\
& = (s/t)H^{\alpha\beta}\del_{\alpha}\big(\grave{\Psi}_{\beta}^0|_{\eta}\big)\widehat{\del_t}\Psi
+ t^{-1}H^{\alpha\beta}\big(\del_{\alpha}\big(\grave{\Psi}_{\beta}^a|_{\eta}\big)- \del_{\alpha}\big(\grave{\Psi}_{\beta}^0|_{\eta}\big)\betab^a\big)\widehat{L_a}\Psi
\\
& \quad +\frac{1}{4}(s/t)H^{\alpha\beta}\del_{\alpha}\big(\grave{\Psi}_{\beta}^0|_{\eta}\big)g^{\mu\nu}\del_{\mu}\cdot\nabla_{\nu}\del_t\cdot\Psi
+\frac{1}{4t}H^{\alpha\beta}\big(\del_{\alpha}\big(\grave{\Psi}_{\beta}^a|_{\eta}\big)- \del_{\alpha}\big(\grave{\Psi}_{\beta}^0|_{\eta}\big)\betab^a\big)g^{\mu\nu}\del_{\mu}\cdot\nabla_{\nu}L_a\cdot\Psi.
\endaligned
\ee
At this juncture, we recall that
\begin{equation}
\big|\del_{\alpha}\big(\grave{\Psi}_{\beta}^0|_{\eta}\big)\big|\lesssim (s/t)^{-2}s^{-1},
\quad
\big|\del_{\alpha}\big(\grave{\Psi}_{\beta}^a|_{\eta}\big)\big|\lesssim (s/t)^{-3}s^{-1}
\end{equation}
and, thanks to \eqref{eq8-31-july-2025},
\begin{equation}
|T_2|_{\vec{n}}\lesssim s^{-1}(s/t)^{-1}\zetab^{-1}|H|[\Psi]_{1,1}
+ s^{-1}(s/t)^{-1}\zetab^{-1}|H||\del H||\Psi|_{\vec{n}},
\end{equation}
which is dealt with by the right-hand side of \eqref{eq13-31-july-2025}.

On the other hand, we also have 
\be
s^{-1}\betab^c\nabla_{\delb_c}\Psi = s^{-1}t^{-1}\betab^c\nabla_{L_c}\Psi 
= s^{-1}t^{-1}\widehat{L_c}\Psi 
+ \frac{1}{4st}g^{\mu\nu}\del_{\mu}\cdot\nabla_{\nu}L_a\cdot\Psi
\ee
therefore, thanks to \eqref{eq2-31-july-2025}, 
\be
\big|s^{-1}\betab^c\nabla_{\delb_c}\Psi\big|_{\vec{n}}
\lesssim s^{-1}t^{-1}\zetab^{-1}[\Psi]_{1,1} + s^{-1}\zetab^{-1}|H||\Psi|_{\vec{n}}.
\ee
It remains to observe that the last term in the right-hand side of the above estimate can be absorbed by the third term in the right-hand side of \eqref{eq13-31-july-2025}. Indeed, this holds since, in $\Mcal^{\Hcal}_{[s_0,s_1]}$, 
\be
s^{-1}\zetab^{-1}\lesssim t^{-1}\zetab^{-3} \Leftrightarrow \zeta^2\lesssim s/t\Leftarrow |H^{\Ncal00}|\lesssim (s/t). \qedhere 
\ee
\ese
\end{proof}

For the second term in the right-hand side of \eqref{eq11-22-july-2025}, we establish the following estimate.

\begin{lemma}
\label{lem3-23-july-2025}
Assume that the uniform spacelike condition \eqref{eq-USA-condition} holds with a sufficiently small $\eps_s$. Then
one has 
\begin{equation}
\aligned
\big|\sigmab^{ab}\nabla_{\delb_a}\big(\nabla_{\delb_b}\Psi\big)\big|_{\vec{n}}
& \lesssim t^{-2}\zetab^{-2}\big([\Psi]_{2,2} + \zetab^{-1}[\Psi]_{1,1}\big)
\\
& \quad + t^{-1}\zetab^{-3}|\del H|[\Psi]_{1,1}
+ t^{-1}\zetab^{-3}|\del H|_1|\Psi|_{\vec{n}}
+ \zetab^{-3}|\del H|^2|\Psi|_{\vec{n}}.
\endaligned
\end{equation}
\end{lemma}


\begin{proof}
\bse
We have 
\be
\sigmab^{ab}\nabla_{\delb_a}(\nabla_{\delb_b}\Psi)
= t^{-2}\sigmab^{ab}\nabla_{L_a}\big(\nabla_{L_b}\Psi\big)
- t^{-3}\sigmab^{ab}x^a\nabla_{L_b}\Psi
=: T_1+T_2, 
\ee
in which the term $T_1$ can be written as
\be
\aligned
T_1 & =  t^{-2}\sigmab^{ab}\nabla_{L_a}\big(\widehat{L_b}\Psi + \frac{1}{4}g^{\mu\nu}\del_{\mu}\cdot\nabla_{\nu}L_b\cdot\Psi\big)
\\
& = t^{-2}\sigmab^{ab}\nabla_{L_a}\big(\widehat{L_b}\Psi\big) 
+ \frac{1}{4}t^{-2}\sigmab^{ab}\nabla_{L_a}\big(g^{\mu\nu}\del_{\mu}\cdot\nabla_{\nu}L_b\big)\cdot\Psi 
+ \frac{1}{4}t^{-2}\sigmab^{ab}g^{\mu\nu}\del_{\mu}\cdot\nabla_{\nu}L_b\cdot\nabla_{L_a}\Psi
\\
& = t^{-2}\sigmab^{ab}\widehat{L_a}(\widehat{L_b}\Psi) 
+ \frac{1}{4}t^{-2}\sigmab^{ab}g^{\mu\nu}\del_{\mu}\cdot\nabla_{\nu}L_a\cdot\widehat{L_b}\Psi
+ \frac{1}{4}t^{-2}\sigmab^{ab}g^{\mu\nu}\del_{\mu}\cdot\nabla_{\nu}L_b\cdot\widehat{L_a}\Psi
\\
& \quad + \frac{1}{16}t^{-2}\sigmab^{ab}g^{\mu\nu}g^{\mu'\nu'}\del_{\mu}\cdot\nabla_{\nu}L_b\cdot \del_{\mu'}\cdot\nabla_{\nu'}L_a\cdot\Psi
+\frac{1}{4}t^{-2}\sigmab^{ab}\nabla_{L_a}\big(g^{\mu\nu}\del_{\mu}\cdot\nabla_{\nu}L_b\big)\cdot\Psi
\\
& = t^{-2}\sigmab^{ab}\widehat{L_a}(\widehat{L_b}\Psi) 
+\frac{1}{2}t^{-2}\sigmab^{ab}g^{\mu\nu}\del_{\mu}\cdot\nabla_{\nu}L_a\cdot\widehat{L_b}\Psi
+\frac{1}{4}t^{-2}\sigmab^{ab}\nabla_{L_a}\big(g^{\mu\nu}\del_{\mu}\cdot\nabla_{\nu}L_b\big)\cdot\Psi
\\
& \quad + \frac{1}{16}t^{-2}\sigmab^{ab}g^{\mu\nu}g^{\mu'\nu'}\del_{\mu}\cdot\nabla_{\nu}L_b\cdot \del_{\mu'}\cdot\nabla_{\nu'}L_a\cdot\Psi
\\
& =:  T_{11} + T_{12} + T_{13} + T_{14},
\endaligned
\ee
while $T_2 =: T_{21} + T_{22}$ is decomposed as  
\be
T_2 = -t^{-3}x^a\sigmab^{ab}\widehat{L_b}\Psi - \frac{1}{4}t^{-3}x^a\sigmab^{ab}g^{\mu\nu}\del_{\mu}\cdot\nabla_{\nu}L_b\cdot\Psi
=: T_{21} + T_{22}.
\ee
It is immediate that
\begin{equation}
|T_{11}|_{\vec{n}}\lesssim t^{-2}\zetab^{-2}[\Psi]_{2,2},
\qquad 
|T_{21}|_{\vec{n}}\lesssim t^{-2}\zetab^{-2}[\Psi]_{1,1}.
\end{equation}
The estimate on $T_{12}$ and $T_{22}$ relies on \eqref{eq2-31-july-2025}, namely 
\begin{equation}
\zetab\big(|T_{12}|_{\vec{n}} + |T_{22}|\big)\lesssim t^{-2}\zetab^{-2}(1+t|\del H|)[\Psi]_{1,1}.
\end{equation}
The term $T_{14}$ is bounded by \eqref{eq4-31-july-2025}:
\begin{equation}
\zetab|T_{14}|_{\vec{n}}\lesssim t^{-2}\zetab^{-2}\big(1+t^2|\del H|^2\big)|\Psi|_{\vec{n}}.
\end{equation}
Finally, the estimate for $T_{13}$ is derived in Lemma \ref{lem3-22-aout-2025}, stated next.  
\ese
\end{proof}


It thus remains one final estimate. 

\begin{lemma}\label{lem3-22-aout-2025}
Assume that the uniform spacelike condition \eqref{eq-USA-condition} holds with a sufficiently small $\eps_s$. Then
one has 
\begin{equation}
\big|\nabla_{L_a}\big(g^{\mu\nu}\del_{\mu}\cdot\nabla_{\nu}L_b\big)\cdot\Psi\big|_{\vec{n}}
\lesssim \zetab^{-1}t\big(|\del H|_1 + t|\del H|^2\big). 
\end{equation}
\end{lemma}

\begin{proof}
\bse
We choose $X=L_b$ in \eqref{eq16-16-aout-2025}, together with $L_a = Y= Y^{\alpha}\del_{\alpha}$, and obtain
\be
\aligned
Y^{\alpha}\nabla_{\alpha}\big(g^{\mu\nu}\del_{\mu}\cdot\nabla_{\nu}X\big) 
& = g^{\mu\nu}X^{\gamma}Y^{\alpha}\del_{\alpha}(\Gamma_{\nu\gamma}^{\beta})\del_{\mu}\cdot\del_{\beta}
- g^{\mu\nu}\Gamma_{\alpha\nu}^{\gamma}Y^{\alpha}\del_{\gamma}X^{\beta} \del_{\mu}\cdot\del_{\beta}
\\
& \quad + g^{\mu\nu}Y^{\alpha}\big(\del_{\nu}X^{\gamma}\Gamma_{\alpha\gamma}^{\beta} 
+ \del_{\alpha}(X^{\gamma})\Gamma_{\nu\gamma}^{\beta}
+ X^{\gamma}\Gamma_{\nu\beta}^{\delta}\Gamma_{\alpha\delta}^{\beta}\big)\del_{\mu}\cdot\del_{\beta}\\
=:  T_1+T_2.
\endaligned
\ee
It is easy to see that
\begin{equation}\label{eq5-23-aout-2025}
|T_2\cdot\Psi|_{\vec{n}}\lesssim \zetab^{-1}t\big(|\del H| + t|\del H|^2\big)|\Psi|_{\vec{n}}
\end{equation}
due to the fact that $\big|\del_{\alpha}X^{\beta}\big|\lesssim 1$, and $|X^{\beta}|,|Y^{\alpha}|\lesssim t$. Then we focus on $T_1$ and observe that
\be
\aligned
T_1 & = g^{\mu\nu}X^{\gamma}Y^{\alpha}\del_{\alpha}\Gamma_{\nu\gamma}^{\beta}\del_{\mu}\cdot\del_{\beta}
\\
& =
\frac{1}{2}g^{\mu\nu}X^{\gamma}\del_{\alpha}\big(g^{\beta\delta})(\del_{\nu}g_{\gamma\delta} + \del_{\gamma}g_{\nu\delta} - \del_{\delta}g_{\nu\gamma}\big)\del_{\mu}\cdot\del_{\beta}
\\
& \quad + \frac{1}{2}g^{\mu\nu}X^{\gamma}g^{\beta\delta}\big(\del_{\alpha}\del_{\nu}g_{\gamma\delta} 
+ \del_{\alpha}\del_{\gamma}g_{\nu\delta}
- \del_{\alpha}\del_{\delta}g_{\nu\gamma}
\big)\del_{\mu}\cdot\del_{\beta}
=:  T_{11} + T_{12}.
\endaligned
\ee
The term $T_{11}$ is quadratic in nature, and we see that
\begin{equation}\label{eq6-23-aout-2025}
|T_{11}\cdot\Psi|_{\vec{n}}\lesssim\zetab^{-1}t^2|\del H|^2. 
\end{equation}
We then turn our attention to $T_2$ and write its expression in the semi-hyperboloidal frame:
\be
\aligned
T_{12} & = \frac{1}{2}g^{\mu\nu}X^{\gamma}Y^{\alpha}g^{\beta\delta}
\big(\del_{\alpha}\del_{\nu}H_{\gamma\delta} 
+ \del_{\alpha}\del_{\gamma}H_{\nu\delta}
- \del_{\alpha}\del_{\delta}H_{\nu\gamma}
\big)\del_{\mu}\cdot\del_{\beta}
\\
& = \frac{1}{2}g^{\mu\nu}X^{\gamma}g^{\beta\delta}
\big(L_a\del_{\nu}H_{\gamma\delta} 
+ L_a\del_{\gamma}H_{\nu\delta}
- L_a\del_{\delta}H_{\nu\gamma}
\big)\del_{\mu}\cdot\del_{\beta}, 
\endaligned
\ee
therefore  
\begin{equation}\label{eq4-23-aout-2025}
|T_{12}\cdot\Psi|_{\vec{n}}\lesssim \zetab^{-1}t|\del H|_1.
\end{equation}
Summing-up \eqref{eq5-23-aout-2025}, \eqref{eq6-23-aout-2025} and \eqref{eq4-23-aout-2025}, we arrive at  the desired estimate.
\ese
\end{proof}


We are now in a position to return to \eqref{eq11-22-july-2025} and, for the first term in the right-hand side of \eqref{eq11-22-july-2025}, we write the following decomposition.

\begin{lemma}\label{lem5-23-july-2025}
Assume that the uniform spacelike condition  \eqref{eq-USA-condition} holds holds with a sufficiently small $\eps_s$. Then one has 
\begin{equation}
\aligned
- \lapsb^{-2}\nabla_{\vec{L}}\big(\nabla_{\vec{L}}\Psi\big) - 3s^{-1}\nabla_{\vec{L}}\Psi 
& =  - \lapsb^{-2}s^{-3/2}\nabla_{\vec{L}}\big(\nabla_{\vec{L}}(s^{3/2}\Psi)\big) 
+ R_1[H,\Psi]
\endaligned
\end{equation}
with
\begin{equation}\label{eq7-25-july-2025}
\aligned
|R_1[H,\Psi]|_{\vec{n}}
& \lesssim  t^{-2}\zetab^{-2}|\Psi|_{\vec{n}} + s^{-1}(s/t)^{-1}|H||\Psi|_{\vec{n}}
\\
& \quad + s^{-2}(s/t)^{-1}\zetab^{-1}|H|([\Psi]_{1,1} + \zetab^{-1}|\Psi|_{\vec{n}})
+ s^{-1}(s/t)^{-2}\zetab^{-2}|H||\del H||\Psi|_{\vec{n}}.
\endaligned
\end{equation}
\end{lemma}

\begin{proof}
\bse
By a direct calculation we have 
\be
\nabla_{\vec{L}}\big(\nabla_{\vec{L}}(s^{3/2}\Psi)\big) 
= s^{3/2}\nabla_{\vec{L}}(\nabla_{\vec{L}}\Psi) + 3s^{1/2}\nabla_{\vec{L}}\Psi
+ \frac{3}{4}s^{-1/2}\Psi, 
\ee
therefore 
\be
- \lapsb^{-2}\nabla_{\vec{L}}\big(\nabla_{\vec{L}}\Psi\big) - 3s^{-1}\nabla_{\vec{L}}\Psi 
=
- \lapsb^{-2}s^{-3/2}\nabla_{\vec{L}}\big(\nabla_{L}(s^{3/2}\Psi)\big) 
+ 3s^{-1}(\lapsb^{-2}-1)\nabla_{\vec{L}}\Psi
+ \frac{3}{4}s^{-2}\lapsb^{-2}\Psi.
\ee
Recalling \eqref{eq4-23-july-2025}, we obtain
\begin{equation}
\aligned
- \lapsb^{-2}\nabla_{\vec{L}}\big(\nabla_{\vec{L}}\Psi\big) - 3s^{-1}\nabla_{\vec{L}}\Psi 
& = - \lapsb^{-2}s^{-3/2}\nabla_{\vec{L}}\big(\nabla_{L}(s^{3/2}\Psi)\big) 
\\
& \quad -   3s^{-1}\big((s/t)^{-2}H^{\N00} + R[H]\big)\nabla_{\vec{L}}\Psi
+ \frac{3}{4}s^{-2}\lapsb^{-2}\Psi.
\endaligned
\end{equation}
Here, we also have 
\begin{equation}\label{eq6-25-july-2025}
|R[H]|\lesssim |H|,\quad (s/t)^{-2}\big|H^{\Ncal00}\big|\lesssim (s/t)^{-2}|H|,\quad
|\lapsb^{-2}|\lesssim (\zeta/\zetab)^2.
\end{equation}
We observe that
\be
\aligned
\nabla_{\vec{L}}\Psi & =  \nabla_{\delb_0}\Psi - \betab^a\nabla_{\delb_a}\Psi 
= J\nabla_t\Psi - t^{-1}\betab^a\nabla_{L_a}\Psi
\\
& = (s/t)\widehat{\del_t}\Psi - t^{-1}\betab^a\widehat{L_a}\Psi
+\frac{1}{4}(s/t)g^{\mu\nu}\del_{\mu}\cdot\nabla_{\nu}\del_t\cdot\Psi
- \frac{1}{4t}g^{\mu\nu}\betab^a\del_{\mu}\cdot\nabla_{\nu}L_a\cdot\Psi.
\endaligned
\ee
In view of \eqref{eq1-31-july-2025} and \eqref{eq2-31-july-2025}, this leads us to 
\begin{equation}\label{eq8-25-july-2025}
\big|\nabla_{\vec{L}}\Psi - (s/t)\widehat{\del_t}\Psi \big|_{\vec{n}}
\lesssim  t^{-1}\zetab^{-2}|\Psi|_{\vec{n}} + t^{-1}\zetab^{-1}[\Psi]_{1,1} 
+ \zetab^{-2}|\del H||\Psi|_{\vec{n}}, 
\end{equation}
where we used \eqref{eq8-31-july-2025}. Finally, \eqref{eq6-25-july-2025} together with \eqref{eq8-25-july-2025} implies \eqref{eq7-25-july-2025}.
\ese
\end{proof}


Now we are ready to derive the desired decomposition.

\begin{proposition}[Line decomposition of the spinorial D'Alembert operator]
\label{prop2-23-july-2025}
Assume that the uniform spacelike condition \eqref{eq-USA-condition} holds with a sufficiently small $\eps_s$. Then for any sufficiently regular solution $\Psi$ to the equation
\begin{equation}\label{eq1-24-july-2025}
\opDirac\Psi + \mathrm{i}M\Psi = \Phi, 
\end{equation}
one has 
\begin{equation}\label{eq2-26-july-2025}
\nabla_{\vec{L}}\big(\nabla_{\vec{L}}(s^{3/2}\Psi)\big)  + \lapsb^2M^2(s^{3/2}\Psi) 
=
\lapsb^2s^{3/2}\big(S[H,\Phi] + R[H,\Psi] - \nabla_W\Psi\big)
\end{equation}
with
\begin{equation}\label{eq4-27-july-2025}
S[H,\Phi] : =  s^{3/2}\big(\opDirac\Phi - \mathrm{i}M\Phi\big),
\end{equation}
and
\begin{equation}\label{eq5-25-july-2025}
\aligned
& \big|R[H,\Psi]\big|_{\vec{n}}
\\
& \lesssim 
t^{-2}\zetab^{-2}\big([\Psi]_{2,2} + \zetab^{-1}[\Psi]_{1,1}\big) 
+ |R_g||\Psi|_{\vec{n}}
\\
& \quad+ t^{-1}\zetab^{-2}|H|_1[\Psi]_{1,1}\sum_{k=0}^2\big((s/t)^{-2}|H|\big)^k
+ t^{-1}\zetab^{-3}|\Psi|_{\vec{n}}|H|_1\sum_{k=0}^2\big((s/t)^{-2}|H|\big)^k
\\
& \quad+ \zetab^{-3}|\Psi|_{\vec{n}}|\del H||H|_1|\sum_{k=0}^2\big((s/t)^{-2}|H|\big)^k
+ s^{-2}(s/t)^{-2}\zetab^{-1}|\del H||H||\Psi|_{\vec{n}}
\\
& \quad+t^{-1}\zetab^{-3}|\del H|[\Psi]_{1,1} +  t^{-1}\zetab^{-3}|\del H|_1|\Psi|_{\vec{n}}
+ s^{-2}(s/t)^{-1}\zetab^{-2}|H||\Psi|_{\vec{n}},
\endaligned
\end{equation}
where $R_g$ is scalar curvature associated with the metric $g$.
\end{proposition}

\begin{proof}
\bse
We recall Lemma~\ref{lem1-21-feb-2025} and write 
\begin{equation}\label{eq1-01-aout-2025}
- \Box_g\Psi + M^2\Psi= \opDirac\Phi - \mathrm{i}M\Phi + \frac{1}{4}R\Psi.
\end{equation}
Then \eqref{eq11-22-july-2025} implies
\begin{equation}
\aligned
& s^{-3/2}\lapsb^{-2}\big(\nabla_{\vec{L}}\big(\nabla_{\vec{L}}(s^{3/2}\Psi)\big) 
+\lapsb^2M^2(s^{3/2}\Psi)\big)
\\
& = 
S[H,\Phi] - \nabla_W\Psi 
+
R_1[H,\Psi] + \sigmab^{ab}\nabla_{\delb_a}(\nabla_{\delb_b}\Psi) + R_3[H,\Psi] 
+ \frac{1}{4}R\Psi.
\endaligned
\end{equation}
Here, in the last term, the contributions $R_1,R_3$ and $\sigmab^{ab}\nabla_{\delb_a}(\nabla_{\delb_b}\Psi)$ are bounded by Lemmas~\ref{lem4-23-july-2025}, \ref{lem3-23-july-2025}, and \ref{lem5-23-july-2025}.
\ese
\end{proof}

}


\subsection{ The orthogonal curves}
\label{section===10-3}

{ 

Next, let us study the geometry of the integral curves of the vector fields $\vec{L}$ (or $\vec{n}$, modulo an reparameterization), which we call the {\bf orthogonal curves} associated with the foliation $\bigcup_{s}\Mcal_s$. By definition, they are orthogonal everywhere to each hypersurfaces $\Mcal_s$. For any point $(t,x)$ in the interior of $\Mcal^{\Hcal}_{[s_0,s_1]}$, we denote by
\be
\gamma_{t,x}: (s- \eps,s+\eps)\mapsto \Mcal^{\Hcal}_{[s_0,s_1]}
\ee
the integral curve of $\vec{L}$ satisfying $\gamma_{t,x}(s) = (t,x)$ with $s = \sqrt{t^2-|x|^2}$. It is clear that when $g=\eta$ we have
\begin{equation}
\vec{L} = \del_s + (x^a/s)\delb_a \qquad \text{ when $g=\eta$}
\end{equation}
and, in this case,
\be
\gamma_{t,x}(\lambda) = (\lambda t/s,\lambda x/s), \qquad \text{ when $g=\eta$}
\ee

For convenience in the discussion, we also introduce the following notation. Given any point $(t,x)\in\Mcal^{\Hcal}_{[s_0,s_1]}$ we define 
\be
\aligned
s_{t,x}^* :& =\inf_{\tau\in[s_0,s]}\{\gamma_{t,x}([\tau,s])\subset\Mcal^{\Hcal}_{[s_0,s]}\}.
\\
\Mcal^{\theta}_{[s_0,s_1]} :& =\Mcal^{\Hcal}_{[s_0,s_1]}\cap\{(s/t)\leq \theta\}
\quad \text{for}\quad 0<\theta\leq 1.
\endaligned
\ee
Observe that when $\theta<\theta'$, we have the ordering property $\Mcal^{\theta}_{[s_0,s_1]}\subset\Mcal^{\theta'}_{[s_0,s_1]}$. Furthermore, when 
$
\theta < \theta_0 = \frac{2s_0}{s_0^2+1},
$
we have  
\be
(t,x)\in \Mcal^{\theta}_{[s_0,s_1]}\quad\Rightarrow\quad s = \sqrt{t^2-|x|^2}> s_0,
\ee
and
\be
\del(\Mcal^{\theta}_{[s_0,s_1]})\cap\Mcal_{s_0} = \emptyset.
\ee
For the region near the light cone, i.e., $\Mcal^{\frac{3}{5}}_{[s_0,s_1]}$, we have the following description of the orthogonal curves.

\begin{proposition}[Geometry of the orthogonal curves to the foliation]
\label{lem2-24-july-2025}
Assume that the uniform spacelike condition \eqref{eq-USA-condition} hold with a sufficiently small $\eps_s$. Furthermore, assume that
\begin{equation}\label{eq-KG-decay-condition}
|H| \leq \delta\eps_s s^{- \delta}, \qquad \delta>0.
\end{equation}
Then for any $(t,x)\in \Mcal^{\frac{3}{5}}_{[s_0,s_1]} = \Mcal^{\Hcal}_{[s_0,s_1]}\cap \{r\geq 4t/5\}$, one has 
\bse
\begin{equation}
\gamma_{t,x}([s_{t,x}^*,s])\subset \Mcal^{\frac{\sqrt{7}}{4}}_{[s_0,s_1]} = \{r\geq 3t/4\}\cap\Mcal^{\Hcal}_{[s_0,s_1]},
\end{equation}
and
\begin{equation}\label{eq5-19-june-2025}
(s/t)|_{\gamma_{t,x}(\tau)}\leq \frac{10}{9}(s/t)|_{\gamma_{t,x}(s)}.
\end{equation}
for $\tau\in[s_{t,x}^*,s]$. When $s_0=2$, $\gamma_{t,x}(s_{t,x}^*)\in\{r=t-1\}$. Furthermore, for any scalar-valued function $f$ defined in the hyperboloidal domain $\Mcal^{\Hcal}_{[s_0,s_1]}$ and satisfying
\begin{equation}
|f(t,x)|\lesssim C_fb(s/t)^a s^{-1-b}, \qquad b>0, a\in \RR
\end{equation}
and for any $(t,x)\in\Mcal^{\Hcal}_{[s_0,s_1]}$, one has the integral bound 
\begin{equation}\label{eq6-19-june-2025}
\int_{s_{t,x}^*}^s|f|_{t,x}(\lambda)\diff \lambda\lesssim C_f(s/t)^{a+b}. 
\end{equation}
\ese
\end{proposition}


\begin{proof} 
\bse
A direct calculation shows that
\be
\vec{L}(r/t) = - \frac{sr}{t^3} - \frac{s^2}{t^3}\frac{x^b}{r}\betab^b
\ee
and, consequently thanks to \eqref{eq1-13-aout-2025}, 
\be
\aligned
\vec{L}(r/t) & =  - \frac{sr}{t^3} - \frac{s^2}{t^3}\frac{x^b}{r}\lapsb^2\gb^{0b} 
= - \frac{sr}{t^3} - \frac{s^2}{t^3}\frac{x^b}{r}\lapsb^2\Big(- \frac{x^b}{s} + \Hb^{0b}\Big) 
= \frac{sr}{t^3}(\lapsb^2-1) - \frac{s^2x^b}{t^3r}\lapsb^2\Hb^{0b}. 
\endaligned
\ee
On the other hand, we observe that, thanks to \eqref{eq3-13-june-2025},
\be
\aligned
\vec{L}\big((s/t)^2\big) & =  - \vec{L}\big((r/t)^2\big) = -2(r/t)\vec{L}(r/t)
\\
& =  -2s^{-1}(s/t)^2(r/t)^2(\lapsb^2-1) + 2s^{-1}(s/t)^2\lapsb^2\Hu^{0b}(x^b/t). 
\endaligned
\ee
By integrating this latter identity along the curve $\gamma_{t,x}$, from the time $\tau$ to the tile $s$ with $s_0\leq\tau\leq s$, we obtain
\begin{equation}\label{eq5-13-june-2025}
\ln\Big(\frac{(s/t)^2|_{\gamma_{t,x}(s)}}{(s/t)^2|_{\gamma_{t,x}(\tau)}}\Big) 
= 2\int_{\tau}^s\frac{(r/t)^2(1- \lapsb^2)|_{\gamma_{t,x}(\lambda)}}{\lambda} \diff \lambda
+2\int_{\tau}^s \frac{\lapsb^2\Hu^{0b}(x^b/t)|_{\gamma_{t,x}(\lambda)}}{\lambda}\diff \lambda.
\end{equation}

It convenient to distinguish between two cases. 

\bei 

\item[$\bullet$] Assume first that $(t,x)\in\Mcal^{\frac{3}{5}}_{[s_0,s_1]}$. When $\gamma_{t,x}([\tau,s])\in\Mcal^{\frac{\sqrt{7}}{4}}_{[s_0,s_1]} = \{r\geq 3t/4\}\cap\Mcal^{\Hcal}_{[s_0,s_1]}$, we obserce that, in the right-hand side of the above expression and thanks to the light-bending assumption $H^{\Ncal00}>0$,  the first term can be written as
\be
\aligned
& \int_{\tau}^s
\bigg(\frac{-(s/t)^{-2}(r/t)^2H^{\Ncal00}}{1-(s/t)^{-2}H^{\Ncal00} - R[H]} - \frac{(r/t)^2R[H]}{1-(s/t)^{-2}H^{\Ncal00} - R[H]}\bigg)_{\gamma_{t,x}(\lambda)}\lambda^{-1}\diff \lambda
\\
& \geq - C\int_{\tau}^s|H|_{\gamma_{t,x}(\lambda)}\lambda^{-1}\diff\lambda.
\endaligned
\ee
Thanks to \eqref{eq-KG-decay-condition}, this leads us to
\begin{equation}\label{eq6-13-june-2025}
\ln\Big(\frac{(s/t)^2|_{\gamma_{t,x}(\tau)}}{(s/t)^2|_{\gamma_{t,x}(s)}}\Big) 
\leq C\int_{s_0}^{+\infty}|H|_{\gamma_{t,x}(\lambda)}\lambda^{-1}\diff \lambda\lesssim \eps_s.
\end{equation}
Provided $\eps_s$ sufficiently small, we have 
\be
(s/t)_{\gamma_{t,x}(\tau)}<\frac{\sqrt{7}}{4}.
\ee
Therefore, by a continuity argument, we can ensure that $\gamma_{t,x}(\tau)\in \Mcal^{\frac{\sqrt{7}}{4}}_{[s_0,s_1]}$ when $\tau\in[s_{t,x}^*,s]$. The estimate \eqref{eq5-19-june-2025} is guaranteed also by \eqref{eq6-13-june-2025}.

When $s_0=2$, we have $\theta_0 = 4/5> \sqrt{7}/4$, therefore $(t^*,x^*) := \gamma(s_{t,x}^*)\in\{r-1\}$. Simultaneously, we have $s_{t,x}^*/t^*\leq (10/9)(s/t)$ and, therfore, 
\begin{equation}
(s/t)\lesssim s_{t,x}^*.
\end{equation}
Thanks to \eqref{eq5-19-june-2025}, the estimate \eqref{eq6-19-june-2025} holds true thanks to the following calculation:
\be
\int_{s_{t,x}^*}^s|f|_{t,x}(\lambda)\diff \lambda
\lesssim 
\int_{s_{t,x}^*}^{+\infty}C_fb(s/t)^a|_{\gamma_{t,x}(\lambda)}\lambda^{-1-b}\diff\lambda
\lesssim (s/t)^aC_f(s_{t,x}^*)^{-b}
\lesssim C_f(s/t)^{a+b}.
\ee

\item[$\bullet$] Next, for a general point $(t,x)\in\Mcal^{\Hcal}_{[s_0,s_1]}$, we only need to point out that when $(t,x)\in\{r\leq 4t/5\}$, we have $(s/t)\geq 3/5$ and, therefore,  \eqref{eq6-19-june-2025} becomes trivial.
\eei
\ese
\end{proof}

}


\subsection{ Integration argument and conclusion}
\label{section===10-4}

{ 

Now we are ready to establish the main conclusion of the present section.

\begin{proposition}[Pointwise estimate for massive Dirac fields on a uniform spacelike spacetime]
\label{prop2-14-aout-2025}
Assume that the uniform spacelike condition \eqref{eq-USA-condition} holds with a sufficiently small $\eps_s$. For any sufficiently regular solution $\Psi$ to the massive Dirac equation
\begin{equation}\label{eq1'-24-july-2025}
\opDirac\Psi + \mathrm{i}M\Psi = \Phi, 
\end{equation}
one has 
\begin{equation}\label{eq8-26-aout-2025}
\aligned
&s^{3/2}\Big((s/t)\sum_{\alpha}|\widehat{\del_\alpha}\Psi|_{\vec{n}} + \lapsb M|\Psi|_{\vec{n}}\Big)(t,x)
\\
& \lesssim_M 
\Big(\big(\big|\nabla_{\vec{L}}(s^{3/2}\Psi)\big|_{\vec{n}}+ \lapsb M|s^{3/2}\Psi|_{\vec{n}}\Big)_{\gamma_{t,x}(s_{t,x}^*)} 
+\int_{s_{t,x}^*}^s |F|_{\vec{n},t,x}(\lambda) \diff \lambda
\\
& \quad+s^{1/2}\lapsb\big((s/t)[\Psi]_{1,1} + \lapsb|\Psi|_{\vec{n}}\big)
+ s^{1/2}[\Psi]_{1,1}
+ s^{3/2}\zetab^{-1}\lapsb|\del H||\Psi|_{\vec{n}}, 
\endaligned
\end{equation}
where $\mathcal{I}_{[s_0,s_1]} = \Mcal^{\Hcal}_{s_0}\cup\{r=t-1|s_0\leq s\leq s_1\}$ 
and
\begin{equation}
|F|_{\vec{n}}\leq s^{3/2}\lapsb^2\big|\opDirac\Phi - \mathrm{i}M\Phi\big|_{\vec{n}}
+ s^{3/2}\lapsb^2|\nabla_W\Psi|_{\vec{n}} + s^{3/2}\lapsb^2\big|R[H,\Psi]\big|_{\vec{n}}
\end{equation}
with
\begin{equation}\label{eq8-22-aout-2025}
\aligned
\big|R[H,\Psi]\big|_{\vec{n}}
& \lesssim 
t^{-2}\zetab^{-2}\big([\Psi]_{2,2} + \zetab^{-1}[\Psi]_{1,1}\big) 
+ |R_g| \, |\Psi|_{\vec{n}}
\\
& \quad+ t^{-1}\zetab^{-2}|H|_1[\Psi]_{1,1}\sum_{k=0}^2\big((s/t)^{-2}|H|\big)^k
+ t^{-1}\zetab^{-3}|\Psi|_{\vec{n}}|H|_1\sum_{k=0}^2\big((s/t)^{-2}|H|\big)^k
\\
& \quad+ \zetab^{-3}|\Psi|_{\vec{n}}|\del H||H|_1\sum_{k=0}^2\big((s/t)^{-2}|H|\big)^k
+ s^{-2}(s/t)^{-2}\zetab^{-1}|\del H||H||\Psi|_{\vec{n}}
\\
& \quad+t^{-1}\zetab^{-3}|\del H|[\Psi]_{1,1} +  t^{-1}\zetab^{-3}|\del H|_1|\Psi|_{\vec{n}} 
+ s^{-2}(s/t)^{-1}\zetab^{-2}|H||\Psi|_{\vec{n}}, 
\endaligned
\end{equation}
and $R_g$ denotes the scalar curvature of the metric $g$.
\end{proposition}

\begin{proof} {\bf 1. Derivation of the ODE.} 
\bse
For the simplicity in the calculations, we set $c^2(t,x) := \lapsb^2M^2$. 
We rewrite \eqref{eq2-26-july-2025} in the form
\be
\aligned
& \nabla_{\vec{L}}(\nabla_{\vec{L}}(s^{3/2}\Psi))  + c^2(s^{3/2}\Psi) = F, 
\\
& F:=\lapsb^2s^{3/2}\big(S[H,\Phi] + R[H,\Psi] - \nabla_W\Psi\big). 
\endaligned
\ee
Multiplying this above equation by $\vec{n}$ and then taking the Dirac form with $\nabla_{\vec{L}}(s^{3/2}\Psi)$, we obtain
\be
\la \nabla_{\vec{L}}(s^{3/2}\Psi),\vec{n}\cdot\nabla_{\vec{L}}(\nabla_{\vec{L}}(s^{3/2}\Psi))\ra_{\ourD}
+ c^2\la \nabla_{\vec{L}}(s^{3/2}\Psi),\vec{n}\cdot(s^{3/2}\Psi)\ra_{\ourD}
= \la \nabla_{\vec{L}}(s^{3/2}\Psi) ,\vec{n}\cdot F\ra_{\ourD}.
\ee
Taking then the Hermitian transpose of this later equation, we find 
\be
\la \nabla_{\vec{L}}(\nabla_{\vec{L}}(s^{3/2}\Psi)),\vec{n}\cdot\nabla_{\vec{L}}(s^{3/2}\Psi)\ra_{\ourD}
+ c^2\la (s^{3/2}\Psi),\vec{n}\cdot\nabla_{\vec{L}}(s^{3/2}\Psi)\ra_{\ourD}
= \la  F,\vec{n}\cdot\nabla_{\vec{n}}(s^{3/2}\Psi) \ra_{\ourD}.
\ee
Summing-up the two equations, we thus obtain
\begin{equation}\label{eq3-24-july-2025}
\aligned
& \vec{L}\big(|\nabla_{\vec{L}}(s^{3/2}\Psi)|_{\vec{n}}^2\big) 
+\vec{L}( c^2|(s^{3/2}\Psi)|_{\vec{n}}^2) 
\\
& = 2\Re(\la\nabla_{\vec{L}}(s^{3/2}\Psi),\vec{n}\cdot F \ra_{\ourD})
\\
& \quad+ \la\nabla_{\vec{L}}(s^{3/2}\Psi),\nabla_{\vec{L}}\vec{n}\cdot \nabla_{\vec{L}}(s^{3/2}\Psi) \ra_{\ourD}
+ \la(s^{3/2}\Psi),\nabla_{\vec{L}}(c^2\vec{n})\cdot(s^{3/2}\Psi) \ra_{\ourD}.
\endaligned
\end{equation}
We now rely on the family of orthogonal curves $\gamma_{t,x}$, that is, the integral curves of $\vec{L}$. For any scalar-valued function $v$ defined in (a subset of) $\Mcal^{\Hcal}_{[s_0,s_1]}$, we define $v_{t,x}(\lambda) :=  v\circ\gamma_{t,x}(\lambda)$. With this notation, \eqref{eq3-24-july-2025} becomes an ODE for the functions 
\be
v = \big|\nabla_{\vec{L}}(s^{3/2}\Psi)\big|_{\vec{n}}^2, 
\quad 
w:=|(s^{3/2}\Psi)|_{\vec{n}}^2,
\ee
namely 
\begin{equation}\label{eq4-25-july-2025}
\frac{\diff}{\diff\lambda}\big(v_{t,x} + c^2w_{t,x}\big)(\lambda)
= 2\Re(\la\nabla_{\vec{L}}(s^{3/2}\Psi),\vec{n}\cdot F \ra_{\ourD})\big|_{t,x}(\lambda) 
+ G_{t,x}(\lambda). 
\end{equation}
Here, we have 
\be
G = G[H,\Psi]:=\la\nabla_{\vec{L}}(s^{3/2}\Psi),\nabla_{\vec{L}}\vec{n}\cdot \nabla_{\vec{L}}(s^{3/2}\Psi) \ra_{\ourD}
+ \la(s^{3/2}\Psi),\nabla_{\vec{L}}(c^2\vec{n})\cdot(s^{3/2}\Psi) \ra_{\ourD}.
\ee
The term $G$ is bounded by Lemma~\ref{lem2-25-july-2025}, which is stated and proven next, in combination with \eqref{eq8-25-july-2025}:
\be
\aligned
\big|\la s^{3/2}\Psi,\nabla_{\vec{L}}(c^2\vec{n})\cdot(s^{3/2}\Psi)\ra_{\ourD}\big|
& \lesssim s^{1/2}(s/t)^{-1}\zetab^{-1}\lapsb^5|H|_1|\Psi|_{\vec{n}}\,\big|s^{3/2}\Psi\big|_{\vec{n}}
\\
& = s^{3/2}\lapsb\big(t^{-1}(s/t)^{-2}\zetab^{-1}\lapsb^4 |H|_1|\Psi|_{\vec{n}}\big)\,\big|s^{3/2}\Psi\big|_{\vec{n}}
=:  G_0\lapsb|s^{3/2}\Psi|_{\vec{n}}.
\endaligned
\ee
Thanks to Claim~\ref{cor1-16-june-2025}, we have
\begin{equation}\label{eq1-14-aout-2025}
\lapsb\lesssim (\zeta/\zetab^{-1})\lesssim \lapsb\lesssim 1,
\end{equation}
therefore 
\be
G_0\lesssim t^{-1}\zetab^{-3}|H|_1|\Psi|_{\vec{n}}.
\ee
In the same manner and thanks to \eqref{eq8-25-july-2025}, we find 
\be
\aligned
& \big|\la\nabla_{\vec{L}}(s^{3/2}\Psi),\nabla_{\vec{L}}\vec{n}\cdot\nabla_{\vec{L}}(s^{3/2}\Psi)\big\ra_{\ourD}\big|
\\
& \lesssim s^{1/2}(s/t)^{-1}\zetab^{-1}\lapsb^3|H|_1\big|\nabla_{\vec{L}}\Psi\big|_{\vec{n}}
\big|\nabla_{\vec{L}}(s^{3/2}\Psi)\big|_{\vec{n}}
+ s^{-1/2}(s/t)^{-1}\zetab^{-1}\lapsb^3|H|_1
\big|\Psi|_{\vec{n}}|\nabla_{\vec{L}}(s^{3/2}\Psi)\big|_{\vec{n}}
\\
& \lesssim s^{3/2}\big(s^{-1}\zetab^{-1}\lapsb^3|H|_1|\widehat{\del_t}\Psi|_{\vec{n}}
+ s^{-2}\zetab^{-3}\lapsb^3|H|_1|\Psi|_{\vec{n}} 
+s^{-2}\zetab^{-2}\lapsb^3|H|_1[\Psi]_{1,1}\big)\big|\nabla_{\vec{L}}(s^{3/2}\Psi)\big|_{\vec{n}}
\\
& \quad+ s^{3/2}\big(
s^{-1}(s/t)^{-1}\zetab^{-3}\lapsb^3|H|_1|\del H|[\Psi]_{1,1}
+s^{-2}\zetab^{-2}\lapsb^3|H|_1|\Psi|_{\vec{n}}
\big)
\big|\nabla_{\vec{L}}(s^{3/2}\Psi)\big|_{\vec{n}}
\\
& \lesssim s^{3/2}\big(t^{-1}\zetab^{-2}\lapsb^2|H|_1|\Psi|_{\vec{n}} 
+t^{-1}\zetab^{-3}\lapsb^{3}|H|_1|\Psi|_{\vec{n}}
+ t^{-1}\zetab^{-2}\lapsb^3|H|_1[\Psi]_{1,1}
\big)|\nabla_{\vec{L}}(s^{3/2}\Psi)\big|_{\vec{n}}
\\
& \quad + s^{3/2}\big(
\zetab^{-3}\lapsb^3|H|_1|\del H||\Psi|_{\vec{n}}
+s^{-2}\zetab^2\lapsb^3|H|_1|\Psi|_{\vec{n}}
\big)|\nabla_{\vec{L}}(s^{3/2}\Psi)\big|_{\vec{n}}, 
\endaligned
\ee
hence
\be
\aligned
& \big|\la\nabla_{\vec{L}}(s^{3/2}\Psi),\nabla_{\vec{L}}\vec{n}\cdot\nabla_{\vec{L}}(s^{3/2}\Psi)\big\ra_{\ourD}\big|
\\
& \lesssim s^{3/2}\big(t^{-1}\zetab^{-3}|H|_1||\Psi|_{\vec{n}} 
+ t^{-1}\zetab^{-2}|H|_1[\Psi]_{1,1} + \zetab^{-3}\lapsb^2|H|_1|\del H||\Psi|_{\vec{n}}
\big)|\nabla_{\vec{L}}(s^{3/2}\Psi)\big|_{\vec{n}}
\\
& =: G_1 \big|\nabla_{\vec{L}}(s^{3/2}\Psi)\big|_{\vec{n}}.
\endaligned
\ee
We thus rewrite \eqref{eq4-25-july-2025} into the form
\begin{equation}
\Big|\frac{\diff}{\diff \lambda}\big(v_{t,x} + cw_{t,x}\big)^{1/2}\Big| 
\leq \big(|F|_{\vec{n}}\big)_{t,x}(\lambda) +  s^{3/2}(G_1+G_0)\big|_{t,x}(\lambda).
\end{equation}
Observe that the second term in the right-hand side of the above equation can be absorbed (thanks to \eqref{eq5-25-july-2025}) into the terms $s^{3/2}R[H,\Psi]$. We thus arrive at
\begin{equation}\label{eq3-27-july-2025}
\Big|\frac{\diff}{\diff \lambda}\big(v_{t,x} + cw_{t,x}\big)^{1/2}\Big| 
\lesssim \big(|F|_{\vec{n}}\big)_{t,x}(\lambda). 
\end{equation}
\ese
%


\vskip.3cm

\bse
\noindent{\bf 2. Integration argument.} Next, for any given point $(t,x)\in\Mcal^{\Hcal}_{[s_0,s_1]}$, we integrate \eqref{eq3-27-july-2025} along $\gamma_{t,x}$ on the interval $[s_{t,x}^*,s]$ and obtain
\begin{equation}\label{eq4-27-julyt-2025}
\aligned
\big(v_{t,x} + cw_{t,x}\big)^{1/2}(s) 
& =  \big(v_{t,x} + cw_{t,x}\big)^{1/2}(s_{t,x}^*) 
+\int_{s_{t,x}^*}^s |F|_{\vec{n},t,x}(\lambda) \,\diff \lambda, 
\endaligned
\end{equation}
which leads us to
\begin{equation}\label{eq6-27-july-2025}
\aligned
& \big(\big|\nabla_{\vec{L}}(s^{3/2}\Psi)\big|_{\vec{n}} + \lapsb M|s^{3/2}\Psi|_{\vec{n}}\big)(t,x)
\\
& \lesssim_M 
\sup_{\mathcal{I}_{[s_0,s_1]}}
\Big(\big(\big|\nabla_{\vec{L}}(s^{3/2}\Psi)\big|_{\vec{n}}+ \lapsb M|s^{3/2}\Psi|_{\vec{n}}\Big) 
+ \int_{s_{t,x}^*}^s |F|_{\vec{n},t,x}(\lambda) \,\diff \lambda. 
\endaligned
\end{equation}
On the other hand, thanks to \eqref{eq8-25-july-2025}, we have 
\be
\aligned
\big|\nabla_{\vec{L}}(s^{3/2}\Psi) - s^{3/2}(s/t)\widehat{\del_t}\Psi\big|_{\vec{n}} 
& \lesssim   s^{1/2}|\Psi|_{\vec{n}} + s^{1/2}(s/t)\lapsb\big([\Psi]_{1,1} + \zetab^{-1}|\Psi|_{\vec{n}}\big)
\\
& \quad +s^{3/2}\zetab^{-1}\lapsb|\del H||\Psi|_{\vec{n}}, 
\endaligned
\ee
where we used $(s/t)\zetab^{-1}\lesssim \lapsb$. Since 
\be
\widehat{\del_a}\Psi 
= - \frac{x^a}{t}\widehat{\del_t}\Psi + t^{-1}\widehat{L_a}\Psi 
+ \frac{1}{4}g^{\mu\nu}\del_{\mu}\cdot\Big(t^{-1}\nabla_{\nu}L_a - \frac{x^a}{t}\nabla_{\nu}\del_t \Big)\cdot\Psi,
\ee
it follows that  
\begin{equation}
\aligned
s^{3/2}(s/t)\big|\widehat{\del_a}\Psi\big|_{\vec{n}}
& \lesssim s^{3/2}(s/t)\big|\widehat{\del_t}\Psi\big|_{\vec{n}} + s^{1/2}(s/t)^2[\Psi]_{1,1}
\\
& \quad 
+ s^{1/2}(s/t)\lapsb|\Psi|_{\vec{n}}
+ s^{3/2}\lapsb|\del H||\Psi|_{\vec{n}}.
\endaligned
\end{equation}
This leads us to
\begin{equation}\label{eq5-27-july-2025}
\aligned
& s^{3/2}\Big((s/t)\sum_{\alpha}|\widehat{\del_\alpha}\Psi|_{\vec{n}} + \lapsb M|\Psi|_{\vec{n}}\Big)
\\
& \lesssim \big(\big|\nabla_{\vec{n}}(s^{3/2}\Psi)\big|_{\vec{n}} + \lapsb M|s^{3/2}\Psi|_{\vec{n}}\big) 
+s^{1/2}\lapsb\big((s/t)[\Psi]_{1,1} 
\\
& \quad + \lapsb|\Psi|_{\vec{n}}\big)
+ s^{1/2}[\Psi]_{1,1} + s^{3/2}\zetab^{-1}\lapsb|\del H||\Psi|_{\vec{n}}.
\endaligned
\end{equation}
Together with \eqref{eq6-27-july-2025}, this provides us with the desired estimate, provided we also check Lemma \ref{lem2-25-july-2025}, stated next. 
\ese
\end{proof}

\begin{lemma}\label{lem2-25-july-2025}
Assume that the uniform spacelike condition \eqref{eq-USA-condition} holds for some sufficiently small $\eps_*$. Then one has 
\begin{equation}\label{eq5-26-july-2025}
\la\Psi,\nabla_{\vec{L}}\vec{n}\cdot\Psi \ra_{\ourD}\lesssim s^{-1}(s/t)^{-1}\zetab^{-1}\lapsb^3|H|_1|\Psi|_{\vec{n}}^2,
\end{equation}
\begin{equation}\label{eq1-27-july-2025}
\la\Psi,\nabla_{\vec{L}}(l\vec{L})\cdot\Psi \ra_{\ourD}\lesssim s^{-1}(s/t)^{-1}\zetab^{-1}\lapsb^5|H|_1|\Psi|_{\vec{n}}^2. 
\end{equation}
\end{lemma}

\begin{proof} 
\bse
In view of $(\nabla_{\vec{n}}\vec{n},\vec{n})_g = \frac{1}{2}\vec{n}\big((\vec{n},\vec{n})_g\big) = 0$, it is clear that the orthogonality property $\nabla_{\vec{n}}\vec{n}\perp \vec{n}$ holds. On the other hand, we have 
\begin{equation}
(\nabla_{\vec{n}}\vec{n},\delb_a)_g = -(\vec{n},\nabla_{\vec{n}}\delb_a)_g 
= (\vec{n},[\delb_a,\vec{n}])_g - (\vec{n},\nabla_{\delb_a}\vec{n})_g = (\vec{n},[\delb_a,\vec{n}])_g.
\end{equation}
A direct calculation also shows that
\be
[\delb_a,\vec{n}] = \delb_a(\lapsb^{-1})(\delb_0 - \beta^b\delb_b) - \delb_a(\betab^b)\delb_b, 
\ee
hence
\be
(\nabla_{\vec{n}}\vec{n},\delb_a)_g = - \lapsb\delb_a\big(\lapsb^{-1}\big)
= \lapsb^{-1}\delb_a\lapsb.
\ee
We thus obtain
\begin{equation}\label{eq3-26-july-2025}
\nabla_{\vec{n}}\vec{n} = \lapsb^{-1}\delb_a\lapsb\,\sigmab^{ab}\delb_b.
\end{equation}
We then apply Lemma~\ref{lem1-26-july-2025} and obtain 
\begin{equation}\label{eq2-27-july-2025}
\big|\nabla_{\vec{n}}\vec{n}\big|_{\vec{n}}\lesssim s^{-1}(s/t)^{-1}\zetab^{-1}\lapsb^2|H|_1.
\end{equation}
Again since $\vec{L} = \lapsb\vec{n}$, we obtain \eqref{eq5-26-july-2025}. 
Finally, to deal with \eqref{eq1-27-july-2025}, we simply observe that
\be
\nabla_{\vec{L}}(\lapsb^{2}\vec{n}) 
= \vec{L}(\lapsb^2)\vec{n} + \lapsb^3\nabla_{\vec{n}}\vec{n}. 
\ee 
\ese
\end{proof}

}


\section{High-order operators and commutators}
\label{section=N8}

\subsection{ High-order operators acting on products}
\label{section===11-1}

{

The following proposition gathers product rules at high order for the commuted norms $[\cdot]_{p,k}$ adapted to our Euclidean--hyperboloidal framework. Item~(1) is a Moser-type inequality: expanding $Z^{\le p}\del^{\le k}(u\Psi)$ by Leibniz and summing over all partitions of $(p,k)$ yields \eqref{eq1-17-july-2025}. Item~(2) uses the almost-flatness structure: derivatives adapted to the geometry acting on a vector field $X$ are controlled by $|X|$ up to coefficients depending on $H$, and the smallness $|H|_{[p/2]+1}\le 1$ allows these geometric coefficients to be absorbed; the factor $\zetab^{-1}$ comes from the change of frame. For $X\cdot\Psi$, we iterate Leibniz and commute Clifford multiplications using $\del_\alpha\cdot\del_\beta+\del_\beta\cdot\del_\alpha=-2\eta_{\alpha\beta}$, so that all geometric remainders are organized in terms involving $|H|$. Finally, Item~(3) follows by iterating the argument for $n$ Clifford multiplications and the same commutator scheme, producing a single $\zetab^{-1}$ loss and only polynomial dependence on $|H|$. 
The proof of the estimates below is based on \eqref{eq1-15-july-2025}, together with estimates for high-order Lie and Clifford-adapted derivatives; see Section~\ref{section=N23}.

For the convenience in the notation, we introduce
\begin{equation}
\Cubic_{p,k}[X,Y,Z] :=\sum_{p_1+p_2+p_3\leq p\atop k_1+k_2+k_3\leq k}|X|_{p_1,k_1}|Y|_{p_2,k_2}|Z|_{p_3,k_3}
\end{equation}
\begin{equation}
\Quart_{p,k}[X_1,X_2,X_3,X_4] = \sum_{p_1+p_2+p_3+p_4\leq p\atop k_1+k_2+k_3+k_4\leq k}
|X_1|_{p_1,k_1}|X_2|_{p_2,k_2}|X_3|_{p_3,k_3}|X_4|_{p_4,k_4}, 
\end{equation}
where $X_i,Y,Z$ denote scalars or vectors. When $X$ is a spinor, $|X|_{p,k}$ is understood as $[X]_{p,k}$, and we set 
\be
\Cubic_p[X,Y,Z] := \Cubic_{p,p}[X,Y,Z],
\ee
and analogue for $\Quart_{p}$.

\begin{proposition}[High-order estimates on products.]
\label{prop1-15-july-2025}
Suppose that the condition \eqref{eq-USA-condition} holds for some sufficiently small $\eps_s$.
Consider a spinor field $\Psi$ defined in (a subset of) $\Mcal_{[s_0,s_1]}$, together with a real-valued field $u$, and a vector field $X$. Let $p,k\ge 0$ be integers. 
\bei

\item[1.] If the function $u$ satisfies 
\begin{equation}\label{eq-scal-condition}
|u|_{[p/2]}\leq 1, 
\end{equation}
then one has 
\begin{equation}\label{eq1-17-july-2025}
[u\Psi]_{p,k}\lesssim_p \sum_{p_1+p_2\leq p\atop k_1+k_2\leq k}|u|_{p_1,k_1}[\Psi]_{p_2,k_2}.
\end{equation}

\item[2.] If, furthermore, the metric satisfies 
\begin{equation}\label{eq-vect-condition}
|H|_{[p/2]+1}\leq 1,
\end{equation}
then one has 
\bse
\begin{equation}\label{eq5-16-july-2025}
[X]_{p,k}\lesssim_p \zetab^{-1}|X|_{p,k} 
+ \zetab^{-1}\sum_{p_1+p_2\leq p\atop k_1+k_2\leq k}|X|_{p_1,k_1}|H|_{p_2,k_2}
\end{equation}
and, in addition, 
\begin{equation}\label{eq3-16-july-2025}
[X\cdot\Psi]_{p,k}\lesssim_p\sum_{p_1+p_2\leq p\atop k_1+k_2\leq k}\zetab^{-1}[\Psi]_{p_1,k_1}|X|_{p_2,k_2}
+\sum_{p_1+p_2+p_3\leq p\atop k_1+k_2+k_3\leq k,p_3\geq 1}\zetab^{-1}[\Psi]_{p_1,k_1}|X|_{p_2,k_2}|H|_{p_3,k_3}.
\end{equation}
\ese

\item[3.] Under the condition \eqref{eq-vect-condition}, for any integer $n\ge 1$ the following estimate also holds:
\begin{equation}\label{eq3-23-july-2025}
\big[\del_{\alpha_1}\cdot\del_{\alpha_2}\cdots\cdot\del_{\alpha_n}\cdot\Psi\big]_{p,k}
\lesssim_{n,p}\sum_{p_1+p_2\leq p\atop k_1+k_2\leq k}
\zetab^{-1}[\Psi]_{p_1,k_1}(1+|H|_{p_2,k_2}).
\end{equation}
\eei
\end{proposition}

}


\subsection{ An ordering result}
\label{section===11-2}

{ 

\paragraph{Preliminary technical estimates.}

We collect here auxiliary bounds used repeatedly in the derivation of our high-order estimates. For convenience in the discussion, we introduce a short-hand notation for the first and second (modified) derivatives of spinors, where $\widehat{\del_\alpha}$ denotes the modified derivative from~\eqref{equa-511m}, 
and which is the analogues of the scalar case: 
\bel{equa-notation-spineur}
[\del\Psi]_{p,k} = \max_{\alpha}[\widehat{\del_{\alpha}}\Psi]_{p,k},
\quad
[\del\del\Psi]_{p,k} = \max_{\alpha,\beta}\big[\widehat{\del_{\beta}}\big(\widehat{\del_{\alpha}}\Psi\big)\big]_{p,k}. 
\ee
As a preparation, we then list the deformation tensors of the basic commutation fields $\{\del_\mu,L_d,\Omega_{ab}\}$. Recall that $\pi[Z]_{\alpha\beta}:=(\mathcal{L}_Z g)_{\alpha\beta}$.

\begin{lemma}\label{lem3-06-oct-2025(l)}
When $Z = \del_{\mu}$, one has 
\begin{equation}
\pi[\del_\mu]_{\alpha\beta} = \del_\mu H_{\alpha\beta},
\quad 
\pi[\del_\delta]^{\alpha\beta} = g^{\alpha\mu}g^{\beta\nu}\del_\delta H_{\mu\nu}
\end{equation}
and, when $Z = L_d,\Omega_{ab}$,
\begin{equation}\label{eq4-18-july-2025}
\aligned
\pi[L_d]_{\alpha\beta} 
& = \, L_dH_{\alpha\beta} + \delta_{\alpha0}H_{d\beta} + \delta_{\beta0}H_{\alpha d} 
+ \delta_{\beta d}H_{\alpha 0} + \delta_{\alpha d}H_{0\beta},
\quad
\pi[L_d]^{\alpha\beta} = g^{\alpha\mu}g^{\beta\nu}\pi[L_d]_{\mu\nu}.
\\
\pi[\Omega_{ab}]_{\alpha\beta} & = \, \Omega_{ab}H_{\alpha\beta} 
+ \delta_{\alpha a}H_{b\beta} + \delta_{\beta a}H_{\alpha b}
- \delta_{\alpha b}H_{a\beta} - \delta_{\beta b}H_{\alpha a},
\quad
\pi[\Omega_{ab}]^{\alpha\beta} =  g^{\alpha\mu}g^{\beta\nu}\pi[\Omega_{ab}]_{\mu\nu}.
\endaligned
\end{equation}
\end{lemma}

We also control the inverse-metric components at high order, as follows. 

\begin{lemma}\label{lem1-11-july-2025}
Assume that \eqref{eq-USA-condition} holds with sufficiently small $\eps_s$. Suppose that \eqref{eq-vect-condition} holds. Then
\begin{equation}\label{eq4-16-july-2025}
|H^{\alpha\beta}|_{p,k}\lesssim_p  |H|_{p,k}.
\end{equation}
\end{lemma}

\begin{proof}
\bse
This is the high-order version of \eqref{eq4-26-july-2025}. We recall that, for any vector field $X$,
\begin{equation}\label{eq4-06-oct-2025(l)}
XH^{\alpha\beta} = Xg^{\alpha\beta} 
= -g^{\alpha\mu}\big(XH_{\mu\nu}\big)g^{\nu\beta}
\end{equation}
therefore, by induction,
\begin{equation}\label{eq1-11-july-2025}
\mathscr{Z}^IH^{\alpha\beta} = \sum_{k=1}^{|I|}(-1)^k\!\!\!\sum_{I_1\odot I_2\cdots\odot I_k=I}\!\!\!
g^{\alpha\mu_1}\big(\mathscr{Z}^{I_1}H_{\mu_1\nu_1}\big)g^{\nu_1\mu_2}\big(\mathscr{Z}^{I_2}H_{\mu_2\nu_2}\big)\cdots\big(\mathscr{Z}^{I_k}H_{\mu_k\nu_k}\big)g^{\nu_k\beta}.
\end{equation}
When $k=1$ and $\ord(I) = p,\rank(Z) = k$, we have 
\be
\big|g^{\alpha\mu}\big(\mathscr{Z}^IH_{\mu\nu}\big)g^{\nu\beta}\big|\lesssim |H|_{p,k}, 
\ee
provided that \eqref{eq-USA-condition} holds. When $k\geq 2$, we observe that in each product, $|I_1|+|I_2|+\cdots |I_k|=|I|$. Without loss of generality, we assume that $|I_1|\geq |I_2|\geq \cdots\geq |I_k|$. Then $\ord(I_1)\leq p, \rank(I_1)\leq k$ and $|I|_j\leq [p/2]$ for $j=2,\cdots,k$. Hence, by \eqref{eq-vect-condition}, we obtain 
\be
\big|g^{\alpha\mu_1}\big(\mathscr{Z}^{I_1}H_{\mu_1\nu_1}\big)g^{\nu_1\mu_2}\big(\mathscr{Z}^{I_2}H_{\mu_2\nu_2}\big)\cdots\big(\mathscr{Z}^{I_k}H_{\mu_k\nu_k}\big)g^{\nu_k\beta}\big|\lesssim |H|_{p,k},
\ee
which provides us with the desired result.
\ese
\end{proof}


\paragraph{Deformation tensors.}

We now quantify the size of the deformation tensors $\pi[Z]=\mathcal{L}_Z g$ in the high-order norms introduced above. Using the explicit formulas from the previous section and the smallness assumption~\eqref{eq-USA-condition}, we derive bounds for both covariant and contravariant components, valid for every commutation field $Z\in\mathscr{Z}:=\{\del_\mu,L_d,\Omega_{ab}\}$. These estimates show that each $\pi[Z]$ costs at most one additional derivative and one additional rank, up to lower-order products controlled by $H$. We first state the resulting estimates for the components of the deformation tensors. 

\begin{corollary}\label{cor1-21-aout-2025}
Assume that \eqref{eq-vect-condition} holds, and that \eqref{eq-USA-condition} holds for some sufficiently small $\eps_s$. Then, for any $Z\in\mathscr{Z}$ one has 
\begin{equation}\label{eq3-21-aout-2025}
|\pi[Z]_{\alpha\beta}|_{p,k}\lesssim_p |ZH|_{p,k} + |H|_{p,k}\lesssim |H|_{p+1,k+1},
\end{equation}
\begin{equation}\label{eq4-21-aout-2025}
\aligned
|\pi[Z]^{\alpha\beta}|_{p,k} & \lesssim_p   |ZH|_{p,k} 
+ \sum_{p_1+p_2\leq p\atop k_1+k_2\leq k}|ZH|_{p_1,k_1}|H|_{p_2,k_2}
\\
& \lesssim |H|_{p+1,k+1} 
+ \sum_{p_1+p_2\leq p\atop k_1+k_2\leq k}|H|_{p_1+1,k_1+1}|H|_{p_2,k_2}.
\endaligned
\end{equation}
In particular, when $Z = \del_{\delta}$, one finds 
\begin{equation}\label{eq1-21-aout-2025}
|\pi[\del_{\delta}]_{\alpha\beta}|_{p,k} \lesssim |\del H|_{p,k},
\end{equation}
\begin{equation}\label{eq2-21-aout-2025}
|\pi[\del_{\delta}]^{\alpha\beta}|_{p,k} \lesssim_p |\del H|_{p,k} 
+ \sum_{p_1+p_2\leq p\atop k_1+k_2\leq k}|\del H|_{p_1,k_1} |H|_{p_2,k_2}{ \,.}
\end{equation}
\end{corollary}

\begin{lemma}\label{lem2-18-july-2025}
Assume \eqref{eq-vect-condition} holds, and \eqref{eq-USA-condition} holds for some sufficiently small $\eps_s$. Then one has 
\begin{equation}
\big|[\del_{\delta},Z]\Psi\big|_{p,k}\lesssim_p [\del\Psi]_{p,k} + \zetab^{-1}\Hcom_{p}[\Psi], 
\end{equation}
in which 
\begin{equation}
\big|\Hcom_{p}[\Psi]\big|_{\vec{n}}
\lesssim_p \sum_{p_1+p_2+p_3\leq p}|H|_{p_1+1}|\del H|_{p_2}[\Psi]_{p_3}.
\end{equation}
\end{lemma}

\begin{proof} 
\bse
We rely on Proposition~\ref{prop1-22-june-2025}. In fact, when $X=\del_{\delta}$ and $Y=Z\in\mathscr{Z}$, we have 
\begin{equation}\label{eq5-16-aout-2025}
[\widehat{\del_{\delta}},\widehat{Z}]\Psi = \widehat{[\del_{\delta},Z]}\Psi 
- \frac{1}{8}\Big(g^{\alpha\beta}\pi[Z]^{\mu\nu}\pi[\del_{\delta}]_{\beta\nu}\del_{\mu}\cdot\del_{\alpha} 
+ \pi[\del_{\delta}]^{\alpha\beta}\pi[Z]_{\alpha\beta}\Big)\cdot \Psi.
\end{equation}
It is easy to check that
\be
\big[\widehat{[\del_{\delta},Z]}\Psi\big]_{p,k}\lesssim_p [\del\Psi]_{p,k}. 
\ee
We need to control the remaining cubic terms. The last term is a scalar multiplied on $\Psi$, thus can be bounded via \eqref{eq1-17-july-2025}. Using Corollary~\ref{cor1-21-aout-2025}, we obtain 
\be
\big|\pi[\del_{\delta}]^{\alpha\beta}\pi[Z]_{\alpha\beta}\big|_{p,k}\lesssim_p
\sum_{p_1+p_2\leq p\atop k_1+k_2\leq k}|\del H|_{p_1,k_1}|H|_{p_2+1,k_2+1}
\ee
and, therefore, 
\begin{equation}\label{eq14-18-july-2025}
[\pi[\del_{\delta}]^{\alpha\beta}\pi[Z]_{\alpha\beta}\cdot\Psi]_{p,k}
\lesssim_p\sum_{p_1+p_2+p_3\leq p\atop k_1+k_2+k_3\leq k}
[\Psi]_{p_1,k_1}|\del H|_{p_2,k_2}|H|_{p_3+1,k_3+1}.
\end{equation}
For the first cubic term, we set 
\be
\Phi:=g^{\alpha\beta}\pi[Z]^{\mu\nu}\pi[\del_{\delta}]_{\beta\nu}\del_{\mu}\cdot\del_{\alpha}\cdot\Psi.
\ee
In view of Lemma~\ref{lem1-11-july-2025} and Corollary~\ref{cor1-21-aout-2025}, we obtain 
\begin{equation}\label{eq1-26-aout-2025}
\, [\Phi]_{p,k}\lesssim_p\sum_{p_1+p_2+p_3\leq p\atop k_1+k_2+k_3\leq k}
|H|_{p_1+1,k_1+1}\,|\del H|_{p_2,k_2}\, [\Psi]_{p_3,k_3}. 
\end{equation}
Combining \eqref{eq1-26-aout-2025} with \eqref{eq3-23-july-2025} and the product estimate \eqref{eq1-17-july-2025} yields the desired bound. 
\ese
\end{proof}


\paragraph{Commutator calculus for modified derivatives.}

We next relate the high-order norm $[\del\Psi]_{p,k}$ to commuted modified derivatives of the form $\widehat{\del_\alpha}(\mathscr{Z}^I\Psi)$ and quantify the commutator errors. The key point is that $[\mathscr{Z}^I,\widehat{\del_\alpha}]$ produces at most lower-order terms, together with cubic remainders controlled by a term denoted $\zetab^{-1}\Hcom_{p-1}[\Psi]$, which we control below. This allows us to pass back and forth between applying $\widehat{\del_\alpha}$ before or after the admissible operators $\mathscr{Z}^I$, which will be used repeatedly in the high-order estimates. We will prove the following estimate.

\begin{proposition}\label{prop1-21-july-2025}
Assume that the conditions \eqref{eq-vect-condition} and \eqref{eq-USA-condition} hold for some sufficiently small $\eps_s$. Then for any spinor field $\Psi$ and any integers $p,k$, one has 
\begin{equation}\label{eq1-22-july-2025}
[\del\Psi]_{p,k}
\lesssim \sum_{\alpha}\sum_{\ord(I)\leq p\atop \rank(I)\leq k}\big|\widehat{\del_{\alpha}}(\mathscr{Z}^I\Psi)\big|_{\vec{n}}
+ \zetab^{-1}\Hcom_{p-1}[\Psi],\quad \text{in } \Mcal_{[s_0,s_1]}.
\end{equation}
For any type $(p,k)$ admissible operator $\mathscr{Z}^I$, one also has 
\begin{equation}\label{eq6-15-aout-2025}
\big|\widehat{\del_{\alpha}}\mathscr{Z}^I\Psi\big|_{\vec{n}}
\lesssim [\del \Psi]_{p,k} + \zetab^{-1}\Hcom_{p-1}[\Psi], 
\end{equation}
in which 
\begin{equation}
\big|\Hcom_{p}[\Psi]\big|
\lesssim_p \sum_{p_1+p_2+p_3\leq p}[\Psi]_{p_1}|H|_{p_2+1}|\del H|_{p_3}. 
\end{equation}
\end{proposition}

The proof is based on the following observations.

\begin{lemma}\label{lem1-16-aout-2025}
Assume that the conditions \eqref{eq-vect-condition} and \eqref{eq-USA-condition} hold for some sufficiently small $\eps_s$. Then for $\mathscr{Z}^I$ of  type $(p,k)$,
\begin{equation}\label{eq2-16-aout-2025}
\big|\big[\mathscr{Z}^I,\widehat{\del_{\alpha}}\big]\Psi\big|_{\vec{n}}
\lesssim_p  \sum_{\beta}\sum_{\ord(J)<p\atop \rank(J)<k}
\big|\widehat{\del_{\beta}}\mathscr{Z}^J\Psi\big|_{\vec{n}} 
+ \zetab^{-1}\Hcom_{p-1}[\Psi] \quad \text{in } \Mcal_{[s_0,s_1]},
\end{equation}
\begin{equation}\label{eq3-16-aout-2025}
\aligned
\big|[\mathscr{Z}^I,\widehat{\del_{\alpha}}]\Psi\big|_{\vec{n}}
& \lesssim_p  \sum_{\beta}\sum_{\ord(J)<p\atop \rank(J)<k}
\big|\mathscr{Z}^J\widehat{\del_{\beta}}\Psi\big|_{\vec{n}} + \zetab^{-1}\Hcom_{p-1}[\Psi]
\\
& \lesssim_p  [\del\Psi]_{p-1,k-1}+\zetab^{-1}\Hcom_{p-1}[\Psi] 
\quad\quad\quad\quad \quad\quad \text{in } \Mcal_{[s_0,s_1]}.
\endaligned
\end{equation}
\end{lemma}

\begin{proof}
\bse
We proceed by induction on $|I|$, and we first treat \eqref{eq2-16-aout-2025}. When $|I|=1$, this is direct by Proposition~\ref{prop1-22-june-2025}. In fact, we have 
\begin{equation}\label{eq4-16-aout-2025}
[\widehat{\del_{\delta}},\widehat{Z}]\Psi 
= \Gamma(Z)_{\delta }^{\gamma}\widehat{\del_{\gamma}}\Psi
- \frac{1}{8}\big(g^{\alpha\beta}\pi[Z]^{\mu\nu}\pi[\del_{\delta}]_{\beta\nu}\del_{\mu}\cdot\del_{\alpha} + \pi[\del_{\delta}]^{\alpha\beta}\pi[Z]_{\alpha\beta}\big)\cdot\Psi,
\end{equation} 
The last term is cubic and in fact can be bounded by
\be
\zetab^{-1}|\Psi|_{\vec{n}}|H|_1|\del H|\lesssim \zetab^{-1}\Hcom_0[\Psi].
\ee
Now we consider $\mathscr{Z}^{I'} = \mathscr{Z}^I\widehat{Z}$ of type $(p',k')$ with $p'=p+1,k' = k $ or $k' = k+1$ . Then
\begin{equation}\label{eq6-16-aout-2025}
[\mathscr{Z}^I\widehat{Z},\widehat{\del_{\delta}}]\Psi 
= \mathscr{Z}^I([\widehat{Z},\widehat{\del_{\delta}}]\Psi) + [\mathscr{Z}^I,\widehat{\del_{\delta}}](\widehat{Z}\Psi).
\end{equation}
For the first term, we take \eqref{eq4-16-aout-2025}:
\begin{equation}\label{eq7-16-aout-2025}
\aligned
\mathscr{Z}^I\big([\widehat{\del_{\delta}},\widehat{Z}]\Psi\big)
& = \Gamma(Z)_{\delta}^{\gamma}\mathscr{Z}^I\widehat{\del_{\gamma}}\Psi 
\\
& \quad -  \frac{1}{8}\mathscr{Z}^I\Big(\big(g^{\alpha\beta}\pi[Z]^{\mu\nu}\pi[\del_{\delta}]_{\beta\nu}\del_{\mu}\cdot\del_{\alpha} + \pi[\del_{\delta}]^{\alpha\beta}\pi[Z]_{\alpha\beta}\big)\cdot\Psi\Big)
\\
& = \Gamma(Z)_{\delta}^{\gamma}\widehat{\del_{\gamma}}\mathscr{Z}^I\Psi
+ \Gamma(Z)_{\delta}^{\gamma}[\mathscr{Z}^I,\widehat{\del_{\gamma}}]\Psi
\\
& \quad -  \frac{1}{8}\mathscr{Z}^I\Big(\big(g^{\alpha\beta}\pi[Z]^{\mu\nu}\pi[\del_{\delta}]_{\beta\nu}\del_{\mu}\cdot\del_{\alpha} + \pi[\del_{\delta}]^{\alpha\beta}\pi[Z]_{\alpha\beta}\big)\cdot\Psi\Big).
\endaligned
\end{equation}
The last term in the right-hand side of the above expression is bounded by $\zetab^{-1}\Hcom_{p}[\Psi]$ in the proof of Lemma~\ref{lem2-18-july-2025}. For the first two terms, we distinguish between two cases. When $Z = \del_{\alpha}$, $\Gamma(Z)_{\delta}^{\gamma} = 0$, and thus 
\begin{equation}
\aligned
\big|\mathscr{Z}^I\big([\widehat{Z},\widehat{\del_{\delta}}]\Psi\big)\big|_{\vec{n}}
& \lesssim_p \zetab^{-1}\Hcom_p[\Psi], 
\endaligned
\end{equation} 
which fits the induction. When $Z = L_a$, the $k' = k+1$, and we apply the induction, namely 
\begin{equation}
\aligned
\big|\mathscr{Z}^I\big([\widehat{Z},\widehat{\del_{\delta}}]\Psi\big)\big|_{\vec{n}}
& \lesssim_p  \big|\widehat{\del_{\gamma}}\mathscr{Z}^I\Psi\big|_{\vec{n}}
+ \sum_{\beta}\sum_{\ord(J)<p\atop \rank(J)<k}
\big|\widehat{\del_{\beta}}\mathscr{Z}^J\Psi\big|_{\vec{n}} 
+ \zetab^{-1}\Hcom_{p-1}[\Psi] + \zetab^{-1}\Hcom_p[\Psi]
\\
& \lesssim_p  \sum_{\beta}\sum_{\ord(J)<p\atop \rank(J)<k}
\big|\widehat{\del_{\beta}}\mathscr{L}^J\Psi\big|_{\vec{n}}
+\zetab^{-1}\Hcom_p[\Psi]
\endaligned
\end{equation}
For the second term in the right-hand side of \eqref{eq6-16-aout-2025}, we also apply the induction for order $p$:
\be
\big|[\mathscr{Z}^I,\widehat{\del_{\delta}}](\widehat{Z}\Psi)\big|_{\vec{n}}
\lesssim_p
\sum_{\beta}\sum_{\ord(J)<p\atop \rank(J)<k}
\big|\widehat{\del_{\beta}}\mathscr{Z}^J\widehat{Z}\Psi\big|_{\vec{n}} 
+ \zetab^{-1}\Hcom_{p-1}[\widehat{Z}\Psi],
\ee
which also fits the induction and we have established \eqref{eq2-16-aout-2025}.

For \eqref{eq3-16-aout-2025}, we observe that \eqref{eq4-16-aout-2025} also yields the case $|I| = 1$. Then we consider $\mathscr{Z}^{I'} = \mathscr{Z}^I\widehat{Z}$ with $\ord(I) = p, \rank(I) = k$. Then in view of \eqref{eq6-16-aout-2025}, we arrive at 
\begin{equation}\label{eq8-16-aout-2025}
\aligned
\mathscr{Z}^I\big([\widehat{\del_{\delta}},\widehat{Z}]\Psi\big)
& = \Gamma(Z)_{\delta}^{\gamma}\mathscr{Z}^I\widehat{\del_{\gamma}}\Psi 
\\
& \quad -  \frac{1}{8}\mathscr{Z}^I\Big(\big(g^{\alpha\beta}\pi[Z]^{\mu\nu}\pi[\del_{\delta}]_{\beta\nu}\del_{\mu}\cdot\del_{\alpha} + \pi[\del_{\delta}]^{\alpha\beta}\pi[Z]_{\alpha\beta}\big)\cdot\Psi\Big).
\endaligned
\end{equation} 
The last term in the right-hand side of \eqref{eq7-16-aout-2025} is also bounded by $\zetab^{-1}\Hcom_p[\Psi]$. The first term already fits the induction (here also remark that when $\Gamma(\del_{\alpha})_{\delta}^{\gamma} = 0$). For the second term in the right-hand side of \eqref{eq6-16-aout-2025}, we apply the induction on $[\mathscr{Z}^I,\widehat{\del_{\delta}}]\Psi$ and obtain
\be
\big|[\mathscr{Z}^I,\widehat{\del_{\delta}}](\widehat{Z}\Psi)\big|_{\vec{n}}
\lesssim_p \sum_{\beta}\sum_{\ord(J)<p\atop \rank(J)<k}
\big|\mathscr{Z}^J\widehat{\del_{\beta}}(\widehat{Z}\Psi)\big|_{\vec{n}} + \zetab^{-1}\Hcom_{p-1}[\Psi].
\ee
In order to bound the first term in the right-hand side of the above inequality, we observe that for any $|J|\leq p-1$,
\be
\mathscr{Z}^J\big(\widehat{\del_{\alpha}}(\widehat{Z}\Psi)\big) 
= \mathscr{Z}^J\big(\widehat{Z}(\widehat{\del_{\alpha}}\Psi)\big) 
+ \mathscr{Z}^J\big([\widehat{\del_{\alpha}},\widehat{Z}]\Psi\big).
\ee
The first term is bounded directly by $[\del\Psi]_{p,k}$. For the second term, we apply again \eqref{eq8-16-aout-2025} and obtain
\be
\aligned
\big|\mathscr{Z}^J\big([\widehat{\del_{\alpha}},\widehat{Z}]\Psi\big) \big|_{\vec{n}}
& \lesssim_p  
\Gamma(Z)_{\alpha}^{\beta}|\mathscr{Z}^J\widehat{\del_{\beta}}\Psi|_{\vec{n}} 
+ 
\big|\mathscr{Z}^J\big((g^{\alpha\beta}\pi[Z]^{\mu\nu}\pi[\del_{\delta}]_{\beta\nu}\del_{\mu}\cdot\del_{\alpha} + \pi[\del_{\delta}]^{\alpha\beta}\pi[Z]_{\alpha\beta})\cdot\Psi\big)\big|_{\vec{n}}
\\
& \lesssim_p |\mathscr{Z}^J\widehat{\del_{\beta}}\Psi|_{\vec{n}} + \zetab^{-1}\Hcom_{p-1}[\Psi].
\endaligned
\ee
We thus close the induction and conclude that \eqref{eq3-16-aout-2025} holds true.
\ese
\end{proof}


\begin{proof}[Proof of Proposition~\ref{prop1-21-july-2025}]
\bse
Based on Lemma~\ref{lem1-16-aout-2025}, for any admissible operator $\mathscr{Z}^I$ of type $(p,k)$,
\be
\big|\mathscr{Z}^I(\widehat{\del_{\alpha}}\Psi)\big|_{\vec{n}}
\lesssim_p \big|\widehat{\del_{\alpha}}(\mathscr{Z}^I\Psi)\big|_{\vec{n}} 
+ \big|[\widehat{\del_{\alpha}},\mathscr{Z}^I]\Psi\big|_{\vec{n}}.
\ee
Then by \eqref{eq2-16-aout-2025}, we obtain \eqref{eq1-22-july-2025}. In the same manner, we have 
\be
\big|\widehat{\del_{\alpha}}(\mathscr{Z}^I\Psi)\big|_{\vec{n}} 
= \big|\mathscr{Z}^I(\widehat{\del_{\alpha}}\Psi)\big|_{\vec{n}} 
+ \big|[\widehat{\del_{\alpha}},\mathscr{Z}^I]\Psi\big|_{\vec{n}}.
\ee
Then, in view of \eqref{eq3-16-aout-2025}, we obtain \eqref{eq6-15-aout-2025}.
\ese
\end{proof}

}


\subsection{ High-order estimates for Dirac commutators}
\label{section===11-3}

{ 

\paragraph{Dealing with translations, boosts, and rotations.}

For an admissible vector field $Z$, Proposition~\ref{prop1-16-july-2025} yields the decomposition 
\begin{equation}\label{eq5-21-aout-2025}
\aligned
\, [\widehat{Z},\opDirac]\Psi  
& =   - \frac{1}{2} \pi[Z]^{\alpha\beta} \del_\alpha \cdot\nabla_{\beta} \Psi
- \frac{1}{4}\pi[Z]^{\alpha\beta}\nabla_{\alpha}\del_{\beta}\cdot\Psi
- \frac{1}{4}[Z,W]\cdot\Psi
\\
& =  - \frac{1}{2}\pi[Z]^{\alpha\beta}\del_{\alpha}\cdot\widehat{\del_{\beta}}\Psi
- \frac{1}{8}g^{\mu\nu}\pi[Z]^{\alpha\beta}\del_{\alpha}\cdot\del_{\mu}
\cdot\nabla_{\nu}\del_{\beta}\cdot \Psi
\\
& \quad - \frac{1}{4}\pi[Z]^{\alpha\beta}\nabla_{\alpha}\del_{\beta}\cdot\Psi
- \frac{1}{4}[Z,W]\cdot\Psi.
\endaligned
\end{equation} 
The first term in this identity is quadratic, the last one depends on the wave gauge condition, and the remaining terms are cubic in nature. Moreover, all of these terms can be written as a single vector multiplied (via the Clifford product) by a spinor. In order to apply Proposition~\ref{prop1-15-july-2025}, we need to control their $|\cdot|_{p,k}$ norms. 

\begin{lemma}[Dealing with translations]
\label{lem2-16-july-2025}
Suppose that the conditions \eqref{eq-vect-condition} and \eqref{eq-USA-condition} hold for a sufficiently small $\eps_s$. Then, with the translation vector fields, one has the following three estimates: 
\begin{equation}\label{eq6-16-july-2025}
\zetab\big[\pi[\del_{\delta}]^{\alpha\beta}\del_{\alpha}\cdot\widehat{\del_{\beta}}\Psi\big]_{p,k}
\lesssim_p \sum_{p_1+p_2\leq p\atop k_1+k_2\leq k} 
|\del H|_{p_1,k_1}[\widehat{\del_{\beta}}\Psi]_{p_2,k_2}
+ \sum_{p_1+p_2+p_3\leq p\atop k_1+k_2+k_3\leq k}
|\del H|_{p_1,k_1}|H|_{p_2,k_2}[\widehat{\del_{\beta}}\Psi]_{p_3,k_3},
\end{equation}
\begin{equation}\label{eq7-16-july-2025}
\aligned
\zetab\big[\pi[\del_{\delta}]^{\alpha\beta}\nabla_{\alpha}\del_{\beta}\cdot\Psi\big]_{p,k}
& \lesssim_p  \sum_{p_1+p_2+p_3\leq p\atop k_1+k_2+k_3\leq k} |\del H|_{p_1,k_1}|\del H|_{p_2,k_2}[\Psi]_{p_3,k_3}
\\
& \quad + \sum_{p_1+p_2+p_3+p_4\leq p\atop k_1+k_2+k_3+k_4\leq k}|\del H|_{p_1,k_1}|\del H|_{p_2,k_2}|H|_{p_3,k_3}[\Psi]_{p_4,k_4},
\endaligned
\end{equation}
\begin{equation}\label{eq8-16-july-2025}
\aligned
\zetab\big[g^{\mu\nu}\pi[\del_{\delta}]^{\alpha\beta}\del_{\alpha}\cdot\del_{\mu}\cdot\nabla_{\nu}\del_{\beta}\cdot\Psi\big]_{p,k}
& \lesssim_p  \sum_{p_1+p_2+p_3\leq p\atop k_1+k_2+k_3\leq k}
|\del H|_{p_1,k_1}|\del H|_{p_2,k_2}[\Psi]_{p_3,k_3} 
\\
& \quad + \sum_{p_1+p_2+p_3+p_4\leq p\atop k_1+k_2+k_3+k_4\leq k}
|\del H|_{p_1,k_1}|\del H|_{p_2,k_2}|H|_{p_3,k_3}[\Psi]_{p_4,k_4}.
\endaligned
\end{equation} 
\end{lemma}

\begin{proof} 
\bse
To derive \eqref{eq6-16-july-2025}, we apply \eqref{eq3-16-july-2025} in combination with \eqref{eq2-21-aout-2025}. To establish \eqref{eq7-16-july-2025}, we write
\be
\aligned
\pi[\del_{\delta}]^{\alpha\beta}\nabla_{\alpha}\del_{\beta} 
& =  \pi[\del_{\delta}]^{\alpha\beta}\Gamma_{\alpha\beta}^{\gamma}
=\frac{1}{2}g^{\alpha\mu}g^{\beta\nu}g^{\gamma\kappa}
\big(\del_{\alpha}H_{\beta\kappa} 
+ \del_{\beta}H_{\alpha\kappa} - \del_{\kappa}H_{\alpha\beta}\big)\del_{\delta}H_{\mu\nu}\del_{\gamma},
\endaligned
\ee
so that, thanks to \eqref{eq-vect-condition} and \eqref{eq2-21-aout-2025},
\begin{equation}
|\pi[\del_{\delta}]^{\alpha\beta}\nabla_{\alpha}\del_{\beta}|_{p,k}
\lesssim_p 
\sum_{p_1+p_2\leq p\atop k_1+k_2\leq k}|\del H|_{p_1,k_1}|\del H|_{p_2,k_2} 
+ \sum_{p_1+p_2+p_3\leq p\atop k_1+k_2+k_3\leq k}|\del H|_{p_1,k_1}|\del H|_{p_2,k_2}|H|_{p_3,k_3}.
\end{equation}
Then by the same argument used for \eqref{eq6-16-july-2025}, we thus obtain \eqref{eq7-16-july-2025}. 

To derive \eqref{eq8-16-july-2025}, we write 
\be
g^{\mu\nu}\pi[\del_{\delta}]^{\alpha\beta}\del_{\alpha}\cdot\del_{\mu}\cdot\nabla_{\nu}\del_{\beta}\cdot\Psi
=
\frac{1}{2}\del_{\alpha}\cdot\del_{\mu}\cdot\del_{\gamma}\cdot \big(g^{\mu\nu}\pi[\del_{\delta}]^{\alpha\beta}g^{\gamma\kappa}
(\del_{\nu}g_{\kappa\beta} + \del_{\beta}g_{\nu\kappa} - \del_{\kappa}g_{\nu\beta})\Psi\big).
\ee
Introducing the short-hand notation
\be
\Phi := g^{\mu\nu}\pi[\del_{\gamma}]^{\alpha\beta}g^{\delta\kappa}
(\del_{\nu}g_{\kappa\beta} + \del_{\beta}g_{\nu\kappa} - \del_{\kappa}g_{\nu\beta})\Psi,
\ee
we obtain 
\begin{equation}\label{eq3-18-july-2025}
\aligned
\, [\Phi]_{p,k}\lesssim_p 
& \sum_{p_1+p_2+p_3\leq p\atop k_1+k_2+k_3\leq k}
|\del H|_{p_1,k_1}|\del H|_{p_2,k_2}[\Psi]_{p_3,k_3}
\\
& \quad +\sum_{p_1+p_2+p_3+p_4\leq p\atop k_1+k_2+k_3+k_4\leq k}
|\del H|_{p_1,k_1}|\del H|_{p_2,k_2}|H|_{p_3,k_3}[\Psi]_{p_4,k_4}. 
\endaligned
\end{equation}
On the other hand, thanks to \eqref{eq3-23-july-2025}, we have 
\begin{equation}\label{eq4-19-july-2025}
\big|\del_{\alpha}\cdot\del_{\beta}\cdot\del_{\gamma}\cdot\Phi\big|_{\vec{n}}
\lesssim_p \zetab^{-1}\sum_{p_0+p_1\leq p\atop k_0+k_1\leq k}(1+|H|_{p_0,k_0})[\Phi]_{p_1,k_1}.
\end{equation}
Combining the above two estimates, we arrive at the third inequality \eqref{eq8-16-july-2025}.
\ese
\end{proof}


Similar estimates hold for the boosts $Z = L_d$ and the rotations $Z = \Omega_{ab}$, as follows. 

\begin{lemma}[Dealing with Lorentz boosts and spatial rotations]
Suppose that the conditions \eqref{eq-vect-condition} and \eqref{eq-USA-condition} hold for a sufficiently small $\eps_s$. Then for $Z = L_d,\Omega_{ab}$ one has the three estimates 
\begin{equation}\label{eq6-18-july-2025}
\aligned
\zetab\big[\pi[Z]^{\alpha\beta}\del_{\alpha}\cdot \widehat{\del_{\beta}}\Psi\big]_{p,k} 
& \lesssim_p  \sum_{p_1+p_2\leq p\atop k_1+k_2\leq k}
|H|_{p+1,k+1}[\widehat{\del_{\beta}}\Psi]_{p_2,k_2}
\\
& \quad +\sum_{p_1+p_2+p_3\leq p\atop k_1+k_2+k_3\leq k}
|H|_{p_1+1,k_1+1}|H|_{p_2,k_2}[\widehat{\del_{\beta}}\Psi]_{p_3,k_3},
\endaligned
\end{equation}
\begin{equation}\label{eq7-18-july-2025}
\aligned
\zetab\big[\pi[Z]^{\alpha\beta}\nabla_{\alpha}\del_{\beta}\cdot \Psi\big]_{p,k}
& \lesssim_p  \sum_{p_1+p_2+p_3\leq p\atop k_1+k_2+k_3\leq k}
|H|_{p_1+1,k_1+1}|\del H|_{p_2,k_2}[\Psi]_{p_3,k_3}
\\
& \quad +\sum_{p_1+p_2+p_3+p_4\leq p\atop k_1+k_2+k_3+k_4\leq k}
|H|_{p_1+1,k_1+1}|\del H|_{p_2,k_2}|H|_{p_3,k_3}[\Psi]_{p_4,k_4},
\endaligned
\end{equation}
\begin{equation}\label{eq8-18-july-2025}
\aligned
\zetab\big[g^{\mu\nu}\pi[Z]^{\alpha\beta}\del_{\alpha}\cdot\del_{\mu}\cdot\nabla_{\nu}\del_{\beta} \cdot\Psi \big]_{p,k}
& \lesssim_p  \sum_{p_1+p_2+p_3\leq p\atop k_1+k_2+k_3\leq k}
|H|_{p_1+1,k_1+1}|\del H|_{p_2,k_2}[\Psi]_{p_3,k_3}
\\
& \quad +\sum_{p_1+p_2+p_3+p_4\leq p\atop k_1+k_2+k_3+k_4\leq k}
|H|_{p_1+1,k_1+1}|\del H|_{p_2,k_2}|H|_{p_3,k_3}[\Psi]_{p_4,k_4}. 
\endaligned
\end{equation}
\end{lemma}

\begin{proof}
\bse
The proof is similar to the one of Lemma~\ref{lem2-16-july-2025}. For \eqref{eq6-18-july-2025} we apply \eqref{eq3-16-july-2025} in combination with Corollary~\ref{cor1-21-aout-2025}. For \eqref{eq7-18-july-2025}, we have
\be
\pi[Z]^{\alpha\beta}\nabla_{\alpha}\del_{\beta} 
= \frac{1}{2}\pi[Z]^{\alpha\beta}g^{\gamma\kappa}
\big(\del_{\alpha}H_{\beta\kappa} + \del_{\beta}H_{\alpha\kappa} - \del_{\kappa}H_{\alpha\beta}\big)\del_{\gamma},
\ee
so that 
\be
\big|\pi[Z]^{\alpha\beta}\nabla_{\alpha}\del_{\beta}\big|_{p,k}
\lesssim_p \sum_{p_1+p_2\leq p\atop k_1+k_2\leq k}|H|_{p_1+1,k_1+1}|\del H|_{p_2,k_2}. 
\ee
This leads us to \eqref{eq7-18-july-2025}. 

On the other hand, the derivation of \eqref{eq8-18-july-2025} is similar to the one of \eqref{eq8-16-july-2025}. We also set 
\be
\Phi = g^{\mu\nu}\pi[Z]^{\alpha\beta}g^{\gamma\kappa}
(\del_{\nu}g_{\kappa\beta} + \del_{\beta}g_{\nu\kappa} - \del_{\kappa}g_{\nu\beta})\Psi
\ee
hence, thanks to \eqref{eq4-21-aout-2025},
\begin{equation}
[\Phi]_{p,k}\lesssim_p \sum_{p_1+p_2+p_3\leq p\atop k_1+k_2+k_3\leq k}
|H|_{p_1+1,k_1+1}|\del H|_{p_2,k_2}[\Psi]_{p_3,k_3}.
\end{equation}
This together with \eqref{eq3-23-july-2025} implies \eqref{eq8-18-july-2025}.
\ese
\end{proof} 


\paragraph{High-order estimates for Dirac commutators.} 

Finally, we reach the following conclusion.

\begin{proposition}[High-order estimates for Dirac commutators]
\label{prop1-23-july-2025}
Assume that the conditions \eqref{eq-vect-condition} and \eqref{eq-USA-condition} hold for a sufficiently small $\eps_s$. Consider any sufficiently regular spinor field $\Psi$ and any admissible vector field $Z$. Then one has 
\begin{equation}\label{eq11-18-july-2025}
\aligned
\zetab\big[[Z,\opDirac]\Psi\big]_{p,k}& \lesssim_p  
\sum_{p_1+p_2\leq p\atop k_1+k_2\leq k}[\del\Psi]_{p_1,k_1}|H|_{p_2+1,k_2+1}
+ \Hcom_{p}[\Psi] + \Hwave_{p}[\Psi]
\endaligned
\end{equation}
and, when $Z = \del_{\delta}$,
\begin{equation}\label{eq12-18-july-2025}
\aligned
\zetab\big[[\del_{\delta},\opDirac]\Psi\big]_{p,k}
& \lesssim_p   \sum_{p_1+p_2\leq p\atop k_1+k_2\leq k} 
[\del\Psi]_{p_1,k_1}|\del H|_{p_2,k_2}
+ \Hcom_{p}[\Psi] + \Hwave_{p,k}[\del,\Psi],
\endaligned
\end{equation}
in which 
\begin{equation}
\big|\Hcom_{p}[\Psi]\big|
\lesssim_p \sum_{p_1+p_2+p_3\leq p}[\Psi]_{p_1}|H|_{p_2+1}|\del H|_{p_3},
\end{equation}
and
\begin{equation}
\big|\Hwave_{p}[\Psi]\big|
\lesssim_p\sum_{p_1+p_2\leq p}[\Psi]_{p_1}|W|_{p_2+1}
+ \sum_{p_1+p_2+p_3\leq p}[\Psi]_{p_1}|W|_{p_2+1}|H|_{p_3},
\end{equation}
\begin{equation}
\big|\Hwave_{p,k}[\del,\Psi]\big|
\lesssim_p\sum_{p_1+p_2\leq p\atop k_1+k_2\leq k}[\Psi]_{p_1,k_1}|\del W|_{p_2,k_2}
+ \sum_{p_1+p_2+p_3\leq p}[\Psi]_{p_1}|\del W|_{p_2}|H|_{p_3}. 
\end{equation}
\end{proposition}

\begin{proof}
\bse\label{eqs1-07-oct-2025}
We need to bound the term depending upon the generalized wave gauge vector $W$. Observe that
\begin{equation}
|[Z,W]|_{p,k}\lesssim_p
\begin{cases}
|W|_{p+1,k+1},\quad &Z = L_d,\Omega_{ab}, 
\\
|\del W|_{p,k},\quad &Z = \del_{\delta}.
\end{cases}
\end{equation}
Hence, thanks to \eqref{eq3-16-july-2025} we find 
\be
\aligned
\zeta\big|[Z,W]\cdot\Psi\big|_{p,k}
& \lesssim_p \sum_{p_1+p_2\leq p\atop k_1+k_2\leq k}[\Psi]_{p_1,k_1}\big|[Z,W]\big|_{p_2,k_2}
+ \sum_{p_1+p_2+p_3\leq p\atop k_1+k_2+k_3\leq k}[\Psi]_{p_1,k_1}\big|[Z,W]\big|_{p_2,k_2}|H|_{p_3,k_3}
\\
& \lesssim_p 
\left\{
\aligned
& \sum_{p_1+p_2\leq p\atop k_1+k_2\leq k}[\Psi]_{p_1,k_1}|W|_{p_2+1,k_2+1}
+ \sum_{p_1+p_2+p_3\leq p\atop k_1+k_2+k_3\leq k}[\Psi]_{p_1,k_1}|W|_{p_2+1,k_2+1}|H|_{p_3,k_3},
\\
& \sum_{p_1+p_2\leq p\atop k_1+k_2\leq k}[\Psi]_{p_1,k_1}|\del W|_{p_2,k_2}
+ \sum_{p_1+p_2+p_3\leq p\atop k_1+k_2+k_3\leq k}[\Psi]_{p_1,k_1}\big|\del W|_{p_2,k_2}|H|_{p_3,k_3}.
\endaligned
\right.
\endaligned
\ee
The estimate \eqref{eq12-18-july-2025} is then direct in view of Lemma~\ref{lem2-16-july-2025}. We observe also that it can be covered by the right-hand side of \eqref{eq11-18-july-2025}.
\ese
\end{proof}

}


\subsection{ High-order commutator estimates for the Dirac operator}
\label{section===11-4}

{ 

\begin{proposition}[High-order commutator estimates for the Dirac operator]
\label{prop1-19-july-2025}
Assume that the metric satisfies 
\begin{equation}\label{eq-com-condition}
[H]_{[p/2]+1}\lesssim \zeta
\end{equation}
and the condition \eqref{eq-USA-condition} holds for a sufficiently small $\eps_s$. Then for all $I$ of type $(p,k)$, one has 
\begin{equation}
\aligned
\zetab|[\mathscr{Z}^I,\opDirac]\Psi|_{\vec{n}}
& \lesssim_p  \sum_{p_1+p_2\leq p\atop k_1+k_2\leq k}[\del\Psi]_{p_1-1,k_1-1}|H|_{p_2+1,k_2+1} 
+\sum_{p_1+p_2\leq p\atop k_1+k_2\leq k}[\del\Psi]_{p_1-1,k_1}|\del H|_{p_2,k_2}
\\
& \quad + \Hcom_{p-1}[\Psi] 
+ \Hwave_{p-1}[\Psi],
\endaligned
\end{equation}
in which the first term vanishes when $k_1=0$ and the second term vanishes when $p_1=k_1$.
\end{proposition}

\begin{proof} 
\bse
Observe that \eqref{eq-com-condition} covers directly \eqref{eq-vect-condition}. We proceed by induction on the length $|I|$. When $|I| = 1$, we have $\ord(I) = 1$ and $\rank(I)=1$ or $\rank(I)=0$. These two cases are checked easily in view of \eqref{eq11-18-july-2025} and \eqref{eq12-18-july-2025}. Next, we consider an admissible operator $\mathscr{Z}^{I'} = \mathscr{Z}^IZ$ with $I$ of type $(p,k)$. We distinguish between two cases, as follows. 

\vskip.15cm

$\bullet$ When $Z = \del_{\delta}$, then $I'$ is of type $(p+1,k)$ and we have 
\be
\aligned
\zetab\big|[\mathscr{Z}^IZ,\opDirac]\Psi\big|_{\vec{n}} 
& \lesssim \zetab\big|\mathscr{Z}^I\big([\del_{\delta},\opDirac]\Psi\big)\big|_{\vec{n}} 
+ \zetab\big|[\mathscr{Z}^I,\opDirac](\del_{\delta}\Psi)\big|_{\vec{n}}
\\
& \lesssim_p \zetab\big[[\del_{\delta},\opDirac]\Psi\big]_{p,k} 
+ \zetab\big|[\mathscr{Z}^I,\opDirac](\del_{\delta}\Psi)\big|_{\vec{n}} =: A, 
\endaligned
\ee
with 
\be
\aligned
A & \lesssim_p  \sum_{p_1+p_2\leq p\atop k_1+k_2\leq k}[\del\Psi]_{p_1,k_1}|\del H|_{p_2,k_2}
+ \Hcom_p[\Psi] + \Hwave_{p,k}[\del,\Psi]
\\
& \quad + \sum_{p_1+p_2\leq p\atop k_1+k_2\leq k}[\del\del\Psi]_{p_1-1,k_1-1}|H|_{p_2+1,k_2+1}
+\sum_{p_1+p_2\leq p\atop k_1+k_2\leq k}[\del\del\Psi]_{p_1-1,k_1}|\del H|_{p_2,k_2}
\\
& \quad + \Hcom_{p-1}[\del_{\delta}\Psi] + \Hwave_{p-1}[\del_{\delta}\Psi]
\\
& = \sum_{p_1'+p_2\leq p+1\atop k_1+k_2\leq k}[\del\Psi]_{p_1'-1,k_1-1}|H|_{p_2+1,k_2+1}
+\sum_{p_1'+p_2\leq p+1\atop k_1+k_2\leq k}[\del\Psi]_{p_1'-1,k_1}|\del H|_{p_2,k_2}
\\
& \quad + \Hcom_{p}[\Psi] + \Hwave_p[\Psi]
+ \Hcom_{p-1}[\del_{\delta}\Psi] + \Hwave_{p-1}[\del_{\delta}\Psi].
\endaligned
\ee
We also observe that, for $Z\in\mathscr{Z}$,
\begin{equation}\label{eq5-19-july-2025}
\big|\Hcom_{p}[Z\Psi]\big|\lesssim\Hcom_{p+1}[\Psi],
\end{equation}
\begin{equation}\label{eq3-19-july-2025}
\big|\Hwave_{p}[Z\Psi]\big|\lesssim \Hwave_{p+1}[\Psi].
\end{equation}
This concludes the induction argument.
\ese

\vskip.15cm

\bse
\noindent$\bullet$ When $Z = L_d,\Omega_{ab}$, then $I'$ is of type $(p+1,k+1)$ and we write 
\begin{equation}\label{eq13-18-july-2025}
\aligned
\zetab\big|[\mathscr{Z}^IZ,\opDirac]\Psi\big|_{\vec{n}} 
& \lesssim \zetab\big|\mathscr{Z}^I\big([Z,\opDirac]\Psi\big)\big|_{\vec{n}} 
+ \zetab\big|[\mathscr{Z}^I,\opDirac](Z\Psi)\big|_{\vec{n}}
\\
& \lesssim_p \zetab\big[[Z,\opDirac]\Psi\big]_{p,k} 
+ \zetab\big|[\mathscr{Z}^I,\opDirac](Z\Psi)\big|_{\vec{n}}
\\
& \lesssim_p  \sum_{p_1+p_2\leq p\atop k_1+k_2\leq k}[\del\Psi]_{p_1,k_1}|H|_{p_2+1,k_2+1}
+ \Hcom_{p}[\Psi] + \Hwave_{p}[\Psi]
\\
& \quad + \sum_{p_1+p_2\leq p\atop k_1+k_2\leq k}[\del Z\Psi]_{p_1-1,k_1-1}|H|_{p_2+1,k_2+1}
+\sum_{p_1+p_2\leq p\atop k_1+k_2\leq k}[\del Z\Psi]_{p_1-1,k_1}|\del H|_{p_2,k_2}
\\
& \quad + \Hcom_{p-1}[Z\Psi] + \Hwave_{p-1}[Z\Psi].
\endaligned
\end{equation}
In order to bound the fourth term in the right-hand side of the above expression, we apply Lemma~\ref{lem2-18-july-2025} and obtain
$$
[\del Z\Psi]_{p,k}\leq [Z\del \Psi]_{p,k} + C[\del\Psi]_{p,k} + C\zetab^{-1}\Hcom_p[\Psi], 
$$
where $C$ is constant depending upon $p$. Thus we find
\begin{equation}\label{eq3-09-oct-2025(l)}
[\del Z\Psi]_{p,k}\lesssim_p [\del\Psi]_{p+1,k+1} + \zetab^{-1}\Hcom_p[\Psi].
\end{equation}
Substituting this result into \eqref{eq13-18-july-2025}, we obtain
\begin{equation}\label{eq2-19-july-2025}
\aligned
\zetab\big|[\mathscr{Z}^IZ,\opDirac]\Psi\big|_{\vec{n}}
& \lesssim_p  \sum_{p_1'+p_2\leq p+1\atop k_1'+k_2\leq k+1}
[\del\Psi]_{p_1'-1,k_1'-1}|H|_{p_2+1,k_2+1} 
+ \Hcom_{p}[\Psi] + \Hwave_{p}[\Psi]
\\
& \quad + \sum_{p_1'+p_2\leq p+1\atop k_1'+k_2\leq k+1}[\del \Psi]_{p_1'-1,k_1'-1}|H|_{p_2+1,k_2+1}
+\zetab^{-1}\sum_{p_1+p_2\leq p}\Hcom_{p_1}[\Psi]|H|_{p_2+1}
\\
& \quad +\sum_{p_1'+p_2\leq p+1\atop k_1'+k_2\leq k+1}[\del \Psi]_{p_1'-1,k_1'}|\del H|_{p_2,k_2}
+ \zetab^{-1}\sum_{p_1+p_2\leq p}\Hcom_p[\Psi]|\del H|_{p_2}
\\
& \quad + \Hcom_{p-1}[Z\Psi] + \Hwave_{p-1}[Z\Psi].
\endaligned
\end{equation}
Next, thanks to \eqref{eq-com-condition} we have 
\begin{equation}\label{eq1-19-july-2025}
\zetab^{-1}\sum_{p_1+p_2\leq p}\Hcom_{p_1}[\Psi]\big(|H|_{p_2+1}+|\del H|_{p_2}\big)
\lesssim_p \sum_{p_1+p_2+p_3\leq p}[\Psi]_{p_1}|H|_{p_2+1}|\del H|_{p_3}.
\end{equation}
Substituting the above estimate together with \eqref{eq5-19-july-2025} and \eqref{eq3-19-july-2025}, we conclude the induction.
\ese
\end{proof}


For an application later on, we also state the following property.

\begin{lemma}[Top-order commutator estimates for the Dirac operator]
\label{lem1-09-oct-2025}
Assume that the metric satisfies \eqref{eq-com-condition}
and that the condition \eqref{eq-USA-condition} holds for a sufficiently small $\eps_s$. Then for all $I$ of type $(p,p)$, i.e., $\mathscr{Z}$ is purely composed by $L_a,\Omega_{ab}$ and does not contain $\del_{\alpha}$, one has 
\begin{equation}\label{eq2-09-oct-2025}
\aligned
\zetab|[\mathscr{Z}^I,\opDirac]\Psi|_{\vec{n}}
& \lesssim_p  \sum_{p_1+p_2\leq p\atop k_1+k_2\leq k}[\del\Psi]_{p_1-1,k_1-1}|H|_{p_2+1,k_2+1}  + \Hcom_{p-1}[\Psi] 
+ \Hwave_{p-1}[\Psi].
\endaligned
\end{equation}
\end{lemma}

\begin{proof}
This is proven by induction. When $p=1$, i.e., $\mathscr{Z}^I = L_a$ or $\Omega_{ab}$. Then \eqref{eq2-09-oct-2025} can be checked in view of \eqref{eq5-21-aout-2025} and \eqref{eq3-21-aout-2025}. Consider an admissible operator $\mathscr{Z}^{I'} = \mathscr{Z}^IZ$ with $Z = L_a$ or $\Omega_{ab}$, and assume that \eqref{eq2-09-oct-2025} is valid for $\mathscr{Z}^I$. Then thanks to \eqref{eq11-18-july-2025}, we have
\be
\aligned
\zetab[\mathscr{Z}^{I'},\opDirac]\Psi & = 
\zetab\big|[\mathscr{Z}^IZ,\opDirac]\Psi\big|_{\vec{n}} 
\lesssim \zetab\big|\mathscr{Z}^I\big([Z,\opDirac]\Psi\big)\big|_{\vec{n}} 
+ \zetab\big|[\mathscr{Z}^I,\opDirac](Z\Psi)\big|_{\vec{n}}
\\
& \lesssim_p  \zetab\big[[Z,\opDirac]\Psi\big]_{p,k} 
+ \zetab\big|[\mathscr{Z}^I,\opDirac](Z\Psi)\big|_{\vec{n}}
\\
& \lesssim_p  \sum_{p_1+p_2\leq p\atop k_1+k_2\leq k}[\del\Psi]_{p_1,k_1}|H|_{p_2+1,k_2+1}
+ \Hcom_{p}[\Psi] + \Hwave_{p}[\Psi] 
+ \zetab\big|[\mathscr{Z}^I,\opDirac](Z\Psi)\big|_{\vec{n}}.
\endaligned
\ee
For the last term we just apply the assumption of induction together with \eqref{eq5-19-july-2025}, \eqref{eq3-19-july-2025} and \eqref{eq3-09-oct-2025(l)}.
\end{proof}

}


\section{Decomposition of Dirac commutator in the near-light-cone region}
\label{section=N9}

\subsection{Estimates from generalized wave coordinate conditions}

We recall \eqref{eq6-05-oct-2025} for the definition of $\Mcal^{\near}_{\ell,[s_0,s_1]}$. Since in this region, $\del H$ will not provide sufficient decay, we need a more delicate analysis on the Dirac commutator. For simplicity in the notation, we introduce
\be
\Mcal^{\near}_{[s_0,s_1]} = \Mcal^{\near}_{\frac{1}{4},[s_0,s_1]} = \Mcal^{\ME}_{[s_0,s_1]}\cap\{r\geq 4t/3\}.
\ee

We introduce the frame $\{\delS_{\alpha}\}$ in the domain $\Mcal^{\near}_{[s_0,s_1]}$:
\begin{equation}\label{eq8-05-oct-2025}
\delS_0 = \del_t,\quad \delS_1 = \vec{V} = \del_t+\del_r,\quad \delS_{2} = \vec{S}_2,\quad \delS_3= \vec{S}_3
\end{equation}
where $\vec{S}_i$ are smooth orthonormal vectors of the sphere $\mathbb{S}^2$ with respect to the Euclidean metric. 

It is clear that
\begin{equation}\label{eq7-05-oct-2025}
\vec{V} = \frac{t-r}{t}\del_t + t^{-1}\frac{x^a}{r}L_a.
\end{equation}
For convenience of discussion, we define 
\be
\Omega_a = \Omega_{23},\quad \Omega_2 = \Omega_{31},\quad \Omega_3 = \Omega_{12}.
\ee
We also observe that $\Omega_{ab} = (x^a/t)L_b - (x^b/t)L_a$, so
\begin{equation}\label{eq9-05-oct-2025}
\vec{S}_i = t^{-1}S_i^aL_a, 
\end{equation}
where $S_i^a$ are smooth functions defined on $\mathbb{S}^2$ (and, therefore, are uniformly bounded). 

We denote by $\Phi^{\Scal}$ and $\PsiS$ the transition matrices in the following sens:
\begin{equation}
\delS_{\alpha} = \PhiS_{\alpha}^{\beta}\del_{\beta},\quad \del_{\alpha} = \PsiS_{\alpha}^{\beta}\delS_{\beta}.
\end{equation}
Form \eqref{eq8-05-oct-2025} (for $\vec{V},\del_t$) and \eqref{eq9-05-oct-2025} (for $\vec{S}_i$), we remark that $\PhiS_{\alpha}^{\beta}$ and $\PsiS_{\alpha}^{\beta}$ are homogeneous functions of degree zero defined in the conical region $\{3t/4\leq  r\leq 4t/3\}$. 

For any two-tensor $T_{\alpha\beta}\diff x^{\alpha}\otimes\diff x^{\beta}$, we define 
\be
T^{\Scal}_{\alpha\beta} := T(\delS_{\alpha},\delS_{\beta}). 
\ee
Especially, the spacetime metric $g$ can be expressed in this frame. Their components are written as $\gScal_{\alpha\beta}$ and $\gScal^{\alpha\beta}$. We remark the following estimate,
\begin{lemma}\label{lem1-06-oct-2025}
Assume that \eqref{eq-USA-condition} holds with sufficiently small $\eps_s$. Then in $\Mcal^{\near}_{[s_0,s_1]}$,
\begin{equation}
|\HScal_{\alpha\beta}|_{p,k} + |\HScal^{\alpha\beta}|_{p,k}\lesssim |H|_{p,k},
\end{equation}
where $H_{\alpha\beta} := g_{\alpha\beta}- \eta_{\alpha\beta}$.
\end{lemma}
\begin{proof}
This is because of the homogeneity of $\PhiS_{\alpha}^{\beta}$ and $\PsiS_{\alpha}^{\beta}$. The estimate on $\HScal^{\alpha\beta}$ is similar to Lemma~\ref{lem1-11-july-2025}. We only need to remark that the Minkowski metric
\begin{equation}\label{eq1-06-oct-2025}
\big({\eta_{\Scal}}_{\alpha\beta}\big)
=\left(\begin{array}{cccc}
-1 & \quad -1 &0 &0
\\
-1 &0 &0 &0
\\
0 &0 &1 &0
\\
0 &0 &0 &1
\end{array}
\right)
\end{equation}
is invertible, and its inverse remain bounded under small uniform perturbations (in the pointwise sense).
\end{proof}


Moreover, in view of \eqref{eq1-06-oct-2025} we have 
\begin{equation}
\aligned
& \gScal_{0i} = \HScal_{0i},\quad \gScal_{1i} = \HScal_{1i},\quad i=2,3,
\\
& \gScal_{ij} = \HScal_{ij},\quad i=2,j=3,
\\
& \gScal_{11} = \HScal_{11}.
\endaligned
\end{equation}
For the simplicity of expression, we write
$
\delSs 
$
for $\delS_a, a>0$. Then we have the following estimate.

\begin{lemma}\label{lem2-06-oct-2025}
In $\Mcal^{\near}_{[s_0,s_1]}$, for all sufficiently regular functions $u$, one has
\begin{equation}\label{eq3-06-oct-2025}
|\delSs u|\lesssim \frac{\la t-r\ra}{t}|\del u| + t^{-1}|Lu| 
\end{equation}
and, more generally,
\begin{equation}
|\delSs u|_{p,k}\lesssim \frac{\la r-t\ra}{r}|\del u|_{p,k} + |\delsN u|_{p,k}.
\end{equation}
\end{lemma}

\begin{proof}
For $\vec{S}_i$, we only need to remark that $r^{-1}S_i^a$ are homogeneous of degree $(-1)$ in $\{3t/4\leq r\leq 4t/3\}$, and $\Omega_{ab} = \frac{x^a}{t}L_b - \frac{x^b}{t}L_a$. For $\vec{V}$, we recall the following estimate (cf. \cite[Lemma~6.7]{PLF-YM-PDE}, and remark that any admissible operator can be written as a finite linear combination of  $\del^IL^J$ with homogeneous coefficients of degree zero in $\{3t/4\leq r\leq 4t/3\}$):
\begin{equation}
\Big|\frac{t-r}{t}\Big|_p\lesssim_p\frac{\la r-t\ra}{t},\quad \text{ in }  \Mcal^{\near}_{[s_0,s_1]}.
\end{equation}
\end{proof}

We write the generalized wave coordinate conditions $\Gamma^{\lambda} = W^{\lambda}$ into the following form:
\begin{equation}
g_{\beta\delta} \del_{\alpha}H^{\alpha\beta}
= - \frac{1}{2}g^{\alpha\beta}\del_{\delta}H_{\alpha\beta} - W_{\delta}.
\end{equation}
We write this decomposition in the frame $\{\delS_{\alpha}\}$:
\begin{equation}
\aligned
\gScal_{\nu\lambda} \delS_{\mu}\HScal^{\mu\nu} 
& =- \frac{1}{2}\gScal^{\mu\nu}\delS_{\lambda}\HScal_{\mu\nu}
- {W_{\Scal}}_{\lambda}
\\
& \quad - \PhiS_{\lambda}^{\delta}g_{\beta\delta}\del_{\alpha}\big(\PhiS_{\mu}^{\alpha}\PhiS_{\nu}^{\beta}\big)\HScal^{\mu\nu} 
- \frac{1}{2}g^{\alpha\beta}\delS_{\lambda}\big(\PsiS_{\alpha}^{\mu}\PsiS_{\beta}^{\nu}\big)\HScal_{\mu\nu}.
\endaligned
\end{equation}
This can be written as
\begin{equation}\label{eq2-06-oct-2025}
\aligned
{\eta_{\Scal}}_{\nu\lambda}\del_t\HScal^{0\nu} & =- \gScal_{\nu\lambda}\delSs_a\HScal^{a\nu} 
- \HScal_{\nu\lambda}\del_t\HScal^{0\nu}
- \frac{1}{2}\gScal^{\mu\nu}\delS_{\lambda}\HScal_{\mu\nu}
- {W_{\Scal}}_{\lambda}
\\
& \quad - \PhiS_{\lambda}^{\delta}g_{\beta\delta}\del_{\alpha}\big(\PhiS_{\mu}^{\alpha}\PhiS_{\nu}^{\beta}\big)\HScal^{\mu\nu} 
- \frac{1}{2}g^{\alpha\beta}\delS_{\lambda}\big(\PsiS_{\alpha}^{\mu}\PsiS_{\beta}^{\nu}\big)\HScal_{\mu\nu}.
\endaligned
\end{equation}
We thus ready to establish the following estimates.

\begin{proposition}\label{prop1-06-oct-2026}
Assume that \eqref{eq-vect-condition} and \eqref{eq-USA-condition} hold for some sufficiently small $\eps_s$. Then one has 
\begin{equation}
\aligned
|\del_t\HScal^{00}|_{p,k} + |\del_t \HScal^{02}|_{p,k} + |\del_t \HScal^{03}|_{p,k}& \lesssim \frac{\la r-t\ra}{r}|\del H|_{p,k} + |\delsN H|_{p,k} 
\\
& \quad + \sum_{p_1+p_2\leq p}|H|_{p_1}|\del H|_{p_2}. 
\endaligned
\end{equation}
\end{proposition}

\begin{proof}
We fix $\lambda = 1,2,3$ in \eqref{eq2-06-oct-2025}. When $\lambda=1$, thanks to \eqref{eq1-06-oct-2025}, in the left-hand side we only have ${\eta_{\Scal}}_{\nu=0,\lambda=1} = -1 \neq 0$. Thus
\begin{equation}
\aligned
- \del_t\HScal^{00} & =- \gScal_{\nu1}\delSs_a\HScal^{a\nu} 
- \HScal_{\nu1}\del_t\HScal^{0\nu}
- \frac{1}{2}\gScal^{\mu\nu}\vec{V}\HScal_{\mu\nu}
- {W_{\Scal}}_{1}
\\
& \quad - \PhiS_{1}^{\delta}g_{\beta\delta}\del_{\alpha}\big(\PhiS_{\mu}^{\alpha}\PhiS_{\nu}^{\beta}\big)\HScal^{\mu\nu} 
- \frac{1}{2}g^{\alpha\beta}\delS_{1}\big(\PsiS_{\alpha}^{\mu}\PsiS_{\beta}^{\nu}\big)\HScal_{\mu\nu}.
\endaligned
\end{equation}
Then we take the $|\cdot|_{p,k}$ norm on each term in the right-hand side by applying Lemma~\ref{lem1-06-oct-2025} and the homogeneity of $\PhiS,\PsiS$ together with \eqref{eq-vect-condition} for first, third, fifth and last term in order to simplify the high-order terms. Then we apply Lemma~\ref{lem2-06-oct-2025} on the first term on the right-hand side for the derivative $\delSs$. 

The case $\lambda =2,3$ leads us to similar estimates on $\HScal^{02}, \HScal^{03}$. We omit the details.
\end{proof}


We then turn our attention to the deformation tensor $\pi[\del_t]$. 

\begin{proposition}\label{prop2-06-oct-2025}
Assume that \eqref{eq-vect-condition} and \eqref{eq-USA-condition} hold for some sufficiently small $\eps_s$. Then one has 
\begin{equation}
\aligned
\big|{\pi[\del_{\alpha}]_{\Scal}}^{00}\big|_{p,k} + \big|{\pi[\del_{\alpha}]_{\Scal}}^{02}\big|_{p,k} + \big|{\pi[\del_{\alpha}]_{\Scal}}^{03}\big|_{p,k}& \lesssim \frac{\la r-t\ra}{r}|\del H|_{p,k} + |\delsN H|_{p,k} 
\\
& \quad + \sum_{p_1+p_2\leq p}|H|_{p_1}|\del H|_{p_2}. 
\endaligned
\end{equation}
\end{proposition}

\begin{proof}
\bse
Recalling Lemma~\ref{lem3-06-oct-2025(l)} and \eqref{eq4-06-oct-2025(l)}, we obtain
\be
\pi[\del_{\alpha}]^{\alpha\beta} = g^{\alpha\mu}g^{\beta\nu}\del_{\alpha}H_{\mu\nu} = - \del_{\alpha}H^{\alpha\beta}.
\ee
We also have the identity
\begin{equation}
{\pi[\del_{\alpha}]_{\Scal}}^{\alpha\beta} = - \del_{\alpha}\HScal^{\alpha\beta} 
+ H^{\mu\nu}\del_{\alpha}\big(\PsiS_{\beta}^{\alpha}\PsiS_{\mu}^{\beta}\big).
\end{equation}
We thus obtain the desired result for $\pi[\del_t]$ by Proposition~\ref{prop1-06-oct-2026}. For the remaining translations, we note that
\be
\del_a = t^{-1}L_a - (x^a/t)\del_t
\ee
and
\be
{\pi[\del_a]_{\Scal}}^{\alpha\beta} 
= -(x^a/t)\del_t\HScal^{\alpha\beta} - t^{-1}L_a\HScal^{\alpha\beta} 
+ H^{\mu\nu}\del_a\big(\PsiS_{\beta}^{\alpha}\PsiS_{\mu}^{\beta}\big), 
\ee
which leads us to the estimates on $\pi[\del_a]$.
\ese
\end{proof}


\subsection{Estimates from spacetime integration}

We thus need to rely on $\mathrm{grad}(t-r)$. For convenience i the discussion, we recall a notation introduced in Proposition~\ref{prop1-05-oct-2025}:
\begin{equation}\label{eq13-06-oct-2025}
\vec{\gamma} := \mathrm{grad}(r-t).
\end{equation}
Under the assumptions \eqref{eq-USA-condition} and \eqref{eq2-09-oct-2025(l)}, $\vec{\gamma}$ is timelike and future oriented. Then we remark the following Cauchy-Schwartz inequality:
\begin{equation}
\la \Phi,\vec{\gamma}\cdot\Psi\ra_{\ourD}\leq |\Phi|_{\vec{\gamma}}|\Psi|_{\vec{\gamma}},
\end{equation} 
where $|\Phi|_{\vec{\gamma}}^2 = \la \Phi,\vec{\gamma}\cdot\Phi\ra_{\ourD}$. Based on the above preparations, we establish the following estimate.

\begin{lemma}\label{lem5-06-oct-2025}
Assume \eqref{eq-vect-condition} holds and \eqref{eq-USA-condition} holds for a sufficiently small $\eps_s$. Then one has 
\begin{equation}
\big|\la\Phi,\vec{V}\cdot\Psi \ra_{\ourD}\big|
\lesssim |\Phi|_{\vec{\gamma}}|\Psi|_{\vec{\gamma}} 
+ \zetab^{-1}|H||\Phi|_{\vec{n}}|\Psi|_{\vec{n}}.
\end{equation}
\end{lemma}

\begin{proof}
\bse
A direct calculation leads us to
\begin{equation}\label{eq6-06-oct-2025}
\mathrm{grad}(r-t) = \vec{V} + H^{\alpha\beta}\del_{\alpha}(r-t)\del_{\beta}.
\end{equation}
Thus we obtain 
\be
\la \Phi,\mathrm{grad}(r-t)\cdot\Psi\ra_{\ourD} 
= \la\Phi,\vec{V}\cdot\Psi\ra_{\ourD} + \la\Phi,H^{\alpha\beta}\del_{\alpha}(r-t)\del_{\beta}\cdot\Psi \ra_{\ourD}.
\ee
and we apply Lemma~\ref{lem2-13-july-2025}.
\ese
\end{proof}


We will also establish a high-order version.

\begin{proposition}\label{prop1-06-oct-2025}
Assume \eqref{eq-vect-condition} holds and \eqref{eq-USA-condition} holds for a sufficiently small $\eps_s$. Then for any admissible operator of type $(p,k)$, one has 
\begin{equation}\label{eq10-06-oct-2025}
\aligned
\big|\la \Phi,\mathscr{Z}^I(\vec{V}\cdot\Psi)\ra_{\ourD}\big|
& \lesssim_p 
|\Phi|_{\vec{\gamma}}[\Psi]_{\vec{\gamma},p,k}
+ \frac{\la r-t\ra}{r}\zetab^{-1}|\Phi|_{\vec{n}}[\Psi]_{p,k} 
\\
& \quad + \zetab^{-1}|\Phi|_{\vec{n}}\sum_{p_1+p_2\leq p}[\Psi]_{p_1}|H|_{p_2},
\endaligned
\end{equation}
with the notation
\be
[\Psi]_{\vec{\gamma},p,k}:=\max_{\ord(J)\leq p\atop \rank(J)\leq k}[\mathscr{Z}^J\Psi]_{\vec{\gamma}}.
\ee
\end{proposition}

\begin{proof}
\bse
The proof is based on the decomposition in Lemma~\ref{lem1-15-july-2025}. We remark that in the proof of Lemma~\ref{lem2-23-july-2025}, we have already given the estimate on the high-order terms in the right-hand side of \eqref{eq7-04-july-2025}, which is the estimate on \eqref{eq7-06-oct-2025(l)}. We thus only need to focus on the part of Lie derivatives.

Observe first that 
\begin{equation}\label{eq11-06-oct-2025}
\aligned
\Lcal_{L_a}\vec{V} & =- \frac{x^a}{r}\vec{V} - \frac{x^a}{r}\frac{t-r}{t}\del_r + \frac{t-r}{t}\del_a,
\\
\Lcal_{\del_{\alpha}}\vec{V} & =t^{-1}B_{\alpha}^b\del_b, 
\endaligned
\end{equation}
where $V^{\alpha}\del_{\alpha} = \vec{V}$, and $B_{\alpha}^b$ are homogeneous coefficients of degree zero in $\{3t/4\leq r\leq 4t/3\}$. Thus $V^{\alpha}$ are homogeneous functions of degree zero defined in $\{3t/4\leq r\leq 4t/3\}$. Then we claim 
\begin{equation}\label{eq8-06-oct-2025}
\Lcal_{\mathscr{Z}}^J\vec{V} = \frac{t-r}{t}A^{Ja}\del_a + t^{-1}B^{Jb}\del_b + C^J\vec{V},
\end{equation}
where $A,B$ are sum of homogeneous functions of degree non-positive. This can be checked by induction on $\ord(J)$. The key is that $L_a$ acting on a homogeneous functions of degree zero gives again the same type function, and
\begin{equation}\label{eq9-06-oct-2025}
L_a\Big(\frac{t-r}{t}\Big) = A\frac{t-r}{t}
\end{equation}
with $A$ a homogeneous function of degree zero in $\{3t/4\leq r\leq 4t/3\}$. Suppose that \eqref{eq8-06-oct-2025} holds for $\ord(I)\leq p$. We consider an admissible operator 
$\mathscr{Z}^{I'} = Z\mathscr{Z}^I$. Then when $Z = L_a$
\be
\Lcal_{L_a}\big(\Lcal_{\mathscr{Z}}^{I'}\vec{V}\big) = \Lcal_{L_a}\Big(\frac{t-r}{t}A^{Ja}\del_a + t^{-1}B^{Jb}\del_b + C^J\vec{V}\Big).
\ee
Recall the relation
\be
\Lcal_{L_a}\del_{\alpha} = [L_a,\del_{\alpha}] = \Theta_{a\alpha}^{\beta}\del_{\beta}
\ee
where $\Theta_{a\alpha}^{\beta}$ are constants. Then by \eqref{eq9-06-oct-2025} and the homogeneity, we obtain the desired estimate. When $Z= \del_{\delta}$, we only need to apply the argument on homogeneity. When $Z = \Omega_{ab}$, we apply again the relation $\Omega_{ab} = (x^a/t)L_b-(x^b/t)L_a$.
\ese
%

%
\bse
Now we return to the proof of \eqref{eq10-06-oct-2025} and note that
\be
\big|\la \Phi,\mathscr{Z}^I(\vec{V}\cdot\Psi)\ra_{\ourD}\big| 
=
\sum_{\ord(I_1)+\ord(I_2)\leq p\atop \rank(I_1)+\rank(I_2)\leq k}
\big|\la \Phi,\mathscr{Z}^{I_1}\vec{V}\cdot\mathscr{Z}^{I_2}\Psi\ra_{\ourD} \big|.
\ee
Then we apply Lemma~\ref{lem1-15-july-2025} and Lemma~\ref{lem2-23-july-2025} (for the high-order terms in \eqref{eq7-04-july-2025}),
\be
\big|\la \Phi,\mathscr{Z}^{I_1}\vec{V}\cdot\mathscr{Z}^{I_2}\Psi\ra_{\ourD} \big|
\lesssim_p \big|\la \Phi,\Lcal_{\mathscr{Z}}^{I_1}\vec{V}\cdot\mathscr{Z}^{I_2}\Psi\ra_{\ourD} \big|
+ \zetab^{-1}|\Phi|_{\vec{n}}\sum_{p_1+p_2\leq p}[\Psi]_{p_1}|H|_{p_2}.
\ee
In view of \eqref{eq8-06-oct-2025} and Lemma~\ref{lem5-06-oct-2025}, we obtain
\be
\aligned
\big|\la \Phi,\Lcal_{\mathscr{Z}}^{I_1}\vec{V}\cdot\mathscr{Z}^{I_2}\Psi\ra_{\ourD} \big|& \lesssim_p  
\zetab^{-1}\frac{\la r-t\ra}{r}|\Phi|_{\vec{n}}[\Psi]_{p,k}  
+ \big|\la \Phi,\vec{V}\cdot\mathscr{Z}^{I_2}\Psi\ra_{\ourD}\big|
\\
& \lesssim_p  |\Phi|_{\vec{\gamma}}[\Psi]_{\vec{\gamma},p,k}
+ \zetab^{-1}\la r-t\ra r^{-1}|\Phi|_{\vec{n}}[\Psi]_{p,k}
+ \zetab^{-1}|H||\Phi|_{\vec{n}}[\Psi]_{p,k}.
\endaligned
\ee
This leads us to the desired estimate.
\ese
\end{proof}


\subsection{Decomposition of the Dirac commutator}

For convenience in the notation, we define
\be
\Psihat_a:= t^{-1}\widehat{L_a}\Psi,
\quad
[\Psihat]_{p,k} := \max_{a=1,2,3}[t^{-1}\widehat{L_a}\Psi]_{p,k}.
\ee
Thanks to the homogeneity property of the admissible vector fields, we have 
\begin{equation}\label{eq2-13-oct-2025}
[\Psihat]_{p,k}\lesssim_p t^{-1}[L\Psi]_{p,k}\lesssim t^{-1}[\Psi]_{p+1,k+1},\quad [\Psihat]_{p,k}\lesssim_p [\del\Psi]_{p,k},
\quad
\text{ in }  \Mcal^{\near}_{[s_0,s_1]}.
\end{equation}

We write \eqref{eq5-21-aout-2025} in the following form:
\begin{equation}\label{eq12-06-oct-2025}
[\widehat{\del_{\delta}},\opDirac]\Psi = - \frac{1}{2}{\pi[\del_{\delta}]_{\Scal}}^{\alpha\beta}\delS_{\alpha}\cdot\nabla_{\delS_{\beta}}\Psi
- \frac{1}{4}\pi[\del_{\delta}]^{\alpha\beta}\nabla_{\alpha}\del_{\beta}\cdot\Psi
- \frac{1}{4}[\del_{\delta},W]\cdot\Psi.
\end{equation}
We focus our attention on the first term in the right-hand side, which is decomposed again as
\begin{equation}\label{eq4-06-oct-2025}
\aligned
{\pi[\del_{\delta}]_{\Scal}}^{\alpha\beta}\delS_{\alpha}\cdot\nabla_{\delS_{\beta}}\Psi
& = 
{\pi[\del_{\delta}]_{\Scal}}^{\alpha b}\delS_{\alpha}\cdot\nabla_{\delSs_b}\Psi
+ {\pi[\del_{\delta}]_{\Scal}}^{\alpha0}\delS_{\alpha}\cdot\nabla_{\del_{\delta}}\Psi
\\
& = {\pi[\del_{\delta}]_{\Scal}}^{\alpha b}\delS_{\alpha}\cdot\nabla_{\delSs_b}\Psi
+ \sum_{\alpha\neq 1}
{\pi[\del_{\delta}]_{\Scal}}^{\alpha 0}\delS_{\alpha}\cdot\nabla_{\del_{\delta}}\Psi
+{\pi[\del_{\delta}]_{\Scal}}^{1 0}\vec{V}\cdot\nabla_{\del_{\delta}}\Psi
\endaligned
\end{equation}
We then perform the following calculation for the first term in the right-hand side:
\begin{equation}
\aligned
\nabla_{\vec{V}}\Psi & =\frac{t-r}{t}\nabla_{\del_{\delta}}\Psi + t^{-1}(x^a/r)\nabla_{L_a}\Psi
\\
& = \frac{t-r}{t}\widehat{\del_{\delta}}\Psi + \frac{x^a}{rt}\widehat{L_a}\Psi
+  \frac{t-r}{4t}g^{\mu\nu}\del_{\mu}\cdot\nabla_{\nu}\del_{\delta}\cdot\Psi 
+ \frac{x^a}{4rt}g^{\mu\nu}\cdot\del_{\mu}\cdot\nabla_{\nu}L_a\cdot\Psi
\\
& = \frac{t-r}{t}\widehat{\del_{\delta}}\Psi + \frac{x^a}{r}\Psihat_a
+  \frac{t-r}{4t}g^{\mu\nu}\del_{\mu}\cdot\nabla_{\nu}\del_{\delta}\cdot\Psi 
+ \frac{x^a}{4rt}g^{\mu\nu}\cdot\del_{\mu}\cdot\nabla_{\nu}L_a\cdot\Psi,
\\
\nabla_{\vec{S}_i}\Psi & =S_i^at^{-1}\nabla_{L_a}\Psi = S_i^a\Psihat_a + \frac{1}{4t}S_i^ag^{\mu\nu}\del_{\mu}\cdot\nabla_{\nu}L_a\cdot\Psi.
\endaligned
\end{equation}
where we recall that $S_i^a$a are homogeneous of degree zero. 


Most importantly, the ``bad'' components enjoy good coefficients, while the ``bad'' coefficients meet the good components. The terms $\widehat{\del}\Psi$ have a factor $\frac{t-r}{t}$, while the terms $\Psihat_a$ (as we will see in the bootstrap argument) enjoy better estimates. We thus establish the following estimate.

\begin{lemma}\label{lem6-06-oct-2025}
Assume \eqref{eq-vect-condition} holds and \eqref{eq-USA-condition} holds for a sufficiently small $\eps_s$, the following estimate holds in $\Mcal^{\near}_{[s_0,s_1]}$:
\begin{equation}
\aligned
& \zetab\big[{\pi[\del_{\delta}]_{\Scal}}^{\alpha b}\delS_{\alpha}\cdot\nabla_{\delSs_b}\Psi\big]_{p,k}
+\zetab\sum_{\alpha\neq 1}\big[{\pi[\del_{\delta}]_{\Scal}}^{\alpha 0}\delS_{\alpha}\cdot\nabla_{\del_{\delta}}\Psi\big]_{p,k}
\\
& \lesssim_p\frac{\la r-t\ra}{r}\hspace{-0.3cm}\sum_{p_1+p_2\leq p\atop k_1+k_2\leq k}|\del H|_{p_1,k_1}[\del\Psi]_{p_2,k_2} 
+\sum_{p_1+p_2\leq p\atop k_1+k_2\leq k}|\delsN H|_{p_1,k_1}[\del \Psi]_{p_2,k_2}
\\
& \quad + \sum_{p_1+p_2\leq p\atop k_1+k_2\leq k}|\del H|_{p_1,k_1}\big([\Psihat]_{p_2,k_2} + r^{-1}[\Psi]_{p_2,k_2}\big)
+\Cubic^+_p[\Psi] + \Cubic^+_p[\del\Psi].
\endaligned
\end{equation}
\end{lemma}

\begin{proof}
\bse
We need to bound the following terms (where we omit the zero-degree homogeneous coefficients):
\be
\aligned
& \frac{t-r}{t}{\pi[\del_{\delta}]_{\Scal}}^{\alpha b}\delS_{\alpha}\cdot\widehat{\del_t}\Psi,
\quad
&&{\pi[\del_{\delta}]_{\Scal}}^{\alpha b}\delS_{\alpha}\cdot \Psihat_a,
\\
& \frac{t-r}{t}g^{\mu\nu}{\pi[\del_{\delta}]_{\Scal}}^{\alpha b}\delS_{\alpha}\cdot\del_{\mu}\cdot\nabla_{\nu}\del_{\gamma}\cdot\Psi,
\quad
&&t^{-1}g^{\mu\nu}{\pi[\del_{\delta}]_{\Scal}}^{\alpha b}\delS_{\alpha}\cdot\del_{\mu}\cdot\nabla_{\nu}L_a\cdot\Psi,
\\
&{\pi[\del_{\delta}]_{\Scal}}^{\alpha 0}\del_{\gamma}\cdot\widehat{\del_t}\Psi,
\quad
&&g^{\mu\nu}{\pi[\del_{\delta}]_{\Scal}}^{\alpha0}\delS_{\alpha}\cdot\del_{\mu}\cdot\nabla_{\nu}\del_t\cdot\Psi,
\quad
\alpha=0,2,3.
\endaligned
\ee
Recalling that the elements of transition matrices $\PhiS$ and $\PsiS$ are also homogeneous of zero order, we omit these coefficients and we only need to estimate
\begin{equation}\label{eq1-07-oct-2025}
\aligned
& \frac{t-r}{t}
\pi[\del_{\delta}]^{\alpha b}\del_{\gamma}\cdot\widehat{\del_{\mu}}\Psi,
\quad
&& \pi[\del_{\delta}]^{\alpha b}\del_{\gamma}\cdot t^{-1}\widehat{L_a}\Psi,
\\
&{\pi[\del_{\delta}]_{\Scal}}^{\alpha 0}\del_{\gamma}\cdot\widehat{\del_t}\Psi,\quad \alpha=0,2,3,
\\
&g^{\mu\nu}\pi[\del_{\delta}]^{\alpha b}\del_{\gamma}\cdot\del_{\mu}\cdot\nabla_{\nu}\del_{\gamma}\cdot\Psi,
\quad
&&t^{-1}g^{\mu\nu}\pi[\del_{\delta}]^{\alpha b}\del_{\gamma}\cdot\del_{\mu}\cdot\nabla_{\nu}L_a\cdot\Psi.
\endaligned
\end{equation}
The first term is bounded by \eqref{eq6-16-july-2025}, except that in the present case we have our standard coefficient $\frac{t-r}{t}$ for $\widehat{\del}\Psi$. 
Then we apply \eqref{eq6-16-july-2025}. The last two terms we remark that
\be
g^{\mu\nu}\nabla_{\nu}\del_\gamma = g^{\mu\nu}\Gamma_{\nu\gamma}^{\lambda}\del_{\lambda},
\quad 
g^{\mu\nu}\nabla_{\nu}L_a = g^{\mu\nu}\del_{\nu}L_a^{\gamma}\del_{\gamma} 
+ g^{\mu\nu}L_a^{\gamma}\Gamma_{\mu\gamma}^{\lambda}\del_{\lambda}. 
\ee
Thanks to Lemma~\ref{lem1-11-july-2025} and \eqref{eq-vect-condition}, we also have 
\begin{equation}
|\Gamma_{\nu\gamma}^{\lambda}|_{p,k}\lesssim_p |\del H|_{p,k} + \sum_{p_1+p_2\leq p\atop k_1+k_2\leq k}|H|_{p_1,k_1}|\del H|_{p_2,k_2}.
\end{equation}
Then thanks to \eqref{eq3-23-july-2025},
\begin{equation}
\aligned
\big[g^{\mu\nu}\pi[\del_{\delta}]^{\alpha b}\del_{\gamma}\cdot\del_{\mu}\cdot\nabla_{\nu}\del_{\gamma}\cdot\Psi\big]_{p,k}
& \lesssim_p  \Cubic_{p,k}(H,\del H,\Psi) + \Cubic(\del H,\del H,\Psi)\lesssim \Cubic^+_p[\Psi],
\\
\big[g^{\mu\nu}\pi[\del_{\delta}]^{\alpha b}\del_{\gamma}\cdot\del_{\mu}\cdot\nabla_{\nu}L_a\cdot\Psi\big]_{p,k}
& \lesssim_p \sum_{p_1+p_2\leq p\atop k_1+k_2\leq k}|\del H|_{p_1,k_1}[\Psi]_{p_2,k_2} 
\\
& \quad + t\Cubic_{p,k}(H,\del H,\Psi) + t\Cubic_{p,k}(\del H,\del H,\Psi)
\\
& \lesssim_p \sum_{p_1+p_2\leq p\atop k_1+k_2\leq k}|\del H|_{p_1,k_1}[\Psi]_{p_2,k_2} + t\Cubic^+_p[\Psi].
\endaligned
\end{equation}
For the third term in \eqref{eq1-07-oct-2025}, we rely on Proposition~\ref{prop2-06-oct-2025}:
\be
\aligned
& \zetab\big[{\pi[\del_{\delta}]_{\Scal}}^{\alpha0}\del_{\gamma}\cdot\widehat{\del_t}\Psi\big]_{p,k}
\\
& \lesssim_p\sum_{p_1+p_2\leq p\atop k_1+k_2\leq k}
\big|{\pi[\del_{\delta}]_{\Scal}}^{\alpha0}\big|_{p_1,k_1}[\widehat{\del_t}\Psi]_{p_2,k_2}
\\
& \lesssim_p \frac{\la r-t\ra}{r}\sum_{p_1+p_2\leq p\atop k_1+k_2\leq k}
|\del H|_{p_1,k_1}[\del\Psi]_{p_2,k_2} 
+ r^{-1}\sum_{p_1+p_2\leq p\atop k_1+k_2\leq k}|\delsN H|_{p_1,k_1}[\del\Psi]_{p_2,k_2}
\\
& \quad+\Cubic_p(H,\del H,\del\Psi). \qedhere 
\endaligned
\ee
\ese
\end{proof}


We then establish the high-order estimate on the commutator $[\del_{\delta},\opDirac]$.

\begin{lemma}\label{lem1-13-oct-2025}
Assume \eqref{eq-vect-condition} holds and \eqref{eq-USA-condition} holds for a sufficiently small $\eps_s$. Let $\Phi,\Psi$ be sufficiently regular spinor fields and $\mathscr{Z}^I$ be an admissible operator of type $(p,k)$. Then in $\Mcal^{\near}_{[s_0,s_1]}$, one has
\begin{equation}
\aligned
& \zetab\big|\la\Phi,\mathscr{Z}^I([\del_{\delta},\opDirac]\Psi)\ra_{\ourD}\big|
\\
& \lesssim_p |\Phi|_{\vec{n}}\sum_{p_1+p_2\leq p\atop k_1+k_2\leq k}\big([\del\Psi]_{p_1,k_1} + [\Psi]_{p_1,k_1}\big)\Big(\frac{\la r-t\ra}{r}|\del H|_{p_2,k_2} + |\delsN H|_{p_2,k_2}\Big)
\\
& \quad+\zetab|\Phi|_{\vec{\gamma}}\sum_{p_1+p_2\leq p\atop k_2+k_2\leq k}
[\Psi]_{\vec{\gamma},p_1,k_1}|\del H|_{p_2,k_2}
+|\Phi|_{\vec{n}}\sum_{p_1+p_2\leq p\atop k_1+k_2\leq k}[\Psihat]_{p_1,k_1}|\del H|_{p_2,k_2}
\\
& \quad
+ |\Phi|_{\vec{n}}\big(\Hcom_p[\Psi] + \Hcom_p[\del\Psi]\big).
\endaligned
\end{equation}
\end{lemma}

\begin{proof}
\bse
We write the decomposition~\eqref{eq12-06-oct-2025} and \eqref{eq4-06-oct-2025} into the following form:
\begin{equation}
\aligned
\, [\widehat{\del_{\delta}},\opDirac]\Psi & =- \frac{1}{2}{\pi[\del_{\delta}]_{\Scal}}^{1 0}\vec{V}\cdot\nabla_{\del_{\delta}}\Psi
- \frac{1}{2}\Big({\pi[\del_{\delta}]_{\Scal}}^{\alpha b}\delS_{\alpha}\cdot\nabla_{\delSs_b}\Psi
+ \sum_{\alpha\neq 1}
{\pi[\del_{\delta}]_{\Scal}}^{\alpha 0}\delS_{\alpha}\cdot\nabla_{\del_{\delta}}\Psi\Big)
\\
& \quad - \frac{1}{4}\big(\pi[\del_{\delta}]^{\alpha\beta}\nabla_{\alpha}\del_{\beta}\cdot\Psi
+ [\del_{\delta},W]\cdot\Psi\big)
\\
=:& T_1 + T_2 + T_3.
\endaligned
\end{equation}
The terms $T_3$ are bounded by \eqref{eq7-18-july-2025} and \eqref{eqs1-07-oct-2025}:
\begin{equation}
[T_3]_{p,k}\lesssim_p \Cubic^+_p[\Psi] + \Hwave_p[\Psi].
\end{equation}
For $T_2$, we apply Lemma~\ref{lem6-06-oct-2025}. For the term $T_1$, we establish the following estimate:
\begin{equation}
\aligned
& \big|\big\la\Phi,\mathscr{Z}^I\big({\pi[\del_{\delta}]_{\Scal}}^{10}\vec{V}\cdot\nabla_{\del_{\delta}}\Psi\big)\big\ra_{\ourD}\big|
\\
& \lesssim_p|\Phi|_{\vec{\gamma}}\sum_{p_1+p_2\leq p\atop k_2+k_2\leq k}
|\del H|_{p_1,k_1}[\Psi]_{\vec{\gamma},p_2,k_2}
+\frac{\la r-t\ra}{r}\zetab^{-1}|\Phi|_{\vec{n}}
\sum_{p_1+p_2\leq p\atop k_1+k_2\leq k}[\del H]_{p_1,k_1}[\Psi]_{p_2,k_2}
\\
& \quad+\zetab^{-1}|\Phi|_{\vec{n}}\Cubic^+_p[\Psi].
\endaligned
\end{equation}
To see this, we note that the elements of $\PhiS,\PsiS$ are homogeneous of degree zero. We thus only need to treat
\be
\big\la\Phi,\mathscr{Z}^I\big(\pi[\del_{\delta}]^{\alpha\beta}\vec{V}\cdot\widehat{\del_t}\Psi\big)\big\ra_{\ourD}.
\ee
We apply Proposition~\ref{prop1-06-oct-2025} and obtain 
\be
\big\la\Phi,\mathscr{Z}^J\big(\pi[\del_{\delta}]^{\alpha\beta}\vec{V}\cdot\widehat{\del_t}\Psi\big)\big\ra_{\ourD}
= \sum_{J_1\odot J_2=J}\mathscr{Z}^{J_1}(\pi[\del_{\delta}]^{\alpha\beta})\big\la\Phi,\mathscr{Z}^{J_2}\big(\vec{V}\cdot\Psi\big) \big\ra_{\ourD}.
\ee
We estimate each term in the right-hand side. Recalling \eqref{eq1-17-july-2025} and Proposition~\ref{prop1-06-oct-2025}
\begin{equation}
\aligned
& \big|\big\la\Phi,\mathscr{Z}^J\big(\pi[\del_{\delta}]^{\alpha\beta}\vec{V}\cdot\widehat{\del_t}\Psi\big)\big\ra_{\ourD}\big|
\\
& \lesssim_p  \sum_{J_1\odot J_2=J}
\big|\mathscr{Z}^{J_1}(\pi[\del_{\delta}]^{\alpha\beta})\big|
\big|\big\la\Phi,\mathscr{Z}^{J_2}\big(\vec{V}\cdot\Psi\big) \big\ra_{\ourD}\big|
\\
& \lesssim_p  |\Phi|_{\vec{\gamma}}\sum_{p_1+p_2\leq p\atop k_2+k_2\leq k}
|\del H|_{p_1,k_1}[\Psi]_{\vec{\gamma},p_2,k_2}
+\frac{\la r-t\ra}{r}\zetab^{-1}|\Phi|_{\vec{n}}
\sum_{p_1+p_2\leq p\atop k_1+k_2\leq k}[\del H]_{p_1,k_1}[\Psi]_{p_2,k_2}
\\
& \quad +\zetab^{-1}|\Phi|_{\vec{n}}\Hcom_p[\Psi].
\endaligned
\end{equation}
\ese
\end{proof}


\subsection{Conclusion}

Now we are ready to give the estimate on Dirac commutators in $\Mcal^{\near}_{[s_0,s_1]}$.

\begin{proposition}\label{prop1-14-oct-2025}
Assume that the conditions \eqref{eq-vect-condition} and \eqref{eq-USA-condition} hold for a sufficiently small $\eps_s$. Consider any sufficiently regular spinor field $\Psi$ and any admissible vector field $Z$. Then one has the following estimate for any admissible operator $\mathscr{Z}^I$ of type $(p,k)$ in $\Mcal^{\near}_{[s_0,s_1]}$:
\begin{equation}
\aligned
& \zetab\big|\big\la\Phi,[\mathscr{Z}^I,\opDirac]\Psi\big\ra_{\ourD}\big|
\\
& \lesssim_p   
|\Phi|_{\vec{n}}\hspace{-0.3cm}\sum_{p_1+p_2\leq p\atop k_1+k_2\leq k}
[\del\Psi]_{p_1-1,k_1-1}|H|_{p_2+1,k_2+1}
\\
& \quad+|\Phi|_{\vec{n}}\hspace{-0.3cm}\sum_{p_1+p_2\leq p\atop k_1+k_2\leq k}
\big([\del\Psi]_{p_1-1,k_1} + [\Psi]_{p_1-1,k_1}\big)
\Big(\frac{\la r-t\ra}{r}|\del H|_{p_2,k_2} + |\delsN H|_{p_2,k_2}\Big)
\\
& \quad+\zetab|\Phi|_{\vec{\gamma}}\sum_{p_1+p_2\leq p\atop k_1+k_2\leq k}
[\Psi]_{\vec{\gamma},p_1-1,k_1}|\del H|_{p_2,k_2}
+t^{-1}|\Phi|_{\vec{n}}\hspace{-0.3cm}\sum_{p_1+p_2\leq p\atop k_1+k_2\leq k}
[\Psi]_{p_1,k_1+1}|\del H|_{p_2,k_2}
\\
& \quad+|\Phi|_{\vec{n}}\big(\Cubic^+_{p-1}[\Psi] + \Cubic^+_{p-1}[\del\Psi] + \Hwave_{p-1}[\Psi]\big).
\endaligned
\end{equation}
\end{proposition}
\begin{proof}
This is based on an induction on $\ord(I)$. Assume that for $\ord(I)\leq p$ the desired estimate holds. We consider $\mathscr{Z}^{I'} = \mathscr{Z}^IZ$. 

\noindent{$\bullet$} When $Z = \del_{\delta}$, $p':=\ord(I') = p+1$, $k':=\rank(I') = k$.
\begin{equation}\label{eq1-09-oct-2025}
[\mathscr{Z}^I\widehat{\del_{\delta}},\opDirac]\Psi 
= \mathscr{Z}^I([\widehat{\del_{\delta}},\opDirac]\Psi) 
+ [\mathscr{Z}^I,\opDirac](\widehat{\del_{\delta}}\Psi)
=: T_1+T_2.
\end{equation}

For the term $T_1$, we apply Lemma~\ref{lem1-13-oct-2025} together with \eqref{eq2-13-oct-2025}:
$$
\aligned
& \zetab\big|\la\Phi,\mathscr{Z}^I([\del_{\delta},\opDirac]\Psi)\ra_{\ourD}\big|
\\
& \lesssim_p |\Phi|_{\vec{n}}\sum_{p_1+p_2\leq p\atop k_1+k_2\leq k}\big([\del\Psi]_{p_1,k_1} + [\Psi]_{p_1,k_1}\big)\Big(\frac{\la r-t\ra}{r}|\del H|_{p_2,k_2} + |\delsN H|_{p_2,k_2}\Big)
\\
& \quad+ \zetab|\Phi|_{\vec{\gamma}}\sum_{p_1+p_2\leq p\atop k_2+k_2\leq k}
|\del H|_{p_1,k_1}[\Psi]_{\vec{\gamma},p_2,k_2}
+|\Phi|_{\vec{n}}\sum_{p_1+p_2\leq p\atop k_1+k_2\leq k}[\Psihat]_{p_1,k_1}|\del H|_{p_2,k_2}
\\
&
\quad+ |\Phi|_{\vec{n}}\big(\Hcom_p[\Psi] + \Hcom_p[\del\Psi]\big)
\\
& = |\Phi|_{\vec{n}}\sum_{p_1'+p_2\leq p'\atop k_1+k_2\leq k}\big([\del\Psi]_{p_1'-1,k_1} + [\Psi]_{p_1'-1,k_1}\big)\Big(\frac{\la r-t\ra}{r}|\del H|_{p_2,k_2} + |\delsN H|_{p_2,k_2}\Big)
\\
& \quad+ \zetab|\Phi|_{\vec{\gamma}}\sum_{p_1'+p_2\leq p'\atop k_2+k_2\leq k}
[\Psi]_{\vec{\gamma},p_1'-1,k_1}|\del H|_{p_2,k_2}
+t^{-1}|\Phi|_{\vec{n}}\sum_{p_1'+p_2\leq p'\atop k_1+k_2\leq k}[\Psi]_{p_1',k_1+1}|\del H|_{p_2,k_2}
\\
& \quad
+ |\Phi|_{\vec{n}}\big(\Hcom_{p'-1}[\Psi] + \Hcom_{p'-1}[\del\Psi]\big).
\endaligned
$$
This fits the induction. For the term $T_2$ introduced in \eqref{eq1-09-oct-2025}, we only need to apply the assumption of induction.

\noindent{$\bullet$} When $Z = L_a,\Omega_{ab}$, $p'=p+1,k'=k+1$. We apply the argument of \eqref{eq13-18-july-2025}. 
$$
[\mathscr{Z}^I\widehat{Z},\opDirac]\Psi = \mathscr{Z}^I([\widehat{Z},\opDirac]\Psi) + [\mathscr{Z}^I,\opDirac](\widehat{Z}\Psi) = :T_3+T_4.
$$
For the term $T_3$, we remark that
$$
\aligned
\zetab\big|\la\Phi,\mathscr{Z}^IT_3\ra_{\ourD}\big|
& = \zetab\big|\la \Phi,\mathscr{Z}^I([\widehat{Z},\opDirac]\Psi)\ra_{\ourD}\big|
\lesssim\zetab\big|[\widehat{Z},\opDirac]\Psi\big|_{p,k} 
\\
& \lesssim_p   
\sum_{p_1+p_2\leq p\atop k_1+k_2\leq k}[\del\Psi]_{p_1,k_1}|H|_{p_2+1,k_2+1}
+ \Hcom_{p}[\Psi] + \Hwave_{p}[\Psi]
\\
& = \sum_{p_1'+p_2\leq p'\atop k_1'+k_2\leq k'}[\del\Psi]_{p_1'-1,k_1'-1}|H|_{p_2+1,k_2+1}
+ \Hcom_{p-1}[\Psi] + \Hwave_{p-1}[\Psi],
\endaligned
$$
where we have applied \eqref{eq11-18-july-2025}. This fixes the induction. For the term $T_4$ we apply directly the assumption of induction and the relation \eqref{eq3-09-oct-2025(l)}, \eqref{eq5-19-july-2025} and \eqref{eq3-19-july-2025}. Here for the term $\Phihat$, we need to treat 
$$
t^{-1}\widehat{L_a}\widehat{Z}\Psi.
$$ 
We only need to remark that it is a linear combination of $\widehat{\del_{\alpha}}\Phi$ with homogeneous coefficients of degree zero (in $\{3t/4\leq r\leq 4t/3\}$). Thus $t^{-1}\widehat{L_a}\widehat{Z}\Psi$ can still estimated by \eqref{eq3-09-oct-2025(l)}.
\end{proof}
\begin{remark}
In the above estimate, the term
$
[\Psihat]_{p_1-1,k_1}|\del H|_{p_2,k_2}
$ 
was bounded by
\begin{equation}\label{eq1-13-oct-2025}
[\Psihat]_{p_1-1,k_1}|\del H|_{p_2,k_2}\lesssim t^{-1}[L\Psi]_{p_1,k_1+1}|\del H|_{p_2,k_2}.
\end{equation} 
This leads us to a satisfactory decay. One may worry about the $(k_1+1)$ index in the right-hand side of \eqref{eq1-13-oct-2025}, however, due to the restriction $k_1+1\leq p_1\leq N$ in the bootstrap argument, the term $[\Psi]_{p_1,k_1+1}$ can always be bounded by $[\Psi]_N$. However, this term causes some difficulty when we consider the commutator associated with $\widehat{\del_{\delta}}\Psi$. That is precise reason why we need particular estimates on the spinor fields $\Psihat_a$.
\end{remark}


\section{Pointwise estimate on massive spinor fields}
\label{section=N10}

\subsection{ Analysis based on the massive Dirac operator}

{ 

\paragraph{First result of this section.}

We establish pointwise estimates for the massive Dirac equation on a curved spacetime, exploiting that the mass parameter $M>0$ is strictly positive. The estimates below imply the schematic decay
\be
|\Psi|_{\vec{n}}\;\lesssim_N\; (s/t)^N \qquad \text{for all integers } N\ge1,
\ee
which captures the physical fact that a massive fermion cannot reach the light cone: it cannot be accelerated to the speed of light, and its propagation remains strictly timelike. As in preceding sections, we assume natural bounds on the metric perturbation $H$ and on the wave--coordinate vector field $W$. Recall that $\zeta^2 := 1-|\delb_r T|^2$ is the (hyperboloidal) energy coefficient and that $\zetab := \sqrt{\zeta^2+|H^{\Ncal 00}|}$ (see \eqref{eq9-08-aout-2025}). 

\begin{proposition}[Pointwise estimates for massive Dirac fields. I]
\label{prop1-24-july-2025}
Assume that \eqref{eq-USA-condition} holds together with
\begin{equation}\label{eq8-22-july-2025}
|H|_{p+1}\lesssim \eps_s\,\zeta,
\qquad
|W|_{p+1}\lesssim \eps_s,
\end{equation}
for a sufficiently small $\eps_s>0$. Then every sufficiently regular solution $\Psi$ to the massive Dirac equation
\begin{equation}\label{eq2-14-aout-2025}
\opDirac\Psi + \mathrm{i}M\Psi = 0
\end{equation}
satisfies the pointwise bound
\begin{equation}
[\Psi]_{p,k}\;\lesssim_{M,p}\; \zetab \,[\Psi]_{p+1,k+1}
\qquad \text{in } \Mcal_{[s_0,s_1]}.
\end{equation}
\end{proposition}

The proof relies on the structure of the massive Dirac operator and on the commutator bounds developed earlier. 


\paragraph{Technical estimates.}

We begin with two preliminary lemmas.

\begin{lemma}\label{lem2-01-aout-2025}
When \eqref{eq-USA-condition} holds for a sufficiently small $\eps_s$, one has 
\begin{equation}\label{eq7-01-aout-2025}
\big|\gu^{00}\del_t + \gu^{a0}\delu_a\big|_{\vec{n}}\lesssim \zetab \qquad \text{ in } \Mcal_{[s_0,s_1]}\cap\{r\leq 3t\}.
\end{equation}
\end{lemma}

\begin{proof} 
\bse
We begin with the following relations valid in $\Mcal_{[s_0,s_1]}\cap\{r\leq 3t\}$:
\begin{equation}
\aligned
\gu^{00}  & =  J^2\gb^{00} + \frac{2x^a}{r}g^{a0}\Big(\del_rT- \frac{r}{t}\Big) 
+ \frac{x^ax^b}{r^2}g^{ab}\Big(\frac{r^2}{t^2}-|\del_rT|^2\Big) := J^2\gb^{00} + R^0,
\\
\gu^{a0}  & =  J\gb^{a0} + \frac{x^b}{r}g^{ab}\Big(\del_rT - \frac{r}{t}\Big) =: J\gb^{a0} + R^a.
\endaligned
\end{equation}
and we recall that $\delu_a = \delb_a + \frac{x^a}{r}\Big(\frac{r}{t}- \del_rT\Big)$. 
By linear combination, we find 
\be
\gu^{00}\del_t + \gu^{a0}\delu_a = J\gb^{00}\delb_s  + J\gb^{0a}\delb_a
+ \gu^{0a}\frac{x^a}{r}\Big(\frac{r}{t} - \del_rT\Big)\del_t + R^0\del_t + R^a\delb_a.
\ee
By recalling Lemma~\ref{lem1-12-june-2025} and $\betab^a = \lapsb^2\gb^{0a} = \gb_{0a}\sigmab^{ab}$.
it follows that 
\be
\gb^{00}\delb_s + \gb^{0a}\delb_a = - \lapsb^{-2}\delb_0 + \gb^{0a}\delb_a
=- \lapsb^{-2}\big(\del_0 - \betab^a\delb_a\big) = \lapsb^{-1}\vec{n}, 
\ee
therefore
\begin{equation}\label{eq6-01-aout-2025}
\gu^{00}\del_t + \gu^{a0}\delu_a = -J\lapsb^{-1}\vec{n} + \gu^{0a}\frac{x^a}{r}\Big(\frac{r}{t} - \del_rT\Big)\del_t + R^0\del_t + R^a\delb_a.
\end{equation}
\ese

\bse
It remains to bound the last three terms. We observe first that 
\be
\Big|\frac{r}{t}- \del_rT\Big| \leq \Big|\frac{r-t}{t}\big| + \big|\frac{\zeta^2}{1+\del_rT}\big|
\lesssim \zeta^2 + \frac{|r-t|}{t}\lesssim \zeta^2, 
\ee
where we have used 
\begin{equation}
\frac{|r-t|}{t}\lesssim \zeta^2 \quad\text{ in }  \Mcal^{\EM}_{[s_0,s_1]}\cap\{r\leq 2t\}.
\end{equation} 
This inequality holds, since 
\be
\frac{|r-t|}{t}\lesssim 1 = \zeta^2 \qquad \text{ in } \Mcal^{\Ecal}_{[s_0,s_1]}\cap\{r\leq 3t\}
\ee
and 
\be
\frac{|r-t|}{t}\lesssim t^{-1}\lesssim r^{-1} \lesssim \frac{s^2}{s^2+r^2}\leq \zeta^2
\qquad \text{ in } \Mcal^{\M}_{[s_0,s_1]}. 
\ee
Each of the last three terms in the right-hand side of \eqref{eq6-01-aout-2025} contains a factor $\big((r/t)- \del_rT\big)$. This completes the proof. 
\ese
\end{proof}

We establish first the estimate in Proposition~\ref{prop1-24-july-2025} at the zero-order, as follows. 

\begin{lemma}
\label{lem2-01-auot-2025}
Assume \eqref{eq-USA-condition} holds together with 
\begin{equation}\label{eq6-30-july-2025}
|\del H|\lesssim \zeta\eps_s
\end{equation}
for a sufficiently small $\eps_s$. Then for any sufficiently regular solution $\Psi$  to the massive Dirac equation 
\begin{equation}\label{eq3-29-july-2025}
\opDirac \Psi + \mathrm{i}M\Psi = \Phi
\end{equation}
in $\Mcal_{[s_0,s_1]}\cap\{r\leq 3t\}$, one has 
\begin{equation}
|\Psi|_{\vec{n}}\lesssim_{M}\zetab|\widehat{\del_t}\Psi|_{\vec{n}} + t^{-1}\zetab^{-1}[\Psi]_{1,1} + |\Phi|_{\vec{n}}\lesssim \zetab^{-1}[\Psi]_{1,1}.
\end{equation}
\end{lemma}

\begin{proof}
\bse
We write the Dirac operator in the semi-hyperboloidal frame, that is, 
\begin{equation}\label{eq2-22-july-2025}
\opDirac\Psi = \big(\gu^{00}\del_t + \gu^{a0}\delu_a\big)\cdot\nabla_t\Psi 
+ t^{-1}\big(\gu^{0a}\del_t + \gu^{ba}\delu_b\big)\cdot\nabla_{L_a}\Psi, 
\end{equation}
which leads us to
\begin{equation}\label{eq5-22-july-2025}
\aligned
\mathrm{i}M\Psi  & =  -(\gu^{00}\del_t + \gu^{a0}\delu_a)\cdot\widehat{\del_t}\Psi 
- t^{-1}(\gu^{0a}\del_t + \gu^{ba}\delu_b)\cdot\widehat{L_a}\Psi + \Phi
\\
& \quad - \frac{1}{4}(\gu^{00}\del_t + \gu^{a0}\delu_a)\cdot
g^{\mu\nu}\del_{\mu}\cdot\nabla_{\nu}\del_t\cdot\Psi 
- \frac{1}{4t}(\gu^{0a}\del_t + \gu^{ba}\delu_b)\cdot 
g^{\mu\nu}\del_{\mu}\cdot\nabla_{\nu}L_a\cdot\Psi.
\endaligned
\end{equation}
Thanks to \eqref{eq7-01-aout-2025}, we have 
\be
\big|(\gu^{00}\del_t + \gu^{a0}\delu_a)\cdot\widehat{\del_t}\Psi \big|_{\vec{n}}
\lesssim \zetab \, [\Psi]_{1,0}.
\ee
We also remark that
\be
\big|t^{-1}(\gu^{0a}\del_t + \gu^{ba}\delu_b)\cdot\widehat{L_a}\Psi\big|_{\vec{n}}
\lesssim\zetab^{-1}t^{-1}[\Psi]_{1,1}
\lesssim\zetab \, [\Psi]_{1,1}, 
\ee
due to the fact that $t^{-1}\lesssim \zetab^2$.  

For the fourth term of \eqref{eq5-22-july-2025}, we obtain 
\be
\big|(\gu^{00}\del_t + \gu^{a0}\delu_a)\cdot
g^{\mu\nu}\del_{\mu}\cdot\nabla_{\nu}\del_t\cdot\Psi \big|_{\vec{n}}
\lesssim \big|(\gu^{00}\del_t + \gu^{a0}\delu_a)\big|_{\vec{n}}|g^{\mu\nu}\del_{\mu}\cdot\nabla_{\nu}\del_t\cdot\Psi|_{\vec{n}}.
\ee
Then by \eqref{eq7-01-aout-2025} and \eqref{eq2-31-july-2025}, and \eqref{eq6-30-july-2025},
\begin{equation}\label{eq4-30-july-2025}
\aligned
\big|(\gu^{00}\del_t + \gu^{a0}\delu_a)\cdot
g^{\mu\nu}\del_{\mu}\cdot\nabla_{\nu}\del_t\cdot\Psi \big|_{\vec{n}}
& \lesssim |\del H||\Psi|_{\vec{n}}\lesssim \eps_s\zeta|\Psi|_{\vec{n}}. 
\endaligned
\end{equation}
The estimate on the last term of \eqref{eq5-22-july-2025} is more involved. We write
\be
\gu^{0a}\del_t + \gu^{ab}\delu_b = \del_a + \Hu^{a0}\del_t + \Hu^{ab}\delu_b 
= \big(\Hu^{a0} + (x^b/t)\Hu^{ab}\big)\del_t + \big(\delta_a^b+\Hu^{ab}\big)\del_b
\ee
and then find  
\be
\aligned
&t^{-1}(\gu^{0a}\del_t + \gu^{ba}\delu_b)\cdot 
g^{\mu\nu}\del_{\mu}\cdot\nabla_{\nu}L_a\cdot\Psi
\\
& = t^{-1}\big(\gu^{0a}\del_t+\gu^{ba}\delb_b\big)
\cdot(g^{\mu0}\del_{\mu}\cdot\del_a + g^{\mu a}\del_{\mu}\cdot\del_t)\cdot\Psi
\\
& \quad +\big(\gu^{0a}\del_t+\gu^{ba}\delb_b\big)
\cdot(g^{\mu\nu}\del_{\mu}\cdot\nabla_{\nu}\del_a + (x^a/t)g^{\mu\nu}\del_{\mu}\cdot\nabla_{\nu}\del_t)\cdot\Psi
\\
& = t^{-1}\Big(\big(\Hu^{a0} + (x^b/t)\Hu^{ab}\big)\del_t + \big(\delta_a^b+\Hu^{ab}\big)\del_b\Big)
\cdot(g^{\mu0}\del_{\mu}\cdot\del_a + g^{\mu a}\del_{\mu}\cdot\del_t)\cdot\Psi
\\
& \quad +\Big(\big(\Hu^{a0} + (x^b/t)\Hu^{ab}\big)\del_t + \big(\delta_a^b+\Hu^{ab}\big)\del_b\Big)
\cdot(g^{\mu\nu}\del_{\mu}\cdot\nabla_{\nu}\del_a 
+ (x^a/t)g^{\mu\nu}\del_{\mu}\cdot\nabla_{\nu}\del_t)\cdot\Psi. 
\endaligned
\ee
Observe that by \eqref{eq2-14-june-2025}
\begin{equation}\label{eq2-30-july-2025}
|\Hu^{a0}| + |\Hu^{ab}|\lesssim 1 \qquad \text{ in } \Mcal^{\ME}_{[s_0,s_1]}\cap\{r\leq 2t\}
\end{equation}
and also recall \eqref{eq3-30-july-2025}. Then by Lemma~\ref{lem1'-18-july-2025},
\begin{equation}\label{eq5-30-july-2025}
\aligned
\big|t^{-1}(\gu^{0a}\del_t + \gu^{ba}\delu_b)\cdot 
g^{\mu\nu}\del_{\mu}\cdot\nabla_{\nu}L_a\cdot\Psi\big|_{\vec{n}}
& \lesssim  t^{-1}\zeta^{-1}(1+t|\del H|)|\Psi|_{\vec{n}}
\\
& \lesssim  \zeta|\Psi|_{\vec{n}} + \zetab^{-1}|\del H||\Psi|_{\vec{n}}.
\endaligned
\end{equation}
Then we apply \eqref{eq6-30-july-2025} and, provided that $\eps_s$ is sufficiently small, we obtain the desired estimate.
\ese
\end{proof}

\begin{proof}[Proof of Proposition~\ref{prop1-24-july-2025}]
\bse
We differentiate the equation \eqref{eq2-14-aout-2025} with an admissible operator $\mathscr{Z}^I$ of type $(p,k)$:
\be
\opDirac\mathscr{Z}^I\Psi + \mathrm{i}M\mathscr{Z}^I\Psi = -[\mathscr{Z}^I,\opDirac]\Psi
\ee
and apply Lemma~\ref{lem2-01-auot-2025} together with Proposition~\ref{prop1-19-july-2025}. At this stage, we observe that \eqref{eq8-22-july-2025} covers \eqref{eq-vect-condition}. Then we have 
\be
\aligned
|\mathscr{Z}^I\Psi|_{\vec{n}} & \lesssim  
\zetab \, [\Psi]_{p+1,k+1} + \big|[\mathscr{Z}^I,\opDirac]\Psi\big|_{\vec{n}}
\\
& \lesssim  \zetab \, [\Psi]_{p+1,k+1} 
\sum_{p_1+p_2\leq p\atop k_1+k_2\leq k}[\del\Psi]_{p_1-1,k_1-1}|H|_{p_2+1,k_2+1} 
+\sum_{p_1+p_2\leq p\atop k_1+k_2\leq k}[\del\Psi]_{p_1-1,k_1}|\del H|_{p_2,k_2}
\\
& \quad + \Hcom_{p-1}[\Psi] 
+ \Hwave_{p-1}[\Psi]
\\
& \lesssim   \zetab \, [\Psi]_{p+1,k+1} + \eps_s[\Psi]_{p,k}, 
\endaligned
\ee
where for the last inequality we have applied the condition \eqref{eq8-22-july-2025}. Provided that $\eps_s$ is sufficiently small, we obtain the desired estimate.
\ese
\end{proof}

}


\subsection{ Analysis based on massive Klein-Gordon operator}
\label{section===12-2}

{ 

We now establish a second result, whose proof relies on the structure of the massive Klein-Gordon operator. We recall \eqref{eq6-29-july-2025} and write the wave operator in the semi-hyperboloidal frame \eqref{eq7-29-july-2025}: 
\begin{equation}
\Box_g\Psi 
= \gu^{\alpha\beta}\nabla_{\delu_{\alpha}}\big(\nabla_{\delu_{\beta}}\Psi\big) - \nabla_W\Psi
+ g^{\alpha\beta}\nabla_{\alpha}\big(\Psiu_{\beta}^{\nu}\big)\nabla_{\delu_{\nu}}\Psi.
\end{equation}
Based on this identity, we establish the following decomposition in the extended light cone domain. 

\begin{lemma}
\label{lem1-30-july-2025} 
Consider the domain $\Mcal_{[s_0,s_1]}\cap\{r\leq 3t\}$. Assume that \eqref{eq-USA-condition} holds and, in addition,
\begin{equation}
|H|_1\lesssim 1,\quad |\del H| + |\del\del H|\lesssim \zeta\eps_s.
\end{equation}
for a sufficiently small $\eps_s$.  Then, in $\Mcal_{[s_0,s_1]}\cap\{r\leq 3t\}$ one has 
\begin{equation}
\Box_g\Psi = \gu^{00}\widehat{\del_t}\big(\widehat{\del_t}\Psi\big) + R_2[H,\Psi] + R_1[H,\Psi] + R_0[H,\Psi] - \nabla_W\Psi, 
\end{equation}
in which the remainders satisfy   
\begin{equation}
\aligned
\big|R_2[H,\Psi]\big|_{\vec{n}} & \lesssim  t^{-1}[\del\Psi]_{1,1}
\\
\big|R_1[H,\Psi]\big|_{\vec{n}} & \lesssim 
t^{-1}\zetab^{-1}(1+t|\del H|)|\del \Psi|_{\vec{n}} + t^{-2}\zetab^{-1}(1+t|\del H|)[\Psi]_{1,1},
\\
\big|R_0[H,\Psi]\big|_{\vec{n}} & \lesssim  \eps_s|\Psi|_{\vec{n}}.
\endaligned
\end{equation}
\end{lemma}

\begin{proof} 
\bse
{\bf 1.} First of all, recalling that $\delu_a = t^{-1}L_a$, we obtain
\be
\aligned
\Box_g\Psi  & =  \gu^{00}\nabla_{t}(\nabla_t\Psi) + \gu^{a0}t^{-1}\nabla_{L_a}(\nabla_t\Psi) 
+ \gu^{0a}\nabla_t(t^{-1}\nabla_{L_a}\Psi) + \gu^{ab}\nabla_{\delb_b}\big(t^{-1}\nabla_{L_a}\Psi\big)
\\
& \quad + g^{\alpha\beta}\del_{\alpha}\big(\Psiu_{\beta}^{\nu}\big)\nabla_{\delu_\nu}\Psi
- \nabla_W\Psi.
\endaligned
\ee
Then we apply the formulas 
\be
\nabla_{\alpha}\Psi = \widehat{\del_{\alpha}}\Psi 
+ \frac{1}{4}g^{\mu\nu}\del_{\mu}\cdot\nabla_{\nu}\del_{\alpha}\cdot\Psi,
\quad
\nabla_{L_a}\Psi = \widehat{L_a}\Psi 
+ \frac{1}{4}g^{\mu\nu}\del_{\mu}\cdot\nabla_{\nu}L_a\cdot\Psi, 
\ee
\ese
together with Proposition~\ref{prop1-22-june-2025}. Finally, by a tedious calculation we obtain
\begin{subequations}\label{eq15-16-aout-2025}
\begin{equation}\label{eq15a-16-aout-2025}
R_2[H,\Psi]:=t^{-1}\Big(2\gu^{0a}\widehat{L_a}(\widehat{\del_t}\Psi) 
+ (x^b/t)\gu^{ab}\widehat{L_a}(\widehat{\del_t}\Psi) 
+ \gu^{ab}\widehat{L_a}(\widehat{\del_b}\Psi)\Big),
\end{equation}
\begin{equation}\label{eq15b-16-aout-2025}
\aligned
R_1[H,\Psi] 
: & =  \frac{1}{2}\gu^{00}g^{\mu\nu}\del_{\mu}\cdot\nabla_{\nu}\del_t\cdot\widehat{\del_t}\Psi
+\frac{1}{2t}\gu^{a0}g^{\mu\nu}\del_{\mu}\cdot\nabla_{\nu}L_a\cdot\widehat{\del_t}\Psi
+g^{\alpha\beta}\del_{\alpha}\big(\Psiu_{\beta}^0\big)\widehat{\del_t}\Psi 
\\
& \quad +t^{-1}\gu^{0a}\widehat{\del_a}\Psi 
+ t^{-1}\gu^{ab}\delta_{ab}\widehat{\del_t}\Psi
+ (x^b/t^2)\gu^{ab}\widehat{\del_a}\Psi
\\
& \quad +\frac{1}{2t}\gu^{a0}g^{\mu\nu}\del_{\mu}\cdot\nabla_{\nu}\del_t\cdot\widehat{L_a}\Psi
+\frac{1}{4t}\gu^{ab}g^{\mu\nu}\del_{\mu}\cdot\big((x^a/t)\nabla_{\nu}\del_t + \nabla_{\nu}\del_b\big)\cdot\widehat{L_a}\Psi
\\
& \quad -t^{-2}\gu^{0a}\widehat{L_a}\Psi
+\frac{1}{4t^2}\gb^{ab}g^{\mu\nu}\del_{\mu}\cdot\nabla_{\nu}L_b\cdot\widehat{L_a}\Psi
- \frac{x^a}{t^3}\gb^{ab}\widehat{L_a}\Psi,
\endaligned
\end{equation}
\begin{equation}\label{eq15c-16-aout-2025}
\aligned
R_0[H,\Psi] 
&:= \frac{1}{4}\gu^{00}\nabla_t\big(g^{\mu\nu}\del_{\mu}\cdot\nabla_{\nu}\del_t\big)\cdot\Psi
+\frac{1}{4t}\gu^{ab}\nabla_{\delb_a}\big(g^{\mu\nu}\del_{\mu}\cdot\nabla_{\nu}L_b\big)\cdot\Psi
\\
& \quad+\frac{1}{4t}\gu^{0a}\nabla_t\big(g^{\mu\nu}\del_{\mu}\cdot\nabla_{\nu}L_a\big)\cdot\Psi
+\frac{1}{4t}\gu^{a0}\nabla_{L_a}\big(g^{\mu\nu}\del_{\mu}\cdot\nabla_{\nu}\del_t\big)\cdot\Psi
\\
& \quad+\frac{1}{16}\gu^{00}g^{\mu\nu}g^{\mu'\nu'}\del_{\mu}\cdot\nabla_{\nu}\del_t\cdot\del_{\mu'}\cdot\nabla_{\nu'}\del_t\cdot\Psi
\\
& \quad+\frac{1}{16t}\gu^{a0}g^{\mu\nu}g^{\mu'\nu'}\del_{\mu}\cdot\nabla_{\nu}L_a\cdot\del_{\mu'}\cdot\nabla_{\nu'}\del_t\cdot\Psi
\\
& \quad+\frac{1}{16t}\gu^{0a}g^{\mu\nu}g^{\mu'\nu'}\del_{\mu}\cdot\nabla_{\nu}\del_t\cdot\del_{\mu'}\cdot\nabla_{\nu'}L_a\cdot\Psi
\\
& \quad+\frac{1}{16t^2}\gu^{ab}g^{\mu\nu}g^{\mu'\nu'}\del_{\mu}\cdot\nabla_{\nu}L_b\cdot\del_{\mu'}\cdot\nabla_{\nu'}L_a\cdot\Psi
\\
& \quad- \frac{1}{4t^2}\gu^{0a}g^{\mu\nu}\del_{\mu}\cdot\nabla_{\nu}L_a\cdot\Psi
- \frac{1}{4t^2}(x^b/t)\gu^{ab}g^{\mu\nu}\del_{\mu}\cdot\nabla_{\nu}L_a\cdot\Psi
\\
& \quad+ \frac{1}{4}g^{\alpha\beta}\del_{\alpha}\big(\Psiu_{\beta}^0\big)g^{\mu\nu}\del_{\mu}\cdot\nabla_{\nu}\del_t\cdot\Psi
\\
& \quad- \frac{1}{8t}\gu^{0a}\big(g^{\alpha\beta}\pi[L_a]^{\mu\nu}\pi[\del_t]_{\beta\nu}\del_{\mu}\cdot\del_{\alpha} + \pi[\del_t]^{\alpha\beta}\pi[L_a]_{\alpha\beta}\big)\cdot\Psi
\\
& \quad - \frac{1}{8t}\gb^{ab}\big(g^{\alpha\beta}\pi[L_a]^{\mu\nu}\pi[\del_b]_{\beta\nu}\del_{\mu}\cdot\del_{\alpha} + \pi[\del_b]^{\alpha\beta}\pi[L_a]_{\alpha\beta}\big)\cdot\Psi
\\
& \quad
- \frac{x^b}{8t^2}\gb^{ab}\big(g^{\alpha\beta}\pi[L_a]^{\mu\nu}\pi[\del_t]_{\beta\nu}\del_{\mu}\cdot\del_{\alpha} + \pi[\del_t]^{\alpha\beta}\pi[L_a]_{\alpha\beta}\big)\cdot\Psi
\endaligned
\end{equation}
\ese

\vskip.3cm

{\bf 2.} For the derivation of the desired estimates, we apply Lemma~\ref{lem1-01-aout-2025} and the following observation
\be
\aligned
\big|\del_{\alpha}\big(\Psiu_0^{\beta}\big)\big| & \lesssim  t^{-1},
\\
\big|\pi[L_a]^{\mu\nu}\big| & \lesssim  |H|_1,
\qquad
\big|\pi[\del_{\alpha}]_{\alpha\beta}\big|\lesssim |\del H|. \hfill \qedhere
\endaligned
\ee
\end{proof}


Now we are ready to establish the following preliminary, zero-order estimate. 

\begin{lemma}
\label{lem3-01-aout-2025}
Consider the extended light cone domain $\Mcal_{[s_0,s_1]}\cap\{ r\leq 3t\}$. Assume \eqref{eq-USA-condition} and  
\begin{equation}
|H|_1\lesssim 1,\quad |\del H| + |\del\del H|\lesssim \eps_s\zeta
\end{equation}
hold for a sufficiently small $\eps_s$, together with the curvature condition
\begin{equation}\label{eq3-01-aout-2025}
|R_g|\leq 2M^2.
\end{equation}
Then for any sufficiently regular solution $\Psi$  to the massive Dirac equation 
\begin{equation}\label{eq2-01-aout-2025}
\opDirac \Psi + \mathrm{i}M\Psi = \Phi,
\end{equation}
in the domain in $\Mcal_{[s_0,s_1]}\cap\{ r\leq 3t\}$, one has 
\begin{equation}
\aligned
|\Psi|_{\vec{n}}\lesssim_M& \Big(\frac{|t-r|}{t} + |H^{\Ncal00}|\Big)\big|\widehat{\del_t}(\widehat{\del_t}\Psi)\big|_{\vec{n}}
\\
& \quad + t^{-1}[\del\Psi]_{1,1}
+t^{-1}\zetab^{-1}(1+t|\del H|)|\del \Psi|_{\vec{n}} + t^{-2}\zetab^{-1}(1+t|\del H|)[\Psi]_{1,1},
\\
& \quad + \big|\opDirac\Phi - \mathrm{i}M\Phi\big|_{\vec{n}} + |\nabla_W\Psi|_{\vec{n}}.
\endaligned
\end{equation}
\end{lemma}

\begin{proof} 
\bse
We derive first an estimate on the metric component $\gu^{00}$. 
\begin{equation}\label{eq14-16-aout-2025}
\aligned
\gu^{00}  & =  \underline{\eta}^{00} + \Hu^{00} 
= - \frac{t^2-r^2}{t^2} + H^{00} - \frac{2x^a}{t}H^{a0} + \frac{x^b}{t^2}H^{ab}
\\
& =  \frac{r^2-t^2}{t^2} + H^{\Ncal00} + \frac{2x^a}{r}H^{a0}\Big(1- \frac{r}{t}\Big) 
+ \frac{x^ax^b}{r^2}H^{ab}\Big(\frac{r^2}{t^2}-1\Big).
\endaligned
\end{equation}
Then provided that \eqref{eq2-14-june-2025} holds with a sufficiently small $\eps_s$, we have $|H^{\alpha\beta}|\lesssim 1$ and we thus obtain
\begin{equation}\label{eq9-01-aout-2025}
\big|\gu^{00}\big|\lesssim \frac{|t-r|}{t} + |H^{\N00}|. 
\end{equation}
Next, recalling \eqref{eq1-01-aout-2025} in combination with Lemma~\ref{lem1-30-july-2025}, we have
\begin{equation}\label{eq1-16-aout-2025}
\aligned
M^2\Psi  & =  \Box_g\Psi + \frac{1}{4}R\Psi + \opDirac\Phi - \mathrm{i}M\Phi
\\
& =  \gu^{00}\widehat{\del_t}\big(\widehat{\del_t}\Psi\big) + R_2[H,\Psi] + R_1[H,\Psi] + R_0[H,\Psi] + \frac{1}{4}R\Psi + \opDirac\Phi - \mathrm{i}M\Phi - \nabla_W\Psi
\endaligned
\end{equation}
which leads us to, provided that $|R|\leq 2M^2$ and $\eps_s$ sufficiently small,
$$
\aligned
|\Psi|_{\vec{n}}\lesssim_M&|\gu^{00}|\big|\widehat{\del_t}(\widehat{\del_t}\Psi)\big|_{\vec{n}}
+ t^{-1}[\del\Psi]_{1,1}
+t^{-1}\zetab^{-1}(1+t|\del H|)|\del \Psi|_{\vec{n}} + t^{-2}\zetab^{-1}(1+t|\del H|)[\Psi]_{1,1},
\\
& \quad + \big|\opDirac\Phi - \mathrm{i}M\Phi\big|_{\vec{n}} + |\nabla_W\Psi|_{\vec{n}}.
\endaligned
$$
This concludes the proof. 
\ese
\end{proof}


\paragraph{Second main result.}

We finally reach the following main result of this section. 

\begin{proposition}[Pointwise estimates for massive Dirac fields. II]
\label{prop1-17-aout-2025}
Assume that \eqref{eq-USA-condition} and, in addition,
\bse
\label{eq-KG-condition} 
\bel{eq9-16-aout-2025}
|H|_{p+1}\lesssim \eps_s\zeta,\qquad |W|_{p+1}\lesssim \eps_s, 
\ee
\bel{eq9-16-aout-2025-bis}
|\del H|_{p,k} + |\del\del H|_{p,k} \lesssim \eps_s\zeta,\quad |R|_{p,k}\lesssim \eps_s, 
\end{equation}
\ese
hold for a sufficiently small $\eps_s$. Then for any sufficiently regular solution $\Phi$ to the Dirac equation (with $M>0$) 
\begin{equation}\label{eq1-15-aout-2025}
\opDirac \Psi + \mathrm{i}M\Psi = 0
\end{equation}
in the near-light-cone domain $\Mcal_{[s_0,s_1]}\cap \{3t/4\leq r\leq 3t\}$,
one has 
\bel{eq12-19-aout-2025}
[\Psi]_{p,k}\lesssim_{M,p} \Big(\frac{|t-r|+1}{t} + |H^{\Ncal00}|_{p,k}\Big) [\del\del\Psi]_{p,k}
+ \big(t^{-1} + |\del H|\big)[\del\Psi]_{p+1,k+1} + t^{-1}[\Psi]_{p+1,k+1}.
\ee
\end{proposition}

\begin{proof} {\bf 1.} 
\bse
We only need a proof in the region $\Mcal_{[s_0,s_1]}\cap\{3t/4\leq r\leq 3t\}$, since when $r\leq 3t/4$ or $r\geq 3t$, $\frac{|t-r|+1}{t}\geq 1$.
When \eqref{eq1-15-aout-2025} holds, \eqref{eq1-16-aout-2025} reduces to
\be
M^2\Psi = \gu^{00}\widehat{\del_t}\big(\widehat{\del_t}\Psi\big) + R_2[H,\Psi] + R_1[H,\Psi] + R_0[H,\Psi] + \frac{1}{4}R\Psi - \nabla_W\Psi.
\ee
Then for any admissible operator $\mathscr{Z}^I$ of type $(p,k)$, we can write 
\begin{equation}\label{eq5-17-aout-2025}
\aligned
M^2\mathscr{Z}^I\Psi  & =  \mathscr{Z}^I\big(\gu^{00}\widehat{\del_t}(\widehat{\del_t}\Psi)\big) 
+ \mathscr{Z}^I\big(R_2[H,\Psi] + R_1[H,\Psi] + R_0[H,\Psi]\big) 
\\
& \quad - \frac{1}{4}\mathscr{Z}^I(R\Psi) - \mathscr{Z}^I\big(\nabla_W\Psi\big).
\endaligned
\end{equation}
We need to control each term in the right-hand side.
\ese

\hskip.3cm
\bse
\noindent{\bf 2.} For the first term, we observe that
\begin{equation}
\Big|\mathscr{Z}^I\Big(\frac{r^2-t^2}{t^2}\Big)\Big|\lesssim_p \frac{|t-r|+1}{t},\quad\text{in }\Mcal_{[s_0,s_1]}\cap\{3t/4\leq r\leq 3t\}.
\end{equation}
This is due to the homogeneity of $\frac{r^2-t^2}{t^2}$ in $\Mcal_{[s_0,s_1]}\cap\{r\leq 3t\}$ and the identity
\be
L_a\Big(\frac{r^2-t^2}{t^2}\Big) = \frac{2x^a}{t}\Big(\frac{r^2-t^2}{t^2}\Big).
\ee
In the same manner, we have 
\begin{equation}
\Big|\mathscr{Z}^I\Big(\frac{t-r}{t}\Big)\big|\lesssim_p \frac{|t-r|+1}{t} \qquad \text{ in } \Mcal_{[s_0,s_1]}\cap\{3t/4\leq r\leq 3t\}.
\end{equation}
Thus from \eqref{eq14-16-aout-2025} we obtain 
\be
\big|\mathscr{Z}^I(\gu^{00})\big|\lesssim_p \frac{|t-r|+1}{t}\big(1 + |H|_{p,k}\big) + |H^{\Ncal00}|_{p,k}.
\ee
Then thanks to \eqref{eq9-16-aout-2025} we have 
\begin{equation}
\big|\mathscr{Z}^I(\gu^{00})\big|\lesssim_p \frac{|t-r|+1}{t}+ |H^{\Ncal00}|_{p,k}.
\end{equation}
Then thanks to \eqref{eq1-17-july-2025},
\begin{equation}\label{eq1-17-aout-2025}
\big[\,\gu^{00}\widehat{\del_t}(\widehat{\del_t}\Psi)\big]_{p,k}
\lesssim_p \Big(\frac{|t-r|+1}{t}+|H^{\Ncal00}|_{p,k}\Big)[\del\del\Psi]_{p,k}.
\end{equation}
\ese

\hskip.3cm

\bse
\noindent{\bf 3.} For the terms $R_i[H,\Psi]$, the estimate relies on the explicit expressions \eqref{eq15-16-aout-2025} and Proposition~\ref{prop1-15-july-2025}. We also need to point out that, thanks to the homogeneity of $x^a,t$ and $(x^a/r), (x^a/t)$ in $\{3t/4\leq r\leq 3t\}$,
\begin{equation}
\aligned
&|\gu^{\alpha\beta}|_{p,k} + |g^{\alpha\beta}|_{p,k}\lesssim_p |H|_{p,k}\lesssim \eps_s\zeta,
\\
&|\nabla_{\mu}\del_{\alpha}|_{p,k}\lesssim_p |\del H|_{p,k},
\\
&|\nabla_{\mu}L_a|_{p,k}\lesssim_p (1 + t|\del H|_{p,k}).
\endaligned
\end{equation}
Then we have 
\begin{equation}\label{eq4-17-aout-2025}
\aligned
\big|R_2[H,\Psi]\big|_{\vec{n}}& \lesssim_p  t^{-1}[\del\Psi]_{p+1,k+1},
\\
\big|R_1[H,\Psi]\big|_{\vec{n}}& \lesssim_p  t^{-1}[\Psi]_{p+1,p+1} 
+ \zetab^{-1}\Big(\frac{|t-r|+1}{t} + |H^{\Ncal00}|_{p,k} + 1\Big)|\del H|[\del\Psi]_{p,k} 
\\
& \quad + t^{-1}\zetab^{-1}[\del\Psi]_{p,k},
\\
\big|R_0[H,\Psi]\big|_{\vec{n}}& \lesssim_p  \eps_s[\Psi]_{p,k} + t^{-1}[\Psi]_{p,k}.
\endaligned
\end{equation}
Here, the estimate on the remainder $R_2$ is direct. For $R_1$, we write each term into the one of the following two forms: 
\be
S^{\alpha_1\alpha_2\cdots\alpha_n}\del_{\alpha_1}\cdot\del_{\alpha_2}\cdot\cdots\cdot\del_{\alpha_n}\cdot\widehat{\del_{\beta}}\Psi,
\quad
T^{\alpha_1\alpha_2\cdots\alpha_n}\del_{\alpha_1}\cdot\del_{\alpha_2}\cdot\cdots\cdot\del_{\alpha_n}\cdot\widehat{L_a}\Psi.
\ee
Then we apply \eqref{eq1-17-july-2025} and \eqref{eq3-23-july-2025}. We thus only need to bound the coefficients $S^{\alpha_1\alpha_2\cdots\alpha_n}$ and/or $T^{\alpha_1\alpha_2\cdots\alpha_n}$. For example
\be
\aligned
& t^{-2}\gb^{ab}g^{\mu\nu}\del_{\mu}\cdot\nabla_{\nu}L_b\cdot\widehat{L_a}\Psi
\\
& = t^{-2}\gu^{ab}g^{\mu\nu}\del_{\mu}\cdot \nabla_{\nu}(t\del_b+x^b\del_t)\cdot\widehat{L_a}\Psi
\\
& = t^{-2}\gu^{ab}g^{\mu0}\del_{\mu}\cdot\del_b\cdot\widehat{L_a}\Psi
+ t^{-2}\gu^{ab}g^{\mu b}\del_{\mu}\cdot\del_t\cdot\widehat{L_a}\Psi
\\
& \quad +t^{-1}\gu^{ab}g^{\mu\nu}\Gamma_{\nu b}^{\gamma}\del_{\mu}\cdot \del_{\gamma}\cdot\widehat{L_a}\Psi
+t^{-1}(x^b/t)\gu^{ab}g^{\mu\nu} \Gamma_{\nu 0}^{\gamma}
\del_{\mu}\cdot\del_{\gamma}\cdot\widehat{L_a}\Psi. 
\endaligned
\ee
Then, by the homogeneity of $t^{-1}, (x^a/t)$ and \eqref{lem1-11-july-2025} and provided that \eqref{eq9-16-aout-2025} holds, we arrive at 
\begin{equation}
\aligned
& \big|t^{-2}\gu^{ab}g^{\mu0}\big|_{p,k} 
+ \big|t^{-2}\gu^{ab}g^{\mu b}\big|_{p,k}\lesssim_p t^{-2},
\\
& \big|t^{-1}\gu^{ab}g^{\mu\nu}\Gamma_{\nu b}^{\gamma}\big|_{p,k}
+\big|t^{-1}(x^b/t)\gu^{ab}g^{\mu\nu} \Gamma_{\nu 0}^{\gamma}\big|_{p,k}
\lesssim_pt^{-1}|\del H|_{p,k}.
\endaligned
\end{equation}
Thus thanks to Proposition~\ref{prop1-15-july-2025}, we obtain 
\begin{equation}
\big|t^{-2}\gb^{ab}g^{\mu\nu}\del_{\mu}\cdot\nabla_{\nu}L_b\cdot\widehat{L_a}\Psi\big|_{p,k}
\lesssim \zetab^{-1}(t^{-2}+t^{-1}|\del H|_{p,k})[\Psi]_{p+1,k+1}
\lesssim t^{-1}[\Psi]_{p+1,k+1}.
\end{equation}
For the terms in $R_3$, the strategy is similar. We write each terms in the form
\be
S^{\alpha_1\alpha_2\cdots\alpha_n}\del_{\alpha_1}\cdot\del_{\alpha_2}\cdot\cdots\cdot\del_{\alpha_n}\cdot\widehat{\del_{\beta}}\Psi.
\ee
Observe that for the decomposition of the first four terms in \eqref{eq15c-16-aout-2025}, we apply \eqref{eq16-16-aout-2025}.
\ese

\hskip.3cm

\bse
\noindent {\bf 4.} For the term involving the scalar curvature $R_g$, we have 
\begin{equation}\label{eq2-17-aout-2025}
|R_g \Psi|_{\vec{n}}\lesssim |R_g |_{p,k}[\Psi]_{p,k}\lesssim \eps_s[\Psi]_{p,k}.
\end{equation}
For the term involving the generalized wave gauge condition, we have 
\be
\nabla_W\Psi = W^{\alpha}\nabla_{\alpha}\Psi 
= W^{\alpha}\widehat{\del_{\alpha}}\Psi 
+ \frac{1}{4}g^{\mu\nu}W^{\alpha}\del_{\mu}\cdot\nabla_{\nu}\del_{\alpha}\cdot\Psi
\ee
and then
\begin{equation}\label{eq3-17-aout-2025}
\aligned
\big|\nabla_W\Psi\big|_{\vec{n}}& \lesssim_p  |W|_{p,k}[\del \Psi]_{p,k} + \zetab^{-1}|\del H||W|_{p,k}[\Psi]_{p,k}
\\
& \lesssim  |W|_{p,k}[\del\Psi]_{p,k} + \eps_s[\Psi]_{p,k}.
\endaligned
\end{equation}
\ese

\hskip.3cm

\noindent{\bf 5.} Finally, we substitute \eqref{eq1-17-aout-2025}, \eqref{eq4-17-aout-2025}, \eqref{eq2-17-aout-2025} and \eqref{eq3-17-aout-2025} into \eqref{eq5-17-aout-2025}. We conclude with 
\be
\aligned
\, [\Psi]_{p,k}\lesssim_{M,p}& \Big(\frac{|t-r|+1}{t} + |H^{\Ncal00}|_{p,k}\Big) [\del\del\Psi]_{p,k}
+ t^{-1}[\del\Psi]_{p+1,k+1}+ t^{-1}[\Psi]_{p+1,k+1}
\\
& \quad + \big(t^{-1} + |\del H|\big)\zetab^{-1}[\del\Psi]_{p,k} .
\endaligned
\ee
We then apply Proposition~\ref{prop1-24-july-2025} on $[\del\Psi]_{p,k}$ and arrive at the desired estimate.
\end{proof}

}


\clearpage 

\part{Nonlinear stability of self-gravitating massive Dirac fields}
\label{part-two}

\section{Preliminary to the bootstrap method for the Einstein-Dirac system}
\label{section=N11}

\subsection{ Objective and list of flatness conditions}

{ 

\paragraph{Strategy.}

As a preliminary, we now establish a bridge between our analysis of the Einstein-Dirac system in the present Monograph and the bootstrap framework for the Einstein equations coupled to a massive real-valued field developed in our earlier work~\cite{PLF-YM-two,PLF-YM-PDE}. Our aim is to record, in one place, the estimates that will be invoked repeatedly later, so that we can refer to them without restating material from previous sections.

We group the ingredients of the proof as follows. 
\bei
\item[$\bullet$] Energy estimates, established in Proposition~\ref{prop1-04-april-2025}.

\item[$\bullet$] Sobolev estimates, established in Theorem~\ref{thm2-05-april-2025}.

\item[$\bullet$] Pointwise estimates for the Klein--Gordon equation along orthogonal curves, established in Proposition~\ref{prop2-14-aout-2025}.

\item[$\bullet$] Pointwise inequalities for Dirac commutators, established in Proposition~\ref{prop1-19-july-2025}, Lemma~\ref{lem1-09-oct-2025}.
\eei

\noindent In addition, to streamline the use of the four estimates above, we will rely on the following auxiliary bounds. 
\bei
\item[$\bullet$] Estimates for Clifford products, stated in Proposition~\ref{prop1-15-july-2025}.

\item[$\bullet$] Pointwise estimates for massive spinor fields, stated in Propositions~\ref{prop1-24-july-2025} and~\ref{prop1-17-aout-2025}.

\eei

The results listed here are presented as \textit{formulae praeparatae} (i.e., ready--made templates) so that they can be applied directly within the bootstrap argument. All these estimates hold provided the spacetime metric is sufficiently close to the Minkowski metric, in the sense specified earlier for each statement (smallness of the perturbation $H$, the light-bending condition, and/or the generalized wave coordinate vector $W$, together with the associated decay assumptions). Accordingly, we either assume these flatness conditions at the initial step of the bootstrap or check that they propagate in the bootstrap arguments. Throughout, we work with a collection of $L^2$- and $L^\infty$-type bounds for the metric and matter components up to a sufficiently high order $N$.


\paragraph{Conditions on the spacetime metric.}

It is convenient to list our conditions on the metric perturbation $H$ and the wave coordinate vector $W$. 
\bse\label{eq-spin-condition}
\bei
\item[$\bullet$] \textbf{Uniform spacelike property:}
\begin{equation}\label{eq-US-condition}
\aligned
H^{\Ncal 00}  & :=  g(\diff t- \diff r, \diff t- \diff r)<0,
\qquad && \text{in }\{3t/4\leq r\leq r^{\Ecal}(s)\}\cap \Mcal_{[s_0,s_1]},
\\
|H|  & :=  \max_{\alpha,\beta}|H_{\alpha\beta} | \leq \eps_s\zeta, \quad && \text{in }\Mcal_{[s_0,s_1]}. 
\endaligned
\end{equation}
These are the assumptions in~\eqref{eq-USA-condition}, stating that the metric is close to flat Minkowski, so the slices $\Mcal_s$ are spacelike and the volume form is bounded from below (see Claim~\ref{lem1-02-march-2025}). 

\item[$\bullet$] \textbf{Light-bending property:}
\begin{equation}\label{eq-bending-condition}
H^{\Ncal00} < 0\quad \text{in }\{r\geq 3t/4\}\cap \Mcal_{[s_0,s_1]},
\end{equation}
This generalizes the first condition in~\eqref{eq-US-condition} to a larger domain. Because it will play a distinct role, we list it separately. In Section~\ref{section=N13}, we will explicitly construct a light-bending coordinate chart relative to the reference metric; within the bootstrap mechanism we will show that the perturbation does not destroy this property.

\item[$\bullet$] \textbf{Sobolev-related conditions:}
\begin{equation}\label{eq-USSR-condition}
\aligned
&(s/t)^{-1}|H|_1 + \frac{(s/t)^{-2}|H^{\Ncal00}|_1}{1+(s/t)^{-2}|H^{\Ncal00}|}\lesssim (s/t)^{- \delta},\quad\text{with}\quad 0\leq \delta\leq 1\quad \text{in } \Mcal^{\Hcal}_{[s_0,s_1]},
\\
& \zeta^{-2}|\del H| + \zeta^{-1}|H|_1 + \frac{\zeta^{-2}|H^{\Ncal00}|_1}{1+\zeta^{-2}|H^{\Ncal00}|}\lesssim \zeta^{- \delta},\quad \text{with}\quad 0< \delta\leq 1,\quad \text{ in }  \Mcal^{\ME}_{[s_0,s_1]}.
\endaligned
\end{equation}
These conditions (introduced in~\eqref{eq5a-05-april-2025} and~\eqref{eq6-17-aout-2025}) ensure the metric is sufficiently close to Euclidean so that Klainerman–Sobolev inequalities hold, with an acceptable loss.  These conditions may appear at first to be rather restrictive; however, with a suitable choice of light-bending chart, they are easily satisfied.

\eei
\ese

\bse\label{eq-spin-condition-N}
The two groups of conditions above are zero–order, that is, independent of the differentiation order $N$. In contrast, the following conditions depend on $N$.
\bei

\item[$\bullet$] \textbf{Generalized wave-coordinate condition:}
\begin{equation}\label{eq-France-condition}
|W|_{N-2} = \max_{\alpha}|W^{\alpha}|_{N-2}\lesssim_N \eps_s\mathbbm{1}_{\{r\geq 3t/4\}}r^{-1- \kappa},\quad \quad \kappa> 0, 
\end{equation}
We recall the decomposition $W := W^{\alpha}\del_{\alpha} = g^{\alpha\beta}\nabla_{\alpha}\del_{\beta}$. When $\eps_s=0$, this reduces to the standard wave gauge, in which Einstein's equations form a quasilinear wave system~\cite{CBG}. 
They play a central role in the bootstrap mechanism that was developed for Einstein's vacuum equations~\cite{LR1} and for the Klein-Gordon-Einstein equations~\cite{PLF-YM-two}. Although \eqref{eq-France-condition} may appear in tension with the light-bending condition~\eqref{eq-bending-condition}, Section~\ref{section=N13} will show that the wave structure survives the light-bending change of variables up to lower-order remainders encoded in \eqref{eq-France-condition}.

\item[$\bullet$] \textbf{High-order condition:}
\begin{equation}\label{eq-China-condition}
[H]_{[N/2]+1}\lesssim \zeta.
\end{equation}
We have introduced this condition in order to suppress the effect of irrelevant (cubic and higher-order) terms appearing in commutation relations and estimates on Clifford products. 

\item[$\bullet$] Treating conditions:
\begin{equation}\label{eq-UK-condition}
\aligned
|H|_{N-3} & \lesssim \eps_s\zeta,\qquad\quad  
& |W|_{N-2} & \lesssim \eps_s,
\\
|\del H|_{N-3} + |\del\del H|_{N-3} & \lesssim \eps_s\zeta,\quad \qquad
& |R_g|_{N-3} & \lesssim \eps_s, 
\endaligned
\end{equation}
\eei
where $R_g$ is the scalar curvature. These bounds ensure that Propositions~\ref{prop1-24-july-2025} and~\ref{prop1-17-aout-2025} apply, yielding the desired decay for massive spinors. We stress that all conditions here concern only the spacetime geometry; taking the trace of Einstein’s equations relates $R_g$ to a quadratic expression in the field $\Psi$.

\begin{definition} 
The collection~\eqref{eq-spin-condition} is called the \textbf{$(\eps_s,\delta,\kappa)$–flatness condition}. Unless otherwise stated, we assume~\eqref{eq-spin-condition} throughout our analysis.
\end{definition}

\ese

}


\subsection{ Four key estimates enjoyed by massive spinor fields}

{ 

\paragraph{Notation.}

We consider sufficiently regular solutions $\Psi$ to the massive Dirac equation (with $M>0$)
\be
\opDirac\Psi + \mathrm{i}M\Psi = 0
\ee
defined in the domain $\Mcal_{[s_0,s_1]}$. The spacetime metric $g_{\alpha\beta} = \eta_{\alpha\beta}+H^{\alpha\beta}$ is assumed to satisfy the $(\eps_s,\delta,\kappa)$--flatness condition for some sufficiently small $\eps_s>0$, while $0<\delta \leq 1$ and $0 \leq \kappa\leq 1$. For convenience in the discussion, we introduce the \textbf{high-order Dirac remainders} associated with an arbitrary operator $\mathscr{Z}^I$ 
\be\label{eq10a-18-july-2025}
\aligned
\Phi^I  & :=  \opDirac \big(\mathscr{Z}^I\Psi\big) + \mathrm{i}M\big(\mathscr{Z}^I\Psi\big),
\\
\Phi_{\mu}^I  & :=  \opDirac \big(\mathscr{Z}^I\widehat{\del_{\mu}}\Psi\big) 
+ \mathrm{i}M\big(\mathscr{Z}^I\widehat{\del_{\mu}}\Psi\big).
\endaligned
\ee

Finally, by summation over $\mathscr{Z}^I$ we also introduce the high-order energy densities
\bse
\be
\ebf_{\kappa}^{p,k}[\Psi]: = \sum_{\ord(I)\leq p\atop \rank(I)\leq k}\ebf_{\kappa}[\mathscr{Z}^I\Psi],
\qquad
\ebf_{\kappa}^{p,k}[\del\Psi]: = \sum_{\alpha}\ebf_{\kappa}^{p,k}[\widehat{\del_{\alpha}}\Psi], 
\ee
together with the associated high-order weighted energies
\be
\Ebf_{\kappa}^{p,k}(s,\Psi):= \sum_{\ord(I)\leq p\atop \rank(I)\leq k}\Ebf_{\kappa}(s,\mathscr{Z}^I\Psi),
\qquad
\Ebf_{\kappa}^{p,k}(s,\del\Psi) := \sum_{\alpha}\Ebf_{\kappa}^{p,k}(s,\widehat{\del_{\alpha}}\Psi). 
\ee
\ese


\paragraph{Estimates for massive spinor fields.}

The proof of the following statement is given in Section~\ref{section-prop-spin-energy}. It provides the standard energy estimates for a spinor and its derivative. 

\begin{proposition}[General order energy estimates for massive spinor fields]
\label{prop-spin-energy}
Consider an arbitrary spinor field $\Psi$, while $\Phi^I, \Phi^I_\mu$ denote the associated remainders defined in~\eqref{eq10a-18-july-2025}. 
For any $s\in [s_0,s_1]$ and $0\leq k\leq p\leq N$ and $\kappa\geq 0$, the following two energy estimates hold:
\begin{equation}
\Ebf_{\kappa}^{p,k}(s,\del\Psi)^{1/2}
\leq \Ebf_{g,w}^{p,k}(s_0,\del\Psi)^{1/2}
+ C\sum_{\ord(I)\leq p\atop \rank(I)\leq k}\sum_{\mu}
\int_{s_0}^s\|J\zetab^{-1/2}w|\Phi_{\mu}^I|_{\vec{n}}\|_{L_2(\Mcal_\tau)}\diff \tau,
\end{equation}
\begin{equation}
\Ebf_{\kappa}^{p,k}(s,\Psi)^{1/2}
\leq \Ebf_{g,w}^{p,k}(s,\Psi)^{1/2}
+ C\sum_{\ord(I)\leq p\atop \rank(I)\leq k}
\int_{s_0}^s\|J\zetab^{-1/2}w|\Phi^I|_{\vec{n}}\|_{L_2(\Mcal_\tau)}\diff \tau,  
\end{equation}
in which $C$ denotes a universal constant. 
Furthermore, one has 
\begin{equation}\label{eq5-09-oct-2025}
\aligned
& \Ebf_{\kappa}^{\ME,p,k}(s_1,u) 
+\int_{s_0}^{s_1}s\int_{\Mcal^{\ME}_s}\la r-t\ra^{-1}
\big(\omega^{\kappa}\zeta[\Psi]_{\vec{\gamma},p,k}\big)^2\diff x\diff s
+\Ebf^{\Ccal,p,k}(s_0,s_1;\Psi)
\\
& \leq \Ebf_{g,w}^{\ME,p,k}(s_0,u) 
+  C\sum_{\ord(I)\leq p\atop \rank(J)\leq k}
\int_{s_0}^{s_1}\Ebf_{\kappa}^{\ME,p,k}(s,\Psi)\|s\omega^{\kappa}\zeta^2\zetab^{-1/2}|\Phi^I|_{\vec{n}}\|_{L^2(\Mcal^{\ME}_s)}, 
\endaligned
\end{equation}
where $\vec{\gamma} := \mathrm{grad}(r-t)$, $|\Psi|_{\vec{\gamma}}^2 = \la\Psi,\vec{\gamma}\cdot\Psi\ra_{\ourD}$, and
\be
\Ebf^{\Ccal,p,k}(s_0,s_1;\Psi) := \sum_{\ord(I)\leq p\atop \rank(I)\leq k}
\int_{\Ccal_{[s_0,s_1]}}
\la \mathscr{Z}^I\Psi,\vec{n}_{\Ccal}\cdot\mathscr{Z}^I\Psi\ra_{\ourD}\mathrm{Vol}_{\Ccal}\geq 0.
\ee
\end{proposition}

The proof of the following statement is given in Section~\ref{section-prop-spin-Sobolev}. It provides weighted Sobolev estimates that take advantage of the structure of the Killing fields of Minkowski spacetime. 

\begin{proposition}[General order Sobolev decay for massive spinor fields]\label{prop-spin-Sobolev}
Consider sufficiently regular spinor fields $\Psi$ defined in the slab $\Mcal_{[s_0,s_1]}$.

$\bullet$ In $\Mcal^{\Hcal}_s\subset\Mcal^{\Hcal}_{[s_0,s_1]}$, for $0\leq \delta,\kappa\leq 1$ and all $p\leq N-2$, one has 
\begin{equation}\label{eq1-19-aout-2025}
(s/t)^{-j+1/2+2\delta}t^{3/2}[\Psi]_{p-j,k-j}\lesssim \Ebf_{\kappa}^{\Hcal,p+2,k+2}(s,\Psi)^{1/2},\quad 0\leq j\leq k.
\end{equation}
\bei

\item[$\bullet$] 
In $\{r\leq 3t\}\cap\Mcal^{\ME}_s\subset\{r\leq 3t\}\cap\Mcal^{\ME}_{[s_0,s_1]}$, for $0\leq \kappa\leq 1, 0<\delta\leq 1$ and all $p\leq N-3$ one has
\begin{equation}\label{eq10-03-oct-2025}
\zeta^{-1/2+2\delta}(2+r-t)^{\kappa}(1+r)[\Psi]_{p,k}\lesssim_{\delta}\Ebf_{\kappa}^{\ME,p+3,k+3}(s,\Psi)^{1/2}.
\end{equation}

\item[$\bullet$] In $\{r\geq 3t\}\cap\Mcal^{\ME}_s\subset\{r\geq 3t\}\cap\Mcal^{\ME}_{[s_0,s_1]}$, for $0\leq \kappa\leq 1$ and all $p\leq N-2$ one has 
\begin{equation}\label{eq2-19-aout-2025}
(1+r)^{1+\kappa}[\Psi]_{p,k}\lesssim \Ebf_{\kappa}^{\ME,p+2,k+2}(s,\Psi)^{1/2}.
\end{equation}

\item[$\bullet$] In $\{r\leq 3t\}\cap\Mcal^{\ME}_s\subset\{r\leq 3t\}\cap\Mcal^{\ME}_{[s_0,s_1]}$, for $0\leq \kappa\leq 1, 0<\delta\leq 1$ and all $p\leq N-3$ one has 
\begin{equation}\label{eq3-19-aout-2025}
\aligned
& \zeta^{1/2+2\delta}(2+r-t)^{\kappa}(1+r)[\Psi]_{p,k}
\\
& \lesssim_{\delta}
\Big(\frac{|t-r|+1}{r}+|H^{\Ncal00}|_{p,k}\Big)\Ebf_{\kappa}^{\ME,p+2,k+2}(s,\del\del\Psi)^{1/2}
\\
& \quad+\big(r^{-1}+|\del H|\big)\big(\Ebf_{\kappa}^{\ME,p+3,k+3}(s,\Psi)^{1/2} + \Ebf_{\kappa}^{\ME,p+3,k+3}(s,\del\Psi)^{1/2}\big).
\endaligned
\end{equation}

\item[$\bullet$]
In $\{r\leq 3t\}\cap\Mcal^{\ME}_s\subset\{r\leq 3t\}\cap\Mcal^{\ME}_{[s_0,s_1]}$, for $0\leq \kappa\leq 1, 0<\delta\leq 1$ and all $p\leq N-4$ one has 
\begin{equation}\label{eq7-04-oct-2025}
\aligned
& \zeta^{-1/2+2\delta}(2+r-t)^{\kappa}(1+r)[\Psi]_{p,k}
\\
& \lesssim_{\delta}
\Big(\frac{|t-r|+1}{r}+|H^{\Ncal00}|_{p,k}\Big)\Ebf_{\kappa}^{\ME,p+3,k+3}(s,\del\del\Psi)^{1/2}
\\
& \quad+\big(r^{-1}+|\del H|\big)\big(\Ebf_{\kappa}^{\ME,p+4,k+4}(s,\Psi)^{1/2} + \Ebf_{\kappa}^{\ME,p+3,k+3}(s,\del\Psi)^{1/2}\big).
\endaligned
\end{equation}
\eei 
\end{proposition}


The following statement concerns the remainders $\Phi^I$ and $\Phi_{\mu}^I$ defined in~\eqref{eq10a-18-july-2025}, and is given in Section~\ref{section-prop1-14-aout-2025}.
\begin{proposition}[General order estimates for commutators]\label{prop1-14-aout-2025}
Consider sufficiently regular spinor fields $\Psi$ defined in the slab $\Mcal_{[s_0,s_1]}$.
For any $\mathscr{Z}^I$ an admissible operator of type $(p,k)$, the remainders $\Phi^I$ and $\Phi_{\mu}^I$ satisfy 
\be\label{eq9b-18-july-2025}
\aligned
\zetab \, |\Phi_{\mu}^I|_{\vec{n}}
& \lesssim_p  \sum_{p_1+p_2\leq p\atop k_1+k_2\leq k}
[\del\del\Psi]_{p_1-1,k_1-1}|H|_{p_2+1,k_2+1}
+ \sum_{p_1+p_2\leq p\atop k_1+k_2\leq k}[\del\Psi]_{p_1,k_1}|\del H|_{p_2,k_2}
\\
& \quad +\Hcom_p[\Psi] + \Hwave_p[\Psi],
\endaligned
\ee
where the first term vanishes when $k=0$. On the other hand, one has 
\begin{equation}\label{eq10b-18-july-2025}
\aligned
\zetab \, |\Phi^I|_{\vec{n}} & \lesssim_p  
\sum_{p_1+p_2\leq p\atop k_1+k_2\leq k}[\del\Psi]_{p_1-1,k_1-1}|H|_{p_2+1,k_2+1} 
+\sum_{p_1+p_2\leq p\atop k_1+k_2\leq k}[\del\Psi]_{p_1-1,k_1}|\del H|_{p_2,k_2}
\\
& \quad + \Hcom_{p-1}[\Psi] 
+ \Hwave_{p-1}[\Psi],
\endaligned
\end{equation}
where the first sum vanishes when $k_1=0$ and the second sum vanishes when $p_1=k_1$. 

In the near-light-cone region $\Mcal^{\near}_{[s_0,s_1]} = \Mcal^{\ME}_{[s_0,s_1]}\cap\{r\leq 4t/3\}$, one has
\begin{equation}\label{eq1-14-oct-2025}
\aligned
\zetab\big|\big\la \Phi,\Phi^I_{\mu}\big\ra_{\ourD}\big|
& \lesssim_p
|\Phi|_{\vec{n}}\hspace{-0.3cm}\sum_{p_1+p_2\leq p\atop k_1+k_2\leq k}
[\del\del\Psi]_{p_1-1,k_1-1}|H|_{p_2+1,k_2+1}
\\
& \quad+|\Phi|_{\vec{n}}\hspace{-0.3cm}\sum_{p_1+p_2\leq p\atop k_1+k_2\leq k}
\big([\del\Psi]_{p_1,k_1} + [\Psi]_{p_1,k_1}\big)
\Big(\frac{\la r-t\ra}{r}|\del H|_{p_2,k_2} + |\delsN H|_{p_2,k_2}\Big)
\\
& \quad+\zetab|\Phi|_{\vec{\gamma}}\sum_{p_1+p_2\leq p\atop k_1+k_2\leq k}
[\Psi]_{\vec{\gamma},p_1,k_1}|\del H|_{p_2,k_2}
+t^{-1}|\Phi|_{\vec{n}}\hspace{-0.3cm}\sum_{p_1+p_2\leq p\atop k_1+k_2\leq k}
[\del\Psi]_{p_1,k_1+1}|\del H|_{p_2,k_2}
\\
& \quad+|\Phi|_{\vec{n}}\sum_{p_1+p_2\leq p\atop k_1+k_2\leq k}[\Psihat]_{p_1,k_1}|\del H|_{p_2,k_2}
\\
& \quad+ 
|\Phi|_{\vec{n}}\big(\Hcom_p[\Psi] + \Hwave_p[\Psi]\big).
\endaligned
\end{equation}
\begin{equation}\label{eq2-14-oct-2025}
\aligned
\zetab\big|\big\la \Phi,\Phi^I\big\ra_{\ourD}\big|
& \lesssim_p   
|\Phi|_{\vec{n}}\hspace{-0.3cm}\sum_{p_1+p_2\leq p\atop k_1+k_2\leq k}
[\del\Psi]_{p_1-1,k_1-1}|H|_{p_2+1,k_2+1}
\\
& \quad+|\Phi|_{\vec{n}}\hspace{-0.3cm}\sum_{p_1+p_2\leq p\atop k_1+k_2\leq k}
\big([\del\Psi]_{p_1-1,k_1} + [\Psi]_{p_1-1,k_1}\big)
\Big(\frac{\la r-t\ra}{r}|\del H|_{p_2,k_2} + |\delsN H|_{p_2,k_2}\Big)
\\
& \quad+\zetab|\Phi|_{\vec{\gamma}}\sum_{p_1+p_2\leq p\atop k_1+k_2\leq k}
[\Psi]_{\vec{\gamma},p_1-1,k_1}|\del H|_{p_2,k_2}
+t^{-1}|\Phi|_{\vec{n}}\hspace{-0.3cm}\sum_{p_1+p_2\leq p\atop k_1+k_2\leq k}
[\Psi]_{p_1,k_1+1}|\del H|_{p_2,k_2}
\\
& \quad+|\Phi|_{\vec{n}}\big(\Cubic^+_{p-1}[\Psi] + \Cubic^+_{p-1}[\del\Psi] + \Hwave_{p-1}[\Psi]\big).
\endaligned
\end{equation}
Let $\mathscr{Z}^J$be an admissible operator, $\ord(J) = \rank(J) = p$. Then
\begin{equation}\label{eq4-09-oct-2025}
\zetab|\Phi^J| \lesssim_p \sum_{p_1+p_2\leq p\atop k_1+k_2\leq p}[\del\Psi]_{p_1-1,k_1-1}|H|_{p_2+1,k_2+1}  + \Hcom_{p-1}[\Psi] 
+ \Hwave_{p-1}[\Psi].
\end{equation}
In all the above cases, the right-hand sides $\Hcom_{p}$ and $\Hwave_{p}$ enjoy 
\bse
\begin{equation}
\big|\Hcom_{p}[\Psi]\big|_{\vec{n}}
\lesssim_p \sum_{p_1+p_2+p_3\leq p}[\Psi]_{p_1}|H|_{p_2+1}|\del H|_{p_3}, 
\end{equation}
\begin{equation}
\big|\Hwave_{p}[\Psi]\big|
\lesssim_p\sum_{p_1+p_2\leq p}[\Psi]_{p_1}|W|_{p_2+1}
+ \sum_{p_1+p_2+p_3\leq p}[\Psi]_{p_1}|W|_{p_2+1}|H|_{p_3}.
\end{equation} 
\ese
\end{proposition}


The proof of the following statement is given in Section~\ref{section-prop-spin-orthgonal}. It relies on the technique of integration that we developed for the massive Dirac equation.  

\begin{proposition}[General order pointwise estimates along orthogonal curves]
\label{prop-spin-orthgonal}
Consider sufficiently regular solutions $\Psi$ defined in $\Mcal^{\Hcal}_{[s_0,s_1]}$ to the equation
\begin{equation}\label{eq12-26-aout-2025}
\opDirac\mathscr{Z}^I\Psi + \mathrm{i}M\mathscr{Z}^I\Psi = \Phi^I
\end{equation}
and define the associated pointwise norm 
\be
\Abf_{p,k}[\Psi](s) := \sup_{\Mcal^{\Hcal}_{[s_0,s_1]}}
s^{3/2}\big((s/t)[\del\Psi]_{p,k} + \lapsb M[\Psi]_{p,k}\big).
\ee
Then for all $p\leq N-4$, one has 
\begin{equation}\label{eq11-26-aout-2025}
\aligned
\Abf_{p,k}[\Psi](s)
& \lesssim_{M,N}
(s/t)^{-3/2}\big(\lapsb M[\Psi]_{N-4} + (s/t)^{-1}|\del H|[\Psi]_{N-3} + \Hcom_{p-1}[\Psi]\big)_{\tau_{t,x}(s_{t,x}^*)} 
\\
& \quad+ (s/t)^{-1/2}\big([\Psi]_{N-3}\big)_{\tau_{t,x}(s_{t,x}^*)}   
+s^{1/2}[\Psi]_{N-3} + \zetab^{-1}\Hcom_{p-1}[\Psi]
\\
& \quad + \sum_{p_1+p_2\leq p\atop k_1+k_2\leq k}\int_{s_{t,x}^*}^s 
\Abf_{p_1,k_1}[\Psi](\lambda)\big(\lapsb|LH|_{p_2-1,k_2-1}\big)\big|_{t,x}(\lambda)\diff \lambda
\\
& \quad +\int_{s_{t,x}^*}^s  \Abf_{p,k}[\Psi](\lambda)I_N[H,W]\big|_{t,x}(\lambda) \diff \lambda
+\int_{s_{t,x}^*}^s J_N[H,W,\Psi]\big|_{t,x}(\lambda) \diff \lambda, 
\endaligned
\end{equation}
in which the source contributions $I_N[H,W]$ and $J_N[H,W,\Psi]$ satisfy 
\bse
\begin{equation}
\aligned
I_N[H,W]  & :=  s^{-1}|H|_{N-2} + |\del H|_{N-3} + \zetab^{-1}|\del\Hu^{00}|_{N-3}
\\ & \quad 
+ (s/t)^{-2}\hspace{-0.3cm}\sum_{p_1+p_2\leq N-4}\hspace{-0.4cm}|H|_{p_1+1}|H|_{p_2+1}
+ \big((s/t)^{-1}\lapsb + \zetab^{-1} |\del H|\big)\lapsb|W|,
\endaligned
\end{equation}
\be
\aligned
J_N[H,W,\Psi] & := s^{3/2}\lapsb^2\big(t^{-1} + |\del H|_{N-4}\big)|H|_{N-3}[\del\Psi]_{N-3} 
\\
& \quad + s^{3/2}\lapsb^2\big(|\del W|_{N-3} + |W|_{N-3} + (s/t)^{-1}|\del H|_{N-4}|H|_{N-3}\big)[\Psi]_{N-3}
\\
& \quad + s^{3/2}\lapsb^2\big(s^{-2} + s^{-1}(s/t)^{-1}|\del H|_1\big)[\Psi]_{N-2}
\\
& \quad +s^{3/2}\lapsb^2(s/t)^{-1}\big(t^{-1} + |\del H|\big)|H|_1[\Psi]_{N-2}\sum_{k=0}^2\big((s/t)^{-2}|H|\big)^k.
\endaligned
\end{equation}
\ese
\end{proposition}


\paragraph{Observation on the hierarchy structure.}

In the bootstrap mechanism, we will establish the integral bounds on $I_N[H,W]$ and $J_N[H,W,\Psi]$, namely for some $\eps>0$  
\begin{equation}
\int_{s_{t,x}^*}^sI_N[H,W]|_{t,x}(\lambda)\diff \lambda\lesssim \eps(s/t)^{\alpha},\quad \alpha\geq 0,
\end{equation}
\begin{equation}
\int_{s_{t,x}^*}^sJ_N[H,W]|_{t,x}(\lambda)\diff \lambda\lesssim \eps(s/t)^{\beta},\quad \beta\geq 0,
\end{equation}
as well as the pointwise bound on the spinor field
\begin{equation}
s^{1/2}[\Psi]_{N-3} + \zetab^{-1}\Hcom_{p-1}[\Psi]\lesssim \eps(s/t)^{\beta}. 
\end{equation}
Here, $\eps$ is naturally chosen to measure the smallness of the prescribed initial data. And with the estimates in the merging-Euclidean domain, we will show that
\begin{equation}
\aligned
(s/t)^{-3/2}\big(\lapsb M[\Psi]_{N-4} & \quad + (s/t)^{-1}|\del H|[\Psi]_{N-3} + \Hcom_{p-1}[\Psi]\big)_{\tau_{t,x}(s_{t,x}^*)} 
\\
& \quad + (s/t)^{-1/2}\big([\Psi]_{N-3}\big)_{\tau_{t,x}(s_{t,x}^*)} 
\lesssim \eps(s/t)^{\beta}.
\endaligned
\end{equation}
Here we should observe
\begin{equation}
t^{-1}\simeq (s/t)^2 \qquad \text{ at } \tau_{t,x}(s_{t,x}^*). 
\end{equation}
At this stage, we can apply to~\eqref{eq11-26-aout-2025}  the following version of Gronwall's inequality: 
\be
0 < u(t)\leq \alpha(t) + \int_{t_0}^t\beta(s)u(s)\diff s
\quad\Rightarrow\quad
u(t)\leq \alpha(t) + \int_{t_0}^t\alpha(s)\beta(s)\exp\big(\int_s^t\beta(r)\diff r\big)\diff s. 
\ee
In turn, we conclude that
\be
\Abf_{p,k}[\Psi](s)\lesssim_{M,N} (s/t)^{\beta}\eps + \sum_{p_1+p_2\leq p\atop k_1+k_2\leq k}\int_{s_{t,x}^*}^s 
\Abf_{p_1,k_1}[\Psi](\lambda)\big(\lapsb|LH|_{p_2-1,k_2-1}\big)\big|_{t,x}(\lambda)\diff \lambda.
\ee
The key structure is that, the \textit{sum does not exist when $k=0$}, while, when $k>0$, it only involves the term $\Abf$ at the rank $k-1$. This \textit{hierarchy structure} allows us to apply an induction on $k$, and arrive at the crucial sharp decay estimates.

}


\subsection{ Proof of Proposition~\ref{prop-spin-energy}}
\label{section-prop-spin-energy}

{ 

\paragraph{Estimate of the bulk term.}

Recall that $\omega = \aleph(r-t)$ defined in~\eqref{equa-weight1} and let us use the weight $\omega^{2\kappa}$. 

\paragraph{Proof of Proposition~\ref{prop-spin-energy}.}

\bse
We now apply Proposition~\ref{prop1-04-april-2025} with the weight $w = \omega^{\kappa}+1$ 
and after differentiating~\eqref{eq1a-05-april-2025} with respect to the variable $s_1$, we obtain
\begin{equation}\label{eq2-18-aout-2025}
\frac{\diff}{\diff s}\Ebf_{g,w}(s,\Psi) 
+ \int_{\Mcal_s}\la \Psi,\mathrm{grad}\big(\omega^{2\kappa}\big)\cdot\Psi \ra_{\ourD} \lapsb\mathrm{Vol}_{\sigmab} 
= -2\int_{\Mcal_s}\omega\Re\big(\la\Phi,\Psi\ra_{\ourD}\big)\lapsb\mathrm{Vol}_{\sigmab}.
\end{equation}
Then thanks to Lemma~\ref{lem1-18-aout-2025} and Corollary~\ref{cor1-18-aout-2025}, we find 
\be
\aligned
\Big|\frac{\diff}{\diff s}\Ebf_{g,w}(s,\Psi)\Big|
& \lesssim  \int_{\Mcal_s}\big|w\Re\big(\la\Phi,\Psi\ra_{\ourD}\big)\big|\lapsb\mathrm{Vol}_{\sigmab}
\lesssim\int_{\Mcal_s}w|\Phi|_{\vec{n}}|\Psi|_{\vec{n}}\lapsb\mathrm{Vol}_{\sigmab}
\\
& \lesssim  \Big(\int_{\Mcal_s}w|\Psi|_{\vec{n}}^2\mathrm{Vol}_{\sigmab}\Big)^{1/2}
\Big(\int_{\Mcal_s}w \lapsb^2|\Phi|_{\vec{n}}^2\mathrm{Vol}_{\sigmab}\Big)^{1/2}
\\
& \lesssim  \Ebf_{g,w}(s,\Psi)^{1/2}
\Big(\int_{\Mcal_s}w \lapsb^2|\Phi|_{\vec{n}}^2\mathrm{Vol}_{\sigmab}\Big)^{1/2}.
\endaligned
\ee
Thanks to Claim~\ref{lem1-02-march-2025}, this provides us with 
\be
\frac{\diff }{\diff s}\Ebf_{g,w}(s,\Psi)^{1/2}
\lesssim \big\|w^{1/2}\lapsb|\Phi|_{\vec{n}}\big\|_{L^2_{\sigmab}(\Mcal_s)}
\lesssim \big\|w^{1/2}\lapsb\,\zetab^{1/2}|\Phi|_{\vec{n}}\big\|_{L^2(\Mcal_s)}. 
\ee
Next, we apply Claim~\ref{cor1-16-june-2025} and observe that $\lapsb_{\eta} = J\zeta^{-1}$, hence 
\be
\frac{\diff }{\diff s}\Ebf_{g,w}(s,\Psi)^{1/2}
\lesssim \big\|w^{1/2}J\zetab^{-1/2}|\Phi|_{\vec{n}}\big\|_{L^2(\Mcal_s)}.
\ee
\ese
Integrating this inequality on the time interval $[s_0,s_1]$, we obtain
\bel{eq4-18-aout-2025}
\Ebf_{g,w}(s_1,\Psi)^{1/2}\leq \Ebf_{g,w}(s_0,\Psi)^{1/2} 
+ C\int_{s_0}^{s_1}\big\|w^{1/2}J\zetab^{-1/2}|\Phi|_{\vec{n}}\big\|_{L^2(\Mcal_s)}\diff s, 
\ee
where $C>0$ a universal constant. It remains to apply~\eqref{eq4-18-aout-2025} to~\eqref{eq10a-18-july-2025} for each field $\mathscr{Z}^I$ of type $(p,k)$. This completes the derivation of the desired estimates. 

Finally, the estimate \eqref{eq5-09-oct-2025} is a direct consequence of Proposition~\ref{prop1-05-oct-2025} combined with the notation \eqref{eq10a-18-july-2025}.


}


\subsection{ Proof of Proposition~\ref{prop-spin-Sobolev}} 
\label{section-prop-spin-Sobolev}

{ 

\paragraph{Proof of~\eqref{eq1-19-aout-2025}.}

In Proposition~\ref{prop-spin-Sobolev}, we have three inequalities to establish. We begin with~\eqref{eq1-19-aout-2025}, and combine~\eqref{eq9-16-june-2025} with Proposition~\ref{prop1-24-july-2025}. Indeed, applying~\eqref{eq9-16-june-2025} to the function $\mathscr{Z}^I\Psi$ where $\mathscr{Z}^I$ is of type $(p,k)$, we obtain
\be
(s/t)^{1/2+2\delta}t^{3/2}|\mathscr{Z}^I\Psi|_{\vec{n}}\lesssim \sum_{|J|\leq 2}\|(s/t)^{1/2}\mathscr{L}^J(\mathscr{Z}^I\Psi)\|_{L^2(\Mcal_s)}
\lesssim \Ebf_{\kappa}^{p+2,k+2}(s,\Psi)^{1/2}.
\ee
This leads us to the estimate on $[\Psi]_{p,k}$ corresponding to~\eqref{eq1-19-aout-2025} with $j=0$. Finally, we apply Proposition~\ref{prop1-24-july-2025} inductively ($j$-times) to $[\Psi]_{p-j,k-j}$.

\paragraph{Proof of~\eqref{eq10-03-oct-2025}.} 

This is direct from \eqref{eq11-19-aout-2025} combined with Proposition~\ref{prop1-24-july-2025}.

\paragraph{Proof of~\eqref{eq2-19-aout-2025}.} 

The argument of proof is based on the Sobolev inequality~\eqref{eq11-19-aout-2025}. We only need to point out that, in the region $r\geq 3t$, we have $r\lesssim 2+r-t$ and $\zetab^{-1}\lesssim 1$.


\paragraph{Proof of~\eqref{eq3-19-aout-2025}.} 

We now combine~\eqref{eq5-18-aout-2025} with Proposition~\ref{prop1-17-aout-2025}. We multiply~\eqref{eq12-19-aout-2025} by $\zeta^{1/2+2\delta}(2+r-t)^{\kappa}(1+r)$. We then observe that
\be
\aligned
\zeta^{1/2+2\delta}(2+r-t)^{\kappa}(1+r)[\del\del\Psi]_{p,k}
& \lesssim_{\delta}  \|(2+r-t)^{\kappa}\zeta^{1/2}[\del\del\Psi]_{p+2,k+2}\|_{L^2(\Mcal^{\EM})}
\\
& \lesssim  \Ebf_{\kappa}^{p+2,k+2}(s,\del\del\Psi)^{1/2}.
\endaligned
\ee
For the terms $[\Psi]_{p+1,k+1}$ and $[\del\Psi]_{p+1,k+1}$, we perform the same estimates, and we have thus established~\eqref{eq3-19-aout-2025}. This completes the proof of Proposition~\ref{prop-spin-Sobolev}. 


\paragraph{Proof of~\eqref{eq7-04-oct-2025}.} 

We only need to combine Proposition~\ref{prop1-24-july-2025} with \eqref{eq3-19-aout-2025}.

}


\subsection{ Proof of Proposition~\ref{prop1-14-aout-2025}}
\label{section-prop1-14-aout-2025}

{ 
\bse
Now we are ready to establish Proposition~\ref{prop1-14-aout-2025}. For~\eqref{eq9b-18-july-2025}
in the zero-order case, we begin with
\begin{equation}
\opDirac\big(\widehat{\del_{\mu}}\Psi\big) + \mathrm{i}M\big(\widehat{\del_{\mu}}\Psi\big) = -[\widehat{\del_{\mu}},\opDirac]\Psi
\end{equation}
and, after differentiation,  
\be
\opDirac\big(\mathscr{Z}^I\widehat{\del_{\mu}}\Psi\big) 
+ \mathrm{i}M\big(\mathscr{Z}^I\widehat{\del_{\mu}}\Psi\big)
= -[\mathscr{Z}^I,\opDirac]\widehat{\del_{\mu}}\Psi 
- \mathscr{Z}^I\big([\widehat{\del_{\mu}},\opDirac]\Psi\big) =:\Phi_{\mu}^I. 
\ee
In turn, in view of~\eqref{eq12-18-july-2025} together with Proposition~\ref{prop1-19-july-2025} we deduce that 
\be
\aligned
\zetab \, |\Phi_{\mu}^I|_{\vec{n}}
& \lesssim_p  \sum_{p_1+p_2\leq p\atop k_1+k_2\leq k}
[\del\del\Psi]_{p_1-1,k_1-1}|H|_{p_2+1,k_2+1}
+\sum_{p_1+p_2\leq p\atop k_1+k_2\leq k}[\del\del\Psi]_{p_1-1,k_1}|\del H|_{p_2,k_2}
\\
& \quad +\Hcom_{p-1}[\del\Psi] + \Hwave_{p-1}[\del\Psi] 
+\sum_{p_1+p_2\leq p\atop k_1+k_2\leq k} 
[\del\Psi]_{p_1,k_1}|\del H|_{p_2,k_2}
+ \Hcom_{p}[\Psi] + \Hwave_{p,k}[\del,\Psi], 
\endaligned
\ee
therefore 
\be
\aligned
\zetab \, |\Phi_{\mu}^I|_{\vec{n}}
& \lesssim_p  \sum_{p_1+p_2\leq p\atop k_1+k_2\leq k}
[\del\del\Psi]_{p_1-1,k_1-1}|H|_{p_2+1,k_2+1}
+\sum_{p_1+p_2\leq p\atop k_1+k_2\leq k}[\del\Psi]_{p_1,k_1}|\del H|_{p_2,k_2}
+\Hcom_p[\Psi] + \Hwave_p[\Psi].
\endaligned
\ee
Observe that the first term vanishes when $k=0$. This leads us to~\eqref{eq9b-18-july-2025}. 

A similar argument applies to derive \eqref{eq1-14-oct-2025}, for which we now apply Lemma~\ref{lem1-13-oct-2025} and Proposition~\ref{prop1-14-oct-2025}.
On the other hand, to deal with~\eqref{eq10b-18-july-2025}, we observe that
\be
\opDirac(\mathscr{Z}^I\Psi) + \mathrm{i}M(\mathscr{Z}^I\Psi) = -[\mathscr{Z}^I,\opDirac]\Psi 
\ee
and we apply directly Proposition~\ref{prop1-19-july-2025} and 
Lemma~\ref{lem1-09-oct-2025}
on high-order commutator estimates for the Dirac operator. 
\ese

}


\subsection{ Preliminaries for the proof of Proposition~\ref{prop-spin-orthgonal}}
\label{section-prop-spin-orthgonal-prel}

{ 

\paragraph{First technical estimate.}

We now consider the Dirac equation $\opDirac\mathscr{Z}^I\Psi + \mathrm{i}M\mathscr{Z}^I\Psi = \Phi^I$ in \eqref{eq12-26-aout-2025}, and we are going to apply Proposition~\ref{lem2-24-july-2025}. We need first the following observation about the remainder arising in~\eqref{eq8-22-aout-2025}. 

\begin{lemma}
\label{lem2-22-aout-2025}
Assume the $(\eps_s,\delta,\kappa)$--flatness condition for a sufficiently small $\eps_s>0$ and ${0<\delta\leq 1}$. Then for any admissible operator $\mathscr{Z}^I$ of type $(p,k)$ with $p\leq n-4$, one has 
\begin{equation}\label{eq9-22-aout-2025}
\aligned
\big|R[H,\mathscr{Z}^I\Psi]\big|_{\vec{n}} & \lesssim_p 
s^{-2}[\Psi]_{p+2,k+2} + (s/t)^{-2}t^{-1}|\del H|_1[\Psi]_{p+2,k+2}
\\
& \quad +(s/t)^{-1}\big(t^{-1} + |\del H|\big)|H|_1[\Psi]_{p+2,k+2}\sum_{k=0}^2\big((s/t)^{-2}|H|\big)^k, 
\endaligned
\end{equation} 
where $R[H,\Psi]$ was defined earlier in~\eqref{eq8-22-aout-2025}.
\end{lemma}

\begin{proof} 
\bse
Thanks to Proposition~\ref{prop1-24-july-2025} (i.e.~the pointwise estimates), we have 
\be
[\Psi]_{1,1}\lesssim \zetab[\Psi]_{2,2},\quad |\Psi|_{\vec{n}}\lesssim \zetab[\Psi]_{1,1}.
\ee
We also observe that $\zetab^{-1}\lesssim (s/t)$ in $\Mcal^{\Hcal}_{[s_0,s_1]}$. Therefore,~\eqref{eq8-22-aout-2025} reduces to
\be
\aligned
\big|R[H,\Psi]\big|_{\vec{n}}
& \lesssim  
s^{-2}[\Psi]_{2,2} + |R||\Psi|_{\vec{n}}
\\
& \quad+ (s/t)^{-1}\big(t^{-1} + |\del H|\big)|H|_1[\Psi]_{2,2}\sum_{k=0}^2\big((s/t)^{-2}|H|\big)^k
\\
& \quad
+ s^{-2}\zetab(s/t)^{-2}|\del H||H|[\Psi]_{2,2}
\\
& \quad+t^{-1}\zetab^{-2}|\del H|[\Psi]_{2,2} + t^{-1}\zetab^{-1}|\del H|_1|\Psi|_{2,2}
+ s^{-2}(s/t)^{-1}|H||\Psi|_{2,2}.
\endaligned
\ee
Recalling~\eqref{eq-UK-condition}, we deduce that 
\be\label{equa-28sept-2025b}
\aligned
\big|R[H,\Psi]\big|_{\vec{n}}
& \lesssim  
s^{-2}[\Psi]_{2,2} + |R||\Psi|_{\vec{n}}
+ (s/t)^{-1}\big(t^{-1} + |\del H|\big)|H|_1[\Psi]_{2,2}\sum_{k=0}^2\big((s/t)^{-2}|H|\big)^k
\\
& \quad+t^{-1}\zetab^{-2}|\del H|[\Psi]_{2,2} +  t^{-1}\zetab^{-1}|\del H|_1|\Psi|_{2,2} .
\endaligned
\end{equation}
Next, we differentiate the spinor field and we consider $R[H,\mathscr{Z}^I\Psi]$. It is immediate from the definition that
\be
[\mathscr{Z}^I\Psi]_{p',k'}\lesssim [\Psi]_{p+p',k+k'}, 
\ee
so that~\eqref{equa-28sept-2025b} implies~\eqref{eq9-22-aout-2025}. 
\ese
\end{proof}


\paragraph{Second technical estimate.}

To proceed next with the proof Proposition~\ref{prop-spin-orthgonal}, we need another rather technical observation, stated now and whose proof is postponed to Section~\ref{section=N12}. The estimate below controls how the Dirac operator itself acts on the Dirac remainders in~\eqref{eq10a-18-july-2025}.  

\begin{proposition}[Applying the Dirac operator to the Dirac remainders in the hyperboloidal domain]
\label{lem1-22-aout-2025}
In the hyperboloidal domain $\Mcal^{\Hcal}_{[s_0,s_1]}$ and for any admissible operator $\mathscr{Z}^I$ of type $(p,k)$, one has 
\begin{equation}
\aligned
\big|\opDirac\Phi^I - \mathrm{i}M\Phi^I\big|_{\vec{n}}
\lesssim_N& \sum_{p_1+p_2\leq p\atop k_1+k_2\leq k}[\Psi]_{p_1,k_1}|LH|_{p_2-1,k_2-1} 
+T_N[H]\,[\Psi]_{p,k}
\\
& \quad + S_{N,p}[H,\Psi] + \Hcom_{p}[\Psi] + \Hwave_{p}[\Psi]
\endaligned
\end{equation}
with
\bel{equa-TN-SN}
\aligned
T_N[H] & := s^{-1}|H|_{N-2} + |\del H|_{N-3} + \zetab^{-1}|\del\Hu^{00}|_{N-3}
+ (s/t)^{-2}\hspace{-0.3cm}\sum_{p_1+p_2\leq N-4}\hspace{-0.4cm}|H|_{p_1+1}|H|_{p_2+1},
\\ 
S_{N,p}[H,\Psi]  & :=   \big(t^{-1} + |\del H|_{N-4}\big)|H|_{N-3}[\del\Psi]_{p+1} 
\\
& \quad + \big(|\del W|_{N-3} + |W|_{N-3} + (s/t)^{-1}|\del H|_{N-4}|H|_{N-3}\big)[\Psi]_{p+1}.
\endaligned
\ee
\end{proposition}

Let us observe the following feature of our bootstrap argument. 
The terms $T_N$ in \eqref{equa-TN-SN} will be proven to enjoy uniformly integrable bounds (that is, independent of $(s/t)$) along the orthogonal curves. This property will allow us to apply a Gronwall argument. On the other hand, the terms $S_{N,N-4}$, $\Hcom_{N-4}$ and $\Hwave_{N-4}$ are also integrable along the orthogonal curves and enjoy Sobolev decay. The crucial structure comes in the \textit{first sum in the right-hand side:} this term vanishes when $k_2 = 0\Leftrightarrow k_1\leq k-1$. In particular, when $k=0$, the first sum \underline{does not exist}. This \textit{hierarchy structure} is a key ingredient allowing us to perform an induction on $k$. 

}


\subsection{ Proof of Proposition~\ref{prop-spin-orthgonal}}
\label{section-prop-spin-orthgonal}

{ 

\paragraph{Step 1.}

We now turn our attention to the general order pointwise estimates which are derived by integration along orthogonal curves in the hyperboloidal domain. 
We are going to establish first the following estimate:
\begin{equation}\label{eq6-26-aout-2025}
[\del\Psi]_{p,k}\lesssim_p\sum_{ord(I)\leq p\atop \rank(I)\leq k} \big|\widehat{\del_t}\mathscr{Z}^I\Psi\big|_{\vec{n}}
+ t^{-1}[\Psi]_1 + |\del H||\Psi|_{\vec{n}}
+ \zetab^{-1}\Hcom_{p-1}[\Psi].
\end{equation}
This is a version of Proposition~\ref{prop1-21-july-2025}. In fact we only need to point out that
\begin{equation}\label{eq5-26-aout-2025}
\aligned
\widehat{\del_a}\Psi  & =  \nabla_{\del_a}\Psi 
- \frac{1}{4}g^{\mu\nu}\del_{\mu}\cdot\nabla_{\nu}\del_a\Psi
\\
& =  t^{-1}\nabla_{L_a} - (x^a/t)\nabla_t\Psi 
- \frac{1}{4}g^{\mu\nu}\del_{\mu}\cdot\nabla_{\nu}\del_a\Psi
\\
& = -(x^a/t)\widehat{\del_t}\Psi + t^{-1}\widehat{L_a}\Psi 
\\
& \quad + \frac{1}{4t}g^{\mu\nu}\del_{\mu}\cdot\nabla_{\nu}L_a\cdot\Psi 
- \frac{x^a}{4t}g^{\mu\nu}\del_{\mu}\cdot\nabla_{\nu}\del_t\cdot\Psi
- \frac{1}{4}g^{\mu\nu}\del_{\mu}\cdot\nabla_{\nu}\del_a\cdot\Psi, 
\endaligned
\end{equation}
in which, to deal with the last three terms in the right-hand side, we apply the basic inequalities~\eqref{eq1-31-july-2025} and~\eqref{eq2-31-july-2025}. 

We are then in a position to apply Proposition~\ref{prop2-14-aout-2025} to the equation  
\begin{equation}
\opDirac \big(\mathscr{Z}^I\Psi\big) + \mathrm{i}M\big(\mathscr{Z}^I\Psi\big) = \Phi^I
\end{equation}
for any $\mathscr{Z}^I$ of type $(p,k)$, in which we defined $\Phi^I = -[\mathscr{Z}^I,\opDirac]\Psi$. To proceed, we need to bound several groups of source/error terms, as follows. 


\paragraph{Step 2.} A first group contains the following terms:
\be
P[\Psi] := s^{1/2}\lapsb\big((s/t)[\Psi]_{1,1} + \lapsb|\Psi|_{\vec{n}}\big)
+ s^{1/2}[\Psi]_{1,1}
+ s^{3/2}\zetab^{-1}\lapsb|\del H||\Psi|_{\vec{n}}
\ee
Provided the metric is $(\eps_s,\delta,\kappa)$--flat (so that Corollary~\ref{cor1-16-june-2025} applies), we then have 
\begin{equation}
P[\mathscr{Z}^I\Psi] \lesssim_p s^{1/2}[\Psi]_{N-3} + \eps_s \big(s^{3/2}\lapsb[\Psi]_{p,k}\big).
\end{equation}
We will see that the first term enjoys sufficient decay, while the second term will be absorbed by the left-hand side of the main estimate. 


\paragraph{Step 3.} A second group of terms is associated with the initial data, namely 
\be
D_{t,x}[\Psi]:=
\Big(\big(\big|\nabla_{\vec{L}}(s^{3/2}\Psi)\big|_{\vec{n}}+ \lapsb M|s^{3/2}\Psi|_{\vec{n}}\Big)_{\gamma_{t,x}(s_{t,x}^*)}
\ee
using here a notation introduced in Section~\ref{section===10-3}. We recall that 
\be
\mathcal{I}_{[s_0,s_1]} := \Mcal^{\Hcal}_{s_0}\cup\{r=t-1|s_0\leq s\leq s_1\}
\ee
is the past and null boundary of $\Mcal^{\Hcal}_{[s_0,s_1]}$.
It is immediate that
\be
\aligned
\nabla_{\vec{L}}(s^{3/2}\Psi)  & =  (3/2)s^{1/2}\Psi + (s/t)s^{3/2}\widehat{\del_t}\Psi 
+ s^{3/2}t^{-1}\betab^a\widehat{L_a}\Psi 
\\
& \quad + \frac{s}{4t}s^{3/2}g^{\mu\nu}\del_{\mu}\cdot\nabla_{\nu}\del_t\cdot\Psi
+ \frac{\betab^a}{4t}s^{3/2}\del_{\mu}\cdot\nabla_{\nu}L_a\cdot\Psi.
\endaligned
\ee
We recall that~\eqref{eq8-31-july-2025} leads us to (in view of~\eqref{eq-US-condition})
\be
|\betab^a|\lesssim (s/t)^{-1}.
\ee 
Consequently, thanks to Proposition~\ref{prop1-21-july-2025}, we deduce the estimate
\be
\aligned
\big|\nabla_{\vec{L}}(s^{3/2}\mathscr{Z}^I\Psi)\big|_{\vec{n}}
& \lesssim  s^{1/2}[\Psi]_{N-3} + (s/t)s^{3/2}[\del\Psi]_{N-4} 
\\
& \quad + (s/t)^{-1}\zetab^{-1}s^{3/2}|\del H|[\Psi]_{N-4} + s^{3/2}\Hcom_{p-1}[\Psi].
\endaligned
\ee
Observe that $t\simeq s^2$ along $\{r=t-1\}$, hence 
\begin{equation}\label{eq7-26-aout-2025}
\big|\nabla_{\vec{L}}(s^{3/2}\mathscr{Z}^I\Psi)\big|_{\vec{n}}\lesssim s^{1/2}[\Psi]_{N-3} 
+ (s/t)^{-1}s^{3/2}|\del H|\,[\Psi]_{N-3}
+ s^{3/2}\Hcom_{p-1}[\Psi].
\end{equation}

At this stage, we discuss in connection with the statement in Proposition~\ref{lem2-24-july-2025}, that is, we must pay attention to the initial point $\gamma_{t,x}(s_{t,x}^*)$. We first consider $(t,x)\in \Mcal^{\frac{3}{5}}_{[s_0,s_1]} = \Mcal^{\Hcal}_{[s_0,s_1]}\cap \{r\geq 4t/5\}$ and, in this case, we apply~\eqref{eq5-19-june-2025} and find 
\bse
\be
(s/t)|_{\gamma_{t,x}(s_{t,x}^*)}\lesssim (s/t).
\ee
Also recalling that, in this case, $\gamma_{t,x}(s_{t,x}^*)\in \{r=t-1\}$ (provided $s_0=2$), we deduce that
\be
\gamma_{t,x}(s_{t,x}^*) = (t^*,x^*) 
\ee
with the property  
\begin{equation}
(s/t)^{-2}\lesssim t^*,\quad (s/t)^{-1}\lesssim s^* = \sqrt{|t^*|^2 - |x^*|^2}. 
\end{equation}
To check this last statement, we only point out the relations
\be
\frac{|t^*|^2 - |x^*|^2}{|t^*|^2} = (s/t)_{\gamma_{t,x}(s_{t,x}^*)}\lesssim(s/t),
\quad
t^*-|x^*| = 1.
\ee
\ese
In turn, from \eqref{eq7-26-aout-2025} we infer that 
\be
\aligned
\big[D_{t,x}[\mathscr{Z}^I\Psi]\big]_{\vec{n}}
\lesssim_p &(s/t)^{-3/2}\big(\lapsb M[\Psi]_{p,k} + (s/t)^{-1}|\del H|[\Psi]_{N-3} + \Hcom_{p-1}[\Psi]\big)_{\tau_{t,x}(s_{t,x}^*)}
\\ 
& \quad + (s/t)^{-1/2}[\Psi]_{N-3}.
\endaligned
\ee


\paragraph{Step 4.} 
\bse
A third group of terms are the error terms contained in $F_{t,x}$, which are
\begin{equation}
E[\Psi]:= s^{3/2}\lapsb^2|\nabla_W\Psi|_{\vec{n}} + s^{3/2}\lapsb^2\big|R[H,\Psi]\big|_{\vec{n}}.
\end{equation}
The latter term $s^{3/2}\lapsb^2\big|R[H,\Psi]\big|_{\vec{n}}$ is bounded by Lemma~\ref{lem2-22-aout-2025}, while for the former term we write 
\be
\big|\nabla_W\big(\mathscr{Z}^I\Psi\big)\big|_{\vec{n}}
\lesssim \big((s/t)^{-1}\lapsb^2|W|\big)\big((s/t)s^{3/2}[\del\Psi]_{p,k}\big) 
+ \zetab^{-1}\lapsb |W||\del H|\big(s^{3/2}\lapsb[\Psi]_{p,k}\big).
\ee
We then have
\begin{equation}
\aligned
E\big[\mathscr{Z}^I\Psi\big]_{p,k}
& \lesssim_p  \big((s/t)^{-1}\lapsb^2|W|\big)\big((s/t)s^{3/2}[\del\Psi]_{p,k}\big) 
+ \zetab^{-1}\lapsb |W||\del H|\big(s^{3/2}\lapsb[\Psi]_{p,k}\big)
\\
& \quad +s^{3/2}\lapsb^2\big(s^{-2} + s^{-1}(s/t)^{-1}|\del H|_1\big)[\Psi]_{N-2}
\\
& \quad +s^{3/2}\lapsb^2(s/t)^{-1}\big(t^{-1} + |\del H|\big)|H|_1[\Psi]_{N-2}\sum_{k=0}^2\big((s/t)^{-2}|H|\big)^k.
\endaligned
\end{equation}

Finally we consider the source term by recalling Proposition~\ref{lem1-22-aout-2025}: 
\be
\aligned
s^{3/2}\lapsb^2\big|\opDirac\Phi - \mathrm{i}M\Phi\big|_{\vec{n}}
\lesssim_N& \sum_{p_1+p_2\leq p\atop k_1+k_2\leq k}
\big(s^{3/2}\lapsb [\Psi]_{p_1,k_1}\big)\big(\lapsb|LH|_{p_2-1,k_2-1}\big) 
+\lapsb T_N[H]\big(s^{3/2}\lapsb[\Psi]_{p,k}\big)
\\
& \quad + \lapsb^2s^{3/2} \big(S_{N,p}[H,\Psi] + \Hcom_{p}[\Psi] + \Hwave_{p}[\Psi]\big).
\endaligned
\ee
Now substituting the above estimates into~\eqref{eq8-26-aout-2025} we find
\begin{equation}\label{eq9-26-aout-2025}
\aligned
&s^{3/2}\big((s/t)\sum_{\alpha}|\del_{\alpha}\mathscr{Z}^I\Psi|_{\vec{n}} + \lapsb M|\mathscr{Z}^I\Psi|_{\vec{n}}\big)(t,x)
\\
& \lesssim_M 
(s/t)^{-3/2}\big(\lapsb M[\Psi]_{N-4} + (s/t)^{-1}|\del H|[\Psi]_{N-3} + \Hcom_{p-1}[\Psi]\big)_{\tau_{t,x}(s_{t,x}^*)} 
\\
& \quad+ (s/t)^{-1/2}\big([\Psi]_{N-3}\big)_{\tau_{t,x}(s_{t,x}^*)}   +s^{1/2}[\Psi]_{N-3}
\\
& \quad + \sum_{p_1+p_2\leq p\atop k_1+k_2\leq k}\int_{s_{t,x}^*}^s 
\Abf_{p_1,k_1}[\Psi](\lambda)\big(\lapsb|LH|_{p_2-1,k_2-1}\big)\big|_{t,x}(\lambda)\diff \lambda
\\
& \quad +\int_{s_{t,x}^*}^s  \Abf_{p,k}[\Psi](\lambda)I_N[H,W]\big|_{t,x}(\lambda) \diff \lambda
+\int_{s{t,x}^*}^s J_N[H,W,\Psi]\big|_{t,x}(\lambda) \diff \lambda
\endaligned
\end{equation}
with 
\be
\aligned
I_N[H,W]  & :=  \lapsb T_{N}[H] + \big((s/t)^{-1}\lapsb + \zetab^{-1} |\del H|\big)\lapsb|W|
\\
& =  s^{-1}|H|_{N-2} + |\del H|_{N-3} + \zetab^{-1}|\del\Hu^{00}|_{N-3}
+ (s/t)^{-2}\hspace{-0.3cm}\sum_{p_1+p_2\leq N-4}\hspace{-0.4cm}|H|_{p_1+1}|H|_{p_2+1}
\\
& \quad + \big((s/t)^{-1}\lapsb + \zetab^{-1} |\del H|\big)\lapsb|W|,
\endaligned
\ee
\be
\aligned
J_N[H,W,\Psi]  & := \lapsb^2s^{3/2} \big(S_{N,N-4}[H,\Psi] + \Hcom_{N-4}[\Psi] + \Hwave_{N-4}[\Psi]\big)
\\
& \quad + s^{3/2}\lapsb^2\big(s^{-2} + s^{-1}(s/t)^{-1}|\del H|_1\big)[\Psi]_{N-2}
\\
& \quad +s^{3/2}\lapsb^2(s/t)^{-1}\big(t^{-1} + |\del H|\big)|H|_1[\Psi]_{p+2,k+2}\sum_{k=0}^2\big((s/t)^{-2}|H|\big)^k
\\
& = s^{3/2}\lapsb^2\big(t^{-1} + |\del H|_{N-4}\big)|H|_{N-3}[\del\Psi]_{N-3} 
\\
& \quad + s^{3/2}\lapsb^2\big(|\del W|_{N-3} + |W|_{N-3} + (s/t)^{-1}|\del H|_{N-4}|H|_{N-3}\big)[\Psi]_{N-3}
\\
& \quad + s^{3/2}\lapsb^2\big(s^{-2} + s^{-1}(s/t)^{-1}|\del H|_1\big)[\Psi]_{N-2}
\\
& \quad +s^{3/2}\lapsb^2(s/t)^{-1}\big(t^{-1} + |\del H|\big)|H|_1[\Psi]_{N-2}\sum_{k=0}^2\big((s/t)^{-2}|H|\big)^k.
\endaligned
\ee
Both $I_N$ and $J_N$ are expected to be uniformly integrable (that is, independently of $(s/t)$). Then we apply Proposition~\ref{prop1-21-july-2025} together with~\eqref{eq9-26-aout-2025}, and arrive at~\eqref{eq11-26-aout-2025}. This completes the proof of Proposition~\ref{prop-spin-orthgonal}. 
\ese
}


\section{Source terms in the hyperboloidal domain (Proposition~\ref{lem1-22-aout-2025})}
\label{section=N12}

\subsection{ An algebraic decomposition}
\label{sec1-20-aout-2025}

{ 

We now provide the proof of Proposition~\ref{lem1-22-aout-2025}, which was stated and used in the previous section. Throughout this section, $\Psi$ denotes a sufficiently regular solution to the massive Dirac equation
\begin{equation}\label{eq9-21-aout-2025}
\opDirac\Psi + \mathrm{i}M\Psi = 0 \qquad \text{ in } \Mcal^{\Hcal}_{[s_0,s_1]}. 
\end{equation}
All of our results in this section hold in the hyperboloidal domain $\Mcal^{\Hcal}_{[s_0,s_1]}$ (unless  otherwise specified).

\begin{lemma}[Two Clifford--Dirac anti-commutator identities]
\label{lem1-19-aout-2025}
For any sufficiently regular vector field $X$ and any sufficiently regular spinor field $\Psi$, one has 
\begin{equation}\label{eq5-22-aout-2025}
\{\opDirac,X\cdot\}\Psi:=\opDirac(X\cdot\Psi) + X\cdot\opDirac\Psi  = - 2\nabla_X\Psi + \opDirac(X)\cdot\Psi
\end{equation}
and, moreover, if $X,Y,Z$ are sufficiently regular vector fields one also has the identity 
\begin{equation}\label{eq4-22-aout-2025}
\aligned
\{\opDirac,X\cdot Y\cdot Z\cdot\}\Psi  & :=  \opDirac(X\cdot Y\cdot Z\cdot \Psi) + X\cdot Y\cdot Z\cdot \opDirac\Psi 
\\
& =  2X\cdot Y\cdot\nabla_Z\Psi + 2X\cdot Z\cdot\nabla_Y\Psi - 2Y\cdot Z\cdot\nabla_X\Psi
\\
& \quad + 2X\cdot\nabla_Y(Z)\cdot \Psi - 2\nabla_XY\cdot Z\cdot\Psi - 2Y\cdot\nabla_XZ\cdot \Psi
\\
& \quad + \opDirac(X)\cdot Y\cdot Z\cdot \Psi  - X\cdot\opDirac(Y)\cdot Z\cdot\Psi +X\cdot Y\cdot\opDirac(Z)\cdot\Psi. 
\endaligned
\end{equation}
\end{lemma}

\begin{proof} A direct computation shows that 
\be
\aligned
\opDirac(X\cdot\Psi)  & =  g^{\alpha\beta}\del_{\alpha}\cdot\nabla_{\beta}(X\cdot\Psi)
= g^{\alpha\beta}\del_{\alpha}\cdot X\cdot\nabla_{\beta}\Psi 
+ g^{\alpha\beta}\del_{\alpha}\cdot\nabla_{\beta}X\cdot\Psi
\\
& =  -g^{\alpha\beta}X\cdot\del_{\alpha}\cdot\nabla_{\beta}\Psi 
- 2g^{\alpha\beta}(X,\del_{\alpha})\nabla_{\beta}\Psi
+ \opDirac(X)\cdot\Psi
\\
& =  -X\cdot\opDirac\Psi - 2\nabla_X\Psi + \opDirac(X)\cdot\Psi.
\endaligned
\ee
To derive~\eqref{eq4-22-aout-2025}, we simply apply the identity~\eqref{eq5-22-aout-2025} three times. 
\end{proof}


Given any admissible vector field $Z$, we consider the expression 
\bel{equa-JJD39}
\opDirac([\widehat{Z},\opDirac]\Psi) - \mathrm{i}M[\widehat{Z},\opDirac]\Psi
\ee
where $\widehat{Z}$ denotes the Clifford-adapted derivative. 
Recalling Lemma~\ref{lem6-25-aout-2025} and recalling our ``underline notation'' for components in the semi-hyperboloidal frame, we rewrite the second term in \eqref{equa-JJD39} as 
\begin{equation}\label{eq4-20-aout-2025}
\aligned
- \mathrm{i}M[\widehat{Z},\opDirac]\Psi 
& =   \frac{1}{2}\underline{\pi[Z]}^{\alpha0}\delu_{\alpha}\cdot\widehat{\del_t}(\mathrm{i}M\Psi) 
+ \frac{1}{2t}\underline{\pi[Z]}^{\alpha b}\delu_{\alpha}\cdot\widehat{L_b}(\mathrm{i}M\Psi)
\\
& \quad + \frac{1}{8}g^{\mu\nu}\underline{\pi[Z]}^{\alpha0}\delu_{\alpha}\cdot\del_{\mu}\cdot\nabla_{\nu}\del_t\cdot(\mathrm{i}M\Psi) 
- \frac{1}{8t}g^{\mu\nu}\underline{\pi[Z]}^{\alpha b}
\delu_{\alpha}\cdot\del_{\mu}\cdot\nabla_{\nu}L_b\cdot(\mathrm{i}M\Psi)
\\
& \quad +\frac{1}{4}\pi[Z]^{\alpha\beta}\nabla_{\alpha}\del_{\beta}\cdot(\mathrm{i}M\Psi)
+\frac{1}{4}[Z,W]\cdot(\mathrm{i}M\Psi).
\\
& =:  S_1 + S_2 + S_3 + S_W. 
\endaligned
\end{equation}
Here, $S_1, S_2$ denote the first two terms (the only quadratic terms), while $S_W$ is determined by the generalized wave coordinate vector $W$, and the remaining cubic terms are gathered in $S_3$.

On the other hand, the first term in \eqref{equa-JJD39} can also be written as 
\be 
\aligned
\opDirac\big([\widehat{Z},\opDirac]\Psi\big) 
& = 
- \frac{1}{2}\opDirac\big(\underline{\pi[Z]}^{\alpha0}\delu_{\alpha}\cdot\widehat{\del_t}\Psi\big)
- \frac{1}{2}\opDirac\big(t^{-1}\underline{\pi[Z]}^{\alpha b}
\delu_{\alpha}\cdot\widehat{L_b}\Psi\big)
\\
& \quad -  \frac{1}{8}\opDirac\big(g^{\mu\nu}\underline{\pi[Z]}^{\alpha0}\delu_{\alpha}\cdot\del_{\mu}\cdot\nabla_{\nu}\del_t\cdot\Psi \big)
- \frac{1}{8}\opDirac\big(t^{-1}g^{\mu\nu}\underline{\pi[Z]}^{\alpha b}
\delu_{\alpha}\cdot\del_{\mu}\cdot\nabla_{\nu}L_b\cdot\Psi\big)
\\
& \quad - \frac{1}{4}\opDirac\big(\pi[Z]^{\alpha\beta}\nabla_{\alpha}\del_{\beta}\cdot\Psi\big)
- \frac{1}{4}\opDirac\big([Z,W]\cdot\Psi\big), 
\endaligned
\ee

hence
\begin{equation}\label{eq6-22-aout-2025}
\aligned
\opDirac\big([\widehat{Z},\opDirac]\Psi\big) 
& =  \frac{1}{2}\underline{\pi[Z]}^{\alpha0}
\delu_{\alpha}\cdot\widehat{\del_t}\big(\opDirac\Psi\big)
+ \frac{1}{2}t^{-1}\underline{\pi[Z]}^{\alpha b}
\delu_{\alpha}\cdot\widehat{L_b}\big(\opDirac\Psi\big)
\\
& \quad + \frac{1}{8}g^{\mu\nu}\underline{\pi[Z]}^{\alpha0}
\delu_{\alpha}\cdot\del_{\mu}\cdot\nabla_{\nu}\del_t\cdot\opDirac\Psi
+ \frac{1}{8}t^{-1}g^{\mu\nu}\underline{\pi[Z]}^{\alpha b}
\delu_{\alpha}\cdot\del_{\mu}\cdot\nabla_{\nu}L_b\cdot\opDirac\Psi
\\
& \quad +\frac{1}{4}\pi[Z]^{\alpha\beta}\nabla_{\alpha}\del_{\beta}\cdot\opDirac\Psi
+\frac{1}{4}[Z,W]\cdot\opDirac\Psi
\\
& \quad - \frac{1}{2}\{\opDirac,\underline{\pi[Z]}^{\alpha0}\delu_{\alpha}\cdot\}\widehat{\del_t}\Psi 
+ \frac{1}{2}\underline{\pi[Z]}^{\alpha0}\delu_{\alpha}\cdot[\opDirac,\widehat{\del_t}]\Psi
\\
& \quad - \frac{1}{2}\{\opDirac,t^{-1}\underline{\pi[Z]}^{\alpha b}\delu_{\alpha}\cdot\}\widehat{L_b}\Psi
+\frac{1}{2t}\underline{\pi[Z]}^{\alpha b}\delu_{\alpha}\cdot[\opDirac,\widehat{L_b}]\Psi
\\
& \quad + \frac{1}{8}\{\opDirac, g^{\mu\nu}\underline{\pi[Z]}^{\alpha0}\delu_{\alpha}\cdot\del_{\mu}\cdot\nabla_{\nu}\del_t\cdot\}\Psi
+ \frac{1}{8}\{\opDirac,t^{-1}g^{\mu\nu}\underline{\pi[Z]}^{\alpha b}\delu_{\alpha}\cdot\del_{\mu}\cdot\nabla_{\nu}L_b\cdot\}\Psi 
\\
& \quad + \frac{1}{4}\{\opDirac, \pi[Z]^{\alpha\beta}\nabla_{\alpha}\del_{\beta}\cdot\}\Psi 
+ \frac{1}{4}\{\opDirac,[Z,W]\cdot\}\Psi.
\endaligned
\end{equation} 
Summing up~\eqref{eq4-20-aout-2025} and~\eqref{eq6-22-aout-2025}, and recalling the Dirac equation~\eqref{eq9-21-aout-2025}, we deduce that 
\begin{equation}
\aligned
& \opDirac([\widehat{Z},\opDirac]\Psi) - \mathrm{i}M[\widehat{Z},\opDirac]\Psi
\\
& = - \frac{1}{2}\{\opDirac,\underline{\pi[Z]}^{\alpha0}\delu_{\alpha}\cdot\}\widehat{\del_t}\Psi 
+ \frac{1}{2}\underline{\pi[Z]}^{\alpha0}\delu_{\alpha}\cdot[\opDirac,\widehat{\del_t}]\Psi
\\
& \quad - \frac{1}{2}\{\opDirac,t^{-1}\underline{\pi[Z]}^{\alpha b}\delu_{\alpha}\cdot\}\widehat{L_b}\Psi
+\frac{1}{2t}\underline{\pi[Z]}^{\alpha b}\del_{\alpha}\cdot[\opDirac,\widehat{L_b}]\Psi
\\
& \quad + \frac{1}{8}\{\opDirac, g^{\mu\nu}\underline{\pi[Z]}^{\alpha0}\delu_{\alpha}\cdot\del_{\mu}\cdot\nabla_{\nu}\del_t\cdot\}\Psi
+ \frac{1}{8}\{\opDirac,t^{-1}g^{\mu\nu}\underline{\pi[Z]}^{\alpha b}\delu_{\alpha}\cdot\del_{\mu}\cdot\nabla_{\nu}L_b\cdot\}\Psi 
\\
& \quad + \frac{1}{4}\{\opDirac, \pi[Z]^{\alpha\beta}\nabla_{\alpha}\del_{\beta}\cdot\}\Psi 
+ \frac{1}{4}\{\opDirac,[Z,W]\cdot\}\Psi
\\
& =:  T_1 + T_2+T_3+T_4+T_5 + T_6 + T_7 + T_W.
\endaligned
\end{equation}
Here, $T_1$ contains quadratic terms, while the terms $T_3$ and $T_6$ contain the favorable decay factor $t^{-1}$. The term $T_W$ is due to the generalized wave coordinate condition, and the remaining terms $T_2, T_4, T_5, T_7$ are cubic in nature. We only need to analyze $T_1$ further, as follows. 


Indeed, we recall the decomposition in Lemma~\ref{lem1-19-aout-2025} and write 
\be
\aligned
T_1  & =  \underline{\pi[Z]}^{\alpha0}\nabla_{\delu_{\alpha}}\big(\widehat{\del_t}\Psi\big)
- \frac{1}{2}\opDirac(\underline{\pi[Z]}^{\alpha 0}\delu_{\alpha})\cdot\widehat{\del_t}\Psi
\\
& =  \underline{\pi[Z]}^{00}\nabla_t\big(\widehat{\del_t}\Psi\big)
+ t^{-1}\underline{\pi[Z]}^{b0}\nabla_{L_b}\big(\widehat{\del_t}\Psi\big)
\\
& \quad - \frac{1}{2}g^{\mu\nu}\del_{\mu}\cdot\del_{\nu}\big(\underline{\pi[Z]}^{\alpha0}\big)\delu_{\alpha}\cdot\widehat{\del_t}\Psi 
- \frac{1}{2}g^{\mu\nu}\underline{\pi[Z]}^{\alpha0}\del_{\mu}\cdot\nabla_{\nu}\delu_{\alpha}\cdot\widehat{\del_t}\Psi, 
\endaligned
\ee
hence
\be
\aligned
T_1   
& = \underline{\pi[Z]}^{00}\widehat{\del_t}\big(\widehat{\del_t}\Psi\big)
+\frac{1}{4}g^{\mu\nu}\underline{\pi[Z]}^{00}\del_{\mu}\cdot\nabla_{\nu}\del_t\cdot\widehat{\del_t}\Psi
\\
& \quad + t^{-1}\underline{\pi[Z]}^{b0}\widehat{L_b}\big(\widehat{\del_t}\Psi\big)
+\frac{1}{4t}g^{\mu\nu}\underline{\pi[Z]}^{b0}\del_{\mu}\cdot\nabla_{\nu}L_b\cdot\widehat{\del_t}\Psi
\\
& \quad - \frac{1}{2}g^{\mu\nu}\del_{\nu}\big(\underline{\pi[Z]}^{00}\big)\del_{\mu}\cdot\del_t\cdot\widehat{\del_t}\Psi 
- \frac{1}{2} g^{\mu\nu}\del_{\nu}\big(\underline{\pi[Z]}^{b0}\big)\del_{\mu}\cdot\delu_b\cdot\widehat{\del_t}\Psi
\\
& \quad - \frac{1}{2}g^{\mu\nu}\underline{\pi[Z]}^{\alpha0}\del_{\mu}\cdot\nabla_{\nu}\delu_{\alpha}\cdot\widehat{\del_t}\Psi
\\
& =:  T_{11} + T_{12}, 
\endaligned
\ee
where we have introduced the following quadratic contribution
\begin{equation}
T_{11} := \underline{\pi[Z]}^{00}\widehat{\del_t}\big(\widehat{\del_t}\Psi\big)
- \frac{1}{2}g^{\mu\nu}\del_{\nu}\big(\underline{\pi[Z]}^{00}\big)\del_{\mu}\cdot\del_t\cdot\widehat{\del_t}\Psi 
- \frac{1}{2} g^{\mu\nu}\del_{\nu}\big(\underline{\pi[Z]}^{b0}\big)\del_{\mu}\cdot\delu_b\cdot\widehat{\del_t}\Psi. 
\end{equation}
The remaining part, denoted by $T_{12}$, consists of cubic terms or terms with favorable decay due to the factor $t^{-1}$. We therefore arrive at the following decomposition.

\begin{proposition}[Quadratic--cubic decomposition of the double Dirac commutator]
\label{prop-quad-cubic}
For any sufficiently regular solution $\Psi$ to the Dirac equation~\eqref{eq9-21-aout-2025} and for any admissible vector field $Z$, one has  
\begin{equation}\label{eq7-22-aout-2025}
\opDirac([\widehat{Z},\opDirac]\Psi) - \mathrm{i}M[\widehat{Z},\opDirac]\Psi 
= Q[Z,\Psi] + R[Z,\Psi], 
\end{equation}
where 
\bel{equa-def-QR} 
\aligned
Q[Z,\Psi] & :=  \underline{\pi[Z]}^{00}\widehat{\del_t}\big(\widehat{\del_t}\Psi\big)
- \frac{1}{2}g^{\mu\nu}\del_{\nu}\big(\underline{\pi[Z]}^{00}\big)\del_{\mu}\cdot\del_t\cdot\widehat{\del_t}\Psi 
- \frac{1}{2} g^{\mu\nu}\del_{\nu}\big(\underline{\pi[Z]}^{b0}\big)\del_{\mu}\cdot\delu_b\cdot\widehat{\del_t}\Psi,
\\
R[Z,\Psi]  & :=  t^{-1}\underline{\pi[Z]}^{b0}\widehat{L_b}\big(\widehat{\del_t}\Psi\big)
+\frac{1}{4}g^{\mu\nu}\underline{\pi[Z]}^{00}\del_{\mu}\cdot\nabla_{\nu}\del_t\cdot\widehat{\del_t}\Psi
+\frac{1}{4t}g^{\mu\nu}\underline{\pi[Z]}^{b0}\del_{\mu}\cdot\nabla_{\nu}L_b\cdot\widehat{\del_t}\Psi
\\
& \quad - \frac{1}{2}g^{\mu\nu}\underline{\pi[Z]}^{\alpha0}\del_{\mu}\cdot\nabla_{\nu}\delu_{\alpha}\cdot\widehat{\del_t}\Psi
+\frac{1}{2}\underline{\pi[Z]}^{\alpha0}\delu_{\alpha}\cdot[\opDirac,\widehat{\del_t}]\Psi
+\frac{1}{2t}\underline{\pi[Z]}^{\alpha b}\delu_{\alpha}\cdot[\opDirac,\widehat{L_b}]\Psi
\\
& \quad - \frac{1}{2}\{\opDirac,t^{-1}\underline{\pi[Z]}^{\alpha b}\delu_{\alpha}\cdot\}\widehat{L_b}\Psi
+ \frac{1}{4}\{\opDirac, \pi[Z]^{\alpha\beta}\nabla_{\alpha}\del_{\beta}\cdot\}\Psi 
+ \frac{1}{4}\{\opDirac,[Z,W]\cdot\}\Psi
\\
& \quad + \frac{1}{8}\{\opDirac, g^{\mu\nu}\underline{\pi[Z]}^{\alpha0}\delu_{\alpha}\cdot\del_{\mu}\cdot\nabla_{\nu}\del_t\cdot\}\Psi
+ \frac{1}{8}\{\opDirac,t^{-1}g^{\mu\nu}\underline{\pi[Z]}^{\alpha b}\delu_{\alpha}\cdot\del_{\mu}\cdot\nabla_{\nu}L_b\cdot\}\Psi.
\endaligned
\ee
\end{proposition}

}


\subsection{ Estimates of quadratic terms}

{

We continue our analysis with the study of $Q[Z,\Psi]$ in \eqref{equa-def-QR}. 

\begin{lemma}\label{lem1-25-aout-2025}
Provided the spacetime metric is $(\eps_s,\delta,\kappa)$--flat, one has for any integer $p\leq N-4$ 
\begin{equation}
\aligned
\big[Q[Z,\Psi]\big]_{p,k} & \lesssim_p \big(|\del H|_{p+1} + \zetab^{-1}|\del\Hu^{00}|_{p+1} + s^{-1}|H|_{p+2}\big)[\del\Psi]_{p,k}
\\
& \quad +\left\{
\aligned
& \big(|\del\Hu^{00}|_p + t^{-1}|H|_p\big) [\Psi]_{p,k},\quad &&Z = \del_{\delta},
\\
& \sum_{p_1+p_2\leq p\atop k_1+k_2\leq k}[\Psi]_{p_1,k_1}|L_a\Hu^{00}|_{p_2,k_2},\quad &&Z = L_a,
\endaligned
\right.
\endaligned
\end{equation}
where one has distinguished between translation and boost vectors $Z$. 
\end{lemma}

\begin{proof}
\bse
We will estimate each of the quadratic terms in in \eqref{equa-def-QR}, and prove that 
\begin{equation}\label{eq6-24-aout-2025}
\big[\,\underline{\pi[Z]}^{00}\widehat{\del_t}\big(\widehat{\del_t}\Psi\big)\big]_{p,k}
\lesssim_p 
\left\{
\aligned
& \big(|\del\Hu^{00}|_p + t^{-1}|H|_{p,k}\big) [\Psi]_p,\quad && Z = \del_{\delta},
\\
& \sum_{p_1+p_2\leq p\atop k_1+k_2\leq k}[\Psi]_{p_1,k_1}|L_a\Hu^{00}|_{p_2,k_2},\quad &&Z = L_a, 
\endaligned
\right.
\end{equation}
together with 
\begin{equation}\label{eq7-24-aout-2025}
\aligned
\big[g^{\mu\nu}\del_{\nu}\big(\underline{\pi[Z]}^{00}\big)\del_{\mu}\cdot\del_t\cdot\widehat{\del_t}\Psi \big]_{p,k}
& \lesssim_p  \big(|\del H|_{p+1} + \zetab^{-1}|\del\Hu^{00}|_{p+1}\ + s^{-1}|H|_{p+2}\big)[\del\Psi]_{p,k},
\endaligned
\end{equation}
and 
\begin{equation}\label{eq8-24-aout-2025}
\big[g^{\mu\nu}\del_{\nu}\big(\underline{\pi[Z]}^{b0}\big)\del_{\mu}\cdot\delu_b\cdot\widehat{\del_t}\Psi\big]_{p,k}
\lesssim_p(|\del H|_{p+1} + s^{-1}|H|_{p+2})[\del\Psi]_{p,k}.
\end{equation}
\ese
\bse
For~\eqref{eq6-24-aout-2025}, we apply the high-order estimates for products in~\eqref{eq1-17-july-2025} together with~\eqref{eq8-aout-2024}. Next, to derive~\eqref{eq7-24-aout-2025}, we decompose the expression as
\be
g^{\mu\nu}\del_{\nu}\big(\underline{\pi[Z]}^{00}\big)\del_{\mu}\cdot\del_t\cdot\widehat{\del_t}\Psi
=
\gu^{\mu\nu}\delu_{\nu}\big(\underline{\pi[Z]}^{00}\big)
\delu_{\mu}\cdot\del_t\cdot\widehat{\del_t}\Psi
=: T_0 + T_1 + T_2, 
\ee
where
\be
\aligned
T_0  & = \gu^{00}\del_t\big(\underline{\pi[Z]}^{00}\big)\del_t\cdot\del_t\cdot\widehat{\del_t}\Psi,
\\
T_1  & = \gu^{d0}\del_t\big(\underline{\pi[Z]}^{00}\big)\delu_d\cdot\del_t\cdot\widehat{\del_t}\Psi,
\\
T_2  & =  t^{-1}\gu^{0c}L_c\big(\underline{\pi[Z]}^{00}\big)\del_t\cdot\del_t\cdot\widehat{\del_t}\Psi
+ t^{-1}\gu^{dc}L_c\big(\underline{\pi[Z]}^{00}\big)\delu_d\cdot\del_t\cdot\widehat{\del_t}\Psi.
\endaligned
\ee
The term $T_0$ and $T_2$ are the easiest and we treat them first. For $T_0$, we observe that the coefficient $\gu^{00}$ enjoys the null condition, so 
\begin{equation}\label{eq10-24-aout-2025}
|\gu^{00}|_{p,k}\lesssim_p (s/t)^2 + |\Hu|_{p,k}\lesssim_p (s/t).
\end{equation}
Moreover, the term $T_2$ contains a decay factor $t^{-1}$ and we thus arrive directly at 
\begin{equation}
[T_0]_{p,k}\lesssim_p (s/t)\big(|\del H|_{p+1} + t^{-1}|H|_p\big)[\del\Psi]_{p,k},
\end{equation}
where we used $\del_t\cdot\del_t = -g_{00}$, while 
\begin{equation}
[T_2]_{p,k}\lesssim_p \zetab^{-1}t^{-1}|H|_{p+2}[\del\Psi]_{p,k}\lesssim_p s^{-1}|H|_{p+2}[\del\Psi]_{p,k}.
\end{equation}
Here, we have applied a weaker version of~\eqref{eq8-23-aout-2025} and~\eqref{eq9-23-aout-2025}, namely 
\begin{equation}\label{eq11-24-aout-2025}
\aligned
\big|L_a\underline{\pi[Z]}^{\alpha\beta}\big|_{p,k} & \lesssim_p  |H|_{p+2},
\\
\big|\del\big(\underline{\pi[Z]}^{\alpha\beta}\big)\big|_{p,k} & \lesssim_p 
|\del H|_{p+1} + t^{-1}|H|_p.
\endaligned
\end{equation}
Finally, we turn our attention to the critical term $T_1$ and, thanks to~\eqref{eq9-24-aout-2025}, we have 
\be
\aligned
\big[\gu^{d0}\del_t\big(\underline{\pi[Z]}^{00}\big)\delu_d\cdot\del_t\cdot\widehat{\del_t}\Psi\big]_{p,k}
& \lesssim_p  \zetab^{-1}\sum_{p_1+p_2\leq p\atop k_1+k_2\leq k}
[\del\Psi]_{p_1,k_1}\big|\gu^{d0}\del_t\big(\underline{\pi[Z]}^{00}\big)\big|_{p_2,k_2}
\\
& \lesssim_p  \zetab^{-1}\big(|\del\Hu^{00}|_{p+1} + t^{-1}|H|_p\big)[\del\Psi]_{p,k}. 
\endaligned
\ee
This completes the derivation of \eqref{eq7-24-aout-2025},. 
\ese
\bse
Finally, the estimate~\eqref{eq8-24-aout-2025} is established similarly, since in view of~\eqref{eq11-24-aout-2025}:
\be
\aligned
g^{\mu\nu}\del_{\nu}\big(\underline{\pi[Z]}^{b0}\big)\del_{\mu}\cdot\delu_b\cdot\widehat{\del_t}\Psi
=\gu^{\mu\nu}\delu_{\nu}\big(\underline{\pi[Z]}^{b0}\big)\delu_{\mu}\cdot\delu_b\cdot\widehat{\del_t}\Psi
=: S_0 + S_1 + S_2, 
\endaligned
\ee
with
\be
\aligned
S_0  & :=  \gu^{00}\del_t\big(\underline{\pi[Z]}^{b0}\big)\del_t\cdot\delu_b\cdot\widehat{\del_t}\Psi,
\\
S_1  & :=  \gu^{c0}\del_t\big(\underline{\pi[Z]}^{b0}\big)\delu_c\cdot\delu_b\cdot\widehat{\del_t}\Psi,
\\
S_2  & :=  t^{-1}\gu^{\mu d}L_d\big(\underline{\pi[Z]}^{b0}\big)\delu_\mu\cdot\delu_b\cdot\widehat{\del_t}\Psi.
\endaligned
\ee
We have now 
\be
\aligned
\, [S_0]_{p,k} & \lesssim_p\big((s/t)|\del H|_{p+1} + t^{-1}|H|_p\big)[\del\Psi]_{p,k},
\\
\, [S_2]_{p,k} & \lesssim_p\zetab^{-1}t^{-1}|H|_{p+1}[\del\Psi]_{p,k}\lesssim_ps^{-1}|H|_{p+2}[\del\Psi]_{p,k}
\endaligned
\ee
However, this time the critical term $S_1$ is easier  to handle, due to the fact that
\begin{equation}
\big[\delu_c\cdot\delu_b\cdot\widehat{\del_t}\Psi\big]_{p,k}\lesssim_p[\del\Psi]_{p,k}.
\end{equation}
This is due to Corollary~\ref{cor1-18-july-2025} and Lemma~\ref{eq2-17-july-2025}. That is the reason why $\underline{\pi[Z]}^{b0}$ does not enjoy~\eqref{prop1-24-aout-2025}, yet one arrive also at a satisfactory estimate. In fact by~\eqref{eq11-24-aout-2025},
\be
[S_1]_{p,k}\lesssim_p (|\del H|_{p+1} + t^{-1}|H|_p)[\del\Psi]_{p,k}. \qedhere
\ee
\ese
\end{proof}

}


\subsection{ Estimates of the remaining terms}

{ 

We now turn our attention to the study of $R[Z,\Psi]$ in \eqref{equa-def-QR}. The terms can be classified into three groups, namely 
\bel{equa-Rone}
\aligned
R_1[Z,\Psi] :& =  t^{-1}\underline{\pi[Z]}^{b0}\widehat{L_b}\big(\widehat{\del_t}\Psi\big)
+\frac{1}{4}g^{\mu\nu}\underline{\pi[Z]}^{00}\del_{\mu}\cdot\nabla_{\nu}\del_t\cdot\widehat{\del_t}\Psi
+\frac{1}{4t}g^{\mu\nu}\underline{\pi[Z]}^{b0}\del_{\mu}\cdot\nabla_{\nu}L_b\cdot\widehat{\del_t}\Psi
\\
& \quad - \frac{1}{2}g^{\mu\nu}\underline{\pi[Z]}^{\alpha0}\del_{\mu}\cdot\nabla_{\nu}\delu_{\alpha}\cdot\widehat{\del_t}\Psi
+\frac{1}{2}\underline{\pi[Z]}^{\alpha0}\delu_{\alpha}\cdot[\opDirac,\widehat{\del_t}]\Psi
+\frac{1}{2t}\underline{\pi[Z]}^{\alpha b}\delu_{\alpha}\cdot[\opDirac,\widehat{L_b}]\Psi,
\endaligned
\ee
\bel{equa-Rtwo}
R_2[Z,\Psi]:=- \frac{1}{2}\{\opDirac,t^{-1}\underline{\pi[Z]}^{\alpha b}\delu_{\alpha}\cdot\}\widehat{L_b}\Psi
+ \frac{1}{4}\{\opDirac, \pi[Z]^{\alpha\beta}\nabla_{\alpha}\del_{\beta}\cdot\}\Psi 
+ \frac{1}{4}\{\opDirac,[Z,W]\cdot\}\Psi,
\ee
and
\bel{equa-Rthree}
R_3[Z,\Psi] := \frac{1}{8}\{\opDirac, g^{\mu\nu}\underline{\pi[Z]}^{\alpha0}\delu_{\alpha}\cdot\del_{\mu}\cdot\nabla_{\nu}\del_t\cdot\}\Psi
+ \frac{1}{8}\{\opDirac,t^{-1}g^{\mu\nu}\underline{\pi[Z]}^{\alpha b}\delu_{\alpha}\cdot\del_{\mu}\cdot\nabla_{\nu}L_b\cdot\}\Psi.
\ee
We first estimate the remainder $R_1[\Psi]$.

\begin{lemma}\label{lem2-25-aout-2025}
Assume that the spacetime metric is $(\eps_s,\delta,\kappa)$--flat. Then for any integer $p\leq N-4$ one has 
\begin{equation}\label{eq3-25-aout-2025}
\aligned
\big[R_1[Z,\Psi]\big]_{p,k}
& \lesssim_N   \big(t^{-1}|H|_{N-3} + |H|_{N-3}|\del H|_{N-4} \big)[\del\Psi]_{p+1}
\\
& \quad + \big(|H|_{N-3}|\del H|_{N-4} + |W|_{N-3}\big)|H|_{N-3}[\Psi]_{p+1}.
\endaligned
\end{equation}
\end{lemma}

\begin{proof} We rely on the estimate~\eqref{eq8-23-aout-2025}. For the last two terms in \eqref{equa-Rone}, we apply the following estimate due to~\eqref{eq5-21-aout-2025}:
\begin{equation}
\aligned
& \big|\del_{\alpha}\cdot[\opDirac,\widehat{L_a}]\Psi\big|_{p,k}
\lesssim_p \zetab^{-1}|H|_{p+1}[\del\Psi]_{p,k}
+\zetab^{-1}\big(|H|_{p+1}|\del H|_p + |W|_{p+1}\big) [\Psi]_{p,k},
\\
& \big|\del_{\alpha}\cdot[\opDirac,\widehat{\del_{\delta}}]\Psi\big|_{p,k}
\lesssim_p \zetab^{-1}\big(|\del H|_p+t^{-1}|H|_p\big)[\del\Psi]_{p,k}
+\zetab^{-1}\big(|H|_{p+1}|\del H|_p + |W|_{p+1}\big) [\Psi]_{p,k},
\endaligned
\end{equation}
All of these terms ---but the first one--- can be bounded by $\zetab^{-1}[\Psi]_{p,k}$ or $\zetab^{-1}[\del\Psi]_{p,k}$ multiplied by $|H|_{N-3},|\del H|_{N-4}, |W|_{N-3}$ or $t^{-1}$. For such terms we apply Proposition~\ref{prop1-24-july-2025} and compensate the factor $\zetab^{-1}$ at the cost of replacing $[\del\Psi]_{p,k}$ by the next order term $[\del\Psi]_{p+1,k+1}$ and its control by $[\del\Psi]_{p+1}$. The same argument applies to $[\Psi]_{p,k}$.
\end{proof}


The terms in $R_2,R_3$ in \eqref{equa-Rtwo} and \eqref{equa-Rthree} are anti-commutators, and we can apply the Clifford--Dirac anti-commutator identity in Lemma~\ref{lem1-19-aout-2025}. Before going into the derivation of the estimates, we need several observations. For each of two terms in $R_3[Z,\Psi]$,  we need to bound the corresponding nine terms in the right-hand side of~\eqref{eq4-22-aout-2025}. To this end, we remark that
\be
\aligned
& \{\opDirac,t^{-1}g^{\mu\nu}\underline{\pi[Z]}^{\alpha b}\delu_{\alpha}\cdot\del_{\mu}\cdot\nabla_{\nu}L_b\cdot\}\Psi
\\
& = \{\opDirac,t^{-1}g^{\mu\nu}\underline{\pi[Z]}^{\alpha b}\delu_{\alpha}\cdot\del_{\mu}\cdot\nabla_{\nu}(x^b\del_t + t\del_b)\cdot\}\Psi
\\
& = \delta_{\nu}^b\{\opDirac,t^{-1}g^{\mu\nu}\underline{\pi[Z]}^{\alpha b}\delu_{\alpha}\cdot\del_{\mu}\cdot\del_t\cdot\}\Psi
+\{\opDirac,(x^b/t)g^{\mu\nu}\underline{\pi[Z]}^{\alpha b}\delu_{\alpha}\cdot\del_{\mu}\cdot\nabla_{\nu}\del_t\cdot\}\Psi
\\
& \quad +\delta_{\nu}^0\{\opDirac,t^{-1}g^{\mu\nu}\underline{\pi[Z]}^{\alpha b}\delu_{\alpha}\cdot\del_{\mu}\cdot\del_b\cdot\}\Psi
+\{\opDirac,g^{\mu\nu}\underline{\pi[Z]}^{\alpha b}\delu_{\alpha}\cdot\del_{\mu}\cdot\nabla_{\nu}\del_b\cdot\}\Psi
\\
& = \delta_{\nu}^b\{\opDirac,t^{-1}g^{\mu\nu}\underline{\pi[Z]}^{\alpha b}\delu_{\alpha}\cdot\del_{\mu}\cdot\del_t\cdot\}\Psi
+\delta_{\nu}^0\{\opDirac,t^{-1}g^{\mu\nu}\underline{\pi[Z]}^{\alpha b}\delu_{\alpha}\cdot\del_{\mu}\cdot\del_b\cdot\}\Psi
\\
& \quad +\{\opDirac,(x^b/t)g^{\mu\nu}\Gamma_{\nu0}^{\gamma}\underline{\pi[Z]}^{\alpha b}\delu_{\alpha}\cdot\del_{\mu}\cdot\del_{\gamma}\cdot\}\Psi
+\{\opDirac,g^{\mu\nu}\Gamma_{\nu b}^{\gamma}\underline{\pi[Z]}^{\alpha b}\delu_{\alpha}\cdot\del_{\mu}\cdot\del_{\gamma}\cdot\}\Psi
\endaligned
\ee
and 
\be
\{\opDirac, g^{\mu\nu}\underline{\pi[Z]}^{\alpha0}\delu_{\alpha}\cdot\del_{\mu}\cdot\nabla_{\nu}\del_t\cdot\}\Psi
=\{\opDirac, g^{\mu\nu}\Gamma_{\nu 0}^{\gamma}\underline{\pi[Z]}^{\alpha0}\delu_{\alpha}\cdot\del_{\mu}\cdot\del_{\gamma}\cdot\}\Psi.
\ee
In view of the three-vector notation in~\eqref{eq4-22-aout-2025}, we thus need to consider 
\be
\aligned
&X =
\begin{cases}
t^{-1}g^{\mu\nu}\underline{\pi[Z]}^{\alpha\beta}\delu_{\alpha},\quad & \text{Case A},
\\
ug^{\mu\nu}\Gamma_{\nu\delta}^{\gamma}\underline{\pi[Z]}^{\alpha\beta}\delu_{\alpha},\quad & \text{Case B}.
\end{cases}
\\
&Y = \del_{\delta},
\quad
Z = \del_{\gamma}.
\endaligned
\ee
Here, $ u = 1, (x^b/t)$ are homogeneous coefficients of degree zero. Thus we need to bound the terms in the right-hand side of~\eqref{eq4-22-aout-2025}. Among these terms the estimates on the following two quantities are the most involved: $\nabla_X\Psi$ and $\opDirac(X)\cdot\Psi$. (The notions of homogeneity were introduced in Section~\ref{section===83}.)

For terms in $R_2$, the structure is similar and we also need to bound
$\nabla_X\Psi$ and $\opDirac(X)\cdot\Psi$, now 
with 
\be
X = 
\begin{cases}
t^{-1}\underline{\pi[Z]}^{\alpha b}\delu_{\alpha},\quad & \text{Case C},
\\
\pi[Z]^{\alpha\beta}\Gamma_{\alpha\beta}^{\gamma}\del_{\gamma},\quad & \text{Case D}.
\end{cases}
\ee
We write $X = X^{\alpha}\del_{\alpha}$. Then recall that the elements of the transition matrices $\Psiu_{\mu}^{\alpha}$ are interior-homogeneous of degree zero. Thus $X^{\alpha}$ are finite linear combinations of the following quantities with interior-homogeneous coefficients:
\begin{equation}\label{eq2-26-aout-2025}
X^{\alpha}\simeq 
\begin{cases}
t^{-1}g^{\mu\nu}\pi[Z]^{\alpha\beta},\quad & \text{Case A},
\\
g^{\mu\nu}\Gamma_{\nu\delta}^{\gamma}\pi[Z]^{\alpha\beta},\quad & \text{Case B},
\\
t^{-1}\pi[Z]^{\alpha\beta},\quad & \text{Case C},
\\
\pi[Z]^{\alpha\beta}\Gamma_{\alpha\beta}^{\gamma},\quad & \text{Case D}.
\end{cases}
\end{equation}
This leads us, for any integer $p\leq N-4$, to the bound
\begin{equation}\label{eq2-25-aout-2025}
|X^{\alpha}|_{p,k}\lesssim_p 
\begin{cases}
t^{-1}|H|_{p+1} ,\quad & \text{Case A,C},
\\
|\del H|_p|H|_{p+1},\quad & \text{Case B,D}, 
\end{cases}
\end{equation}
where we used the condition~\eqref{eq-UK-condition}  and~\eqref{eq1-23-aout-2025}.

Then we establish the following intermediate estimates about $\nabla_X\Psi$.

\begin{lemma}\label{lem1-26-aout-2025}
Assume that the spacetime metric is  $(\eps_s,\delta,\kappa)$--flat. Then for any $p\leq N-4$ one has 
\begin{equation}
\aligned
\big[\nabla_X\Psi\big]_{p,k}& \lesssim_N  \big(t^{-1} + |\del H|_{N-4}\big)|H|_{N-3}[\del\Psi]_{p}
\\
& \quad +\big(t^{-1} + |\del H|_{N-4}\big)|H|_{N-3}|\del H|_{N-4}[\Psi]_{p+1}.
\endaligned
\end{equation}
\begin{equation}
\aligned
\big[\del_{\mu}\cdot\del_{\nu}\cdot\nabla_X\Psi\big]_{p,k}
& \lesssim_N  \big(t^{-1} + |\del H|_{N-4}\big)|H|_{N-3}[\del\Psi]_{p+1}
\\
& \quad +\big(t^{-1} + |\del H|_{N-4}\big)|H|_{N-3}|\del H|_{N-4}[\Psi]_{p+1},
\endaligned
\end{equation}
\end{lemma}

\begin{proof}
We recall that $\nabla_X\Psi = X^{\alpha}\widehat{\del_{\alpha}}\Psi + \frac{1}{4}g^{\mu\nu}X^{\alpha}\del_{\mu}\cdot\nabla_{\nu}\del_{\alpha}\cdot\Psi$, which together with~\eqref{eq2-25-aout-2025} implies 
\begin{equation}
\big[\nabla_X\Psi\big]_{p,k}\lesssim_p 
\begin{cases}
t^{-1}|H|_{p+1}[\del\Psi]_{p,k} + t^{-1}\zetab^{-1}|\del H|_p|H|_{p+1}[\Psi]_{p,k},\quad & \text{Case A,C},
\\
|\del H|_p|H|_{p+1}[\del\Psi]_{p,k} + \zetab^{-1}|\del H|_p^2|H|_{p+1}[\Psi]_{p,k},\quad & \text{Case B,D}.
\end{cases}
\end{equation}
Furthermore, thanks to~\eqref{eq3-23-july-2025} and~\eqref{eq-UK-condition} together with~\eqref{eq2-25-aout-2025} (similar to the argument for~\eqref{eq1-26-aout-2025}),
$$
\big[\del_{\mu}\cdot\del_{\nu}\cdot\nabla_X\Psi\big]_{p,k}
\lesssim_p
\begin{cases}
\zetab^{-1}t^{-1}|H|_{p+1}[\del\Psi]_{p,k} + t^{-1}\zetab^{-1}|\del H|_p|H|_{p+1}[\Psi]_{p,k},\quad & \text{Case A,C},
\\
\zetab^{-1}|\del H|_p|H|_{p+1}[\del\Psi]_{p,k} + \zetab^{-1}|\del H|_p^2|H|_{p+1}[\Psi]_{p,k},\quad & \text{Case B,D}.
\end{cases}
$$
Then we apply the pointwise inequality in Proposition~\ref{prop1-24-july-2025} to compensate for the factor $\zetab^{-1}$.
\end{proof}

Next, we proceed with the analysis of $\opDirac(X)\cdot\Psi$.

\begin{lemma}\label{lem2-26-aout-2025}
Assume that the spacetime metric is  $(\eps_s,\delta,\kappa)$--flat. Then for $p\leq N-4$ one has 
\begin{equation}
\aligned
\big[\opDirac(X)\cdot\Psi\big]_{p,k}
\lesssim_N
&(|H|_{N-3}+t^{-1})|\del H|_{N-3}[\Psi]_{p+1} + t^{-2}|H|_{N-3}[\Psi]_{p+1}
\endaligned
\end{equation}
and 
\begin{equation}
\big[\del_{\mu}\cdot\del_{\nu}\cdot\opDirac(X)\cdot\Psi\big]_{p,k}\lesssim_N (|H|_{N-3}+t^{-1})|\del H|_{N-3}[\Psi]_{p+1} + t^{-2}|H|_{N-3}[\Psi]_{p+1}.
\end{equation}
\end{lemma}

\begin{proof}
We obsere that
\be
\aligned
\opDirac(X)  & =  g^{\mu\nu}X^{\alpha}\del_{\mu}\cdot\nabla_{\nu}\del_{\alpha}\cdot\Psi 
+ g^{\mu\nu}\del_{\nu}(X^{\alpha})\del_{\mu}\cdot\del_{\alpha}\cdot\Psi
\\
& = g^{\mu\nu}X^{\alpha}\Gamma_{\nu\alpha}^{\gamma}\del_{\mu}\cdot\del_{\gamma}\cdot\Psi
+g^{\mu\nu}\del_{\nu}(X^{\alpha})\del_{\mu}\cdot\del_{\alpha}\cdot\Psi
=: T_1+T_2.
\endaligned
\ee
The term $T_1$ is bounded by~\eqref{eq2-25-aout-2025}.
For the term $T_2$, thanks to~\eqref{eq2-26-aout-2025} we find 
\begin{equation}
\big|\del_{\nu}(X^{\alpha})\big|_{p,k}
\lesssim
\begin{cases}
t^{-1}|\del H|_{p+1} + t^{-2}|H|_{p+1},\quad & \text{Case A,C},
\\
|\del H|_{p+1}|\del H|_p + |H|_{p+1}|\del\del H|_p,\quad & \text{Case B,D},
\end{cases}
\end{equation}
which gives us the desired bound.
\end{proof}

We are finally in a position to control $R_2[Z,\Psi]$ defined in \eqref{equa-Rtwo}. 

\begin{lemma}\label{lem3-25-aout-2025}
Assume that the spacetime metric is  $(\eps_s,\delta,\kappa)$--flat. Then for $p\leq N-4$, one has 
\begin{equation}
\aligned
\big[R_2[Z,\Psi]\big]_{p,k}& \lesssim_N 
(t^{-1}+|\del H|_{N-4})|H|_{N-3}[\del\Psi]_{p+1} 
+ t^{-1}(|\del H|_{N-3} + t^{-1}|H|_{N-3})[\Psi]_{p+1}
\\
& \quad + |W|_{N-3}[\del\Psi]_{p+1} + \big(|\del W|_{N-3} + |\del H|_{N-4}|W|_{N-3}\big)[\Psi]_{p+1}.
\endaligned
\end{equation}
\end{lemma}

\begin{proof}
\bse
The first two terms are bounded by Lemmas~\ref{lem1-26-aout-2025} and \ref{lem2-26-aout-2025}. For the third term which concerns the generalized wave gauge condition, we remark that
\be
[Z,W] = Z(W^{\alpha})\del_{\alpha} - W^{\beta}\del_{\beta}(Z^{\alpha})\del_{\alpha}.
\ee
For admissible vector fields, $\del_{\beta}Z^{\alpha}$ are constants. Thus
\begin{equation}\label{eq3-26-aout-2025}
|[Z,W]|_{p,k}\lesssim_p|W|_{p+1}.
\end{equation}
On the other hand, thanks to~\eqref{eq5-22-aout-2025},
\be
\{\opDirac,[Z,W]\cdot\}\Psi = -2\nabla_{[Z,W]}\Psi + \opDirac([Z,W])\cdot\Psi =: T_1+T_2.
\ee

To control $T_1$, we write 
\be
\aligned
\nabla_{[Z,W]}\Psi  & =  Z(W^{\alpha})\nabla_{\alpha}\Psi 
+ W^{\beta}\del_{\beta}(Z^{\alpha})\nabla_{\alpha}\Psi
\\
& = Z(W^{\alpha})\widehat{\del_{\alpha}}\Psi 
+ \frac{1}{4}g^{\mu\nu}Z(W^{\alpha})\del_{\mu}\cdot\nabla_{\nu}\del_{\alpha}\cdot\Psi 
\\
& \quad +W^{\beta}\del_{\beta}(Z^{\alpha})\widehat{\del_{\alpha}}\Psi
+ \frac{1}{4}g^{\mu\nu}W^{\beta}\del_{\beta}(Z^{\alpha})\del_{\mu}\cdot\nabla_{\nu}\del_{\alpha}\cdot\Psi 
\endaligned
\ee
Then thanks to~\eqref{eq3-26-aout-2025} and similar to the previous estimates, we obtain 
\be
\aligned
\, [T_1]_{p,k}& \lesssim_N  |W|_{p+1}[\del\Psi]_{p} + \zetab^{-1}|W|_{p+1}|\del H|_p[\Psi]_{p}
\\
& \lesssim_N  |W|_{N-3}[\del\Psi]_{p+1} + |W|_{N-3}|\del H|_p[\Psi]_{p+1}.
\endaligned
\ee
For the term $T_2$ we observe that
\begin{equation}
|\del_{\nu}([Z,W]^{\alpha})|_{p,k}\lesssim |\del W|_{p+1},
\end{equation}
therefore 
\be
[T_2]_{p,k}\lesssim_N \zetab^{-1}\big(|W|_{p+1}|\del H|_p + |\del W|_{p+1}\big)[\Psi]_{p}
\lesssim_N\big(|W|_{N-3}|\del H|_{N-4} + |\del W|_{N-3}\big)[\Psi]_{p+1}.
\ee
\ese
\end{proof}


The last issue is the control of $R_3[Z,W]$. The most involved terms are bounded by Lemmas~\ref{lem1-26-aout-2025} and~\ref{lem2-26-aout-2025}. We omit the details concerning the remaining terms and give directly the relevant estimate: 
\begin{equation}\label{eq4-26-aout-2025}
\aligned
\big[R_3[Z,\Psi]\big]_{p,k}& \lesssim_N  (t^{-1}+|\del H|_{N-4})|H|_{N-3}[\del\Psi]_{p+1} 
\\
& \quad + \big(t^{-2}|H|_{N-3}^2 + (t^{-1}|H|_{N-3} + |\del H|_{N-4})|H|_{N-3}|\del H|_{N-4}\big)[\Psi]_{p+1}.
\endaligned
\end{equation}
For the convenience of the reader (who wants to check each term one by one), we also provide here the required bounds: 
\begin{equation}
\aligned
& \, [\nabla_{\mu}\Psi]_{p,k}\lesssim_N [\del \Psi]_p + \zetab^{-1}|\del H|_{N-4}[\Psi]_{p},
\\
&|\opDirac(\del_{\mu})|_{p,k}\lesssim_N |\del H|_{N-4},
\\
&|\nabla_{\mu}\del_{\nu}|_{p,k}\lesssim|\del H|_{N-4}.
\endaligned
\end{equation}

}


\subsection{ Closing the proof of Proposition~\ref{lem1-22-aout-2025}}

{ 

We first present the following conclusion.

\begin{lemma}\label{lem5-25-aout-2025}
Assume that the spacetime metric is  $(\eps_s,\delta,\kappa)$--flat. Then for any sufficiently regular solution $\Psi$ to the massive Dirac equation~\eqref{eq9-21-aout-2025} and for any integer $p\leq N-4$, one has 
\begin{equation}
\aligned
\big|\opDirac([Z,\opDirac]\Psi) - \mathrm{i}M[Z,\opDirac] \Psi\big|_{p,k}
& \lesssim_N T_{01}[H]\,[\del\Psi]_{p,k} + T_{00}[H]\,[\Psi]_{p,k} + S_{0p}[H,\Psi]
\\
& \quad +\left\{
\aligned
&0,\quad &&Z = \del_{\delta},
\\
& \sum_{p_1+p_2\leq p\atop k_1+k_2\leq k}[\Psi]_{p_1,k_1}|L_a\Hu^{00}|_{p_2,k_2},\quad &&Z = L_a,
\endaligned
\right.
\endaligned
\end{equation}
where
\be
\aligned
T_{01}[H]  & :=   s^{-1}|H|_{N-2} + |\del H|_{N-3} + \zetab^{-1}|\del\Hu^{00}|_{N-3},
\\
T_{00}[H]  & :=  |\del\Hu^{00}|_{N-4} + t^{-1}|H|_{N-4},
\endaligned
\ee
\be
S_{0p}[\Psi] := \big(t^{-1} + |\del H|_{N-4}\big)|H|_{N-3}[\del\Psi]_{p+1} 
+ \big(|\del W|_{N-3} + |W|_{N-3} + |\del H|_{N-4}|H|_{N-3}\big)[\Psi]_{p+1}.
\ee
\end{lemma}

Finally, we are ready to complete this section. 

\begin{proof}[Proof of Proposition~\ref{lem1-22-aout-2025}]
We proceed by induction on $|I|$. The case $|I| = 1$ is guaranteed by Lemma~\ref{lem5-25-aout-2025} by choosing $p=k=0$. We then consider the expression 
\bel{equa-T1T2T3}
\aligned
\opDirac\big([\mathscr{Z}^{I}Z,\opDirac]\Psi\big) - \mathrm{i}M[\mathscr{Z}^IZ,\opDirac]\Psi
& = 
\mathscr{Z}^I\big(\opDirac([Z,\opDirac]\Psi) - \mathrm{i}M[Z,\opDirac]\Psi\big)
+[\opDirac,\mathscr{Z}^I]([Z,\opDirac]\Psi)
\\
& \quad +\opDirac\big([\mathscr{Z}^I](Z\Psi)\big) - \mathrm{i}M[\mathscr{Z}^I,\opDirac](Z\Psi)
\\
& =:  T_1 + T_2 + T_3.
\endaligned
\end{equation}
We define 
\be
p' := \ord(\mathscr{Z}^IZ) = p+1,
\quad
k' := \rank(\mathscr{Z}^IZ) 
= \begin{cases}
k,\quad &Z = \del_{\delta},
\\
k+1,\quad &Z = L_a.
\end{cases}
\ee
To deal with $T_2$ in \eqref{equa-T1T2T3}, we observe first that it is a cubic term. We need the following weaker version of Proposition~\ref{prop1-19-july-2025}:
\begin{equation}
\zetab\big[[\mathscr{Z}^I,\opDirac]\Psi\big]_{p,k}
\lesssim_p \sum_{p_1+p_2\leq p\atop k_1+k_2\leq k}|\Psi|_{p_1,k_1-1}|H|_{p_2+1,k_2+1} 
+ [\Psi]_{p-1,k}|\del H|_p + \Hcom_{p-1}[\Psi] + \Hwave_{p-1}[\Psi],
\end{equation}
where $I$ is of type $(p,k)$, and a rough estimate on $[\widehat{Z},\opDirac]\Psi$ (which is a weaker version of Proposition~\ref{prop1-23-july-2025}): 
\begin{equation}
(s/t)\big[[\widehat{Z},\opDirac]\Psi\big]_{p,k}
\lesssim_p[\del\Psi]_{p,k}|H|_{p+1} + \Hcom_p[\Psi].
\end{equation} 
Then we find 
\begin{equation}
\aligned
\, [T_2]_{\vec{n}}& \lesssim_N  (s/t)^{-2}\hspace{-0.2cm}\sum_{p_1+p_2\leq N-4}\hspace{-0.4cm}
|H|_{p_1+1}|H|_{p_2+1} [\del\Psi]_{p,k} 
\\
& \quad + (s/t)^{-1}|H|_{N-3}|\del H|_{N-4}[\Psi]_{p+1} 
+ \zetab^{-1}\big(\Hcom_p[\Psi] + \Hwave_p[\Psi]\big).
\endaligned
\end{equation}
Here, for the second term we applied Proposition~\ref{prop1-24-july-2025}:
\be
[\del\Psi]_{p-1}\leq [\Psi]_p\lesssim \zetab[\Psi]_{p+1}.
\ee
Consequently, $[T_2]_{\vec{n}}$ fits the induction assumption.

We then apply Lemma~\ref{lem5-25-aout-2025} on $T_1$, and the assumption of induction on $T_3$: 
\\
$\bullet$ When $Z = \del_\delta$. In this case, $p' = p+1, k' = k$ and, by Lemma~\ref{lem5-25-aout-2025},
\be
\, [T_1]_{\vec{n}}\lesssim_N T_{01}[H][\del\Psi]_{p,k} + T_{00}[H]\,[\Psi]_{p,k} + S_{0p}[\Psi] 
\lesssim_N T_N[H][\Psi]_{p',k'} + S_{N,p'}[\Psi],
\ee
which fits the induction. We have 
\be
\aligned
\, [T_3]_{\vec{n}}& \lesssim_N  
\sum_{p_1+p_2\leq p\atop k_1+k_2\leq k}[\widehat{\del_{\delta}}\Psi]_{p_1,k_1}|LH|_{p_2-1,k_2-1} 
+T_N[H]\,[\widehat{\del_{\delta}}\Psi]_{p,k}
\\
& \quad + S_{N,p}[\widehat{\del_{\delta}}\Psi] + \Hcom_p[\widehat{\del_{\delta}}\Psi] + \Hwave_p[\widehat{\del_{\delta}}\Psi]
\\
& \lesssim_N  
\sum_{p_1+p_2\leq p\atop k_1+k_2\leq k}[\Psi]_{p_1+1,k_1}|LH|_{p_2-1,k_2-1} 
+T_N[H]\,[\Psi]_{p+1,k}
\\
& \quad + S_{N,p+1}[\Psi] + \Hcom_{p+1}[\Psi] + \Hwave_{p+1}[\Psi]
\\
& = \sum_{p_1'+p_2\leq p+1\atop k_1+k_2\leq k}[\Psi]_{p_1',k_1}|LH|_{p_2-1,k_2-1} 
+T_N[H]\,[\Psi]_{p',k'}
\\
& \quad + S_{N,p'}[\Psi] + \Hcom_{p'}[\Psi] + \Hwave_{p'}[\Psi], 
\endaligned
\ee
which fits the induction.

\noindent$\bullet$ When $Z = L_a$. In this case $p'=p+1,k'=k+1$, and 
\be
\aligned
\, [T_1]_{\vec{n}} & \lesssim_N  T_{01}[H]\,[\del\Psi]_{p,k} + T_{00}[H]\,[\Psi]_{p,k} + S_{0p}[\Psi]
+\sum_{p_1+p_2\leq p\atop k_1+k_2\leq k}[\Psi]_{p_1,k_1}|LH|_{p_2,k_2}
\\
& \lesssim_N  T_N[H]\,[\Psi]_{p,k} + S_{N,p}[\Psi]
+\sum_{p_1+p_2\leq p\atop k_1+k_2\leq k}[\Psi]_{p_1,k_1}|LH|_{p_2,k_2}
\\
& \lesssim_N T_N[H]\,[\Psi]_{p,k} + S_{N,p}[\Psi]
+\sum_{p_1+p_2'\leq p'\atop k_1+k_2'\leq k'}[\Psi]_{p_1,k_1}|LH|_{p_2'-1,k_2'-1},
\endaligned
\ee
which fits the induction. We find 
\be
\aligned
\, [T_3]_{\vec{n}}& \lesssim_N \sum_{p_1+p_2\leq p\atop k_1+k_2\leq k}[\widehat{L_a}\Psi]_{p_1,k_1}|LH|_{p_2-1,k_2-1} 
+T_N[H]\,[\widehat{L_a}\Psi]_{p,k}
\\
& \quad + S_{N,p}[\widehat{L_a}\Psi] + \Hcom_p[\widehat{L_a}\Psi] + \Hwave_p[\widehat{L_a}\Psi]
\\
& \lesssim_N \sum_{p_1+p_2\leq p\atop k_1+k_2\leq k}[\Psi]_{p_1+1,k_1+1}|LH|_{p_2-1,k_2-1} 
+T_N[H]\,[\Psi]_{p+1,k+1}
\\
& \quad + S_{N,p+1}[\Psi] + \Hcom_{p+1}[\Psi] + \Hwave_{p+1}[\Psi]
\\
& = \sum_{p_1'+p_2\leq p'\atop k_1'+k_2\leq k'}[\Psi]_{p_1',k_1'}|LH|_{p_2-1,k_2-1} 
+T_N[H]\,[\Psi]_{p',k'}
\\
& \quad + S_{N,p'}[\Psi] + \Hcom_{p'}[\Psi] + \Hwave_{p'}[\Psi], 
\endaligned
\ee
which fits the induction. This completes the proof of Proposition~\ref{lem1-22-aout-2025}. 
\end{proof}
}


\section{Einstein equations in a light-bending and wave-coordinate chart}
\label{section=N13}

\subsection{ Einstein equations in a general coordinate chart}
\label{subsec1-20-sept-2025}

{ 

\paragraph{Decomposition of the Ricci curvature.}

We essentially follow the notation in our previous work~\cite{PLF-YM-PDE}. We consider a spacetime $(\Mcal,g)$, in which $\Mcal$ is diffeomorphic to $\RR^4$ and the metric $g$ is Lorentzian. We assume\footnote{As we pointed out earlier, it is convenient to put a tilde on these coordinates, which will only be used temporarily.} that there is a global coordinate chart $\{\xsans^{\alpha}\}$, and denote by $\Mks$ the Minkowski metric with respect to $\{\xsans^{\alpha}\}$. Then we have 
\be 
\widetilde{\Mks}_{\alpha\beta} = \etasans_{\alpha\beta} 
= \begin{cases}
-1,\quad & \alpha=\beta =0,
\\
1,\quad & \alpha=\beta=1,2,3,
\\
0,\quad & \text{otherwise},
\end{cases}
\ee
and we set 
\be
\gsans_{\alpha\beta} := g(\delsans_{\alpha},\delsans_{\beta}),
\qquad
\Hsans_{\alpha\beta} := \gsans_{\alpha\beta}- \etasans_{\alpha\beta}.
\ee
We denote by $\Gammasans_{\alpha\beta}^{\gamma}$ the Christoffel symbols associated with the metric $g$ and the coordinates $\{\xsans^{\alpha}\}$, and we also set 
$\Gammasans^{\lambda}:= \gsans^{\alpha\beta}\Gammasans_{\alpha\beta}^{\lambda}$. We recall that the Ricci curvature $R[g] = \widetilde{R}[g]_{\alpha\beta}\diff\xsans^{\alpha}\otimes\diff\xsans^{\beta}$ associated with the metric $g$ reads 
\begin{subequations}\label{eqs1-exp}
\begin{equation} \label{eq1-13-sept-2025}
2\widetilde{R}[g]_{\alpha\beta} = - \gsans^{\mu\nu} \delsans_{\mu} \delsans_{\nu} \gsans_{\alpha\beta} 
+ \Fbb(\gsans,\delsans\gsans)_{\alpha\beta} + 2\mathbb{G}(\gsans,\delsans\gsans,\Gammasans,\delsans\Gammasans)_{\alpha\beta},
\end{equation}
where $\Gammasans_{\alpha} = \gsans_{\alpha\lambda}\Gammasans^{\lambda}$ and 
\be
2\mathbb{G}(\gsans,\delsans\gsans,\Gammasans,\delsans\Gammasans)_{\alpha\beta} 
:= (\delsans_{\alpha} \Gammasans_{\beta} + \delsans_{\beta} \Gammasans_{\alpha}) 
+ \Gammasans^{\delta} \delsans_{\delta}\gsans_{\alpha\beta} 
- \Gammasans_{\alpha} \Gammasans_{\beta},
\ee
\ese
and
\bse
\be
\Fbb_{\alpha\beta} = \Pbb_{\alpha\beta} + \Qbb_{\alpha\beta},
\ee
\begin{equation}\label{eq2-13-sept-2025} 
\aligned
& \Qbb_{\alpha\beta}(\gsans,\delsans \gsans) 
\\
& := 
\gsans^{\mu\mu'} \gsans^{\nu\nu'} \delsans_\mu \gsans_{\alpha\nu} \delsans_{\mu'} \gsans_{\beta\nu'}
- \gsans^{\mu\mu'} \gsans^{\nu\nu'} \big(\delsans_\mu \gsans_{\alpha\nu'} \delsans_\nu \gsans_{\beta\mu'} - \delsans_\mu \gsans_{\beta\mu'} \delsans_\nu \gsans_{\alpha\nu'} \big)
\\
& \quad + \gsans^{\mu\mu'} \gsans^{\nu\nu'} \big(\delsans_\alpha \gsans_{\mu\nu} \delsans_{\nu'} \gsans_{\mu'\beta} - \delsans_\alpha \gsans_{\mu'\beta} \delsans_{\nu'} \gsans_{\mu\nu} \big)
+ \frac{1}{2} \gsans^{\mu\mu'} \gsans^{\nu\nu'} \big(\delsans_\alpha \gsans_{\mu\beta} \delsans_{\mu'} \gsans_{\nu\nu'} - \delsans_\alpha \gsans_{\nu\nu'} \delsans_{\mu'} \gsans_{\mu\beta} \big)
\\
& \quad + \gsans^{\mu\mu'} \gsans^{\nu\nu'} \big(\delsans_\beta \gsans_{\mu\nu} \delsans_{\nu'} \gsans_{\mu'\alpha} - \delsans_\beta \gsans_{\mu'\alpha} \delsans_{\nu'} \gsans_{\mu\nu} \big)
+ \frac{1}{2} \gsans^{\mu\mu'} \gsans^{\nu\nu'} \big(\delsans_\beta \gsans_{\mu\alpha} \delsans_{\mu'} \gsans_{\nu\nu'} - \delsans_\beta \gsans_{\nu\nu'} \delsans_{\mu'} \gsans_{\mu\alpha} \big)
\endaligned
\end{equation} 
and
\begin{equation} \label{eq3-13-sept-2025} 
\aligned
\Pbb_{\alpha\beta} (\gsans,\delsans \gsans) 
:= - \frac{1}{2} \gsans^{\mu\mu'} \gsans^{\nu\nu'} \delsans_\alpha \gsans_{\mu\nu} \delsans_\beta \gsans_{\mu'\nu'} + \frac{1}{4} \gsans^{\mu\mu'} \gsans^{\nu\nu'} \delsans_\alpha \gsans_{\mu\mu'} \delsans_\beta \gsans_{\nu\nu'}.
\endaligned
\end{equation}
\ese
For the clarity in the presentation, when dealing with a tensor such as the Ricci curvature, we explicitly distinguish its dependence on the metric components $\gsans_{\alpha\beta}$ and of the contracted Christoffel symbols $\Gammasans^{\lambda}$, namely 
\be
\Rbb(\gsans,\delsans\gsans,\delsans\delsans\gsans,\Gammasans,\delsans\Gammasans)_{\alpha\beta} := - \frac{1}{2}\gsans^{\mu\nu}\delsans_{\nu}\delsans_{\mu}\gsans_{\alpha\beta}
+ \frac{1}{2}\Fbb(\gsans,\delsans\gsans)_{\alpha\beta} 
+ \mathbb{G}(\gsans,\delsans\gsans,\Gammasans,\delsans\Gammasans)_{\alpha\beta}.
\ee
Dropping the contribution of the contracted Christoffel symbols, we also introduce
\begin{equation}\label{eq9-13-sept-2025}
\Rbb^{\mathrm{w}}(\gsans,\delsans\gsans,\delsans\delsans\gsans)_{\alpha\beta}
:= - \frac{1}{2}\gsans^{\mu\nu}\delsans_{\nu}\delsans_{\mu}\gsans_{\alpha\beta}
+ \frac{1}{2}\Fbb(\gsans,\delsans\gsans)_{\alpha\beta}
\end{equation}
which is referred to as the \textbf{wave-reduced Ricci curvature}, i.e., the expression of the Ricci curvature when the coordinate chart $\{\xsans^{\alpha}\}$ satisfies the wave coordinate condition. Hence, the Ricci curvature can be expressed in terms of $(\gsans,\Gammasans)$ via ${\Rbb^{\mathrm{w}}}_{\alpha\beta}$. It is worth observing that the (second-order) principal part of the operator ${\Rbb^{\mathrm{w}}}(\gsans,\delsans \gsans,\delsans\delsans\gsans)_{\alpha\beta}$ forms a diagonal quasilinear wave operator acting on $\gsans_{\alpha\beta}$.

We recall the Einstein equation in the general from
\begin{equation}\label{eq4-13-sept-2025}
\widetilde{G}_{\alpha\beta} = \widetilde{T}_{\alpha\beta}, 
\qquad \widetilde{G}_{\alpha\beta} := \widetilde{R}_{\alpha\beta} - \frac{1}{2}R_g\gsans_{\alpha\beta}, 
\end{equation}
where $R_g$ denotes the scalar curvature associated with $g$. In turn,  thanks to \eqref{eqs1-exp}, the equations \eqref{eq4-13-sept-2025} can be written as
\begin{equation}\label{eq5-13-sept-2025}
\Rbb^{\mathrm{w}}(\gsans,\delsans\gsans,\delsans\delsans\gsans)_{\alpha\beta} 
= \Tsans_{\alpha\beta} -  
\frac{1}{2}\mathrm{Tr}_{\gsans}(\widetilde T)\gsans_{\alpha\beta}
- \mathbb{G}(\gsans,\delsans\gsans,\Gammasans,\delsans\Gammasans)_{\alpha\beta}.
\end{equation}


\paragraph{Admissible reference metric.}

As in the study of Einstein's vacuum equations~\cite{LR1,LR2,Lindblad-3} and in our earlier investigation of the Einstein–scalar field system, we consider a small perturbation around a \textit{reference metric}, which was the Schwarzschild metric in~\cite{LR1,LR2,Lindblad-3} and, more generally, a metric denoted by $\gref$. This reference metric plays a crucial role in the analysis of the Einstein equations, owing to the typically slow decay in the spatial directions. This slow decay is a consequence of the positive mass theorem, and the Schwarzschild decay rate $w{1/r}$ represents the fastest admissible decay. In the present analysis, we could essentially take an approximate solution to Einstein's vacuum equations by smoothly merging the Minkowski and Schwarzschild metrics, as proposed in~\cite{LR2}. However, it is convenient to allow for a slight generalization of this standard reference. 

To proceed, we also introduce the {\bf extended family of admissible vector fields}, defined by
\be
\mathscr{K}:=\mathscr{Z}\cup\{S\},\qquad S := t\del_t + x^a\del_a, 
\ee
which forms a family of conformal Killing vector fields for the Minkowski metric. We denote by 
$\mathscr{K}^I$ a differential operator of arbitrary order, built from vector fields in $\mathscr{K}$, and we refer to it as an \textbf{extended admissible operator}. We also introduce the notions of \textbf{order} and \textbf{rank} in this context: an extended admissible operator is said to be of order~$p$ if it is composed of~$p$ vector fields in~$\mathscr{K}$, and of rank~$k$ if it contains exactly~$k$ boosts, rotations, and/or occurrences of~$S$. For a scalar field~$u$, we use the notation $|u|^{\Kscr}_{p,k}$ as a natural extension of our earlier notation associated with~$\mathscr{Z}$. With the above preparations, we now introduce three conditions. 

\bei 

\item[$\bullet$] 
{\bf Asymptotically Minkowski condition.} 
We fix an integer~$N$ sufficiently large\footnote{Taking $N \geq 22$ is sufficient.} and, in the coordinates~$\{\xsans^{\alpha}\}$, we require the following decay on the metric components $\greft_{\alpha\beta} := \gref(\delsans_{\alpha},\delsans_{\beta})$:
\begin{equation}\label{eq2-14-sept-2025}
\aligned
&|\Hreft_{\alpha\beta}|^{\Kscr}_{N+2} \lesssim \epss\, (\tsans + \rsans + 1)^{-1+\theta}, 
\\
&|\delsans \widetilde{\slashed{H}}_{\mathrm{r}}| \lesssim \epss\, (\tsans + \rsans + 1)^{-1}
\qquad \text{near the light cone, i.e., for } (1- \ell)\tsans \leq \rsans \leq (1- \ell)^{-1}\tsans,
\endaligned
\end{equation}
where $0 < \ell < \tfrac{1}{2}$ and $\theta > 0$ denotes a sufficiently small constant\footnote{Taking $\theta \simeq 10^{-4}$ is sufficient.}. 
Here, $\widetilde{\slashed{H}}$ designates the {\bf ``good'' components} of the metric in the null frame associated with the coordinates~$\{\xsans^{\alpha}\}$, that is,
\be 
\widetilde{\slashed{H}} \in \bigl\{\,\widetilde{H}^{\Ncal}_{\alpha\beta} \,\big|\, (\alpha,\beta)\neq(0,0)\bigr\}.
\ee
We refer to~\eqref{eq2-14-sept-2025} as the asymptotically Minkowski condition, which coincides with the properties established for the reference metric in~\cite{LR2}. See also Remark~\ref{rmk1-21-sept-2025} below.

\item[$\bullet$] 
{\bf Generalized wave coordinate condition.} 
We also require that the coordinate functions~$\{\xsans^{\alpha}\}$ are approximately wave coordinates with respect to~$\gref$. More precisely, let $\widetilde{W}_{\mathrm{r}}^{\alpha}$ be four sufficiently regular functions defined on~$\Mcal$, satisfying
\begin{equation}\label{eq8-21-sept-2025}
\big|\widetilde{\Kscr}^I \widetilde{W}_{\mathrm{r}}^{\alpha}\big|
\lesssim \epss\,(\tsans+\rsans+1)^{-4-(p-k)+\theta},
\end{equation}
where $\Kscr^I$ denotes an extended admissible operator of type~$(p,k)$ with respect to~$\{\xsans^{\alpha}\}$ and $p \leq N$. Then, if
\begin{equation}\label{eq7-21-sept-2025}
\greft^{\mu\nu}\Gammasans[\gref]_{\mu\nu}^{\alpha} = \widetilde{W}_{\mathrm{r}}^{\alpha},
\end{equation}
the reference metric~$\gref$ is said to satisfy the generalized wave coordinate condition in the chart~$\{\xsans^{\alpha}\}$. We observe that the vacuum solution constructed in~\cite{LR2} satisfies the exact wave coordinate condition, that is, $\widetilde{W}_{\mathrm{r}}^{\alpha} \equiv 0$ if we takes the vacuum solution as the reference.

\item[$\bullet$] 
{\bf Almost Ricci-flat condition.}
We also require that~$\gref$ is approximately Ricci-flat, that is, anapproximate solution to the vacuum Einstein equations. More precisely, there exist sufficiently regular, symmetric tensor fields~$\widetilde{U}_{\mathrm{r}\,\alpha\beta}$ satisfying, for any extended admissible operator~$\Kscr^I$ associated with~$\{\xsans^{\alpha}\}$ of type~$(p,k)$ with~$p\le N$, 
\begin{equation}\label{eq8-13-sept-2025}
\big|\widetilde{\Kscr}^I \widetilde{U}_{\mathrm{r}\,\alpha\beta}\big|
\lesssim \epss^2\,(\tsans+\rsans+1)^{-4-(p-k)+\theta},
\end{equation}
and 
\begin{equation}
\widetilde{R}_{\mathrm{r}\,\alpha\beta} = \widetilde{U}_{\mathrm{r}\,\alpha\beta}.
\end{equation} 
This is called the almost Ricci-flat condition. 
We note that the vacuum solution constructed in~\cite{LR2} satisfies 
$\widetilde{U}_{\mathrm{r}\,\alpha\beta}\equiv0$, so that the condition above holds identically.
\eei


A reference metric satisfying the above three conditions is called a 
\textbf{$(N,\theta,\epss,\ell)$--admissible reference metric}. 
Furthermore, an $(N,\theta,\epss,\ell)$--admissible reference metric is called 
\textbf{light-bending} if it satisfies
\begin{equation}\label{def-61} 
\Hsans_{\mathrm{r}}^{\Ncal00} := 
\gref(\diff \rsans- \diff \tsans,\diff \rsans- \diff \tsans) < 0
\quad \text{in the extended exterior region } 
\{\rsans \geq \tfrac{3}{4}\tsans\}.
\end{equation}
We emphasize that a slightly weaker version of this requirement was imposed in~\cite{PLF-YM-PDE}. On the other hand, this condition, although satisfied by the Schwarzschild metric, may appear restrictive for more general reference metrics. 
However, as we shall show in the remainder of this section, the light-bending property can always be achieved after performing a suitable coordinate transformation.

\begin{remark}[Asymptotically Minkowski condition]
\label{rmk1-21-sept-2025}
Let us consider the inequalities \eqref{eq2-14-sept-2025}. 
In~\cite{PLF-YM-PDE}, we also imposed the following estimates in the exterior region $\{\rsans \geq \tsans - 1\}$:
\be
\aligned
& \big|\delsans^m \Hreft_{\alpha\beta}\big|_{N-m}
\lesssim \epss (\tsans + \rsans + 1)^{-1+\theta} (1 + |\rsans - \tsans|)^{- \kappa - m},
\quad m = 0, 1,
\\
& \big|\delsans^m \widetilde{\slashed{\del}} \Hreft_{\alpha\beta}\big|_{N-m}
\lesssim \epss (\tsans + \rsans + 1)^{-1- \kappa} (1 + |\rsans - \tsans|)^{-m},
\quad m = 0, 1.
\endaligned
\ee
However, these inequalities are \emph{not independent} from~\eqref{eq2-14-sept-2025}. 
Indeed, recalling~\cite[Eqs.~(5.1)--(5.4) and Lemma~5.1]{LR2}, we have
\begin{equation}\label{eq3-01-oct-2025}
\aligned
(1 + |\tsans - \rsans|)^m 
\big|\delsans^m \Hreft_{\alpha\beta}\big|^{\Kscr}_{N - m + 1} 
& \lesssim 
\big|\Hreft_{\alpha\beta}\big|^{\Kscr}_{N + 1}
\lesssim \epss (\tsans + \rsans + 1)^{-1+\theta},
\quad m \leq N + 1,
\\
(1 + \tsans + \rsans)
\big|\widetilde{\slashed{\del}} \Hreft\big|^{\Kscr}_N
& \lesssim 
\big|\Hreft_{\alpha\beta}\big|^{\Kscr}_{N + 1}
\lesssim \epss (\tsans + \rsans + 1)^{-2+\theta}.
\endaligned
\end{equation}
\end{remark}

}


\subsection{ The coordinate transformation}
\label{section===61}

{

\paragraph{Proposed strategy.}

In order to ensure the light-bending property~\eqref{def-61} for the reference metric, we now introduce a suitable change of coordinates.  
The purpose of this transformation is to adjust the spacetime foliation so that the level sets of the new radial variable align more closely with the null directions of the metric~$\gref$.  
This construction preserves the asymptotic flatness and the generalized wave coordinate structure, while improving the behavior of the metric near the light cone.  
In particular, the transformation will slightly modify the spatial coordinates in the exterior region, bending them toward the outgoing null directions.  
As a result, the new coordinate system provides a natural framework for establishing the hyperboloidal foliation and for controlling the nonlinear interaction terms that appear in the Einstein--massive field systems.  

To achieve the condition~\eqref{def-61}, we follow the Kauffman-Lindblad strategy (see~\cite{KauffmanLindblad}), which treated the Schwarzschild metric.  
Let $\{\xsans^{\alpha}\}$ be a global coordinate chart defined on~$\Mcal$, and let $g = \gsans_{\alpha\beta}\diff \xsans^{\alpha}\otimes\diff\xsans^{\beta}$ be a general Lorentzian metric satisfying an \emph{asymptotic flatness condition} of order~$\rsans^{-1+\theta}$ with $\theta\in [0,1)$, namely
\be
\big|\gsans_{\alpha\beta}- \etasans_{\alpha\beta}\big|\lesssim \eps_s\, \rsans^{-1+\theta}.
\ee
We then introduce the \emph{modified distance function}
\begin{equation}\label{eq1-08-sept-2025}
\ravec := \rsans + C_m\,\chi(\rsans/\tsans)\, f(\rsans)
=
\begin{cases}
\rsans + C_m\,\theta^{-1}\chi(\rsans/\tsans)\rsans^{\theta}, & \theta>0,\\[0.3em]
\rsans + C_m\,\chi(\rsans/\tsans)\ln(\rsans), & \theta=0,
\end{cases}
\end{equation}
where the function~$f$ satisfies
\be
\frac{\del f}{\del \rsans} = \rsans^{-1+\theta}.
\ee
In the following discussion, we restrict attention to the case $\theta>0$.

Here, $C_m$ is a sufficiently small constant (not necessarily positive) to be determined later, while $\chi$ is a \emph{smooth cutoff function} such that
\begin{equation}\label{eq8-24-sept-2025}
\chi(\tau) =
\begin{cases}
0, & \tau < 1/2,\\
1, & \tau > 5/8.
\end{cases}
\end{equation}
Finally, we introduce the \emph{modified coordinate functions}
\begin{equation}\label{eq4-10-april-2025}
\xavec^0 = \tavec := \tsans, 
\qquad 
\xavec^a := (\ravec/\rsans)\,\xsans^a.
\end{equation}


\paragraph{Derivatives of the new coordinates.}

We need to establish several estimates. 
The new coordinates depend smoothly upon their arguments and the Jacobian $\Psisans{} = \left(\frac{\del \xavec{}^{\alpha}}{\del \xsans^{\beta}} \right)$ between the charts $\{\xsans^{\alpha} \}$ and $\{\xavec{}^{\alpha} \}$ is given by
\begin{equation}
\aligned
& \frac{\del \xavec{}^0}{\del \xsans^0} = 1, 
\qquad\frac{\del \xavec{}^0}{\del \xsans^b} = 0,
\\
& \frac{\del \xavec{}^a}{\del \xsans^0} 
= - C_m \chi'( \rsans / \tsans)\frac{ f(\rsans)}{\tsans^2}\xsans^a
=
- C_m\theta^{-1}\chi'(\rsans/\tsans)\frac{\rsans^{\theta}}{\tsans^2}\,\xsans^a,
\\
& \frac{\del \xavec{}^a}{\del \xsans^b} \,=\, \delta_b^a\,
+\, C_m\chi'(\rsans/\tsans)\Big(\frac{f(\rsans)}{\tsans}\Big)\frac{\xsans^a\xsans^b}{\rsans^2} 
+ C_m\chi(\rsans/\tsans)\Big(\rsans^{-1}f(\rsans)\Big)'\frac{\xsans^a\xsans^b}{\rsans}
+ C_m\chi(\rsans/\tsans)\frac{f(\rsans)}{\rsans}\delta_b^a
\\
& \quad\ \ \ 
=
\delta_b^a
+ C_m\theta^{-1}\chi'(\rsans/\tsans)\Big(\frac{\rsans^{\theta}}{\tsans}\Big)\frac{\xsans^a\xsans^b}{\rsans^2}
- C_m\theta^{-1}(1- \theta)\chi(\rsans/\tsans)\rsans^{-1+\theta}\frac{\xsans^a\xsans^b}{\rsans^2} 
+ C_m\theta^{-1}\chi(\rsans/\tsans)\rsans^{-1+\theta}\delta_b^a.
\endaligned
\end{equation}
Since $\chi'(\rsans/\tsans)$ is supported in $\{\tsans/2\leq \rsans\leq 5\tsans/8\}$, we have 
\begin{equation}\label{equa-18mai-1} 
\Psisans{} = \mathrm{I}_4 + 
C_m\theta^{-1}\frac{\rsans^{\theta}}{\rsans} D(\tsans,\rsans),
\end{equation}
where the elements of the matrix $D$ are \emph{homogeneous functions of degree zero}, supported in the extended exterior light-cone region $\{\rsans \geq \tsans/2\}$. It is clear that, provided $|C_m\theta^{-1}|$ is sufficiently small, $\Psisans{}$ is invertible, hence $\{\xavec{}^{\alpha} \}$ forms a coordinate chart on $\Mcal_{[s_0, s_1]}$. 

We also introduce the natural frame associated with the ``old'' coordinates $\{\xsans^{\alpha}\}$,
\be
\delsans_{\alpha} := \frac{\del}{\del \xsans^{\alpha}}, 
\ee
and the natural frame associated with the modified coordinates $\{\xavec^{\alpha}\}$,
\be
\delavec_{\alpha} := \frac{\del}{\del \xavec^{\alpha}}. 
\ee
The matrix $\Phisans{} = \Psisans{}^{-1}$ provides us with the transition relations between $\{\delsans_{\alpha} \}$ and $\{\delavec{}_{\alpha} \}$, specifically 
\be
\delavec{}_0 = \delsans_0, 
\qquad 
\delavec{}_a = \frac{\del}{\del \xavec{}^a} = \frac{\del \xsans^b}{\del \xavec{}^a} \delsans_b.
\ee
With the notation $\Phisans{}_{\alpha}^{\beta} = \big(\Phisans{} \big)_{\beta\alpha}$ and $\Psisans{}_{\alpha}^{\beta} = \big(\Psisans{} \big)_{\beta\alpha}$, we then have 
\begin{equation}
\delavec{}_{\alpha} = \Phisans{}_{\alpha}^{\beta} \, \delsans_{\beta},
\quad 
\delsans_{\alpha} = \Psisans{}_{\alpha}^{\beta} \, \delavec{}_{\beta}. 
\end{equation}


\paragraph{Estimates for the new coordinates.}

For clarity, we define the vector fields 
\be
\Ssans := \tsans\delsans_t + \xsans^a\delsans_a,
\quad
\Lsans_a := \xsans^a\delsans_t + \tsans\delsans_a,
\quad
\Omegasans_{ab} := \xsans^a\delsans_b - \xsans^b\delsans_a.
\ee
Together with $\delsans_{\alpha}$, these vectors form the family of extended admissible vector fields $\widetilde{\Kscr}$ associated with the coordinates $\{\xsans^{\alpha}\}$. We denote by $\Zsans$ any one in $\widetilde{\Kscr}$. For the coordinates $\{\xavec^{\alpha}\}$ we continue to apply the notation $L_a,\Omega_{ab},\del_{\alpha}$ and $\Kscr$, $Z$ for the corresponding objects. We observe that when \eqref{eq1-08-sept-2025} and \eqref{eq4-10-april-2025} hold
\begin{equation}
\{r\leq t/2\} = \{\rsans\leq \tsans/2\} \ \text{and hence} \ \{r> t/2\} = \{\rsans> \tsans/2\}.
\end{equation}
That is a consequence of the fact that $r=\rsans$ and $t=\tsans$ in the region $\{\rsans\leq \tsans/2\}$.
Furthermore, provided $|C_m\theta^{-1}|$ is sufficiently small, we have 
\begin{equation}
\{r\geq 3t/4\}\subset \{\rsans\geq 5\tsans/8\}.
\end{equation}

We also observe that when \eqref{eq1-08-sept-2025} and \eqref{eq4-10-april-2025} hold with $|C_m\theta^{-1}|$ sufficiently small, 
\begin{equation}
r+t\lesssim  \rsans + \tsans\lesssim r + t.
\end{equation}
This is because the transition matrices \eqref{equa-18mai-1} are uniformly close to the identity.
In the following discussion, we often change $(1+\tsans+\rsans)^\alpha$ by $(1+t+r)^{\alpha}$ without further explanation. We denote by 
\be
\Xi_{\alpha}^{\beta} := \Psisans_{\alpha}^{\beta} - \delta_{\alpha}^{\beta},\quad
\Theta_{\alpha}^{\beta} := \Phisans_{\alpha}^{\beta} - \delta_{\alpha}^{\beta}.
\ee
In addition, it is clear that
\begin{equation}\label{eq3-11-sept-2025}
\aligned
\delavec_{\alpha} & = \Phisans_{\alpha}^{\beta}\delsans_{\beta} = \delsans_{\alpha} + \Theta_{\alpha}^{\beta}\delsans_{\beta},
\quad
\\ 
S & = \widetilde{S} + \tsans\Theta_0^{\alpha}\delsans_{\alpha} 
+ (r/\rsans)\xsans^a\Theta_a^{\beta}\delsans_{\beta} + (r/\rsans-1)\xsans^a\delsans_a,
\\
L_a
& = \Lsans_a + (r/\rsans -1)\xsans^a\Phisans_0^{\beta}\delsans_{\beta} 
+ \tsans\Theta_a^{\beta}\delsans_{\beta}
+ (r/\rsans)\xsans^a\Theta_a^{\beta}\delsans_{\beta},
\\
\Omega_{ab} & = (r/\rsans)\Omegasans_{ab} 
+ (r/\rsans)\big(\xsans^a\Theta_b^{\beta}
- \xsans^b\Theta_a^{\beta}\big)\delsans_{\beta}.
\endaligned
\end{equation}
We rewrite each new admissible vector
$Z\in\Kscr$ (built on $\{\xavec^\alpha\}$) as an old one $\Zsans\in\widetilde{\Kscr}$ (built on $\{\xsans^\alpha\}$)
plus a lower–order correction $\Gamma[Z]^{\beta}\delsans_\beta$. And we thus use the short-hand notation 
\begin{equation}
Z = \Zsans + \Gamma[Z]^{\beta}\delsans_{\beta}, 
\end{equation}
where the coefficients $\Gamma[Z]^{\beta}$ are given by \eqref{eq3-11-sept-2025}. 

We then establish the following estimates, which quantify the effect of the change of variables by estimating the transition matrices $\Psisans$ and $\Phisans=\Psisans^{-1}$ from \eqref{equa-18mai-1}. 

\begin{lemma}[Bounds for transition coefficients]
Assume that the constant $|C_m\theta^{-1}|$ is sufficiently small. Then provided ${\theta \in (0,1)}$ , 
any admissible operator $\Kscr^I$ of type $(p,k)$ one has 
\begin{equation}\label{eq1-11-sept-2025}
\big|\widetilde{\Kscr}^I(\Gamma[\del_{\gamma}]^{\beta})\big|\lesssim_p 
|C_m\theta^{-1}| \, r^{-1-(p-k)+\theta},
\end{equation}
\begin{equation}\label{eq2-11-sept-2025}
\big|\widetilde{\Kscr}^I(\Gamma[Z]^{\beta})\big|\lesssim_p 
|C_m\theta^{-1}| \, r^{-(p-k)+\theta}, \qquad Z = S, L_a,\Omega_{ab}. 
\end{equation}
\end{lemma}


\begin{proof} 
\bse
In the region $\{\rsans\geq \tsans/2\}$, for any homogeneous function $u$ of degree zero, the following estimate holds (by homogeneity):
\begin{equation}\label{eq2-08-sept-2025}
\aligned
\big|\widetilde{\Kscr}^I(\rsans^{-1+\theta}u)\big|
\lesssim_p \rsans^{-1 - (p-k) +\theta}\lesssim r^{-1 - (p-k)+\theta}.
\endaligned
\end{equation}
Thanks to \eqref{equa-18mai-1}, this observation implies 
\begin{equation}\label{eq10-13-sept-2025}
\big|\widetilde{\Kscr}^I\Xi_{\alpha}^{\beta}\big|\lesssim_p 
|C_m\theta^{-1}| \, r^{-1-(p-k)+\theta}.
\end{equation}
Since $\Phisans=\Psisans^{-1}= \mathrm{I}_4 + O\!\left(C_m\theta^{-1}\rsans^{-1+\theta}\right)$, the argument of Lemma~\ref{lem1-11-july-2025} leads to
\begin{equation}\label{eq1-12-sept-2025}
\big|\widetilde{\Kscr}^I\Theta_{\alpha}^{\beta}\big|\lesssim_p 
|C_m\theta^{-1}| \, r^{-1-(p-k)+\theta}.
\end{equation}
This leads us to \eqref{eq1-11-sept-2025}. 
\ese
\bse
To deal with \eqref{eq2-11-sept-2025}, we recall the formulas \eqref{eq3-11-sept-2025}. It follows that the coefficients $\Gamma[Z]^{\beta}$ are finite linear combinations of the terms 
\be
\xsans^{\gamma}\Theta_{\beta}^{\alpha},\quad (r/\rsans-1)\xsans^{\alpha},\quad (r/\rsans-1)\xsans^{\gamma}\Theta_{\alpha}^{\beta}. 
\ee
In the above expressions, we have the identity 
\begin{equation}
r/\rsans-1 = C_m\chi(\rsans/\tsans)\rsans^{-1}f(\rsans).
\end{equation}
Hence,  in view of \eqref{eq2-08-sept-2025} and the homogeneity of  $\rsans^{-1}$ and $\chi(\rsans/\tsans)$, 
we arrive at
\be
\big|\widetilde{\Kscr}^I\big(r/\rsans -1\big)\big| \lesssim_p 
|C_m\theta^{-1}| \, r^{-1-(p-k)+\theta}.
\ee
In combination with \eqref{eq1-12-sept-2025}, we conclude that \eqref{eq2-11-sept-2025} holds.
\ese
\end{proof} 


We are ready to rewrite derivatives in the new chart as combinations of old-chart derivatives with controllable weights. We observe that the transition matrices differ from the identity by $O\!\left(C_m\theta^{-1}r^{-1+\theta}\right)$, so each geometric vector field expansion contributes a small factor $\sim |C_m\theta^{-1}|\,r^{\theta}$ and at most one extra base derivative.

\begin{lemma}[Transferring the estimates for admissible operators under the coordinate change]
\label{lem1-13-sept-2025}
Consider sufficiently regular scalar fields $u$ defined in $\Mcal_{[s_0,s_1]}$. Assume that $|C_m\theta^{-1}|$ is sufficiently small. Then for any admissible operator $\Kscr^I$ of type $(p,k)$ associated with the coordinates $\{\xavec^{\alpha}\}$, one has 
\begin{equation}\label{eq2-12-sept-2025}
\big|\Kscr^Iu\big|\lesssim_p 
\sum_{\ord(K)=0}^k\big(|C_m\theta^{-1}r^{\theta}|\big)^{\ord(K)}\sum_{\ord(J)\leq p- \ord(K)\atop \rank(J)\leq k- \ord(K)}
\big|\widetilde{\Kscr}^J\delsans^K u\big|.
\end{equation}
Here, $\delsans^K$ denotes any $\ord(K)$--order admissible operator associated with $\{\xsans^{\alpha}\}$ but consisting only of $\delsans_{\alpha}$.
\end{lemma}

\begin{proof}
\bse
We proceed by induction on $\ord(I)$. When $\ord(I)=1$, \eqref{eq2-12-sept-2025} follows directly from \eqref{eq3-11-sept-2025}. Let $\Kscr^{I'} = \Kscr^I Z$ with $Z=\del_{\gamma},S, L_a$ or $\Omega_{ab}$, and denote by $(p,k)$ the type of $\Kscr^I$ and by $(p',k')$ the type of $\Kscr^{I'}$.

\vskip.15cm

\noindent$\bullet$ When $Z=\del_{\gamma}$, $(p',k')=(p+1,k)$, we find
\be
\aligned
\Kscr^{I'}u & = \Kscr^I(\del_{\gamma}u) 
= \Kscr^I(\delsans_{\gamma}u + \Gamma[\del_{\gamma}]^{\beta}\delsans_{\beta}u)
\\
& = \Kscr^I(\delsans_{\gamma}u) 
+ \sum_{I_1\odot I_2=I}\Kscr^{I_1}(\Gamma[\del_{\gamma}]^{\beta})\Kscr^{I_2}(\delsans_{\beta}u)
=: T_1+T_2.
\endaligned
\ee
For $T_1$, apply the induction hypothesis. For $T_2$, use \eqref{eq1-11-sept-2025}:
\be
|T_2|\lesssim_p |C_m\theta^{-1}|\, r^{-1+\theta}\!\!\sum_{\ord(J)\le p\atop \rank(J)\le k}
\big|\widetilde{\Kscr}^J\delsans u\big|,
\ee
which fits \eqref{eq2-12-sept-2025} (absorbing $r^{-1}\le 1$ into the $\ord(K)=1$ slice).

\vskip.15cm

\noindent$\bullet$ When $Z\in\{S,L_a,\Omega_{ab}\}$, $(p',k')=(p+1,k+1)$, we have
\be
\aligned
\Kscr^{I'}u & = \Kscr^I(\Zsans u + \Gamma[Z]^{\beta}\delsans_{\beta}u)
= \Kscr^I(\Zsans u) 
+ \sum_{I_1\odot I_2=I}\Kscr^{I_1}(\Gamma[Z]^{\beta})\Kscr^{I_2}(\delsans_{\beta}u)
=: T_1+T_2.
\endaligned
\ee
The term $T_1$ is controlled by the induction hypothesis (recall $[\delsans^K,\Zsans]$ is a linear combination of $\delsans^{K'}$ of the same order with constant coefficients). For $T_2$, by \eqref{eq2-11-sept-2025} and the induction hypothesis,
\be
\aligned
\big|\Kscr^{I_1}(\Gamma[Z]^{\beta})\Kscr^{I_2}(\delsans_{\beta}u)\big|
& \lesssim_p |C_m\theta^{-1}|\, r^{\theta}\!\!\sum_{\ord(K)=0}^{k}\!\big(|C_m\theta^{-1}|\, r^{\theta}\big)^{\ord(K)}
\!\!\sum_{\ord(J)\le p- \ord(K)\atop \rank(J)\le k- \ord(K)}
\big|\widetilde{\Kscr}^J\delsans^{K}\delsans_{\beta} u\big|
\\
& \lesssim_p \sum_{\ord(K')=1}^{k'}\big(|C_m\theta^{-1}|\, r^{\theta}\big)^{\ord(K')}
\!\!\sum_{\ord(J)\le p'- \ord(K')\atop \rank(J)\le k'- \ord(K')}
\big|\widetilde{\Kscr}^J\delsans^{K'} u\big|,
\endaligned
\ee
where we wrote $\delsans^{K'}=\delsans^{K}\delsans_{\beta}$ and used $p'=p+1$ and $k'=k+1$.
\ese
\end{proof}


We thus have the following conclusion on the transition matrices, whose proof is immediate by recalling \eqref{eq10-13-sept-2025}, \eqref{eq1-12-sept-2025}, and applying Lemma~\ref{lem1-13-sept-2025}.

\begin{corollary}[Bounds for transition matrices under the coordinate change]
\label{cor1-14-sept-2025}
Assume that $|C_m\theta^{-1}|$ is sufficiently small. Then one has 
\begin{equation}
|\Kscr^I\Xi_{\alpha}^{\beta}|+|\Kscr^I\Theta_{\alpha}^{\beta}|\lesssim_p
|C_m\theta^{-1}|(t+r+1)^{-1-(p-k)+\theta}, 
\end{equation}
where $\Kscr^I$ is any admissible operator associated with $\{\xavec^{\alpha}\}$ of type $(p,k)$.
\end{corollary}

\begin{corollary}[Derivatives in the new chart of the reference metric components]
\label{cor2-21-sept-2025}
Assume that $|C_m\theta^{-1}|$ is sufficiently small and the reference metric $\gref$ is $(N,\theta,\epss,\ell)$-admissible. Then for any extended admissible operator $\Kscr^I$ of type $(p,k)$ associated with the new coordinates $\{\xavec^{\alpha}\}$, with $p\leq N+2$, one has 
\begin{equation}\label{eq5-21-sept-2025}
\big|\Kscr^I\Hreft_{\alpha\beta}\big| + \big|\Kscr^I\Hreft^{\alpha\beta}\big|\lesssim_p \epss\big(1 + |C_m\theta^{-1}| \, r^{k\theta}(1+|r-t|)^{-k}\big)
(1+|r-t|)^{-(p-k)}(t+r+1)^{-1+\theta}.
\end{equation}
\end{corollary}


\begin{proof} 
\bse
For any $(p,k)$ type extended admissible operator $\widetilde{\Kscr}^I$ associated with $\{\xsans^{\alpha}\}$, we have 
\begin{equation}
\big|\widetilde{\Kscr}^I \Hreft_{\alpha\beta}\big|\lesssim \epss (1+|r-t|)^{-(p-k)}(1+t+r)^{-1+\theta}, 
\end{equation}
as essentially already pointed out in Remark~\ref{rmk1-21-sept-2025}. More precisely, recalling \cite[Eqs. (5.1)--(5.4) and Lemma~5.1]{LR2}, we find 
\begin{equation}
\big|\widetilde{\Kscr}^J\delsans^j u\big|\lesssim (1+|t-r|)^{-j}\sum_{\ord(J')\leq \ord(J)+j}\big|\widetilde{\Kscr}^{J'}u\big|.
\end{equation}
This means, in short, that one can trade a unit of rank for one derivative order at the cost of a factor $(1+|r-t|)^{-1}$.
\ese
\bse
Now let us consider the estimate \eqref{eq5-21-sept-2025}. By a commutation relation of the family $\Kscr$, we can write $\Kscr^I$ as a finite linear combination of 
\be
\Kscr^{J}\del^{p-k},\quad \ord(J) = \rank(J)\leq k.
\ee
For example, since $[\del_t,S]=\del_t$, we have $\del_t S = S\del_t + \del_t$. This is easily checked by induction. Consequently, we apply \eqref{eq2-12-sept-2025} and obtain
\be
\aligned
\big|\Kscr^{J}\del^m u\big|& \lesssim_p  
\sum_{\ord(K)=0}^{\ord(J)}\big(|C_m\theta^{-1}| \, r^{\theta}\big)^{\ord(K)}\sum_{\ord(J')\leq{ \ord(J)- \ord(K)}}\big|\widetilde{\Kscr}^{J'}\delsans^K\delsans^m u\big|
\\
& \lesssim \sum_{\ord(K)=0}^{\ord(J)}\big(|C_m\theta^{-1}| \, r^{\theta}\big)^{\ord(K)}(1+|r-t|)^{- (\ord(K)+m)}\sum_{\ord(J'')\leq \ord(J)+m}\big|\widetilde{\Kscr}^{J''}u\big|
\\
& \lesssim  \Big(1 + |C_m\theta^{-1}| \, r^{\ord(J)\theta}(1+|r-t|)^{- \ord(J)}\Big)(1+|r-t|)^{-m}
\sum_{\ord(J'')\leq \ord(J)+m}\hspace{-0.8cm}\big|\widetilde{\Kscr}^{J''}u\big|.
\endaligned
\ee
We fix $m=p-k$ and observe that $\ord(J)\leq k$. Then from \eqref{eq2-14-sept-2025} we obtain the estimate on $\Hreft_{\alpha\beta}$. For $\Hreft^{\alpha\beta}$, we use the standard series control of the inverse metric (as in the proof of Lemma~\ref{lem1-11-july-2025}) to transfer the same bounds to the contravariant components. 
\ese
\end{proof}

}


\subsection{ Reference metric in the new coordinate chart}

{ 

\paragraph{Light-bending condition.}

In the above notation, we can compute the null component of the metric perturbation
\be
\Havec{}^{\Ncal 00} := (\diff \xavec{}^0 - \diff \ravec{}, \diff \xavec{}^0- \diff \ravec{})_g.
\ee 
We now establish that given any $(N,\theta,\epss,\ell)$--admissible reference metric, there always exists a new coordinate system \eqref{eq1-08-sept-2025} in which the sign condition $\Havec{}^{\Ncal 00}<0$ is satisfied.

\begin{proposition}[The light-bending condition holds in the new coordinates]
\label{lem1-12-sept-2025}
Let $\gsans_{\alpha\beta} = \etasans_{\alpha\beta} + \Hsans_{\alpha\beta}$ be a metric defined in $\Mcal$ and suppose that  
\begin{equation}\label{eq1-07-sept-2025}
| \Hsans | :=\max_{\alpha, \beta} | \Hsans_{\alpha\beta}| < \epss \rsans^{-1+\theta} 
\quad  \text{ in } \Mcal \cap \{\rsans \geq (5/8) \tsans \} \quad\text{ for some } \theta> 0. 
\end{equation}
Then there exist two universal constants $K_0,K_1$ and a constant $\eps_0>0$ such that for any $0<\epss\leq \epsm\leq \eps_0$ and any $C_L\ge 0$ satisfying 
\be
0\leq C_L\leq \frac{1}{K_1\epsm}-K_0,
\ee
the following property holds. In the new coordinates $\{\tavec, \xavec^a = (\ravec{}/\rsans)\xsans^a \}$ defined by
\bel{equa-coord-transform}
\tavec{} = \tsans, \qquad \ravec{} = \rsans - K\theta^{-1}\epsm \chi(\rsans / \tsans) \rsans^{\theta},
\end{equation}
with $K = K_0+C_L$, the null metric component satisfies 
\begin{equation}\label{eq3-07-sept-2025}
\Havec^{\Ncal 00} := (\diff \tavec{} - \diff \ravec{}, \diff \tavec{} - \diff \ravec{})_g 
< -  C_L\epsm \rsans^{-1+\theta}
\qquad \text{ in } \Mcal \cap\{ \rsans \geq (5/8) \tsans \}.
\end{equation}
\end{proposition}


\begin{proof} 
\bse
We compute 
\be
\diff(t-r) = \diff\tsans - \diff\rsans 
+ K\theta^{-1}\epsm\chi'(\rsans/\tsans)\rsans^{-1+\theta}\Big(\frac{\rsans}{\tsans}\diff\rsans - (\rsans/\tsans)^2\diff\tsans\Big)
+K\epsm\chi(\rsans/\tsans)\rsans^{-1+\theta}\diff\rsans.
\ee
Since $\chi'(\rsans/\tsans) \equiv 0$ in the region $\Mcal \cap \{\rsans \geq (5/8) \tsans \}$, the above expression reduces to
\be
\diff(t-r) = \diff\tsans - \diff\rsans + K\epsm\rsans^{-1+\theta}\diff\rsans.
\ee
Then we deduce that 
\be
\Havec^{\Ncal00} = \Hsans^{\Ncal00} 
+ 2K\epsm\rsans^{-1+\theta}\big(\diff\tsans- \diff\rsans,\diff\rsans\big)_g
+ \big(K\epsm\rsans^{-1+\theta}\big)^2(\diff\rsans,\diff\rsans)_g.
\ee
\ese
This leads us to
\begin{equation}\label{eq2-07-sept-2025}
\aligned
\Havec^{\Ncal00} & = \Hsans^{\Ncal00} - 2K\epsm \rsans^{-1+\theta} + R[\Hsans],
\\
R[\Hsans] :& = K\epsm\rsans^{-1+\theta}\Big(\sum_a\frac{\xsans^a}{\rsans}\Hsans^{0a}
- \sum_{a,b}\frac{\xsans^a\xsans^b}{\rsans^2}\Hsans^{ab} \Big)
+ K^2\epsm^2\rsans^{-2+2\theta}\sum_{a,b}\frac{\xsans^a\xsans^b}{\rsans^2}\gsans^{ab}.
\endaligned
\end{equation}
Here we used $(\diff\tsans- \diff\rsans,\diff\rsans)_{\etasans}=-1$ and $(\diff\rsans,\diff\rsans)_{\etasans}=1$, so the Minkowski contribution yields the term $-2K\epsm\rsans^{-1+\theta}$.)

In view of \eqref{eq1-07-sept-2025}, there exists a universal constant $K_0>0$ such that 
\begin{equation}
\rsans^{1- \theta} \big|\Hsans^{\Ncal00}\big|\leq K_0\epss,\quad \text{ in }  \{\rsans\geq 5\tsans/8\}.
\end{equation}
On the other hand, we have 
\be
\rsans^{1- \theta} |R[\Hsans]|\leq \frac{1}{2}K\epsm^2(C_0 + K_1K)
\ee
with positive universal constants $C_0,K_1$. Thus \eqref{eq2-07-sept-2025} implies 
\be
\rsans^{1- \theta} \Havec^{\Ncal00}\leq K_0\epsm + \frac{1}{2}K\epsm^2(C_0 + K_1K) - 2K\epsm.
\ee
\bse
Provided that 
\begin{equation}\label{eq4-07-sept-2025}
\epsm\leq \frac{1}{C_0},\quad K\leq \frac{1}{K_1\epsm}, \quad\epsm\leq \frac{1}{K_0K_1}, 
\end{equation}
we deduce that 
\be
\epsm|C_0+K_1K|\leq 2,\quad K_0\leq \frac{1}{K_1\epsm}.
\ee
In this case, we take $0\leq C_L\leq \frac{1}{K_1\epsm} - K_0$ and we pick $K = K_0 + C_L$. 
In this case, \eqref{eq3-07-sept-2025} holds and we guarantee \eqref{eq4-07-sept-2025} by choosing 
\be
\eps_0 = \min\big\{C_0^{-1},(K_0K_1)^{-1}\big\}. \qedhere
\ee
\ese
\end{proof}


\paragraph{Next aim.}

The coordinate transformation \eqref{equa-coord-transform} produces a light-bending coordinate chart.  The price to pay is that it does not preserve the wave coordinate condition. Nevertheless, the condition is only mildly violated: as we will prove in Proposition~\ref{lem1-26-sept-2025} below, the error term can be quantitatively controlled.

\paragraph{Asymptotic Minkowski condition.}

First, we translate the asymptotically Minkowski condition for an admissible reference metric into the new coordinate chart $\{\xavec^{\alpha}\}$. This change of coordinates was constructed in Proposition~\ref{lem1-12-sept-2025} applied to $g=\gref$, i.e.,
\begin{equation}
t=\tsans,\quad x^{a}=\big(r/\rsans\big)\,\xsans^{a},\quad
r=\rsans- \Kref\,\epsm\,\chi(\rsans/\tsans)\,\rsans^{\theta},
\end{equation}
with $\Kref:=(K_0+C_L)\theta^{-1}$ and we assume $\Kref\epsm$ is sufficiently small.

For convenience, we define 
\be
\Hreff_{\alpha\beta} := \gref_{\alpha\beta} - \eta_{\alpha\beta}, 
\ee
which is the deviation of $\gref$ from the Minkowski metric associated with the coordinates $\{\xavec^{\alpha}\}$. It is clear that
\begin{equation}\label{eq5-14-sept-2025}
\Hreff_{\alpha\beta} = \Href_{\alpha\beta} + \Phisans_{\alpha}^{\mu}\Phisans_{\beta}^{\nu}\etasans_{\mu\nu} - \eta_{\alpha\beta}
\quad \text{with}\quad \Href_{\alpha\beta} := \Phisans_{\alpha}^{\mu}\Phisans_{\beta}^{\nu}\Hreft_{\mu\nu}.
\end{equation} 
Here, $\Href_{\alpha\beta}$ are the components of the tensor $\Hreft_{\alpha\beta}\diff \xsans^{\alpha}\otimes \diff \xsans^{\beta}$ in the coordinates $\{\xavec^{\alpha}\}$.  It is important to point out that, generally speaking, $\Hreff_{\alpha\beta}\neq \Href_{\alpha\beta}$, but these two objects are close to each other. In fact we have the following estimates.

\begin{lemma}[Comparing metric perturbations]
\label{lem1-14-sept-2025}
Assume that $\Kref\epsm$ is sufficiently small. Then one has 
\begin{equation}
\big|\Href_{\alpha\beta} - \Hreff_{\alpha\beta}\big|^{\Kscr}_{p,k}\lesssim_p \Kref\epsm\,(t+r+1)^{-1-(p-k)+\theta}.
\end{equation}
\end{lemma}

\begin{proof} By a direct calculation, we find 
\be
\aligned
\Phisans_{\alpha}^{\mu}\Phisans_{\beta}^{\nu}\etasans_{\mu\nu} - \eta_{\alpha\beta}
& = (\delta_{\alpha}^{\mu}+\Theta_{\alpha}^{\mu})(\delta_{\beta}^{\nu}+\Theta_{\beta}^{\nu})\etasans_{\mu\nu}- \eta_{\alpha\beta}
\\
& = \Theta_{\alpha}^{\mu}\,\etasans_{\mu\beta} 
+ \Theta_{\beta}^{\nu}\,\etasans_{\alpha\nu}
+ \Theta_{\alpha}^{\mu}\Theta_{\beta}^{\nu}\,\etasans_{\mu\nu}.
\endaligned
\ee
Applying \eqref{eq1-12-sept-2025} ---and absorbing the quadratic term, which is $O\big((\Kref\epsm)^2(t+r+1)^{-2-(p-k)+2\theta}\big)$, into the stated bound for $\Kref\epsm$ sufficiently small)--- yields the desired estimate.
\end{proof}


\begin{lemma}[Asymptotic control of the reference perturbation in the new chart]\label{lem1-21-sept-2025}
Let $\gref$ be an $(N,\theta,\epss,\ell)$--admissible reference. Then, provided $\Kref\epsm$ is sufficiently small, one has 
\begin{equation}\label{eq6-14-sept-2025}
|\Href_{\alpha\beta}|^{\Kscr}_{p,k}\lesssim \epss\big(1+r^{k\theta}(1+|r-t|)^{-k}\big)(1+|r-t|)^{k-p}(t+r+1)^{-1+\theta},\quad p\leq N+2,
\end{equation}
\begin{equation}\label{eq6-21-sept-2025}
\aligned
|\Hreff_{\alpha\beta}|^{\Kscr}_{p,k}& \lesssim  \Kref\epsm(1+t+r)^{-1-(p-k)+\theta}
\\
& \quad +\epss\big(1+r^{k\theta}(1+|r-t|)^{-k}\big)(1+|r-t|)^{k-p}(t+r+1)^{-1+\theta},\quad p\leq N+2,
\endaligned
\end{equation}
\begin{equation}\label{eq7-14-sept-2025}
|\del{\slashed{H}_{\mathrm{r}}}_{\alpha\beta}|\lesssim \epss(t+r+1)^{-1}.
\end{equation}
\end{lemma}

\begin{proof}
\bse
The estimate \eqref{eq6-14-sept-2025} follows directly from \eqref{eq5-14-sept-2025} together with Corollaries~\ref{cor1-14-sept-2025} and \ref{cor2-21-sept-2025}. The estimate \eqref{eq6-21-sept-2025} is based on Lemma~\ref{lem1-14-sept-2025} and \eqref{eq6-14-sept-2025}.
To deal with \eqref{eq7-14-sept-2025}, we observe that 
\be
\aligned
\Href(\del_t,\del^{\Ncal}_a) & =  \frac{x^a}{r}\Href_{00} + \Href_{0a}
= \frac{x^a}{r}\Phisans_a^{\alpha}\Phisans_0^{\beta}\Hreft_{\alpha\beta} + \Phisans_0^{\alpha}\Phisans_a^{\beta}\Hreft_{\alpha\beta},
\\
\Href(\del^{\Ncal}_a,\del^{\Ncal}_b) & = \frac{x^ax^b}{r^2}\Href_{00} + \frac{2x^a}{r}\Href_{a0} 
+ \Href_{ab}
\\
& =\frac{x^ax^b}{r^2}\Phisans_0^{\alpha}\Phisans_0^{\beta}\Hreft_{\alpha\beta} 
+ \frac{2x^a}{r}\Phisans_a^{\alpha}\Phisans_0^{\beta}\Hreft_{\alpha\beta}
+ \Phisans_a^{\alpha}\Phisans_b^{\beta}\Hreft_{\alpha\beta}.
\endaligned
\ee
We also have 
\begin{equation}\label{eq1-27-sept-2025}
\del^{\Ncal}_a 
= \delsans^{\Ncal}_a 
+ \big(\Theta_a^{\beta} + (\xsans^a/\rsans)\Theta_0^{\beta}\big)\delsans_{\beta},
\quad
\del_t = \delsans_t + \Theta_0^{\beta}\delsans_{\beta}.
\end{equation}
Thus it follows that 
\be
\aligned
\Href(\del_t,\del^{\Ncal}_a) & = \Href(\delsans_t,\delsans^{\Ncal}_a) 
+ \Theta_0^{\beta}\Href(\delsans_{\beta},\delsans^{\Ncal}_a)
+ \big(\Theta_a^{\gamma} + (\xsans^a/\rsans)\Theta_0^{\gamma}\big)\Href(\delsans_t,\delsans_\gamma)
\\
& \quad +\Theta_0^{\beta}\big(\Theta_a^{\gamma} +(\xsans^a/\rsans)\Theta_0^{\gamma}\big)\Href(\delsans_{\beta},\delsans_{\gamma}),
\\
\Href(\del^{\Ncal}_a,\del^{\Ncal}_b) & =  \Href(\delsans^{\Ncal}_a,\delsans^{\Ncal}_b)
+ \big(\Theta_b^{\gamma} + (\xsans^b/\rsans)\Theta_0^{\gamma}\big)
\Href(\delsans^{\Ncal}_a,\delsans_{\gamma})
+ \big(\Theta_a^{\beta} + (\xsans^a/\rsans)\Theta_0^{\beta}\big)
\Href(\delsans_{\beta},\delsans^{\Ncal}_c)
\\
& \quad +\big(\Theta_a^{\beta} + (\xsans^a/\rsans)\Theta_0^{\beta}\big)
\big(\Theta_b^{\gamma} + (\xsans^b/\rsans)\Theta_0^{\gamma}\big)
\Href(\delsans_{\beta},\delsans_{\gamma}).
\endaligned
\ee
In other words, we have
\be
{\slashed{H}_{\mathrm{r}}}_{\alpha\beta} = {\slashed{\Hsans}_{\mathrm{r}}}_{\alpha\beta} + A_{\alpha\beta}^{\gamma\delta}\Hsans_{\gamma\delta}, 
\ee 
where $A_{\alpha\beta}^{\gamma\delta}$ are homogeneous functions of degree $(-1+\theta)$ and/or $(-2+2\theta)$ in the cone $\{\rsans\geq 5\tsans/8\}$. 
\ese
\bse
On the other hand, we recall \eqref{eq1-27-sept-2025} and obtain 
\begin{equation}\label{eq2-27-sept-2025}
\del_{\delta}{\slashed{H}_{\mathrm{r}}}_{\alpha\beta} 
= \delsans_{\delta}{\slashed{\Hsans}_{\mathrm{r}}}_{\alpha\beta} 
+ B_{\delta\alpha\beta}^{\lambda\mu\nu}\delsans_{\lambda}\Hsans_{\mu\nu}
+ C_{\delta\alpha\beta}^{\mu\nu}\Hsans_{\mu\nu}. 
\end{equation}
Here, $B_{\delta\alpha\beta}^{\lambda\mu\nu}$ are finite linear combinations of homogeneous functions of degree not exceeding $(-1+\theta)$ with constant coefficients, and $C_{\delta\alpha\beta}^{\mu\nu}$ are finite linear combinations of homogeneous functions of degree not exceeding $(-2+\theta)$ with constant coefficients. Then, in \eqref{eq2-27-sept-2025}, we bound the leading term by the second estimate in \eqref{eq2-14-sept-2025}, while the remaining terms are easily controlled via \eqref{eq2-14-sept-2025} combined with the homogeneity of $B_{\delta\alpha\beta}^{\lambda\mu\nu}$ and $C_{\delta\alpha\beta}^{\mu\nu}$.
\ese
\end{proof}


In practical applications, we may also use the following weaker version.

\begin{corollary}[Coarse bounds under small $\theta$]\label{cor1-23-sept-2025}
Let $\gref$ be an $(N,\theta,\epss,\ell)$--admissible reference and assume, in addition, that $(N+3)\theta\leq 1$. Then, provided $\Kref\epsm$ is sufficiently small, one has 
\begin{subequations}\label{eqs2-23-sept-2025}
\begin{equation}\label{eq4-23-sept-2025}
|\Hreff_{\alpha\beta}\big|^{\Kscr}_p + \big|\Hreff^{\alpha\beta}\big|^{\Kscr}_p
\lesssim_N \Kref\epsm(t+r+1)^{-1+(p+1)\theta}
\lesssim \Kref\epsm,\quad p\leq N+2.
\end{equation}
\begin{equation}\label{eq5-23-sept-2025}
\aligned
\big|\del_{\mu}\gref^{\alpha\beta}\big|^{\Kscr}_p 
 + \big|\del_{\nu}\gref_{\alpha\beta}\big|^{\Kscr}_p
& \lesssim_N  \Kref \epsm(1+|r-t|)^{-1}(t+r+1)^{-1+(p+1)\theta}
\\
& \lesssim  \Kref\epsm(1+|r-t|)^{-1},\quad p\leq N+1,
\endaligned
\end{equation}
\begin{equation}
|\del_{\mu}\del_{\nu}\gref_{\alpha\beta}|^{\Kscr}_p
+ |\del_{\mu}\del_{\nu}\gref^{\alpha\beta}|^{\Kscr}_p
\lesssim 
\Kref \epsm(1+|r-t|)^{-2}(t+r+1)^{-1+(p+1)\theta},\quad p\leq N,
\end{equation}
\begin{equation}
|\del_{\mu}\dels_{\nu}\gref_{\alpha\beta}|^{\Kscr}_p
+ |\del_{\mu}\dels_{\nu}\gref^{\alpha\beta}|^{\Kscr}_p
\lesssim \Kref \epsm(1+|r-t|)^{-1}(t+r+1)^{-2+(p+1)\theta},\quad p\leq N. 
\end{equation}
Here, $\dels_{\nu}$ denotes the ``good'' derivatives, that is, 
\be
\aligned
\dels_{\nu} & = \delu_a,\quad a=1,2,3 \quad & \text{ in the domain } \Mcal^{\Hcal}_{[s_0,+\infty)}, 
\\
\dels_{\nu} & = \del^{\Ncal}_a,\quad a=1,2,3 \quad & \text{ in the domain } \Mcal^{\ME}_{[s_0,+\infty)}. 
\endaligned
\ee
\end{subequations}
\end{corollary}

\begin{proof} We apply \eqref{eq6-21-sept-2025} together with the condition $(N+3)\theta\leq 1$. For the last two estimates we apply \eqref{eq3-01-oct-2025}.
\end{proof}


\paragraph{Generalized wave gauge condition.} 

Next, we need to estimate the contracted Christoffel symbols ${\Gamma_{\mathrm{r}}}^{\lambda} := \gref^{\alpha\beta}{\Gamma_{\mathrm{r}}}_{\alpha\beta}^{\lambda}$. 

\begin{proposition}\label{lem1-26-sept-2025}
Let $\gref$ be an $(N,\theta,\epss,\ell)$--admissible reference with $(N+3)\theta\leq 1$. Then, provided $\Kref\epsm$ is sufficiently small, one has 
\begin{equation}\label{eq3-22-sept-2025}
\big|{\Gamma_{\mathrm{r}}}^{\lambda}\big|^{\Kscr}_{N+1}\lesssim \Kref\epsm(t+r+1)^{-2+\theta},
\end{equation}
\begin{equation}\label{eq4-22-sept-2025}
\big|\del{\Gamma_{\mathrm{r}}}^{\lambda}\big|^{\Kscr}_N\lesssim \Kref\epsm(1+|r-t|)^{-1}(t+r+1)^{-2+\theta}.
\end{equation}
\end{proposition}

\begin{proof}
\bse
We recall the following relation, essentially Eq.~(5.5) in \cite{KauffmanLindblad}, namely in our notation:
\begin{equation}\label{eq4-21-sept-2025}
\gref^{\alpha\beta}{\Gamma_{\mathrm{r}}}_{\alpha\beta}^{\lambda}
=\gref^{\alpha\beta}\del_{\alpha}\big(\Phisans_{\beta}^{\delta}\big)\Phisans_{\delta}^{\lambda} + \greft^{\mu\nu}\widetilde{\Gamma_{\mathrm{r}}}_{\mu\nu}^{\delta}\Phisans_{\delta}^{\lambda}
= \gref^{\alpha\beta}\del_{\alpha}\big(\Phisans_{\beta}^{\delta}\big)\Phisans_{\delta}^{\lambda} + \widetilde{\Gamma_{\mathrm{r}}}^{\delta}\Phisans_{\delta}^{\lambda}. 
\end{equation}
To derive it, we observe that 
\be
\del_{\alpha}\del_{\beta} = \Phisans_{\alpha}^{\mu}\Phisans_{\beta}^{\nu}\delsans_{\mu}\delsans_{\nu} + \del_{\alpha}(\Phisans_{\beta}^{\delta})\delsans_{\delta}, 
\ee
which leads us to
\be
g^{\alpha\beta}\del_{\alpha}\del_{\beta} 
= \gsans^{\mu\nu}\delsans_{\mu}\delsans_{\nu} 
+ g^{\alpha\beta}\del_{\alpha}(\Phisans_{\beta}^{\delta})\delsans_{\delta}.
\ee
\ese

On the other hand, by the invariance of the wave operator, we have 
\be
\gref^{\alpha\beta}\del_{\alpha}\del_{\beta} 
- \gref^{\alpha\beta}{\Gamma_\mathrm{r}}_{\alpha\beta}^{\lambda}\del_{\lambda} 
= 
\greft^{\alpha\beta}\delsans_{\alpha}\delsans_{\beta} 
- \greft^{\alpha\beta}\widetilde{\Gamma_{\mathrm{r}}}_{\alpha\beta}^{\lambda} \,\delsans_{\lambda}.
\ee
We thus obtain \eqref{eq4-21-sept-2025}.

\bse
Next, we recall Corollary~\ref{cor1-14-sept-2025}, and thus 
\begin{equation}\label{eq1-22-sept-2025}
\big|\Kscr^I\del_{\alpha}\big(\Phisans_{\beta}^{\delta}\big)\big|\lesssim_N \Kref\epsm (t+r+1)^{-2+\theta}.
\end{equation}
Thanks to \eqref{eq8-21-sept-2025} and Lemma~\ref{lem1-13-sept-2025}, we also have 
\begin{equation}\label{eq2-22-sept-2025}
\big|\Kscr^I\widetilde{\Gamma_{\mathrm{r}}}^{\lambda}\big|
\lesssim_N \epss(t+r+1)^{-4+(N+2)\theta} \leq \epss(t+r+1)^{-3},
\quad \ord(I)\leq N+1,
\end{equation}
where the last inequality follows from $(N+3)\theta\leq 1$. Substituting \eqref{eq1-22-sept-2025} and \eqref{eq2-22-sept-2025} together with \eqref{eq4-23-sept-2025} into \eqref{eq4-21-sept-2025}, we obtain \eqref{eq3-22-sept-2025}.
\ese
\bse
For \eqref{eq4-22-sept-2025}, using again Corollary~\ref{cor1-14-sept-2025}, we have 
\begin{equation}
\big|\Kscr^I\del_{\alpha}\del_{\mu}\big(\Phisans_{\beta}^{\delta}\big)\big|\lesssim_N \Kref\epsm (1+t+r)^{-3+\theta},\quad \ord(I)\leq N.
\end{equation}
Then we substitute this estimate together with \eqref{eq2-22-sept-2025} (applied to $\Kscr^I\del_{\alpha}\widetilde{\Gamma_{\mathrm{r}}}^{\lambda}$ with $|I|\leq N$) and \eqref{eq5-23-sept-2025} into
\be
\aligned
\del_{\gamma}\big(\gref^{\alpha\beta}{\Gamma_{\mathrm{r}}}_{\alpha\beta}^{\lambda}\big)
& =\del_{\gamma}\gref^{\alpha\beta}\del_{\alpha}\big(\Phisans_{\beta}^{\delta}\big)\Phisans_{\delta}^{\lambda} 
+\gref^{\alpha\beta}\del_{\gamma}\del_{\alpha}\big(\Phisans_{\beta}^{\delta}\big)\Phisans_{\delta}^{\lambda} 
+\gref^{\alpha\beta}\del_{\alpha}\big(\Phisans_{\beta}^{\delta}\big)\del_{\gamma}\big(\Phisans_{\delta}^{\lambda} \big)
\\
& \quad + \del_{\gamma}\big(\widetilde{\Gamma_{\mathrm{r}}}^{\delta}\big)\Phisans_{\delta}^{\lambda}
+\widetilde{\Gamma_{\mathrm{r}}}^{\delta}
\del_{\gamma}\big(\Phisans_{\delta}^{\lambda}\big).
\endaligned
\ee
We thus conclude that \eqref{eq4-22-sept-2025} holds. 
\ese
\end{proof}


\paragraph{Conclusion.}

We arrive at a complete description of the new coordinates.

\begin{proposition}[Existence of a light-bending coordinate chart]
\label{prop1-01-oct-2025}
Consider a spacetime $(\Mcal,\gref)$ where $\Mcal$ is diffeomorphic to $\RR^{1+3}_+ = \{(t,x)\,|\,x\in\RR^3,\ t\geq 0\}$.  Assume that there is a global coordinate chart $\{\xsans^{\alpha}\}$ such that $\gref$ is an $(N,\theta,\epss,\ell)$--admissible reference metric with $(N+3)\theta\leq 1$. Then there exist two universal constants $K_0,K_1$ and some $\eps_0>0$ such that for any $0<\epss\leq\epsm\leq\eps_0$ and any $C_L$ satisfying 
\be
0\leq C_L\leq \frac{1}{K_1\epsm} - K_0,
\ee
so that in the new coordinate chart $\{x^{\alpha}\}$ defined by
\begin{equation}
t = \tsans,\quad x^a = (r/\rsans)\xsans^a
\end{equation}
with
\begin{equation}
r = \rsans - \Kref\epsm\chi(\rsans/\tsans)\rsans^{\theta},\quad \Kref = (K_0+C_L)\theta^{-1}, 
\end{equation}
where $\chi$ is the cut-off function defined in \eqref{eq8-24-sept-2025}, the reference metric $\gref$ satisfies the following properties, provided that $\Kref\epsm$ is sufficiently small. 

\begin{subequations}
\noindent $\bullet$ Almost Minkowski condition:
\begin{equation}\label{eq-Mink-conditions}
\aligned
|\Hreff_{\alpha\beta}\big|^{\Kscr}_p + \big|\Hreff^{\alpha\beta}\big|^{\Kscr}_p
& \lesssim_N  \Kref\epsm(t+r+1)^{-1+(p+1)\theta},\quad 
&& p\leq N+2,
\\
\big|\del_{\mu}\gref^{\alpha\beta}\big|^{\Kscr}_p
+ \big|\del_{\mu}\gref_{\alpha\beta}\big|^{\Kscr}_p
& \lesssim_N  \Kref \epsm(1+|r-t|)^{-1}(t+r+1)^{-1+(p+1)\theta},\quad 
&& p\leq N+1,
\\
\big|\dels_{\mu}\gref^{\alpha\beta}\big|^{\Kscr}_p
 + \big|\dels_{\mu}\gref_{\alpha\beta}\big|^{\Kscr}_p
& \lesssim_N  \Kref\epsm (t+r+1)^{-2+(p+1)\theta},\quad 
&& p\leq N+1,
\\
\big|\del_{\mu}\del_{\nu}\gref^{\alpha\beta}\big|^{\Kscr}_p
+ 
 \big|\del_{\mu}\del_{\nu}\gref_{\alpha\beta}\big|^{\Kscr}_p
& \lesssim_N  \Kref \epsm(1+|r-t|)^{-2}(t+r+1)^{-1+(p+1)\theta},\quad 
&& p\leq N,
\\
\big|\del_{\mu}\dels_{\nu}\gref^{\alpha\beta}\big|^{\Kscr}_p
+ \big|\del_{\mu}\dels_{\nu}\gref_{\alpha\beta}\big|^{\Kscr}_p
& \lesssim_N  \Kref \epsm(1+|r-t|)^{-1}(t+r+1)^{-2+(p+1)\theta},\quad 
&& p\leq N.
\endaligned
\end{equation}

\noindent$\bullet$ Almost Ricci-flat condition:
\begin{equation}\label{eq-Ricci-condition}
\big|R[\gref]_{\alpha\beta}\big|_{N} + (1+r+t)\big|\del_{\gamma}R[\gref]_{\alpha\beta}\big|_{N-1}\lesssim_N \epss^2(1+t+r)^{-4+\theta}. 
\end{equation}

\noindent $\bullet$ Generalized wave coordinate condition:
\begin{equation}\label{eq-wave-condition}
\big|{\Gamma_{\mathrm{r}}}^{\lambda}\big|_{N+1} 
+ (1+|r-t|)\big|\del_{\gamma}{\Gamma_{\mathrm{r}}}^{\lambda}\big|_N\lesssim_N \Kref\epsm(t+r+1)^{-2+\theta};
\end{equation}
\noindent$\bullet$ Uniform light-bending condition:
\begin{equation}\label{eq-lb-condition}
\Hreff^{\Ncal 00} := (\diff \tavec{} - \diff \ravec{}, \diff \tavec{} - \diff \ravec{})_{\gref} 
< -  C_L\epsm r^{-1+\theta}
\qquad \text{ in } \Mcal_{[s_0,s_1]} \cap\{ r \geq (3/4) t \}.
\end{equation}
\end{subequations}
\end{proposition}

\begin{remark}
Interestingly, the generalized wave coordinate condition \eqref{eq-wave-condition} introduced here is weaker than the assumptions imposed in \cite{PLF-YM-PDE} (for Class~B reference metrics), and it is also weaker than the original bounds in \eqref{eq8-21-sept-2025}. This seems surprising at first; however, in the next section, by exploiting a \emph{better} formulation of the Einstein equations, we show that the coupling between the components of the wave coordinate condition is much weaker.
\end{remark}

}


\subsection{ Formulation of Einstein equations in the new coordinates}

{ 

\paragraph{Generalized wave coordinate condition.}

We now consider an admissible reference $\gref$ of the form defined in Section~\ref{subsec1-20-sept-2025} in the coordinates $\{\xsans^{\alpha}\}$. Then we seek a solution $g$ to the Einstein equations \eqref{eq4-13-sept-2025} which satisfies the generalized wave coordinate condition in these coordinates $\{\xsans^{\alpha}\}$:
\begin{equation}\label{eq5-24-sept-2025}
\Gammasans^{\lambda} = \widetilde{\Gamma_{\mathrm{r}}}^{\lambda} = \Wsans^{\lambda}.
\end{equation} 
We recall (see, e.g., the textbook~\cite[Section~6.8.1, Lemmas~8.1--8.2]{YCB}) that if the generalized wave coordinate condition \eqref{eq5-24-sept-2025} is satisfied on an initial spacelike hypersurface by a sufficiently regular solution to \eqref{eq4-13-sept-2025}, then this condition is preserved for as long as the (sufficiently regular) solution exists.

First, in the new coordinate chart $\{x^{\alpha}\}$, we have 
\begin{equation}\label{eq4-20-sept-2025}
\Rbb^{\mathrm{w}}(g,\del g,\del\del g)_{\alpha\beta} 
= \Big(T_{\alpha\beta} - \frac{1}{2}\mathrm{Tr}_g(T)g_{\alpha\beta}\Big) 
- \mathbb{G}(g,\del g,\Gamma,\del\Gamma)_{\alpha\beta}.
\end{equation} 
Second, we claim that the reference metric $\gref$ satisfies 
\begin{equation}\label{eq5-20-sept-2025}
\Rbb^{\mathrm{w}}(\gref,\del\gref,\del\del\gref)_{\alpha\beta} = \Phisans_{\alpha}^{\mu}\Phisans_{\beta}^{\nu}\widetilde{U_{\mathrm{r}}}_{\mu\nu} 
- \mathbb{G}(\gref,\del \gref,\Gamma_{\mathrm{r}},\del\Gamma_{\mathrm{r}})_{\alpha\beta}, 
\end{equation}
since, in the coordinates $\{\xsans^{\alpha}\}$, we assume that $\gref$ satisfies the almost Ricci-flat condition. Indeed, we first note that
\begin{equation}
\Rbb(\greft,\delsans\greft,\delsans\delsans\greft,\Gammasans,\delsans\Gammasans)_{\alpha\beta}
= \widetilde{U_{\mathrm{r}}}_{\alpha\beta}.
\end{equation}
On the other hand, by the tensoriality of the Ricci curvature, for any metric $g$ we can write 
\begin{equation}
\Rbb(g,\del g,\del\del g, \Gamma,\del\Gamma)_{\alpha\beta}\diff x^{\alpha}\otimes\diff x^{\beta} = \Rbb(\gsans,\delsans\gsans,\delsans\delsans\gsans, \Gammasans,\delsans\Gammasans)_{\mu\nu}\diff \xsans^{\mu}\otimes\diff \xsans^{\nu},
\end{equation}
which, in our notation, implies that 
\begin{equation}
\Rbb(g,\del g,\del\del g, \Gamma,\del\Gamma)_{\alpha\beta} 
= \Phisans_{\alpha}^{\mu}\Phisans_{\beta}^{\nu}
\Rbb(\gsans,\delsans\gsans,\delsans\delsans\gsans, \Gammasans,\delsans\Gammasans)_{\mu\nu}, 
\end{equation}
which establishes \eqref{eq5-20-sept-2025}.

Next, comparing \eqref{eq4-20-sept-2025} and \eqref{eq5-20-sept-2025}, we conclude that 
\begin{equation}\label{eq6-20-sept-2025}
\aligned
& \Rbb^{\mathrm{w}}(g,\del g,\del\del g)_{\alpha\beta} - \Rbb^{\mathrm{w}}(\gref,\del\gref,\del\del\gref)_{\alpha\beta} 
\\
& =\Big(T_{\alpha\beta} - \frac{1}{2}\mathrm{Tr}_g(T)g_{\alpha\beta}\Big) 
- \Phisans_{\alpha}^{\mu}\Phisans_{\beta}^{\nu}\widetilde{U_{\mathrm{r}}}_{\mu\nu} 
- \big(\mathbb{G}(g,\del g,\Gamma,\del\Gamma)_{\alpha\beta} - \mathbb{G}(\gref,\del \gref,\Gamma_{\mathrm{r}},\del\Gamma_{\mathrm{r}})_{\alpha\beta}\big).
\endaligned
\end{equation}
At this juncture, we emphasize that, in \eqref{eq6-20-sept-2025}, rather than working with the single term $\mathbb{G}(g,\del g,\Gamma,\del\Gamma)_{\alpha\beta}$, it is preferable to consider the \emph{difference} $\big(\mathbb{G}(g,\del g,\Gamma,\del\Gamma)_{\alpha\beta} - \mathbb{G}(\gref,\del \gref,\Gamma_{\mathrm{r}},\del\Gamma_{\mathrm{r}})_{\alpha\beta}\big)$, which enjoys significantly better decay properties (as shown in the discussion below). This is why we can accommodate a mildly decaying wave–gauge condition \eqref{eq-wave-condition} in the new coordinate chart $\{x^{\alpha}\}$. (For the case $\gref=\Mks$, see also \cite[Appendix~B]{KauffmanLindblad}.) 

\paragraph{Einstein equations formulated in the new coordinates.}

We now write \eqref{eq6-20-sept-2025} as a system of PDEs. To proceed, we introduce
\be
u_{\alpha\beta}:=g_{\alpha\beta}- \gref_{\alpha\beta}.
\ee
Recall that the nonlinear term $\Fbb$ is a multilinear form which is quadratic in the first derivatives and depends (rationally) on the metric. We record it as
\be
\Fbb(g,\del g)=\Fbb(g,g;\del g,\del g),
\ee
where $\Fbb(\cdot,\cdot;\star,\star)$ is linear in each argument and symmetric in the first two and in the last two entries. Then we have 
\begin{equation}
\aligned
\Fbb(g,g;\del g,\del g)- \Fbb(\gref,\gref;\del\gref,\del\gref)
& = \Fbb(g,g;\del u,\del u)
+2\,\Fbb(g,g;\del\gref,\del u)
\\
& \quad
+2\,\Fbb(u,\gref;\del\gref,\del\gref)
+ \Fbb(u,u;\del\gref,\del\gref).
\endaligned
\end{equation}
Consequently, we find 
\begin{equation}
\aligned
& \Rbb^{\mathrm{w}}(g,\del g,\del\del g)_{\alpha\beta}
- \Rbb^{\mathrm{w}}(\gref,\del\gref,\del\del\gref)_{\alpha\beta}
\\
& \qquad
= - \frac{1}{2}g^{\mu\nu}\del_{\mu}\del_{\nu}u_{\alpha\beta}
+\frac{1}{2}\Fbb(g,g;\del u,\del u)
+\frac{1}{2}u^{\mu\nu}\del_{\mu}\del_{\nu}\gref_{\alpha\beta}
\\
& \qquad\quad
+\Fbb(g,g;\del\gref,\del u)
+\Fbb(u,\gref;\del\gref,\del\gref)
+\frac{1}{2}\Fbb(u,u;\del\gref,\del\gref).
\endaligned
\end{equation}
These terms were already analyzed in our previous work~\cite{PLF-YM-PDE}. In the present setting, however, $u^{\alpha\beta}$ measures the correction of $g$ from $\gref$ in the \emph{new} coordinate chart, and we will use the following previous estimate (stated next and proved for completeness).

\begin{lemma}\label{lem1-24-sept-2025}
Consider $u^{\alpha\beta} = g^{\alpha\beta} - \gref^{\alpha\beta}$ where $\gref$ is an $(N,\theta,\epss,\ell)$–admissible reference on $\Mcal$ with $(N+3)\theta\leq 1$ and $\Kref\epsm$ is sufficiently small. Assume, in addition, that
\begin{equation}\label{eq1-23-sept-2025}
|u|_{[N/2]+1}\leq \tfrac12 \quad \text{in } \Mcal_{[s_0,s_1]}.
\end{equation}
Then, for every $(p,k)$ with $p\le N+1$, one has 
\begin{equation}
|u^{\alpha\beta}|_{p,k} \lesssim_N |u|_{p,k}.
\end{equation}
\end{lemma}


\begin{proof} 
\bse
{\bf 1.} 
We define $g = (g_{\alpha\beta})_{\alpha\beta}$, $\gref = (\gref_{\alpha\beta})_{\alpha\beta}$, and $u = \big(u_{\alpha\beta}\big)_{\alpha\beta}$ the matrices formed by the components of the corresponding tensors. Then we have 
\be
\big(u^{\alpha\beta}\big)_{\alpha\beta} = g^{-1}- \gref^{-1},
\ee
and the following identity holds 
\begin{equation}\label{eq2-24-sept-2025}
g^{-1}- \gref^{-1} = \sum_{k=1}^{\infty}(-1)^k(\gref^{-1} u)^k\,\gref^{-1},
\end{equation}
provided the series converges, which follows from the smallness assumption \eqref{eq1-23-sept-2025} (after possibly lowering the smallness threshold so that $\|\gref^{-1}u\|_{[N/2]+1}\le \tfrac12$).

When we differentiate the right-hand side, we have, for each term,
\begin{equation}\label{eq7-23-sept-2025}
\mathscr{Z}^I\big((\gref^{-1} u)^k\gref^{-1}\big) 
= \sum_{I_0\odot I_1\odot\cdots\odot I_k = I} \,  \Big(\prod_{j=1}^k \mathscr{Z}^{I_j}(\gref^{-1}u)\Big)\,\mathscr{Z}^{I_0}(\gref^{-1}).
\end{equation}
Recalling  \eqref{eq4-23-sept-2025} from Corollary~\ref{cor1-23-sept-2025}, for any admissible operator $\mathscr{Z}^I$ of type $(p,k)$ we have 
\begin{equation}
\big|\mathscr{Z}^I(\gref^{-1}u)\big|\lesssim_N |u|_{p,k}.
\end{equation} 
In each product of \eqref{eq7-23-sept-2025}, choose one distinguished factor to bear all high-order derivatives (contributing $\lesssim_N |u|_{p,k}$) and bound all other $k-1$ factors in the low norm by \eqref{eq1-23-sept-2025}:
\be
\|\gref^{-1}u\|_{[N/2]+1}\le \tfrac12.
\ee
Hence every term in \eqref{eq7-23-sept-2025} is bounded by $2^{-(k-1)}|u|_{p,k}$ up to a constant depending only on $N$.
\ese

\vskip.15cm

{\bf 2.} 
To guarantee the convergence, we need to count the number of terms in the right-hand side of \eqref{eq7-23-sept-2025}. We claim that this number, say $C_p(k+1)$, is controlled by a polynomial of degree $\leq N$  acting on $k$ with coefficients determined by $N$. This combinatorial problem can be solved as follows.

We distribute the integers $A_p = \{1,2,\cdots,p\}$ (corresponding to the vectors contained in $\mathscr{Z}^I$) into $(k+1)$ baskets (corresponding to the factors). Each basket can receive more than one integer, and can also be left empty. One manner of distribution corresponds to one term in the right-hand side of \eqref{eq7-23-sept-2025} (of course there will be identical terms that we can collect together, but the coefficients count). This number $C_p(k)$ is strictly increasing with respect to $p$ because for the case $p+1$, one can merge two integers as one (always distribute them together). Thus we only need to control $C_N(k)$. Remark that for a fixed $1\leq j\leq N$, we can decompose $A_N = \{1,2,\cdots,N\}$ into $j$ non-empty subsets. The number of different decompositions is denoted by $D_N(j)$ which is determined by $N$ and $j$. Then we distribute these $j$ subsets (different from each other) into $(k+1)$ baskets, and there are 
\be
(k\!+\!1)k(k\!- \!1)\cdots(k\!- \!j\!+\!2)  =  \frac{(k\!+\!1)!}{(k\!+\!1\!- \!j)!}  \leq (k\!+\!1)^j
\ee
manners of distribution. Then in total, we have 
\begin{equation}\label{eq1-24-sept-2025}
C_N(k)\leq (k+1) + \sum_{j=2}^ND_N(j)(k+1)^j. 
\end{equation}

Next, based on \eqref{eq7-23-sept-2025} and \eqref{eq1-24-sept-2025}, we see that 
\be
\big|\mathscr{Z}^I\big((\gref^{-1}u)^k\gref^{-1}\big)\big|
\leq \frac{C_N(k+1)}{2^{\,k-1}} \,|u|_{p,k}.
\ee
This leads us to the absolute and uniform convergence of  
\begin{equation}
\sum_{k=1}^{\infty}\mathscr{Z}^I\big((\gref^{-1}u)^k\gref^{-1}\big)
\end{equation}
and 
\be
\Big|\sum_{k=1}^{\infty}\mathscr{Z}^I\big((\gref^{-1}u)^k\gref^{-1}\big)\Big|
\leq |u|_{p,k}\sum_{k=1}^{\infty} \frac{C_N(k+1)}{2^{\,k-1}} 
\lesssim_N \, |u|_{p,k}.
\ee
Thus we can differentiate \eqref{eq2-24-sept-2025} term-by-term, and obtain the desired estimate.
\end{proof}


Then we turn our attention to the terms defining $\mathbb{G}$, which are entirely new ones in comparison to our previous work for the Einstein--massive scalar field system. We need the following estimate.

\begin{lemma}
\label{lem2-24-sept-2025}
Let $\gref$ be an $(N,\theta,\epss,\ell)$--admissible reference where $(N+3)\theta\leq 1$ and $\Kref\epsm$ is sufficiently small. Assume, in addition, that \eqref{eq5-24-sept-2025} holds, together with 
\begin{equation}
|u|_{[N/2]+1}\leq 1/2 \quad \text{ in } \Mcal_{[s_0,s_1]}.
\end{equation} 
Then one has 
\begin{equation}\label{eq6-24-sept-2025}
\big|\Gamma^{\lambda} - {\Gamma_{\mathrm{r}}}^{\lambda}\big|_p\lesssim_N \Kref\epsm(t+r+1)^{-2+\theta}|u|_p,\qquad p\leq N+1,
\end{equation}
\begin{equation}\label{eq7-24-sept-2025}
\big|\del_{\alpha}\big(\Gamma^{\lambda} - {\Gamma_{\mathrm{r}}}^{\lambda}\big)\big|_p
\lesssim_N \Kref\epsm(t+r+1)^{-2+\theta}|\del u|_p 
+ \Kref\epsm(t+r+1)^{-3+\theta}|u|_p,\quad p\leq N.
\end{equation}
Here, one recalls that $|u|_p:=\max_{\alpha,\beta}|u_{\alpha\beta}|_p$. 
\end{lemma}

\begin{proof} 
\bse
Similar as in \eqref{eq4-21-sept-2025}, we have 
\begin{equation}
\aligned
g^{\alpha\beta}\Gamma_{\alpha\beta}^{\lambda} 
& = g^{\alpha\beta}\del_{\alpha}\big(\Phisans_{\beta}^{\delta}\big)\Phisans_{\delta}^{\lambda} + \gsans^{\mu\nu}\Gammasans_{\mu\nu}^{\delta}\Phisans_{\delta}^{\lambda}.
\endaligned
\end{equation}
Then again in view of \eqref{eq4-21-sept-2025} and \eqref{eq5-24-sept-2025}, we obtain 
\begin{equation}\label{eq2-21-sept-2025}
\Gamma^{\lambda} - {\Gamma_{\mathrm{r}}}^{\lambda} 
= \del_{\alpha}\big(\Phisans_{\beta}^{\delta}\big)\Phisans_{\delta}^{\lambda} u^{\alpha\beta}.
\end{equation}
Here, $u^{\alpha\beta} = g^{\alpha\beta} - \gref^{\alpha\beta}$. For the term $\del_{\alpha}\big(\Phisans_{\beta}^{\delta}\big)$ we apply \eqref{eq1-22-sept-2025} together with Lemma~\ref{lem1-24-sept-2025}, and obtain 
\be
|\Gamma^{\lambda}-{\Gamma_{\mathrm{r}}}^{\lambda}|_p
\lesssim_N \Kref\epsm(t+r+1)^{-2+\theta}\max_{\alpha,\beta}|u^{\alpha\beta}|_p
\lesssim_N \Kref\epsm(t+r+1)^{-2+\theta}|u|_p, 
\ee
which is \eqref{eq6-24-sept-2025}. In the same manner, we write 
\be
\aligned
\big|\del_{\gamma}\big(\Gamma^{\lambda} - {\Gamma_{\mathrm{r}}}^{\lambda}\big)\big|_p
& \lesssim_N 
\Kref\epsm(t+r+1)^{-2+\theta}\max_{\alpha,\beta}|\del_{\gamma}u^{\alpha\beta}|_p 
+ \Kref\epsm(t+r+1)^{-3+\theta}\max_{\alpha,\beta}|u^{\alpha\beta}|_p.
\endaligned
\ee
Then we apply again Lemma~\ref{lem1-24-sept-2025}, and reach \eqref{eq7-24-sept-2025}.
\ese
\end{proof}


Then we have the wave coordinate condition for $g$ in the new coordinates.

\begin{corollary}\label{lem3-24-sept-2025}
Under the assumption of Lemma~\ref{lem1-24-sept-2025}, one has 
\begin{equation}
|\Gamma^{\lambda}|_{p,k}\lesssim_N K_{\theta }\epsm(t+r+1)^{-2+\theta}(1+|u|_{p,k}). 
\end{equation}
\end{corollary}

Then we turn our attention to the calculation of
\begin{equation}\label{eq3-21-sept-2025}
\aligned
&2\mathbb{G}(g,\del g,\Gamma,\del\Gamma)_{\alpha\beta} - 2\mathbb{G}(\gref,\del \gref,\Gamma_{\mathrm{r}},\del\Gamma_{\mathrm{r}})_{\alpha\beta}
\\
& =\del_{\alpha}\big(g_{\beta\lambda}\Gamma^{\lambda} - \gref_{\beta\lambda}{\Gamma_\mathrm{r}}^{\lambda}\big)
+\del_{\beta}\big(g_{\alpha\lambda}\Gamma^{\lambda} - \gref_{\alpha\lambda}{\Gamma_\mathrm{r}}^{\lambda}\big)
\\
& \quad +\Gamma^{\lambda}\del_{\lambda}g_{\alpha\beta} 
- {\Gamma_{\mathrm{r}}}^{\lambda}\del_{\lambda}\gref_{\alpha\beta}
+ g_{\alpha\lambda}g_{\beta\delta}\Gamma^{\lambda}\Gamma^{\delta} - \gref_{\alpha\lambda}\gref_{\beta\delta}{\Gamma_{\mathrm{r}}}^{\lambda}{\Gamma_{\mathrm{r}}}^{\delta}
\\
& =\big(\del_{\alpha}g_{\beta\lambda} + \del_{\beta}g_{\alpha\lambda}\big)\big(\Gamma^{\lambda} - {\Gamma_\mathrm{r}}^{\lambda}\big)
+(g_{\beta\lambda}\del_{\alpha} + g_{\alpha\lambda}\del_{\beta})\big(\Gamma^{\lambda} - {\Gamma_\mathrm{r}}^{\lambda}\big)
\\
& \quad +{\Gamma_{\mathrm{r}}}^{\lambda}\big(\del_{\alpha}u_{\beta\lambda} + \del_{\beta}u_{\alpha\lambda}\big)
+ \big(u_{\beta\lambda}\del_{\alpha} + u_{\alpha\lambda}\del_{\beta}\big){\Gamma_{\mathrm{r}}}^{\lambda}
\\
& \quad + \big(\Gamma^{\lambda} - {\Gamma_\mathrm{r}}^{\lambda}\big)\del_{\lambda}g_{\alpha\beta} 
+ {\Gamma_{\mathrm{r}}}^{\lambda}\del_{\lambda}u_{\alpha\beta}
\\
& \quad + u_{\alpha\lambda}g_{\beta\delta}\Gamma^{\lambda}\Gamma^{\delta}
+ \gref_{\alpha\lambda}u_{\beta\delta}\Gamma^{\lambda}\Gamma^{\delta}
+ \gref_{\alpha\lambda}\gref_{\beta\delta}\big(\Gamma^{\lambda} - {\Gamma_{\mathrm{r}}}^{\lambda}\big)
+ \gref_{\alpha\lambda}\gref_{\beta\delta}{\Gamma_{\mathrm{r}}}^{\lambda}\big(\Gamma^{\delta} - {\Gamma_{\mathrm{r}}}^{\delta}\big).
\endaligned
\end{equation}
For $\mathbb{G}$, we thus arrive at the following estimate.

\begin{proposition}\label{lem4-24-sept-2025}
Let $\gref$ be an $(N,\theta,\epss,\ell)$--admissible reference with $(N+3)\theta\leq 1$ and $\Kref\epsm$ sufficiently small. Assume, in addition, that $g$ is expressed in wave gauge with respect to $\gref$, and
\begin{equation}
|u|_{[N/2]+1}\leq 1/2,\quad \text{ in }  \Mcal_{[s_0,s_1]}.
\end{equation} 
Then one has 
\begin{equation}\label{eq1-04-oct-2025}
\aligned
& \big|\big(\mathbb{G}(g,\del g,\Gamma,\del\Gamma)_{\alpha\beta} - \mathbb{G}(\gref,\del \gref,\Gamma_{\mathrm{r}},\del\Gamma_{\mathrm{r}})_{\alpha\beta}\big)\big|_{p,k}
\\
& \lesssim \Kref \epsm(t+r+1)^{-2+\theta}(1+|t-r|)^{-1}|u|_{p,k} 
+ \Kref\epsm(t+r+1)^{-2+\theta}|\del u|_{p,k}
\\
& \quad + \Kref\epsm(t+r+1)^{-1+(N+1)\theta}(1+|r-t|)^{-1}\sum_{p_1+p_2\leq p\atop k_2+k_2\leq k}|\del u|_{p_1,k_1}|u|_{p_2,k_2}.
\endaligned
\end{equation}
\end{proposition}

\begin{proof}
We need to control each term in the right-hand side of \eqref{eq3-21-sept-2025}. This is done by applying Lemma~\ref{lem2-24-sept-2025}, \eqref{eqs2-23-sept-2025}, and \eqref{eq4-22-sept-2025}. For the calculation, we recall the relation $g_{\alpha\beta} = u_{\alpha\beta} + \gref_{\alpha\beta}$. High-order terms such as $|u^2|_p$ are simplified as follows: 
$$
|u^2|_p\lesssim \sum_{p_1+p_2\leq p}|u|_{p_1}|u|_{p_2}\lesssim |u|_{[N/2]}|u|_p\lesssim |u|_p. \qedhere
$$
\end{proof}


\paragraph{Conclusion.}

In the new coordinate chart $\{x^{\alpha}\}$, the Einstein equations can be formulated as a quasi-linear wave system with the following structure. 

\begin{proposition}\label{prop1-27-sept-2025}
Let $\gref$ be an $(N,\theta,\epss,\ell)$--admissible reference with $(N+3)\theta\leq 1$ and $\Kref\epsm$ sufficiently small. Let $g$ be a solution to \eqref{eq4-13-sept-2025} satisfying the generalized wave coordinate condition
\begin{equation}
\Gammasans^{\lambda} = \widetilde{\Gamma_{\mathrm{r}}}^{\lambda}.
\end{equation}
Then $g$ is a solution to the following quasilinear wave equations in the new coordinates $\{x^{\alpha}\}$:
\begin{equation}\label{eq3-27-sept-2025}
\aligned
g^{\mu\nu}\del_{\mu}\del_{\nu}u_{\alpha\beta} 
& = \Fbb(g,g;\del u,\del u)_{\alpha\beta} 
+ u^{\mu\nu}\del_{\mu}\del_{\nu}\gref_{\alpha\beta} + 2\Fbb(g,g;\del\gref,\del u)_{\alpha\beta}
\\
& \quad + 2\Fbb(u,\gref;\del \gref,\del \gref)_{\alpha\beta} 
+ \Fbb(u,u;\del\gref,\del\gref)_{\alpha\beta}
\\
& \quad -2\big(\mathbb{G}(g,\del g,\Gamma,\del\Gamma)_{\alpha\beta} - \mathbb{G}(\gref,\del \gref,\Gamma_{\mathrm{r}},\del\Gamma_{\mathrm{r}})_{\alpha\beta}\big)
\\
& \quad + \mathrm{Tr}_g(T)g_{\alpha\beta} - 2T_{\alpha\beta}
\endaligned
\end{equation}
in $\RR^{1+3}_+ = \{(t,x)|t\geq 0,x\in\RR^3\}$, satisfying the generalized wave coordinate condition
\begin{equation}\label{eq9-24-sept-2025}
\Gamma^{\lambda}  
= {\Gamma_{\mathrm{r}}}^{\lambda} 
+ \del_{\alpha}\big(\Phisans_{\beta}^{\delta}\big)\Phisans_{\delta}^{\lambda} (g^{\alpha\beta} - \gref^{\alpha\beta}).
\end{equation}
\end{proposition}

We conclude this section with two observations. 

\bei 

\item 
The coordinate condition \eqref{eq9-24-sept-2025} is preserved once it holds on an initial slice; they contain known functions $\Phisans_{\beta}^{\alpha},\gref^{\alpha\beta},\Gamma_{\mathrm{r}}^{\lambda}$ and the metric $g^{\alpha\beta}$, but not the derivatives of the metric $g$ (cf.~also~\cite[Appendix B.0.1]{KauffmanLindblad}).

\item In the right-hand side of \eqref{eq3-27-sept-2025}, the terms  $\Fbb(g,g;\del u,\del u)_{\alpha\beta}$, $u^{\mu\nu}\del_{\mu}\del_{\nu}\gref_{\alpha\beta}$, and $\Fbb(g,g;\del\gref,\del u)_{\alpha\beta}$ have been treated in our previous work \cite{PLF-YM-PDE}. The terms $\mathrm{Tr}_g(T)g_{\alpha\beta} - 2T_{\alpha\beta}$ are interaction terms due to the coupling source, which in the present setup is written as \eqref{equa-81} and will be analyzed in full details in the next section. The new terms $\big(\mathbb{G}(g,\del g,\Gamma,\del\Gamma)_{\alpha\beta} - \mathbb{G}(\gref,\del \gref,\Gamma_{\mathrm{r}},\del\Gamma_{\mathrm{r}})_{\alpha\beta}\big)$ are due to the coordinate transform, and were dealt with in Proposition~\ref{lem4-24-sept-2025}. We are thus ready to run the existence bootstrap mechanism for Einstein equations.

\eei

}


\section{Einstein-Dirac system in generalized wave coordinates}
\label{section=N14}

\subsection{ Derivation of the massive Dirac equation} 

{ 

The Vielbein energy-momentum tensor defined, for instance, in \cite[equation (3.274)]{Ortin} 
is derived through the variation of the action associated with the spinor fields (which are coupled to an external gravitational field) with respect to the Vielbein or frame fields. It takes the form 
\begin{equation}
\label{equa-81}
T[\Psi]_{\mu\nu} 
= 
\frac{\mathrm{i}}{4} \Big(\la\Psi,\del_\mu \cdot\nabla_\nu \Psi \ra_{\ourD}
+ 
\la\Psi,\del_{\nu} \cdot\nabla_{\mu} \Psi \ra_{\ourD} \Big)
- 
\frac{\mathrm{i}}{4} \Big(\la \del_\mu \cdot\nabla_\nu \Psi,\Psi\ra_{\ourD}
+
\la \del_{\nu} \cdot\nabla_{\mu} \Psi,\Psi\ra_{\ourD}
\Big).
\end{equation}
Thus the Einstein equations coupled with a spinor field read
\begin{equation}
G_{\mu\nu} = T_{\mu\nu}.
\end{equation}
It might seem counter-intuitive that the mass parameter of the spinor field does not appear in \eqref{equa-81}. 

\begin{proposition}[Compatibility property for the Einstein-Dirac system]
Consider sufficiently regular spinor fields $\Psi$ defined in an open set of a spacetime $(\mathcal{M},g)$ and satisfying the massive Dirac equation
\begin{equation} \label{eq1-09-march-2025}
\opDirac \Psi + \mathrm{i}M\Psi = 0
\end{equation}
for some \textbf{mass parameter} $M$, taken here to be an \emph{arbitrary} real-valued field. Then, one has the divergence-law
\begin{equation} \label{eq3-09-march-2025}
\nabla^\mu \bigl( T[\Psi]_{\mu\nu} \bigr) = 0.
\end{equation}
\end{proposition}

{\sl For instance,} rather than a constant $M$, we could pick a mass depending upon the spacetime scalar curvature or the norm $\la \Psi,\Psi\ra_{\ourD}$ of te spinor field. 
\be
M = 1+R, \qquad M = 1 + R_{\alpha\beta}R^{\alpha\beta}, \qquad M = 1 + \la \Psi,\Psi\ra_{\ourD}. 
\ee
In the present Monograph, we restrict attention to a \emph{constant and positive mass} $M>0$. 

\begin{proof} 
\bse
By a direct calculation we obtain 
\be
\aligned
\nabla^\mu T[\Psi]_{\mu\nu} 
& =  \frac{\mathrm{i}}{4} \Big(\la \del_\mu \nabla^\mu  \Psi,\nabla_\nu \Psi \ra_{\ourD}
+
\la\Psi,\del_\mu \cdot\nabla^\mu  \nabla_\nu \Psi \ra_{\ourD} 
\Big)
\\
& \quad + \frac{\mathrm{i}}{4} \Big(\underline{\la\nabla^\mu  \Psi,\del_{\nu} \cdot\nabla_{\mu} \Psi  \ra_{\ourD}}
+ \la \Psi,\del_{\nu} \cdot\nabla^\mu  \nabla_{\mu} \Psi\ra_{\ourD} \Big)
\\
& - \frac{\mathrm{i}}{4} \Big(\la\del_\mu \cdot\nabla^\mu  \nabla_\nu \Psi,\Psi \ra_{\ourD}
+\la\del_\mu \cdot\nabla_\nu \Psi,\nabla^\mu  \Psi \ra_{\ourD}
\Big)
\\
& - \frac{\mathrm{i}}{4} \Big(
\la \del_{\nu} \cdot\nabla^\mu  \nabla_{\mu} \Psi,\Psi\ra_{\ourD}
+
\underline{\la \del_{\nu} \cdot\nabla_{\mu} \Psi,\nabla^\mu  \Psi\ra_{\ourD}}
\Big).
\endaligned
\ee
By recalling the notation $\opDirac \Psi = \del^{\mu} \cdot\nabla_{\mu} \Psi$ and $\Box_g\Psi := \nabla^\mu  \nabla_{\mu} \Psi$, we find  
\begin{equation} \label{eq2-09-march-2025}
\aligned
\nabla^\mu T[\Psi]_{\mu\nu}
& = \frac{\mathrm{i}}{4} \Big(\la \opDirac\Psi,\nabla_\nu \Psi\ra_{\ourD}
-
\la \nabla_\nu \Psi,\opDirac\Psi\ra_{\ourD}
\Big) 
+ \frac{\mathrm{i}}{4} \Big(\la \del_{\nu} \cdot\Psi,\Box_g\Psi \ra_{\ourD}
- \la \Box_g\Psi,\del_{\nu} \cdot\Psi \ra_{\ourD}
\Big)
\\
& \quad + \frac{\mathrm{i}}{4} \Big(\la\Psi,\del_\mu \cdot\nabla^\mu  \nabla_\nu \Psi \ra_{\ourD} 
-
\la\del_\mu \cdot\nabla^\mu  \nabla_\nu \Psi,\Psi \ra_{\ourD}
\Big). 
\endaligned
\end{equation}
We focus on the last term and write 
\be
\aligned
\del_\mu \cdot\nabla^\mu  \nabla_\nu \Psi 
& =  g^{\mu\alpha} \del_\alpha \cdot\nabla_{\mu} \nabla_\nu \Psi
= g^{\mu\alpha} \del_\alpha \cdot\nabla_\nu \nabla_{\mu} \Psi + g^{\mu\alpha} \del_\alpha \cdot \opCurve_{\mu\nu} \Psi
\\
& = \nabla_\nu \big(g^{\mu\alpha} \del_\alpha \cdot\nabla_{\mu} \Psi\big) 
+ \frac{1}{2}g^{\mu\alpha}R_{\mu\nu} \del_\alpha \cdot\Psi 
= \nabla_{\nu}(\opDirac \Psi) + \frac{1}{2}g^{\mu\alpha}R_{\mu\nu} \del_\alpha \cdot\Psi, 
\endaligned
\ee
where we applied \eqref{eq2-24-feb-2025}. Consequently, we obtain 
\be
\aligned
& \frac{\mathrm{i}}{4} \Big(\la\Psi,\del_\mu \cdot\nabla^\mu  \nabla_\nu \Psi \ra_{\ourD} 
- \la\del_\mu \cdot\nabla^\mu  \nabla_\nu \Psi,\Psi \ra_{\ourD}
\Big)
=  
\frac{\mathrm{i}}{4} \big( \la \Psi,\nabla_{\nu}(\opDirac\Psi)\ra_{\ourD} 
- \la \nabla_{\nu}(\opDirac\Psi),\Psi\ra_{\ourD} \big)
\\
& =  \frac{\mathrm{i}}{4} \big( \la \Psi,\nabla_{\nu}(\opDirac\Psi + \mathrm{i}M\Psi)\ra_{\ourD} 
- \la \nabla_{\nu}(\opDirac\Psi + \mathrm{i}M\Psi),\Psi\ra_{\ourD} \big)
+ \frac{1}{4} \big(\la \Psi,\nabla_{\nu}(M\Psi)\ra_{\ourD} + \la \nabla_{\nu}(M\Psi),\Psi\ra_{\ourD} \big)
\\
& = \frac{\mathrm{i}}{4} \big( \la \Psi,\nabla_{\nu}(\opDirac\Psi + \mathrm{i}M\Psi)\ra_{\ourD} 
- \la \nabla_{\nu}(\opDirac\Psi + \mathrm{i}M\Psi),\Psi\ra_{\ourD} \big)
\\
& \quad +  \frac{M}{4} \big(\la \Psi,\nabla_\nu \Psi\ra_{\ourD} 
+ \la \nabla_\nu \Psi,\Psi\ra_{\ourD} \big)
+\frac{\del_{\nu}M}{2} \la \Psi,\Psi\ra_{\ourD}.
\endaligned
\ee
\ese
For the second term in the right-hand side of \eqref{eq2-09-march-2025}, we apply Lemma~\ref{lem1-21-feb-2025} and obtain
\bse
\be
\aligned
& \frac{\mathrm{i}}{4} \Big(\la \del_{\nu} \cdot\Psi,\Box_g\Psi \ra_{\ourD}
- \la \Box_g\Psi,\del_{\nu} \cdot\Psi \ra_{\ourD}
\Big) 
 =  \frac{\mathrm{i}}{4} \Big(\la\del_{\nu} \cdot\opDirac^2\Psi,\Psi\ra_{\ourD} - \la\Psi,\del_{\nu} \cdot\opDirac^2\Psi\ra_{\ourD} 
\Big)
\\
& = \frac{\mathrm{i}}{4} \Big(\la \del_{\nu} \cdot\opDirac(\opDirac\Psi + \mathrm{i}M\Psi),\Psi\ra_{\ourD}
-
\la\Psi,\del_{\nu} \cdot\opDirac(\opDirac\Psi + \mathrm{i}M\Psi)\ra_{\ourD}
\Big)
\\
& - \frac{1}{4} \Big(
\la\del_{\nu} \cdot\opDirac(M\Psi),\Psi \ra_{\ourD} + \la \Psi,\del_{\nu} \cdot\opDirac(M\Psi)\ra_{\ourD}
\Big)
\\
& = \frac{\mathrm{i}}{4} \Big(\la \del_{\nu} \cdot\opDirac(\opDirac\Psi + \mathrm{i}M\Psi),\Psi\ra_{\ourD}
-
\la\Psi,\del_{\nu} \cdot\opDirac(\opDirac\Psi + \mathrm{i}M\Psi)\ra_{\ourD}
\Big)
\\
& - \frac{M}{4} \Big(
\la\del_{\nu} \cdot\opDirac\Psi,\Psi \ra_{\ourD} + \la \Psi,\del_{\nu} \cdot\opDirac\Psi\ra_{\ourD}
\Big) 
- \frac{\del_{\nu}M}{2} \la \Psi,\Psi\ra_{\ourD}.
\endaligned
\ee
Here,  for the last equality, we used
\be
\aligned
& \la \del_{\nu} \cdot g^{\alpha\beta} \del_{\beta}M\del_\alpha \cdot\Psi,\Psi \ra_{\ourD}
+\la \Psi,\del_{\nu} \cdot g^{\alpha\beta} \del_{\beta}M\del_\alpha \cdot\Psi\ra_{\ourD}
\\
& =  g^{\alpha\beta} \del_{\beta}M
\big(\la\del_{\nu} \cdot\del_\alpha \Psi,\Psi\ra_{\ourD} 
+ \la\Psi,\del_{\nu} \cdot\del_\alpha \Psi\ra_{\ourD} \big)
= g^{\alpha\beta} \del_{\beta}M\la \Psi,(\del_\alpha \cdot\del_{\nu}+\del_{\nu} \cdot\del_{\alpha})\cdot\Psi\ra_{\ourD}
\\
& = g^{\alpha\beta}g_{\alpha\nu} \del_{\beta}M\la \Psi,\Psi\ra_{\ourD}
= 2\del_{\nu}M \la \Psi,\Psi\ra_{\ourD}.
\endaligned
\ee
For the second term in the right-hand side of \eqref{eq2-09-march-2025}, we thus find 
\be
\aligned
& \frac{\mathrm{i}}{4} \Big(\la \del_{\nu} \cdot\Psi,\Box_g\Psi \ra_{\ourD}
- \la \Box_g\Psi,\del_{\nu} \cdot\Psi \ra_{\ourD}
\Big) 
\\
& = \frac{\mathrm{i}}{4} \Big(\la \del_{\nu} \cdot\opDirac(\opDirac\Psi + \mathrm{i}M\Psi),\Psi\ra_{\ourD}
- \la\Psi,\del_{\nu} \cdot\opDirac(\opDirac\Psi + \mathrm{i}M\Psi)\ra_{\ourD}
\Big)
\\
& \quad - \frac{M}{4} \Big(
\la\del_{\nu} \cdot(\opDirac\Psi + \mathrm{i}M\Psi),\Psi \ra_{\ourD} 
+ \la \Psi,\del_{\nu} \cdot(\opDirac\Psi+\mathrm{i}M\Psi)\ra_{\ourD}
\Big) 
- \frac{\del_{\nu}M}{2} \la \Psi,\Psi\ra_{\ourD}.
\endaligned
\ee
For the first term in the right-hand side of \eqref{eq2-09-march-2025}, we obtain 
\be
\aligned
\frac{\mathrm{i}}{4} \Big(\la \opDirac\Psi,\nabla_\nu \Psi\ra_{\ourD}
-
\la \nabla_\nu \Psi,\opDirac\Psi\ra_{\ourD}
\Big) 
& =  \frac{\mathrm{i}}{4} \Big(\la \opDirac\Psi + \mathrm{i}M\Psi,\nabla_\nu \Psi \ra_{\ourD}
- 
\la \nabla_\nu \Psi,\opDirac\Psi + \mathrm{i}M\Psi\ra_{\ourD} \Big)
\\
& \quad - \frac{M}{4} \big(\la \Psi,\nabla_\nu \Psi\ra_{\ourD} + \la \nabla_\nu \Psi,\Psi\ra_{\ourD} \big). 
\endaligned
\ee
We thus arrive at 
\begin{equation}
\aligned
\nabla^\mu T[\Psi]_{\mu\nu} 
& = 
\frac{\mathrm{i}}{4} \Big(\la \opDirac\Psi + \mathrm{i}M\Psi,\nabla_\nu \Psi \ra_{\ourD}
- 
\la \nabla_\nu \Psi,\opDirac\Psi + \mathrm{i}M\Psi\ra_{\ourD} \Big)
\\
& \quad + \frac{\mathrm{i}}{4} \Big(\la \del_{\nu} \cdot\opDirac(\opDirac\Psi + \mathrm{i}M\Psi),\Psi\ra_{\ourD}
-
\la\Psi,\del_{\nu} \cdot\opDirac(\opDirac\Psi + \mathrm{i}M\Psi)\ra_{\ourD}
\Big)
\\
& \quad - \frac{M}{4} \Big(
\la\del_{\nu} \cdot(\opDirac\Psi + \mathrm{i}M\Psi),\Psi \ra_{\ourD} 
+ \la \Psi,\del_{\nu} \cdot(\opDirac\Psi+\mathrm{i}M\Psi)\ra_{\ourD}
\Big) 
\\
& \quad + \frac{\mathrm{i}}{4} \big( \la \Psi,\nabla_{\nu}(\opDirac\Psi + \mathrm{i}M\Psi)\ra_{\ourD} 
- \la \nabla_{\nu}(\opDirac\Psi + \mathrm{i}M\Psi),\Psi\ra_{\ourD} \big). 
\endaligned
\end{equation}
This establishes \eqref{eq3-09-march-2025}.
\ese
\end{proof}

}


\subsection{ Formulation of the Einstein-Dirac system in PDE form}
\label{subsec1-22-oct-2025}
{ 

Recalling that we always apply the modified derivatives $\widehat{\del}_{\alpha}$ in the analysis for spinor fields, we establish the following identity. Observe that the last two terms in \eqref{eq6-27-sept-2025} are cubic in nature. 

\begin{lemma}\label{lem1-01-oct-2025}
Let $\Psi$ be a sufficiently regular spinor field and 
\be
\widehat{\del}_{\alpha}\Psi =: \nabla_{\alpha}\Psi 
+ \frac{1}{4}g^{\mu\nu}\del_{\mu}\cdot\nabla_{\nu}\del_{\alpha}\cdot\Psi.
\ee
Assume furthermore that \eqref{eq1-09-march-2025} holds. Then one has
\begin{equation}\label{eq6-27-sept-2025}
\aligned
&T_{\alpha\beta} - \frac{1}{2}\mathrm{Tr}_g(T)g_{\alpha\beta}
\\
& =\frac{\mathrm{i}}{4} \Big(\la\Psi,\del_\mu \cdot\widehat{\del_{\nu}} \Psi \ra_{\ourD}
+ 
\la\Psi,\del_{\nu} \cdot\widehat{\del_{\mu}} \Psi \ra_{\ourD} \Big)
- 
\frac{\mathrm{i}}{4} \Big(\la \del_\mu \cdot\widehat{\del_\nu} \Psi,\Psi\ra_{\ourD}
+
\la \del_{\nu} \cdot\widehat{\del_{\mu}} \Psi,\Psi\ra_{\ourD}
\Big)
\\
& \quad - \frac{M}{2}g_{\alpha\beta}\la\Psi,\Psi\ra_{\ourD}
+\frac{\mathrm{i}}{4}g^{\gamma\beta}\del_{\beta}g_{\mu\nu}\la \Psi,\del_{\gamma}\cdot\Psi\ra_{\ourD} 
- \frac{\mathrm{i}}{8}g^{\gamma\beta}(\del_{\mu}g_{\nu\beta} + \del_{\nu}g_{\mu\beta})\la \Psi,\del_{\gamma}\cdot\Psi\ra_{\ourD}.
\endaligned
\end{equation}
\end{lemma}

\begin{proof}
\bse
We perform the following calculation: 
\begin{equation}\label{eq5-27-sept-2025}
\aligned
\mathrm{Tr}_g(T) :& =g^{\mu\nu}T_{\mu\nu} 
= \frac{\mathrm{i}}{4}\big(\la\Psi,\opDirac\Psi\ra_{\ourD} + \la \Psi,\opDirac\Psi\ra_{\ourD}
\big)
- \frac{\mathrm{i}}{4}\big(\la\opDirac\Psi,\Psi \ra_{\ourD} + \la\opDirac\Psi,\Psi \ra_{\ourD}\big)
\\
& =M\la\Psi,\Psi \ra_{\ourD}.
\endaligned
\end{equation}
On the other hand, we have 
\be
\nabla_{\nu}\Psi = \widehat{\del_{\nu}}\Psi 
+ \frac{1}{4}g^{\alpha\beta}\del_{\alpha}\cdot\nabla_{\beta}\del_{\nu}\cdot\Psi
= \widehat{\del_{\nu}}\Psi +  \frac{1}{4}\opDirac\del_{\nu}\cdot\Psi
\ee
and 
\begin{equation}\label{eq4-27-sept-2025}
\aligned
T_{\mu\nu} & =\frac{\mathrm{i}}{4} \Big(\la\Psi,\del_\mu \cdot\nabla_\nu \Psi \ra_{\ourD}
+ 
\la\Psi,\del_{\nu} \cdot\nabla_{\mu} \Psi \ra_{\ourD} \Big)
- 
\frac{\mathrm{i}}{4} \Big(\la \del_\mu \cdot\nabla_\nu \Psi,\Psi\ra_{\ourD}
+
\la \del_{\nu} \cdot\nabla_{\mu} \Psi,\Psi\ra_{\ourD}
\Big)
\\
& = \frac{\mathrm{i}}{4} \Big(\la\Psi,\del_\mu \cdot\widehat{\del_{\nu}} \Psi \ra_{\ourD}
+ 
\la\Psi,\del_{\nu} \cdot\widehat{\del_{\mu}} \Psi \ra_{\ourD} \Big)
- 
\frac{\mathrm{i}}{4} \Big(\la \del_\mu \cdot\widehat{\del_\nu} \Psi,\Psi\ra_{\ourD}
+
\la \del_{\nu} \cdot\widehat{\del_{\mu}} \Psi,\Psi\ra_{\ourD}
\Big)
\\
& \quad +\frac{\mathrm{i}}{16}\la \Psi, \del_{\mu}\cdot\opDirac\del_{\nu}\cdot\Psi\ra_{\ourD}
+\frac{\mathrm{i}}{16}\la\Psi,\del_{\nu}\cdot\opDirac\del_{\mu}\cdot\Psi\ra_{\ourD}
\\
& \quad - \frac{\mathrm{i}}{16}\la \Psi, \opDirac\del_{\nu}\cdot\del_{\mu}\cdot\Psi\ra_{\ourD}
- \frac{\mathrm{i}}{16}\la\Psi,\opDirac\del_{\mu}\cdot\del_{\nu}\cdot\Psi\ra_{\ourD}.
\endaligned
\end{equation}
We note that
\be
\aligned
\del_{\mu}\cdot\opDirac\del_{\nu} - \opDirac\del_{\nu}\cdot\del_{\mu}
& = 
g^{\alpha\beta}\Gamma_{\beta\nu}^{\gamma}
\big(\del_{\mu}\cdot\del_{\alpha}\cdot\del_{\gamma}
- \del_{\alpha}\cdot\del_{\gamma}\cdot\del_{\mu}
\big)
\\
& = g^{\alpha\beta}\Gamma_{\beta\nu}^{\gamma}
\big(-2g_{\alpha\mu}\del_{\gamma}- \del_{\alpha}\cdot\del_{\mu}\cdot\del_{\gamma} - \del_{\alpha}\cdot\del_{\gamma}\cdot\del_{\mu}\big)
\\
& = g^{\alpha\beta}\Gamma_{\beta\nu}^{\gamma}
\big(-2g_{\alpha\mu}\del_{\gamma}+2g_{\mu\gamma}\del_{\alpha}\big)
\\
& = 2g_{\mu\gamma}g^{\alpha\beta}\Gamma_{\beta\nu}^{\gamma}\del_{\alpha} 
- 2\Gamma_{\mu\nu}^{\gamma}\del_{\gamma}.
\endaligned
\ee
Then we deduce that 
\begin{equation}
\del_{\mu}\cdot\opDirac\del_{\nu} - \opDirac\del_{\nu}\cdot\del_{\mu}
=
2g^{\gamma\beta}\big(\del_{\beta}g_{\mu\nu} - \del_{\mu}g_{\nu\beta}\big)\del_{\gamma}, 
\end{equation}
therefore 
\be
\aligned
&2g_{\mu\gamma}g^{\alpha\beta}\Gamma_{\beta\nu}^{\gamma}\del_{\alpha} 
- 2\Gamma_{\mu\nu}^{\gamma}\del_{\gamma}
\\
& =g_{\mu\gamma}g^{\alpha\beta}g^{\gamma\delta}\big(\del_{\beta}g_{\nu\delta} + \del_{\nu}g_{\beta\delta} - \del_{\delta}g_{\beta\nu}\big)\del_{\alpha}
- g^{\gamma\delta}(\del_{\mu}g_{\nu\delta} + \del_{\nu}g_{\mu\delta} - \del_{\delta}g_{\mu\nu})\del_{\gamma}
\\
& = g^{\gamma\beta}\big(\del_{\beta}g_{\mu\nu} + \del_{\nu}g_{\beta\mu} - \del_{\mu}g_{\beta\nu}\big)\del_{\gamma}
-g^{\gamma\beta}(\del_{\mu}g_{\nu\beta} + \del_{\nu}g_{\mu\beta} - \del_{\beta}g_{\mu\nu})\del_{\gamma}
\\
& = 2g^{\gamma\beta}\big(\del_{\beta}g_{\mu\nu} - \del_{\mu}g_{\nu\beta}\big)\del_{\gamma}.
\endaligned
\ee
Recalling \eqref{eq4-27-sept-2025}, we then obtain
\begin{equation}
\aligned
T_{\mu\nu} & =\frac{\mathrm{i}}{4} \Big(\la\Psi,\del_\mu \cdot\widehat{\del_{\nu}} \Psi \ra_{\ourD}
+ 
\la\Psi,\del_{\nu} \cdot\widehat{\del_{\mu}} \Psi \ra_{\ourD} \Big)
- 
\frac{\mathrm{i}}{4} \Big(\la \del_\mu \cdot\widehat{\del_\nu} \Psi,\Psi\ra_{\ourD}
+
\la \del_{\nu} \cdot\widehat{\del_{\mu}} \Psi,\Psi\ra_{\ourD}
\Big)
\\
& \quad +\frac{\mathrm{i}}{4}g^{\gamma\beta}\del_{\beta}g_{\mu\nu}\la \Psi,\del_{\gamma}\cdot\Psi\ra_{\ourD} 
- \frac{\mathrm{i}}{8}g^{\gamma\beta}(\del_{\mu}g_{\nu\beta} + \del_{\nu}g_{\mu\beta})\la \Psi,\del_{\gamma}\cdot\Psi\ra_{\ourD}.
\endaligned
\end{equation}
This combined with \eqref{eq5-27-sept-2025} leads us to the desired result.
\ese
\end{proof}


For simplicity in the notation, we set
\begin{equation}\label{eq11-03-oct-2025}
\aligned
& \Sbb(g,\del g,\Psi,\widehat{\del}\Psi)_{\alpha\beta} 
:= \mathrm{Tr}_g(T)g_{\alpha\beta} - 2T_{\alpha\beta}
\\
& = 
\frac{\mathrm{i}}{2} \Big(\la \del_\mu \cdot\widehat{\del_\nu} \Psi,\Psi\ra_{\ourD}
+
\la \del_{\nu} \cdot\widehat{\del_{\mu}} \Psi,\Psi\ra_{\ourD}
\Big)
- \frac{\mathrm{i}}{2} \Big(\la\Psi,\del_\mu \cdot\widehat{\del_{\nu}} \Psi \ra_{\ourD}
+
\la\Psi,\del_{\nu} \cdot\widehat{\del_{\mu}} \Psi \ra_{\ourD} \Big)
\\
& \quad +Mg_{\alpha\beta}\la\Psi,\Psi\ra_{\ourD}
- \frac{\mathrm{i}}{2}g^{\gamma\beta}\del_{\beta}g_{\mu\nu}\la \Psi,\del_{\gamma}\cdot\Psi\ra_{\ourD} 
+ \frac{\mathrm{i}}{4}g^{\gamma\beta}(\del_{\mu}g_{\nu\beta} + \del_{\nu}g_{\mu\beta})\la \Psi,\del_{\gamma}\cdot\Psi\ra_{\ourD}.
\endaligned
\end{equation}
Based on Proposition~\ref{prop1-27-sept-2025} and Lemma \ref{lem1-01-oct-2025}, we formulate the Cauchy problem associated with the Einstein-Dirac system into the following form.
\begin{equation}\label{eq03-02-oct-2025}
\aligned
g^{\mu\nu}\del_{\mu}\del_{\nu}u_{\alpha\beta} 
& =\Fbb(g,g;\del u,\del u)_{\alpha\beta} 
+ u^{\mu\nu}\del_{\mu}\del_{\nu}\gref_{\alpha\beta} 
+ \Sbb(g,\del g,\Psi,\widehat{\del}\Psi)_{\alpha\beta}
\\
& \quad + 2\Fbb(g,g;\del\gref,\del u)_{\alpha\beta} 
+ 2\Fbb(u,\gref;\del \gref,\del \gref)_{\alpha\beta} 
+ \Fbb(u,u;\del\gref,\del\gref)_{\alpha\beta}
\\
& \quad -2\big(\mathbb{G}(g,\del g,\Gamma,\del\Gamma)_{\alpha\beta} - \mathbb{G}(\gref,\del \gref,\Gamma_{\mathrm{r}},\del\Gamma_{\mathrm{r}})_{\alpha\beta}\big),
\\
\opDirac\Psi + \mathrm{i}M\Psi & = 0,
\endaligned
\end{equation}
\begin{equation}\label{eq11-04-oct-2025}
\Gamma^{\lambda}  
= {\Gamma_{\mathrm{r}}}^{\lambda} 
+ \del_{\alpha}\big(\Phisans_{\beta}^{\delta}\big)\Phisans_{\delta}^{\lambda} u^{\alpha\beta},
\end{equation}
with the following initial data
\begin{equation}
\aligned
&u_{\alpha\beta}|_{\{t=1\}} = u_{0\alpha\beta},\quad 
\del_tu_{\alpha\beta}|_{\{t=1\}} = u_{1\alpha\beta},
\\
& \Psi|_{\{t=1\}} = \Psi_0.
\endaligned
\end{equation}

}


\section{Quantitative statement of the nonlinear stability result}
\label{section=N15}

\subsection{Initial smallness conditions}

The following notation was introduced earlier. For a scalar function $u$ defined in $\Mcal_{[s_0,s_1]}$, the associated energy (related to the flat wave operator $\eta^{\mu\nu}\del_{\mu}\del_{\nu}$, cf. \cite[Eq.(3.18)]{PLF-YM-PDE}) is defined as
\begin{equation}\label{eq1-01-oct-2025}
\aligned
\Ebf_{\kappa}(s,u) & =
\int_{\Mcal_s}\Big(|\del_t u|^2 + \sum_a|\del_a u|^2 + 2(x^a/r)\del_rT(s,r)\del_tu\del_au\Big)\omega^{2\kappa}\diff x
\\
& = \int_{\Mcal_s}\Big(|\zeta\del_t u|^2 + \sum_a|\delb_a u|^2\Big)\omega^{2\kappa}\diff x.
\endaligned
\end{equation}
We also recall its high-order version:
\be
\Ebf_{\kappa}^{p,k}(s,u) := \sum_{\ord(I)\leq p\atop\rank(I)\leq J}\hspace{-0.3cm}\Ebf_{\kappa}(s,\mathscr{Z}^Iu),
\ee
where $\mathscr{Z}^I$ represents an admissible operator. We emphasize that when talking about energies associated with a {\sl scalar}, we always refer to \eqref{eq1-01-oct-2025}, while when talking the energies associated with a {\sl spinor}, we always refer to \eqref{equa-28-sept-2025a}.

In a curved spacetime, the above scalar energy takes the form
\begin{equation}\label{eq4-02-oct-2025}
\Ebf_{g,\kappa}(s,u) = \int_{\Mcal_s}\Big((-g^{00}|\del_tu|^2 + g^{ab}\del_au\del_bu ) + 2(x^a/r)\del_rT g^{a\beta}\del_tu\del_{\beta}u\big)\omega^{2\kappa}\diff x.
\end{equation}
We also introduce its high-order version
$$
\Ebf_{g,\kappa}^{p,k}(s,u) := \sum_{\ord(I)\leq p\atop\rank(I)\leq J}\hspace{-0.3cm}\Ebf_{g,\kappa}(s,\mathscr{Z}^Iu).
$$

We recall \cite[Section 10.3, 10.4]{PLF-YM-PDE} for the construction from the geometric initial data set to the PDE initial data set. In the present article we consider the same situation for the metric perturbation, i.e., we take
\begin{equation}
\|\la r\ra^{\kappa+\ord(I)}\del^I\del_au_{0\alpha\beta}\|_{L_2(\RR^3)} 
+ \|\la r\ra^{\kappa+\ord(I)}\del^I u_{1\alpha\beta}\|_{L^2(\RR^3)}
\leq \eps_{\Ebf},\quad \ord(I)\leq N,
\end{equation}
and for the spinor field, we assume the following smallness condition:
\begin{equation}
\|\la r\ra^{\mu+N+1}|\widehat{\del}^J\Psi|_{\vec{n}}\|_{L^2(\RR^3)}\leq \eps_{\Ebf},\quad \ord(J)\leq N+1.
\end{equation}
Here $\kappa,\mu$ are positive constants. Different choice of $(\kappa,\mu)$ leads us to (slightly) different strategy of bootstrap argument. Remark that these bounds leads us to the following smallness condition on the restriction of the local solution on $\Mcal_{s_0}$ (cf. \cite[Proposition~10.5]{PLF-YM-PDE}):
\begin{equation}\label{eq5-02-oct-2025}
\Ebf_{g,\kappa}^N(s_0,u)^{1/2} + \Ebf_{\mu}^N(s_0,\Psi)^{1/2} + \Ebf_{\mu}^N(s_0,\widehat{\del}\Psi)^{1/2}\leq C_g\eps_{\Ebf},
\end{equation}
where $C_g$ is a constant determined by $N$ and the system. This implies the following smallness conditions that we will apply in the bootstrap argument:
\begin{equation}\label{eq3-02-oct-2025}
\Ebf_{\kappa}^N(s_0,u)^{1/2} + \Ebf_{\mu}^N(s_0,\Psi)^{1/2} + \Ebf_{\mu}^N(s_0,\widehat{\del}\Psi)^{1/2}\leq C_0\eps_{\Ebf},
\end{equation}
where $\Ebf_{\kappa}^N(s_0,u)^{1/2} := \sum_{\alpha,\beta}\Ebf_{\kappa}^N(s_0,u_{\alpha\beta})^{1/2}$, and $C_0$ a constant determined by $N$. This can be guaranteed by \eqref{eq5-02-oct-2025}. In fact when $\eps$ is sufficiently small, the local theory tells that the solution $u_{\alpha\beta}$ remain small (controlled by $\eps$) in the region limited by $\{t=1\}$ and $\Mcal_{s_0}$ with a fixed $s_0$ (in the following discussion we fix $s_0=2$). Then we apply \cite[Lemma~17.4]{PLF-YM-PDE}, in which we have essentially showed that when $g$ is sufficiently close to the Minkowski metric, then $\Ebf_{g,\kappa}(s,u)$ and $\Ebf_{\kappa}(s,u)$ are equivalent. 


We now turn our attention to the reference metric $\gref$, which is a $(N,\theta,\epss,\ell)$- admissible reference in $\{\xsans\}$. We recall Proposition~\ref{prop1-01-oct-2025} and we fix
\begin{equation}\label{eq2-01-oct-2025}
\eps_s = \eps = \max\{\epss ,\eps_{\Ebf}\}.
\end{equation}
Then we obtain the following decay properties for $\gref$ form Proposition~\ref{prop1-01-oct-2025}:
\begin{subequations}\label{eq-reference-bounds}
\begin{equation}\label{eq1-03-oct-2025}
\aligned
|\Hreff_{\alpha\beta}\big|^{\Kscr}_p + \big|\Hreff^{\alpha\beta}\big|^{\Kscr}_p
\lesssim_N& \Kref\eps(t+r+1)^{-1+(p+1)\theta},\quad p\leq N+2,
\\
\big|\del_{\mu}\gref^{\alpha\beta}\big|^{\Kscr}_p
+ |\del_{\mu}\gref^{\alpha\beta}|^{\Kscr}_p
\lesssim_N& \Kref\eps(1+|r-t|)^{-1}(t+r+1)^{-1+(p+1)\theta},\quad p\leq N+1,
\\
\big|\dels_{\mu}\gref^{\alpha\beta}\big|^{\Kscr}_p
+ |\del_{\mu}\gref^{\alpha\beta}|^{\Kscr}_p
\lesssim_N& \Kref\eps (t+r+1)^{-2+(p+1)\theta},\quad p\leq N+1,
\\
\big|\del_{\mu}\del_{\nu}\gref^{\alpha\beta}\big|^{\Kscr}_p + |\del_{\mu}\del_{\nu}\gref^{\alpha\beta}|^{\Kscr}_p
\lesssim_N& \Kref\eps(1+|r-t|)^{-2}(t+r+1)^{-1+(p+1)\theta},\quad p\leq N,
\\
\big|\del_{\mu}\dels_{\nu}\gref^{\alpha\beta}\big|^{\Kscr}_p + |\del_{\mu}\dels_{\nu}\gref^{\alpha\beta}|^{\Kscr}_p
\lesssim_N& \Kref\eps(1+|r-t|)^{-1}(t+r+1)^{-2+(p+1)\theta},\quad p\leq N;
\endaligned
\end{equation}
\begin{equation}
\big|R[\gref]_{\alpha\beta}\big|^{\Kscr}_{N} + (1+r+t)\big|\del_{\gamma}R[\gref]_{\alpha\beta}\big|^{\Kscr}_{N-1}\lesssim_N \eps^2(1+t+r)^{-4+\theta},
\end{equation}
\begin{equation}\label{eq1-05-oct-2025}
\big|{\Gamma_{\mathrm{r}}}^{\lambda}\big|^{\Kscr}_{N+1} 
+ (1+|r-t|)\big|\del_{\gamma}{\Gamma_{\mathrm{r}}}^{\lambda}\big|^{\Kscr}_N\lesssim_N \Kref\eps(t+r+1)^{-2+\theta},
\end{equation}
\begin{equation}\label{eq7-02-oct-2025}
\Hreff^{\Ncal 00} := (\diff \tavec{} - \diff \ravec{}, \diff \tavec{} - \diff \ravec{})_{\gref}
< -  C_L\eps r^{-1+\theta}
\qquad \text{ in }\quad \Mcal_{[s_0,s_1]} \cap\{ r \geq (3/4) t \}.
\end{equation}
\end{subequations}
\subsection{Stability of the light-bending property}\label{subsect1-02-oct-2025}
As we have seen in the {\sl flatness conditions} introduced in Section~\ref{section=N11}, the analysis on spinor fields depends essentially on the light-bending condition \eqref{eq-bending-condition}. Although \eqref{eq7-02-oct-2025} guarantees this property for the reference, we want it holds for the total metric 
\be
g_{\alpha\beta}  = h_{\alpha\beta} + \eta_{\alpha\beta} = u_{\alpha\beta} + \Hreff_{\alpha\beta} + \eta_{\alpha\beta}.
\ee
To this juncture, we can chose $C_L$ sufficiently large, such that the property on $\Hreff$ would not be destroyed by $u_{\alpha\beta}$. More precisely, consider the free-linear wave system:
\begin{equation}
\Box v_{\alpha\beta} = 0,\quad v_{\alpha\beta}|_{\{t=1\}} = u_{0\alpha\beta},\quad \del_tv_{\alpha\beta}|_{\{t=1\}} = u_{1\alpha\beta}.
\end{equation}
Then the solution enjoys the following estimate (cf. \cite[Proposition~10.6]{PLF-YM-PDE}):
\begin{equation}\label{eq5-05-oct-2025}
|v_{\alpha\beta}|\lesssim L_0\eps(t+r+1)^{-1}
\end{equation}
where $L$ is a constant determined by $\kappa$. Therefore, when $\eps$ is sufficiently small, we have
\begin{equation}\label{eq8-02-oct-2025}
|v^{\Ncal00}|\leq L\eps, 
\end{equation} 
where
\be
v^{\Ncal00} = v(\diff t- \diff r,\diff t- \diff r), \quad\text{with}\quad v(\del_{\alpha},\del_{\beta}):=v_{\alpha\beta}
\ee
and $L$ a constant determined by $\kappa$. We then chose $C_L$ sufficiently large, such that
\begin{equation}\label{eq9-02-oct-2025}
\Kref\leq 2C_L\theta^{-1},\quad L\leq \frac{1}{3}C_L.
\end{equation}
Then the \textbf{light-bending condition} 
\begin{equation}\label{eq10-02-oct-2025}
h^{\Ncal00}\leq - \frac{1}{3}C_L\eps s^{-1+\theta},
\end{equation}
holds at least locally in time. This is because our solution $g_{\alpha\beta} = u_{\alpha\beta} + \gref_{\alpha\beta}$, and $u_{\alpha\beta}$ can be again decomposed as the initial-data contribution $v_{\alpha\beta}$ plus the nonlinear source contribution. In a given time interval, the latter is of size $\eps^2$ and it is the initial-data contribution dominates. We also see that from \eqref{eq7-02-oct-2025} that the reference gives a negative contribution with large $C_L$, while \eqref{eq8-02-oct-2025} and \eqref{eq9-02-oct-2025} show that the initial-data contribution does not destroy completely the light-bending property, at least  locally in time. Thus we conclude that there exists a time $s_1>0$ such that \eqref{eq10-02-oct-2025} holds.

In our global analysis, we will show that (by the bootstrap argument), the $\eps^2-$smallness of the nonlinear source contribution is in fact uniform with respect to time. Thus it will never destroy the light-bending property.

\subsection{Statement of the nonlinear
 stability result in the PDEs setting }
\label{subsec2-22-oct-2025}
We are now ready to state the nonlinear stability result in a quantitative manner.

\begin{theorem}
\label{subsec2-22-oct-2025-theo}
Let $\gref$ be a $(N,\theta,\epss,\ell)$ admissible reference with $N$ sufficiently large, and $\theta,\epss$ sufficiently small. Then there exists $\eps_0>0$, such that when
$$
\max\{\epss,\eps_{\Ebf}\}< \eps_0
$$
then any local solution to \eqref{eq03-02-oct-2025} together with \eqref{eq11-04-oct-2025} satisfying the following initial smallness conditions
\begin{equation}
\Ebf_{\kappa}^N(s_0,u)^{1/2} + \Ebf_{\mu}^N(s_0,\Psi)^{1/2} 
+ \Ebf_{\mu}^N(s_0,\widehat{\del}\Psi)^{1/2}\leq C_0\eps_{\Ebf},
\end{equation}
with $\kappa\in (1/2,3/4)$ and $\mu\in(3/4,1)$ extends to $\Mcal_{[s_0,+\infty)}$. 
\end{theorem}

\section{Bootstrap argument in $\Mcal^{\ME}_{[s_0,s_1]}$}
\label{section=N16}

\subsection{Bootstrap assumptions and choice of constants}\label{subsec1-02-oct-2025}

We (essentially) follow the {\bf Class B} proposed in~\cite{PLF-YM-PDE}. Considering the energy estimates in $
[s_0,s_1]$, we have the following bouns. 
\\
$\bullet$ Energy bounds: 
\begin{subequations}\label{eq1-02-oct-2025}
\begin{equation}\label{eq1-02-oct-2025-H}
\Ebf_{\kappa}^{\ME,N}(s,u)^{1/2} 
+ s^{-1}\Ebf_{\mu}^{\ME,N}(s,\Psi)^{1/2} 
+ s^{-1}\Ebf_{\mu}^{\ME,N}(s,\widehat{\del}\Psi)^{1/2} \leq C_1\eps s^{\delta},
\end{equation}
\begin{equation}\label{eq1-02-oct-2025-L}
\Ebf_{\kappa}^{\ME,N-5}(s,u)^{1/2} 
+ \Ebf_{\mu}^{\ME,N-5}(s,\Psi)^{1/2} 
+ \Ebf_{\mu}^{\ME,N-5}(s,\widehat{\del}\Psi)^{1/2} \leq C_1\eps s^{\delta},
\end{equation}
\begin{equation}\label{eq1-11-oct-2025}
\sum_{\ord(J) = p\atop \rank(J) = p}\Ebf_{\mu}^{\ME}(s,t^{-1}\widehat{L_a}(\mathscr{Z}^J\Psi))^{1/2}\leq C_1\eps, \quad p\leq N,
\end{equation}
\begin{equation}
\Ebf_{\kappa}^{\Hcal,0}(s,u)^{1/2}\lesssim C_1\eps s^{\delta}.
\end{equation}
\end{subequations}
Here $N$ is a sufficiently large integer. In the present article $N=22$ is sufficient. For energy weight indices $(\kappa,\mu)$, we chose 
\begin{equation}
1/2<\kappa<3/4,\quad 3/4<\mu<9/10.
\end{equation} 

The constant $C_1>0$ is chosen to be sufficiently large so that $C_1\eps s_0^{\delta}<\eps_{\Ebf}$. Then due to \eqref{eq3-02-oct-2025} and the continuity, \eqref{eq1-02-oct-2025} are valid on a interval $[s_0,s_1]$ with $s_1>s_0$.

The positive constant $\delta$ describes the increasing rate of the energies and will be fixed during the process of estimates. We begin by demand
\begin{equation}
0<\delta<10^{-1} \min\{\kappa-1/2,\mu-3/4\}.
\end{equation}
Furthermore, as in \cite{PLF-YM-PDE} (but more explicitly), we assume  
\begin{equation}\label{eq4-15-oct-2025}
0<\theta \ll_N \delta.
\end{equation}
More explicitly, we will see that $(12N^2+4N)\theta<\delta$ will be sufficient. This is a restriction on the reference metric $\gref$. However, this is not that restrictive, because our typical reference, the vacuum global solution constructed by \cite{LR2} in enjoys a decay rate for {\sl any} $\theta>0$.

\noindent$\bullet$  Light-bending condition: 
\begin{equation}\label{eq6-02-oct-2025}
h^{\Ncal00}\leq - \frac{1}{3}C_L\eps s^{-1+\theta},\quad \text{ in }  \Mcal_{[s_0,s_1]}.
\end{equation}
The constant $C_L$ is chosen to be sufficiently large and is determined by $\kappa$ as explained in~Section~\ref{subsect1-02-oct-2025}. Then $\Kref$ is fixed by the relation
$
\Kref = (K_0+C_L)\theta^{-1},
$
where $K_0$ is a universal constant. We then demand $C_L\geq K_0$ such that $\Kref\leq 2\theta^{-1}C_L$. In the following discussion we always fix 
\begin{equation}\label{eq2-03-oct-2025}
C_1\geq \Kref
\end{equation}
in the following discussion. We emphasize that the choice of the constants $N,\theta,\ell,\mu,\nu,\delta,C_1$ are independent on $\eps$. We thus demand $ C_1\eps<1/2$ in the following discussion.  In the following, we will apply the notation $\lesssim$ for a ``smaller or equal to up to a constant determined by \eqref{eq03-02-oct-2025} and $N$''. The dependence on the remaining constants will be explicitly expressed.

Then we will show the following strictly stronger estimates holds on same interval:
\begin{subequations}\label{eq11-02-oct-2025}
\begin{equation}
\Ebf_{\kappa}^{\ME,N}(s,u)^{1/2} 
+ s^{-1}\Ebf_{\mu}^{\ME,N}(s,\Psi)^{1/2} 
+ s^{-1}\Ebf_{\mu}^{\ME,N}(s,\widehat{\del}\Psi)^{1/2} \leq \frac{1}{2}C_1\eps s^{\delta},
\end{equation}
\begin{equation}
\Ebf_{\kappa}^{\ME,N-5}(s,u)^{1/2} 
+ \Ebf_{\mu}^{\ME,N-5}(s,\Psi)^{1/2} 
+ \Ebf_{\mu}^{\ME,N-5}(s,\widehat{\del}\Psi)^{1/2} \leq \frac{1}{2}C_1\eps s^{\delta},
\end{equation}
\begin{equation}\label{eq3-11-oct-2025}
\sum_{\ord(J) = p\atop \rank(J) = p}
\Ebf_{\mu}^{\ME}(s,t^{-1}\widehat{L_a}(\mathscr{Z}^J\Psi))^{1/2}
\leq \frac{1}{2}C_1\eps, \quad p\leq N,
\end{equation}
\end{subequations}
\begin{equation}
h^{\Ncal00}\leq - \frac{1}{2}C_L\eps s^{-1+\theta},\quad \text{ in }  \Mcal_{[s_0,s_1]}.
\end{equation}
Then improvement on $\Ebf_{\kappa}^{\Hcal,0}(s,u)$ is left to the analysis in the hyperboloidal region.

If we compare the present setting with the {\bf Class B} of \cite{PLF-YM-PDE}, we will find the following inconsistence:
\\
$\bullet$ We do not have the index $\lambda$, which describes the decay rate of general components of $\Hreff$ in \cite[Section~12]{PLF-YM-PDE}. Here the role is taken by $1- \theta$. However, in \cite{PLF-YM-PDE} we also demanded $\theta<1- \lambda$. This problem will be treated next. On the other hand, the {\sl tame decay} \cite[Eq. (12.8)]{PLF-YM-PDE} are automatically satisfied.
\\
$\bullet$ The exact wave gave condition satisfied by $g$ now becomes a generalized wave coordinate conditions \eqref{eq1-05-oct-2025}, and the wave gauge condition satisfied by the reference $\gref$ is different. These two points will be discussed in Section~\ref{subsec3-05-oct-2025}.
\\
$\bullet$ The main system contains two groups of new terms. Their estimates are included in~Section~\ref{subsec1-05-oct-2025}.
\\
$\bullet$ We need to guarantee the {\sl flatness conditions} introduced in Section~\ref{section=N11}. This will be considered shortly below. 
\\
$\bullet$ We have the bootstrap assumptions on the ``good derivatives'' of $\Psi$, i.e., \eqref{eq1-11-oct-2025}. We will see that these spinors enjoy particular estimates. These new energy bounds will be carefully used and reinforced in~Section~\ref{subsec1-11-oct-2025}.
\subsection{Direct estimates for metric components and spinor field}
In order to follow the argument of \cite{PLF-YM-PDE}, we take $\lambda = 1-2\theta$ which satisfies the restriction $0<\theta<1- \lambda$.

As in \cite{PLF-YM-PDE}, we have, from \eqref{eq1-02-oct-2025-H} and \eqref{eq3-19-aout-2025} of Proposition~\ref{prop-spin-Sobolev}, the following estimates on metric perturbation $u$:
\begin{equation}\label{eq3-03-oct-2025}
r\la r-t \ra^{\kappa}|\del u|_{N-3} + r^{1+\kappa}|\delsN u|_{N-3}\lesssim C_1\eps s^{\delta},
\end{equation}
which, combined with \eqref{eq1-03-oct-2025} (the second and third estimate, remark that we have assume $C_1\geq \Kref$), 
\begin{equation}\label{eq4-03-oct-2025}
r\la r-t \ra^{\kappa}|\del h|_{N-3} + r^{1+\kappa}|\delsN h|_{N-3}\lesssim C_1\eps s^{\delta},
\end{equation} 
where we have demanded that 
$$
(N-2)\theta\leq \delta/2,\quad \theta\leq 1/10.
$$ 
We recall \cite[Proposition~7.2]{PLF-YM-PDE}.
\begin{proposition}[Hardy-Poincar\'e inequality for high-order derivatives] 
\label{eq3-15-05-2020}
For any $\kappa = 1/2 + \delta$ with $\delta>0$ and any sufficiently decaying function $u$ defined in $\Mcal_{[s_0,s_1]}$ and 
for all $s \in [s_0, s_1]$ one has 
$$
\| \la r-t \ra^{-1 + \kappa} |u|_{p,k}\|_{L^2(\MME_s)} 
\lesssim \big(1+\delta^{-1} \big) \, \Ebf_\kappa^{\ME,p,k}(s,u)^{1/2} + 
\Ebf_\kappa^0(s,u)^{1/2}. 
$$
\end{proposition} 
This leads us to, similar to \cite[Eq. (12.30)]{PLF-YM-PDE},
\begin{equation}
\|\la r-t\ra^{-1+\kappa}u\|_{L^2(\Mcal^{\ME}_s)}\lesssim \delta^{-1}C_1\eps s^{\delta}.
\end{equation}
Also, similar to \cite[Eq. (12.33)]{PLF-YM-PDE} which is essentially the above Proposition~\ref{eq3-15-05-2020} combined with the scalar Sobolev inequality, and obtain 
\begin{equation}\label{eq5-03-oct-2025}
r\la r-t\ra^{\kappa-1}|u|_{N-2}\lesssim \delta^{-1}C_1\eps s^{\delta}.
\end{equation}
This is a  result of the Proposition~7.2 of \cite{PLF-YM-PDE} (the weighted Hardy-Poincar\'e inequality) combined with \eqref{eq1-02-oct-2025-H}. Then recalling \eqref{eq1-03-oct-2025} (the first and second estimate), we have
\begin{equation}\label{eq7-03-oct-2025}
|h|_{N-2}\lesssim \delta^{- \lambda}C_1\eps r^{-1}\la r-t\ra^{1- \kappa}s^{\delta},\quad \lambda = 1-2\theta.
\end{equation}

It is important that these pointwise estimates guarantee the {\sl flatness conditions} in $\Mcal^{\ME}_{[s_0,s_1]}$. In fact we have the following estimate on the generalized wave gauge condition:
\begin{equation}
|W^{\lambda}|_{N-2} = |\Gamma^{\lambda}|_{N-2}\lesssim \delta^{-1}C_1\eps r^{-2+\theta}.
\end{equation}
To see this, we recall Corollary~\ref{lem3-24-sept-2025} and obtain
\begin{equation}
|W^{\lambda}|_{N-2}\lesssim \delta^{-1}C_1\eps r^{-2+\theta}.
\end{equation}
Then we observe that the Light-bending property is guaranteed by the bootstrap assumption \eqref{eq6-02-oct-2025}. The Uniform spacelike property is guaranteed by \eqref{eq4-03-oct-2025} and \eqref{eq6-02-oct-2025}. The Generalized wave gauge property, the High-order condition and the Treating conditions are also direct. We focus on the Sobolev-related conditions. In fact we only need to verify the second one.
\begin{lemma}[The Sobolev-related conditions]\label{lem1-03-oct-2025}
In $\Mcal^{\ME}_{[s_0,s_1]}$, the following estimate holds:
\begin{equation}\label{eq8-03-oct-2025}
\aligned
\zeta^{-2}|\del h| + \zeta^{-1}|h|_1\lesssim 
+ \frac{\zeta^{-2}|h^{\Ncal00}|_1}{1+\zeta^{-2}|h^{\Ncal00}|}\lesssim \delta^{-1}C_1\eps\zeta^{- \delta/2}.
\endaligned
\end{equation}
\end{lemma}
\begin{proof}
We recall \cite[Lemma~3.4]{PLF-YM-PDE}:
\begin{equation}\label{eq6-03-oct-2025}
\frac{\la r-t\ra}{r}\lesssim \zeta^2,\quad \Mcal^{\Mcal}_{[s_0,s_1]}.
\end{equation}
Thus \eqref{eq4-03-oct-2025} together with \eqref{eq7-03-oct-2025} leads us to 
$$
\zeta^{-2}|\del h| + \zeta^{-1}|h|\lesssim \delta^{-1}C_1\eps s^{\delta}\lesssim \delta^{-1}C_1\eps \zeta^{- \delta/2},\quad \text{ in }  \Mcal^{\Mcal}_{[s_0,s_1]}.
$$
For the third term in the left-hand side of \eqref{eq8-03-oct-2025}, we remark that, thanks to \eqref{eq7-03-oct-2025}
$$
|h^{\Ncal00}|\lesssim \delta^{-1}C_1\eps s^{\delta}.
$$
Thus
$$
\zeta^{-1}|h|_1\lesssim 
+ \frac{\zeta^{-2}|h^{\Ncal00}|_1}{1+\zeta^{-2}|h^{\Ncal00}|}
\lesssim \delta^{-1}C_1\eps s^{\delta}
\lesssim \delta^{-1}C_1\eps \zeta^{- \delta/2},\quad \text{ in }  \Mcal^{\Mcal}_{[s_0,s_1]}.
$$

The estimate is trivial in $\Mcal^{\Ecal}_{[s_0,s_1]}$, because $\zeta\geq 1$ there.
\end{proof}

Thanks to Lemma~\ref{lem1-03-oct-2025}, we are able to apply the spinorial Sobolev inequalities \eqref{eq10-03-oct-2025}and \eqref{eq2-19-aout-2025}. We obtain
\begin{equation}\label{eq10-04-oct-2025}
[\Psi]_{p} + [\del\Psi]_p\lesssim_{\delta} 
\begin{cases}
C_1\eps\zeta^{1/2- \delta}\la r-t\ra^{- \mu}r^{-1}s^{1+\delta},\quad &p\leq N-3,
\\
C_1\eps\zeta^{1/2- \delta}\la r-t\ra^{- \mu}r^{-1}s^{\delta},\quad &p\leq N-8.
\end{cases}
\end{equation}
Also, by \eqref{eq3-19-aout-2025}
\begin{equation}\label{eq5-04-oct-2025}
[\Psi]_p\lesssim_{\delta}
\begin{cases}
C_1\eps \zeta^{-1/2- \delta}\la r-t\ra^{1- \mu}r^{-2}s^{1+2\delta},\quad & p\leq N-3,
\\
C_1\eps \zeta^{-1/2- \delta}\la r-t\ra^{1- \mu}r^{-2}s^{2\delta},\quad& p\leq N-8.
\end{cases}
\end{equation}
Then by \eqref{eq7-04-oct-2025}, we obtain
\begin{equation}\label{eq6-04-oct-2025}
[\Psi]_p\lesssim_{\delta}
\begin{cases}
C_1\eps \zeta^{1/2- \delta}\la r-t\ra^{1- \mu}r^{-2}s^{1+2\delta},\quad & p\leq N-4,
\\
C_1\eps \zeta^{1/2- \delta}\la r-t\ra^{1- \mu}r^{-2}s^{2\delta},\quad& p\leq N-9.
\end{cases}
\end{equation}

\subsection{Estimates on new terms}
\label{subsec1-05-oct-2025}
We are ready to bound the ``new terms'' respect to \cite{PLF-YM-PDE}. Then are:
\begin{equation}
\aligned
\mathbb{S}(g,\del g,\Psi,\widehat{\del}\Psi)_{\alpha\beta},
\quad
\big(\mathbb{G}(g,\del g,\Gamma,\del\Gamma)_{\alpha\beta} - \mathbb{G}(\gref,\del \gref,\Gamma_{\mathrm{r}},\del\Gamma_{\mathrm{r}})_{\alpha\beta}\big).
\endaligned
\end{equation}
We need to give the pointwise and $L^2$ bounds on these terms. In order to not disturb the argument in \cite{PLF-YM-PDE}, we need to show that these terms enjoy integrable $L^2$ estimates, and their pointwise bounds should give $\la t+r\ra^{-1}$ in the Kirchhoff's formula, i.e., the ``sub-critical'' case of \cite[Proposition~8.1]{PLF-YM-PDE}. Terms satisfying the above two properties will be called {\sl trivial terms} in the following discussion.
\begin{proposition}\label{prop1-04-oct-2025}
In $\Mcal^{\ME}_{[s_0,s_1]}$, the following estimates hold: 
\begin{equation}
\big|\big(\mathbb{G}(g,\del g,\Gamma,\del\Gamma)_{\alpha\beta} - \mathbb{G}(\gref,\del \gref,\Gamma_{\mathrm{r}},\del\Gamma_{\mathrm{r}})_{\alpha\beta}\big)\big|_{N-3}\lesssim_{\delta}(C_1\eps)^2\la r-t\ra^{- \kappa}r^{-3+2\delta},
\end{equation}
\begin{equation}
\big\|s\zeta\la r-t\ra^{\kappa}\big|\mathbb{G}(g,\del g,\Gamma,\del\Gamma)_{\alpha\beta} - \mathbb{G}(\gref,\del \gref,\Gamma_{\mathrm{r}},\del\Gamma_{\mathrm{r}})_{\alpha\beta}\big|_N\big\|_{L^2(\Mcal^{\ME}_s)}
\lesssim_{\delta} (C_1\eps)^2s^{-3+4\delta}.
\end{equation}
\end{proposition}
\begin{proof}
This is by Lemma~\ref{lem4-24-sept-2025}. For the pointwise bound, we substitute \eqref{eq3-03-oct-2025} and \eqref{eq5-03-oct-2025} into \eqref{eq1-04-oct-2025}, and obtain the desired estimate.

For the $L^2$ estimate, we need to apply Proposition~\ref{eq3-15-05-2020}.
\end{proof}


We now turn our attention to the stress-energy tensor of $\Psi$, which is the interaction from $\Psi$ to the metric $g$. Recalling the decomposition \eqref{eq11-03-oct-2025}, we will establish the following estimate.
\begin{proposition}\label{prop1-03-oct-2025}
Under the bootstrap assumptions and let $N\geq 19$, the following estimates hold:
\begin{equation}
|\mathbb{S}(g,\del g,\Psi,\widehat{\del}\Psi)_{\alpha\beta}|_{p,k}
\lesssim (C_1\eps)^2\la r-t\ra^{-1+(1/4-3\delta)}r^{-2-(1/4-2\delta)},
\end{equation}
\begin{equation}
\big\|s\zeta\la r-t\ra^{\kappa}\big|\mathbb{S}(g,\del g,\Psi,\widehat{\del}\Psi)_{\alpha\beta}\big|_N\big\|_{L^2(\Mcal^{\ME}_s)}
\lesssim_{\delta} (C_1\eps)^2s^{-1- \delta}.
\end{equation}
\end{proposition}

These are parallel to the \cite[Lemma~12.4 and Eq. (12.42)]{PLF-YM-PDE}. The proof is composed by several steps. We first deal with the estimate for the Dirac form $\la \Phi,\Psi \ra_{\ourD}$.

\begin{lemma}
Assume that \eqref{eq-China-condition} holds in $\Mcal_{[s_0,s_1]}$. Then one has 
\begin{equation}\label{eq1-28-sept-2025}
|\mathrm{div}(Z)|_{p,k}\lesssim_N  (t+r)|\del H|_{p,k}
\end{equation}
\begin{equation}\label{eq2-28-sept-2025}
\aligned
\, [\mathscr{Z}^I\la\Phi,\Psi\ra_{\ourD}]_{p,k}\lesssim_N& 
\sum_{p_a+p_b\leq p\atop k_a+k_b\leq k}[\Phi]_{p_a,k_a}[\Psi]_{p_b,k_b} 
+ \HDirac_{p,k}[\del H,\Phi,\Psi]
\endaligned
\end{equation}
with
\be
\HDirac_{p,k}[\del H,\Phi,\Psi]
\lesssim_N\sum_{j=1}^p(t+r)^j\hspace{-0.8cm}\sum_{p_1+\cdots+p_j+p_a+p_b\leq p\atop k_1+\cdots+k_j+k_a+k_b\leq k}\prod_{l=1}^j|\del H|_{p_l,k_l}[\Phi]_{p_a,k_a}[\Psi]_{p_b,k_b}.
\ee
\end{lemma}

\begin{proof}
\bse
Recall that
\begin{equation}
\mathrm{div}(Z) = \del_{\alpha}Z^{\alpha} + Z^{\gamma}\Gamma_{\gamma\alpha}^{\alpha}.
\end{equation}
We remark that $\del_{\alpha}Z^{\beta}$ are constants. Furthermore, when $Z = \del_{\mu}$, $Z^{\alpha}$ are constants. When $Z = L_a$ or $\Omega_{ab}$, an explicit calculation shows that
\be
\del_{\alpha}Z^{\alpha} = 0.
\ee
Then thanks to \eqref{eq-China-condition} and Lemma~\ref{lem1-11-july-2025},
\begin{equation}
|\Gamma_{\alpha\beta}^{\gamma}|_{p,k}\lesssim_N |\del H|_{p,k}.
\end{equation}
Thus for $p\leq N$,
\be
\big|\mathscr{Z}^I\big(\mathrm{div}(Z)\big)\big|\lesssim_N (t+r)|\del H|_{p,k}.
\ee
This leads us to \eqref{eq1-28-sept-2025}.
\ese

For \eqref{eq2-28-sept-2025}, we establish the decomposition
\begin{subequations}
\begin{equation}
\mathscr{Z}^I(\la\Phi,\Psi\ra_{\ourD}) \simeq 
\la\mathscr{Z}^J\Phi,\mathscr{Z}^K\Psi \ra_{\ourD},
\quad 
\prod_{l=1}^j\mathscr{Z}^{I_l}\big(\mathrm{div}(Z_{\imath_l})\big)\la\mathscr{Z}^J\Phi,\mathscr{Z}^K\Psi\ra_{\ourD},
\end{equation}
with $1\leq j\leq p$ and
\begin{equation}
\sum_{l=1}^j\ord(I_l) + \ord(J) + \ord(K) \leq p,
\quad
\sum_{l=1}^j\rank(I_l) + \rank(J) + \rank(K) \leq k. 
\end{equation}
\end{subequations}
where $\simeq$ means the left-hand side is a finite linear combination of the right-hand side with constant coefficients. Thanks to Lemma~\ref{lemma-19juillet2025-a}, this can be easily checked by induction on $p=\ord(I)$.
\end{proof}

\begin{lemma}[Estimates on quadratic terms]\label{lem1-04-oct-2025}
Under the bootstrap assumptions and suppose that $N\geq 1p$, the following estimates hold:
\begin{equation}\label{eq2-04-oct-2025}
\sum_{p_1+p_2\leq N-3}\zeta^{-1}[\Psi]_{p_1}\big([\Psi]_{p_2} + [\del\Psi]_{p_2}\big)
\lesssim_{\delta} (C_1\eps)^2\la r-t\ra^{-1/2-3\delta}r^{-5/2+2\delta},
\end{equation}
\begin{equation}\label{eq9-04-oct-2025}
\aligned
\sum_{p_1+p_2\leq N}\!\!
\big\|s\zeta\la r-t\ra^{\kappa}\cdot \zeta^{-1}\cdot[\Psi]_{p_1}[\Psi]_{p_2}\big\|_{L^2(\Mcal^{\ME}_s)}  
\lesssim_\delta& (C_1\eps)^2s^{-3/2},
\\
\sum_{p_1+p_2\leq N}\!\!
\big\|s\zeta\la r-t\ra^{\kappa}\cdot\zeta^{-1}\cdot[\Psi]_{p_1}[\del\Psi]_{p_2}\big\|_{L^2(\Mcal^{\ME}_s)}
\lesssim_\delta& (C_1\eps)^2s^{-3/2}.
\endaligned
\end{equation}
\end{lemma}
\begin{proof}
For \eqref{eq9-04-oct-2025}, we only need to substitute \eqref{eq10-04-oct-2025} and \eqref{eq6-04-oct-2025}. 

For the $L^2$ estimates, we only write the second one in details. When $N\geq 19$, we have $[N/2]\leq N-10$ and then
\be
\aligned
& \big\|s\zeta\la r-t\ra^{\kappa}\cdot\zeta^{-1} \cdot[\Psi]_{p_1}[\del\Psi]_{p_2}\big\|_{L^2(\Mcal^{\ME}_s)}
\\
& \lesssim_{\delta}
\begin{cases}
C_1\eps s^{1+2\delta}\big\|\zeta^{-1/2- \delta}\la r-t\ra^{1-2\mu+\kappa}r^{-2}\cdot\zeta\la r-t\ra^{\mu}[\Psi]_N\big\|_{L^2(\Mcal^{\ME}_s)}
\\
C_1\eps s^{1+2\delta}\big\|\zeta^{-1/2- \delta}\la r-t\ra^{1-2\mu+\kappa}r^{-2}\cdot\zeta\la r-t\ra^{\mu}[\del\Psi]_N\big\|_{L^2(\Mcal^{\ME}_s)}
\end{cases}
\\
& \lesssim_{\delta}
\begin{cases}
C_1\eps s^{1+2\delta}\big\|\la r-t\ra^{3/4-2\mu+\kappa}r^{-7/4}\cdot\zeta\la r-t\ra^{\mu}[\Psi]_N\big\|_{L^2(\Mcal^{\ME}_s)}
\\
C_1\eps s^{1+2\delta}\big\|\la r-t\ra^{3/4-2\mu+\kappa}r^{-7/4}\cdot\zeta\la r-t\ra^{\mu}[\del\Psi]_N\big\|_{L^2(\Mcal^{\ME}_s)}
\end{cases}
\\
& \lesssim_{\delta}
\begin{cases}
C_1\eps s^{-5/2+2\delta}\Ebf^{\ME,N}_{\kappa}(s,\Psi)^{1/2}
\\
C_1\eps s^{-5/2+2\delta}\Ebf^{\ME,N}_{\kappa}(s,\widehat{\del}\Psi)^{1/2}
\end{cases}
\\
& \lesssim_{\delta}(C_1\eps)^2s^{-1- \delta}.
\endaligned
\ee
\end{proof}

\begin{lemma}[Estimates on high-order terms]\label{lem2-04-oct-2025}
Assume the bootstrap assumptions and \eqref{eq1-03-oct-2025} hold. Assume in addition that
\begin{equation}
N\delta\leq 1/2,\quad N\geq 19.
\end{equation}
Then one has
\begin{equation}\label{eq3-04-oct-2025}
\zeta^{-1}\HDirac_{N-3}[\del h,\Psi,\Psi] + \zeta^{-1}\HDirac_{N-3}[\del h,\Psi,\del\Psi]
\lesssim_{\delta} (C_1\eps)^3\la r-t\ra^{-3/4-3\delta}r^{-7/4+2\delta}.
\end{equation}
\begin{equation}
\big\|s\la r-t\ra^{\kappa}\HDirac_N[\del h,\Psi,\Psi]\big\|_{L^2(\Mcal^{\ME}_s)}
+\big\|s\la r-t\ra^{\kappa}\HDirac_N[\del h,\Psi,\del\Psi]\big\|_{L^2(\Mcal^{\ME}_s)}
\lesssim_{\delta}(C_1\eps)^3s^{-1- \delta}.
\end{equation}
\end{lemma}

\begin{proof}
For the pointwise estimate, we observe that 
$$
(t+r)|\del h|_{N-3}\lesssim C_1\eps \la r-t\ra^{- \kappa}s^{\delta}\lesssim C_1\eps \la r-t\ra^{\kappa}r^{\delta/2}. 
$$
Then we have
\begin{equation}\label{eq4-04-oct-2025}
\aligned
(t+r)^j\hspace{-0.8cm}\sum_{p_1+\cdots+p_j\leq N-3}\prod_{l=1}^j|\del h|_{p_l}
& \lesssim_{\delta} C_1\eps\big(\la r-t\ra^{- \kappa}r^{\delta/2}\big)^j
\lesssim_{\delta} (C_1\eps)\la r-t\ra^{- \kappa}r^{j\delta/2}
\\
& \lesssim (C_1\eps)\la r-t\ra^{- \kappa}r^{1/4}.
\endaligned
\end{equation}
Then combined with \eqref{eq2-04-oct-2025}, we obtain \eqref{eq3-04-oct-2025}.

For the $L^2$ bounds, we consider the product
$$
\prod_{l=1}^jr|\del h|_{p_l}[\Psi]_{p_a}[\Psi]_{p_b}, \quad\sum_{l=1}^jp_l + p_a+p_b\leq N.
$$
When $\sum p_l\leq [N/2]$, we apply \eqref{eq4-04-oct-2025} and \eqref{eq6-04-oct-2025}:
$$
\aligned
& \Big\|s\la r-t\ra^{\kappa}\prod_{l=1}^jr|\del h|_{p_l}[\Psi]_{p_a}[\Psi]_{p_b}\Big\|_{L^2(\Mcal^{\ME}_s)}
\\
& \lesssim_{\delta}(C_1\eps)^2s^{1+2\delta}\|\la r-t\ra^{- \kappa}r^{1/8}\cdot\zeta^{-1/2- \delta}\la r-t \ra^{1+\kappa-2\mu} r^{-2}\cdot\zeta\la r-t\ra^{\mu}[\Psi]_{N}\|_{L^2(\Mcal^{\ME}_s)}
\\
& \lesssim_{\delta}(C_1\eps)^2s^{-1- \delta}.
\endaligned
$$
When $\max{p_l}\geq [N/2]$, we have 
$$
\aligned
\prod_{l=1}^jr|\del h|_{p_l}[\Phi]_{p_a}[\Psi]_{p_b}
& \lesssim r^{1+N\delta/4}|\del u|_{N}[\Psi]_{[N/2]}^2 + \la r-t\ra^{-1}r^{N\delta/4}[\Psi]_{N}[\Psi]_{[N/2]}
\\
& \lesssim_{\delta} (C_1\eps)^2 \zeta^{1-2\delta}s^{4\delta}\la r-t\ra^{2-2\mu}r^{-23/8}|\del u|_{N} \\
& \quad+ (C_1\eps)^2 s^{2\delta}\zeta^{1/2- \delta}\la r-t\ra^{- \mu}r^{-15/8}[\Psi]_N.
\endaligned
$$
Then we obtain
$$
\aligned
& \Big\|s\la r-t\ra^{\kappa}\prod_{l=1}^jr|\del h|_{p_l}[\Phi]_{p_a}[\Psi]_{p_b}\Big\|_{L^2(\Mcal^{\ME}_s)}
\\
& \lesssim_{\delta}
(C_1\eps)^2 s^{1+4\delta}
\big\|\zeta^{-2\delta}\la r-t\ra^{2-2\mu}r^{-23/8}\cdot\zeta\la r-t\ra^{\kappa}|\del u|_N\big\|_{L^2(\Mcal^{\ME}_s)}
\\
& \quad+(C_1\eps)^2s^{1+2\delta}
\big\|\zeta^{-1/2- \delta}\la r-t\ra^{\kappa-2\mu}r^{-15/8}\cdot\zeta\la r-t\ra^{\mu}[\Psi]_N\big\|_{L^2(\Mcal^{\ME}_s)}
\\
& \lesssim_{\delta} (C_1\eps)^3s^{-1- \delta}.        
\endaligned
$$
The estimates concerning $\widehat{\del}\Psi$ is similar, and we omit the details.
\end{proof}


\begin{proof}[Proof of Proposition~\ref{prop1-03-oct-2025}]
We need to estimate each term in the right-hand side of \eqref{eq6-27-sept-2025}. For simplicity of expression, we denote by
\be
T_1 :=\la\Psi,\del_\mu \cdot\widehat{\del_{\nu}} \Psi \ra_{\ourD},
\quad
T_2 :=g_{\alpha\beta}\la\Psi,\Psi\ra_{\ourD},
\quad
T_3:=g^{\gamma\beta}\del_{\beta}g_{\mu\nu}\la \Psi,\del_{\gamma}\cdot\Psi\ra_{\ourD}. 
\ee
The terms in \eqref{eq11-03-oct-2025} are linear combinations with constant coefficients of them. First, thanks to \eqref{eq2-28-sept-2025},
\begin{equation}
\aligned
\big|T_2\big|_{p,k}
\lesssim_N& \sum_{p_1+p_2\leq p\atop k_1+k_2\leq k}[\Psi]_{p_1,k_1}[\Psi]_{p_2,k_2}
+\Cubic_{p,k}[h,\Psi,\Psi]
\\
& \quad + \HDirac_{p,k}[\del h,\Psi,\Psi] + \sum_{p_1+p_2\leq p\atop k_1+k_2\leq k}|h|_{p_1,k_1}\HDirac_{p_2,k_2}[\del h,\Psi,\Psi].
\endaligned
\end{equation}
Thanks to \eqref{eq3-23-july-2025} and \eqref{eq2-28-sept-2025}, we deduce that 
\begin{equation}\label{eq1-29-sept-2025}
\aligned
\zetab|\la\Phi,\del_{\gamma}\cdot\Psi\ra_{\ourD}|_{p,k}
\lesssim_N& \sum_{p_1+p_2\leq p\atop k_1+k_2\leq k}[\Phi]_{p_1,k_1}[\Psi]_{p_2,k_2} 
+ \Cubic_{p,k}[H,\Phi,\Psi]
\\
& \quad + \HDirac_{p,k}[\del H,\Phi,\Psi] 
+ \sum_{p_0+p_1\leq p\atop k_0+k_1\leq k}\hspace{-0.3cm}|H|_{p_0,k_0}\HDirac_{p_1,k_1}[\del H,\Phi,\Psi].
\endaligned
\end{equation}
Thus we have 
\begin{equation}
\aligned
\zetab|T_1|_{p,k}\lesssim_N& 
\sum_{p_1+p_2\leq p\atop k_1+k_2\leq k}[\Psi]_{p_1,k_1}[\del\Psi]_{p_2,k_2} 
+ \Cubic_{p,k}[h,\Psi,\del\Psi]
\\
& \quad + \HDirac_{p,k}[\del h,\Psi,\del\Psi] 
+ \sum_{p_0+p_1\leq p\atop k_0+k_1\leq k}|h|_{p_0,k_0}\HDirac_{p_1,k_1}[\del h,\Psi,\del\Psi].
\endaligned
\end{equation}
For the term $T_3$, we note that
$$
T_3
=\eta^{\beta\gamma}\del_{\beta}h_{\mu\nu}\la \Psi,\del_{\gamma}\cdot\Psi\ra_{\ourD} 
+ h^{\beta\gamma}\del_{\beta}h_{\mu\nu}\la \Psi,\del_{\gamma}\cdot\Psi\ra_{\ourD}.
$$

Finally, the cubic terms enjoy sufficient decay (which can be easily observed via Lemma~\ref{lem1-04-oct-2025}). Similarly, the terms concerning $\HDirac_{p,k}$ can also be bounded via Lemma~\ref{lem2-04-oct-2025}.
\end{proof}


\subsection{Modifications on the generalized wave coordinate conditions}
\label{subsec3-05-oct-2025}

In the present context, the metric $g$ non longer satisfies the exact {\sl wave gauge conditions} $\Gamma^{\lambda} = 0$. The generalized wave coordinate conditions \eqref{eq11-04-oct-2025} is not so far from that. More clearly,
\begin{equation}\label{eq12-04-oct-2025}
\aligned
& \big|\Gamma^{\lambda}\big|_{p,k}\lesssim C_1\eps\la r+t\ra^{-2+\theta} (1 + |u|_{p,k}),\quad&& \text{Present case},
\\
& \big|\Gamma^{\lambda}\big|_{p,k} = 0,\quad && \text{Class B of \cite{PLF-YM-PDE}}.
\endaligned
\end{equation}
Furthermore, in the present case, the wave gauge conditions on the reference is different for \cite{PLF-YM-PDE} (adjusted to the present system of notation):
\begin{equation}\label{eq2-05-oct-2025}
\aligned
&|{\Gamma_{\mathrm{r}}}^{\gamma}|\lesssim C_1\eps \la r+t\ra^{-2+\theta},\quad && \text{Present case},
\\
&|{\Gamma_{\mathrm{r}}}^{\gamma}|\lesssim C_1\eps \la r+t\ra^{-1}\la r-t\ra^{-1- \varsigma},\quad \varsigma>0,\quad && \text{Case B of \cite{PLF-YM-PDE}}.
\endaligned
\end{equation} 

However, we will explain why these difference do not destroy the argument in \cite{PLF-YM-PDE}. For this purpose, we recall the role of these conditions in \cite{PLF-YM-PDE} is to turn certain derivatives of bad direction to good directions (viz. tangent derivatives of $\Mcal_s$ or the light-cone). More precisely, they are applied in \cite[Lemma~11.2]{PLF-YM-PDE} for the decomposition of the quasi-null terms, and in \cite[Lemma~11.3 and Lemma~11.4]{PLF-YM-PDE} for the gradient of the component $g^{\Ncal00}$.

\paragraph{Decomposition of quasi-null terms.}

Let us consider the decomposition of quasi-null terms. A first change arrives due to \eqref{eq12-04-oct-2025} in the \cite[Eq. (11.11)]{PLF-YM-PDE}:
\begin{equation}
|\eta^{0b}\del_t h^{\Ncal}_{bc}|_{p,k}\lesssim \underline{|\Gamma^{\lambda}|_{p,k}} + |\dels^{\Ncal}h|_{p,k} + r^{-1}|h|_{p,k} 
+ \sum_{p_1+p_2=p\atop k_1+k_2=k}|\del h|_{p_1,k_1}|h|_{p_2,k_2}
\end{equation}
with \underline{one more term} in the right-hand side. This change leads us to two more terms in the right-hand side of \cite[Eq. (11.18)]{PLF-YM-PDE}:
\begin{equation}
|\Gamma^{\lambda}|_{p_1,k_1}|\del u|_{p_2,k_2},\quad |\Gamma^{\lambda}|_{p_1,k_1}|\Gamma^{\lambda'}|_{p_2,k_2}
\end{equation}
and two more terms in the estimate on $|\mathbb{P}^{\star\Ncal}_{00}[u,v]|_{p,k}$ in \cite[Lemma~11.2]{PLF-YM-PDE}, which are (adjusted in the present system of notation)
\begin{equation}
|\Gamma^{\lambda}|_{p_1,k_1}|\del \Hreff|_{p_2,k_2}
\quad
|{\Gamma_{\mathrm{r}}}^{\lambda}|_{p_1,k_1}|\Gamma^{\lambda'}|_{p_2,k_2}
\end{equation}
Fortunately, all enjoy sufficient decay and are {\sl trivial terms} in the sens defined at the beginning of~Section~\ref{subsec1-05-oct-2025}. That is the estimates on 
\begin{equation}
|{\Gamma_{\mathrm{r}}}^{\lambda}|_{p_1,k_1}|\del u|_{p_2,k_2}
\end{equation}
contained in the right-hand side of the estimate on $|\mathbb{P}^{\star\Ncal}_{00}[u,v]|_{p,k}$ in \cite[Lemma~11.2]{PLF-YM-PDE} (more precisely, hidden in $\mathbb{W}^{\ME}_{p_1,k_1}[v]|\del u|_{p_2,k_2}$. The remaining terms in this expression are cubic in nature). However, the decay present decay rate $r^{-2+\theta}$ is still sufficient such that the above term becomes {\sl trivial}.


\paragraph{The gradient of $g^{\Ncal00}$.}

We recall \cite[Lemma~11.3 and Lemma~11.4]{PLF-YM-PDE}. The estimates on the gradient of the component $g^{\Ncal00}$ now becomes (adjusted to the present system of notation):
\begin{equation}
\aligned
|\del u^{\Ncal00}|_{p,k}& \lesssim \, \underline{|\Gamma^{\lambda}|_{p,k}} + |\dels^{\Ncal}u|_{p,k} + r^{-1}|u|_{p,k} + \sum_{\gamma}|{\Gamma_{\mathrm{r}}}^{\gamma}|_{p,k} 
\\
& \quad + \sum_{p_1+p_2=p\atop k_1+k_2=k}\big(|\Hreff|_{p_1,k_1}|\del u|_{p_2,k_2} + |\del\Hreff|_{p_1,k_1}|u|_{p_2,k_2} + |u|_{p_1,k_1}|\del u|_{p_2,k_2}\big),
\endaligned
\end{equation}
\begin{equation}
\aligned
|\del_t\del_t h^{\Ncal00}|_{p,k}& \lesssim \, \underline{|\del_t\Gamma^{\lambda}|_{p,k}} + |\del\dels^{\Ncal}h|_{p,k} + r^{-1}|h|_{p,k} 
\\
& \quad + \sum_{p_1+p_2=p\atop k_1+k_2=k}\big(|h|_{p_1,k_1}|\del h|_{p_2,k_2} + |\del\del h|_{p_1,k_1}|\Hreff|_{p_2,k_2} + |\del h|_{p_1,k_1}|\del h|_{p_2,k_2}\big)
\endaligned
\end{equation}
with the \underline{underlined} terms in supplementary. We observe that these new terms do not exceed the main term of interest, viz. $\delsN u$ or $\del\delsN h$ in neither weighted$-L^2$ nor pointwise sens.

\subsection{Sharp decay estimates and improved energy estimates}
\label{subsec2-05-oct-2025}
We follow the argument of \cite{PLF-YM-PDE} on the $u_{\alpha\beta}$, and we will arrive at the following sharp decay estimates (cf. \cite[Proposition~16.1]{PLF-YM-PDE}):
\begin{equation}\label{eq3-05-oct-2025}
\aligned
|\del \slashed{h}|_{N-4,k} + |\del \slashed{u}|_{N-4,k}
\lesssim_{\delta,\ell}& \, C_1\eps \la r-t\ra^{-1/2- \delta/2} r^{-1+k\theta},
\\
& \text{ in }  \Mcal_{\ell,[s_0,s_1]}^{\near}, \quad 0\leq k\leq N-4,
\\
|\del\del \slashed{h}|_{N-5,k} + |\del \slashed{u}|_{N-5,k}
\lesssim_{\delta,\ell}& \, C_1\eps \la r-t\ra^{-1/2- \delta/2} r^{-1+k\theta},
\\
& \text{ in }  \Mcal_{\ell,[s_0,s_1]}^{\near}, \quad 0\leq k\leq N-5.
\endaligned
\end{equation}
\begin{equation}\label{eq4-05-oct-2025}
|u_{\source}|_k\lesssim C_1\eps r^{-1+(3N+1)\theta},\quad \text{ in }  \Mcal^{\ME}_{[s_0,s_1]},\quad 0\leq N-5.
\end{equation}
Here 
$$
\Mcal_{\ell,[s_0,s_1]}^{\near} = \Mcal^{\ME}_{[s_0,s_1]}\cap\{r\leq (1- \ell)^{-1}t\},
$$
and $u_{\source,\alpha\beta}$ the nonlinear source contribution of $u_{\alpha\beta}$, viz.,
$$
\aligned
g^{\mu\nu}\del_{\mu}\del_{\nu}u_{\source,\alpha\beta} & =\Fbb(g,g;\del u,\del u)_{\alpha\beta} 
+ u^{\mu\nu}\del_{\mu}\del_{\nu}\gref_{\alpha\beta} 
+ \Sbb(g,\del g,\Psi,\widehat{\del}\Psi)_{\alpha\beta}
\\
& \quad + 2\Fbb(g,g;\del\gref,\del u)_{\alpha\beta} 
+ 2\Fbb(u,\gref;\del \gref,\del \gref)_{\alpha\beta} 
+ \Fbb(u,u;\del\gref,\del\gref)_{\alpha\beta}
\\
& \quad -2\big(\mathbb{G}(g,\del g,\Gamma,\del\Gamma)_{\alpha\beta} - \mathbb{G}(\gref,\del \gref,\Gamma_{\mathrm{r}},\del\Gamma_{\mathrm{r}})_{\alpha\beta}\big),
\endaligned
$$
with zero initial data:
$$
u_{\source,\alpha\beta}|_{\Mcal_{s_0}} = \del_tu_{\source,\alpha\beta}|_{\Mcal_{s_0}} = 0.
$$

We also need the following result.
\begin{lemma}
Under the bootstrap assumptions, the hypersurface $\Ccal_{[s_0,s_1]}$ is spacelike.
\end{lemma}

\begin{proof}
We note that $\{\del_t+\del_r,\Omega_{ab}\}$ generates the tangent space of $\Ccal_{[s_0,s_1]}$, while $\Omega_{ab}$ are obviously spacelike. We then recall \eqref{eq2-02-oct-2025(l)}:
$$
g(\del_t+\del_r,\del_t+\del_r) = -h^{\Ncal00} + h(|h|^2)\geq \frac{1}{3}C_L\eps r^{-1+\theta} + C(C_1\eps)^2r^{-2+4\theta+\delta/2}.
$$
Provided that $\theta,\delta$ sufficiently small, and $\eps$ sufficiently small, 
$$
g(\del_t+\del_r,\del_t+\del_r)>0. \qedhere
$$
\end{proof}

Equipped with these results together with the integrable energy bounds on the new terms, i.e., Propositions~\ref{prop1-04-oct-2025} and \ref{prop1-03-oct-2025}, we arrive at the improved energy estimates on $u_{\alpha\beta}$. 


\subsection{Improved energy estimates on spinor field}
\label{subsec1-11-oct-2025}

We first establish the estimate on $\Psihat_a$.

\begin{proposition}\label{prop1-15-oct-2025}
Under the bootstrap assumptions, the following estimate holds:
\begin{equation}\label{eq4-11-oct-2025}
\big\|\zeta^{1/2}\la r-t\ra^{\mu}[\Psihat]_{N}\big\|_{L^2(\Mcal^{\near}_s)}
\lesssim C_1\eps,
\end{equation}
\begin{equation}\label{eq5-11-oct-2025}
\big\|t\zeta^{1/2}\la r-t\ra^{\mu}[\Psihat]_{p,p-1}\big\|_{L^2(\Mcal^{\near}_s)}\lesssim 
\begin{cases}
C_1\eps s^{1+\delta},\quad& p\leq N,
\\
C_1\eps s^{\delta},\quad & p\leq N-5.
\end{cases}
\end{equation}
\end{proposition}
\begin{proof}

We observe that, thanks to \eqref{eq2-13-oct-2025}, 
\begin{equation}
\|t\zeta^{1/2}\la r-t\ra^{\mu}[\Psihat]_{p-1}\|_{L^2(\Mcal^{\near}_s)}
\lesssim 
\begin{cases}
C_1\eps s^{1+\delta},\quad& p\leq N,
\\
C_1\eps s^{\delta},\quad & p\leq N-5.
\end{cases}
\end{equation}

For the top-order estimates, we consider \eqref{eq5-11-oct-2025}. Let $\mathscr{Z}^I$ be an admissible operator with $\ord(I) = N$ and $\rank(I) = k\leq N-1$. Then there exist admissible operators $\mathscr{Z}^{I_1},\mathscr{Z}^{I_2}$ such that
\be
\mathscr{Z}^I = \mathscr{Z}^{I_1}\circ\widehat{\del_{\delta}}\circ\mathscr{Z}^{I_2},
\quad
\ord(I_1)+\ord(I_2) = N-1,
\quad
\rank(I_1) + \rank(I_2) = k\leq N-1. 
\ee
We then apply \eqref{eq2-16-aout-2025} and obtain
\begin{equation}
\big|\mathscr{Z}^I\Psi\big|_{\vec{n}}\lesssim_N \sum_{\ord(J)\leq N-1\atop rank(J)\leq k}
|\widehat{\del_{\delta}}\mathscr{Z}^J\Psi|_{\vec{n}}
+ \Hcom_{N-1}[\Psi].
\end{equation}
Thus
\begin{equation}
\aligned
|\mathscr{Z}^I\Psihat_a|_{\vec{n}} = \big|\mathscr{Z}^I(t^{-1}\widehat{L_a})\Psi\big|_{\vec{n}}
\lesssim_N& t^{-1}\sum_{\delta}\sum_{\ord(J)\leq N-1\atop \rank(J)\leq N-1}
\big|\widehat{\del}_{\delta}\mathscr{Z}^J\widehat{L_a}\Psi\big|_{\vec{n}}
+ t^{-1}\Hcom_{N-1}[\widehat{L_a}\Psi]
\\
& \lesssim 
t^{-1}\sum_{\delta}\sum_{\ord(K)\leq N\atop \rank(K)\leq N}
\big|\widehat{\del}_{\delta}\mathscr{Z}^K\Psi\big|_{\vec{n}}
+ t^{-1}\Hcom_{N}[\Psi].
\endaligned
\end{equation}
Then by the energy bounds \eqref{eq1-02-oct-2025-H} and \eqref{eq1-02-oct-2025-L} one obtain the desired $L^2$ bounds on the linear terms in the right-hand side. For the high-order terms we also need the Sobolev decay estimates \eqref{eq4-03-oct-2025}, \eqref{eq7-03-oct-2025} and \eqref{eq6-04-oct-2025}. We omit the details.

We now turn our attention to to \eqref{eq4-11-oct-2025}. The only case left our is when $\ord(I) = \rank(I) = N$. We perform the following calculation:
\be
\mathscr{Z}^I(\Psihat_a) = \mathscr{Z}^I\big(t^{-1}\widehat{L_a}\Psi\big) 
= \sum_{I_1\odot I_2 = I}\mathscr{Z}^{I_1}(t^{-1})\mathscr{Z}^{I_2}(\widehat{L_a}\Psi).
\ee
By the homogeneity of $t^{-1}$, one obtains, in $\Mcal^{\near}_{[s_0,s_1]}$:
\be
|\mathscr{Z}^I\Psihat_a|_{\vec{n}}\lesssim_N t^{-1}\sum_{\ord(J)\leq N+1}\mathscr{Z}^{J}\Psi,
\ee
where $\rank(J) = \ord(J)$. Thus we obtain 
\begin{equation}
|\mathscr{Z}^I\Psihat_a|_{\vec{n}}\lesssim_N t^{-1}\sum_{a}\sum_{\ord(J') = N\atop \rank(J')=N}
\widehat{L_a}\mathscr{Z}^{J'}\Psi + t^{-1}\sum_{\ord(J')\leq N\atop \rank(J')\leq N}\Psi.
\end{equation}
Then we derive an $L^2$ estimate. Thanks to relation
\be
s\lesssim t^{1/2},\quad \text{in}\quad\Mcal^{\ME}_{[s_0,s_1]},
\ee
the second term in the right-hand side can be bounded via \eqref{eq1-02-oct-2025-H}, \eqref{eq1-02-oct-2025-L}. The first term in the right-hand side is bounded by \eqref{eq1-11-oct-2025}.
\end{proof}

We now turn our attention to the component $t^{-1}\widehat{L_a}\mathscr{Z}^J\Psi$ with $\ord(J) = \rank(J)\leq N$, 
and we establish the following estimates.

\begin{lemma}
Let $\mathscr{Z}^J$ be such that $\ord(J) = \rank(J)$ and set
$\mathscr{Z}^{J_a} := L_a\mathscr{Z}^J$. Thenone has 
\begin{equation}\label{eq6-10-oct-2025}
\aligned
& \int_{s_0}^{s_1}s\Ebf_{\mu}^{\ME}(s,t^{-1}\mathscr{Z}^{J_a}\Psi)^{1/2}
\|t^{-1}\omega^{\mu}\zeta^2\zetab^{-1/2}|\Phi^{J_a}|_{\vec{n}}\|_{L^2(\Mcal^{\ME}_s)}\diff s
\\
& \lesssim_{\delta}(C_1\eps)^2
\int_{s_0}^{s_1}s^{-1- \delta}\Ebf_{\mu}^{\ME}(s,t^{-1}\mathscr{Z}^{J_a}\Psi)^{1/2}\diff s,
\endaligned
\end{equation}
\begin{equation}\label{eq7-10-oct-2025}
\int_{s_0}^{s_1}s\|t^{-1}|h|\|_{L^{\infty}(\Mcal^{\ME}_s)}\Ebf_{\mu}^{\ME}(s,t^{-1}\mathscr{Z}^{J_a}\Psi) \,\diff s
\lesssim C_1\eps \int_{s_0}^{s_1}s^{-1- \delta}\Ebf_{\mu}^{\ME}(s,t^{-1}\mathscr{Z}^{J_a}\Psi) \,\diff s.
\end{equation}
\end{lemma}

\begin{proof}
For \eqref{eq6-10-oct-2025}, we apply \eqref{eq4-09-oct-2025}.
$$
\aligned
\zetab|\Phi^{J_a}|_{\vec{n}}& \lesssim \sum_{p_1+p_2\leq p+1\atop k_1+k_2\leq p+1}
[\del\Psi]_{p_1-1,k_1-1}|h|_{p_2+1,k_2+1} 
+ \Hcom_p[\Psi] +  \Hwave_p[\Psi]
\\
& \lesssim \sum_{p_1+p_2\leq p\atop k_1+k_2\leq p}
[\del\Psi]_{p_1,k_1}|\Hreff|_{p_2+1,k_2+1}
+\sum_{p_1+p_2\leq p\atop k_1+k_2\leq p}[\del\Psi]_{p_1,k_1}|u|_{p_2+1,k_2+1}
\\
& \quad +\Hcom_p[\Psi] + \Hwave_p[\Psi]
=: T_1 + T_2 + \Hcom_p[\Psi] + \Hwave_p[\Psi].
\endaligned
$$

For the term $T_1$, we apply \eqref{eq1-03-oct-2025},
\begin{equation}\label{eq3-10-oct-2025}
\aligned
\big\|t^{-1}\omega^{\mu}\zeta^2\zetab^{-3/2}|T_1|\big\|_{L^2(\Mcal^{\ME}_s)}
& \lesssim C_1\eps \| t^{-1}r^{-1+\theta}\zetab^{1/2}\omega^{\mu}[\del \Psi]_{N}\|_{L^2(\Mcal^{\ME}_s)}
\\
& \lesssim C_1\eps s^{-4+2\theta}\Ebf^{\ME,N}_{\kappa}(s,\Psi)^{1/2}
\lesssim(C_1\eps)^2 s^{-3+\delta+2\theta}.
\endaligned
\end{equation}
Term $T_2$ is a bit more complicated. When $p_1\leq N-10$, we apply \eqref{eq6-04-oct-2025} to $[\del\Psi]_{p_1}$:
$$
\aligned
& \|t^{-1}\omega^{\mu}\zeta^2\zetab^{-3/2}[\del\Psi]_{p_1,k_1}|u|_{p_2+1,k_2+1}\|_{L^2(\Mcal^{\ME}_s)}
\\
& \lesssim 
C_1\eps s^{2\delta}\|t^{-1}\zeta^{1- \delta}\la r-t\ra r^{-2} |u|_{N}\|_{L^2(\Mcal^{\ME}_s)}
\\
& \quad +C_1\eps s^{2\delta}\sum_{\ord(J')\leq N}\sum_b\|t^{-1}\zeta^{1- \delta}\la r-t\ra r^{-2} |L_b\mathscr{Z}^{J'}u|\|_{L^2(\Mcal^{\ME}_s)}
\\
& \quad +C_1\eps s^{2\delta}\sum_{\ord(J')\leq N}\sum_{\alpha}\|t^{-1}\zeta^{1- \delta}\la r-t\ra r^{-2} |\del_{\alpha}\mathscr{Z}^{J'}u|\|_{L^2(\Mcal^{\ME}_s)}
\\
& \lesssim C_1\eps s^{2\delta}\|t^{-1}\zeta^{1- \delta}\la r-t\ra^{\kappa}r^{-2}  \la r-t\ra^{1- \kappa}|u|_{N}\|_{L^2(\Mcal^{\ME}_s)}
\\
& \quad + C_1\eps s^{2\delta}\sum_{\ord(J')\leq N}\sum_b
\|t^{-1}\zeta^{1- \delta}\la r-t\ra^{1- \kappa} r^{-1} |\la r-t \ra^{\kappa}\delN_b\mathscr{Z}^{J'}u|\|_{L^2(\Mcal^{\ME}_s)}
\\
& \quad +C_1\eps s^{2\delta}\|t^{-1}\zeta^{1- \delta}r^{-2}|\del \mathscr{Z}^{J'}u|\|_{L^2(\Mcal^{\ME}_s)}
\\
& \quad + C_1\eps s^{2\delta}\|t^{-1}\zeta^{1- \delta}\la r-t\ra r^{-2} |\del u|_N  \|_{L^2(\Mcal^{\ME}_s)}
\\
\lesssim_{\delta}& (C_1\eps)^2 s^{-2-2\kappa + 4\delta}.
\endaligned
$$
When $p_2\leq N-2$, we apply \eqref{eq5-03-oct-2025} oto$|u|$ and obtain
$$
\aligned
& \|t^{-1}\omega^{\mu}\zeta^2\zetab^{-3/2}[\del\Psi]_{p_1,k_1}|u|_{p_2+1,k_2+1}\|_{L^2(\Mcal^{\ME}_s)}
\\
\lesssim_{\delta}& C_1\eps s^{\delta}
\|t^{-1}r^{-1}\la r-t \ra^{1- \kappa}\omega^{\mu}\zetab^{1/2}[\del\Psi]_{N}\|_{L^2(\Mcal^{\ME}_s)}
\\
\lesssim_{\delta}& C_1\eps s^{-2-2\kappa + \delta}\Ebf_{\kappa}^{\ME,N}(s,\Psi)^{1/2}
\lesssim_{\delta} (C_1\eps)^2 s^{-2- \delta}.
\endaligned
$$
So we have
\begin{equation}
\big\|t^{-1}\omega^{\kappa}\zeta^2\zetab^{-3/2}|T_2|\big\|_{L^2(\Mcal^{\ME}_s)}
\lesssim_{\delta} (C_1\eps)^2 s^{-2- \delta}.
\end{equation}
The estimates on high-order terms $\Hcom, \Hwave$ are easier. We omit the details. This concludes the derivation of~\eqref{eq6-10-oct-2025}. For \eqref{eq7-10-oct-2025}, we only need to apply \eqref{eq7-03-oct-2025}.
\end{proof}


\begin{lemma}
Under the bootstrap assumptions, the following energy estimate holds:
\begin{equation}
\sum_{\ord(J) = p\atop \rank(J) = p}\Ebf_{\mu}^{\ME}(s,t^{-1}\widehat{L_a}(\mathscr{Z}^J\Psi))^{1/2}
\leq C_0\eps + C(C_1\eps)^{3/2},\quad p\leq N.
\end{equation}
\end{lemma}
\begin{proof}
Recall the energy estimate \eqref{eq2-10-oct-2025} and \eqref{eq6-10-oct-2025},
\begin{equation}
\aligned
& \Ebf^{\ME}_{\mu}(s_1,t^{-1}\widehat{L_a}\mathscr{Z}^J\Psi) 
- \Ebf^{\ME}_{\mu}(s_0,t^{-1}\widehat{L_a}\mathscr{Z}^J\Psi) 
+ \Ebf^{\Ccal}(s_0,s_1;t^{-1}\widehat{L_a}\mathscr{Z}^J\Psi)
\\
& \quad + \int_{s_0}^{s_1}\int_{\Mcal^{\ME}_s}
t^{-1}\big(\omega^{\mu}\zeta|t^{-1}\widehat{L_a}\mathscr{Z}^J\Psi|_{\del_t}^2\big) 
+ \la r-t\ra^{-1}\big(\omega^{\mu}\zeta|t^{-1}\widehat{L_a}\mathscr{Z}^J\Psi|_{\vec{\gamma}}^2\big)
\,\diff s
\\
& \lesssim_{\delta} (C_1\eps)^3,
\endaligned
\end{equation}
for $\ord(J) = \rank(J) = N$. Then we establish the uniform-in-time $L^2$ bounds on $t^{-1}\widehat{L_a}\mathscr{Z}^J\Psihat_a$.
\end{proof}


Consequently, if we demand
\begin{equation}
C_1>2C_0,\quad \eps\leq \Big(\frac{C_1-2C_0}{2CC_1^{3/2}}\Big)^2 
\end{equation}
it follows that \eqref{eq3-11-oct-2025} is established.

Finally, we sketch the energy estimates on $\Psi$ and $\widehat{\del_{\delta}}\Psi$. These rely on the weighted $L^2$ estimates on $\Phi_{\mu}^I$ and $\Phi^I$ as we have done in \cite[Section~18]{PLF-YM-PDE} for real-valued field. The commutators are also decomposed into {\bf easy} terms and {\bf hard} terms. For simplicity in the discussion, we introduce the following classification. 
\begin{subequations}
\noindent{$\bullet$} ``Near-light-cone subcritical terms'' for $\Phi^I$:
\begin{equation}\label{eq2-15-oct-2025}
\aligned
& \zetab^{-1}\sum_{p_1+p_2\leq p\atop k_1+k_2\leq k}
\big([\del\Psi]_{p_1-1,k_1} + [\Psi]_{p_1-1,k_1}\big)\Big(\frac{\la r-t\ra}{r}|\del h|_{p_2,k_2} + |\delsN h|_{p_2,k_2}\Big),
\\
& \zetab^{-1}t^{-1}\sum_{p_1+p_2\leq p\atop k_1+k_2\leq k}[\Psi]_{p_1,k_1+1}|\del h|_{p_2,k_2},
\quad
\zetab^{-1}\Hcom_{p-1}[\Psi] + \zetab^{-1}\Hcom_{p-1}[\del\Psi] + \zetab^{-1}\Hwave_{p-1}[\Psi],
\endaligned
\end{equation}
\noindent{$\bullet$} ``Near-light-cone subcritical terms'' for $\Phi_{\mu}^I$:
\begin{equation}\label{eq3-15-oct-2025}
\aligned
& \zetab^{-1}\sum_{p_1+p_2\leq p\atop k_1+k_2\leq k}
\big([\del\Psi]_{p_1,k_1} + [\Psi]_{p_1,k_1}\big)
\Big(\frac{\la r-t\ra}{r}|\del h|_{p_2,k_2} + |\delsN h|_{p_2,k_2}\Big),
\\
&
\zetab^{-1}t^{-1}\hspace{-0.3cm}\sum_{p_1+p_2\leq p\atop k_1+k_2\leq k}
[\del\Psi]_{p_1,k_1+1}|\del h|_{p_2,k_2},
\quad
\zetab^{-1}\sum_{p_1+p_2\leq p\atop k_1+k_2\leq k}[\Psihat]_{p_1,k_1}|\del h|_{p_2,k_2},
\\
& \zetab^{-1}\Hcom_p[\Psi] + \zetab^{-1}\Hwave_p[\Psi].
\endaligned
\end{equation}

\noindent{$\bullet$} ``Near-light-cone absorbable terms'' for $\Phi^I$ and $\Phi_{\mu}^I$:
\begin{equation}
[\del\Psi]_{\vec{\gamma},p,k}\sum_{p_1+p_2\leq p\atop k_1+k_2\leq k}
[\Psi]_{\vec{\gamma},p_1,k_1}|\del h|_{p_2,k_2},
\quad
[\Psi]_{\vec{\gamma},p,k}\sum_{p_1+p_2\leq p\atop k_1+k_2\leq k}
[\Psi]_{\vec{\gamma},p_1,k_1}|\del h|_{p_2,k_2}
\end{equation}

\noindent{$\bullet$} ``Away-from-light-cone subcritical terms'' for $\Phi^I$ and $\Phi_{\mu}^I$,
\begin{equation}
\sum_{p_1+p_2\leq p\atop k_1+k_2\leq k}[\del\Psi]_{p_1,k_1}|\del h|_{p_2,k_2},
\quad
\Hcom_{p}[\Psi] + \Hwave_{p}[\Psi]
\end{equation}

\noindent{$\bullet$} ``Hard terms'' for $\Phi^I$ and $\Phi_{\mu}^I$:
\begin{equation}
\zetab^{-1}\sum_{p_1+p_2\leq p\atop k_1+k_2\leq k}\big([\del\Psi]_{p_1-1,k_1-1} + [\del\del\Psi]\big)|h|_{p_2+1,k_2+1}.
\end{equation}
\end{subequations}
The ``Hard terms'' correspond to the hard terms in \cite[Proposition~9.2]{PLF-YM-PDE}. The remaining terms correspond to the ``easy terms''.

Let us compare the above terms and the commutator in the scalar case. For the ``hard'' terms, we write them into the following form:
\begin{equation}\label{eq1-15-oct-2025}
\zetab^{-1}\sum_{p_1+p_2\leq p\atop k_1+k_2\leq k}[\del\Psi]_{p_1-1,k_1-1}|L h|_{p_2,k_2},
\quad
\zetab^{-1}[\del\Psi]_{p-1,k-1}|h|
\end{equation}
which enjoy the same hierarchy structure of $W^{\bf hard}_{p,k}$ there in \cite[Proposition~9.2]{PLF-YM-PDE}, except that there we have distinguished different components of $H$ due to the fact that only the components $\{rr, 0a, 00\}$ of the reference $h^{\star}$ therein enjoy a decay rate $r^{-1+\theta}$ ({\sl tame decay conditions}). But this decay rate is supposed for general components of reference in the present case (and this is valid for the vacuum solution  of \cite{LR2}), thus we do not need this distinction. Then the treatment of \eqref{eq1-15-oct-2025} is essentially the same as we have done in \cite[Proposition~18.2]{PLF-YM-PDE}. The reader may worry about the unpleasant weight $\zetab^{-1}$, however, it will not bring difficulty. Namely, if we consider \eqref{eq4-10-oct-2025(l)}, there is a standard weight $\zeta^2\zetab^{-1/2}\lesssim \zeta^{3/2}$. Then when one makes $L^2$ estimates on $[\del\Psi]_{p,k}$ or $[\Psi]_{p,k}$, it remains a $\zeta^{1/2}\leq \zetab^{1/2}$ which is exactly the weight in the energy density for spinor fields, see \eqref{equa-28-sept-2025a} (This is also the case in \cite{PLF-YM-PDE}, where the weight in the energy sour term $J\zeta^{-1}$ gives a $\zeta$, which is exactly the weight demanded in the scalar energy estimate). When one does $L^2$ estimates on $|L u|$, one can always treat regularity for the weight $\zetab$ by Proposition~\ref{prop1-24-july-2025}. In practical, this means for example that one can  always apply \eqref{eq5-04-oct-2025} instead of \eqref{eq5-04-oct-2025} to obtain one order $\zetab$, provided that $N$ sufficiently large.

For the ``near-light-cone subcritical terms'', we remark that all terms in \eqref{eq2-15-oct-2025} and the ones other than the third in \eqref{eq3-15-oct-2025} are essentially the same to $W^{\bf easy}_{p,k}$ defined in \cite[Proposition~9.2]{PLF-YM-PDE}. The key structure is that, for metric, we have the standard factor $\frac{\la r-t\ra}{r}$ for general components of the gradient, and the conical derivatives $\delsN h$, which enjoy a decay faster that $t^{-1}$, while in  $W^{\bf easy}_{p,k}$ if we apply the wave coordinate conditions, we also turn $\del h^{\Ncal00}$ into $|\delsN h|$ plus high-order corrections. We refer the interested reader to~\cite[Lemma~11.13, Lemma~11.14]{PLF-YM-PDE}. 

The ``away-from-light-cone subcritical terms'' are integrated in $\Mcal^{\ME}\cap\{r\geq4t/3\}$, in which
\be
\la r-t\ra^{-1}\lesssim r^{-1}.
\ee
Thus $|\del h|$ enjoy sufficient decay, which turn the estimates on these terms trivial.

Finally, we arrive the essential new terms, which are ``near-light-cone absorbable'' terms and the third one in \eqref{eq3-15-oct-2025}.
\begin{lemma}
Under the bootstrap assumptions, the following energy estimate holds for $p_1+p_2\leq p,k_1+k_2\leq k$:
\begin{equation}
\big\|\omega^{\mu}\zeta^2\zetab^{-3/2}[\Psihat]_{p_1,k_1}|\del h|_{p_2,k_2}\big\|_{L^2(\Mcal^{\ME}_s)}
\lesssim 
\begin{cases}
(C_1\eps)^2 s^{-2+\delta},\quad & p\leq N,
\\
(C_1\eps)^2 s^{-4+\delta},\quad & p\leq N-5.
\end{cases}
\end{equation}
Thus
\begin{equation}\label{eq5-15-oct-2025}
\aligned
& \int_{s_0}^{s_1} s\Ebf_{\mu}^{\ME,p,k}(s,\Psi)^{1/2}\big\|\omega^{\mu}\zeta^2\zetab^{-3/2}[\Psihat]_{p_1,k_1}|\del h|_{p_2,k_2}\big\|_{L^2(\Mcal^{\ME}_s)} \,\diff s
\\
& \lesssim
\begin{cases}
(C_1\eps)^3s_1^{1+2\delta},\quad &p\leq N,
\\
(C_1\eps)^3 ,\quad &p\leq N-5.
\end{cases}
\endaligned
\end{equation}
\end{lemma}

\begin{proof}
The estimate relies on Proposition~\ref{prop1-15-oct-2025}. In fact, let
$$
[\Psihat]_{p_1,k_1}|\del h|_{p_2,k_2} \leq [\Psihat]_{p_1,k_1}|\del \Hreff|_{p_2,k_2} 
+ [\Psihat]_{p_1,k_1}|\del u|_{p_2,k_2} = : T_1 + T_2.
$$
For the term $T_1$, we recall \eqref{eq1-03-oct-2025}
$$
\aligned
\big\|\omega^{\mu}\zeta^2\zetab^{-3/2}|T_1|\big\|_{L^2(\Mcal^{\ME}_s)}
& \lesssim 
\begin{cases}
C_1\eps s^{-2+2(N+1)\theta}\big\|\zeta^{1/2}\omega^{\mu}[\Psihat]_{p}\big\|_{L^2(\Mcal^{\ME}_s)},
\quad&  p\leq N,
\\
C_1\eps s^{-4+2(N+1)\theta}\big\|\zeta^{1/2}\omega^{\mu}[\Psihat]_p\big\|_{L^2(\Mcal^{\ME}_s)},
\quad &p\leq N-5,
\end{cases}
\\
& \lesssim 
\begin{cases}
(C_1\eps)^2 s^{-2+2(N+1)\theta},\quad &p\leq N,
\\
(C_1\eps)^2 s^{-4+2(N+1)\theta},\quad &p\leq N-5.
\end{cases}
\endaligned
$$

For the term $T_2$, we need to remark that when $p_2\leq N-3$, we apply \eqref{eq3-03-oct-2025} on $|\del u|_{p_2,k_2}$ and obtain
$$
\aligned
\big\|\omega^{\mu}\zeta^2\zetab^{-3/2}|T_2|\big\|_{L^2(\Mcal^{\ME}_s)}
\lesssim &
\begin{cases}
C_1\eps s^{-2+\delta}\big\|\zeta^{1/2}\omega^{\mu}[\Psihat]_p\big\|_{L^2(\Mcal^{\ME}_s)},\quad p\leq N,
\\
C_1\eps s^{-4+\delta}\big\|\zeta^{1/2}\omega^{\mu}[\Psihat]_p\big\|_{L^2(\Mcal^{\ME}_s)},\quad p\leq N-5,
\end{cases}
\\
& \lesssim 
\begin{cases}
(C_1\eps)^2s^{-2+\delta},\quad &p\leq N,
\\
(C_1\eps)^2s^{-4+\delta},\quad &p\leq N-5.
\end{cases}
\endaligned
$$
When $p_1\leq N-9$, we recall \eqref{eq2-13-oct-2025} and \eqref{eq6-04-oct-2025},
$$
\aligned
\big\|\omega^{\mu}\zeta^2\zetab^{-3/2}|T_2|\big\|_{L^2(\Mcal^{\ME}_s)}
& \lesssim 
C_1\eps s^{-4-2\kappa+2\delta}\big\|\zeta^{1/2}\omega^{\kappa}|\del u|_{N}\big\|_{L^2(\Mcal^{\ME}_s)}
\lesssim (C_1\eps)^2s^{-4+\delta}.
\endaligned
$$
Recalling \eqref{eq4-15-oct-2025}, we arrive the desired estimate.
\end{proof}
\begin{remark}
From \eqref{eq4-10-oct-2025(l)} of Proposition~\ref{prop1-05-oct-2025}, we observe that the contribution of this term in the energy estimate is bounded by \eqref{eq5-15-oct-2025}. Remark particularly that this brings an $s^{1/2 + \delta}$ increasing rate on $\Ebf_{\mu}^{\ME,p,k}(s,\Psi)^{1/2}$ when $p\geq N-4$. But it is still weaker than the bootstrap assumption which is $s^{1+\delta}$. When $p\leq N-5$, the uniform contribution of this term is still benign as far as the $s^{\delta}$ increasing rate in \eqref{eq1-02-oct-2025-L} is concerned.
\end{remark}


Then we turn our attention to the ``near-light-cone absorbable'' terms.

\begin{lemma}
Under the bootstrap assumptions, the following energy estimate holds for $p_1+p_2\leq p,k_1+k_2\leq k$:
\begin{equation}\label{eq6-15-oct-2025}
\aligned
& \int_{s_0}^{s_1}s\int_{\Mcal^{\near}_s}\omega^{2\mu}\zeta^2
\big([\del\Psi]_{\vec{\gamma},p,k} + [\Psi]_{\vec{\gamma},p,k}\big)
[\Psi]_{\vec{\gamma},p_1,k_1}|\del h|_{p_2,k_2}\,\diff x\diff s
\\
& \lesssim_{\delta} C_1\eps\int_{s_0}^{s_1}s\int_{\Mcal^{\ME}_s}
\frac{(\omega^{\mu}\zeta[\del\Psi]_{\vec{\gamma},p,k})^2 + (\omega^{\mu}\zeta[\Psi]_{\vec{\gamma},p,k})^2}{\la r-t\ra} \diff x\diff s + (C_1\eps)^3.
\endaligned
\end{equation}
\end{lemma}
\begin{proof}
We remark that, when $p_2\leq N-3$, we apply \eqref{eq4-03-oct-2025} and obtain
$$
\aligned
& \int_{\Mcal^{\near}_s}\omega^{2\mu}\zeta^2[\del\Psi]_{\vec{\gamma},p,k}
[\Psi]_{\vec{\gamma},p_1,k_1}|\del h|_{p_2,k_2}\,\diff x
\\
& \lesssim C_1\eps s^{\delta}\int_{\Mcal^{\near}_s}r^{-1}\omega^{2\mu- \kappa}\zeta^2 [\del\Psi]_{\vec{\gamma},p,k}
[\Psi]_{\vec{\gamma},p_1,k_1}\,\diff x
\\
& \lesssim C_1\eps s^{\delta}\int_{\Mcal^{\near}_s} r^{-1}\la r-t\ra^{1- \kappa}
(\la r-t\ra^{-1/2}\omega^{\mu}\zeta[\del\Psi]_{\vec{\gamma},p,k})\,
(\la r-t\ra^{-1/2}\omega^{\mu}\zeta[\Psi]_{\vec{\gamma},p,k})
\\
& \lesssim C_1\eps s^{\delta} \int_{\Mcal^{\near}_s}r^{-1}\la r-t\ra^{1- \kappa}
(\la r-t\ra^{-1/2}\omega^{\mu}\zeta[\del\Psi]_{\vec{\gamma},p,k})^2 \diff x
\\
& \quad+ C_1\eps s^{\delta}\int_{\Mcal^{\near}_s}r^{-1}\la r-t\ra^{1- \kappa}
(\la r-t\ra^{-1/2}\omega^{\mu}\zeta[\Psi]_{\vec{\gamma},p,k})^2 \diff x
\\
& \lesssim C_1\eps s^{-2\kappa+\delta}\int_{\Mcal^{\near}_{s}}
\frac{(\omega^{\mu}\zeta[\del\Psi]_{\vec{\gamma},p,k})^2 + (\omega^{\mu}\zeta[\Psi]_{\vec{\gamma},p,k})^2}{\la r-t\ra}\,\diff x.
\endaligned
$$

When $p_1\leq N-9$, we remark that $\del h = \del \Hreff + \del u$. The integral
$$
\int_{\Mcal^{\near}_s}\omega^{2\mu}\zeta^2 [\del\Psi]_{\vec{\gamma},p,k}
[\Psi]_{\vec{\gamma},p_1,k_1}|\del \Hreff|_{p_2,k_2}\,\diff x
$$ 
is bounded in the same manner as above by applying \eqref{eq1-03-oct-2025} on $|\del \Hreff|$ (which is better than the decay of $|\del h|$ in the above estimate). We omit the details. Then we focus on the part concerning $\del u$. Thanks to \eqref{eq6-04-oct-2025},
$$
\aligned
& \int_{\Mcal^{\near}_s}\omega^{2\mu}\zeta^2 [\del\Psi]_{\vec{\gamma},p,k}
[\Psi]_{\vec{\gamma},p_1,k_1}|\del u|_{p_2,k_2}\,\diff x
\\
& \lesssim C_1\eps s^{2\delta}\int_{\Mcal^{\near}_s}\zeta^{1/2- \delta}\la r-t\ra^{1- \mu}r^{-2}\omega^{2\mu}\zeta^2 [\del\Psi]_{\vec{\gamma},p,k}
|\del u|_{p_2,k_2}\,\diff x
\\
& \lesssim C_1\eps s^{2\delta}\int_{\Mcal^{\near}_s}r^{-2}\la r-t\ra^{3/2- \kappa}\big(\la r-t\ra^{-1/2}\omega^{\mu}\zeta[\del\Psi]_{\vec{\gamma},p,k}\big)\big(\omega^{\kappa}\zeta|\del u|_{p,k}\big)\,\diff x
\\
& \lesssim C_1\eps s^{-1-2\kappa + 2\delta}\int_{\Mcal^{\near}_s}\frac{(\omega^{\mu}\zeta[\del\Psi]_{\vec{\gamma},p,k})^2}{\la r-t\ra}\,\diff x
+C_1\eps s^{-1-2\kappa + 2\delta}\Ebf_{\kappa}^{\ME,p,k}(s,u).
\endaligned
$$
The last term leads in the right-hand side leads us to the integrable bound $(C_1\eps)^3$ in the right-hand side of \eqref{eq6-15-oct-2025}.
\end{proof}

\begin{remark}
When we consider the energy contribution of the above ``near-light-cone absorbable'' terms, we recall \eqref{eq4-10-oct-2025(l)}. When $C_1\eps$ sufficiently small, the first term in the right-hand side of \eqref{eq6-15-oct-2025} can be absorbed by the third term in the left-hand side of \eqref{eq4-10-oct-2025(l)}. So the final contribution of this term is also a uniform in time term $(C_1\eps)^3$ which is also benign as far as  the bootstrap assumptions are concerned.
\end{remark}


\clearpage

\part{Technical material}
\label{part-three}

\section{Spin structure and algebraic properties} 
\label{section=N17}

\subsection{ Basic concepts} 
\label{section=31}

{

\paragraph{Aim of this section.}

We develop here the fundamental concepts of calculus on spinor bundles in a \emph{curved spacetime} with trivial tangent bundle, with the ultimate goal of defining and analyzing Dirac spinor fields. We follow the presentation and (mostly) the notation in Hamilton's textbook~\cite{Hamilton-2017}.  In this Section~\ref{section=31}, we first present the necessary algebraic framework, including the Lorentz Clifford algebra and the associated gamma matrices, which underpin the notion of spinors. The geometric setting relies on the notion of  \emph{principal bundles}, whose fibers correspond to local Lorentz or spin frames, and we explain how to lift these structures to the spin group \(\mathrm{Spin}_{3,1}^+\). In particular, we clarify the transition between different global frames and demonstrate the equivalence between the corresponding spin bundles (cf.~Proposition~\ref{prop1-19-dec-2024}). We also present the key notion of a \emph{spin structure}, which is essential for formulating covariant derivatives of spinor fields in curved spacetimes. Two different global orthonormal frames generally yield different spin bundles with different spin connections. However, the canonical isomorphism described below ensures that these bundles are \emph{equivalent} in a precise sense, allowing us, in applications, to \emph{switch from one frame to another} without altering the underlying spin structure. These results will be used extensively to investigate the main analytic properties of spinor fields and their interactions with the underlying geometry. 


\paragraph{Lorentz Clifford algebra.}

We begin with a presentation of the fundamental Lorentz Clifford structure. Let $(\RR^{3,1}, \eta)$ be Minkowski spacetime equipped with its standard (flat) metric \(\eta\) of signature \((-,+,+,+)\).  Let \((e_0, e_1, e_2, e_3)\) be the canonical orthonormal basis, so that \(\eta(e_i,e_j) = \eta_{ij} \) is given by
\be
\eta_{00} = -1, 
\qquad 
\eta_{\ih\ih} = 1 \;\text{ for } \ih=1,2,3,
\qquad 
\eta_{\ih \jh} = 0 \;\text{ for } \ih \neq \jh.
\ee
When dealing with spinors, it is convenient\footnote{Since we will adopt a gauge-independent presentation this convention is of limited use.}
to adopt the convention that indices $i,j,k,l$ run over $\{0,1,2,3\}$, while indices $\ih,\jh,\kh,\lh$ run over $\{1,2,3\}$. Let us now introduce\footnote{Although these explicit matrices are not used directly, they help illustrate the abstract presentation adopted later in this section.}
the following \(4\times 4\) \emph{gamma matrices} (in the so-called chiral representation): 
\be
\gamma_0 = 
\begin{pmatrix}
0    & I_2 \\
I_2  & 0 
\end{pmatrix},
\qquad
\gamma_i
= 
\begin{pmatrix}
0       & - \sigma_i \\[6pt]
\sigma_i & 0        
\end{pmatrix},
\quad 
i=1,2,3,
\ee
where \(\sigma_1, \sigma_2, \sigma_3\) denote the standard \(2\times2\) Pauli matrices, defined as 
\be
\sigma_1 = 
\begin{pmatrix}
0 & 1 \\
1 & 0
\end{pmatrix},
\quad
\sigma_2 = 
\begin{pmatrix}
0 & - \,\mathrm{i} \\
\mathrm{i} &  0
\end{pmatrix},
\quad
\sigma_3 = 
\begin{pmatrix}
1 & 0 \\
0 & -1
\end{pmatrix}, 
\ee
where $\mathrm{i}^2 = -1$. The famous anti-commutativity property of these matrices is
\bel{equa-29D} 
\aligned
& \{\gamma_i, \gamma_j\} 
= \gamma_i\,\gamma_j + \gamma_j\,\gamma_i
= -2\,\eta_{ij} \,I_4  \qquad 
&&(i,j=0,1,2,3), 
\qquad
\\
& \gamma_0^{\dagger} = \gamma_0,
\qquad \qquad \qquad 
\gamma_{\ih}^{\dagger} = - \,\gamma_{\ih} \quad
&& (\ih=1,2,3),
\endaligned
\ee
where \(I_2\) and \(I_4\) denote the \(2\times 2\) and \(4\times 4\) identity matrices, respectively, and  the ${}^\dagger$ symbol denotes the
Hermitian transpose of a matrix.  The real algebra generated by these four matrices \(\{\gamma_i\} \) forms a sub-algebra of \(\mathrm{End}(\mathbb{C}^4)\), which is precisely the \emph{real Clifford algebra} \(\mathrm{Cl}(3,1)\) defined over $\RR$ and
represented over \(\mathbb{C}^4\). Moreover, the complexification of the real Clifford algebra is isomorphic to $M_4(\mathbb{C})$, the algebra of all $4\times 4$ complex matrices, which may be stated as 
$\mathrm{Cl}(3,1)\,\otimes_{\mathbb{R}} \,\mathbb{C}
\cong M_4(\mathbb{C})$;  see~\cite[Section 6.3.4]{Hamilton-2017}.

\bse
Furthermore, with this notation, we consider the \emph{Gamma linear map}
\begin{equation} \label{equa-map}
\gamma : \RR^{3,1} \to  \mathrm{Cl}(3,1),
\quad
v = v^i\,e_i \mapsto \gamma(v) = v^i\,\gamma_i, 
\end{equation}
which induces an action of Minkowski space \(\RR^{3,1} \) on \(\mathbb{C}^4\) (in this context referred to as the space of spinors), defined by
\be
v\cdot \psi = v^i\,\gamma_i\,\psi,
\qquad
\psi \in \mathbb{C}^4,
\ee
referred to as the \emph{Clifford multiplication.} Observe that $v\cdot \psi$ is a spinor in the same spinor representation space. 
\ese
%


\paragraph{Proper orthochronous Lorentz spin group.}

Let $\mathrm{SO}_{3,1}^+$ denote the group of matrices of $\text{GL}(\RR^{3,1})$ that preserves the metric $\eta$, as well as the time-orientation and the spatial orientation;  this subgroup is referred to as the \emph{proper orthochronous Lorentz group}. By definition, the \emph{spin group}, denoted by 
\(\mathrm{Spin}^+(3,1)\),  
is the double cover of \(\mathrm{SO}^+_{3,1} \); we now explain this construction. First of all, we introduce the hyperboloidal hypersurfaces
\be
\ourS_{3,1}^+ := \{v\in \RR^{3,1}| \eta(v,v) =1\}, \qquad \ourS_{3,1}^- := \{v\in \RR^{3,1}|\eta(v,v) = -1\},
\ee
which consist, respectively, of suitably normalized, spacelike and timelike vectors. Next, we introduce the \emph{orthochronous Lorentz spin group\footnote{The plus sign in $\mathrm{Spin}_{3,1}^+$ indicates the connected component containing the identity element in $\mathrm{SO}^+(3,1)$.}} over Minkowski spacetime $\RR^{3,1}$
\be
\mathrm{Spin}_{3,1}^+ = \Big\{ \gamma(v_1) \gamma(v_2) \ldots \gamma(v_{2p})  \gamma(w_1) \gamma(w_2)\ldots \gamma(w_{2q})
\, \Big / \, v_m\in \ourS_{3,1}^+, w_n\in \ourS_{3,1}^- \Big\}, 
\ee
which consists of the frames preserving both orientations, and forms a group under the multiplication inherited from $\text{Cl}(3,1)$. To check this group property, we observe that  for $v,w\in \RR^{3,1}$ and with the short-hand notation $v\cdot w = \gamma(v) \gamma(w)$ and applying the anti-commutation rules \eqref{equa-29D} yields $v\cdot w = -2\eta(v,w) \textrm{Id} - w\cdot v$. Furthermore, for any $v\in \ourS_{3,1}^+$ and $w\in \ourS_{3,1}^-$ we have 
$$
\aligned
&v\cdot v = - \eta(v,v)  \textrm{Id}= - \textrm{Id} \quad \Rightarrow \quad v^{-1} = -v,
\\
&w\cdot w = - \eta(w,w)  \textrm{Id} =  \textrm{Id} \quad \Rightarrow \quad w^{-1} = w.
\endaligned
$$
It is clear also that $\mathrm{Spin}_{3,1}^+$ is a Lie group —that is, a group endowed with a smooth manifold structure— which is connected and simply connected. Moreover, it double-covers $\mathrm{SO}_{3,1}^+$ via a two-fold covering map $\lambda: \mathrm{Spin}_{3,1}^+ \to \mathrm{SO}_{3,1}^+$; its kernel is given by $\big\{ \pm I_4 \big\}$, defined in \eqref{eq2-18-dec-2024}, below. Furthermore, the following property holds. 

\begin{claim}[cf. Hamilton {\cite[Proposition 6.5.22]{Hamilton-2017}}]
The Lie algebra of $\mathrm{Spin}_{3,1}^+$ (i.e., the tangent space at the identity)  is the subspace of the real Clifford algebra $\mathrm{Cl}(3,1)$ given by 
$$
\mathfrak{spin}_{3,1}^+ = \mathrm{\bf span} \big\{\gamma_i\gamma_j \, \big| \, 0\leq i \leq j \leq 3 \big\} \subset \mathrm{Cl}(3,1). 
$$
\end{claim} 


\paragraph{The Lie group homomorphism $\lambda$.}

The relation between $\mathrm{Spin}_{3,1}^+$ and $\mathrm{SO}_{3,1}^+$ can be described by the following action:
\begin{equation} \label{eq2-03-dec-2024}
R: \mathrm{Spin}_{3,1}^+\times \RR^{3,1} \to \RR^{3,1}, \qquad (u,x)\mapsto u\cdot x\cdot u^{-1}.
\end{equation}
Here the multiplication is understood as the multiplication in $\mathrm{Cl}(3,1)$. However, we need to show that $u\cdot x\cdot u^{-1}$ is still in $\RR^{3,1}$. This is checked by induction, as follows. For any vector $v\in \RR^{3,1}$ satisfying $\eta(v,v) = \pm 1$, we can write 
\be
- v\cdot x\cdot v^{-1} = (2\eta(v,x) + x\cdot v)\cdot v^{-1} =  -2\eta(v,x)\frac{v}{\eta(v,v)} + x\in\RR^{3,1}, 
\ee
which represents the reflection of $x$ with respect to the hyperplane with normal vector $v$. 
Define the map 
\be
\lambda: \mathrm{Spin}_{3,1}^+ \to \mathrm{SO}_{3,1}^+
\ee
by requiring that for the canonical basis \((e_0,e_1,e_2,e_3)\) of \(\RR^{3,1}\)
\begin{equation} \label{eq2-18-dec-2024-0}
\lambda(u)(e_i) = u\cdot e_i\cdot u^{-1}, \quad i=0,1,2,3
\end{equation} 
or, in matrix notation, 
\begin{equation} \label{eq2-18-dec-2024}
(e_0,e_1,e_2,e_3)\lambda(u) = u\cdot(e_0,e_1,e_2,e_3)\cdot u^{-1}
\end{equation}
We easily check that \(\lambda\) is a Lie group homomorphism. We emphasize that $\mathrm{Spin}_{3,1}^+$ is the universal double cover of $\mathrm{SO}_{3,1}^+$ with kernel $\{\pm I_4\}$, ensuring that $\lambda$ is surjective onto the proper orthochronous Lorentz group, as now stated.  
Indeed we can check the following property~\cite[Theorem~6.5.13]{Hamilton-2017}.

\begin{claim}
\label{thm1-31-dec-2024}
The map $\lambda$ is a surjective Lie group homomorphism whose kernel is $\{\pm I_4\}$. Furthermore, it is a 2-covering and the Lie algebra  of the Lie group  $\mathrm{Spin}_{3,1}^+$ and $\mathrm{SO}_{3,1}^+$, and satisfies 
\be
\lambda_*: \mathfrak{spin}_{3,1}^+\to \mathfrak{so}_{3,1}^+, \qquad z\mapsto \{\lambda_*(z):x\mapsto [z,x]\}. 
\ee 
\end{claim}

Observe in passing that the relation \eqref{eq2-18-dec-2024} 
can equivalently be written, in matrix notation, as 
\begin{equation} \label{eq3-18-dec-2024}
(\gamma_0,\gamma_1,\gamma_2,\gamma_3)\lambda(u) = u\cdot (\gamma_0,\gamma_1,\gamma_2,\gamma_3)\cdot u^{-1}.
\end{equation}
 

\paragraph{General formalism.}

For a general Lorentz space \((V,\eta)\) endowed with a quadratic form $\eta$, we define
\[
\ourS(V)^+ := \{ v \in V \mid \eta(v,v) = 1 \}, \qquad \ourS(V)^- := \{ w \in V \mid \eta(w,w) = -1 \}.
\]
The Clifford algebra on \(V\) is defined intrinsically as
\[
\mathrm{Cl}(V) := T(V) \Big/ \langle\, v\otimes v - \eta(v,v)\cdot 1 \,\rangle,
\]
where \(T(V)\) denotes the tensor algebra of \(V\) and $\la v\otimes v - \eta(v,v)\cdot 1\ra$ denotes the ideal generated by the elements $ v\otimes v - \eta(v,v)\cdot 1$. We then have $v\cdot w := [v\otimes w]$,  
and we emphasize that this algebra does not depend on the choice of base.

For a fixed orthonormal basis \(e = \{e_i\}\) of \(V\), the coordinate map
\be
\gamma_e: V \to \mathrm{Cl}(3,1), \quad v = v^i e_i \mapsto v^i\gamma_i,
\ee
induces an isomorphism of real algebras \(\mathrm{Cl}(V) \cong \mathrm{Cl}_{3,1}\). Then, we define
\be
\mathrm{Spin}(V)^+ := \Bigl\{ v_1\cdot v_2 \cdots v_{2p} \cdot w_1\cdot w_2 \cdots w_{2q} \,\Big|\, v_m\in \ourS(V)^+,\, w_n\in \ourS(V)^- \Bigr\}.
\ee
Given an orthonormal basis \((e_0,e_1,e_2,e_3)\) of \(V\), the coordinate map \(\gamma_e\) induces a Lie group isomorphism
\be
\gamma_e: \mathrm{Spin}(V)^+ \to \mathrm{Spin}_{3,1}^+.
\ee
Finally, if \(\mathrm{SO}(V)^+\) denotes the subgroup of \(\mathrm{GL}(V)\) preserving the metric, time-orientation, and spatial orientation, then the covering map is given by
\begin{equation} \label{eq2-15-dec-2024}
\lambda_V: \mathrm{Spin}(V)^+ \to \mathrm{SO}(V)^+, \quad \lambda_V(u): v \mapsto u\cdot v\cdot u^{-1}.
\end{equation}


\paragraph{Dirac sesquilinear forms.}
We next consider complex-valued sesquilinear forms $(\cdot,\cdot)$ defined on $\mathbb{C}^4$, which are assumed to be invariant under the action of $\mathrm{Spin}_{3,1}^+$, namely complex sesquilinear ($\RR$-conjugate-linear in the first, $\RR$-linear in the second slot) forms satisfying, by definition, 
\be
\aligned
&(u\psi,u\phi) = (\psi,\phi) \quad && \text {for all }  \phi,\psi\in \mathbb{C}^4 \text { and }  u\in\mathrm{Spin}_{3,1}^+\subset \text{End}(\mathbb{C}^4),
\\
&(\psi,c\phi) = c(\psi,\phi) = (c^*\psi,\phi)\quad && \text { for all }  c\in \mathbb{C} \text{ and }  \phi,\psi\in\mathbb{C}^4, 
\endaligned
\ee
where $c^*$ denotes the complex conjugate to $c$. An important class of examples is given as follows. 

\begin{claim}\label{lem1-15-dec-2024}
Let $\la \cdot,\cdot\ra$ be the standard Hermitian product on $\mathbb{C}^4$, i.e.
$\la\psi,\phi\ra := \psi^{\dag} \phi$ (the dag denoting the Hermitian transpose). 
Then the map 
\be
(\psi,\phi)\mapsto \la \gamma_0\cdot\psi,\phi\ra
\ee
is an invariant sesquilinear form. 
\end{claim}

The proof of this result is provided in Section~\ref{section=N19}, below.
 Furthermore, an invariant sesquilinear form $( \cdot,\cdot)$ is called a \emph{Dirac sesquilinear form} if, in addition, 
\begin{equation}
( \phi,X\cdot \psi) = (X\cdot \phi,\psi)
\qquad \text { for all }  X\in \RR^{3,1} \text{ and }  \phi,\psi\in\mathbb{C}^4
\end{equation}
(without complex conjugation). This condition guarantees that the form ``commutes'' with Clifford multiplication. In other words, if one multiplies a spinor by a vector  and then takes the Dirac form, the result is the same as if one had first taken the form and then inserted the vector on the other side. This compatibility is essential when defining the adjoint of the Dirac operator and when ensuring that conserved currents (derived from the Dirac equation) have the correct transformation properties. 


A typical example of a Dirac form (with a convenient sign convention) is
\begin{equation}\label{eq1-12-jan-2025}
(\phi,\psi) :=  \langle \gamma_0\phi,\psi\rangle = - \langle \gamma^0\phi,\psi\rangle,
\end{equation}
which we refer to as the \emph{standard Dirac form}. Here, \(\langle \cdot,\cdot \rangle\) is the standard Hermitian product on \(\mathbb{C}^4\), defined by \(\langle \psi,\phi\rangle = \psi^{\dagger}\phi\). In what follows, this form is sometimes denoted by
$$
\langle \phi,\psi \rangle_{\ourD} =  \langle \gamma_0\phi,\psi\rangle
\quad \text{for all } \phi,\psi\in \mathbb{C}^4.
$$
with the subscript \(\ourD\) indicating ``Dirac from''.
}


\subsection{ Spin structure over a topologically trivial spacetime} 
\label{subsec1-08-jan-2025}
{
\paragraph{\(\mathrm{SO}^+_{3,1} \)-bundle.}

We use the language of principal bundles and refer the reader to Section~\ref{annex-principal-bundle} for basic terminology. We are interested in the description of the space of spinors in a curved spacetime, and we provide the geometric context in which the results of the present Monograph will be formulated in a gauge-invariant manner. 

We are now in a position to apply the previous formalism and, for simplicity, from now on we assume that the base spacetime \(\mathcal{M} \) is diffeomorphic to \(\mathbb{R}^4\). Hence there exists a global orthonormal frame \(\{e_0,e_1,e_2,e_3\} \) satisfying, by definition, $g(e_i,e_j) = \eta_{ij}$. Obviously, one procedure to construct such a frame is by starting with any globally defined frame and then applying the Schmidt orthogonalization procedure. On the other hand for spacetimes with non-trivial topology, such a global frame need not exist, which leads us to topological obstructions to the existence of spin structures. 

It is convenient to use the notation $e : x \,\mapsto\, (e_0,e_1,e_2,e_3)$ for the collection of \emph{frames} at each point. We consider the principal bundle
\[
\pi_{\mathrm{SO}} : \mathrm{SO}^+(\mathcal{M})  \to \mathcal{M}
\]
with structure group \(\mathrm{SO}^+_{3,1} \), consisting of all positively oriented and time-oriented orthonormal frames on the spacetime \(\mathcal{M} \). Since \(\mathcal{M} \cong \mathbb{R}^4\), by recalling the  \emph{given} frame there is an obvious global trivialization
\be
\aligned
\mathcal{T}_\mathrm{SO}^e:\quad 
& \mathrm{SO}^+(\mathcal{M}) \quad\;\;\;\longmapsto\quad \mathcal{M} \times \mathrm{SO}^+_{3,1},
\\
&(f_0,f_1,f_2,f_3)\;\longmapsto\; (x,A) \quad \text{such that} \quad 
(f_0,f_1,f_2,f_3) = (e_0,e_1,e_2,e_3)\,A.
\endaligned
\ee
It is clear that \(\mathrm{SO}^+_{3,1} \) acts on \(\mathrm{SO}^+(\mathcal{M})\) from the right, that is, 
\be
\aligned
\mathrm{SO}^+(\mathcal{M}) \times \mathrm{SO}^+_{3,1}
& \;\longrightarrow\; \mathrm{SO}^+(\mathcal{M}),
\qquad \quad 
\bigl((f_0,f_1,f_2,f_3), A\bigr)
& \;\longmapsto\; (f_0,f_1,f_2,f_3)\,A,
\endaligned
\ee
preserving each fiber and acting freely on it.  These properties indeed allow us to deal with \(\mathrm{SO}^+(\mathcal{M})\) as a principal bundle over \(\mathcal{M} \) with structure group \(\mathrm{SO}^+_{3,1} \). In short (with the orientation being tacitly assumed), we may refer to it as the \emph{principal bundle of orthonormal frames.} 


\paragraph{\(\mathrm{Spin}^+_{3,1} \)-bundle.}

Next, let us consider 
\be
\mathrm{Spin}^+(\mathcal{M}) 
= \bigsqcup_{x\in \mathcal{M}} \mathrm{Spin}(T_x\mathcal{M})^+,
\quad
\pi_{\mathrm{Spin}}: \mathrm{Spin}^+(\mathcal{M})  \to \mathcal{M}.
\ee
Again, using the \emph{given} global orthonormal frame \(\{e_0,e_1,e_2,e_3\} \), we obtain a trivialization
\be
\mathcal{T}_\mathrm{Spin}^e:\ 
\mathrm{Spin}^+(\mathcal{M}) \;\longrightarrow\; \mathcal{M} \times \mathrm{Spin}^+_{3,1},
\quad 
p \;\longmapsto\; \bigl(x,\gamma_e(p)\bigr),
\ee
where, for all \(v\in T_x\mathcal{M} \),
\be
\gamma_e(v) = v^i\,\gamma_i,
\quad \text{with} \quad
v = v^i\,e_i.
\ee
This induces a Lie group isomorphism \(\gamma_e: \mathrm{Spin}(T_x\mathcal{M})^+  \to \mathrm{Spin}_{3,1}^+\).  The group \(\mathrm{Spin}_{3,1}^+\) acts from the right on \(\mathrm{Spin}^+(\mathcal{M})\) by
\begin{equation} \label{eq1-16-dec-2024}
\mathrm{Spin}^+(\mathcal{M})\times \mathrm{Spin}_{3,1}^+
\to \mathrm{Spin}^+(\mathcal{M}),
\quad \qquad 
(p,u) \mapsto p\star u = p\cdot \gamma_e^{-1}(u).
\end{equation}
Thus, \(\mathrm{Spin}^+(\mathcal{M})\) becomes a principal bundle over \(\mathcal{M} \) with structure group \(\mathrm{Spin}_{3,1}^+\).  This is called the  \emph{spin bundle} with respect to \(\{e_0,e_1,e_2,e_3\} \) and is denoted by \(\mathrm{Spin}_e^+(\mathcal{M})\) to indicate its dependence on the chosen frame.

We emphasize that the algebraic structure of $\mathrm{Spin}^+_e(\mathcal{M})$ does not depend on the collection of frames $e$, however, the right action of the structure group $\mathrm{Spin}_{3,1}^+$ is determined by the frames $e$ via the coordinate map $\gamma_e$. We now define the map
\be
\Lambda_e : \mathrm{Spin}_e^+(\mathcal{M})
\;\longrightarrow\; \mathrm{SO}^+(\mathcal{M}),
\quad 
p \;\longmapsto\; 
p\cdot (e_0,e_1,e_2,e_3)\cdot p^{-1},
\ee
where the dot denotes the Clifford multiplication.

\begin{lemma}
\label{lem1-16-dec-2024}
With the above notation, one has 
\[
\Lambda_e\bigl(p\star u\bigr) = \Lambda_e(p)\,\lambda(u),
\]
so that the following diagram commutes:
\be
\begin{tikzcd}
{\mathrm{Spin}_e^+(\mathcal{M})\times \mathrm{Spin}_{3,1}^+} 
\arrow[rr] 
\arrow[dd, "\Lambda_e\times \lambda"] 
&  & \mathrm{Spin}_e^+(\mathcal{M}) 
\arrow[dd, "\Lambda_e"] 
\arrow[rd, "\pi_{\mathrm{Spin}}"] 
&   \\
&  &  & \mathcal{M} \\
{\mathrm{SO}^+(\mathcal{M})\times \mathrm{SO}_{3,1}^+} 
\arrow[rr]                                    
&  & \mathrm{SO}^+(\mathcal{M}) 
\arrow[ru, "\pi_{\mathrm{SO}}"]                       
&  
\end{tikzcd}
\ee
\end{lemma}

The proof is postponed to Section~\ref{section=N19}. In the trivializations \(\mathcal{T}_{\mathrm{Spin}}^e\) and \(\mathcal{T}_{\mathrm{SO}}^e\), the map \(\Lambda_e\) is exactly given by \(\lambda\), as now stated.

\begin{lemma}
With the notation above one has 
\[
\lambda\circ \mathcal{T}_{\mathrm{Spin}}^e 
 =  
\mathcal{T}_{\mathrm{SO}}^e \circ \Lambda_e.
\]
\end{lemma}

\begin{proof}
Let \(p\in\mathrm{Spin}^+(\mathcal{M})\) and observe that
\[
(e_0,e_1,e_2,e_3)\,\mathcal{T}_{\mathrm{SO}}^e\circ \Lambda_e(p) 
 =  
p\cdot (e_0,e_1,e_2,e_3)\cdot p^{-1}.
\]
Applying \(\gamma_e\) to both sides, we get
\[
(\gamma_0,\gamma_1,\gamma_2,\gamma_3)\,\mathcal{T}_{\mathrm{SO}}^e\circ \Lambda_e(p)
 = 
\gamma_e(p)\,\bigl(\gamma_0,\gamma_1,\gamma_2,\gamma_3\bigr)\,\gamma_e(p)^{-1}
 = 
(\gamma_0,\gamma_1,\gamma_2,\gamma_3)\,\lambda\circ\gamma_e(p),
\]
which implies $
\mathcal{T}_{\mathrm{SO}}^e\circ \Lambda_e(p) = \lambda\circ \mathcal{T}_{\mathrm{Spin}}^e(p)$.
\end{proof}

Based on the above result and the property of the map $\lambda$ in Claim~\ref{thm1-31-dec-2024}, 
we have also the following result.

\begin{proposition}
The map $\Lambda_e: \mathrm{Spin}_e^+(\mathcal{M})\mapsto \mathrm{SO}^+(\mathcal{M})$ is a 2-covering and especially, is a local diffeomorphism. 
\end{proposition}


\paragraph{The spin structure.}

\begin{definition}[cf.\ {Hamilton~\cite[Definition 5.3.1]{Hamilton-2017}}]
\label{def1-17-dec-2024}
Let \(\pi: \mathcal{P} \to \mathcal{M} \) and \(\pi': \mathcal{P}'\to \mathcal{M} \) be two principal bundles over a smooth manifold \(\mathcal{M} \) with the same structure group \(G\).  A diffeomorphism \(F: \mathcal{P} \to \mathcal{P}'\) is called a  \emph{bundle isomorphism} if it is fiber-preserving and \(G\)-equivariant, i.e.,
\be
\aligned
\pi\bigl(F(p)\bigr) & = \pi(p)\quad && \text{ for all } \,p\in \mathcal{P},
\\
F\bigl(p\star g\bigr) & = F(p)\,*\,g
\quad && \text { for all } \,p\in \mathcal{P}, \,  g\in G,
\endaligned
\ee
where \(\star\) and \(*\) are the right-actions of \(G\) on \(\mathcal{P} \) and \(\mathcal{P}'\), respectively.
\end{definition}

Returning to spin bundles, suppose
\((f_0,f_1,f_2,f_3)\) is another positively oriented and time-oriented orthonormal frame on \(\mathcal{M} \). Our goal is to show that \(\mathrm{Spin}_f^+(\mathcal{M})\) and \(\mathrm{Spin}_e^+(\mathcal{M})\) are isomorphic. Let
\be
T_{e \to f} : \mathcal{M}  \to \mathrm{SO}^+_{3,1}
\quad\text{such that} \quad
(f_0,f_1,f_2,f_3) = (e_0,e_1,e_2,e_3)\,T_{e \to f}.
\ee
Denote by \(\Lambda_e\) and \(\Lambda_f\) the respective maps generated by these two frames.  Then, observe that the action of $p$ is real-linear: 
\be
\Lambda_f(p) 
 =  
p\cdot(f_0,f_1,f_2,f_3)\cdot p^{-1}
 = 
p\cdot (e_0,e_1,e_2,e_3)\cdot p^{-1} \,T_{e \to f}
 = 
\Lambda_e(p)\,T_{e \to f}.
\ee
Recall the universal covering map \(\lambda: \mathrm{Spin}_{3,1}^+ \to \mathrm{SO}^+_{3,1} \).  We choose a lift
$
\tau_{e \to f}: \mathcal{M}  \to \mathrm{Spin}_{3,1}^+
$
defined by the key identity 
\begin{equation} \label{eq1-18-dec-2024}
\tau_{e \to f} \,\cdot(\gamma_0,\gamma_1,\gamma_2,\gamma_3)\,\cdot\tau_{e \to f}^{-1}
 =  (\gamma_0,\gamma_1,\gamma_2,\gamma_3)\,T_{e \to f}.
\end{equation}
Then, by Lemma~\ref{lem1-16-dec-2024} we have 
\be
\Lambda_f(p) 
 =  \Lambda_e(p)\,\lambda(\tau_{e \to f}) 
 =  
\Lambda_e\!\bigl(p \star \tau_{e \to f} \bigr)
 = 
\Lambda_e\!\bigl(p\cdot \gamma_e^{-1}(\tau_{e \to f})\bigr)
\ee
and we introduce the mapping 
\begin{equation} \label{eq1-17-dec-2024}
F_{f \to e}: \mathrm{Spin}_f^+(\mathcal{M}) \to \mathrm{Spin}_e^+(\mathcal{M}),
\quad \qquad 
p \mapsto p \star \tau_{e \to f} = p \,\cdot\, \gamma_e^{-1} \!\bigl(\tau_{e \to f} \bigr).
\end{equation}
We arrive at the following important observation.

\begin{proposition}
\label{prop1-19-dec-2024}
Let \((e_0,e_1,e_2,e_3)\) and \((f_0,f_1,f_2,f_3)\) be two orthonormal frames on \(\mathcal{M} \) with positive orientation and time-orientation.  Define \(F_{f \to e} \) by \eqref{eq1-17-dec-2024}.  Then \(F_{f \to e} \) is a bundle isomorphism and, moreover,
\[
\Lambda_f = \Lambda_e \circ F_{f \to e}.
\]
In other words, the following diagram commutes:
\[
\begin{tikzcd}
\mathrm{Spin}_f^+(\mathcal{M}) 
\arrow[rd, "\Lambda_f"] 
\arrow[rr, "F_{f \to e}"]
&  & \mathrm{Spin}_e^+(\mathcal{M}) 
\arrow[ld, "\Lambda_e"'] 
\\
& \mathrm{SO}^+(\mathcal{M}) &                                                    
\end{tikzcd}
\]
\end{proposition}

By definition, an \emph{equivalence class} of such spin bundles over \(\mathcal{M} \) is called a  \emph{spin structure}. In other words, any two spin bundles arising from different global frames are isomorphic if they induce the same spin structure. This distinction between a spin bundle (which depends on a global choice of frame) and the spin structure (which is an equivalence class) is crucial for gauge invariance, and later on developing a gauge-invariant theory of existence for the Dirac operator.

}


{

\subsection{ Spinor bundle over a topologically trivial spacetime} 
\label{section===23}

\paragraph{Associated vector bundle.}

Let \(\pi : \mathcal{P} \to \mathcal{M} \) be a principal bundle with structure group \(G\).  Suppose \(V\) is a real vector space and \(\rho : G \to \mathrm{End}(V)\) is a linear representation.
The  \emph{associated vector bundle} \(\mathcal{P} \times_{\rho}V\mapsto \mathcal{M} \) is defined by
$
\mathcal{E} = \bigl(\mathcal{P} \times V\bigr)\,\bigl/\, G,
$
where \(G\) acts on \(\mathcal{P} \times V\) via
\be
(p,v)\,\cdot\, g 
= 
\bigl(p\cdot g,\; \rho(g)^{-1} \,v\bigr).
\ee
It can be shown that the above action is a {\sl principal action} (cf.~\cite[Lemma~4.7.1]{Hamilton-2017}).  Consequently, there is a unique smooth structure on \(\mathcal{E} \) making
\be
\Phi \;:\; \mathcal{P} \times V \;\longrightarrow\; \bigl(\mathcal{P} \times V\bigr)\,\bigl/\,G = \mathcal{E}
\ee
a surjective submersion.  By a standard property of surjective submersions (cf.\ \cite[Lemma~3.7.5]{Hamilton-2017}), there is a unique surjective submersion \(\pi_{\mathcal{E}} : \mathcal{E} \to \mathcal{M} \) such that
\[
\begin{tikzcd}[column sep=3em, row sep=3em]
\mathcal{P} \,\times\,V 
\arrow[rrd, "\;\pi \,\circ\,\mathrm{pr}_1"'] 
\arrow[d, "\Phi"'] 
&  &             \\
\mathcal{E} 
\arrow[rr, "\pi_{\mathcal{E}}"']                       
&  & \mathcal{M}
\end{tikzcd}
\]
Given \([p,v]\in \mathcal{E} \), we define a linear structure on each fiber \(\mathcal{E}_x := \pi_{\mathcal{E}}^{-1}(\{x\})\) by
\[
\alpha\,[p,v] \;+\; \beta\,[p,w] 
= [\,p,\;\alpha\,v + \beta\,w\,].
\]
One checks that \(\mathcal{E}_x\) is linearly isomorphic to \(V\). Hence \(\pi_{\mathcal{E}}:\mathcal{E} \to \mathcal{M} \) defines a vector bundle over \(\mathcal{M} \).


\paragraph{\(\mathrm{SO}^+_{3,1} \)-bundle and the tangent bundle.}

Returning to our topologically trivial spacetime, we note that for \(V = \mathbb{R}^{3,1} \), there is an associated vector bundle of \(\mathrm{SO}^+(\mathcal{M})\) is isomorphic to the tangent bundle \(T\mathcal{M} \).  In fact, one simply observes that (with \(\rho\) the canonical action of \(\mathrm{SO}^+_{3,1} \) on \(\mathbb{R}^{3,1} \))
\be
\mathrm{SO}^+(\mathcal{M})\times_{\rho} \mathbb{R}^{3,1} \;\longrightarrow\; T\mathcal{M},
\quad
\bigl[(f_0,f_1,f_2,f_3),v\bigr]\;\longmapsto\; v^i\,f_i
\ee
is well-defined and provides an isomorphism of vector bundles (i.e.\ a diffeomorphism that is fiber-preserving and linear on each fiber).


\paragraph{\(\mathrm{Spin}^+_{3,1} \)-bundle and Dirac spinor bundle.}

In a similar manner, consider the associated bundle
\[
\ourS_e(\mathcal{M}) 
= \mathrm{Spin}_e^+(\mathcal{M})\times_{\kappa} \mathbb{C}^4,
\]
where \(\kappa\) is the canonical representation of \(\mathrm{Spin}^+_{3,1} \) on \(\mathbb{C}^4\) induced by Clifford multiplication:
\be
v \,\cdot\, \psi = \gamma_e(v)\psi 
= v^i\,\gamma_i\,\psi,\quad v^ie_i = v.
\ee
If \((\,f_0,f_1,f_2,f_3)\) denotes another global orthonormal basis, then the induced isomorphism reads 
\[
F_{f \to e} : \ourS_f(\mathcal{M}) \;\longrightarrow\; \ourS_e(\mathcal{M}),
\quad
[p,\psi] \;\longmapsto\; \bigl[F_{f \to e}(p),\,\psi\bigr].
\]

\begin{lemma}
\label{lem1-05-jan-2025}
The map \(F_{f \to e} \) is well-defined on \(\ourS_f(\mathcal{M})\) and is an isomorphism of real vector bundles.
\end{lemma}

\begin{proof}
Let \(\bigl[q,\phi\bigr] = \bigl[p,\psi\bigr]\in \ourS_f(\mathcal{M})\).  Then \((q,\phi)\) and \((p,\psi)\) are related by some \(u\in\mathrm{Spin}_{3,1}^+\), so
\[
q = p\,\star\,u = p \,\cdot\, \gamma_f^{-1}(u),
\quad
\phi = u^{-1} \,\psi.
\]
By Proposition~\ref{prop1-19-dec-2024}, \(F_{f \to e} \) is a bundle isomorphism and equivariant under \(\mathrm{Spin}_{3,1}^+\), hence 
\[
F_{f \to e}(q) 
 =  F_{f \to e} \bigl(p\,\star\,u\bigr) 
 =  F_{f \to e}(p)\,\star\,u.
\]
Therefore, we have 
\[
\bigl[F_{f \to e}(q),\,\phi\bigr]
 = 
\bigl[F_{f \to e}(p)\,\star\,u,\;\phi\bigr]
 = 
\bigl[F_{f \to e}(p),\;u^{-1} \phi\bigr]
 = 
\bigl[F_{f \to e}(p),\,\psi\bigr],
\]
showing \(F_{f \to e} \) is well-defined. It is also direct that $F_{f \to e}$ preserves the fiber and on each fiber $F_{f \to e}$ is linear. 
\end{proof}
\medskip

Moreover, we particularly note that
\begin{equation} \label{eq5-01-jan-2025}
F_{f \to e} \bigl([p,\psi]\bigr) 
 =  
\bigl[p,\;\tau_{e \to f}^{-1} \,\psi\bigr].
\end{equation}
For $v\in T_x\mathcal{M}$, we also define the  \emph{Clifford multiplication} as follows (cf. Hamilton~\cite[Proposition 6.9.13]{Hamilton-2017}).} Let $s\in \mathrm{Spin}^+(T_x\mathcal{M})$. Then
\begin{equation}
v\cdot [s,\psi] := [s,\gamma_e(s^{-1}\cdot v\cdot s) \psi], 
\end{equation}
where $s^{-1}\cdot v\cdot s$ is a product in $\mathrm{Spin}^+(T_x\Mcal)$. It is easy to check that this definition does not depend on the chose of representative, and it can be verified that
\begin{equation}
v\cdot[1,\psi] = [1,\gamma_e(v)\psi],
\end{equation}
where $1$ is the unit element in $\mathrm{Spin}^+(T_x\Mcal)$. That is, we are acutely defining the Clifford multiplication via the coordinates (see the notion of {\sl coordinate} below \eqref{eq2-30-aout-2025}). Furthermore,
\begin{equation} \label{eq5-08-jan-2025}
F_{f \to e}(v\cdot[s,\psi]) = v\cdot F_{f \to e}([s,\psi]).
\end{equation}
To see this, we perform the following calculation:
\be
\aligned
v\cdot F_{f \to e}([s,\psi]) & =  v\cdot [F_{f \to e}(s),\psi] \stackrel{\eqref{eq5-01-jan-2025}} = v\cdot[s,\tau_{e \to f}^{-1}\psi] = [s,\gamma_e(s^{-1}\cdot v\cdot s)\cdot\tau_{e \to f}^{-1}\psi],
\\
F_{f \to e}(v\cdot[s,\psi]) & =  F_{f \to e}([s,\gamma_f(s^{-1}\cdot v\cdot s)\psi]) = [F_{f \to e}(s),\gamma_f(s^{-1}\cdot v\cdot s)\psi]
\\ 
& \stackrel{\eqref{eq5-01-jan-2025}} = [s,\tau_{e \to f}^{-1}\cdot\gamma_f(s^{-1}\cdot v\cdot s)\psi].
\endaligned
\ee
Then we apply the following lemma.

\begin{lemma}\label{lem1-30-aout-2025}
For any $u\in\mathrm{Spin}^+(T_x\Mcal)$, one has 
\begin{equation}\label{eq1-30-aout-2025}
\gamma_f(u) = \tau_{e \to f}\cdot\gamma_e(u)\cdot\tau_{e \to f}^{-1}.
\end{equation}
\end{lemma}

\begin{proof}
\bse
When $u$ is a vector, viz. $u\in T_x\Mcal$, let
\be
\gamma_e(u) = 
(\gamma_0,\gamma_1,\gamma_2,\gamma_3)
\left(\begin{array}{c}
v^0
\\
v^1
\\
v^2
\\
v^3
\end{array}
\right),
\quad
\gamma_f(u) = 
(\gamma_0,\gamma_1,\gamma_2,\gamma_3)
\left(\begin{array}{c}
w^0
\\
w^1
\\
w^2
\\
w^3
\end{array}
\right).
\ee
Then, denote by $v = (v^0,v^1,v^2,v^3)^{\mathrm{T}}$, $w = (w^0,w^1,w^2,w^3)^{\mathrm{T}}$,
\be
T_{e \to f}w = v 
\ee
which leas to, thanks to \eqref{eq1-18-dec-2024},
\be
\aligned
\gamma_f(u) & =(\gamma_0,\gamma_1,\gamma_2,\gamma_3)w 
= (\gamma_0,\gamma_1,\gamma_2,\gamma_3)T_{e \to f}v 
= \tau_{e \to f}\cdot(\gamma_0,\gamma_1,\gamma_2,\gamma_3)\cdot\tau_{e \to f }^{-1}v
\\
\stackrel{*}{} & = \tau_{e \to f}\cdot(\gamma_0,\gamma_1,\gamma_2,\gamma_3)v\cdot\tau_{e \to f }^{-1}
=\tau_{e \to f}\cdot\gamma_e(u)\cdot\tau_{e \to f }^{-1},
\endaligned
\ee
where the equality $*$ is due to the fact that $\tau_{e \to f}^{-1}\cdot$ is real-linear.

When $u$ is a Clifford product of finite many vectors, it is obvious that \eqref{eq1-30-aout-2025} also holds.

\ese
\end{proof}

{
Finally, we close this presentation by emphasizing that a smooth section of $\ourS_e(\mathcal{M})$ is called a {\bf spinor field} and, more precisely,
\begin{equation}\label{eq2-30-aout-2025}
\Psi: \mathcal{M}\mapsto \ourS_e(\mathcal{M}), \quad x\mapsto [p(x),\psi(x)].
\end{equation}
We can fix $p(x) = 1\in\mathrm{Spin}_e^+(\mathcal{M})$ and, consequently, $\Psi = [1,\psi]$. Therefore, the section $\Psi$ is associated with a smooth $\mathbb{C}^4$--valued function, which is called the \emph{coordinates of the spinor field} $\Psi$. It is clear that different choices of $p$ lead to different $\psi$. Clearly, when $q = p\cdot u$, we find 
\be
[p,\psi] = [q,\phi] \quad\longrightarrow \quad \phi = u\psi.
\ee

}


\section{Gauge-invariant differentiation of spinor fields}
\label{section=N18}

\subsection{ Spin connection}
\label{subsec1-09-jan-2025}

{

\paragraph{Connections on a principal bundle and the associated vector bundles.}

We continue our description in Section~\ref{section=N17} of a framework to work with spinor fields in a curved spacetime. Let $\pi:\mathcal{P} \mapsto \mathcal{M}$ be a principal bundle with structure group $G$. We recall the notion of vertical distribution of $\mathcal{P}$, that is, 
\be
\mathscr{V}_p = \ker(\diff_p\pi) = T_{p} \mathcal{P}_x, \qquad x = \pi(p), 
\ee
which is the tangent space to the fiber $\mathcal{P}_x$ (an embedded submanifold of $\mathcal{P})$.
An  \emph{Ehresmann connection} (also called a principal connection) defined on $\mathcal{P}$ is a right-invariant horizontal distribution, that is, a distribution $\mathscr{H}$ of $\mathcal{P}$ such that 
\begin{equation}
\aligned
&T_p\mathcal{P} = \mathscr{V}_p\oplus \mathscr{H}_p,
\qquad\quad
&(r_g)_*\mathscr{H}_p = \mathscr{H}_{p\cdot g}\quad \text { for all }  g\in G. 
\endaligned
\end{equation}
Let $\mathcal{E} = \mathcal{P} \times_{\rho}V$ be an associated vector bundle of structure group $G$. We also recall the vertical distribution
\be
\mathscr{V}^{\mathcal{E}}_{[p,v]} := \ker(\diff_{[p,v]} \pi_{\mathcal{E}}) = T_{[p,v]} \mathcal{E}_x, \qquad x=\pi_{\mathcal{E}}([p,v]) = \pi(p).
\ee
A  \emph{connection} is a distribution $\mathscr{H}^\mathcal{E}$ defined on $\mathcal{E}$ such that
\be
\mathscr{H}^{\mathcal{E}}_{[p,v]} \oplus \mathscr{V}^{\mathcal{E}}_{[p,v]} = T\mathcal{E}_{[p,v]}.
\ee
We establish now that a connection defined on a principal bundle induces a connection defined on the associated vector bundle.

\begin{proposition} \label{prop1-01-jan-2025}
Let $\mathscr{H}$ be an Ehresmann connection defined on a principal bundle $\mathcal{P}$. Let $\Phi: \mathcal{P} \times V\mapsto \mathcal{E} = \mathcal{P} \times_{\rho}V$ be the canonical projection. Then 
\be
\mathscr{H}^{\mathcal{E}}_{[p,v]} := \{\diff_{(p,v)} \Phi(Y,0)|Y\in \mathscr{H}_p\} 
\ee
is a connection defined on $\mathcal{E}$, that is,
\begin{equation} \label{eq2-22-dec-2024}
\mathscr{H}^{\mathcal{E}} \oplus \mathscr{V}^{\mathcal{E}} = T\mathcal{E},
\end{equation}
which is referred to as the \emph{associated connection} on the vector bundle. 
\end{proposition}

\begin{proof}
We first show that $\mathscr{H}^{\mathcal{E}} \cap \mathscr{V}^{\mathcal{E}} = \{0\}$. Recall that $\pi_{\Ecal}\circ\Phi = \pi$. Thus at $(p,v)\in P\times V$, we have 
$$
\diff_{[p,v]} \pi_{\mathcal{E}} \circ \diff_{(p,v)} \Phi(Y,0) = \diff_p\pi(Y).
$$
Let \(Y\in \mathscr{H}_p\) and set \(X := \diff_{(p,v)} \Phi(Y,0)\in \mathscr{H}^{\mathcal{E}}{[p,v]} \). If X also belongs to the vertical subspace \(\mathscr{V}^{\mathcal{E}}_{[p,v]} = \ker(\diff\pi_{\mathcal{E}})\), then
$$
\diff_{[p,v]} \pi_{\mathcal{E}} \bigl(\diff_{(p,v)} \Phi(Y,0)\bigr) = \diff_p\pi(Y) = 0.
$$
Since $\diff_p\pi$ is injective on the horizontal space \(\mathscr{H}_p\) (since \(T_p\mathcal{P} = \mathscr{H}_p \oplus \mathscr{V}_p\)), it follows that $Y=0$ and hence $X=0$. Then the desired relation $\mathscr{H}^{\mathcal{E}} \cap \mathscr{V}^{\mathcal{E}} = \{0\}$ follows.

On the other hand, since the map 
$
\diff_p\pi|_{\mathscr{H}_p}
$
is linear and surjective, the map 
$$
\diff_{[p,v]} \pi_{\mathcal{E}} \circ \diff_{(p,v)} \Phi|_{\mathscr{H} \times\{0\}}  = \diff_p\pi|_{\mathscr{H}_p}: \mathscr{H} \times\{0\} \mapsto T_{\pi(p)} \mathcal{M}
$$
is linear and surjective. We deduce that 
$$\dim(\mathcal{M}) = \dim(\mathscr{H}_p)\geq \dim(\diff_{(p,v)}\Phi|_{\mathscr{H} \times\{0\}}) \geq \dim(\mathcal{M}),
$$ which leads us to \(\mathscr{H}^{\mathcal{E}} \oplus \mathscr{V}^{\mathcal{E}} = T\mathcal{E} \). In other words, every tangent vector in $T_{[p,v]} \mathcal{E}$ can be uniquely decomposed into a horizontal component (lifted from $\mathscr{H}_p$ via $\Phi$) and a vertical component (tangent to the fiber), which is precisely the defining property of a connection on $\mathcal{E}$. This leads us to \eqref{eq2-22-dec-2024}.
\end{proof}


\paragraph{Spin connection associated with a choice of frame.}

Let us state the following classical result whose proof is given in~Section~\ref{Appendix--B5}.  

\begin{proposition}[Existence of the Levi-Civita connection on $\mathrm{SO}^+$] 
\label{thm2-08-jan-2025}
Let $(\mathcal{M},g)$ be a oriented and time-oriented spacetime with trivial topology, and let $\mathrm{SO}^+(\mathcal{M})$ denote the bundle of orthochronous frames on $\mathcal{M}$. Then there exists a unique Ehresmann connection $\mathscr{H}^{\mathrm{SO}}$ that \emph{induces the Levi-Civita connection} on $T\mathcal{M}$. 
\end{proposition}

We recall that by Proposition~\ref{thm1-31-dec-2024}, the map $\Lambda_e$ is a local diffeomorphism. Thus we define, for $p\in \mathrm{Spin}^+(\mathcal{M})$
$$
\mathscr{H}^{\mathrm{Spin}}_{e,p}:= \big\{Y\in T_p\mathrm{Spin}^+(\mathcal{M})| \diff_p\Lambda_e(Y)\in\mathscr{H}^{\mathrm{SO}}_{\Lambda_e(p)} \big\}.
$$
Then we establish the following property. At this stage, we are still relying on a frame $e$ (which will be eventually ``suppressed''). 

\begin{proposition}
The subspace $\mathscr{H}^{\mathrm{Spin}}_{e,p}$ forms a smooth distribution and defines a connection on $\mathrm{Spin}^+_e(\mathcal{M})$.
\end{proposition}

\begin{proof}
Recall that $\Lambda_e$ is a local diffeomorphism. Thus $\mathscr{H}^{\mathrm{Spin}}_{e}$ is a smooth distribution since, locally, a smooth basis of $\mathrm{SO}^+(\mathcal{M})$ can be pulled back on $\mathscr{H}^{\mathrm{Spin}}_e$ and still be smooth. On the other hand, we observe that
$
\pi_{\mathrm{SO}} \circ \Lambda_e = \pi_{\mathrm{Spin}}
$
and this leads us to
$
\diff\pi_{\mathrm{SO}} \circ \diff\Lambda_e = \diff\pi_{\mathrm{Spin}}.
$
Recall that $\diff_p\Lambda_e$ is a linear isomorphism. Thus we find 
$$
\ker(\diff_p\pi_\mathrm{Spin}) = \Lambda^*(\ker(\diff_{\Lambda_e(p)} \pi_{\mathrm{SO}})), 
$$
which yields 
$
\mathscr{V}^{\mathrm{Spin}}_p = \Lambda^*\big(\mathscr{V}^{\mathrm{SO}}_{\Lambda_e(p)} \big).
$
Thus we obtain
\begin{equation}
\mathscr{H}^{\mathrm{Spin}}_{e,p} \oplus \mathscr{V}^{\mathrm{Spin}}_{e,p} = T_p\mathrm{Spin}^+(\mathcal{M}).
\end{equation}

We then check that $\mathscr{H}^{\mathrm{Spin}}_e$ is invariant under the right action of $\mathrm{Spin}_{3,1}^+$. In fact, we only need to show
\begin{equation} \label{eq1-01-jan-2025}
(r_g)_*(\mathscr{H}^{\mathrm{Spin}}_{e,p}) \subset \mathscr{H}^{\mathrm{Spin}}_{e,p\cdot g}
\end{equation}
due to the symmetry. To see this, let $Y\in \mathscr{H}^{\mathrm{Spin}}_{e,p}$. We need to show that $\diff_pr_g(Y)\in  \mathscr{H}^{\mathrm{Spin}}_{e,p\cdot g}$. We observe that Lemma~\ref{lem1-16-dec-2024} gives
$$
\Lambda_e\circ r_g = r_{\lambda(g)}\circ\Lambda_e.
$$
Then
$$
\diff_{p\cdot g} \Lambda_e\circ \diff_pr_g = \diff_{\Lambda_e(p)}r_{\lambda(g)} \circ \diff_p\Lambda_e
$$
by taking the differentials at $p$. 
Thus, remark that $\diff_p\Lambda_e(Y) \in\mathscr{H}^{\mathrm{SO}}_{\Lambda_e(p)}$, we obtain
$$
\diff_{p\cdot g} \Lambda_e\circ\diff_pr_g(Y) = \diff_{\Lambda_e(p)}r_{\lambda(g)} \circ \diff_p\Lambda_e(Y)\in (r_{\lambda(g)})_*\big(\mathscr{H}^{\mathrm{SO}}_{\Lambda_e(p)} \big) 
= \mathscr{H}^{\mathrm{SO}}_{\Lambda_e(p)\cdot\lambda(g)} = \mathscr{H}^{\mathrm{SO}}_{\Lambda_e(p\cdot g)}, 
$$
where the last two identity is due to the fact that $\mathscr{H}^{\mathrm{SO}}$ is right-invariant and Lemma~\ref{lem1-16-dec-2024}. Thus we find
$
\diff_pr_g(Y)\in \mathscr{H}^{\mathrm{Spin}}_{e,p\cdot g}$, which leads us to \eqref{eq1-01-jan-2025}. 
\end{proof}


\paragraph{Spin connection for arbitrary frames.}

Finally, we will show that the above construction of right-invariant horizontal distribution on $\mathrm{Spin}^+(\mathcal{M})$ does not depend on the choice of   frame.  

\begin{proposition}
If $\{e_0,e_1,e_2,e_3\}$ and $\{f_0,f_1,f_2,f_3\}$ are two orthonormal frames defined on $\mathcal{M}$, the corresponding right-invariant horizontal distributions coincide:
\begin{equation} \label{eq2-01-jan-2025}
\mathscr{H}^{\mathrm{Spin}}_e = \mathscr{H}^{\mathrm{Spin}}_f.
\end{equation}
\end{proposition}

\begin{proof} We recall that Proposition~\ref{prop1-19-dec-2024} yields 
$
\Lambda_f = \Lambda_e\circ F_{f \to e} = \Lambda_e\circ r_g, 
$
in which, by \eqref{eq1-17-dec-2024},
$
g = \gamma_e^{-1}(\tau_{e \to f}).
$
This leads us to
\begin{equation} \label{eq3-01-jan-2025}
\diff_p\Lambda_f = \diff_{p\cdot g} \Lambda_e\circ \diff_pr_g.
\end{equation}
Then we return to \eqref{eq2-01-jan-2025}. By symmetry, we only need to show 
\begin{equation} \label{eq4-01-jan-2025}
\mathscr{H}^{\mathrm{Spin}}_e \subset \mathscr{H}^{\mathrm{Spin}}_f.
\end{equation}
Now let $Y\in \mathscr{H}^{\mathrm{Spin}}_{e,p}$, thus by the right-invariance of $\mathscr{H}^{\mathrm{Spin}}_{e,p}$, we have 
$
\diff_pr_g(Y) \in \mathscr{H}^{\mathrm{Spin}}_{e,p\cdot g}.
$
This gives us 
$
\diff_{p\cdot g} \Lambda_e\circ \diff_pr_g(Y) \in \mathscr{H}^{\mathrm{SO}}_{\Lambda_e(p\cdot g)}
$
and, by \eqref{eq3-01-jan-2025}, 
$$
\diff_p\Lambda_f(Y) \in \mathscr{H}^{\mathrm{SO}}_{\Lambda_e(p\cdot g)} = \mathscr{H}^{\mathrm{SO}}_{\Lambda_f(p)} \quad \Rightarrow \quad Y\in \mathscr{H}^{\mathrm{Spin}}_{f,p}, 
$$
which gives \eqref{eq4-01-jan-2025}.
\end{proof}

We arrive at the following notion.  

\begin{definition} The  \emph{spin connection} is defined as the right-invariant horizontal distribution
\begin{equation}
\mathscr{H}^{\mathrm{Spin}} := \mathscr{H}^{\mathrm{Spin}}_e, 
\end{equation}
where $e = \{e_0,e_1,e_2,e_3\}$ is \emph{any choice} of orthonormal frame with positive time-orientation and positive spatial orientation.
\end{definition}

This notion of a spin connection precisely ensures the desired local Lorentz (or spin) invariance under changes of frames; indeed, the spin connection describes how one can parallel transports spinors in a manner consistent with local $\mathrm{SO}_{3,1}^+$ (or $\mathrm{Spin}_{3,1}^+$) symmetry.  In other words, the notion of right-invariant horizontal distribution on a spin or frame bundle, and the related notion of covariant derivative, determines how spinors are parallel transported on the manifold. This connection encapsulates the interplay between the geometric structure of the manifold and the algebraic properties of spinors. 


\paragraph{Connections on spinor bundles.}

We then define the connection on the spinor bundle $\ourS_e(\mathcal{M})$, which follows from Proposition~\ref{prop1-01-jan-2025}. With the notation
$$
P_e: \mathrm{Spin}^+_e(\mathcal{M})\times \mathbb{C}^4 \mapsto \ourS_e(\mathcal{M}), \qquad (p,\psi)\mapsto [p,\psi], 
$$
the distribution 
$$
\mathscr{H}^e_{[p,\psi]} := \{\diff_{(p,\psi)}P_e(Y,0)| Y\in\mathscr{H}^{\mathrm{Spin}}_p\}
$$
is called as the  \emph{spinor connection} on $\ourS_e(\mathcal{M})$. At this stage of the discussion, it is important to point out that \emph{different orthonormal frames} on $\mathcal{M}$ lead to \emph{different spinor bundles} equipped with different connections (although they are generated from the same spin connection $\mathscr{H}^{\mathrm{Spin}}$). However, the canonical bundle isomorphism $F_{f \to e}$ (see Lemma~\ref{lem1-05-jan-2025}) enjoys the following property \eqref{eq2-08-jan-2025}, which shows that the \emph{canonical bundle isomorphisms between two spinor bundles} generated by two orthonormal bases are equivalent and furthermore, this isomorphism also preserves the spinor connections defined on the respective spinor bundles.

\begin{proposition}
\label{thm1-08-jan-2025}
With the notation above, if $\{e_0,e_1,e_2,e_3\}$ and $\{f_0,f_1,f_2,f_3\}$ are two orthonormal frames defined on $\mathcal{M}$, the corresponding right-invariant horizontal distributions are related as follows: 
\begin{equation} \label{eq2-08-jan-2025}
(F_{f \to e})_*(\mathscr{H}_{[p,\psi]}^f) = \mathscr{H}_{F_{f \to e}([p,\psi])}^e.
\end{equation}
\end{proposition}

\begin{proof} 
\bse
Thanks to \eqref{eq5-01-jan-2025}, we have $F_{f \to e} \circ P_f(p,\psi) = [F_{f \to e}(p),\psi] = [p,\tau_{e \to f}^{-1}\psi]$. 
Let $\gamma:(- \vep,\vep)\mapsto \mathrm{Spin}^+_f(\mathcal{M})$ with $\gamma(0) = p$ and $\dot\gamma(0) = Y\in\mathscr{H}^{\mathrm{Spin}}_p$. Then we have 
$$
F_{f \to e} \circ P_f(\gamma(t),\psi) = [\gamma(t),\tau_{e \to f}^{-1} \psi] = P_e(\gamma(t),\tau_{e \to f}^{-1}\psi).
$$
Differentiate the above identity with respect to $t$ at $t=0$ (regarded as curves in $\ourS_e(\mathcal{M})$), we obtain
\be
\diff_{[p,\psi]}F_{f \to e} \circ \diff_{(p,\psi)}P_f(Y,0) = \diff_{(p,\tau_{e \to f}^{-1}\psi)}P_e(Y,0).
\ee
Recall that 
\be
P_e(p,\tau_{e \to f}^{-1}\psi) = [p,\tau_{e \to f}^{-1}\psi] = [F_{f \to e}(p), \psi] = F_{f \to e}([p,\psi]),
\ee
then $\diff_{(p,\tau_{e \to f}^{-1}\psi)}P_e(Y,0)\in \mathscr{H}_{F_{f \to e}([p,\psi])}^e$ and this leads us to
\be
(F_{f \to e})_*(\mathscr{H}_{[p,\psi]}^f) \subset \mathscr{H}_{F_{f \to e}([p,\psi])}^e.
\ee
The reverse inclusion also holds by symmetry, an we reach \eqref{eq2-08-jan-2025}.
\ese
\end{proof}

}


\subsection{ Spinorial covariant derivative}
\label{subsec2-09-jan-2025}

{

\paragraph{Covariant derivative on vector bundles.}

Let $\pi_\mathcal{E}: \mathcal{E} \mapsto\mathcal{M}$ be a real vector bundle. Let $\mathscr{H}^\mathcal{E}$ be a horizontal distribution. Let 
$
E: U\mapsto \mathcal{P}_U
$
be a local section of $\mathcal{E}$. Then the covariant derivative of $E$ is defined as follows. Let $x\in U$ and $X\in T_x\mathcal{M}$. Let 
\be
\gamma:(- \vep,\vep)\mapsto \mathcal{M}, \qquad \gamma(0) = x, \qquad \dot\gamma(0) = X
\ee
be a curve on $\mathcal{M}$. Then
\begin{equation} \label{eq3-08-jan-2025}
\nabla_XE\big|_x := \Big(\frac{\mathrm{d}}{\mathrm{d}t} \Big|_{t=0}E\circ\gamma(t)\Big)^{\text{v}},
\end{equation}
in which 
$
(\cdot)^\text{v}: T_{E(x)} \mathcal{E} \mapsto T_{E(x)} \mathcal{E}_x
$
is the vertical projection with respect to the direct sum $T_{E(x)} \mathcal{E} = \mathscr{H}^{\mathcal{E}}_{E(x)} \oplus \mathscr{V}^{\mathcal{E}}_{E(x)}$. It is immediate to check the following expected properties:
\begin{subequations} \label{eqs1-27-feb-2025}
\begin{equation}
\nabla_{(\alpha X+\beta Y)}E = \alpha \nabla_XE + \beta\nabla_Y E, \qquad \alpha,\beta\in \RR, \qquad X,Y\in \Gamma(T\mathcal{M}),
\end{equation}
\begin{equation}
\nabla_X(fE) = \mathrm{d}f(X)E + f\nabla_XE, \qquad f\in C^{\infty}(\mathcal{M}).
\end{equation}
\end{subequations}
Observe that $E_x$ is a real vector field. We can identify $T_{E(x)} \mathcal{E}_x$ with $\mathcal{E}_x$, therefore 
$
\nabla_XE\big|_x\in \mathcal{E}_x.
$


\paragraph{Covariant derivatives on spinor bundles.}

Let us return to the spinor bundle. The following result is an alternative version of Proposition~\ref{thm1-08-jan-2025}. The proof is given Section~\ref{section-JD44}.

\begin{proposition}
\label{proposition-JD44} 
Let $\Phi: \mathcal{M} \mapsto S_f(\mathcal{M})$ be a section, and $\nabla^e$ and $\nabla^f$ are the covariant derivatives determined by the connection $\mathscr{H}^e$ and $\mathscr{H}^f$. Then for a (local) section $\Phi$ of $\mathrm{Spin}^+(\Mcal)$, one has 
\begin{equation} \label{eq4-08-jan-2025}
F_{f \to e} \big(\nabla_X^f\Phi\big) = \nabla^e_X\big(F_{f \to e} \circ \Phi\big). 
\end{equation}
\end{proposition}


\paragraph{Dirac operator.}

Let $\Phi$ be a section of $\mathrm{Spin}^+_e(\mathcal{M})$ with $e = (e_0,e_1,e_2,e_3)$ being a global orthonormal frame. Then the Dirac operator \emph{with respect to the frame $e$} reads 
\begin{equation} 
\opDirac_e \Phi:= \eta^{ij}e_i\nabla^e_{e_j} \Phi.
\end{equation}
The following observation is essential. It shows that the Dirac operator can be defined on the classes of spinor fields which are defined in different spinor bundles by different orthonormal frames. In the following analysis, we do not distinguish between a spinor field and its class.  Henceforth, any two spinor fields differing by a choice of global frame are identified via the bundle isomorphisms.  In particular, the Dirac operator is understood as an operator be defined with respect to \emph{any given orthonormal frame.} 

\begin{proposition}
If  $e = (e_0,e_1,e_2,e_3), f = (f_0,f_1,f_2,f_3)$ are two orthonormal frames with transition matrix given by 
$
(f_0,f_1,f_2,f_3) = (e_0,e_1,e_2,e_3)T_{e \to f}, 
$
one has
\begin{equation}
F_{f \to e} \big( \opDirac_f \Phi\big) = \opDirac_e   (F_{f \to e} \circ\Phi).
\end{equation}
\end{proposition}

\begin{proof} We recall \eqref{eq5-08-jan-2025} together with \eqref{eq4-08-jan-2025}, and compute 
$$
F_{f \to e} \big( \opDirac_f \Phi\big) = F_{f \to e} \big(\eta^{ij}f_i\cdot \nabla^f_{f_j} \Phi\big) 
= \eta^{ij}F_{f \to e}(f_i\cdot \nabla^f_{f_j} \Phi) = \eta^{ij}f_i\cdot \nabla^e_{f_j}(F_{f \to e} \circ\Phi)
= \opDirac_e (F_{f \to e} \circ\Phi). \qedhere
$$
\end{proof} 

}


\subsection{ Dirac sesquilinear forms}
\label{section===33}
{ 
 We generalize the notion of Dirac forms introduced in \eqref{eq1-12-jan-2025} to the setting of spinor bundles. Let $e = (e_0,e_1,e_2,e_3)$ be a global orthonormal frame, and let $\Phi,\Psi$ be two sections on $\mathrm{Spin}^+_e(\mathcal{M})$. We have the coordinate map 
$
\gamma_e: \mathrm{Spin}^+(T_x\mathcal{M})\to \mathrm{Spin}_{3,1}^+.
$
Similarly to the spinorial covariant derivative, we have the following property. Based on this observation,  for $\Phi,\Psi$ as representative elements of isomorphic sections, we can define the sesquilinear form 
\begin{equation}
\langle \Phi,\Psi \rangle_{\ourD} := \langle \gamma_0 \gamma_e(\Phi),\gamma_e(\Psi)\rangle.
\end{equation}

\begin{proposition}
If $e = (e_0,e_1,e_2,e_3)$ and $f = (f_0,f_1,f_2,f_3)$ are two orthonormal frames with transition matrix given by
$(f_0,f_1,f_2,f_3) = (e_0,e_1,e_2,e_3)T_{e \to f}$, then one has 
\begin{equation}
\langle \gamma_0 \gamma_e(F_{f \to e} \circ\Phi),\gamma_e(F_{f \to e} \circ\Psi) \rangle = \langle \gamma_0\gamma_f(\Phi),\gamma_f(\Psi)\rangle.
\end{equation}
\end{proposition}

\begin{proof}
We only need to observe that $F_{f \to e} \circ\Phi = \Phi \cdot \gamma_e^{-1}(\tau_{e \to f})$, 
which implies 
$
\gamma_e(F_{f \to e} \circ\Phi) = \gamma_e(\Phi)\tau_{e \to f}.
$
Then, by the spin-invariance of $\langle \gamma_0\cdot,\cdot\rangle$, the desired result follows.
\end{proof}

Next, we define a generalization of Dirac forms on spinor bundles. Let $e = (e_0,e_1,e_2,e_3)$ be an orthogonal frame. Let $\Phi,\Psi$ be two spinor fields with $\Phi: \mathcal{M} \to \ourS_e(\mathcal{M})$ and $ \Psi: \mathcal{M} \to \ourS_e(\mathcal{M})$
such that $\Phi = [p,\phi]$, $\Psi = [p,\psi]$, where $p$ is a local section $U \subset \mathcal{M} \to \mathrm{Spin}_e^+(\mathcal{M})$. We then define 
\begin{equation}
\langle \Phi,\Psi \rangle_{\ourD e} :=\langle \gamma_0\phi,\psi\rangle.
\end{equation}
We must check that $\langle \cdot,\cdot\rangle_{\ourD e}$ is well-defined. In fact, let $q: U \subset \mathcal{M} \to\mathrm{Spin}_e^+(\mathcal{M})$ be another local section. Suppose that
$\Phi = [q,\widetilde{\phi}]$ and $\Psi = [q,\widetilde{\psi}]$. 
Since $\mathrm{Spin}_{3,1}^+$ acts freely and transitively on each fiber $\mathrm{Spin}(T_x\mathcal{M})^+$ for $x \in U$, there exists a unique $\tau_{qp}$ such that 
$$
p(x) = q(x) \star \tau_{qp}(x) = q(x) \cdot \gamma_e^{-1}(\tau_{qp}).
$$
Then we have $\widetilde{\phi} = \tau_{qp} \phi, \quad \widetilde{\psi} = \tau_{qp} \psi$, and we conclude that 
$
\langle \gamma_0\widetilde{\phi},\widetilde{\psi} \rangle = \langle \gamma_0\phi,\psi\rangle,
$
due to the invariance of $\langle\gamma_0\cdot,\cdot \rangle$ under the spin group action.

Furthermore, this Dirac form extends naturally to equivalence classes of spinor fields.  

\begin{proposition}
If $e = (e_0,e_1,e_2,e_3)$ and $f = (f_0,f_1,f_2,f_3)$ are two orthonormal frames, then one has 
\begin{equation}
\langle F_{f \to e}(\Phi),F_{f \to e}(\Psi)\rangle_{\ourD e} = \langle \Phi,\Psi\rangle_{\ourD f}.
\end{equation}
\end{proposition}

\begin{proof}
Let $T_{e \to f} \in\mathrm{SO}_{3,1}^+$ such that
$
(f_0,f_1,f_2,f_3) = (e_0,e_1,e_2,e_3)T_{e \to f}.
$
Let $p:\mathcal{M} \to \mathrm{Spin}_f^+(\mathcal{M})$ and define:
$$
\Phi = [p,\phi], \quad \Psi = [p,\psi].
$$
Then we have 
$$
F_{f \to e}(\Phi) = [p,\tau_{e \to f}^{-1} \phi], \quad F_{f \to e}(\Psi) = [p,\tau_{e \to f}^{-1} \psi], 
$$
therefore
$$
\aligned
\langle F_{f \to e}(\Phi),F_{f \to e}(\Psi)\rangle_{\ourD e} 
& = \langle \gamma_0\tau_{e \to f}^{-1} \phi,\tau_{e \to f}^{-1} \psi\rangle 
 = \langle \gamma_0\phi,\psi\rangle = \langle \Phi,\Psi\rangle_{\ourD_f},
\endaligned
$$
using again the invariance of $\langle \gamma_0\cdot,\cdot\rangle$ under the spin group.
\end{proof}
From now on, we do not specify the frame and write simply
\be
\langle \Phi,\Psi\rangle_{\ourD}
\ee
for equivalence classes of spinor fields.
}


\section{Additional material on the gauge-invariant framework}
\label{section=N19}

\subsection{ Principal bundle structure} 
\label{annex-principal-bundle}

{ 

\paragraph{Action of a Lie group on a manifold.}

In order to support the geometric context in which we formulate our results of the present Monograph in a gauge-invariant manner, we outline here the abstract framework that is relevant to the description of spinors in a curved spacetime. We follow the presentation and (mostly) the notation in Hamilton's textbook~\cite{Hamilton-2017}. Let \(\mathcal{M} \) be a smooth manifold and \(G\) be a Lie group. By definition, a  \emph{right action} of \(G\) on \(\mathcal{M} \) is a smooth map
$
\mathcal{M} \times G  \to \mathcal{M},
\quad (x,g)\,\mapsto\, x\cdot g
$
satisfying
\be
x\cdot (g\cdot h) = (x\cdot g)\cdot h,
\quad x\cdot e = x
\qquad \text{for all } x \in \mathcal{M} \text{ and } g,h \in G,
\ee
where \(e\) denotes the identity element of \(G\). Then, given a point \(x\in \mathcal{M} \), the  \emph{orbit map} $\phi_x : G  \to \mathcal{M}$, namely 
\be
\phi_x(g) = x \cdot g
\quad
\text{for each } g \in G,
\ee
induces a family of  \emph{fundamental vector fields} on \(\mathcal{M} \), defined by differentiation at the identity element as 
\[
\widetilde{X}_{\,x\cdot g} = \diff_e \phi_x(X), 
\quad X \in T_e G = \mathfrak{g}.
\]

\bse
Right-translation by an element $g\in G$ maps a fundamental vector field to another fundamental vector field. In particular, if $\widetilde{X}$ is the fundamental vector field associated with $X\in\mathfrak{g}$ (the tangent space at the identity), then the push forward of $\widetilde{X}$ reads
\be
(r_g)*(\widetilde{X}) = \widetilde{Y},
\ee
where the corresponding Lie algebra element Y is given by
\be
Y = \mathrm{Ad}_{g^{-1}}(X).
\ee
That is, the effect of right-translation is equivalent to conjugating $X$ by $g^{-1}$. (Here, $\mathrm{Ad}$ denotes the adjoint representation of the Lie group, defined by the key equation 
\be
\mathrm{Ad}_g(X) = g\,X\,g^{-1}. 
\ee
Its differential at the identity yields the Lie algebra adjoint, which in our notation is given by
\be
\mathrm{ad}_{g}(h) := g^{-1} \,h\,g,
\ee
so that
\be
\diff_e\!\bigl(\mathrm{Ad}_{g^{-1}} \bigr)(X)= \mathrm{ad}_{g^{-1}}(X).
\ee
\ese
%


\paragraph{Basics facts about principal bundles.}

In short, a principal bundle is a manifold that locally looks like a product of a base manifold and a Lie group. For our purpose, we assume that the projection $\pi : \mathcal{P} \to \mathcal{M}$ is a surjective smooth map between smooth manifolds. For \(x\in \mathcal{M} \) and any open set \(U\subset \mathcal{M} \), we set
\[
\mathcal{P}_x = \pi^{-1}(\{x\}), 
\quad 
\mathcal{P}_U = \pi^{-1}(U).
\]
The set \(\mathcal{P}_x\) is called the  \emph{fiber} over \(x\).

\begin{definition}[Principal bundle as a locally trivial space]
Let \(\pi : \mathcal{P} \to \mathcal{M} \) be a smooth surjective map, and let \(G\) be a Lie group acting on \(\mathcal{P} \) from the right. Suppose that the action of \(G\) preserves each fiber:
\[
\mathcal{P}_x \times G  \to \mathcal{P}_x,
\]
and the map 
$
\phi_x : G  \to \mathcal{P}_x, 
\, g \,\mapsto\, x\cdot g 
$
is a bijection for all \(x \in \mathcal{M} \). Suppose also the following two properties. 

\bei 

\item For each \(x \in \mathcal{M} \), there is a local trivialization on an open neighborhood \(U\subset \mathcal{M} \), say 
\[
\phi_U : \mathcal{P}_U  \to U \times G,
\quad p \,\mapsto\, \phi_U(p)\,(\,=\,(\pi(p),\beta(p))\,),
\]
and \(\phi_U\) is a diffeomorphism and \(\mathrm{pr}_1 \circ \phi_U = \pi\) (where $\mathrm{pr}_1$ denotes the projection on the first component). 

\item The trivialization is  \emph{equivariant}, in the sense that 
$\phi_U(p\cdot g) = \phi_U(p)\cdot g$, 
where \(\phi_U(p) = (x,a)\) and
$
(x,h)\cdot g = (x,h\cdot g).
$
\end{itemize}
These properties define the \emph{principal bundle structure} and, by definition, in this context \(G\) is called the  \emph{structure group}.
\end{definition}

We will also use the following terminology. 
\bei

\item 
A  \emph{global section} of a principal bundle is a map \(s : \mathcal{M} \to \mathcal{P} \) such that
$
\pi \circ s = \mathrm{Id}_{\mathcal{M}}.
$

\item 
A  \emph{local section} is a map \(s\) defined on an open set \(U\subset \mathcal{M} \) with $\pi \circ s = \mathrm{Id}_U$. 
\eei 
\noindent 
A local section is also called a  \emph{local gauge}. Observe that specifying a local trivialization is equivalent to choosing a local section, by setting $s(x) = \phi_U^{-1} \bigl(x,e\bigr)$. Conversely, if for every \(x \in \mathcal{M} \) a local section is provided on some neighborhood of \(x\), then we can recover a principal bundle structure, as now stated. For further details, we refer to \cite[Sections 4.1 and 4.2]{Hamilton-2017}.

\begin{claim}
Let \(\pi : \mathcal{P} \to \mathcal{M} \) be a principal bundle with structure group \(G\). 
If \(\phi_U\) is a local trivialization defined on an open set \(U \subset \mathcal{M} \), then there exists a local section
$s : U  \to \mathcal{P}$. 
\end{claim}

\begin{claim}[Principal bundle defined by local gauge]\label{prop2-18-oct-2024}
Let \(\pi : \mathcal{P} \to \mathcal{M} \) be a surjective map between smooth manifolds endowed with a smooth action $\Phi : \mathcal{P} \times G  \to \mathcal{P}$. Suppose that the following two properties hold. 
\begin{itemize}
\item The action \(\Phi\) preserves the fibers of \(\pi\) and is free and transitive on each fiber:
\[
\aligned
& \mathcal{P}_x \times G  \to \mathcal{P}_x, \qquad \qquad 
\quad G_x = \{e\} \ (\text{free}), 
\quad 
\\
& \text { for all } \, p,q \in \mathcal{P}_x,\ \text{ there exists } g \in G \text{ such that } p = q\cdot g\ (\text{transitive}).
\endaligned
\]

\item There is an open covering \(\{U_{\alpha} \} \) of \(\mathcal{M} \) such that on each \(U_{\alpha} \), there is a local section \(s_{\alpha} \).
\end{itemize}
Then, \(\pi : \mathcal{P} \to \mathcal{M} \) is a principal bundle with structure group \(G\).
\end{claim}

\begin{proof}[Sketch of proof] We can construct a local trivialization on each \(U_{\alpha} \) by setting 
\[
\phi_{\alpha}^{-1}: \ U_{\alpha} \times G  \to \mathcal{P}_{U_{\alpha}}, 
\quad 
(x,g)\,\mapsto\, s_{\alpha}(x)\cdot g.
\]
It can be checked that this map is equivariant and that \(\phi_{\alpha} \) is indeed a diffeomorphism. In particular, $\phi_{\alpha}$ is one-to-one and is an immersion. 
\end{proof}

}


\subsection{ Proof of Claim \ref{lem1-15-dec-2024}}

{ 

\bse
In view of the anti-commutation properties \eqref{equa-29D}, we have (for $a=0$) 
\be
\la \gamma_0\gamma_\alpha\psi,\gamma_{\beta} \phi\ra + \la \gamma_0\gamma_{\beta} \psi,\gamma_{\alpha} \phi\ra 
= -2\eta_{\alpha\beta} \la \gamma_0\psi,\phi\ra.
\ee
Given any $v \in \RR^{3,1}$, we have 
\be
\la \gamma_0 v\cdot \psi,v\cdot\phi\ra = - \eta(v,v)\la \gamma_0 \psi,\phi\ra,
\ee
therefore for any $v\in S_{3,1}^{\pm}$ we find 
\be
\la\gamma_0 v\cdot \psi,v\cdot\phi\ra = \mp\la \gamma_0\psi,\phi \ra.
\ee
Any element $u\in \text{Sp}_{3,1}^+$ is composed of an even number of $\gamma(v)$ (for some $v\in \RR^{3,1}$) and, consequently, 
\be
\la \gamma_{\alpha}u\psi,u\phi\ra = \la \gamma_{\alpha} \psi,\phi\ra.
\ee
\ese

}


\subsection{ Proof of Lemma~\ref{lem1-16-dec-2024}}

{ 

As stated in \eqref{eq1-16-dec-2024}, the group \(\mathrm{Spin}_{3,1}^+\) acts from the right on \(\mathrm{Spin}^+(\mathcal{M})\), so for any $u\in\mathrm{Spin}_{3,1}^+$ we can write 
\begin{equation} \label{eq4-16-dec-2024}
\aligned
\Lambda(p\star u) 
& =  \gamma_e^{-1}(\gamma_e(p)\ u)(e_0,e_1,e_2,e_3)\big(\gamma_e^{-1}(\gamma_e(p)\cdot u)\big)^{-1}
\\
& = p\cdot\gamma_e^{-1}(u)(e_0,e_1,e_2,e_3)\big(\gamma_e^{-1}(u)\big)^{-1} \cdot p^{-1}.
\endaligned
\end{equation}
Let $B\in \mathrm{SO}_{3,1}^+$ such that
$$
\gamma_e^{-1}(u)(e_0,e_1,e_2,e_3)\big(\gamma_e^{-1}(u)\big)^{-1} = (e_0,e_1,e_2,e_3)B.
$$
Acting the above identity by $\gamma_e$ and observing that $\gamma_e^{-1}(u^{-1}) = \big(\gamma_e^{-1}(u)\big)^{-1}$, we obtain 
$$
u(\gamma_0,\gamma_1,\gamma_2,\gamma_3)u^{-1} = (\gamma_0,\gamma_1,\gamma_2,\gamma_3)B, 
$$
so this gives us 
\begin{equation} \label{eq3-16-dec-2024}
\gamma_e^{-1}(u)(e_0,e_1,e_2,e_3)\big(\gamma_e^{-1}(u)\big)^{-1} = (e_0,e_1,e_2,e_3)\lambda(u)
\end{equation}
where we emphasize that $\lambda(u)\in\mathrm{SO}_{3,1}^+$. Substituting this into \eqref{eq4-16-dec-2024}, we obtain
$$
\Lambda(p\star u) = p\cdot (e_0,e_1,e_2,e_3)\cdot p^{-1} \lambda(u) = \Lambda(p) \lambda(u).
$$
This completes the proof of Lemma~\ref{lem1-16-dec-2024}. 

}


\subsection{ Proof of Proposition~\ref{prop1-19-dec-2024}}

{ 

We need the following observation. 

\begin{lemma}
Let $(e_0,e_1,e_2,e_3)$ and $(f_0,f_1,f_2,f_3)$ be two orthonormal frames on $\mathcal{M}$. For any $u\in\mathrm{Spin}_{3,1}^+$, one has 
\begin{equation} \label{eq4-18-dec-2024}
\tau_{e \to f} \cdot u\cdot \tau_{e \to f}^{-1} = \gamma_e\circ\gamma_f^{-1}(u).
\end{equation}
\end{lemma}

\begin{proof} This is in fact a reformulation of Lemma~\ref{lem1-30-aout-2025}. For $u$ being a vector,
We decompose $u = u^{\alpha} \gamma_{\alpha}$, and set $U := (u^0,u^1,u^2,u^3)^T$. Then we have 
$$
\gamma_f^{-1}(u) = u^{\alpha}f_{\alpha} = (f_0,f_1,f_2,f_3)U = (e_0,e_1,e_2,e_3)T_{e \to f}U.
$$
We then compute successively 
$$ 
\aligned
& \gamma_e\circ\gamma_f^{-1}(u) = (\gamma_0,\gamma_1,\gamma_2,\gamma_3)T_{e \to f}U 
= \tau_{e \to f} \cdot(\gamma_0,\gamma_1,\gamma_2,\gamma_3)\cdot\tau_{e \to f}^{-1}U 
\\
& = \tau_{e \to f} \cdot\big((\gamma_0,\gamma_1,\gamma_2,\gamma_3)U\big)\cdot\tau_{e \to f}^{-1} 
   =  \tau_{e \to f} \cdot u\cdot \tau_{e \to f}^{-1}. \hskip4.cm  \qedhere
\endaligned 
$$
And for $u$ a product of finite many vectors, the desired property still holds obviously.
\end{proof}

\begin{proof}[Proof of Proposition~\ref{prop1-19-dec-2024}]
We only need to prove the two properties in Definition~\ref{def1-17-dec-2024}. Observe that $F_{f \to e}$ denotes the right action of $\tau_{e \to f}$ on $\text{Sp}^+(\mathcal{M})$, and therefore preserves the fiber. We also recall that
$$
\Lambda_f(p) = \Lambda_e(p)T_{e \to f} = \Lambda_e(p)\lambda(\tau_{e \to f}).
$$
For any $u\in\text{Sp}_{3,1}^+$, we then have 
$$
\aligned
&F_{f \to e}(p\cdot \gamma_f^{-1}(u)) = p\cdot\gamma_f^{-1}(u)\cdot\gamma_e^{-1}(\tau_{e \to f}) 
= p\cdot \gamma_e^{-1}(\gamma_e\circ\gamma_f^{-1}(u)\cdot \tau_{e \to f}) 
\\
& = p\cdot \gamma_e^{-1}(\tau_{e \to f})\cdot \gamma_e^{-1}(\tau_{e \to f}^{-1} \cdot \gamma_e\circ\gamma_f^{-1}(u)\cdot \tau_{e \to f})
  =  F_{f \to e}(p)\cdot \gamma_e^{-1}(u), 
\endaligned
$$
which establishes the equivariance with respect to $G$.
\end{proof}

}


\section{Properties of the spinorial connection}
\label{section=N20}

\subsection{ Construction of the Levi-Civita connection on $\mathrm{SO}^+(\Mcal)$}
\label{Appendix--B5}

{ 

\paragraph{Connection form.}

Let $\mathcal{P}$ be a principal bundle with structure group $G$. A connection $1$-form $A$ is a differential form defined on $\mathcal{P}$ with value in $\mathfrak{g}$ with the following properties:
\bei 

\item[1.] Let $\widetilde{X} \in \mathscr{V}_p$, then $A(\widetilde{X}) = X$;

\item[2.] One has $(r_g)_*(A) = Ad_{g^{-1}} \circ A$.

\eei
\noindent 
The relation between a connection 1-form and a horizontal distribution is described as follows. 

\begin{proposition}[cf. Theorem 5.22 of \cite{Hamilton-2017}]
There is a one-to-one correspondence between the connection 1-forms and the right-invariant horizontal distributions.

\bei 

\item[1.] Let $\mathscr{H}$ be a right-invariant horizontal distribution (an Ehresmann connection). With the notation $\widetilde{X}_p + Y_p\in T_pP$ with $\widetilde{X}_p \in V_p$ and $Y_p\in H_p$, the mapping  
$$
A_p(\widetilde{X}_p + Y_p) := X\in \mathfrak{g}
$$
is a connection 1-form. 

\item[2.] Let $A$ be a connection 1-form. Then $\mathscr{H}_p = \ker A_p$ is a right-invariant horizontal distribution.

\eei 

\end{proposition}

In a local trivialization (or equivalently, a local gauge), the connection $1$-form can be written as follows. Let $s$ be the local section defined on an opens set $U\subset\mathcal{M}$ which determines the local trivialization by
$$
U\times G\mapsto \mathcal{P}_U, \qquad (x,g)\mapsto s(x)\cdot g.
$$
We denote by, for $X\in T_x\mathcal{M}$
$$
A^s_x(X) := A\circ d_xs(X) = s^*A(X)\in \mathfrak{g}.
$$
This $1-$form defined on $\mathcal{M}$ is called a local connection $1-$form.

Let $\mathcal{E} = \mathcal{P} \times_{\rho}V$ be a vector bundle associated with the principal bundle $\mathcal{P}$. Let $s: U\mapsto \mathcal{P}$ be a local section of $\mathcal{P}$. We recall that [cf. Proposition 4.7.6 of Hamilton] any local section of $\mathcal{E}$ can be written as $E(x) = [s(x), v(x)]$, where $v$ is a $C^{\infty}$ map from $U$ to $V$. Then the covariant derivative can be written within the above expression as (cf. Definition 5.9.3 of Hamilton)
\begin{equation} \label{eq6-01-jan-2025}
\nabla_XE = [s(x), \diff\psi(X) + \rho_*(A_s(X))\psi].
\end{equation}


\paragraph{Proof of Proposition~\ref{thm2-08-jan-2025}.}

Proposition~\ref{thm2-08-jan-2025} provides the Levi-Civita connection on $\mathrm{SO}^+(\mathcal{M})$. Let $\nabla$ be the Levi-Civita connection defined on $T\mathcal{M}$. Let $p = (p_0,p_1,p_2,p_3) \in \mathrm{SO}^+(\mathcal{M})$. Let $x = \pi(p)\in \mathcal{M}$. Then $(p_0,p_1,p_2,p_3)$ is an orthonormal basis of $T_x\mathcal{M}$ with positive time and spatial orientation. Let $\gamma:(- \vep,\vep)\mapsto \mathcal{M}$ be a smooth curve with $\gamma(0) = x$, $\dot\gamma(0) = X\in T_x\mathcal{M}$. Let $E_i$ be the vector field defined along $\gamma$ by
\begin{equation} \label{eq2-09-jan-2025}
\aligned
&E_i(0) = p_i(x) \quad (i=0,1,2,3), \qquad \nabla_{\dot{\gamma}(t)}E_i(t) = 0.
\endaligned
\end{equation}
Let 
$
H_p: T_x\mathcal{M} \mapsto T_p\mathrm{SO}^+(\mathcal{M}), \qquad X\mapsto H_p(X)
$
be the mapping defined by
$$
H_p(X) = \frac{\mathrm{d}}{\mathrm{d}t} \Big|_{t=0}(\gamma(t),E_0(t),E_1(t),E_2(t),E_3(t)).
$$

\begin{lemma}
The map $H_p$ is linear and injective.
\end{lemma}

\begin{proof} It is immediate that
$H_p(X) = (X, \dot E_0(0), \dot E_1(0),\dot E_2(0), \dot E_3(0))$. 
Let $(e_0,e_1,e_2,e_3)$ be an orthonormal frame defined around $x = \pi(p)$ and $E_i(t) = e_jE_i^j(t)$. Then
$$
\nabla_XE_i = \dot E_i^j e_j + E_i^j\nabla_Xe_j.
$$
At $t=0$, one has $\dot E_i^j e_j = -E_i^j(0)\nabla_Xe_j$
and we arrive at 
\begin{equation} \label{eq1-09-jan-2025}
H_p(X) = (X, -E_0^j(0)\nabla_Xe_j,-E_1^j(0)\nabla_Xe_j,-E_2^j(0)\nabla_Xe_j,-E_3^j(0)\nabla_Xe_j).
\end{equation}
It is clear that $H_p$ is a linear isomorphism.
\end{proof}

For a different choice of orthonormal frame, we show the following property. Consequently, by following the procedure described in~Sections~\ref{subsec1-09-jan-2025} and \ref{subsec2-09-jan-2025}, we can construct the covariant derivative operator on $\ourS_e(\mathcal{M})$.

\begin{proposition}
\label{la--construction}
For each $p\in\mathrm{SO}^+(\mathcal{M})$, the image of $H_p$, denoted by $\mathscr{H}^{\mathrm{SO}}_p$, forms a Ehresmann connection on $\mathrm{SO}^+(\mathcal{M})$.
\end{proposition}
\begin{proof}
Let $U$ be a open neighborhood of $x = \pi(p)$ on $\mathcal{M}$ and $(x^0,x^1,x^2,x^3)$ be a local coordinate system. It is clear that $H_p(\del_{\alpha}$ generates $\mathscr{H}^{\mathrm{SO}}_p$. For the smoothness, we only need to show that 
$$
Y_{\alpha}: q\mapsto H_q(\del_{\alpha})
$$
are (locally defined) smooth vector fields of $\mathrm{SO}^+(\mathcal{M})$. This can be checked by \eqref{eq1-09-jan-2025} where $H_p(X)$ smoothly depends on $E_i^j$ and $x$. Then it also direct to see that $\mathscr{H}^{\mathrm{SO}}_p\oplus \mathscr{V}_p = T_p\mathrm{SO}^+(\mathcal{M})$ because
$$
\mathscr{V}_p = \ker(d_p\pi) = \{(0,A_i^je_j)| A\in \mathrm{SO}_{3,1}^+\}.
$$
Finally we check that $\mathscr{H}^{\mathrm{SO}}$ is invariant under the right action of $\mathrm{SO}_{3,1}^+$. To see this we observe that for $A\in \mathrm{SO}_{3,1}^+$,
$$
r_A(p) = p\cdot A = (p_0,p_1,p_2,p_3)A.
$$
This is a linear action thus $d_pr_A(p) = r_A(p)$.

If $E_i$ denotes the parallel vector fields defined by \eqref{eq2-09-jan-2025}, then 
$F_j := E_iA_j^i = e_kE_i^kA_j^i$
are the curves translated from $E_i$ by $A$. Due to the fact that $A$ preserves the metric, $F_i$ are still orthonormal and parallel along $\gamma$. Thus
$$
\aligned
H_{p\cdot A}(X) & =  \frac{\mathrm{d}}{\mathrm{d}t} \Big|_{t = 0}(\gamma(t), F_0(t),F_1(t),F_2(t),F_3(t))
\\
& =  (X, -E_i^jA_0^i\nabla_Xe_j,-E_i^jA_1^i\nabla_Xe_j,-E_i^jA_2^i\nabla_Xe_j,-E_i^jA_3^i\nabla_Xe_j)
\\
& = (X, -E_1^j(0)\nabla_Xe_j,-E_2^j(0)\nabla_Xe_j,-E_3^j(0)\nabla_Xe_j)\cdot A, 
\endaligned
$$
which yields the desired result.
\end{proof}

}

\subsection{ Proof of Proposition~\ref{proposition-JD44}}
\label{section-JD44}

{ 

Observe that $\nabla^f_X\Phi = \big(\diff_x\Phi(X)\big)^{\text{v}}$, where $(\cdot)^{\text {v}}$ is the projection with respect to the decomposition
$$
T_{\Phi(x)} \ourS_f(\mathcal{M}) = \mathscr{H}^f_{\Phi(x)} \oplus \mathscr{V}_{\Phi(x)}, 
\qquad
\quad 
\mathscr{V}_{\Phi(x)} = \ker(\diff_{\Phi(x)} \pi_{\ourS_f(\mathcal{M})}).
$$
Consequently, in view of Proposition~\ref{thm1-08-jan-2025}, we have 
$$
\diff_x\Phi(X) - \nabla^f_X\Phi\in  \mathscr{H}^f_{\Phi(x)} \quad \Rightarrow \quad \diff_{\Phi(x)}F_{f \to e}(\diff_x\Phi(X) - \nabla^f_X\Phi) \in \mathscr{H}^e_{F_{f \to e}(\Phi(x))}.
$$
Observe also that
$
\diff_{\Phi(x)}F_{f \to e}(\diff_x\Phi(X)) = \diff_x(F_{f \to e} \circ\Phi)(X) 
$
and therefore, 
$$
\diff_x(F_{f \to e} \circ\Phi)(X)  - \diff_{\Phi(x)}F_{f \to e}(\nabla^f_X\Phi\big|_x) \in \mathscr{H}^e_{F_{f \to e}(\Phi(x))}. 
$$
Since $\nabla^f_X\Phi\in \mathscr{V}_{\Phi(x)}$, we also claim that
\begin{equation}
\diff_{\Phi(x)}F_{f \to e}(Y) \in \mathscr{V}_{\Phi(x)},
\qquad Y\in  \mathscr{V}_{\Phi(x)}.
\end{equation}
Indeed, this follows from the fact (cf. \eqref{eq5-01-jan-2025}) that $F_{f \to e}$ is linear on each fiber $\ourS_f(\mathcal{M})_x$ whose tangent space is precisely $\mathscr{V}_{\Phi(x)}$. By recalling the definition of covariant derivative (cf.~\eqref{eq3-08-jan-2025}), this leads us to
\begin{equation}
\diff_{\Phi(x)}F_{f \to e}(\nabla^f_X\Phi\big|_x) = \big(\diff_x(F_{f \to e} \circ\Phi)(X)\big)^{\text{v}} = \nabla^e_X\big(F_{f \to e} \circ\Phi\big)_x.
\end{equation}
This completes the proof of Proposition~\ref{proposition-JD44}. 
}


\section{Proof of Lemma \ref{lem1-01-march-2025}} 
\label{section=N21}

{ 
We work in a tetrad $\{e_{\alpha}\}$. Let $\psi$ be the coordinate of $\Psi$. We recall the Levi-Civita connection form associate to $\{e_{\alpha}\}$ (cf. \cite[Definition~6.10.1]{Hamilton-2017})
$$
\nabla e_{\alpha} = \omega_{\alpha\beta}\eta^{\beta\gamma}\otimes e_{\gamma},
$$  
and, thanks to \cite[Prpposition~6.10.9]{Hamilton-2017},
\begin{equation}
\nabla_X\psi = \diff\psi (X) + \frac{1}{4}\omega_{\alpha\beta}(X)\gamma^{\alpha\beta}\psi,
\end{equation}
where we recall the Clifford algebra relations $\gamma^{\alpha\beta} : =\frac{1}{2}[\gamma^{\alpha},\gamma^{\beta}]$. We rely on the identity $\diff\omega(X,Y) = X(\omega(Y)) - Y(\omega(X)) - \omega([X,Y])$ and compute
$$
\aligned
\nabla_X(\nabla_Y\psi) 
& =  \nabla_X(\diff\psi(Y)) + \frac{1}{4}\nabla_X\big(\omega_{\alpha\beta}(Y)\gamma^{\alpha\beta} \psi\big)
\\
& = X(Y(\psi)) + \frac{1}{4} \omega_{\mu\nu}(X)\gamma^{\mu\nu}Y(\psi) 
+ \frac{1}{4}X(\omega_{\alpha\beta}(Y)\gamma^{\alpha\beta} \psi)
+ \frac{1}{16} \omega_{\mu\nu}(X)\gamma^{\mu\nu} \omega_{\alpha\beta}(Y)\gamma^{\alpha\beta} \psi
\\
& = X(Y(\psi))  + \frac{1}{4} \omega_{\mu\nu}(X)\gamma^{\mu\nu}Y(\psi)
+ \frac{1}{4} \omega_{\alpha\beta}(Y)\gamma^{\alpha\beta}X(\psi)
+ \frac{1}{4}X(\omega_{\alpha\beta}(Y)\gamma^{\alpha\beta})\psi
\\
& \quad + \frac{1}{16} \omega_{\mu\nu}(X)\gamma^{\mu\nu} \omega_{\alpha\beta}(Y)\gamma^{\alpha\beta} \psi.
\endaligned
$$
then
$$
\gamma^{\alpha\beta} =
\begin{cases}
0,\quad & \alpha=\beta
\\
\gamma^\alpha \gamma^\beta,\quad & \alpha\neq \beta,
\end{cases}
\quad 
[\gamma^{\alpha\beta},\gamma^{\mu\nu}] =
\begin{cases} 
0,& \quad \alpha = \beta,\quad \text{or}\quad \mu=\nu,
\\
[\gamma^\alpha \gamma^\beta,\gamma^\mu \gamma^\nu], & \quad \alpha\neq \beta,\quad \text{and}\quad \mu\neq \nu.
\end{cases}
$$
and also use
$$
\aligned
&[\gamma^\alpha \gamma^\beta,\gamma^\mu \gamma^\nu] = 0, \quad \text{if all indices are distinct,}
\\
&[\gamma^\alpha \gamma^\beta,\gamma^\alpha \gamma^\mu] = 2\eta^{\alpha\alpha} \gamma^\beta \gamma^\mu,
\qquad
&&[\gamma^\alpha \gamma^\beta,\gamma^\mu \gamma^\alpha] = -2\eta^{\alpha\alpha} \gamma^\beta \gamma^\mu,
\\
&[\gamma^\alpha \gamma^\beta,\gamma^\beta \gamma^\mu] = -2\eta^{\beta\beta} \gamma^\alpha \gamma^\mu,
\qquad
&&[\gamma^\alpha \gamma^\beta,\gamma^\mu \gamma^\beta] = 2\eta^{\beta\beta} \gamma^\alpha \gamma^\mu,
\endaligned
$$
and the remaining cases are also zero.
 
We compute the difference of products of connection forms, as follows: 
$$
\aligned
& \omega_{\mu\nu}(X)\gamma^{\mu\nu} \omega_{\alpha\beta}(Y)\gamma^{\alpha\beta} 
- \omega_{\mu\nu}(Y)\gamma^{\mu\nu} \omega_{\alpha\beta}(X)\gamma^{\alpha\beta}
\\
& = - \omega_{\nu\mu}(X)\omega_{\alpha\beta}(Y)[\gamma^{\alpha\beta},\gamma^{\nu\mu}]
\\
& = - \omega_{\alpha\mu}(X)\omega_{\alpha\beta}(Y)[\gamma^\alpha \gamma^\beta, \gamma^\alpha \gamma^\mu]
- \omega_{\mu\alpha}(X)\omega_{\alpha\beta}(Y)[\gamma^\alpha \gamma^\beta, \gamma^\mu \gamma^\alpha]
\\
& \quad - \omega_{\beta\mu}(X)\omega_{\alpha\beta}(Y)[\gamma^\alpha \gamma^\beta, \gamma^\beta \gamma^\mu]
- \omega_{\mu\beta}(X)\omega_{\alpha\beta}(Y)[\gamma^\alpha \gamma^\beta, \gamma^\mu \gamma^\beta]
\\
& = -2\omega_{\alpha\mu}(X)\omega_{\alpha\beta}(Y)\eta^{\alpha\alpha}\gamma^{\beta}\gamma^{\mu}
+ 2\omega_{\mu\alpha}(X)\omega_{\alpha\beta}(Y)\eta^{\alpha\alpha}\gamma^{\beta}\gamma^{\mu}
\\
& \quad + 2\omega_{\beta\mu}(X)\omega_{\alpha\beta}(Y)\eta^{\beta\beta}\gamma^{\alpha}\gamma^{\mu}
- 2\omega_{\mu\beta}(X)\omega_{\alpha\beta}(Y)\eta^{\beta\beta}\gamma^{\alpha}\gamma^{\mu}
\\
& = 8\omega_{\alpha\mu}(X)\omega_{\alpha\beta}(Y)\eta^{\alpha\alpha}\gamma^{\mu}\gamma^{\beta}
\\
& = 4\big(\omega_{\alpha\mu}(X)\omega_{\alpha\beta}(Y) - \omega_{\alpha\beta}(X)\omega_{\alpha\mu}(Y)\big)\gamma^{\mu}\gamma^{\beta}
\\
& = 4\omega_{\mu\alpha}\wedge\omega_{\beta}^{\ \alpha}(X,Y)\gamma^{\mu\beta}
\endaligned
$$
Hence, we find
$$
\aligned
(\nabla_X\nabla_Y - \nabla_Y\nabla_X)\psi 
& = [X,Y](\psi) + \frac{1}{4} \omega_{\alpha\beta}([X,Y]) \gamma^{\alpha\beta} \psi 
+ \frac{1}{4} \diff\omega_{\alpha\beta}(X,Y) \gamma^{\alpha\beta} \psi 
\\
& \quad + \frac{1}{16} \Big(
\omega_{\mu\nu}(X)\gamma^{\mu\nu} \omega_{\alpha\beta}(Y)\gamma^{\alpha\beta}
- \omega_{\mu\nu}(Y)\gamma^{\mu\nu} \omega_{\alpha\beta}(X)\gamma^{\alpha\beta}
\Big)
\\
& = \nabla_{[X,Y]}\psi + \frac{1}{4} \diff\omega_{\alpha\beta}(X,Y) \gamma^{\alpha\beta} \psi 
+\frac{1}{4} \omega_{\alpha\gamma} \wedge \omega_\beta{}^{\gamma}(X,Y) \gamma^{\alpha\beta} \psi
\\
& = \nabla_{[X,Y]}\psi + \frac{1}{4} g(R(X,Y)e_\alpha, e_\beta) \gamma^{\alpha\beta} \psi,
\endaligned
$$
which is the desired identity. 

}


\section{The uniformly-spacelike property}
\label{section=N22}

\subsection{Some linear algebra}

We first establish the following result for preparation.

\begin{lemma}\label{lem1-14-june-2025}
When $|H|\leq \eps_s$ with $\eps_s$ sufficiently small, 
\begin{equation}
g(\delb_r,\delb_r) = \zeta^2 - H^{\Ncal 00} - \frac{2\zeta^2}{1+\delb_rT}(H_{00} + (x^a/r)H_{0a})
+ \frac{\zeta^4}{(1+\delb_rT)^2}H_{00}
 + O(|H|^2).
\end{equation}
\end{lemma}
\begin{proof}
\begin{equation}\label{eq2-02-oct-2025(l)}
\aligned
H^{\Ncal 00} & =  g(\diff t - \diff r, \diff t - \diff r)
= g^{00} - 2\frac{x^a}{r}g^{0a} + \frac{x^ax^b}{r^2}g^{ab}
\\
& = H^{00} - 2\frac{x^a}{r}H^{0a} + \frac{x^ax^b}{r^2}H^{ab}
\\
& = -H_{00} - \frac{x^a}{r}H_{0a} - \frac{x^ax^b}{r^2}H_{ab} + O(|H|^2)
\\
& = -H(\del_t+\del_r,\del_t+\del_r) + O(|H|^2).
\endaligned
\end{equation}
$$
\aligned
g(\delb_r,\delb_r) & =  g\big(\delb_rT\del_t+\del_r, \delb_rT\del_t+\del_r\big)
\\
& =  \zeta^2 + H\big(\delb_rT\del_t+\del_r, \delb_rT\del_t+\del_r\big)
\\
& = 
\zeta^2 + H(\del_t+\del_r,\del_t+\del_r) - 2(1- \delb_rT)H(\del_t,\del_t+\del_r) + (1- \delb_rT)^2H(\del_t,\del_t)
\\
& = \zeta^2 -H^{\Ncal 00} + O(|H|^2) - 2(1- \delb_rT)H(\del_t,\del_t+\del_r) + (1- \delb_rT)^2H(\del_t,\del_t).
\endaligned
$$
\end{proof}

\begin{corollary}\label{cor1-14-june-2025}
Suppose that
\begin{equation}
H^{\Ncal 00}<0
\end{equation}
and for some sufficiently small $\eps_s$,
\begin{equation}
|H|\leq \eps_s\zeta,
\end{equation}
then one has
\begin{equation}
g(\delb_r,\delb_r)>\frac{3}{4}\zeta^2 + |H^{\Ncal 00}|.
\end{equation}
\end{corollary}


\subsection{Proof of Proposition~\ref{prop1-14-june-2025}}

We denote by 
\be
\mathcal{S}(s,r) := \big\{(t,x)| t = T(s,r), |x| = r\big\}\subset \Mcal_s
\ee
the Euclidean sphere contained in $\Mcal_s$. Let $\{(t_i,x_i)\}\in \mathcal{S}(s,r)$ be a set of points. Then for each $i$ there exists an open subset $U_i$ of $\mathcal{S}$ containing $(t_i,x_i)$ with a orthogonal frame $\{X_i,Y_i\}$ defined on it.
Then we have 
\be
(\delb_r, X_i, Y_i) = (\delb_1,\delb_2,\delb_3)A_i
\ee
with
\be
A_i = ({A_i}_a^b) =
\left(
\begin{array}{ccc}
x^1/r &*&*
\\
x^2/r &* &*
\\
x^3/r &* &*
\end{array}
\right), \qquad A_iA_i^{\mathrm{T}} = \mathrm{I}_3.
\ee
Let
$$
P_i := ({P_i}_{ab}) =
\left(
\begin{array}{ccc}
g(\delb_r,\delb_r) &g(\delb_r,X_i) &g(\delb_r ,Y_i)
\\
g(\delb_r,X_i)&g(X_i,X_i) &g(X_i,Y_i)
\\
g(\delb_r ,Y_i)&g(Y_i,X_i) &g(Y_i,Y_i)
\end{array}
\right),
\quad
\overline{g} := 
\left(
\gb_{ab}
\right),
\quad
\overline{\eta} := \left(\overline{\eta}_{ab}\right).
$$
Then
$$
P_i = A_i^{\mathrm{T}}\gb A_i
$$
with 
\begin{equation}\label{eq7-14-june-2025}
E:=\left(\begin{array}{cc}
\zeta^2 &0
\\
0 & \mathrm{I}_2
\end{array}\right)
= A_i^{\mathrm{T}}\overline{\eta} A_i.
\end{equation}
By shrinking (if necessary) a bit the open sets $U_i$, one can suppose that
$$
|{A_i}_a^b|\leq C_1
$$
with $C_1$ a universal constant
\footnote{
This can be observed as follows. First, suppose that $(t_*,x_*) = \big(\sqrt{s^2+r^2}, r,0,0\big)$. Then let $X_* = r^{-1}\del_{\theta}, Y_* = (r\cos\theta)^{-1}\del_{\varphi}$. Then in a neighborhood of $(t_*,x_*)$ on $\mathcal{S}_{s,r}$ $X,Y$ are well defined with ${A_*}_a^b$ being uniformly bounded (independent of $(s,r)$). For a general point $(t_i,x_i)\in\mathcal{S}(s,r)$, $\{X_i,Y_i\}$ and ${A_i}_a^b$ can be obtained by an Euclidean rotation. Then due to the compactness $O(\RR^3)$, the orthogonal matrices are uniformly bounded, which leads us to the uniform boundedness of  ${A_i}_a^b$.
}.
We remark that \eqref{eq7-14-june-2025} leads us to
\begin{equation}\label{eq8-14-june-2025}
\left(\begin{array}{cc}
\zeta^2+{a_i}_{11} &w_i^{\mathrm{T}}
\\
w_i & \mathrm{I}_2 + D_i
\end{array}\right)
:=
E + A_i^{\mathrm{T}}\big(\Hb_{ab}\big)A_i
= 
A_i^{\mathrm{T}}\gb A_i.
\end{equation}
Thus
\begin{equation}\label{eq12-14-june-2025}
|{a_i}_{11}| + |w_i| + |D_i|\lesssim |\Hb|\leq C_2\eps_s\zeta
\end{equation}
with $C_2$ a universal constant. When \eqref{eq1-14-june-2025} and \eqref{eq2-14-june-2025} hold, Corollary~\ref{cor1-14-june-2025} tell us that
\begin{equation}\label{eq11-14-june-2025}
{P_i}_{11} = \zeta^2 + {a_i}_{11} >\frac{3}{4}\zeta^2+|H^{\Ncal 00}|.
\end{equation}
More precisely, Lemma~\ref{lem1-14-june-2025},
\begin{equation}\label{eq2-29-sept-2025}
\aligned
{P_i}_{11} & =  \zeta^2 + |H^{\Ncal 00}| - \frac{2\zeta^2}{1+\delb_rT}(H_{00} + (x^a/r)H_{0a})
+ \frac{\zeta^4}{(1+\delb_rT)^2}H_{00} + O(|H|^2)
\\
=:& \zeta^2 + |H^{\Ncal 00}| + \rho_i(H)
\endaligned
\end{equation}
where we have applied \eqref{eq1-14-june-2025}, i.e., $H^{\Ncal00}<0$. Moreover, thanks to \eqref{eq2-14-june-2025}.
\begin{equation}\label{eq1-15-june-2025}
|\rho_i(H)|\lesssim \zeta^2|H| + |H|^2\lesssim \zeta^2\eps_s.
\end{equation}
Thus \eqref{eq12-14-june-2025} leads us to
\begin{equation}\label{eq9-14-june-2025}
\left(\begin{array}{cc}
1 &0
\\
-{P_i}_{11}^{-1}w_i &1
\end{array}\right)
A_i^{\mathrm{T}}\gb A_i
\left(\begin{array}{cc}
1 & \quad - {P_i}_{11}^{-1}w_i^{\mathrm{T}}
\\
0 &1
\end{array}\right)
=
\left(\begin{array}{cc}
{P_i}_{11} &0
\\
0 & \mathrm{I}_2 + D_i - {P_i}_{11}^{-1}w_iw_i^{\mathrm{T}} 
\end{array}\right).
\end{equation}
Here remark that \eqref{eq8-14-june-2025} leads us to $|D_i - {P_i}_{11}^{-1}w_iw_i^{\mathrm{T}}|\lesssim \eps_s$. Then when $\eps_s$ sufficiently small, the determinant of the right-hand side of \eqref{eq9-14-june-2025} is located in 
$$
\Big[\frac{\zeta^2}{2}+|H^{\Ncal00}|, \frac{3}{2}\zeta^2+|H^{\Ncal00}|\Big].
$$ 
Thus by the uniform boundness of $A_i$,
\begin{equation}
C_0(\zeta^2+|H^{\Ncal00}|\geq\det(\gb))\geq c_0\big(\zeta^{2} + H^{\Ncal 00}\big).
\end{equation}
on each $U_i\subset \mathcal{S}(s,r)$. Then by the compactness of $\mathcal{S}(s,r)$, the above estimate holds on the entire $\mathcal{S}(s,r)$.
This completes the Proposition~\ref{prop1-14-june-2025}.


\subsection{Proof of Proposition~\ref{prop1-27-june-2025}}\label{subsec1-29-sept-2025}

We first establish a technical result.

\begin{lemma}\label{lem1-15-june-2025}
Let $M$ be a real matrix of size $3\times3$, which can be written as
\begin{equation}
M=
\left(
\begin{array}{cc}
a & \zeta^{-1}w^{\mathrm{T}}
\\
\zeta v &D
\end{array}\right)
\end{equation}
where $a\in\RR$, $v,w\in\RR^2$ and $D\in \mathrm{M}_{2\times 2}(\RR)$. Suppose that
\begin{equation}
|a| + |v| + |w| + |D|\leq \eps_s
\end{equation}
with $\eps_s$ sufficiently small, and $0<\zeta$. Then the series 
\begin{equation}
T := \sum_{k=2}^{\infty}(-1)^kM^k
\end{equation}
converges. Furthermore, let 
$$
T = \left(
\begin{array}{cc}
\alpha & \zeta^{-1}\omega^{\mathrm{T}}
\\
\zeta \upsilon & \Delta
\end{array}\right)
$$
with $\alpha\in\RR$, $\upsilon,\omega\in\RR^2$ and $\Delta\in\mathrm{M}_{2\times 2}(\RR)$, then
\begin{equation}\label{eq3-15-june-2025}
|\alpha| + |\upsilon| + |\omega| + |\Delta|\lesssim \eps_s^2.
\end{equation}
\end{lemma}

\begin{proof}
We set
$$
M^k = 
\left(\begin{array}{cc}
a_k & \zeta^{-1}w_k^{\mathrm{T}}
\\
\zeta v_k &D_k
\end{array}\right).
$$
Then
\begin{equation}\label{eq2-15-june-2025}
\aligned
&a_{k+1} = aa_k + w^{\mathrm{T}}v_k,
\quad
&&w_{k+1}^{\mathrm{T}} = aw_k^{\mathrm{T}} + w^{\mathrm{T}}D_k.
\\
&v_{k+1} =  a_kv + Dv_k,
\quad
&&D_{k+1} = vw_k^{\mathrm{T}} + DD_k.
\endaligned
\end{equation}
We consider the convergence of $\{v_k\}, \{w_k\}$ in $L^{\infty}$ norm in $\RR^2$, and $\{D_k\}$ in the norm of $\mathrm{End}(\RR^2)$ with respect to the $L^{\infty}$ norm in $\RR^2$. We observe that
\begin{equation}\label{eq6-27-june-2025}
\max\{|a_k|, |v_k|_{\infty}, |w_k|_{\infty} , \|D_k\|_{\infty}\}\leq 2^{k-1}\eps_s^k.
\end{equation}
This can be checked by induction via \eqref{eq2-15-june-2025} and the property
$$
|v^{\mathrm{T}}w|_{\infty}, |vw^{\mathrm{T}}|_{\infty}\leq |v|_{\infty}|w|_{\infty}.
$$
This leads us to the convergence of $\sum_{k=1}^{\infty}|M|_{\infty}^k$ and $\sum_{k=1}^{\infty}(-1)^kM^k$. Furthermore, the bounds of the sub-matrices of $T$ can be estimated directly. For example for $\omega$,
$$
|\omega|_{\infty}\leq \sum_{k=2}^{\infty}|w_{k}|_{\infty}\leq \sum_{k=2}^{\infty}2^{k-1}\eps_s^k\lesssim C\eps_s^2.
$$
This concludes the desired estimate \eqref{eq3-15-june-2025}.

\end{proof}

\begin{proof}[Proof of Proposition~\ref{prop1-27-june-2025}]
We will make estimate on each $U_i$ and the apply the compactness. We define
$$
R_i = R_i(H) = 
\left(
\begin{array}{cc}
\rho_i(H) &w_i^{\mathrm{T}}
\\
w_i & D_i
\end{array}
\right),
\quad
F = 
\left(
\begin{array}{cc}
\zeta^{-2}+|H^{\Ncal00}| & 0
\\
0 & \mathrm{I}_2
\end{array}\right).
$$
Then \eqref{eq8-14-june-2025} (recall \eqref{eq2-29-sept-2025}) is written as
\begin{equation}\label{eq6-13-july-2025}
F + R_i = A_i^{\mathrm{T}}\gb A_i
\end{equation}
which leads us to
\begin{equation}\label{eq4-15-june-2025}
\big(A_i^{\mathrm{T}}\gb A_i\big)^{-1}
=
F^{-1} + \sum_{k=1}^{\infty}(-F^{-1}R_i)^kF^{-1},
\end{equation}
provided that the right-hand side converges. To see this, we remark that,
$$
F^{-1}R_i = 
\left(\begin{array}{cc}
\frac{\rho_i(H)}{\zeta^2+|H^{\Ncal00}|} & \frac{w_i^{\mathrm{T}}}{\zeta^2+|H^{\Ncal00}|}
\\
w_i & D_i
\end{array}\right).
$$
Then thanks to \eqref{eq1-15-june-2025} and \eqref{eq2-14-june-2025}
$$
\aligned
& \frac{\big|\rho_i(H)\big|}{\zeta^2+|H^{\Ncal00}|}\lesssim \zetab^{-2}(|H|^2+\zeta^2|H|)\lesssim \zetab^{-1}|H|\lesssim\eps_s,
\\ 
& \frac{\big|w_i^{\mathrm{T}}\big|}{\zeta^2+|H^{\Ncal00}|}\lesssim \zetab^{-2}|H|\lesssim \zetab^{-1}(\zetab^{-1}|H|)\lesssim \zetab^{-1}\eps_s,
\\
&|w_i| \lesssim |H|\leq \zetab(\zetab^{-1}|H|)\lesssim \zetab\eps_s, 
\quad 
&&|D_i|\lesssim |H|\leq \zetab^{-1}|H|\lesssim \eps_s.
\endaligned
$$
Then the desired convergence is guaranteed by Lemma~\ref{lem1-15-june-2025}. We remark particularly that
$$
F^{-1}R_iF^{-1} = 
\left(
\begin{array}{cc}
\frac{\rho_i(H)}{(\zeta^2 + |H^{\N00}|)^2} & \frac{w_i^{\mathrm{T}}}{\zeta^2+|H^{\Ncal00}|}
\\
\frac{w_i}{\zeta^2+|H^{\Ncal00}|} &D_i
\end{array}
\right).
$$
That is, 
\begin{equation}\label{eq8-27-june-2025}
\big|F^{-1}R_iF^{-1}\big|\lesssim (\zeta^2+|H^{\Ncal00})^{-2}(\zeta^2|H| + |H|^2) 
+ (\zeta^{2}+|H^{\Ncal00}|)^{-1}|H|\lesssim \zetab^{-2}|H|.
\end{equation}
Moreover, we have the following estimate on the matrix $ \sum_{k=2}^{\infty}(-F^{-1}R_i)^k$:
$$
\sum_{k=2}^{\infty}(-F^{-1}R_i)^k = :
\left(
\begin{array}{cc}
\alpha_i & \zetab^{-1}\omega_i^{\mathrm{T}}
\\
\zetab\upsilon_i & \Delta_i
\end{array}\right),
$$
then in Lemma~\ref{lem1-15-june-2025}, we fix $\eps_s = \zetab^{-1}|H|$ and obtain 
$$
|\alpha_i| + |\upsilon_i| + |\omega_i| + |\Delta_i|\lesssim \zetab^{-2}|H|^2.
$$
We thus obtain, provided that $\zeta^{-1}|H|$ sufficiently small (compared with $1$),
\begin{equation}\label{eq7-27-june-2025}
\Big|\sum_{k=1}^{\infty}(-F^{-1}R_i)^kF^{-1}\Big|\lesssim \zetab^{-2}|H|.
\end{equation}

Based on \eqref{eq4-15-june-2025}, we observe that
\begin{equation}
\sigmab  = Q - M_1Q + M_2Q
\end{equation}
where
$$
Q=:A_iF^{-1}A_i^{-1}, \qquad M_1 = A_iF^{-1}R_iF^{-1}A_i^{-1},
\quad
M_2 = A_i\sum_{k=2}^{\infty}(-F^{-1}R_i)^kF^{-1}A_i^{-1}.
$$
Then by the uniform bounds on $A_i$ and $A_i^{-1}$ together with \eqref{eq8-27-june-2025} and \eqref{eq7-27-june-2025},
\begin{equation}
|MQ| = \big|(-M_1+M_2)Q\big|\lesssim \zetab^{-2}|H|.
\end{equation}

Finally, we give the expression of $Q$. Recall \eqref{eq7-14-june-2025}, and remark that
$$
F^{-1} - E^{-1} = 
\left(\begin{array}{ccc}
\frac{-|H^{\Ncal00}|}{\zeta^2(\zeta^2+|H^{\Ncal00}|)} &0 &0
\\
0 &0 &0
\\
0 &0 &0
\end{array}\right).
$$
Then,
$$
\aligned
Q - \sigmab_{\eta} & = 
\left(\begin{array}{ccc}
x^1/r &* &*
\\
x^2/r &* &*
\\
x^3/r &* &*
\end{array}\right)
\left(\begin{array}{ccc}
\frac{-|H^{\Ncal00}|}{\zeta^2(\zeta^2+|H^{\Ncal00}|)} &0 &0
\\
0 &0 &0
\\
0 &0 &0
\end{array}\right)
\left(\begin{array}{ccc}
x^1/r &x^2/r &x^3/r
\\
* &* &*
\\
* &* &*
\end{array}\right)
\\
& = \frac{-|H^{\Ncal00}|}{\zeta^2(\zeta^2+|H^{\Ncal00}|)}\Big(\frac{x^ax^b}{r^2}\Big)_{ab}.
\endaligned
$$
\end{proof}


\subsection{An estimate on the eigenvalue of $\big(\gb_{ab}\big)$}

\begin{lemma}\label{lem1-13-july-2025}
Assume \eqref{eq-USA-condition} with sufficiently small $\eps_s$.
Then there exists two positive constants $K_0<K_1$ such that for any $(X^1,X^2,X^3)\in\RR^3$,
\begin{equation}
K_0(\zeta^2+|H^{\Ncal00}|)\max_a\{|X^a|^2\}\leq X^a\gb_{ab}X^b\leq K_1\max_a\{|X^a|^2\}.
\end{equation}
\end{lemma}

\begin{proof}
\bse
We treat first the region $\{3t/4\leq r\leq r^{\Ecal}(s)\}$. We recall \eqref{eq6-13-july-2025}. Remark that $A_i$ are orthogonal matrices. We thus only need to control the eigenvalues of $(F+R_i)$. We remark that, let
\be
K = \left(
\begin{array}{cc}
(\zeta^2+|H^{N00}|)^{-1/2} &0
\\
0 & \mathrm{I}_2
\end{array}
\right).
\ee
Then
\begin{equation}
K^{\mathrm{T}}(F+R_i)K
=\mathrm{I}_3 + J
\end{equation}
with
\be
J = \left(
\begin{array}{cc}
\frac{\rho_i(H)}{\zeta^2+|H^{\N00}|} & \frac{w_i^{\mathrm{T}}}{\big(\zeta^2+|H^{\N00}|\big)^{1/2}}
\\
\frac{w_i}{\big(\zeta^2+|H^{\N00}|\big)^{1/2}} &D_i
\end{array}
\right).
\ee
Then recall ing\eqref{eq1-15-june-2025} and \eqref{eq12-14-june-2025}, we have 
\be
|J|\lesssim \eps_s.
\ee
Thus for any $X = (X^1,X^2,X^3)\in\RR^3$, 
\be
(1-C\eps)\sum_{a}|X^a|^2\leq X^{\mathrm{T}}(I_3+J)X \leq (1+C\eps_s)\sum_a|X^a|^2.
\ee
Thus for any $Y = (Y^1,Y^2,Y^3)\in\RR^3$, we find 
\be
\aligned
&(1-C\eps)(|Y^1|(\zeta^2+|H^{\Ncal00}|)+|Y^2|^2+|Y^3|^2)
\\
\leq& Y^{\mathrm{T}}(F + R_i)Y 
\\
\leq& (1+C\eps_s)(|Y^1|(\zeta^2+|H^{\Ncal00}|)+|Y^2|^2+|Y^3|^2).
\endaligned
\ee
That is, provided that $\eps_s$ sufficiently small, the eigenvalue of $(F+R_i)$ is bounded in the interval 
\begin{equation}
\frac{1}{2}(\zeta^2 + |H^{\N00}|)\sum_a|Y^a|^2\leq Y^{\mathrm{T}}(F+R_i)Y\leq \frac{3}{2}\sum_{a}|Y^a|^2
\end{equation}
Here we have remarked that $|H^{\N00}|\lesssim \eps_s$ due to \eqref{eq2-14-june-2025} and \eqref{eq1-14-july-2025}.
Then by the orthogonality of each $A_i$, one obtains the desired result in $\{3t/4\leq r\leq r^{\Ecal}(s)\}$. 
\ese
\bse
When in the region out of $\Mcal{[s_0,s_1]}\setminus \{3t/4\leq r\leq r^{\Ecal}(s)\}$, we remark that $\zeta\geq \sqrt{7}/4$. Thus in this case we only need to prove that 
\begin{equation}\label{eq2-14-july-2025}
K_0\sum_a|Y^a|^2\leq Y^a\gb_{ab}Y^b\leq K_1\sum_{a}|Y^a|^2
\end{equation}
with $K_1>K_0>0$. When $r\geq r^{\Ecal}(s)$, $\delb_rT=0$. Thus $\gb_{ab} = g_{ab} = \delta_{ab} + H_{ab}$. Provided that $|H_{ab}|\lesssim \zeta\eps_s = \eps_s$ sufficiently small, \eqref{eq2-14-july-2025} is guaranteed for $\{r\geq r^{\Ecal}(s)\}$. When $r\leq 3t/4$, we remark that
\be
\gb_{ab} = - \frac{x^ax^b}{t^2} + \delta_{ab} + \Hb_{ab}.
\ee
Then, for any $Y\in\RR^3$, we have 
\be
\Big|\sum_{a,b}\frac{x^ax^b}{t^2}Y^aY^b\Big|\leq \Big(\sum_a|Y^a|^2\Big)^{1/2}\Big(\sum_a|x^a/t|^2\Big)^{1/2}
\leq \frac{3}{4}\Big(\sum_a|Y^a|^2\Big)^{1/2}.
\ee
Thus, provided that $|H|\lesssim \eps_s$ sufficiently small
\be
(1/8) \Big(\sum_a|Y^a|^2\Big)^{1/2}\leq Y^aY^b\gb_{ab}\leq 2\Big(\sum_a|Y^a|^2\Big)^{1/2}.
\ee
\ese
\end{proof}


\section{High-order derivatives acting on products}
\label{section=N23}

\subsection{Estimates on high-order Lie derivatives}

For $I,J$ two multi-indices. If there exists a multi-index $K$ such that $J\odot K=I$, we say $J$ a sub-index of $I$, and we denote by
\be
J\prec I.
\ee
It is clear that when $J\prec I$, one has $\ord(J)\leq\ord(I)$, $\rank(I)\leq \rank(J)$.
For convenience in the discussion, we introduce the notation
\be
|X|_{\vec{n},p,k} := \max_{\ord(I)\leq p\atop \rank(I)\leq k}|\Lcal_{\mathscr{Z}}^IX|_{\vec{n}},
\quad
|X|_{\vec{n},p} := |X|_{\vec{n},p,p}.
\ee
Then we establish the following estimate.

\begin{proposition}\label{prop1-14-july-2025}
Assume that \eqref{eq-USA-condition} holds for a sufficiently small $\eps_s$. Let $X$ be a vector field defined in a (subset) of $\Mcal_{[s_0,s_1]}$. Then for $0\leq k\leq p$,
\begin{equation}\label{eq10-14-july-2025}
|X|_{\vec{n},p,k}\lesssim_p \zetab^{-1}|X|_{p,k}.
\end{equation}
\end{proposition}

This result is based on the following decompositions.

\begin{lemma}\label{lem4-11-july-2025}
For a vector field $X$ and an admissible operator $\mathscr{Z}^I$, one has 
\begin{equation}\label{eq13-11-july-2025}
\Lcal_{\mathscr{Z}}^IX = \big(\mathscr{Z}^IX^{\alpha}\big)\del_{\alpha} 
+ \sum_{J\prec I\atop|J|<|I|}\Gamma_{J\beta}^{\alpha}(\mathscr{Z}^{J}X^{\beta})\del_{\alpha}
\end{equation}
where $\Gamma_{J\beta}^{\alpha}$ are constants. For a tensor field $T$ of type $(0,2)$, one has 
\begin{equation}\label{eq14-11-july-2025}
\Lcal_{\mathscr{Z}}^IT_{\alpha\beta} = \mathscr{Z}^I\big(T_{\alpha\beta}\big) 
+ \sum_{J\prec I\atop|J|<|I|}\Gamma_{\alpha\beta J}^{I \mu\nu} \mathscr{Z}^{J}(T_{\mu\nu})
\end{equation}
with $\Gamma_{\alpha\beta J}^{I \mu\nu}$ constants.
\end{lemma}

\begin{proof}
This can be checked by induction on $|I|$. In fact for zero order, we remark that
\begin{equation}\label{eq12-11-july-2025}
\aligned
\, [\del_{\alpha},X] & = \, \del_{\alpha}X^{\mu}\del_{\mu},
\\
[L_a,X] & = \, L_aX^{\alpha}\del_{\alpha} - X^{\alpha}\delta_\alpha^0\del_a 
- X^{\alpha}\delta_{\alpha}^a\del_t,
\\
[\Omega_{ab},X] & = \, \Omega_{ab}X^{\alpha}\del_{\alpha} - X^{\alpha}\delta_{\alpha}^a\del_b 
+ X^{\alpha}\delta_{\alpha}^b\del_a 
\endaligned
\end{equation}
which verifies \eqref{eq13-11-july-2025}. Then for $\mathscr{Z}^{I'} = Z\mathscr{Z}^I$, we have
$$
\aligned
\Lcal_\mathscr{Z}^{I'}X 
& = \, \Lcal_Z(\Lcal_{\mathscr{Z}}^IX) = \Lcal_Z\Big(\big(\mathscr{Z}^IX^{\alpha}\big)\del_{\alpha} 
+ \sum_{J\prec I\atop|J|<|I|}\Gamma_{J\beta}^{\alpha}(\mathscr{Z}^{J}X^{\beta})\del_{\alpha}\Big)
\\
& = \, \mathscr{Z}^{I'}X^{\alpha}\del_{\alpha} 
- \big(\mathscr{Z}^IX^{\beta}\big)\del_{\beta}Z^{\alpha}\del_{\alpha}
\\
& \quad +\sum_{J\prec I\atop|J|<|I|}\Gamma_{J\beta}^{\alpha}Z(\mathscr{Z}^{J}X^{\beta})\del_{\alpha}
-
\sum_{J\prec I\atop|J|<|I|}\Gamma_{J\beta}^{\gamma}(\mathscr{Z}^{J}X^{\gamma})\del_{\beta}Z^{\alpha}\del_{\alpha}.
\endaligned
$$
Then we only need to remark that $\del_{\alpha}Z^{\beta}$ are constants.

For \eqref{eq14-11-july-2025}, we recall that
$$
\Lcal_ZT_{\alpha\beta} = ZT_{\alpha\beta} + T_{\gamma\beta}\del_{\alpha}Z^{\gamma} 
+ T_{\alpha\gamma}\del_{\beta}Z^{\gamma}.
$$
Recall that $\del_{\alpha}Z^{\beta}$ are constants, the above expression verifies \eqref{eq14-11-july-2025}. Then we consider $\mathscr{Z}^{I'} = Z\mathscr{Z}$.
$$
\aligned
\Lcal_\mathscr{Z}^{I'}T_{\alpha\beta} 
& =  \Lcal_Z(\Lcal_{\mathscr{Z}}^IT)_{\alpha\beta} 
\\
& = Z(\Lcal_{\mathscr{Z}}^IT_{\alpha\beta}) 
+ \big(\Lcal_{\mathscr{Z}}^IT\big)_{\gamma\beta}\del_{\alpha}Z^{\gamma} 
+ \big(\Lcal_{\mathscr{Z}}^IT\big)_{\alpha\gamma}\del_{\beta}Z^{\gamma} .
\endaligned
$$
We thus conclude for the case $|I'| = |I|+1$, due to the fact that $\del_{\alpha}Z^{\beta}$ being constants.
\end{proof}

\begin{proof}[Proof of Proposition~\ref{prop1-14-july-2025}]
From \eqref{eq13-11-july-2025}, considering a multi-index $I$ with $\ord(I) = p, \rank(I) = k$, one has
\begin{equation}\label{eq2-16-july-2025}
\big|\Lcal_{\mathscr{Z}}^IX^{\alpha}\big|\lesssim_p |X|_{p,k}.
\end{equation}
Then  by Lemma~\ref{lem2-13-july-2025}, we obtain \eqref{eq10-14-july-2025}.
\end{proof}
\subsection{Decomposition on high-order adapted derivatives}
{
\begin{lemma}\label{lem1-15-july-2025}
If $X$ is a vector field and $\mathscr{Z}^I$ an admissible operator, then $\mathscr{Z}^IX$ can be expressed as
\begin{equation}\label{eq7-04-july-2025}
\mathscr{Z}^IX = \Lcal_{\mathscr{Z}}^IX 
+ \sum_{J\odot K\prec I\atop |K|\geq 1} \Lcal_\mathscr{Z}^JX^{\alpha}\Lcal_{\mathscr{Z}}^KH_{\beta\gamma}
P_{JK\alpha}^{\beta\gamma\mu}\del_{\mu}
\end{equation}
where $P_{JK\alpha}^{\beta\gamma\mu}$ are polynomials of constant coefficients acting on the variables
$g^{\mu\nu}$ and $\Lcal_\mathscr{Z}^LH_{\mu\nu}$  with $|L|+|K| \leq|I|-|J|$ and $1\leq|L|\leq|K|$.
\end{lemma}
\begin{proof} We use an induction on $|I|$. 
The case $|I|=1$ follows from \eqref{eq1-04-july-2025}:
$$
ZX = \Lcal_ZX + \frac{1}{2}g^{\mu\nu}X^{\alpha}\pi[Z]_{\alpha\nu}\del_{\mu} 
= \Lcal_ZX + \frac{1}{2}g^{\mu\nu}X^{\alpha}\Lcal_ZH_{\alpha\nu}\del_{\mu}
$$
where we have applied that $\Lcal_Z\eta_{\mu\nu} = 0$ for $Z\in\mathscr{Z}$. Suppose that $\mathscr{Z}^IX$ can be expressed in the form \eqref{eq7-04-july-2025}. We consider $\mathscr{Z}^{I'} = Z \mathscr{Z}^I$ and write 
$$
\aligned
Z \mathscr{Z}^IX & =  Z(\mathscr{Z}^IX) = \Lcal_{Z}\big(\Lcal_\mathscr{Z}^IX\big) 
+ \sum_{J\odot K\prec I\atop |K|\geq 1}
\Lcal_{Z}\big(\Lcal_\mathscr{Z}^JX^{\alpha}\Lcal_{\mathscr{Z}}^KH_{\beta\gamma}
P_{JK\alpha}^{\beta\gamma\mu}\del_{\mu}\big)
\\
& \quad + \frac{1}{2}g^{\mu\nu}\Lcal_\mathscr{Z}^IX^{\alpha}\pi[Z]_{\alpha\nu}\del_{\mu}
+ \frac{1}{2}g^{\mu\nu}\Big(\sum_{J\odot K\prec I\atop |K|\geq 1}
\Lcal_\mathscr{Z}^JX^{\delta}\Lcal_{\mathscr{Z}}^KH_{\beta\gamma}
P_{JK\delta}^{\beta\gamma\alpha}\Big)\pi[Z]_{\alpha\nu}\del_{\mu}, 
\endaligned
$$
therefore 
$$
\aligned
Z \mathscr{Z}^IX 
& = \Lcal_{\mathscr{Z}}^{I'}X 
+ \sum_{J\odot K\prec I\atop |K|\geq 1}
\Lcal_{Z}\Lcal_\mathscr{Z}^JX^{\alpha}\Lcal_{\mathscr{Z}}^KH_{\beta\gamma}
P_{JK\alpha}^{\beta\gamma\mu}\del_{\mu}
+ \sum_{J\odot K\prec I \atop |K|\geq 1}
\Lcal_\mathscr{Z}^JX^{\alpha}\Lcal_Z\Lcal_{\mathscr{Z}}^KH_{\beta\gamma}
P_{JK\alpha}^{\beta\gamma\mu}\del_{\mu}
\\
& \quad +  \sum_{J\odot K\prec I \atop |K|\geq 1}
\Lcal_\mathscr{Z}^JX^{\alpha}\Lcal_{\mathscr{Z}}^KH_{\beta\gamma}
\Lcal_ZP_{JK\alpha}^{\beta\gamma\mu}\del_{\mu}
\\
& \quad + \frac{1}{2}g^{\mu\nu}\Lcal_\mathscr{Z}^IX^{\alpha}\pi[Z]_{\alpha\nu}\del_{\mu}
+ \frac{1}{2}g^{\mu\nu}\sum_{J\odot K\prec I\atop |K|\geq 1}
\big(\Lcal_\mathscr{Z}^JX^{\delta}\Lcal_{\mathscr{Z}}^KH_{\beta\gamma}
P_{JK\delta}^{\beta\gamma\alpha}\big)
\pi[Z]_{\alpha\nu}\del_{\mu}.
\\
\endaligned
$$
All of the terms in the right-hand side, but the fourth one, are already in the correct form (recall that 
$\pi[Z] = \Lcal_{Z}g = \Lcal_ZH$). We only need to consider the third term and we observe that 
\begin{equation}\label{eq1-05-july-2025}
\Lcal_{Z}g^{\mu\nu} = - \pi[Z]^{\mu\nu} 
= -g^{\mu\mu'}g^{\nu\nu'}\Lcal_{Z}H_{\mu'\nu'}.
\end{equation}
We remark that $P_{JK\alpha}^{\beta\gamma\mu}$ are finite linear combinations of monomial composed by $g^{\mu\nu}$ and $\Lcal_{\mathscr{Z}}^{L}g_{\mu\nu}$ with $|J|+|K|+|L|\leq |I|$.
Then when we take the Lie derivative on each monomial. We rely on the Leibniz rule and apply \eqref{eq1-05-july-2025} when the Lie derivative acts on the factors $g^{\mu\nu}$. It may happen that a monomial of $\Lcal_{Z}(P_{JK\delta}^{\beta\gamma\alpha})$ contains a Lie derivative acting on $g_{\mu\nu}$, namely $\Lcal_{\mathscr{Z}}^Lg_{\mu\nu}$, with order larger than $|K|$. Then we exchange the role of $\Lcal_{\mathscr{Z}}^Lg_{\mu\nu}$ and $\Lcal_{\mathscr{Z}}^Kg_{\beta\gamma}$. Then we can guarantee that $|K|$ is not smaller than the order of any Lie derivative acting $g_{\mu\nu}$ contained in the corresponding monomial. finally, recall that for $|K|\geq 1$, $\Lcal_{\mathscr{Z}}^Kg = \Lcal_{\mathscr{Z}}^KH$. This completes the proof.
\end{proof} 
}

Then we establish the following estimate:
\begin{lemma}\label{lem2-23-july-2025}
Assume that \eqref{eq-USA-condition} holds for a sufficiently small $\eps_s$, and suppose that
\begin{equation}\label{eq3-15-july-2025}
|H|_{[p/2]}\leq 1.
\end{equation}
Let $X$ be a vector field and $\mathscr{Z}^I$ be an admissible operator of type $(p,k)$. Then
\begin{equation}\label{eq5'-16-july-2025}
|\mathscr{Z}^IX|_{\vec{n}}\lesssim \zetab^{-1}|X|_{p,k} 
+ \sum_{{p_1+p_2\leq p\atop k_1+k_2\leq k}\atop p_2\geq 1}\zetab^{-1}|X|_{p_1,k_1}|H|_{p_2,k_2}.
\end{equation}
\end{lemma}
\begin{proof}
We apply Proposition~\ref{lem1-15-july-2025}. First, by Lemma~\ref{prop1-14-july-2025}
\begin{equation}\label{eq2-15-july-2025}
\big|\Lcal_{\mathscr{Z}}^IX\big|_{\vec{n}}\lesssim \big|X\big|_{\vec{n},p,k}\lesssim \zetab^{-1}|X|_{p,k}.
\end{equation}
For the second term in the right-hand side of \eqref{eq7-04-july-2025}, we remark that it is a finite linear combination with constant coefficient of the monomials if the following form:
\begin{equation}\label{eq1-16-july-2025}
\Lcal_{\mathscr{Z}}^JX^{\alpha}\Lcal_{\mathscr{Z}}^{K}g_{\beta\gamma}
\Lcal_{\mathscr{Z}}^{K_2}g_{\beta_2\gamma_2}\cdots \Lcal_{\mathscr{Z}}^{K_n}g_{\beta_n\gamma_n}g^{\delta_1\lambda_1}g^{\delta_2\lambda_2}\cdots g^{\delta_m\lambda_m}\del_{\mu}
\end{equation}
with $|K|\geq |K_2|\geq\cdots \geq|K_n| \geq 1$. We observe that $|K_j|\leq [p/2]$. Then thanks to \eqref{eq3-15-july-2025}, 
$$
\aligned
&|\Lcal_{\mathscr{Z}}^JX^{\alpha}\Lcal_{\mathscr{Z}}^{K}g_{\beta\gamma}
\Lcal_{\mathscr{Z}}^{K_2}g_{\beta_2\gamma_2}\cdots \Lcal_{\mathscr{Z}}^{K_n}g_{\beta_n\gamma_n}g^{\delta_1\lambda_1}g^{\delta_2\lambda_2}\cdots g^{\delta_m\lambda_m}|
\\
& \lesssim \sum_{J\odot K\prec I\atop |K|\geq 1}
|\Lcal_{\mathscr{Z}}^JX^{\alpha}||\Lcal_{\mathscr{Z}}^KH_{\beta\gamma}|.
\endaligned
$$
Then by Lemma~\ref{lem2-13-july-2025} applied on \eqref{eq1-16-july-2025}, 
\begin{equation}\label{eq7-06-oct-2025(l)}
|\eqref{eq1-16-july-2025}|_{\vec{n}}
\lesssim \zetab^{-1}\sum_{J\odot K\prec I\atop |K|\geq 1}
|\Lcal_{\mathscr{Z}}^JX^{\alpha}||\Lcal_{\mathscr{Z}}^Kg_{\beta\gamma}|
\lesssim \zetab^{-1}\sum_{{p_1+p_2\leq p\atop k_1+k_2\leq k}\atop p_2\geq 1}|X|_{p_1,k_1}|H|_{p_2,k_2}
\end{equation}
where for the last inequality we applied \eqref{eq2-16-july-2025}.
\end{proof}
\subsection{Proof of Proposition~\ref{prop1-15-july-2025}}
The proof of of \eqref{eq1-17-july-2025} is direct. We recall that the first relation in \eqref{eq4-04-july-2025} leads us to the following high-order Leibniz rule:
\begin{equation}
\mathscr{Z}(uf) = \sum_{I_1\odot I_2 = I}\mathscr{Z}^{I_1}u\mathscr{Z}^{I_2}\Psi
\end{equation}
which leads us to \eqref{eq1-17-july-2025}.

The estimate \eqref{eq5-16-july-2025} is guaranteed by \eqref{eq5'-16-july-2025}. Then we recall \eqref{eq1-15-july-2025}. For $\ord(I) = p,\rank(I)=k$,
$$
\mathscr{Z}^I(X\cdot\Psi) = \sum_{I_1\odot I_2 = I}\mathscr{Z}^{I_1}X\cdot\mathscr{Z}^{I_2}\Psi.
$$
Then by Lemma~\ref{lem1-16-july-2025} and Lemma~\ref{lem2-23-july-2025},
$$
\aligned
|\mathscr{Z}^I(X\cdot\Psi)|_{\vec{n}}
\leq&
\sum_{I_1\odot I_2=I}
\big|\mathscr{Z}^{I_1}X\big|_{\vec{n}}\cdot\big|\mathscr{Z}^{I_2}\Psi\big|_{\vec{n}}
\lesssim_p\sum_{p_1+p_2\leq p\atop k_1+k_2\leq k}[X]_{p_1,k_1}[\Psi]_{p_2,k_2}
\\
& \lesssim_p \sum_{p_1+p_2\leq p\atop k_1+k_2\leq k}\zetab^{-1}[\Psi]_{p_1,k_1}|X|_{p_2,k_2}
+\sum_{p_1+p_2+p_3\leq p\atop k_1+k_2+k_3\leq k,p_3\geq 1}\zetab^{-1}[\Psi]_{p_1,k_1}|X|_{p_2,k_2}|H|_{p_3,k_3}.
\endaligned
$$
This leads us to \eqref{eq3-16-july-2025}.

For \eqref{eq3-23-july-2025}, we first remark that, thanks to Lemma~\ref{lem1-15-july-2025},
$$
\mathscr{Z}^I\del_{\alpha} = \Lcal_{\mathscr{Z}}^I\del_{\alpha} 
+ \sum_{J\odot K\prec I\atop |K|\geq 1} (\Lcal_\mathscr{Z}^J\del_{\alpha})^{\delta}\Lcal_{\mathscr{Z}}^KH_{\beta\gamma}
P_{JK\delta}^{\beta\gamma\mu}\del_{\mu}
$$
where $P_{JK\delta}^{\beta\gamma\mu}$ are polynomials of constant coefficients acting on the variables
$g^{\mu\nu}$ and $\Lcal_\mathscr{Z}^LH_{\mu\nu}$  with $|L|+|K| \leq|I|-|J|$ and $1\leq|L|\leq|K|$. The importance is that, 
$$
\big(\Lcal_{\mathscr{Z}}^I\del_{\alpha}\big)^\delta 
$$
are constants, which can be checked by induction. Thus for $I$ being typo $(p,k)$,
\begin{equation}\label{eq2-18-july-2025}
\big|\big(\mathscr{Z}^I\del_{\alpha}\big)^{\delta}\big|\lesssim_p 1 + |H|_{p,k}.
\end{equation}

We now turn our attention to \eqref{eq3-23-july-2025}. Let $\ord(I) = p,\rank(I) = k$. We have 
\be
\aligned
\mathscr{Z}^I\Big(\prod_{\jmath=1}^n\del_{\alpha_{\jmath}}\cdot\Psi\Big)
& =  \sum_{I_1\odot\cdots\odot I_n\odot J=I}
\prod_{\jmath=1}^n\mathscr{Z}^{I_{\jmath}}\del_{\alpha_{\jmath}}\cdot\mathscr{Z}^J\Psi
\\
& =  \sum_{I_1\odot\cdots\odot I_n\odot J=I}
\prod_{\jmath=1}^n\del_{\beta_{\jmath}}\cdot \prod_{\jmath=1}^n(\mathscr{Z}^{I_{\jmath}}\del_{\alpha_{\jmath}})^{\beta_{\jmath}}\mathscr{Z}^J\Psi.
\endaligned
\ee
Then in view of \eqref{eq17'-18-july-2025} we find 
\be
\aligned
\Big|\prod_{\jmath=1}^n\del_{\beta_{\jmath}}\cdot \prod_{\jmath=1}^n(\mathscr{Z}^{I_{\jmath}}\del_{\alpha_{\jmath}})^{\beta_{\jmath}}\mathscr{Z}^J\Psi\Big|_{\vec{n}}
\lesssim_n& \zetab^{-1}\Big|\prod_{\jmath=1}^n(\mathscr{Z}^{I_{\jmath}}\del_{\alpha_{\jmath}})^{\beta_{\jmath}}\mathscr{Z}^J\Psi\Big|_{\vec{n}}
\\
\lesssim_{n,p}& \zetab^{-1}\prod_{\jmath=1}^n(1+|H|_{p'_{\jmath},k'_{\jmath}})|\mathscr{Z}^J\Psi|_{\vec{n}}, 
\endaligned
\ee
where
$
p'_\jmath = \ord(I_\jmath),\quad k'_\jmath = \rank(I_\jmath).
$
We take $p_1=\max\{p'_\jmath\}$. Then by \eqref{eq-vect-condition},
\be
\Big|\prod_{\jmath=1}^n\del_{\beta_\jmath}\cdot \prod_{\jmath=1}^n(\mathscr{Z}^{I_\jmath}\alpha_\jmath)^{\beta_\jmath}\mathscr{Z}^J\Psi\Big|_{\vec{n}}
\lesssim_{n,p} \zetab^{-1}(1+|H|_{p_1,k_1})|\mathscr{Z}^J\Psi|_{\vec{n}}
\lesssim_{n,p} \zetab^{-1}(1+|H|_{p_1,k_1})[\Psi]_{p_2,k_2}
\ee
with $p_1+p_2\leq p,k_1+k_2\leq k$. This leads us to \eqref{eq3-23-july-2025}. 


\section{Further observations on Sobolev inequalities}
\label{section=N24}

\subsection{Sobolev inequalities for scalar-valued functions}

Let $\RR^3_+ = \{(x^1,x^2,x^3)| x^a\geq 0\}$. For $x\in\RR^3_+$, we denote by
\be
\Ccal(x,\rho) := \{y\in\RR^3_+| x^a\leq y\leq y^a+\rho\}, \qquad \Ccal(x):= \{y\in\RR^3_+| x^a\leq y\leq y^a+1\}.
\ee
We then recall two Sobolev inequalities stated earlier in~LeFloch-Ma~\cite[Appendix B]{PLF-YM-PDE}.

\begin{lemma}\label{lem1-11-mai-2025}
Let $u: \RR^3_+\mapsto\RR$ be a sufficiently regular function. Then one has 
\begin{equation}\label{eq1-11-mai-2025}
|u(x)|\lesssim \sum_{|I|\leq 1}\|\del^I u\|_{L^6(\Ccal(x))},
\end{equation}
\begin{equation}\label{eq2-11-mai-2025}
\|u\|_{L^2(\RR^3_+)}\lesssim \sum_{\alpha}\|\del_{\alpha} u\|_{L^2(\RR^3_+)}
\end{equation}
\end{lemma}
A direct result of these two inequalities is the following one (cf \cite[Lemma 4.1]{PLF-YM-PDE})
\begin{equation}\label{eq3-11-mai-2025}
|u(x)|\lesssim\sum_{|I|\leq 2}\|\del^Iu\|_{L^2(\Ccal(x))}.
\end{equation}


\subsection{Klainerman-Sobolev type inequalities in the scalar case}

We recall the Klainerman-Sobolev inequalities adapted for Euclidean-hyperboloidal foliation, which are established in \cite{PLF-YM-PDE}. But here for technical reason, we establish a more precise version. When we consider the parameterization of $\Mcal_s^{\Hcal}$ by $x = (x^1,x^2,x^3)$:
\be
x \mapsto (T(s,r),x), \qquad r^2 = |x|^2 = \sum_{a}|x^a|^2,
\ee
the domain of the parameter $x$ is
$$
\mathcal{D}^{\Hcal}_s = \big\{r\leq r^{\Hcal}(s) = (s^2-1)/2\big\}.
$$
\begin{lemma}[K-S inequalities in hyperboloidal domain]\label{lem2-15-mai-2025}
Let $u$ be a sufficiently regular function defined in $\Mcal^{\Hcal}_{[s_0,s_1]}$ (not necessarily vanishing near the conical boundary $\del\Kcal = \{r=t-1\}$). Then for $(\hat{t},\hat{r})$ satisfying $0\leq \hat{r}\leq \hat{t}-1$ and $s_0\leq \hat{s} = \sqrt{\hat{t}^2- \hat{r}^2}\leq s_1$,
\begin{equation}\label{eq7-15-mai-2025}
\hat{t}^{1/2}|u(\hat{t},\hat{x})|
\lesssim \sum_{|I|\leq 1}\Big(\int_{\Ccal_{\hat{t},\hat{x}}}|L^Iu|^6\diff x\Big)^{1/6},
\end{equation}
\begin{equation}\label{eq8-15-mai-2025}
\hat{t}\Big(\int_{\Ccal_{\hat{t},\hat{x}}}|u|^6\diff x\Big)^{1/6}
\lesssim
\sum_{|I|\leq 1}\Big(\int_{\Mcal^{\Hcal}_{\hat{s}}}|L^Iu|^2\diff x\Big)^{1/2}.
\end{equation}
where
$$
\hat{x} := -(\sqrt{3}\,\hat{r}/3)(1,1,1),
$$
and
$$
\Ccal(\hat{t},\hat{x}) 
= \big\{(t,x)\in \Mcal^{\Hcal}_{\hat{s}}, x\in\Ccal(\hat{x},(\sqrt{3}/12)\hat{t})\big\}.
$$
\end{lemma}
\begin{proof}
We define
\begin{equation}\label{eq7-11-mai-2025}
v(s,x):=u\big(\sqrt{t^2+r^2},x\big)
\end{equation}
and for $(\hat{t},\hat{x})\in \Mcal_{\hat{s}}^{\Hcal}$ with $\hat{x} := - \frac{\sqrt{3}\,\hat{r}}{3}(1,1,1)$,
\begin{equation}
v_{\hat{t},\hat{x}}(y) : = v(\hat{s}, \hat{x} + \lambda y), \qquad \hat{s} = \sqrt{\hat{t}^2- \hat{r}^2}
\end{equation}
where $\lambda\leq r^{\Hcal}(\hat{s})$ is a constant to be determined. Then
\begin{equation}\label{eq9-10-mai-2025}
\frac{\del}{\del y^a} v_{\hat{t},\hat{x}}(y) = \lambda \frac{\del}{\del x^a}v(\hat{s},\hat{x} + \lambda y) 
= \frac{\lambda}{t}L_au(t,x)
\end{equation}
with 
$$
x = \hat{x} + \lambda y, \qquad t = \sqrt{\hat{s}^2+|x|^2}.
$$

For \eqref{eq7-15-mai-2025}, we apply the Sobolev inequality \eqref{eq1-11-mai-2025} in 
$$
\mathcal{C} = \{y|0\leq y^a\leq 1\}
$$
and obtain
\begin{equation}\label{eq6-10-mai-2025}
\aligned
|v_{\hat{t},\hat{x}}(0)|^6& \lesssim \sum_{|I|\leq 1}\int_{\mathcal{C}}  |\del_{y}^Iv_{\hat{t},\hat{x}}(y)|^6 \diff y
\\
& =  \lambda^{-3}\int_{\Mcal_{\hat{s}}\cap \Ccal(\hat{x},\lambda)}
\Big(|u(t,x)|^6 + \Big(\frac{\lambda}{t}\Big)^6\sum_{a}|L_au(t,x)|^6\Big)\diff x
\endaligned
\end{equation}
where 
$$
\Ccal_{\hat{x},\lambda} = \Big\{x\big|- \frac{\sqrt{3}}{3}\hat{r}\leq x^a\leq - \frac{\sqrt{3}}{3}\hat{r}+\lambda\Big\}\subset \mathcal{D}^{\Hcal}_{\hat{s}}
$$

When $\hat{t}-1\geq \hat{r}\geq \frac{3}{4}\hat{t}$, we take $\lambda = \frac{\sqrt{3}}{12}\hat{t}\leq \hat{r}\leq r^{\Hcal}(\hat{s})$. Then in $\Mcal_s\cap\{|x- \hat{x}|\leq \hat{t}/4\}$, $|x|\geq \hat{t}/2$ and thus
$$
t^2 = x^2 + \hat{s}^2 \geq \frac{1}{4}\hat{t}^2+\hat{s}^2\geq \frac{1}{4}\hat{t}^2.
$$
Then \eqref{eq6-10-mai-2025} leads us to
\begin{equation}\label{eq8-10-mai-2025}
|v_{\hat{t},\hat{x}}(0)|\lesssim \hat{t}^{-1/2}\sum_{|I|\leq 1}\Big(\int_{\Mcal_{\hat{s}}\cap \Ccal(\hat{x},(\sqrt{3}/12)\hat{t})}|L^I u(t,x)|^6\diff x\Big)^{1/6}.
\end{equation}

When $\hat{r}\leq \frac{3}{4}\hat{t}$, we take $\lambda = \frac{\hat{s}}{4}$. It is clear that $\lambda \leq r^{\Hcal}(\hat{s})$ provided that $\hat{s}\geq s_0\geq  2$. In this case 
$$
|\hat{s}/t|\leq 1, \qquad \sqrt{7}/4 \leq |\hat{s}/\hat{t}|\leq 1.
$$
Thus \eqref{eq6-10-mai-2025} also leads us to \eqref{eq8-10-mai-2025}. Then we obtain \eqref{eq7-15-mai-2025}.


We now turn our attention to \eqref{eq8-15-mai-2025}. We consider a smooth cut-off function defined on $\RR$ such that $\chi$ is non-increasing, 
$$
\chi(\rho) = 
\begin{cases}
1, \qquad & \rho\leq \frac{\sqrt{3}}{12},
\\
0, \qquad & \rho\geq \frac{\sqrt{3}}{6}.
\end{cases}
$$
and 
$$
\hat{\chi}(x):= \prod_{a=1}^3\chi\big(\hat{t}^{-1}(x^a- \hat{x}^a)\big). 
$$
The function $\hat{\chi}$ is smooth, $\chi(x) = 1$ for $x\in \Ccal(\hat{x}, (\sqrt{3}/12)\hat{t})$ and $\chi(x) = 0$ for $x\in \Ccal(\hat{x},\infty)\setminus \Ccal(\hat{x},(\sqrt{3}/6)\hat{t})$. We denote by
\begin{equation}\label{eq4-11-mai-2025}
w_{\hat{t},\hat{x}}(y) = \hat{\chi}(\hat{x}+\lambda y)v(\hat{s}, \hat{x}+\lambda y)
\end{equation}
One has
\begin{equation}
w_{\hat{t},\hat{x}}(y) 
= 
\begin{cases}
v_{\hat{t},\hat{x}}(y), \qquad &0\leq y^a\leq \frac{\sqrt{3}}{12}(\hat{t}/\lambda), \qquad \forall a = 1,2,3;
\\
0 , \qquad & \frac{\sqrt{3}}{6}(\hat{t}/\lambda)\leq y^a, \qquad \exists a\in \{1,2,3\}.
\end{cases}
\end{equation}
Thus
$$
\int_{\Ccal(0,(\sqrt{3}/12)(\hat{t}/\lambda))}|v_{\hat{t},\hat{x}}(y)|^6\diff y
\leq\int_{\RR^3_+}|w_{\hat{t},\hat{x}}(y)|^6\diff y.
$$
Then by Sobolev inequality \eqref{eq2-11-mai-2025}, one has
$$
\Big(\int_{\Ccal(0,(\sqrt{3}/12)(\hat{t}/\lambda))}|v_{\hat{t},\hat{x}}(y)|^6\diff y\Big)^{1/6}
\lesssim \Big(\int_{\Ccal(0,(\sqrt{3}/6)(\hat{t}/\lambda))}\Big|\frac{\del}{\del y^a}w_{\hat{t},\hat{x}}(y)\Big|^2\diff y\Big)^{1/2}
$$
By \eqref{eq9-10-mai-2025} and \eqref{eq4-11-mai-2025}, one has
$$
\aligned
& \Big(\int_{\Ccal(0,(\sqrt{3}/12)(\hat{t}/\lambda))}|v_{\hat{t},\hat{x}}(y)|^6\diff y\Big)^{1/6}
\\
& \lesssim 
\Big(\int_{\Ccal(0,(\sqrt{3}/12)(\hat{t}/\lambda))}\big| u(t,x)\big|^2\diff y\Big)^{1/2}
+\sum_{a=1}^3\Big(\int_{\Ccal(0,(\sqrt{3}/12)(\hat{t}/\lambda))}\Big|\frac{\lambda}{t}L_au(t,x)\Big|\diff y\Big)^{1/2}
\\
& = \lambda^{-3/2}\Big(\int_{\Mcal_{\hat{s}}\cap \Ccal(\hat{x},(\sqrt{3}/6)\hat{t})}|u(t,x)|^2\diff x\Big)^{1/2}
+\lambda^{-3/2}\sum_{a=1}^3\Big(\int_{\Mcal_{\hat{s}}\cap \Ccal(\hat{x},(\sqrt{3}/6)\hat{t})}\Big|\frac{\lambda}{t}L_au(t,x)\Big|\diff x\Big)^{1/2}.
\endaligned
$$
On the other hand, recall that
$$
\Big(\int_{\Ccal(0,(\sqrt{3}/12)(\hat{t}/\lambda))}|v_{\hat{t},\hat{x}}(y)|^6\diff y\Big)^{1/6}
=
\lambda^{-1/2}\Big(\int_{\Mcal_{\hat{s}}\cap \Ccal(\hat{x}, (\sqrt{3}/12)\hat{t})}|u(t,x)|^2\diff x\Big)^{1/6},
$$
we obtain
\begin{equation}\label{eq5-11-mai-2025}
\aligned
& \lambda\Big(\int_{\Mcal_{\hat{s}}\cap \Ccal(\hat{x}, (\sqrt{3}/12)\hat{t})}|u(t,x)|^2\diff x\Big)^{1/6}
\\
& \lesssim 
\Big(\int_{\Mcal_{\hat{s}}\cap \Ccal(\hat{x},(\sqrt{3}/6)\hat{t})}|u(t,x)|^2\diff x\Big)^{1/2}
+\sum_{a=1}^3\Big(\int_{\Mcal_{\hat{s}}\cap \Ccal(\hat{x},(\sqrt{3}/6)\hat{t})}\Big|\frac{\lambda}{t}L_au(t,x)\Big|\diff x\Big)^{1/2}.
\endaligned
\end{equation}
Here we remark that
$$
\Ccal\big(\hat{x},(\sqrt{3}/6)\hat{t}\big)\subset \mathcal{D}_{\hat{s}}^{\Hcal}.
$$
When $r^{\Hcal}(\hat{s})\geq \hat{r}\geq\frac{3}{4}\hat{t}$, we fix $\lambda  = \hat{t}$. We remark that on $\Mcal_{\hat{s}}\cap \Ccal(\hat{x},(\sqrt{3}/6)\hat{t})$,
$$
t^2 = \hat{s}^2 + x^2 \geq \hat{s}^2 + |\hat{r} - \hat{t}/2|^2 \geq \hat{s}^2 + \hat{t}^2/16 \geq (\hat{t}/4)^2.
$$
Thus \eqref{eq5-11-mai-2025} leads us to
\begin{equation}\label{eq6-11-mai-2025}
\Big(\int_{\Mcal_{\hat{s}}\cap \Ccal(\hat{x},(\sqrt{3}/12)\hat{t})}|u(t,x)|^6\diff x\Big)^{1/6}
\lesssim \hat{t}^{-1}\|u\|_{L^2_f(\Mcal^{\Hcal}_{\hat{s}})} 
+ \hat{t}\sum_{a}\|L_a u\|_{L^2_f(\Mcal^{\Hcal}_{\hat{s}})}.
\end{equation}
When $\hat{r}\leq \frac{3}{4}\hat{t}$, we take $\lambda = \frac{\hat{s}}{4}$. In this case 
$$
|\hat{s}/t|\leq 1, \qquad \sqrt{7}/4\leq |\hat{s}/\hat{t}|\leq 1.
$$
Thus one still concludes by \eqref{eq6-11-mai-2025}. Returning to the definition of $v$, i.e., \eqref{eq7-11-mai-2025}, we obtain \eqref{eq8-15-mai-2025}.
\end{proof}

\begin{lemma}[K-S inequalities in Merging-Euclidean domain]\label{lem2-16-mai-2025}
Let $u$ be a sufficiently regular function defined in $\Mcal^{\ME}_{[s_0,s_1]}$ (not necessarily vanishing near the conical boundary $\del\Kcal = \{r=t-1\}$). Then for $(\hat{t},\hat{r})$ satisfying $0\leq \hat{t}-1\leq \hat{r}$ and $s_0\leq \hat{s}\leq s_1$ with
$
\hat{t} = T(\hat{s},\hat{r}),
$
\begin{equation}\label{eq9-15-mai-2025}
\hat{r}^{1/3}|u(\hat{t},\hat{r})|\lesssim 
\sum_{j+|K|\leq 1}\Big(\int_{\Ccal_{\hat{t},\hat{x}}}|\delb_r^j\Omega^Ku|^6\diff x\Big)^{1/6},
\end{equation}
\begin{equation}\label{eq10-15-mai-2025}
\hat{r}^{2/3}\Big(\int_{\Ccal_{\hat{t},\hat{x}}}|u|^6\diff x\Big)^{1/6}
\lesssim
\sum_{j+|K|\leq 1}\Big(\int_{\Mcal^{\EM}_s}|\delb_r^j\Omega^Ku|^2\diff x\Big)^{1/2},
\end{equation}
where
$$
\hat{x} := (\hat{r},0,0)
$$
and
$$
\Ccal(\hat{t},\hat{x}) = \big\{(t,x)\in\Mcal^{\ME}_{\hat{s}},x = x(y),y\in[0,\pi/6]^3\big\}
$$
with
$$
\aligned
x^1 = (\hat{r} + y^1)\cos y^2\cos y^3,
\quad 
x^2 = (\hat{r} + y^1)\cos y^2\sin y^3,
\quad 
x^3 = (\hat{r} + y^1)\sin y^2.
\endaligned
$$
\end{lemma}
\begin{proof}
We define
\begin{equation}
v(s,x) := u\big(T(s,r),x\big),
\end{equation}
and for $(\hat{t},\hat{x})\in \Mcal^{\ME}_{\hat{s}}$ with $\hat{x} := (\hat{r},0,0)$ and $\hat{t} = T(\hat{s},\hat{r})$, we define
\begin{equation}
v_{\hat{t},\hat{x}}(\varrho,\vartheta,\varphi) := v(\hat{s},x),
\end{equation}
with $x = x(\varrho,\vartheta,\varphi)$ defined by
\begin{equation}\label{eq2-12-mai-2025}
x^1 = (\hat{r} + \varrho)\cos\vartheta\cos\varphi,
\quad 
x^2 = (\hat{r} + \varrho)\cos\vartheta\sin\varphi,
\quad 
x^3 = (\hat{r} + \varrho)\sin\vartheta
\end{equation}
where we only consider the domain of parameterization
$$
(y^1,y^2,y^3) = (\varrho,\vartheta,\varphi)\in[0,\pi/3]^3. 
$$
Here for the simplicity of expression, we denote by $y^1 = \varrho,y^2=\vartheta,y^3 = \varphi$. We remark particularly that $x(0,0,0) = \hat{x} = (\hat{r},0,0)$. Then
\begin{equation}\label{eq1-12-mai-2025}
\aligned
\frac{\del}{\del y^1}=\frac{\del}{\del\varrho}v_{\hat{t},\hat{x}}(\varrho,\vartheta,\varphi) & =  \,\delb_ru(t,x),
\\
\frac{\del}{\del y^2} = \frac{\del}{\del \vartheta}v_{\hat{t},\hat{x}}(\varrho,\vartheta,\varphi) & = \, \big(\cos\varphi \Omega_{13} + \sin\varphi\Omega_{23}\big)u(t,x),
\\
\frac{\del}{\del y^3} = \frac{\del}{\del\varphi}v_{\hat{t},\hat{x}}(\varrho,\vartheta,\varphi) & = \, \Omega_{12}u(t,x)
\endaligned
\end{equation}
with
$$
t = T(\hat{s},x).
$$

We apply the Sobolev inequality \eqref{eq1-11-mai-2025} on $v_{\hat{t},\hat{x}}$ in the set 
$$
\Ccal :=[0,\pi/6]^3.
$$
Considering \eqref{eq1-12-mai-2025},
\begin{equation}
|v_{\hat{t},\hat{x}}(0)|^6\lesssim \sum_{|I|\leq 1}\int_{\Ccal}|\del_y^I u|^6\diff y
= \sum_{j+|K|\leq 1}\int_{\Ccal}|\delb_r^j\Omega^K u|^6\diff r\diff\vartheta\diff\varphi
\end{equation}
Recall that
$$
\diff x = \varrho^2\cos\vartheta\diff r \diff\vartheta\diff\varphi
$$
and in $\mathcal{D}$ one has
\begin{equation}
\cos\vartheta\in[\sqrt{3}/2,1], \qquad \hat{r}\leq \varrho\leq \hat{r} + \frac{\pi}{6}.
\end{equation}
We thus obtain
\begin{equation}\label{eq3-12-mai-2025}
|u(\hat{t},\hat{x})|\lesssim \hat{r}^{-1/3}\sum_{j+|K|\leq 1}
\Big(\int_{\Ccal_{\hat{t},\hat{x}}}|\delb_r^j\Omega^Ku|^6 \diff x\Big)^{1/6}
\end{equation}
with
$$
\Ccal_{\hat{t},\hat{x}} = \big\{(t,x)\in\Mcal^{\ME}_{\hat{s}}| x=x(y)=x(\varrho,\vartheta,\varphi), (\varrho,\vartheta,\varphi)\in\Ccal\big\}
$$
where $x = (\varrho,\vartheta,\varphi)$ is defined by \eqref{eq2-12-mai-2025}. Then \eqref{eq3-12-mai-2025} concludes \eqref{eq9-15-mai-2025}.

For \eqref{eq10-15-mai-2025}, we first introduce a smooth, non-increasing cut-off function $\chi$ defined on $\RR$ such that 
$$
\chi(\rho) = 
\begin{cases}
1, \qquad & \rho\leq \pi/6
\\
0, \qquad & \rho\geq \pi/3
\end{cases}
$$
and define, recall that $(y^1,y^2,y^3) = (\varrho,\vartheta,\varphi)$,
$$
\hat{\chi}(\varrho,\vartheta,\varphi)
:= \prod_{a=1}^3\chi(y^a) 
=  \chi(\varrho)\chi(\vartheta)\chi(\varphi).
$$
Then let 
$$
w_{\hat{t},\hat{x}}(y) = w_{\hat{t},\hat{x}}(\varrho,\vartheta,\varphi) 
:= \hat{\chi}(\varrho,\vartheta,\varphi) v_{\hat{t},\hat{x}}(\varrho,\vartheta,\varphi).
$$
Then we apply the Sobolev inequality \eqref{eq2-11-mai-2025} on $w_{\hat{t},\hat{x}}$ in $\RR^3_+$:
\begin{equation}
\Big(\int_{\RR^3_+}|w_{\hat{t},\hat{x}}|^6\diff \varrho\diff \vartheta \diff\varphi\Big)^{1/6}
\lesssim \sum_{|I|\leq 1}
\Big(\int_{\RR^3_+}|\del_{y}^I w_{\hat{t},\hat{x}}|^2\diff \varrho\diff \vartheta \diff\varphi\Big)^{1/2}
\end{equation}
which leads us to
\begin{equation}\label{eq2-15-mai-2025}
\Big(\int_{\Ccal}|u(t,x)|^6\diff \varrho\diff\vartheta\diff\varphi\Big)^{1/6}
\lesssim
\sum_{|I|\leq 1}\Big(\int_{\Ccal'}
\del_y^Iu(t,x)\diff \varrho\diff\vartheta\diff\varphi\Big)^{1/2}
\end{equation}
where
$$
\Ccal' :=[0,\pi/3]^3.
$$
Remark that when $(\varrho,\vartheta,\varphi)\in \Ccal'$, 
$$
\cos\vartheta\in[1/2,1], \qquad \hat{r}\leq r\leq \hat{r} + \frac{\pi}{3}.
$$
Thus 
$$
\int_{\Ccal}|f|\diff\varrho\diff\vartheta\diff\varphi 
= \int_{\Ccal_{\hat{t},\hat{x}}}|f|\big(r^2\cos\vartheta\big)^{-1}\diff x
\simeq \hat{r}^{-2}\int_{\Ccal_{\hat{t},\hat{x}}}|f|\diff x.
$$
Then \eqref{eq2-15-mai-2025} leads us to
\begin{equation}
\aligned
\hat{r}^{-1/3}\Big(\int_{\Ccal_{\hat{t},\hat{x}}}|u(t,x)|^6\diff x\Big)^{1/6}
\simeq& \Big(\int_{\Ccal_{\hat{t},\hat{x}}}|u(t,x)|^6\diff \varrho\diff\vartheta\diff\varphi\Big)^{1/6}
\\
& \lesssim 
\sum_{|I|\leq 1}\Big(\int_{\Ccal'}
\del_y^Iu(t,x)\diff \varrho\diff\vartheta\diff\varphi\Big)^{1/2}
\\
& \lesssim 
\hat{r}^{-1}
\sum_{j+|K|\leq 1}\Big(\int_{ \Mcal^{\ME}_{\hat{s}}}
|\delb_r^j\Omega^K u|^2\diff x\Big)^{1/2}.
\endaligned
\end{equation}
which concludes \eqref{eq10-15-mai-2025}.
\end{proof}

\section{The null structure in the hyperboloidal region}
\subsection{Some basic facts about semi-hyperboloidal frame}
We recall the relation between the adapted frame $\{\delb_{\alpha}\}$ and the semi-hyperboloidal frame $\{\delu_{\alpha}\}$ in $\Mcal^{\Hcal}_{[s_0,s_1]}$, which leads to \eqref{eq4-13-june-2025}. We recall the following two important properties of the weight $(s/t)$ in $\Mcal^{\Hcal}_{[s_0,s_1]}$ which we have established in our previous work (one can also check directly by induction). For any admissible operator $\mathscr{Z}^I$ of type $(p,k)$,
\begin{equation}\label{eq3-23-aout-2025}
\big|\mathscr{Z}^I(s/t)\big|\lesssim_p (s/t)(t/s^2)^{p-k}\lesssim (s/t),\quad \text{in}\quad\Mcal^{\Hcal}_{[s_0,s_1]}.
\end{equation}

Furthermore, the following estimates holds in $\Mcal^{\Hcal}_{[s_0,s_1]}$, provided that \eqref{eq-USA-condition}:
\begin{equation}\label{eq2-23-aout-2025}
\aligned
&|\delb_0\Hb_{00}|\lesssim (s/t)^2s^{-1}|H| + (s/t)^3|\del H|,
\quad
&&|\delb_c\Hb_{00}|\lesssim (s/t)^3s^{-1}|H| + (s/t)^2|\delu H|,
\\
&|\delb_0\Hb_{b0}|\lesssim (s/t)s^{-1}|H| + (s/t)^2|\del H|,
\quad
&&|\delb_c\Hb_{b0}|\lesssim (s/t)^2s^{-1}|H| + (s/t)|\delu H|,
\\
&|\delb_0\Hb_{ab}|\lesssim (s/t)^2s^{-1}|H| + (s/t)|\del H|,
\quad
&&|\delb_c\Hb_{ab}|\lesssim (s/t)s^{-1}|H| + |\delu H|.
\endaligned
\end{equation}
On the other hand, thanks to \eqref{eq10-03-aout-2025}, in $\Mcal^{\Hcal}_{[s_0,s_1]}$, we can obtain the following estimates through an explicit calculation:
\begin{equation}\label{eq1-23-aout-2025}
\aligned
&|\delb_a\bar{\eta}_{bc}|\lesssim (s/t)s^{-1},
\quad
&&|\delb_a\bar{\eta}_{0b}|\lesssim (s/t)^2s^{-1},
\quad
&&&|\delb_a\bar{\eta}_{00}|\lesssim (s/t)^3s^{-1};
\\
&|\delb_0\bar{\eta}_{ab}|\lesssim (s/t)^2s^{-1},
\quad
&&|\delb_0\bar{\eta}_{0b}|\lesssim  (s/t)s^{-1},
\quad
&&&|\delb_0\bar{\eta}_{00}|\lesssim (s/t)^2s^{-1}.
\endaligned
\end{equation}

In some circumstance we also need to treat homogeneous coefficients such as $(1-r/t)$. We claim the following estimate
\begin{equation}\label{eq1-24-aout-2025}
\big|\mathscr{Z}^I(1-r/t)\big|\lesssim 
\begin{cases}
t^{p-k},\quad & p>k
\\
\frac{|t-r|}{t},\quad &p=k,
\end{cases}
\quad\text{in}\quad \{3t/4\leq r\leq 3t\}
\end{equation}
\subsection{Estimates on deformation tensors}
A direct calculation leads to
\begin{equation}
\aligned
\pi[Z]^{\alpha\beta} =& -Z(g^{\alpha\beta}) + g^{\alpha\mu}\del_{\mu}Z^{\beta} + g^{\beta\nu}\del_{\nu}Z^{\alpha}
\\
=&\begin{cases}
-\del_{\delta}H^{\alpha\beta},\quad & Z = \del_{\delta},
\\
-L_aH^{\alpha\beta} + H^{\alpha\mu}\del_{\mu}Z^{\beta} + H^{\beta\nu}\del_{\nu}Z^{\alpha},\quad &Z = L_a.
\end{cases}
\endaligned
\end{equation}
Then in the semi-hyperboloidal frame,
\begin{equation}\label{eq7-23-aout-2025}
\underline{\pi[Z]}^{\alpha\beta} = 
\begin{cases}
\,-\del_{\delta}\Hu^{\alpha\beta} + \del_{\delta}\big(\Psiu_{\mu}^{\alpha}\Psiu_{\nu}^{\beta}\big)H^{\mu\nu},\quad &Z = \del_{\delta},
\\
{\aligned
&-L_a\Hu^{\alpha\beta} 
+ \Hu^{\alpha\gamma}\delu_{\gamma}\underline{Z}^{\beta} + \Hu^{\gamma\beta}\delu_{\gamma}\underline{Z}^{\alpha}
\\
&\quad+ L_a\big(\Psiu_{\mu}^{\alpha}\Psiu_{\nu}^{\beta}\big)H^{\mu\nu}
- \Hu^{\alpha\gamma}Z^{\nu}\delu_{\gamma}(\Psiu_{\nu}^{\beta})
- \Hu^{\gamma\beta}Z^{\mu}\delu_{\gamma}(\Psiu_{\mu}^{\alpha}),
\endaligned}
\quad& Z = L_a.
\end{cases}
\end{equation}
We remark that in the above expression, $\Psiu_{\mu}^{\alpha}$ are interior-homogeneous of degree zero, and $Z^{\nu}$ are interior-homogeneous of degree $+1$. Then
$$
\deg\big(\del_{\delta}\big(\Psiu_{\mu}^{\alpha}\Psiu_{\nu}^{\beta}\big)\big) = -1,
\quad 
\deg\big(\delu_{\gamma}\underline{Z}^{\alpha}\big) = \deg\big(Z^{\nu}\delu_{\gamma}(\Psiu_{\nu}^{\beta})\big) = \deg\big(L_a\big(\Psiu_{\mu}^{\alpha}\Psiu_{\nu}^{\beta}\big)\big) = 0
$$
We also remark that in $\Mcal^{\Hcal}_{[s_0,s_1]}$,
\begin{equation}\label{eq5-24-aout-2025}
0\leq 1-r/t\lesssim  (s/t)^2. 
\end{equation}
We thus have the following estimates: 
\begin{equation}\label{eq8-23-aout-2025}
\big|\underline{\pi[Z]}^{\alpha\beta}\big|_{p,k}\lesssim_p
\begin{cases}
|\del_{\delta} H|_{p,k} + t^{-1}|H|_{p,k},\quad &Z = \del_{\delta}
\\
|L_aH|_{p,k} + |H|_{p,k},\quad &Z = L_a,
\end{cases}\quad \text{in}\quad \Mcal^{\Hcal}_{[s_0,s_1]},
\end{equation}
\begin{equation}\label{eq9-23-aout-2025}
\big|\del(\underline{\pi[Z]}^{\alpha\beta})\big|_{p,k}
\lesssim_p 
\begin{cases}
|\del\del_{\delta} H|_{p,k} + t^{-1}|\del H|_{p,k} + t^{-2}|H|_{p,k},\quad &Z = \del_{\delta},
\\
|\del L_a H|_{p,k} + |\del H|_{p,k} + t^{-1}| H|_{p,k},\quad &Z = L_a,
\end{cases}\quad \text{in}\quad \Mcal^{\Hcal}_{[s_0,s_1]}.
\end{equation}
These estimates are sufficient in the region $\{r\leq 3t/4\}\cap\Mcal^{\Hcal}_{[s_0,s_1]}$. However when near the light-cone. we need to make a semi-hyperboloidal decomposition in order to clarify the null structure.

W then make the following calculation.
\begin{lemma}\label{lem1-24-aout-2025}
In $\Mcal^{\Hcal}_{[s_0,s_1]}$,
\begin{equation}\label{eq4-24-aout-2025}
\underline{\pi[L_a]}^{00} = -L_a\Hu^{00} - 2(x^a/t)\Hu^{00}.
\end{equation}
\end{lemma}
\begin{proof}
In \eqref{eq7-23-aout-2025}, when $\alpha = \beta = 0$,
\begin{equation}\label{eq3-24-aout-2025}
\aligned
\pi[L_a]^{00} =& -L_a\Hu^{00} 
+ \Hu^{0\gamma}\delu_{\gamma}\underline{Z}^{0} + \Hu^{\gamma0}\delu_{\gamma}\underline{Z}^{0}
\\
&+ L_a\big(\Psiu_{\mu}^0\Psiu_{\nu}^0\big)H^{\mu\nu}
- \Hu^{0\gamma}Z^{\nu}\delu_{\gamma}(\Psiu_{\nu}^0)
- \Hu^{\gamma0}Z^{\mu}\delu_{\gamma}(\Psiu_{\mu}^0).
\endaligned
\end{equation}
Then we remark that when $Z = L_a = t\delu_a$,
\begin{equation}
\underline{Z}^{\gamma} = t\delta_a^{\gamma},
\end{equation}
and
\begin{equation}
L_a\big(\Psiu_{\mu}^0\Psiu_{\nu}^0\big)
=
\begin{cases}
0,\quad &\mu=\nu=0,
\\
-\delta_{ac} + (x^ax^c/t^2) = \underline{\eta}_{ac},\quad &\mu = 0,\nu = c,
\\
\frac{\delta_{ac}x^d + \delta_{ad}x^c}{t} - \frac{2x^ax^cx^d}{t^3},\quad &\mu = c,\nu = d.
\end{cases}
\end{equation}
\begin{equation}
\delu_{\gamma}\big(\Psi_{\nu}^0\big)
=\begin{cases}
0,\quad &\nu=0,
\\
(x^a/t^2),\quad &\gamma = 0,\nu = a,
\\
\frac{x^ax^c}{t^3} - \frac{\delta_{ac}}{t},\quad &\gamma = c,\nu = a.
\end{cases}
\end{equation}
Then we calculate:
\begin{subequations}\label{eq2-24-aout-2025}
\begin{equation}
\Hu^{0\gamma}\delu_{\gamma}\underline{Z}^{0} + \Hu^{\gamma0}\delu_{\gamma}\underline{Z}^{0}
= 0.
\end{equation}
\begin{equation}
\aligned
\frac{1}{2}L_a\big(\Psiu_{\mu}^0\Psiu_{\nu}^0\big)H^{\mu\nu}
=& -H^{a0} + (x^d/t)H^{ad} + (x^ax^c/t^2)\big(H^{0c} - (x^d/t)H^{dc}\big)
\\
=& -\Hu^{a0} + (x^ax^c/t^2)\Hu^{0c} = \underline{\big((x^ax^c/t^2) - \delta_{ac}\big)\Hu^{0c}}.  
\endaligned
\end{equation}
\begin{equation}
\aligned
- \Hu^{0\gamma}Z^{\nu}\delu_{\gamma}(\Psiu_{\nu}^0)
- \Hu^{\gamma0}Z^{\mu}\delu_{\gamma}(\Psiu_{\mu}^0)
=& - t\Hu^{0\gamma}\delu_{\gamma}(\Psiu_{a}^0)
- t\Hu^{\gamma0}\delu_{\gamma}(\Psiu_{a}^0)
\\
=&-2(x^a/t) \Hu^{00}\, \underline{- 2\Hu^{0c}\big((x^ax^c/t^2)-\delta_{ac}\big)}.
\endaligned
\end{equation}
\end{subequations}
Substitute \eqref{eq2-24-aout-2025} into \eqref{eq3-24-aout-2025}, and remark that the \underline{underlined terms} cancel each other (!). Thus we obtain \eqref{eq4-24-aout-2025}
\end{proof}
Now we are ready to establish the following estimates.
\begin{proposition}\label{prop1-24-aout-2025}
Assume that \eqref{eq-com-condition} holds and that \eqref{eq-USA-condition} holds with a sufficiently small $\eps$. Then in $\Mcal^{\Hcal}_{[s_0,s_1]}$,
\begin{equation}\label{eq8-aout-2024}
\big|\underline{\pi[Z]}^{00}\big|_{p,k}\lesssim_p
\begin{cases}
|\del \Hu^{00}|_{p,k} + t^{-1}|H|_{p,k},\quad &Z = \del_{\delta},
\\
|L_a\Hu^{00}|_{p,k} + |\Hu^{00}|_{p,k},\quad &Z = L_a.
\end{cases}
\end{equation}
\begin{equation}\label{eq9-24-aout-2025}
\big|\del\big(\underline{\pi[Z]}^{00}\big)\big|_{p,k}\lesssim_p 
\begin{cases}
|\del \del \Hu^{00}|_{p,k} + t^{-1}|\del H|_{p,k} + t^{-2}|H|_{p,k},\quad &Z = \del_{\delta},
\\
|\del \Hu^{00}|_{p+1,k+1} + t^{-1}|H|_{p,k} ,\quad & Z = L_a.
\end{cases}
\end{equation}
\end{proposition}
\begin{proof}
These are direct consequences of \eqref{eq7-23-aout-2025} (for $Z = \del_{\delta}$) and \eqref{eq4-24-aout-2025} (for $Z = L_a$).
\end{proof}
\begin{remark}
The component $\Hu^{00}$ enjoys the a similar behavior of $H^{\Ncal00}$. In fact, we remark that
$$
\thetau^0 = \diff t - (x^a/t)\diff x^a = \diff t - \diff r + (1-r/t)\diff r.
$$
Then
\begin{equation}\label{eq10-23-aout-2025}
\Hu^{\alpha\beta} = H(\thetau^0,\thetau^0) 
= H^{\Ncal00} + 2(1-r/t)H(\diff t - \diff r, \diff r) + (1-r/t)^2H(\diff r,\diff r).
\end{equation}
The ``error'' terms contain $(1-r/t)$ which is equivalent to $(s/t)^2$.
\end{remark}

\subsection{Semi-hyperboloidal decomposition of Dirac commutators in hyperboloidal region}
\begin{lemma}\label{lem6-25-aout-2025}
Let $\Psi$ be a sufficiently regular spinor field and $Z$ an admissible vector field. Then
\begin{equation}
\aligned
[\widehat{Z},\opDirac]\Psi 
=& -\frac{1}{2}\underline{\pi[Z]}^{\alpha0}\delu_{\alpha}\cdot\widehat{\del_t}\Psi 
- \frac{1}{2t}\underline{\pi[Z]}^{\alpha b}\delu_{\alpha}\cdot\widehat{L_b}\Psi 
\\
&- \frac{1}{8}g^{\mu\nu}\underline{\pi[Z]}^{\alpha0}\delu_{\alpha}\cdot\del_{\mu}\cdot\nabla_{\nu}\del_t\cdot\Psi 
- \frac{1}{8t}g^{\mu\nu}\underline{\pi[Z]}^{\alpha b}
\delu_{\alpha}\cdot\del_{\mu}\cdot\nabla_{\nu}L_b\cdot\Psi
\\
&-\frac{1}{4}\pi[Z]^{\alpha\beta}\nabla_{\alpha}\del_{\beta}\cdot\Psi
-\frac{1}{4}[Z,W]\cdot\Psi.
\endaligned
\end{equation}
\end{lemma}
\begin{proof}
Recall \eqref{eq5-21-aout-2025},
$$
\aligned
[\widehat{Z},\opDirac]\Psi  
& =   - \frac{1}{2} \pi[Z]^{\alpha\beta} \del_\alpha \cdot\nabla_{\beta} \Psi
-\frac{1}{4}\pi[Z]^{\alpha\beta}\nabla_{\alpha}\del_{\beta}\cdot\Psi
-\frac{1}{4}[Z,W]\cdot\Psi
\\
&= -\frac{1}{2}\underline{\pi[Z]}^{\alpha\beta}\delu_{\alpha}\cdot\nabla_{\delu_{\beta}}\Psi
-\frac{1}{4}\pi[Z]^{\alpha\beta}\nabla_{\alpha}\del_{\beta}\cdot\Psi
-\frac{1}{4}[Z,W]\cdot\Psi
\\
&= -\frac{1}{2}\underline{\pi[Z]}^{\alpha0}\delu_{\alpha}\cdot\nabla_{t}\Psi 
- \frac{1}{2t}\underline{\pi[Z]}^{\alpha b}\delu_{\alpha}\cdot\nabla_{L_b}\Psi 
-\frac{1}{4}\pi[Z]^{\alpha\beta}\nabla_{\alpha}\del_{\beta}\cdot\Psi
-\frac{1}{4}[Z,W]\cdot\Psi
\\
&= -\frac{1}{2}\underline{\pi[Z]}^{\alpha0}\delu_{\alpha}\cdot\widehat{\del_t}\Psi 
- \frac{1}{8}g^{\mu\nu}\underline{\pi[Z]}^{\alpha0}\delu_{\alpha}\cdot\del_{\mu}\cdot\nabla_{\nu}\del_t\cdot\Psi 
\\
&\quad - \frac{1}{2t}\underline{\pi[Z]}^{\alpha b}\delu_{\alpha}\cdot\widehat{L_b}\Psi 
- \frac{1}{8t}g^{\mu\nu}\underline{\pi[Z]}^{\alpha b}
\delu_{\alpha}\cdot\del_{\mu}\cdot\nabla_{\nu}L_b\cdot\Psi
\\
&\quad-\frac{1}{4}\pi[Z]^{\alpha\beta}\nabla_{\alpha}\del_{\beta}\cdot\Psi
-\frac{1}{4}[Z,W]\cdot\Psi.
\endaligned
$$
\end{proof}
We remark that among these terms, the first and the third are quadratic, while the rest are cubic in nature. 

\subsection*{Acknowledgments.} 

PLF was supported by the project ANR-23-CE40-0010-02 funded by the Agence Nationale de la Recherche (ANR), and the MSCA Staff Exchange Project 101131233 funded by the European Research Council (ERC). 
YM was supported by the Shaanxi Provincial Basic Science Research Fund (Mathematics and Physics) under the grant number~22JSZ003. 
WDZ was supported by the CSC scholarship program (Project ID: 202406280362).
 

\bibliography{references}

\addcontentsline{toc}{section}{References}
\pagestyle{plain}

\end{document}